\documentclass[journal,12pt, onecolumn, draftclsnofoot]{IEEEtran}
\usepackage{graphicx}
\usepackage{amsmath}
\usepackage{amssymb}
\usepackage{amsfonts}
\usepackage{bm}
\usepackage{bbm}
\usepackage{subcaption}
\usepackage{multirow}
\graphicspath{ {images/} }
\usepackage{hyperref}
\usepackage{cleveref}
\usepackage{xcolor}
\usepackage[nobiblatex]{xurl}
\usepackage{amsthm} 
\newtheorem{theorem}{Theorem}
\usepackage[linesnumbered,ruled,vlined]{algorithm2e}

\SetKwInput{KwInput}{Input}
\SetKwInput{KwOutput}{Output}
\SetKwInput{KwInit}{Initialization}

\newcommand{\R}{\mathbb{R}}
\DeclareMathOperator*{\argmin}{arg\,min}

\usepackage{titling}
\usepackage{xr}

\makeatletter
\newcommand*{\addFileDependency}[1]{
  \typeout{(#1)}
  \@addtofilelist{#1}
  \IfFileExists{#1}{}{\typeout{No file #1.}}
}
\makeatother

\ifCLASSINFOpdf
\else
\fi
\begin{document}
%
\title{
Hybrid Modeling of Regional COVID-19 Transmission Dynamics in the U.S.}
%
%
%

\author{Yue Bai, 
        Abolfazl Safikhani,
        and~George Michailidis
\thanks{Y. Bai is with the Department of Statistics, University of Florida, Gainesville, FL 32611 USA (email: baiyue@ufl.edu).} 
\thanks{A. Safikhani is with the Department of Statistics and the Informatics Institute, University of Florida, Gainesville, FL 32611 USA (email: a.safikhani@ufl.edu).}
\thanks{G. Michailidis is with the Department of Statistics and the Informatics Institute, University of Florida, Gainesville, FL 32611 USA (email: gmichail@ufl.edu). }
}

\maketitle

\begin{abstract}
The fast transmission rate of COVID-19 worldwide has made this virus the most important challenge of year 2020. Many mitigation policies have been imposed by the governments at different regional levels (country, state, county, and city) to stop the spread of this virus. Quantifying the effect of such mitigation strategies on the transmission and recovery rates, and predicting the rate of new daily cases are two crucial tasks. In this paper, we propose a hybrid modeling framework which not only accounts for such policies but also utilizes the spatial and temporal information to characterize the pattern of COVID-19 progression. Specifically, a piecewise susceptible-infected-recovered (SIR) model is developed while the dates at which the transmission/recover rates change significantly are defined as ``break points'' in this model. A novel and data-driven algorithm is designed to locate the break points using ideas from fused lasso and thresholding. In order to enhance the forecasting power and to describe additional temporal dependence among the daily number of cases, this model is further coupled with spatial smoothing covariates and vector auto-regressive (VAR) model. The proposed model is applied to several U.S. states and counties, and the results confirm the effect of ``stay-at-home orders'' and some states' early ``re-openings'' by detecting break points close to such events. Further, the model provided satisfactory short-term forecasts of the number of new daily cases at regional levels by utilizing the estimated spatio-temporal covariance structures. {They were also better or on par with other proposed models in the literature, including flexible deep learning ones.} Finally, selected theoretical results and empirical performance of the proposed methodology on synthetic data are reported which justify the good performance of the proposed method.
\end{abstract}


%

\section{Introduction}
\label{s:intro}

%
%
%
%
\IEEEPARstart{S}{ince} the first officially reported case in China in late December 2019, the SARS-CoV-2 virus spread worldwide within weeks. As of late May 2021, there have been $\sim 34$ million confirmed cases of COVID-19 in the United States alone and more than 170 million worldwide.
In response to the rapid growth of confirmed cases, followed by hospitalizations and fatalities, initially in the Hubei Province and in particular its capital Wuhan in China, and subsequently in Northern Italy, Spain and the tri-state area of New York, New Jersey and Connecticut, various mitigation strategies were put rapidly in place with the most stringent one being ``stay-at-home'' orders. The key purpose of such strategies was to reduce the virus transmission rate and consequently pressure on public health infrastructure \cite{anderson2020will}.
To that end, the California governor issued a ``stay-at-home'' order on March 19, 2020, that was quickly followed by another 42 states by early April.
All states with such orders proceeded with multi-phase reopening plans starting in early May, allowing various non-essential business to operate, possibly at reduced capacity levels to enforce social distancing guidelines.
In addition, mask wearing mandates also came into effect \cite{cdcmask} as emerging evidence from clinical and laboratory studies showed that masks reduce the spread \cite{leung2020respiratory}. However, these reopening plans led to a substantial increase in the number of confirmed COVID-19 cases in many US states, followed by increased number of fatalities throughout the summer of 2020, concentrated primarily in the Southern US states. Different states and local communities adopted and implemented different  non-pharmaceutical interventions to reduce infections, but a ``3rd wave" emerged in the fall of 2020 with cooling temperatures and people spending more time indoors. Further, during late fall of 2020,
various variants of concerns started emerging around the world, characterized by higher transmission capabilities and potentially increased severity based on hospitalizations and fatalities \cite{horby2021nervtag}.
Variants that exhibited a certain degree of spread in the US include B.1.1.7 (first detected in the United Kingdom), B.1.351 (first detected in South Africa), B.1.427 and B.1.429 (first detected in California) and P.1 (first detected in Brazil). The B.1.1.7 variant went on to become the dominant one in the US by March 2021, displacing the original dominant strain B.1.2, while the B.1.427/429 ones represented about 15\% of the total based on genomic surveillance studies
\footnote{\url{https://covid.cdc.gov/covid-data-tracker/##variant-proportions}}.

The emergence of the COVID-19 pandemic led to the development of many data science and signal processing modeling approaches addressing diverse issues, including
forecasting progress of the disease, impact of non-pharmaceutical intervention strategies \cite{flaxman2020estimating}, methods to estimate the Infection and Case Fatality Rates (IFR/CFR) \cite{meyerowitz2020systematic}, pre-existing conditions and clinical factors that impact the CFR \cite{williamson2020factors}, computer audition for diagnosing COVID-19 \cite{qian2021recent}, image analysis for COVID-19 \cite{yaron2021point}, phylogenetic network analysis of Covid-19 genomes \cite{forster2020phylogenetic}, impact of aerosol transmission to public health \cite{anderson2020consideration}, guidelines for reopening critical social activities such as schools \cite{honein2021data}. Note that as Covid-19 progressed together with our knowledge about it, the range of topics addressed significantly expanded, while the focus also exhibited certain shifts. As an example, early on (March 2020) it was believed that the virus can spread through contaminated surfaces, known as fomites, and this informed both the World Health Organization (WHO)
and the Center of Disease Control (CDC) recommendations 
\footnote{\url{https://www.who.int/news-room/commentaries/detail/transmission-of-sars-cov-2-implications-for-infection-prevention-precautions}}
on surface cleaning and disinfection. However, subsequent studies and investigations of outbreaks pointed that the majority of transmissions occur through droplets and aerosols that led to a revision of recommendations by the WHO and the CDC \footnote{\url{https://www.cdc.gov/coronavirus/2019-ncov/science/science-briefs/sars-cov-2-transmission.html}}. 
It also led to new research on aerosol dispersion models and on the role of ventilation to mitigate transmission \cite{bazant2021guideline}. Nevertheless, forecasting the spread of the epidemic throughout its course (initially with the imposition of various mitigation strategies, and more recently through the emergence of more transmissible variants and the increased pace of vaccination campaigns around the world) has remained a key task and a number of signal processing approaches have been developed as briefly summarized next.

\vspace{-0.5cm}

\subsection{Related Work}

A number of epidemic models have been developed to analyze and predict COVID-19 transmission dynamics.
Mathematical models, such as the class of susceptible-infectious-recovered (SIR) models are widely used to model and forecast epidemic spreads. 
\cite{chen2020time} proposed a time-dependent SIR model and tracked transmission and recovery rates at each time point by employing ridge regression while \cite{anastassopoulou2020data} proposed a discrete-time susceptible-infectious-recovered-dead (SIRD) model and provided estimations of the basic reproduction number ($R_0$), and the infection mortality and recovery rates by least squares method. 
Moreover, {\cite{wang2020epidemiological} and \cite{wangping2020extended} built an extended SIR model with time-varying transmission rates and implemented a Markov Chain Monte Carlo algorithm (MCMC) to obtain posterior estimates and credible intervals of the model parameters.}

A number of models focused on identifying a change in the parameters of the underlying model employed. For example, 
\cite{dehning2020inferring} combined the widely used SIR model (see Section \ref{sec:models}) with Bayesian parameter inference through MCMC algorithms, assuming a time-dependent transmission rate.  Instead of directly estimating a change point in the transmission rate and the other parameters in the SIR model, they assumed a fixed value on the number of the change points, and imposed informative prior distributions on their locations, as well as the transmission rate based on information from intervention policies. Further, \cite{jiang2020time} proposed to model the time series of the log-scaled cumulative confirmed cases and deaths of each country via a piecewise linear trend model. They combined the self-normalization (SN) change-point test with the narrowest-over-threshold (NOT) algorithm \cite{baranowski2016narrowest} to achieve multiple change-point estimation. Moreover, \cite{voko2020effect} and \cite{Wagner2020}  analyzed the effect of social distancing measures adopted in Europe and the United States, respectively, using an interrupted time series (ITS) analysis of the confirmed case counts. Their work aim to find a change point in the time series data of confirmed cases counts for which there is a significant change in the growth rate.
In \cite{voko2020effect}'s paper, the change points were determined by linear threshold regression models of the logarithm of daily cases while \cite{Wagner2020} used an algorithm developed in \cite{fearnhead2019detecting}, based on an $L_0$ penalty on changes in slope to identify the change points.  Finally, \cite{bertozzi2020challenges}
utilized a branching process for modeling and forecasting the spread of COVID-19.

Another line of work employed spatio-temporal models for parameter estimation and forecasting the spread of COVID-19. For example, \cite{wang2020spatiotemporal} introduced 
an additive varying coefficient model and coupled it with a non-parametric approach for modeling the data, to study spatio-temporal patterns in the spread of COVID-19 at the county level. Further, \cite{srivastava2020learning} proposed a heterogeneous infection rate model with human mobility from multiple regions and trained it using weighted least squares at regional levels while \cite{qi2020covid} fitted a generalized additive model (GAM) to quantify the province-specific associations between meteorological variables and the daily cases of COVID-19 during the period under consideration.

In addition to mathematical methods, many machine learning/deep learning methods were applied for forecasting of COVID-19 transmission. {For example,
\cite{moftakhar2020exponentially} 
and \cite{chimmula2020time}  employed Artificial Neural Networks (ANN) and Long Short-Term Memory (LSTM) type deep neural networks to forecast future COVID-19 cases in Iran and Canada, respectively,} while \cite{hu2020artificial} developed a modified stacked auto-encoder for modeling the transmission dynamics of the epidemic.  Review paper \cite{rahimi2021review} presents a summary of recent COVID-19 forecasting models. 

The previous brief overview of the literature indicates that there are two streams of models, the first mechanistic and the second statistical in nature. The former (SIR/SIRD) describe key components of the transmission chain and its dynamics and have proved useful in assessing scenarios of the evolution of a contagious disease, by altering the values of key model parameters. However, they are macroscopic in nature and can not easily incorporate additional information provided either by mitigation strategies or other features, such as movement of people assessed through cell phone data. Statistical models can easily utilize such information in the form of covariates to improve their forecasting power. However, they primarily leverage correlation patterns in the available data, that may be noisy, especially at more granular spatio-temporal scales (e.g., county or city level) that are of primary interest to public health officials and policy makers.

To that end, this paper aims to develop an interpretable \textit{hybrid} model that combines a mechanistic and a statistical model, that respects the theoretical transmission dynamics of the former, but also incorporates additional spatio-temporal characteristics resulting in improved forecasting capabilities at fairly granular spatio-temporal scales.

Specifically, we analyze confirmed cases and deaths related to COVID-19 from several states and counties/cities in the United States from March 1st, 2020 to March 31st, 2021. In the absence of non-pharmaceutical interventions, the spread of COVID-19 can be modeled by a SIR model with fixed transmission and recovery rates. One of the main reasons to select that model as the building block for the proposed methodology is that the transmission and recovery rates are easy to interpret and hence can be used in policy decision making. However, many diverse mitigation policies were put in place at different regional levels in the U.S. Thus, the simple SIR model may not be a good fit for the data. Instead, we propose a piecewise stationary SIR model (Model~1), i.e. the SIR model parameters may change at certain (unknown) time points. Such time points are defined as ``break (change) points". Unlike some other methods discussed in the literature review, in our modeling framework, the number of change points and their locations are assumed to be \textit{unknown} and must be inferred from the data. Such flexibility on the modeling front allows inferring potentially different temporal patterns across different regions (states or counties), and yields a data-driven segmentation of the data which subsequently improves the fit (see more details in Section~\ref{sec:results}), but also complicates the model fitting procedures. To that end, a novel data-driven algorithm is developed to detect all break points, and to estimate the model parameters within each stationary segment. Specifically, we define certain time blocks and assume the model parameters are fixed during each block of time points. Then, a fused lasso penalty is used to estimate all model parameters \cite{tibshirani2005sparsity}. This procedure is further coupled with hard-thresholding and exhaustive search steps to estimate the number and location of change points (details provided in Section~\ref{sec:models}). To enhance the forecasting power of the model and to capture additional spatial and temporal dependence not explained through the SIR model, the piecewise constant SIR model is coupled with spatial smoothing (Model~2) and time series components (Model~3). The former is accomplished through the addition of a spatial effect term which accounts for the effect of neighboring regions, while the latter through a Vector Auto-Regressive (VAR)  component to capture additional auto-correlations among new daily cases and deaths. Capturing the spatio-temporal dependence through Model~3 aids in reducing the prediction error significantly (sometimes around 80\%) compared to the piecewise SIR model which confirms the usefulness of a hybrid modeling framework (for more details see Section~\ref{sec:results}). To verify the applicability of the proposed methodology to other data sets with similar characteristics, the developed algorithm is tested over several simulation settings and exhibits very satisfactory performance (details in Section~\ref{sec:sims}) and some theoretical properties of the proposed method (prediction consistency, as well as detection accuracy) are established in Appendix~\ref{sec:main_proof}.

The remainder of the paper is organized as follows. In Section~\ref{sec:models}, proposed statistical models are introduced and data-driven algorithms are described to estimate their parameters. The proposed algorithms are tested on various simulation settings and the results are reported in Section~\ref{sec:sims}. The proposed models are applied to several U.S. states and counties and the results are described in Section~\ref{sec:results}. Finally, some concluding remarks are drawn in Section~\ref{sec:conclusion}.


\section{A Family of Spatio-temporal Heterogeneous SIR Models}\label{sec:models}


The proposed class of hybrid models leverages the framework of the SIR model, which is presented next to set up key concepts.

\subsection{The Standard SIR Model with Fixed Transmission and Recovery Rates}

The standard SIR model \cite{kermack1927contribution} is a mechanistic model, wherein the total population is divided into the following three compartments: susceptible (uninfected),  infected, and recovered (healed or dead). 
It is assumed that each infected individual, on average, infects $\beta$ other individuals per unit time, and each infected individual recovers at rate $\gamma$.  The two key model parameters, the transmission rate $\beta$ and recovery rate $\gamma$, are assumed to be fixed over time.
The temporal evolution of the SIR model is governed by the following system of three ordinary differential equations: 
\begin{equation}\label{SIR-ode}
\frac{d S}{dt} = -\beta\frac{S I_f}{N}, \ \ \ \
\frac{d I_f}{dt}= \beta\frac{S I_f}{N} - \gamma I_f, \ \ \ \ 
\frac{d R}{dt} = \gamma I_f,
\end{equation}
where $S$,  $I_f$ and $R$ represent the individuals in the population in the susceptible, infected and recovered stages, respectively. Note that the variables $S$, $I_f$ and $R$ always satisfy $S + I_f+ R = N$, where $N$ is the total population size.  In this formulation, we ignore the change in the total population, so that $N$ remains constant over time. Due to the fact that COVID-19 records are discrete in time ($\Delta t = 1$ day), we consider the discrete-time version of SIR model, so that for each $t= 1, \dots, T-1$, the system comprises of the following three difference equations
\begin{IEEEeqnarray}{lll}\label{SIR-diff-eqn}
S(t+1)-S(t) =-\beta  \frac{S(t)I_f(t)}{N}, \\ 
I_f(t + 1) - I_f(t) = \beta\frac{S(t)I_f(t)}{N} - \gamma I_f(t),\\ 
R(t + 1) - R(t) = \gamma I_f(t),
\end{IEEEeqnarray}
where $S(t)$ stand for the number of susceptible individuals at time $t$, $I_f(t)$ for the number of infected ones and finally $R(t)$ for those recovered. Note that these three variables $S(t)$, $I_f(t)$ and $R(t)$ still satisfy the constraint $S(t)+I_f(t)+R(t)=N$.

Notice that the number of infected cases $I_f(t)$ is not observable. Specifically, confirmed COVID-19 case counts may not capture the total infected cases due to limited testing availability/capacity, especially at the beginning of the pandemic (testing has been primarily restricted to individuals with moderate to severe symptoms). For example, in the United States, over 90\% of COVID-19 infections were not identified/reported at the beginning of the pandemic \cite{wu2020substantial, lau2020evaluating}. To that end, we define a relationship between the true infected cases and observed/recorded infected cases through an under-reporting function. Specifically, define $\Delta I(t) = \Delta I_f(t) \times  (1-u(t+1))$ for $t = 1, \dots, T-1$, where $\Delta I(t) =  I(t + 1) - I(t)$, $\Delta I_f(t) =  I_f(t + 1) - I_f(t)$, $I_f(t)$ is the true infected cases, $I(t)$ is the observed/recorded infected cases, and finally $u(t)$ is the under-reporting function. We consider a parametric model for the function $u(t)$ (see more details in Section~\ref{sec:sims}). Note that we could also use non-parametric method to solve the under-reporting issue, but since the sample size of the time series is limited, in this paper, we consider a parametric method instead. Given the observed infected cases $I(t)$ and under-reporting function $u(t)$, one can transform the data back to $I_f(t)$ by
\begin{IEEEeqnarray}{lll}\label{eq:transform}
I_f(1) = \frac{I(1)}{1- u(1)},\\
I_f(t)= \frac{\Delta I(t-1)}{1- u(t)} + I_f(t-1) = \sum_{i = 2}^t \frac{\Delta I(i-1)}{1- u(i)}   + \frac{I(1)}{1- u(1)}, 
\end{IEEEeqnarray}
for $t = 2, \dots, T$. 

Combining the difference equation~\eqref{SIR-diff-eqn} normalized by the total population with the transformations stated in Equation~\eqref{eq:transform} yield to the following simple linear equations: 
\begin{IEEEeqnarray}{lll}\label{two_eq_1}
 \underbrace{\left ( \begin{array}{c}
    \frac{\Delta I(t)}{N(1- u(t+1))}\\
        \frac{\Delta R(t)}{N} 
    \end{array}\right)}_{Y_t}
= \underbrace{\left ( \begin{array}{cc}
        \frac{S(t)}{N^2} I_f(t) & - \frac{1}{N} I_f(t) \\
       0  & \frac{1}{N}I_f(t)
    \end{array}\right)}_{X_t}
    \underbrace{\left ( \begin{array}{c}
         \beta  \\
         \gamma
    \end{array}\right)}_{B},
\end{IEEEeqnarray}
for each $t= 1, \dots, T-1$ where $\Delta I(t) =  I(t + 1) - I(t)$ and $\Delta R(t)= R(t + 1) - R(t)$.

Next, we extend the standard SIR model to accommodate temporal and spatial heterogeneity as well as to include stochastic temporal components. The former is achieved, by allowing the transmission and recovery rates to vary over time and through the inclusion of an additional term in \eqref{two_eq_1} that captures spatial effects while the latter is achieved through adding a vector auto-regressive component \cite{lutkepohl2005new}.

\vspace{-0.35cm}

\subsection{Modeling Framework: A Stochastic Piecewise Stationary SIR Model with Spatial Heterogeneity}


Compared to the standard SIR model, the proposed modeling framework makes three major changes/modifications. First, the assumptions underlying the transmission and recovery rates of the standard SIR model are stringent. Both environmental factors and changes in population behavior can lead to time varying behavior and this has been the case for Covid-19; see, e.g., discussion in \cite{giordano2020modelling}. Variants of the SIR model with time varying parameters have been proposed in the literature \cite{liu2012infectious}. For our application, we assume that the transmission and recovery rates are \textit{piecewise constant} over time, reflecting the fact that their temporal evolution is impacted by intervention strategies and environmental factors (Model~1). Second, the standard SIR model and its piecewise stationary counterpart do not account for any influence due to inter-region mobility and travel activity. We incorporate such inter-region information by considering the influence exerted by its few neighboring regions such as cities, counties or states (Model~2); see also  \cite{srivastava2020learning}. Third, the standard homogeneous SIR model, previously discussed, is deterministic; hence, the output of the model is fully determined by the parameter values of the transmission and recovery rates and the initial conditions. Its stochastic counterpart \cite{bartlett1949some, bailey1953total}, possesses some inherent randomness. Alternatively, the general stochastic epidemic model can be approximated by the  stochastic differential equation (see e.g. \cite{greenwood2009stochastic}):
\begin{equation}\label{two_eq_sde}
dX(t) = f(X(t)) dt + G(X(t)) dW(t),
\end{equation}
where the random variables $S(t)$ and $I_f(t)$ are continuous, 
\begin{equation}
X(t) = \left ( \begin{array}{c}
   S(t)\\[2pt]
    I_f(t)
    \end{array}\right)
 , \  
 f =\left ( \begin{array}{c}
  -\beta \frac{SI_f}{N}\\[2pt]
  \beta \frac{SI_f}{N} - \gamma I_f
    \end{array}\right) ,\ 
 G =  \left ( \begin{array}{cc}
  -\sqrt{\beta \frac{SI_f}{N}} & 0 \\[2pt]
 \sqrt{\beta \frac{SI_f}{N}} &  -\sqrt{\gamma I_f }
    \end{array}\right), 
\end{equation}
and $W = (W_1, W_2)^\prime$ is a vector of two independent Wiener processes, i.e., $W_i(t) \sim \mathcal{N}(0, t)$. Given \eqref{two_eq_sde}, the stochastic SIR model can be written as:
{\begin{align}\label{two_eq_approx}
   Y_t
= {X_t}{B}
    +  {\epsilon_t},
\end{align}}
where
{\small\begin{align}\label{two_eq_approx_detail}
 &Y_t =  {\left ( \begin{array}{c}
   \frac{\Delta I(t)}{1- u(t+1)}\\
       \Delta R(t)
    \end{array}\right)},
    B = 
    {\left ( \begin{array}{c}
         \beta  \\
         \gamma
    \end{array}\right)}, \nonumber\\
 &   X_t 
= {\left ( \begin{array}{cc}
       \frac{S(t)}{N} I_f(t) & -I_f(t) \\
       0  & I_f(t)
    \end{array}\right)}, \nonumber\\
  & {\epsilon_t}= {\left ( \begin{array}{c}
      \sqrt{ \frac{\beta S(t)}{N} I_f(t) }\Delta W_1(t) -\sqrt{\gamma I_f(t) }\Delta W_2(t) \\
      \sqrt{\gamma I_f(t) }\Delta W_2(t) 
    \end{array}\right)},\nonumber
\end{align}}
with the increments $\Delta W_1(t)$ and $\Delta W_2(t)$ being two independent normal random variables, i.e., $\Delta W_i(t)\sim \mathcal{N} (0, \Delta t)$. It can be seen that the resulting regression model, based on the discrete analogue of  \eqref{two_eq_approx}, will have an error term exhibiting temporal correlation, driven by $I_f(t)$ and $R(t)$. 
An examination of the temporal correlation patterns in COVID-19 data (see left two panels in Figures~\ref{fig:acf} and \ref{fig:acf_county} in the supplementary material) supports this finding. To that end, we model the error process as a Vector Auto-Regressive (VAR) with lag $p$ (Model~3). The corresponding temporal correlation plots for the residuals after inclusion of the VAR($p$) component are depicted in the right two panels in  Figures~\ref{fig:acf} and \ref{fig:acf_county} in the supplementary material and clearly show the importance of considering such an error structure. The piecewise stationary SIR model with spatial heterogeneity and a VAR($p$) error process is given by
\begin{IEEEeqnarray}{lll}\label{eq:model_var}
    Y_t = \sum_{j = 1}^{m_0+1}  X_t B^{(j)}  \mathbbm{1}_{\{t_{j-1}\leq t <t_j \}} + \alpha Z_{t} \nonumber\\
    \quad  + \phi_1 \varepsilon_{t-1} + \dots + \phi_p \varepsilon_{t-p}  + e_t, t =  1, \dots, T-1, 
\end{IEEEeqnarray}
where $\left \{t_1 , \dots , t_{m_0} \right\}$ are unknown $m_0$ ``change points" such that the transmission and recovery rates exhibit a change from $B^{(j)} = (\beta^{(j)}, \gamma^{(j)})^\prime$ to $B^{(j+1)} = (\beta^{(j+1)}, \gamma^{(j+1)})^\prime$ at time point $t_j$, while it remains fixed until the next break point. Hence, these break points divide the time series data into stationary segments. Moreover, $Z_t = \sum_{j = 1}^q\omega_{j} \left (  \frac{\Delta  I^j(t-1)}{N^j(1- u(t))}, \frac{\Delta  R^j(t-1)}{N^j} \right)^\prime $ is the weighted average of spatial effect over the neighboring regions at time $t$ where $N^j$  denotes the total population in neighboring region $j$; $\alpha$ is a spatial effect parameter; 
$\omega_{j}$'s are spatial weights such that $\sum_{j=1}^q \omega_{j} = 1$. The latter two parameters capture inter-region mobility patterns. Finally, $e_t$ is white noise with mean 0 and variance $\sigma^2$, and $\phi_1, \dots, \phi_p$ are the corresponding autoregressive parameters. In the sequel, model~\eqref{eq:model_var} with $\alpha=\phi_1= \ldots =\phi_p = 0$ is considered as Model~1 (piecewise SIR), the case of only restricting $\phi_1= \ldots =\phi_p = 0$ is considered as Model~2 (piecewise SIR with spatial effects) while the full model with no constrains is considered as Model~3. Notice that the number of change points $m_0$ and their locations are unknown and must be estimated from the data together with all other model parameters including $B^{(j)}$'s, $\alpha$, and $\phi_1, \ldots, \phi_p$. A brief discussion of the proposed algorithm to perform all such estimations is presented next.

\vspace{-0.5cm}

\subsection{Algorithm}

The estimation of the model parameters is accomplished in the following three steps: \textit{Step 1}: Fit Model 1 for each region of interest to obtain the change points; \textit{Step 2}: Obtain the transmission and recovery rates and spatial effect as in Model 2. \textit{Step 3}: Compute the residuals ($\widehat{\epsilon}_t$) from Step 2 and fit a VAR model to them, see \cite{lutkepohl2005new}. The rationale behind this step-wise algorithm is that assuming that the influence of the spatial effect component $Z_t$ is small, we can use Model 1 for each region of interest to estimate both the change points and the corresponding transmission and recovery rates. Denoting the final estimated change points by Model 1 as $ \widetilde{\mathcal{A}}_n^f = \left\lbrace \widetilde{t}_1^f, \ldots, \widetilde{t}_{\widetilde{m}^f}^f  \right\rbrace $, segment-specific transmission and recovery rates combined with an overall spatial effect can be readily estimated using least squares applied to augmented linear model which includes all segments concatenated to each other at time points $\widetilde{t}_j^f$'s and the spatial effect. Finally, the residuals of this augmented linear model utilizing the least squares estimates can be computed and additional least squares estimates on the residuals with its previous values in the design matrix can yield to estimates of autoregressive parameters. The difficult part of the algorithm is to estimate the number and locations of break points. Details of the algorithm are presented in the Appendix~\ref{sec:algo} while a brief summary is provided next. The first  step of the algorithm  aims  to select  candidate change points $\widehat{A}_n$ among blocks by solving a block fused lasso problem. The estimated change points obtained by the block fused lasso step includes all points with non-zero estimated parameters, which leads to overestimating the number of the true change points in the model. Nevertheless, the block fused lasso parameter estimates enjoy a prediction consistency property, which implies that the prediction error converges to zero with high probability as $n \rightarrow + \infty$. This result is stated and proved (see Theorem~1 in the Appendix~\ref{sec:main_proof}) under some mild conditions on the behaviour of the tail distributions of error terms. A hard-thresholding step is then added to reduce the over-selection problem from the fused lasso step by ``thinning out'' redundant change points exhibiting small changes in the estimated coefficients. After the hard-thresholding step, those candidate change points located far from any true change points will be eliminated when the block size is appropriate. On the other hand, there may be more than one selected change points remaining in small neighborhoods of each true change point. To remedy this issue, the remaining estimated change points are clustered while in each cluster, an exhaustive search examines every time point inside the neighborhood search region based on the cluster of candidate change points and selects the best time point as the final estimated change point.

\vspace{-0.3cm}

\section{Simulation Studies}\label{sec:sims}

We evaluate the performance of the proposed models on their predictive accuracy, change point detection and parameter estimation. We consider three simulation scenarios (see additional simulation settings in \ref{sec:sim} in the supplementary material). The details of the simulation settings for each scenario are explained in Section  \ref{Sim_Scenarios}. All results are averaged over 100 random replicates.  

We assess the results for the three models presented: Model 1, the piecewise stationary SIR model; Model 2, the piecewise stationary SIR model with spatial effect; and Model 3, the piecewise stationary SIR model with spatial effect and a VAR($p$) error process. 
The out-of-sample Mean Relative Prediction Error (MRPE) is used as the performance criterion defined as:
\begin{equation}\label{eq:MRPE}
     \mbox{MRPE}(I)= \frac{1}{n^{test}}\sum_{t= T+1}^{T + n^{test}} \left\vert\frac{\widehat{I}(t) - I(t) }{I(t)} \right\vert,
\end{equation}
where $n^{test}$ is the number of time points for prediction,
$\widehat{I}(t)$ is the predicted count of infected cases at time $t$, and $I(t)$ the observed one.
The MRPE of $R(t)$ can be obtained by respectively replacing the $\widehat{I}(t)$  and $I(t)$ with $\widehat{R}(t)$ and $R(t)$.
The predicted number of infected cases and recovered cases are defined as 
\begin{align}
    \widehat{I}(t) = I(t-1) + \widehat{\Delta I}(t-1), \ \widehat{R}(t) = R(t-1) + \widehat{\Delta R}(t-1), 
    \label{IR_est}
\end{align}
for all $t = T+1,\cdots, T+n^{test} $.
For change point detection, we report the locations of the estimated change points and the percentage of replicates that correctly identifies the change point.
This percentage is calculated as the proportion of replicates, where the estimated change points are close to each of the true break points. 
Specifically, to compute the selection rate,  
a selected break point is counted as a ``success'' for the $j$-th true break point, $t_j$, if it falls in the interval $ [ t_j - \frac{t_{j} - t_{j-1}}{5}, t_j + \frac{t_{j+1} - t_{j}}{5}  ] $, $j = 1,\dots, m_0$. 
We also report the mean and standard deviation of estimated parameters for each models. All results are reported in  Table~\ref{table_sim}.

\vspace{-0.4cm}

\subsection{Simulation Scenarios}\label{Sim_Scenarios}
We consider three different simulation settings. The SIR model's coefficients and under-reporting functions are depicted in Figure~\ref{fig:sim_rate} in the Section \ref{sec:sim}. 

\textit{Simulation Scenario A (Model 3 with no under-reporting):} the data are generated based on \eqref{eq:model_var} with piecewise constant transmission and recovery rates.  
We set the number of time points $T=200$, $m_0 =1$, the change point $t_1 = \lfloor \frac{T}{2}\rfloor =  100  $,  $\beta^{(1)} = 0.10$,  $\beta^{(2)} = 0.05$, $\gamma^{(1)} = 0.04$ and $\gamma^{(2)} = 0.04$ . For the spatial effect, we set $\alpha = 1$,
$\beta_s(t)= 0.10  - \frac{0.05t}{T-1}$,
$\gamma_s(t)= 0.04$,
$t = 1, \dots, T-1$.
We first generate the spatial effect data from SIR model in \eqref{eq:model_var} with parameter $\beta_s(t)$ and $\gamma_s(t)$   and generate the error term by VAR(1) model with the covariance matrix of the noise process $\Sigma_\varepsilon = 0.1 \mathbb{I}_2$, where $\mathbb{I}_2$ is the two-dimensional identity matrix.
By plugging in the spatial effect data and error term data, we generate the dataset of the response variable $Y_t$ from \eqref{eq:model_var}.
The autoregressive coefficient matrix has entries 0.8, 0, 0.2, 0.7 from top left to bottom right.  
We assume no under-reporting issue in this scenario, i.e., $u(t) = 0$, hence $\Delta I(t) = \Delta I_f(t)$, for $t = 1, \dots, T-1$.

\textit{Simulation Scenario B (Model 1 with exponentially decreasing under-reporting rate):} In this scenario, 
we set the number of time points $T= 250$, $m_0 =2$, the change points $t_1 =100 $ and $t_2 = 200$. 
We choose $\beta^{(1)} = 0.10$,  $\beta^{(2)} = 0.05$,  $\beta^{(3)} = 0.10$,
$\gamma^{(1)} = 0.04$, $\gamma^{(2)} = 0.06$, $\gamma^{(3)} = 0.04$.
Results are based on data generated from the SIR model in \eqref{eq:model_var} with 
$\beta(t) \sim  \text{Lognormal}(\sum_{j=1}^{m_0+1}\beta^{(j)}\mathbbm{1}_{\{t_{j-1}\leq t <t_j \}} , 0.005)$ and $\gamma(t) \sim  \text{Lognormal}(\sum_{j=1}^{m_0+1}\gamma^{(j)}\mathbbm{1}_{\{t_{j-1}\leq t <t_j \}} , 0.005)$.
 The under-reporting rate is chosen to change over time. Specifically, we set the under-reporting rate $u(t) = 1- \frac{1}{1 + be^{-a(t-1)} }$, $t = 1, \dots, T$, with $a = 0.05$ and $b = 10$. 

\textit{Simulation Scenario C (Model 3 with quadratically
decreasing under-reporting rate):} the data are generated based on \eqref{eq:model_var} with piecewise constant transmission and recovery rates. All the settings are exactly the same as those in scenario A except for the under-reporting rate. In this  scenario, we set under-reporting rate $u(t) =  1- \left(\frac{t + aT}{(1+ a)T}\right)^2$, $t = 1, \dots, T$, with $a = 0.5$. 

\begin{table*}[!ht]
\caption{\label{table_sim} Simulation results including selection rate, estimated parameters, and out-of-sample mean relative prediction error (MRPE). Note that $IR$ includes both $I(t)$ and $ R(t)$. }
\centering
\resizebox{0.70\textwidth}{!}{
\begin{tabular}{cccccccccc} 
 \hline
  \hline
 &  change point & truth   & mean  & std & selection rate \\
   \hline
  Scenario A & 1 & 0.5 & 0.5 & 0 & 1 \\ 
   \multirow{ 2}{*}{Scenario B}& 1 & 0.4 & 0.4 & 0 & 1 \\ 
   & 2 & 0.8 & 0.8 & 4e-04 & 1 \\ 
    Scenario C & 1 & 0.5 & 0.5 & 0 & 1 \\ 
\hline \hline
    & parameter  &true value &  mean & std \\
  \hline
    \multirow{ 5}{*}{Scenario A} & $\beta_1$& 0.1 &0.1 & 3e-04 \\  
     & $\beta_2$& 0.05  & 0.05 & 1e-04 \\  
     & $\gamma_1$ & 0.04 & 0.04 & 1e-04 \\ 
     & $\gamma_2$ & 0.04  & 0.04 & 1e-04 \\  
     & $\alpha$ &1  & 0.9945 & 0.0318 \\  
   \hline
  \multirow{ 7}{*}{Scenario B} & $\beta_1$ & 0.1  & 0.1001 & 0.0078 \\ 
   & $\beta_2$ & 0.05 &  0.0488 & 0.0106 \\ 
   & $\beta_3$ & 0.1 &  0.0986 & 0.0087 \\ 
   & $\gamma_1$ & 0.04 &  0.0392 & 0.0074 \\ 
   & $\gamma_2$ & 0.06 &  0.0589 & 0.0114 \\ 
   & $\gamma_3$ & 0.04 &  0.0391 & 0.0054 \\ 
   & $a$ & 0.05 &  0.0515 & 0.016 \\ 
  \hline
  \multirow{ 6}{*}{Scenario C} & $\beta_1$ & 0.1 & 0.0904 & 0.016 \\ 
  & $\beta_2$ & 0.05 & 0.0427 & 0.0123 \\ 
  & $\gamma_1$ & 0.04 & 0.0336 & 0.0108 \\ 
  & $\gamma_2$ & 0.04 & 0.0341 & 0.0099 \\ 
  & $\alpha$ & 1 & 1.84 & 1.6373 \\ 
  & $a$ & 0.5 & 0.3905 & 0.1977 \\ 
   \hline
    \hline
 & Model &  MRPE($IR$)   & MRPE($I(t)$) & MRPE($R(t)$)\\
   \hline
\multirow{ 3}{*}{Scenario A}
   & Model 1  & 0.000252 & 0.000362 & 0.000142 \\ 
   & Model 2  & 2.6e-05 & 3.9e-05 & 1.3e-05 \\ 
   & Model 3  & 2.4e-05 & 3.7e-05 & 1.2e-05 \\ 
   Scenario B & Model 1 & 0.00415 & 0.005888 & 0.002413 \\
   \multirow{ 2}{*}{Scenario C}
   & Model 3 (Transformed by $u(t)$) & 0.000162 & 0.000175 & 0.000149 \\ 
   & Model 3 (Not transformed) & 0.000911 & 0.000829 & 0.000993 \\ 
   \hline
\end{tabular}}
\end{table*}

\vspace{-0.4cm}

\subsection{Simulation Results}

The mean and standard deviation of the  location of selected change point, relative to the the number of time points $T$ -- i.e., $\widetilde{t}^f_1/T$ -- for all simulation scenarios are summarized in Table~\ref{table_sim}. The results clearly indicate that, in the piecewise constant setting, our procedure accurately detects the location of change points. The results of the estimated transmission rate $\widehat{\beta}$, recovery rate $\widehat{\gamma}$, spatial effect $\widehat{\alpha}$ and parameter of the  under-reporting rate function $\widehat{a}$ suggest that our procedure produces accurate estimates of the parameters, under the various under-reporting function settings ($b$ is assumed to be known). We generate additional 20 days worth of data to measure the prediction performance. The MRPE results for $I(t)$ and $R(t)$ are provided in Table~\ref{table_sim}. The results in scenario A  indicate that adding the spatial effect can significantly improve the prediction, when the spatial component influences the individual data series. The results in scenario C indicate that adding the under-reporting function $u(t)$ can significantly improve the prediction, when there is under-reporting in the data series.

\vspace{-0.3cm}

\section{Application to State and County Level COVID-19 Data in the U.S.}\label{sec:results}

\subsection{Data Description}

The COVID-19 data used in this study are obtained from \cite{NY:2020}. The curated data and code used in the analysis are available at the authors' GitHub  repository\footnote{\url{https://github.com/ybai69/COVID-19-Change-Point-Detection}}. The analysis is performed both at the state and county level and the raw data include both cases and deaths, as reported by state and local health departments and compiled by the NY Times. However, due to lack of complete information on recovered individuals (which is an important covariate in the models considered, but the daily number of recovered cases is only reported at the national level \cite{wang2020comparing}), we calculate the number of recovered cases for each region (state/county) as follows: 
the number of deaths in the region, multiplied by the nationwide cumulative recovered cases and divided by the nationwide deaths. Specifically, we assume that the recovery versus deceased ratio for each state/county is fixed, and can be well approximated by the nationwide recovery-to-death ratio.         
As coronavirus infections increase, while laboratory testing faces capacity constraints, reporting only confirmed cases and deaths leads to (possibly severe) under-estimation of the disease’s impact.
On April 14, 2020, CDC advised states to count both confirmed and probable cases and deaths. As more states and localities did so since then, in this study we focus on the combined cases, which include both confirmed and probable cases. 
The populations of states and counties are obtained from \cite{population2020}. Further, to decide which neighboring states/counties to include in Model 2, their distance to the target state/county of interest is used. The latter is obtained from the \cite{distance2010}. In the results presented, we define regions within 500 miles for states and 100 miles for counties/cities as neighboring ones in Models 2 and 3. For those areas with a large number of neighboring regions, such as New York state, we only consider the top five regions with the smallest distances.
We assume the probability rate of becoming susceptible again after having recovered from the infection to be 0.5\%. The reinfection rate in the short run ($\sim 6$ months) is believed to be very low. Some evidence from health care workers (median age 38 years) estimates it at 0.2\% \cite{lumley2020antibody}. Hence, the selected one seems a reasonable upper bound for the task at hand. Note that small variation of this rate did not impact the results.

We analyzed the daily count of COVID-19 cases at the state-level from March 1, 2020, to March 31, 2021 for the selected five states (NY, OR, FL, CA, TX), and at the county-level from the first day after March 1, 2020, when the region records at least one positive COVID-19 case to March 31, 2021. In addition, we analyzed the COVID-19 cases in the state of Michigan from March 1, 2020, to May 15, 2021.
In particular, the sample size $n$ for the five states presented next (NY, OR, FL, CA, TX) is $n = 395$ while for state of Michigan, Riverside County (CA) and Santa Barbara County (CA), $n= 432, 389, 381$, respectively. 
{For the under-reporting rate function $u(t)$, both quadratic function -$u(t) =  1- \left(\frac{t + aT}{(1+ a)T}\right)^2$- and exponential function -$u(t) = 1- \frac{1}{1 + be^{-a(t-1)} }$- are considered. The quadratic function achieved better performance in change point selection in the real data application (change point detection results using the exponential function are presented in Table~\ref{table_plan_all_exp} in the Supplement). Therefore, all presented results in this Section are based on the quadratic function.} Finally, to estimate the under-reporting parameter $a$, we perform a grid search within the interval $[0.1,0.3]$. The main reason for selecting this interval is that it matches with around 90\% of COVID-19 under-reporting rate at the beginning of the pandemic as investigated and reported in \cite{wu2020substantial,lau2020evaluating}. Note that we also assume that the COVID-19 under-reporting rate after December is very low as most of the regions built-up their testing capacity. Therefore, we set $u(t) = 1$ after December 2020.

Most of the states were selected due to being severely affected for a certain period of time during the course of COVID-19. The remaining regions illustrate interesting patterns gleaned from the proposed models. Let $I(t)$ and $R(t)$ denote the number of infected and recovered individuals (cases) on day $t$. Day 1 refers to the first day after March for which the region records at least one positive COVID-19 case. Figure \ref{fig:numbers_states} in the Supplement depicts the actual case numbers $I(t)$ and $R(t)$ in the six states considered.

\vspace{-0.4cm}

\subsection{Results for Selected U.S. States}
\label{sec:state}

The various models considered are applied on scaled versions (divided by their standard deviations) of the predictors matrix $X_t$  and the response vector $Y_t$. 
For Model 2, we consider four different types of weights: equal weights (Model 2.1),  distance-based weights (Model 2.2),  similarity-based weights (Models 2.3 and 2.4).
In Models 2.1 and 2.2, the neighboring regions are selected based on distance. For states, a threshold of 500 miles is used and the resulting neighbors are displayed in Table \ref{table_adj} in the supplementary material. When the spatial resolution is high (county level aggregated data), neighboring regions may exhibit similar patterns in terms of the evolution of transmission and recovery rates. Therefore, constructing weight matrices based on distance is meaningful as it is a common practice in spatial statistics \cite{cressie2015statistics}. Such spatial smoothing through proper weight matrices is especially helpful in increasing the statistical power through increasing the sample size, thus yielding more accurate predictions.

However, the evolution of COVID-19 may exhibit different patterns across neighboring states due to the coarse spatial resolution. Thus, defining weight matrices based on distance may not be ideal. Hence, Models 2.3 and 2.4 select regions based on similarities of infected/recovered cases. Similarity between the region of interest and the $j$-th potential similar region is defined as
\begin{equation}
    s_j =  \sqrt{\sum_{t=1}^{T-1} \left(\frac{\Delta I_f(t)}{N} - \frac{\Delta I^j_f(t)}{N^j}\right)^2 +\sum_{t=1}^{T-1} \left(\frac{\Delta R(t)}{N} - \frac{\Delta R^j(t)}{N^j}\right)^2 },
\end{equation}
where  $\Delta I_f(t) = \frac{\Delta I(t)}{ (1-u(t+1))}$ and $\Delta I_f^j(t) = \frac{\Delta I^j(t)}{ (1-u(t+1))}$.
Model 2.3 uses the top five regions with the smallest similarity score, while Model 2.4 uses all states in the country. For states, the resulting neighbors for Model 2.3 are displayed in Table \ref{table_similar_all}.

In the equal weight setting (Model 2.1), $\omega_j= 1/q$ for any $j = 1, \dots, q$.
In both distance-based weight and similarity-based weight settings, power distance weights are used, wherein the weight is a function of the distance/similarity to the neighboring region $ \omega_j = \frac{1}{d_j}, \quad \omega_j = \frac{1}{s_j}$ where $d_j$ is the distance score and $s_j$ is the similarity score for the $j$-th region. Under the constraint that $\sum_{j=1}^q{\omega_j}=1$, we obtain the normalized weights as $\omega_j = \frac{d_j^{-1}}{\sum_{k=1}^q d_k^{-1}}, \quad  \omega_j = \frac{s_j^{-1}}{\sum_{k=1}^q s_k^{-1}}$.

\begin{table}
\caption{\label{table_similar_all} 
Neighboring states and cities/counties by similarity score (for Model 2.3). (states: selected from all states in the country; counties: in the same state) }
\centering
{
\begin{tabular}{l l}
  \hline
  \hline
Region &  Neighboring Regions     \\
  \hline
New York & Massachusetts, New Jersey, D.C., Pennsylvania, Michigan \\
 Oregon & Maine, Washington, Vermont, West Virginia, New Hampshire  \\
   Florida & South Carolina, Nevada, Texas, Alabama, Mississippi  \\ 
   California & 
    North Carolina, Texas, Nevada, West Virginia, Mississippi \\
   Texas & 
   California, North Carolina, South Carolina, Nevada, Alabama \\
Riverside& Orange, Los Angeles, Ventura, San Diego, Monterey \\ 
  Santa Barbara& San Diego, Ventura, Orange, San Francisco, Contra Costa \\ 
 \hline
\end{tabular}}
\end{table}


Before applying Model 3, we compare the in-sample and out-of-sample MRPEs (defined in Section~\ref{sec:sims}) in all four variants of Model 2 and select the parameter values estimated by the best-performing model. In subsequent analysis, the results from Model 2.3 are reported, since it proved to be the best performing one.

As expected, change points detected for state data are related to ``stay-at-home'' orders, or phased reopening dates issued by state governments. We define the reopening date as the time when either the ``stay-at-home'' order expired or state governments explicitly lifted orders and allowed (selected or even all) businesses to reopen \cite{NYtimesreopen}. The ``stay-at-home'' and reopening dates for all states are shown in Table \ref{table_plan_all}.

\begin{table*}[ht!]
\caption{\label{table_plan_all}
Statewide ``Stay-at-home'' plan and reopening plan begin dates, along with the detected change points (CPs) in states and counties/cities.}
\centering
\resizebox{0.98\textwidth}{!}{
\begin{tabular}{lccccc} 
  \hline
  \hline
Region  & ``Stay-at-home'' plan  &Reopening plan (statewide) & Detected change points \\
\hline
New York & March 22 & July 6 (Phase 3) &  April 04 2020  \\ 
Oregon  & March 23 & May 15 (Phase 1) &   April 06 2020, June 06 2020, July 18 2020, Jan 16 2021 \\
Florida &   April 3 & June 5 (Phase 2) &   April 13 2020, June 17 2020, July 25 2020\\
California  & March 19 & June 12 (Phase 2)  &April 11 2020, April 29 2020, Aug 12 2020, Dec 06 2020, Jan 19 2021\\
Texas  & April 2  & June 3 (Phase 3) &  April 18 2020,  June 03 2020, June 15 2020, July 23 2020, Sep 20 2020, Jan 03 2021, Feb 10 2021\\
Michigan & March 24 & June 8 (phase 5)&  April 11 2020, June 04 2020, Nov 01 2020, Nov 21 2020, Dec 11 2020, March 21 2021\\
Riverside (CA) & March 19 & June 12 (Phase 2)  & April 17 2020,  June 20 2020, July 26 2020, Nov 28 2020\\
Santa Barbara (CA) & March 19 & June 12 (Phase 2)  & April 15 2020, May 11 2020, June 10 2020, July 25 2020, Dec 27 2020\\
 \hline
 \multicolumn{4}{l}{*NY reopening dates vary in counties. Here, we use the NYC reopening dates instead.}\\
\end{tabular}}
\end{table*}

In Model 1, a change point is detected from March to April for all five states: New York, Oregon, Florida, California and Texas.  
These change points coincide with the onset of ``stay-at-home'' orders and correspond to a significant decrease in the transmission rate.
The first change points detected are around two weeks after the state's Governors signed a statewide ``stay-at-home'' order (three weeks for CA), which is consistent with the fact that COVID-19 symptoms develop 2 days to 2 weeks following exposure to the virus. As can be seen from Figure \ref{fig:rates_states_smooth} in the supplementary material, in addition to the downward trend after lockdowns have been put in place,  Oregon, Florida and Texas have a clear upward trend after state reopenings began in May. 
The model detects a change point in June for these states, which relate to their reopenings.

Note that the restriction in either phase 2 reopening plan in Florida and the phase 3 reopening plan in Texas are quite similar in terms of restaurants, bars, and entertainment businesses.  
Starting June 5, restaurants and bars in Florida could increase their indoor seating to 50\% capacity. Movie theaters, concert venues, arcades, and other entertainment businesses could also open at 50\% capacity.
Starting June 3, all businesses in Texas could expand their occupancy to 50\% with certain exceptions. Moreover, bars could increase their capacity to 50\% as long as patrons are seated.

Assuming that Florida had not begun the phase 2 reopening plan on June 5, our model predicts 
36,626
infected cases by June 12 while the actual number of infected cases is 
39,327
($7.3\%$ higher). 
Similarly, by June 19, our model predicts 
41,728
infected cases, while the actual number of infected cases is 55,607
($33.3\%$ higher).
Similarly, suppose that Texas had not begun the phase 3 reopening plan on June 3, 
then by June 17, our model predicts 
74,204
infected cases, while the actual number of infected cases is 76,377
($1\%$ higher).

In July, many states paused plans to reopen, amid rising infected case counts\footnote{ \url{https://www.usatoday.com/story/news/nation/2020/06/30/covid-cases-states-pausing-reopening-plans-list/3284513001/}}. These pausing actions effectively slowed the spread of COVID-19, as can be clearly seen in the downward trend of the transmission rate in July and August in Oregon, Florida, California and Texas. Interestingly, a change point in July is detected in Oregon, Florida and Texas, and a change point in August is detected in California, mainly related to this pausing.

The left panel of Figure~\ref{fig:under-reporting_states} depicts the estimated function $u(t)$ while the right panel displays the performed daily test normalized by the populations for the six states under consideration. Comparing these two plots indicate that states at which the testing capacity has been limited, the under-reporting rate was estimated higher and vise versa. For example, New York had more daily tests based on its population compared with the other states, which coincides with its lower estimated under-reporting rate at the beginning while the estimated under-reporting rate is higher for Texas and Florida which could be due to limited capacity in daily testing at the beginning of pandemic for these states. Note that the under-reporting rate is estimated around $90\%$ in March for all states. Such high under-reporting rates at the beginning of the pandemic are reported in other countries as well, such as China \cite{li2020substantial}.


\begin{figure*}[!ht]
     \centering
\resizebox{0.60\textwidth}{!}{
          \begin{subfigure}[b]{0.3\textwidth}
         \centering
         \includegraphics[width=\textwidth]{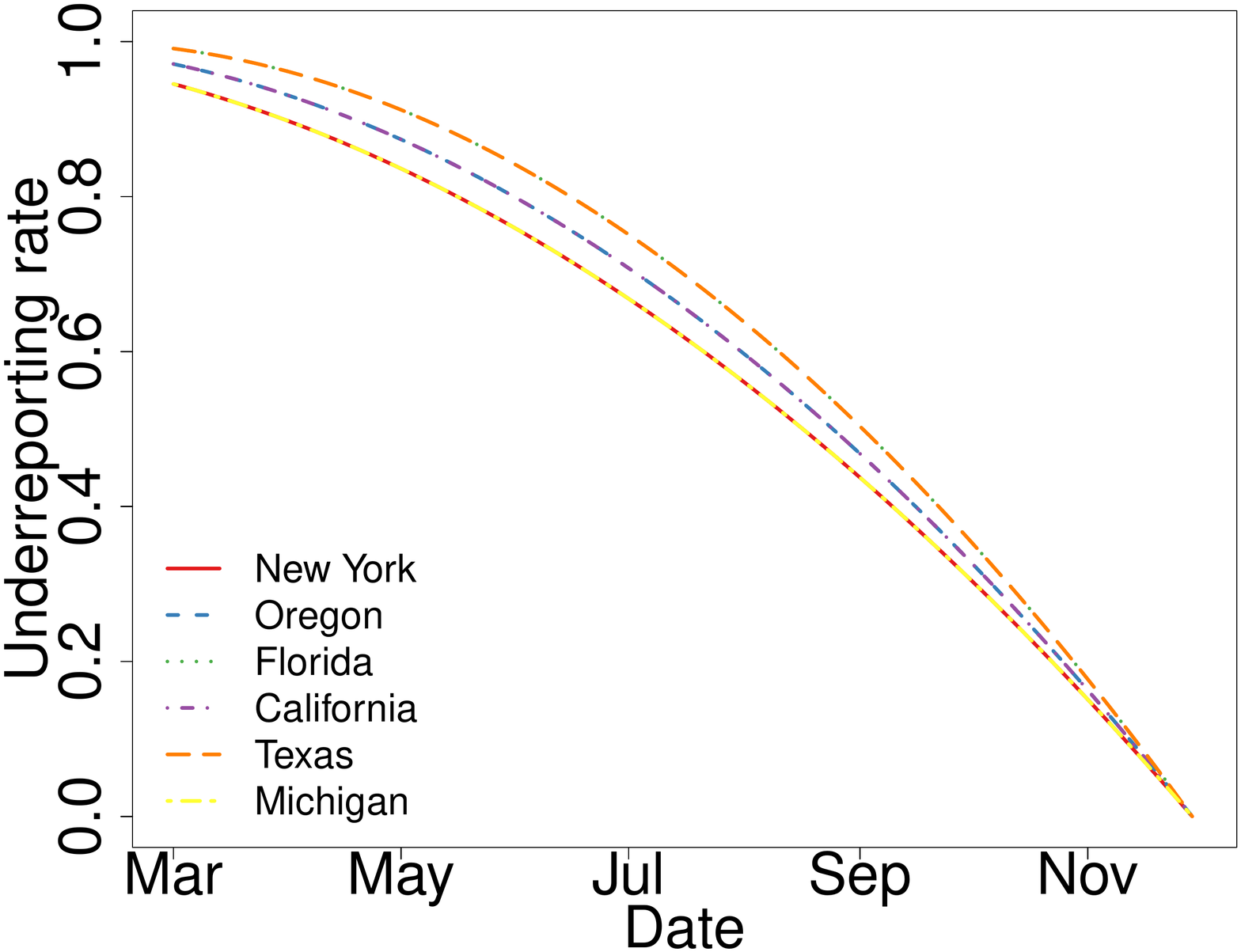}
         \subcaption{under-reporting function}
     \end{subfigure}
     \begin{subfigure}[b]{0.3\textwidth}
         \centering
         \includegraphics[width=\textwidth]{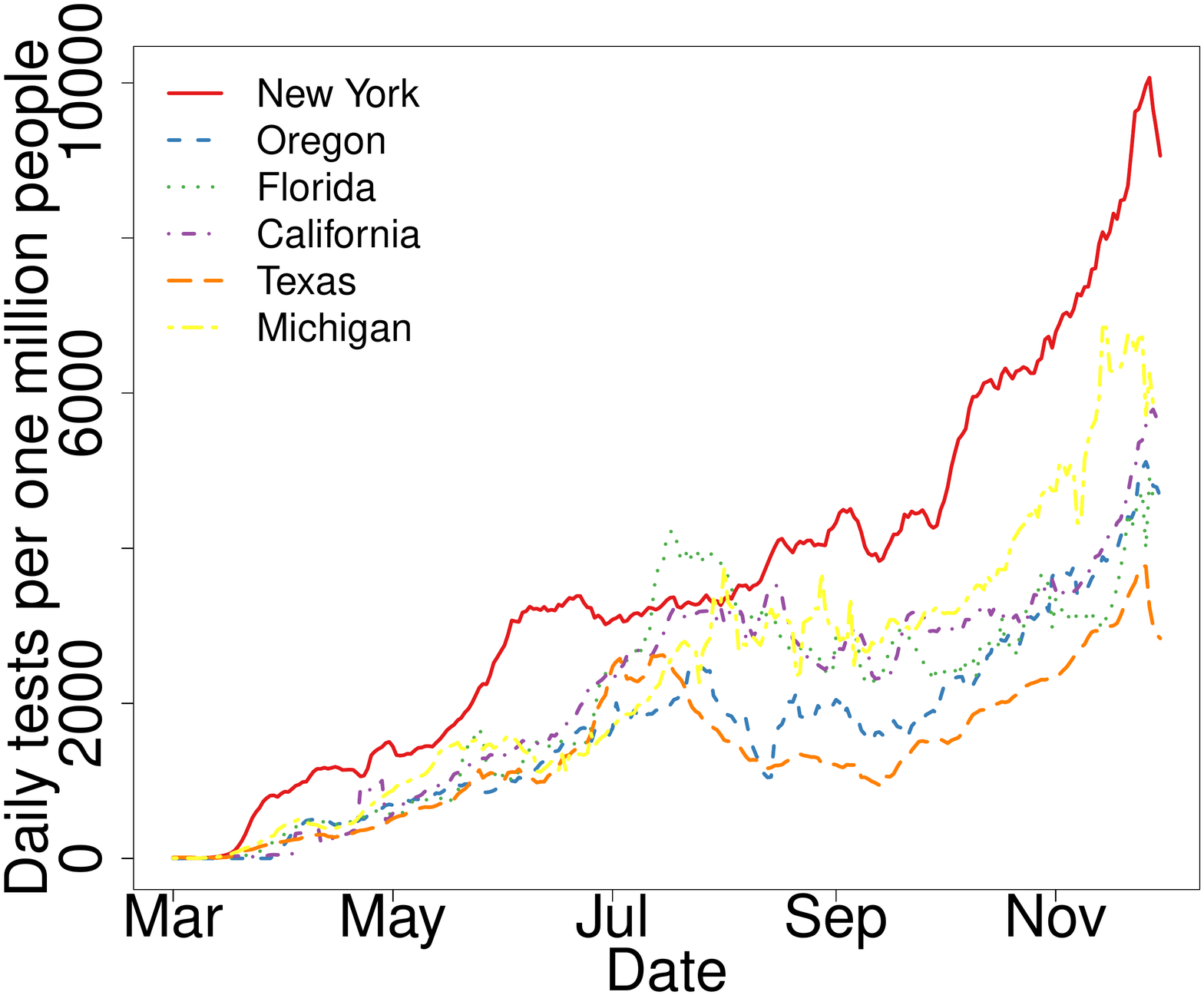}
         \subcaption{daily testing}
     \end{subfigure}}
        \caption{
        (left panel) Estimated under-reporting rate functions for all six states; (right panel) 7-day moving average of daily testing for all six states.
        }
        \label{fig:under-reporting_states}
\end{figure*}

The observed and fitted number of infected cases are displayed in Figure \ref{fig:number of infected}. Note that the fitted number of infected cases and recovered cases are defined as 
\begin{align}
    \widetilde{I}(t) = I(1) + \sum_{k=1}^{t-1}\widehat{\Delta I}(k), \quad \widetilde{R}(t) = R(1) + \sum_{k=1}^{t-1}\widehat{\Delta R}(k),
\end{align}
for all $t = 2,\dots, T$. In summary, the piecewise constant SIR model (Model 1) with detected change points significantly improves the performance in prediction of the number of infected cases 
compared with the piecewise constant SIR model with pre-specified change points. Pre-specified change points are defined as stay-at-home or reopening order dates for each region. It can be seen from the first two rows in Figure~\ref{fig:number of infected} that estimating when break points occurred  using the developed algorithm improves the fit significantly, especially for states in which multiple change points are selected such as Oregon, California and Texas. Model 2.3 further improves the fit of the data for Florida and California, due to the addition of a spatial smoothing effect. Finally, Model 3 provides further improvements, in particular for New York, which justifies empirically the use of hybrid modeling to analyze regional transmission dynamics of COVID-19.

\begin{figure*}[ht!]
     \centering
     \resizebox{0.91\textwidth}{!}{
     \begin{subfigure}[b]{0.19\textwidth}
         \centering
         \includegraphics[width=\textwidth]{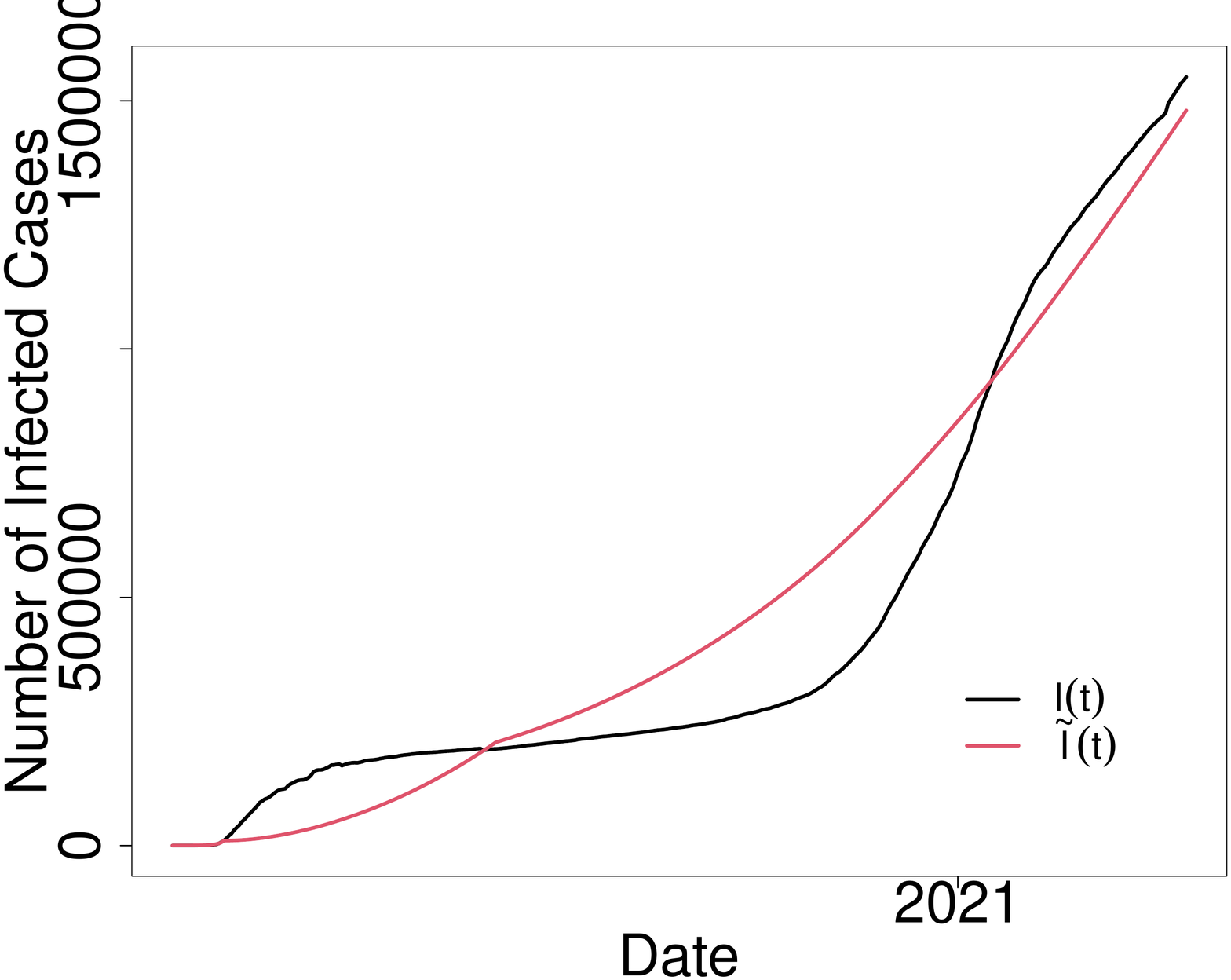}
         \subcaption{NY (Model 1 with pre-specified change points)}
     \end{subfigure}
    \begin{subfigure}[b]{0.19\textwidth}
         \centering
         \includegraphics[width=\textwidth]{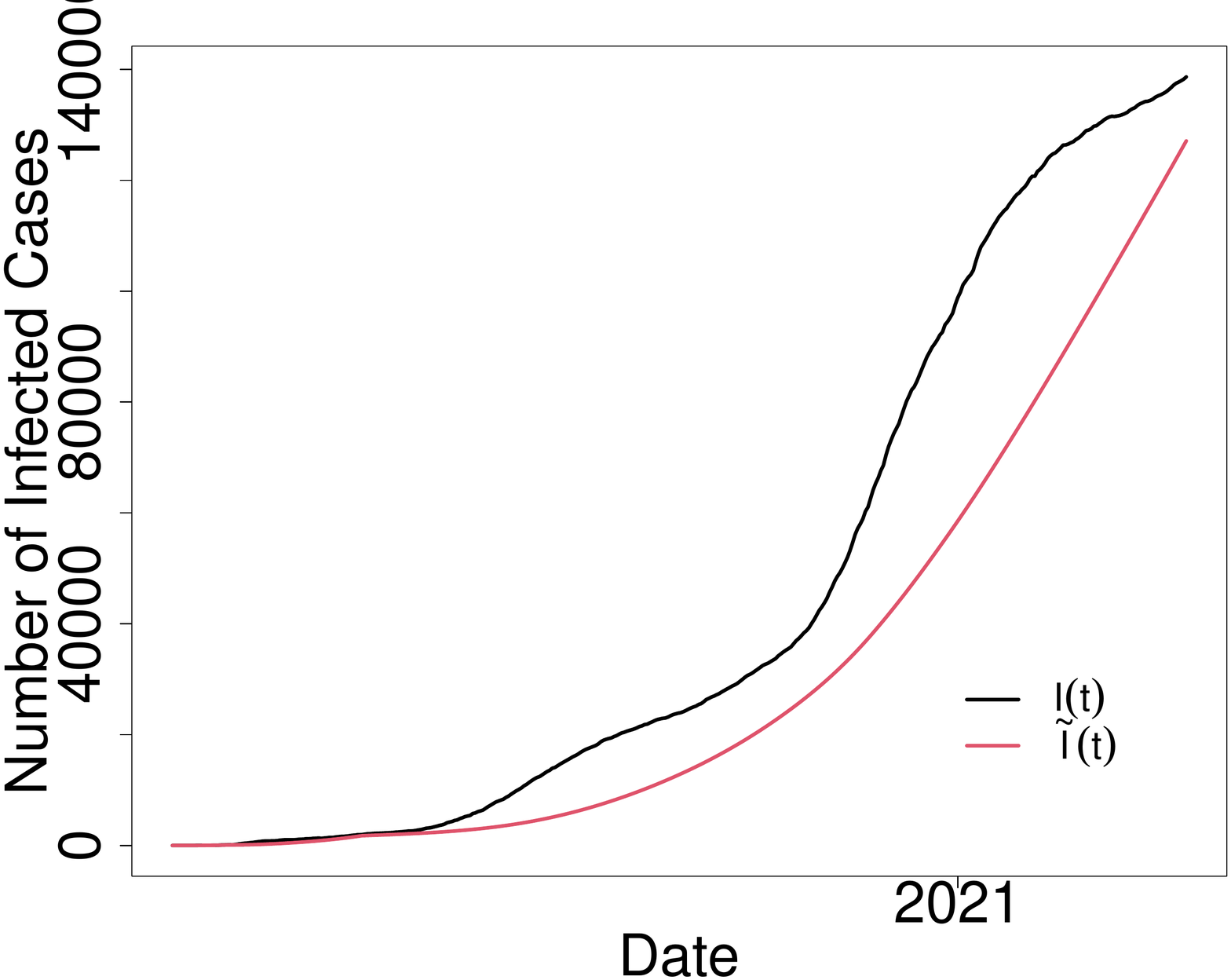}
         \subcaption{OR (Model 1 with pre-specified change points)}
     \end{subfigure}
     \begin{subfigure}[b]{0.19\textwidth}
         \centering
         \includegraphics[width=\textwidth]{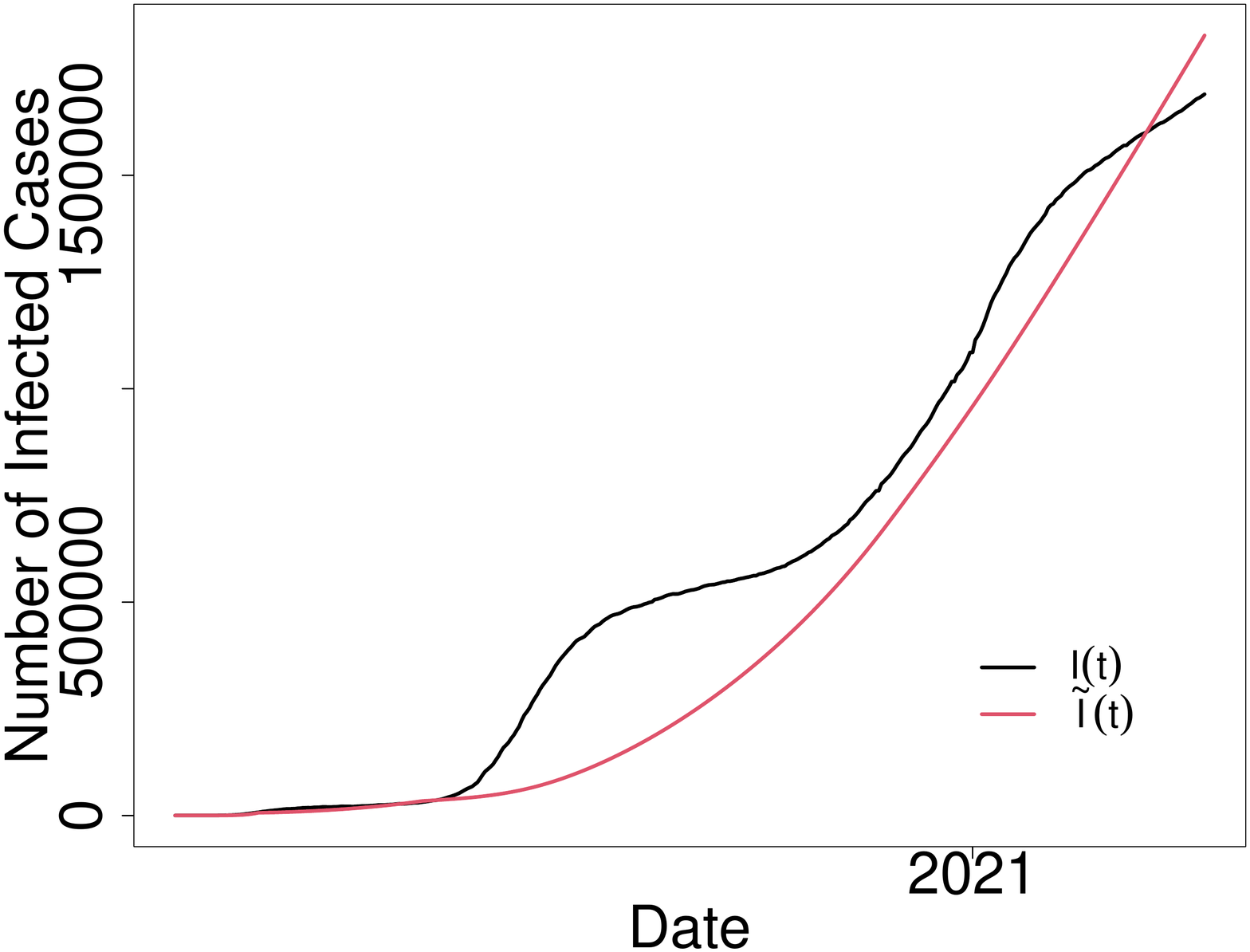}
          \subcaption{FL (Model 1 with pre-specified change points)}
     \end{subfigure}
     \begin{subfigure}[b]{0.19\textwidth}
         \centering
         \includegraphics[width=\textwidth]{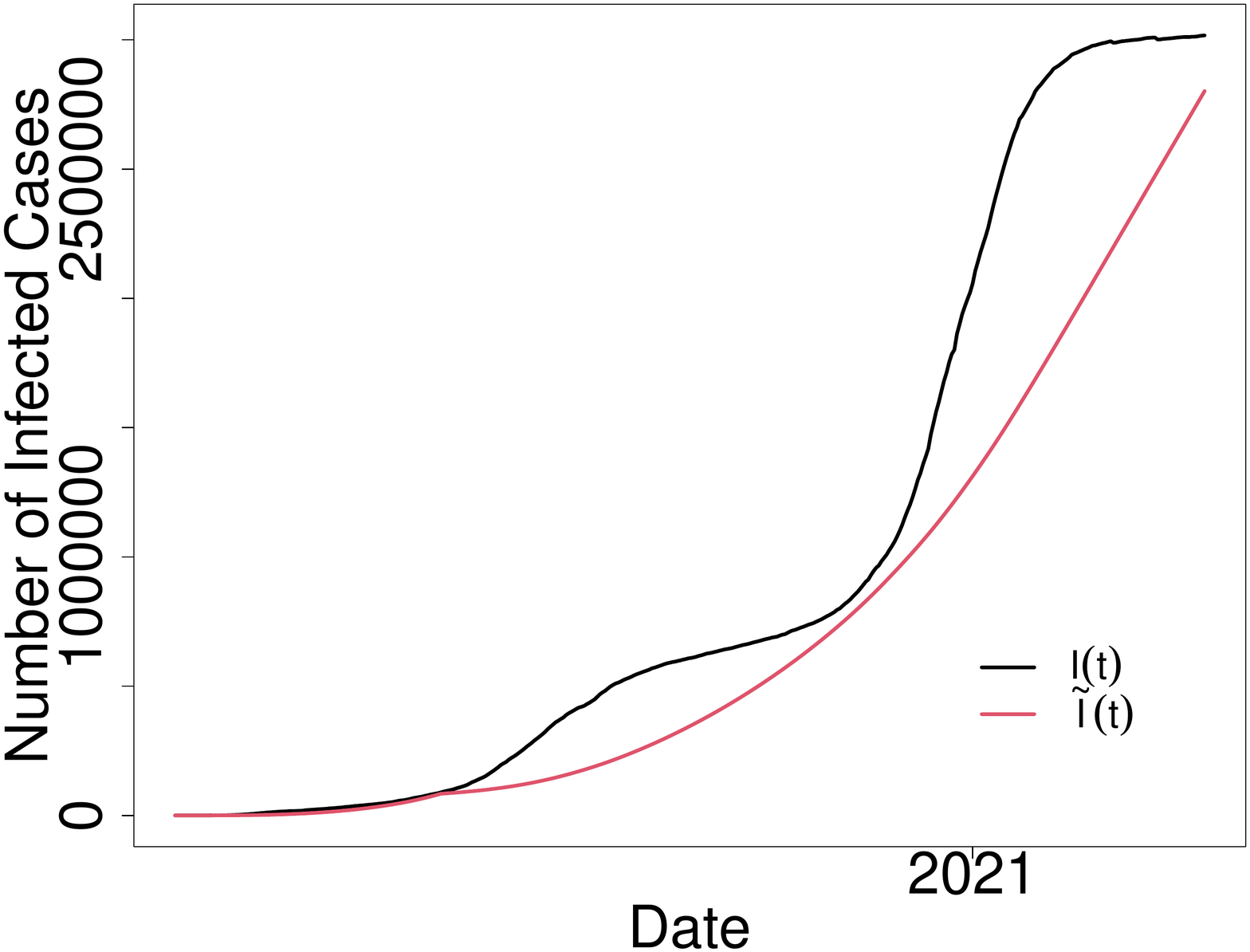}
         \subcaption{CA (Model 1 with pre-specified change points)}
     \end{subfigure}
     \begin{subfigure}[b]{0.19\textwidth}
         \centering
         \includegraphics[width=\textwidth]{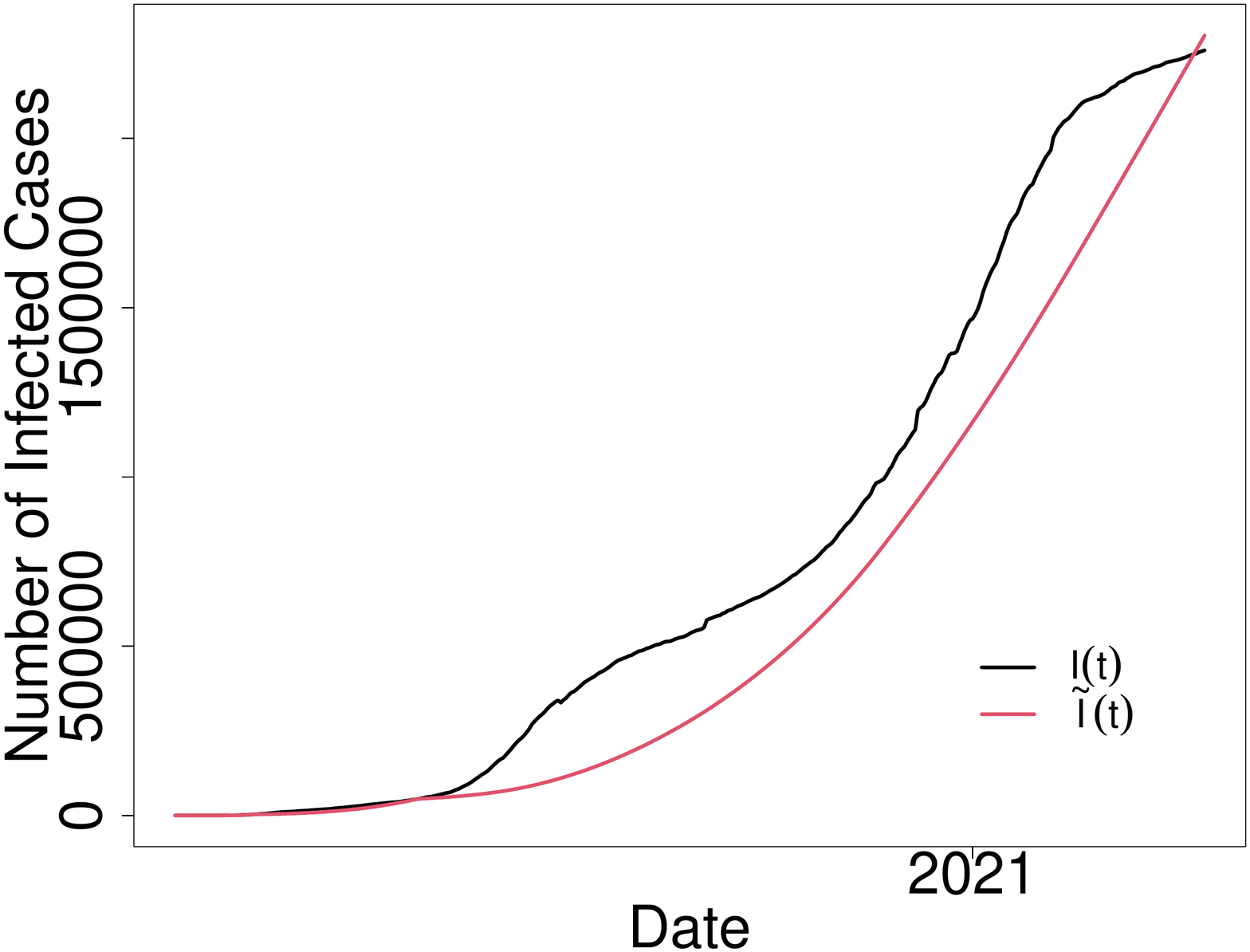}
         \subcaption{TX (Model 1 with pre-specified change points)}
     \end{subfigure}}
     
     \resizebox{0.91\textwidth}{!}{
     \begin{subfigure}[b]{0.19\textwidth}
         \centering
         \includegraphics[width=\textwidth]{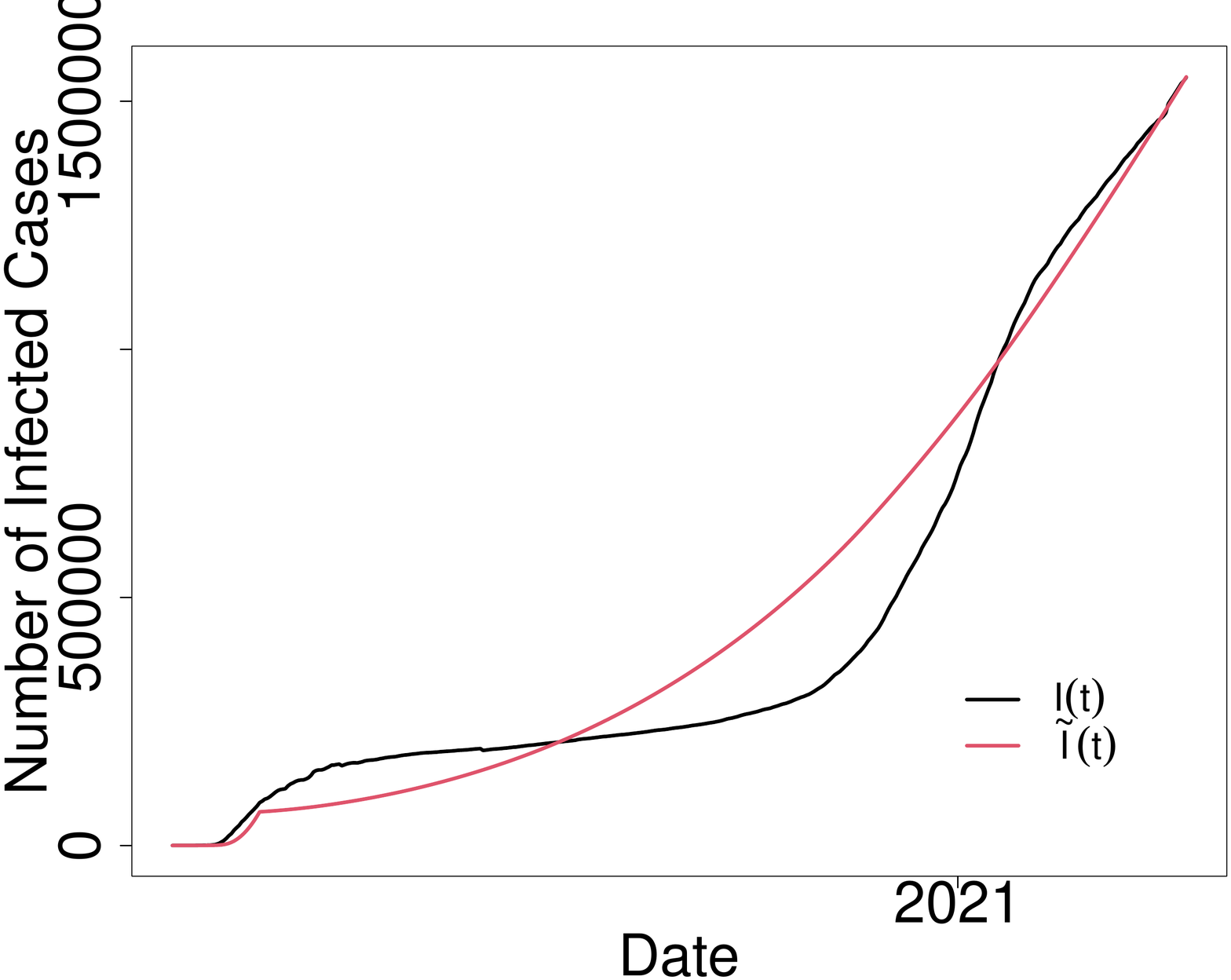}
         \subcaption{NY (Model 1 with detected change points)}
     \end{subfigure}
    \begin{subfigure}[b]{0.19\textwidth}
         \centering
         \includegraphics[width=\textwidth]{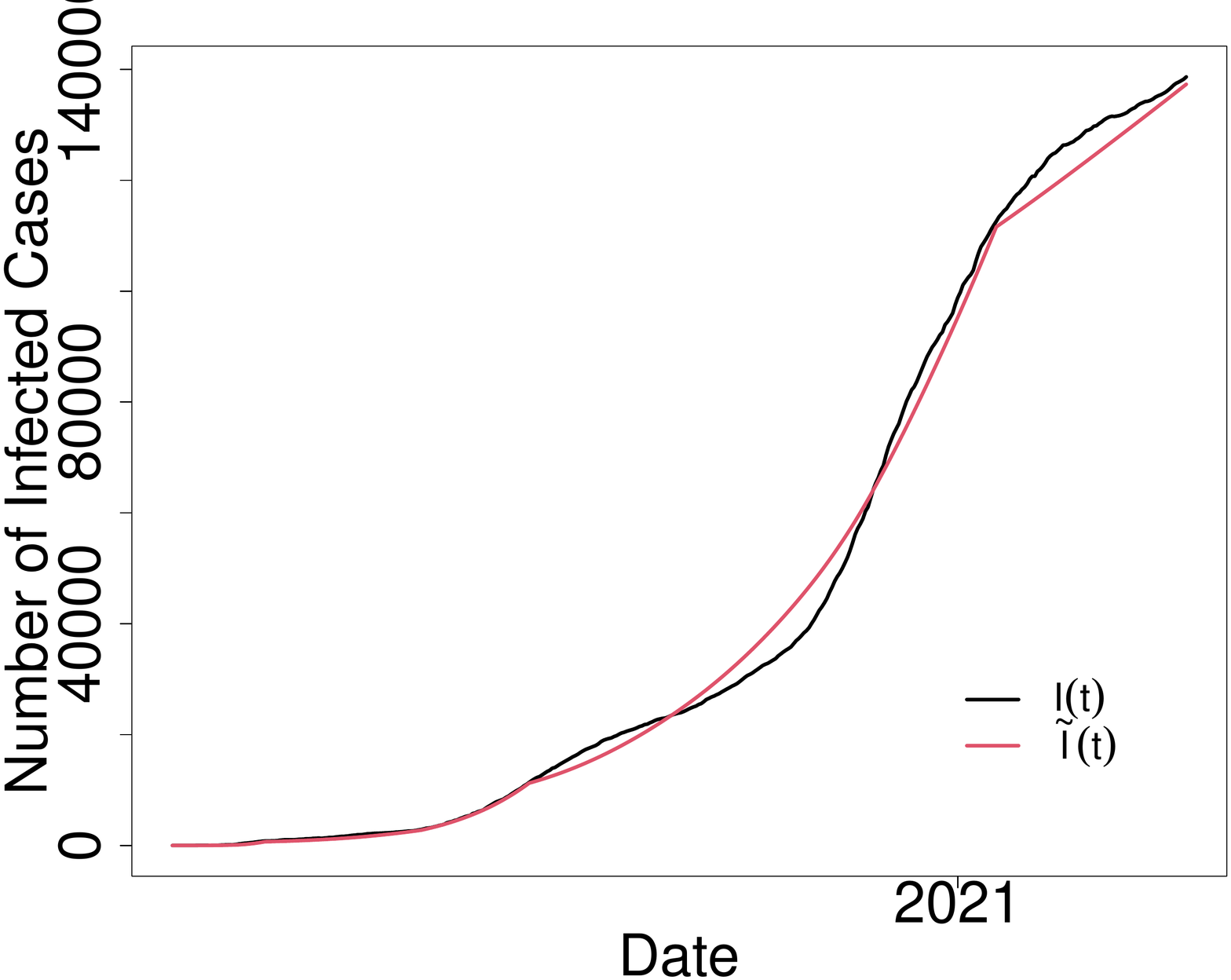}
         \subcaption{OR (Model 1 with detected change points)}
     \end{subfigure}
     \begin{subfigure}[b]{0.19\textwidth}
         \centering
         \includegraphics[width=\textwidth]{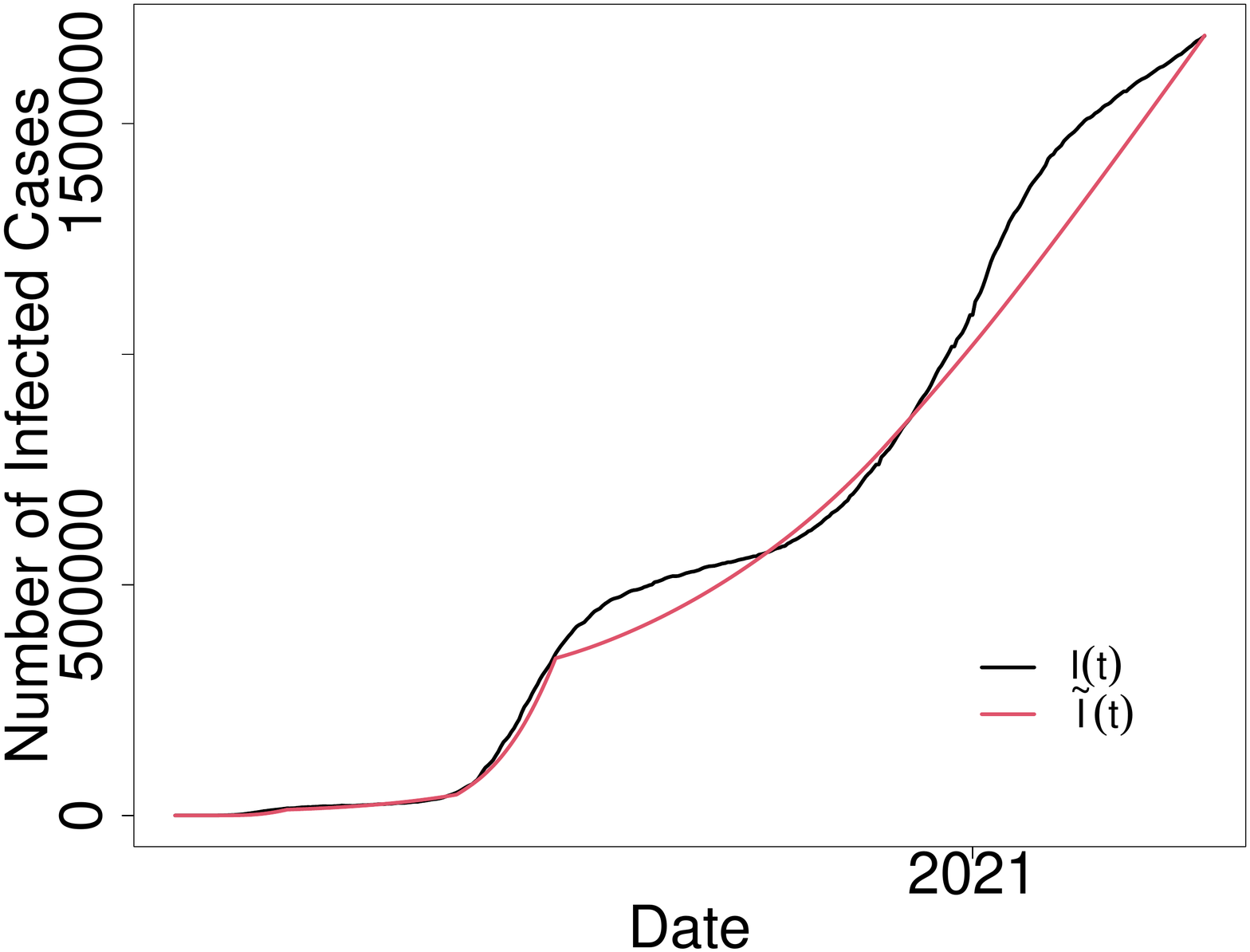}
          \subcaption{FL (Model 1 with detected change points)}
     \end{subfigure}
     \begin{subfigure}[b]{0.19\textwidth}
         \centering
         \includegraphics[width=\textwidth]{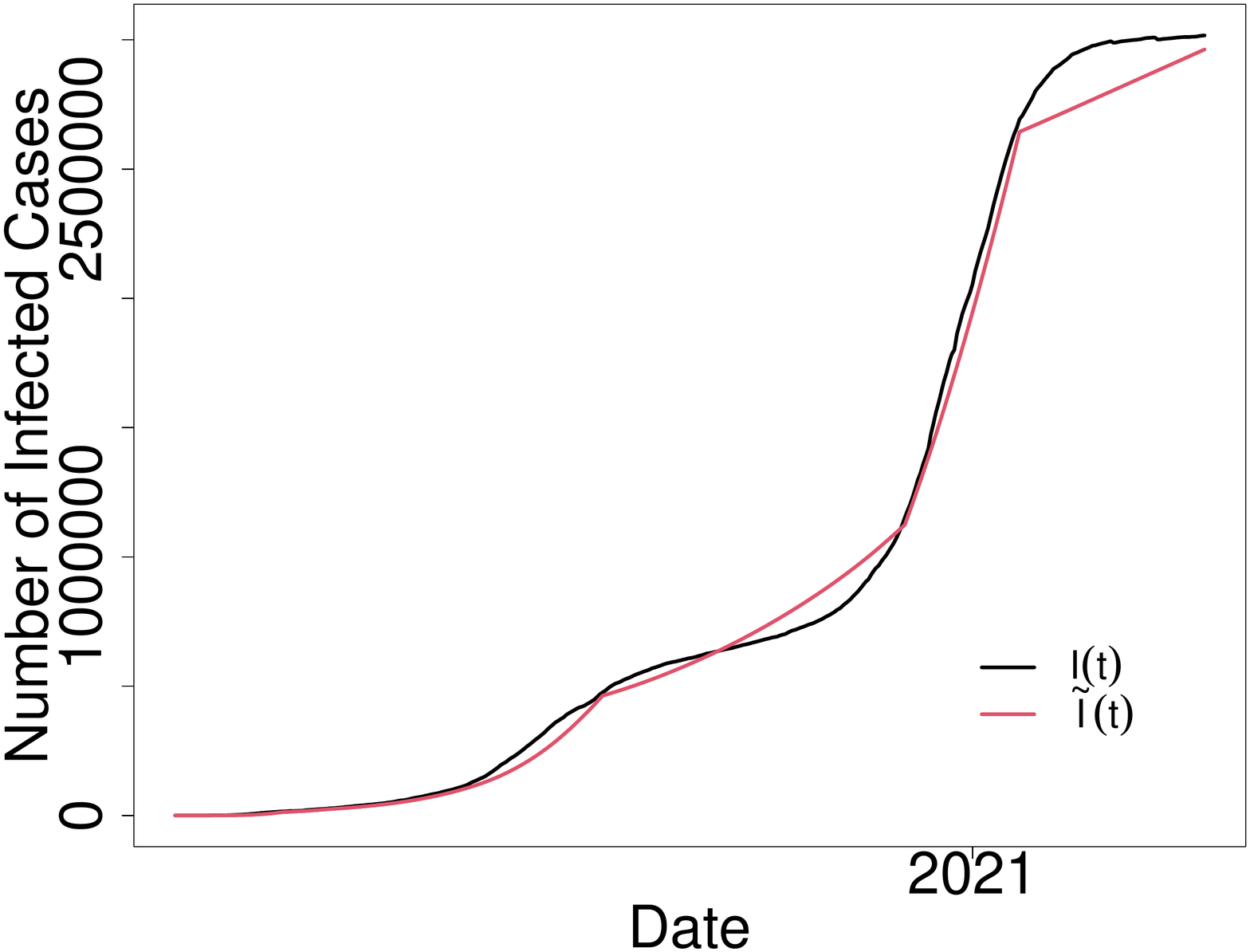}
         \subcaption{CA (Model 1 with detected change points)}
     \end{subfigure}
     \begin{subfigure}[b]{0.19\textwidth}
         \centering
         \includegraphics[width=\textwidth]{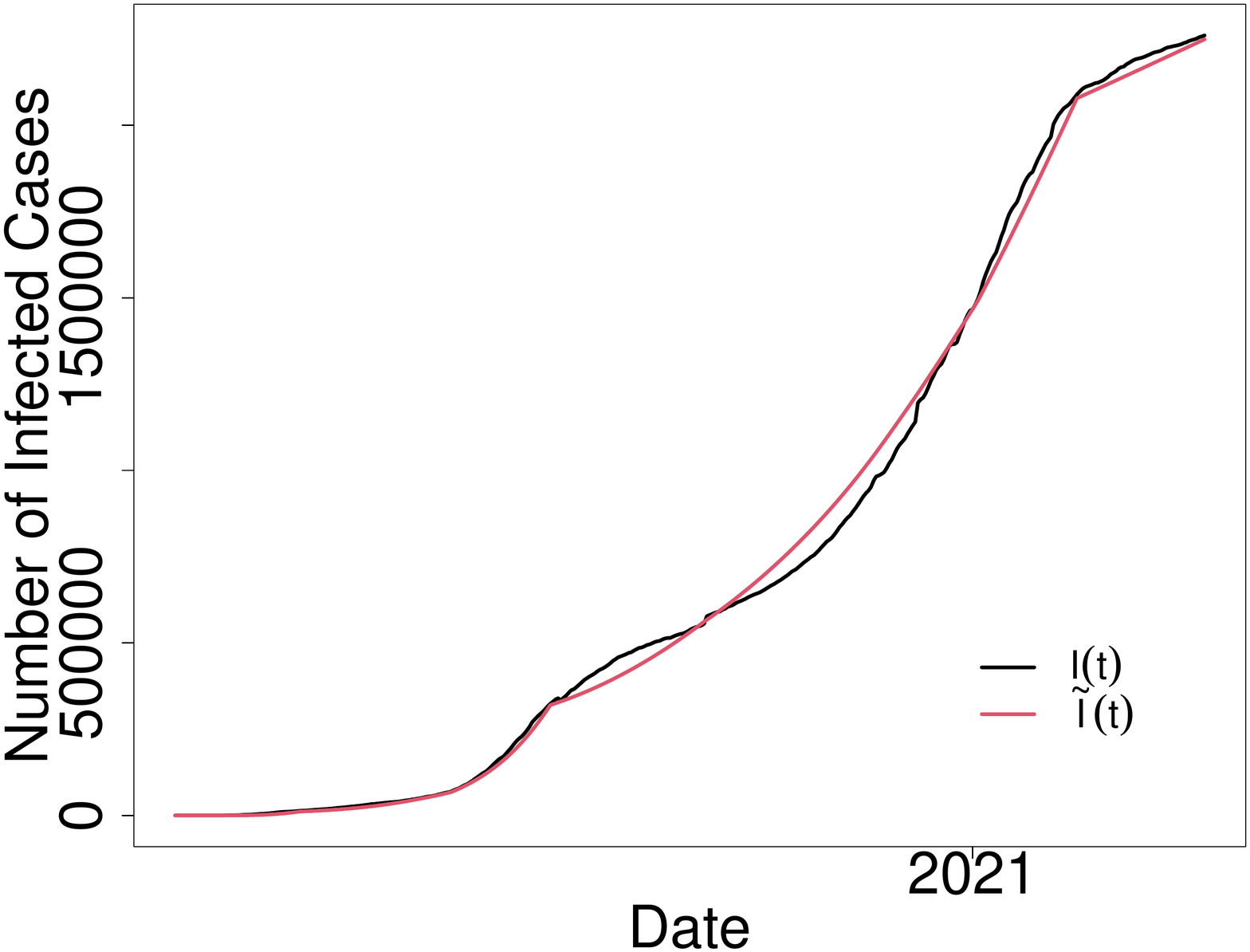}
         \subcaption{TX (Model 1 with detected change points)}
     \end{subfigure}}
    
    \resizebox{0.91\textwidth}{!}{
     \begin{subfigure}[b]{0.19\textwidth}
         \centering
         \includegraphics[width=\textwidth]{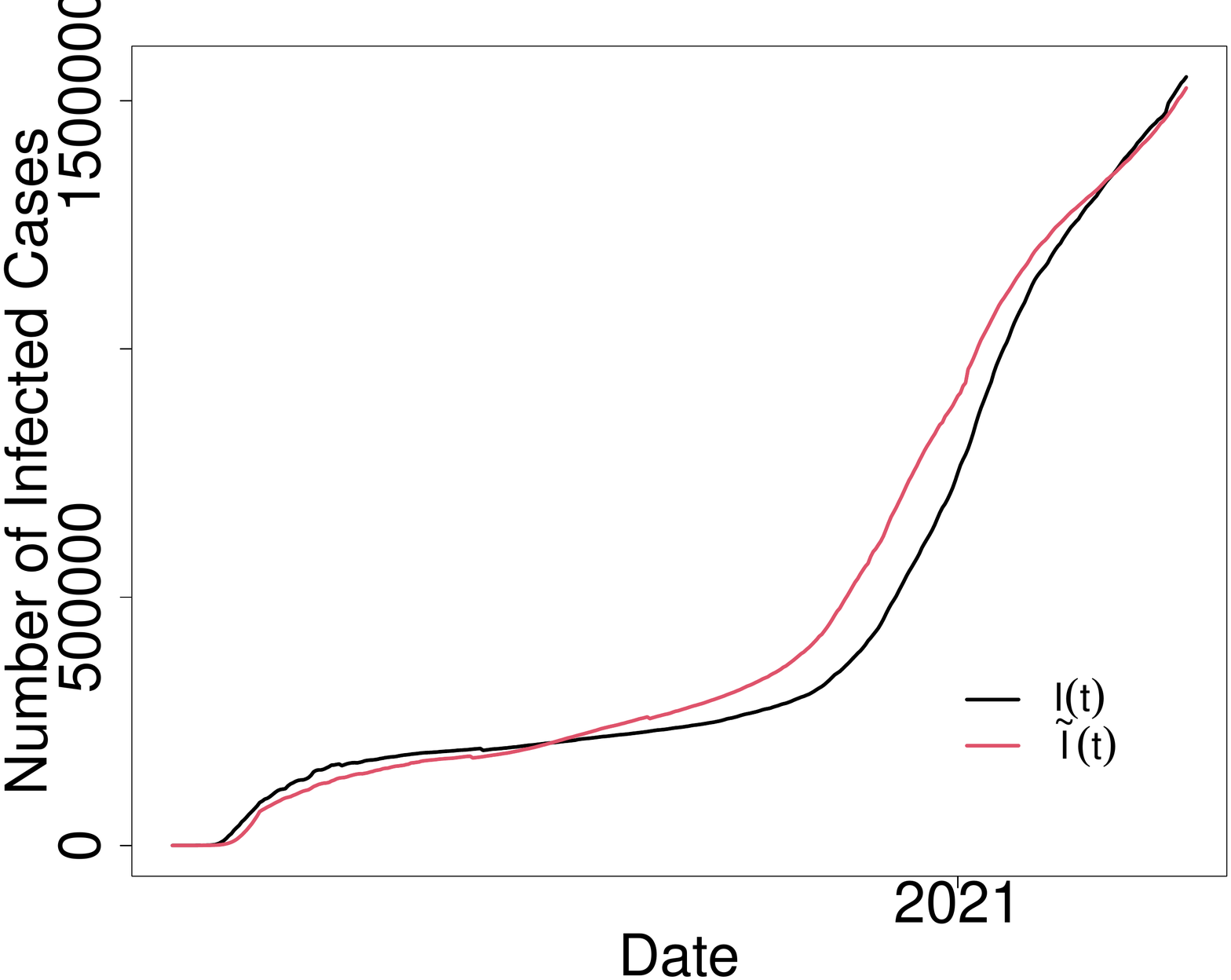}
         \subcaption{NY (Model 2.3)}
     \end{subfigure}
    \begin{subfigure}[b]{0.19\textwidth}
         \centering
         \includegraphics[width=\textwidth]{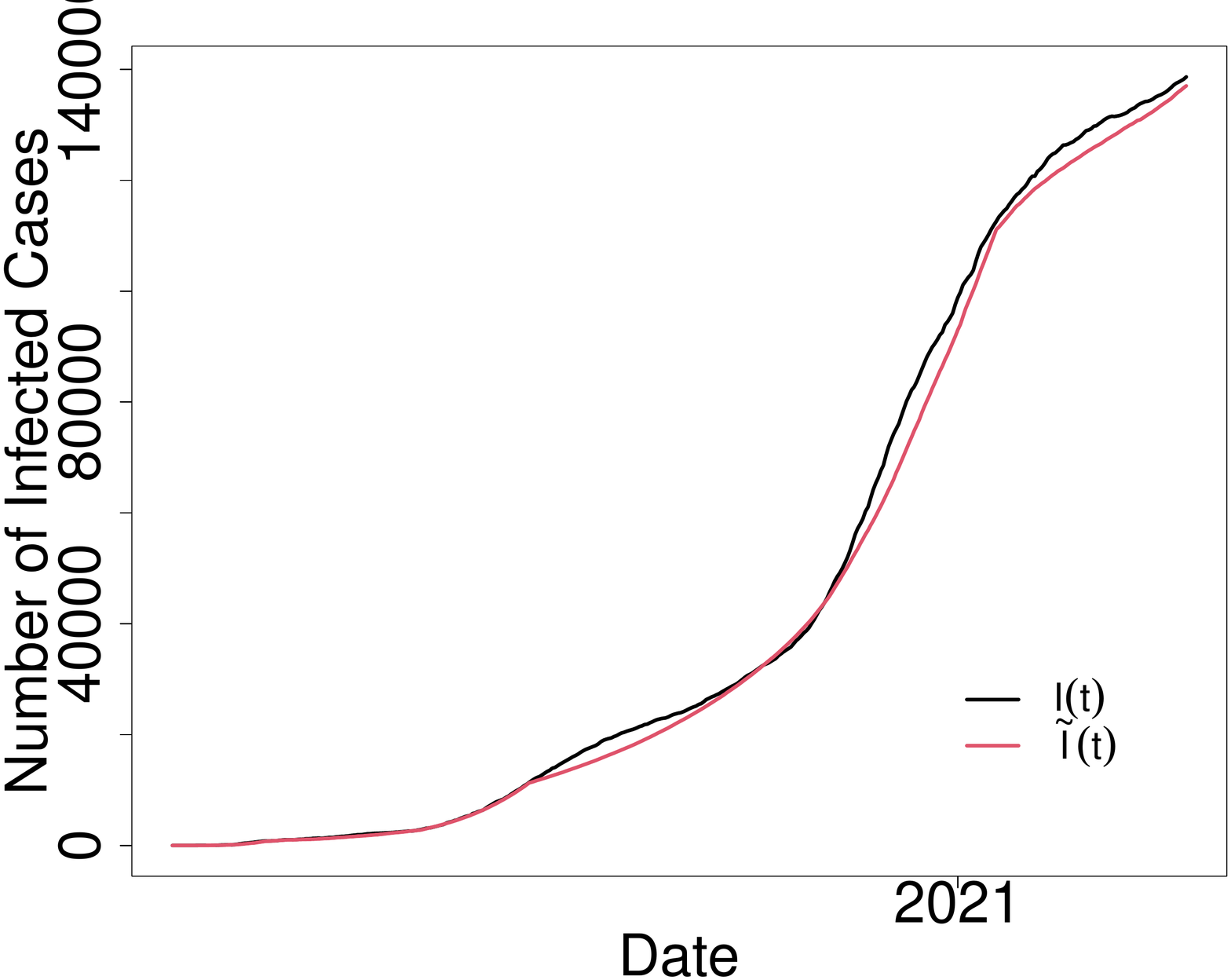}
         \subcaption{OR (Model 2.3)}
     \end{subfigure}
     \begin{subfigure}[b]{0.19\textwidth}
         \centering
         \includegraphics[width=\textwidth]{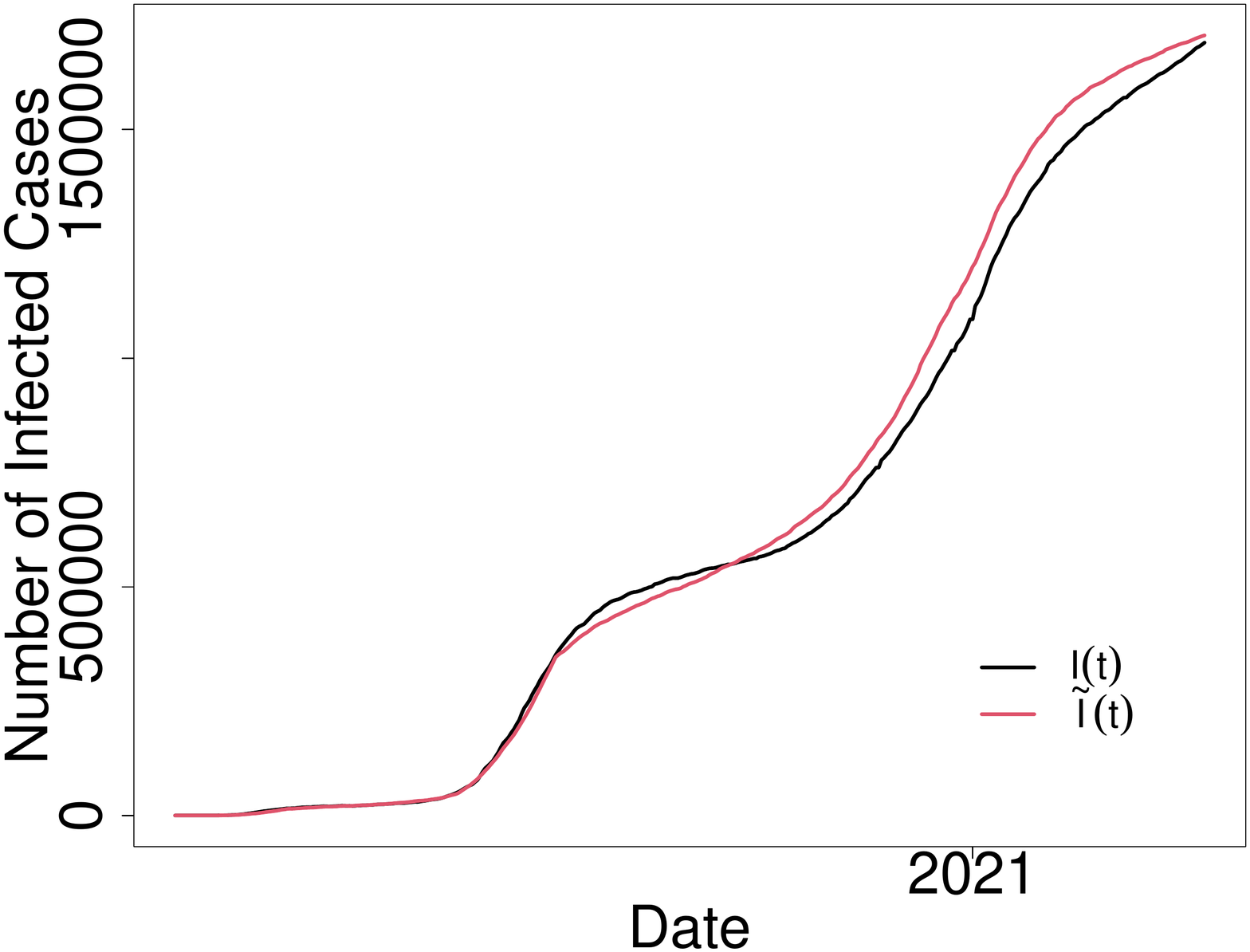}
          \subcaption{FL (Model 2.3)}
     \end{subfigure}
     \begin{subfigure}[b]{0.19\textwidth}
         \centering
         \includegraphics[width=\textwidth]{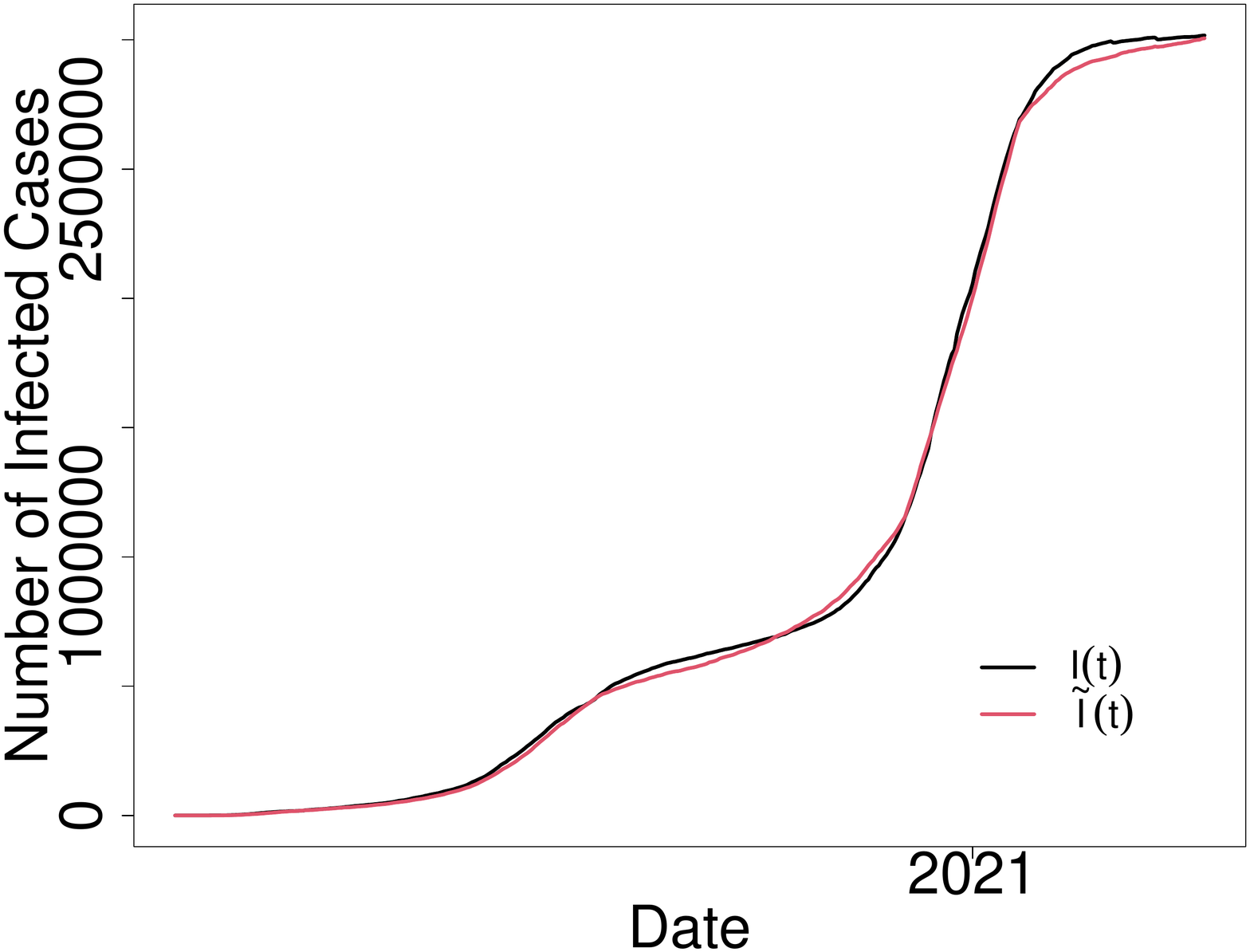}
         \subcaption{CA (Model 2.3)}
     \end{subfigure}
     \begin{subfigure}[b]{0.19\textwidth}
         \centering
         \includegraphics[width=\textwidth]{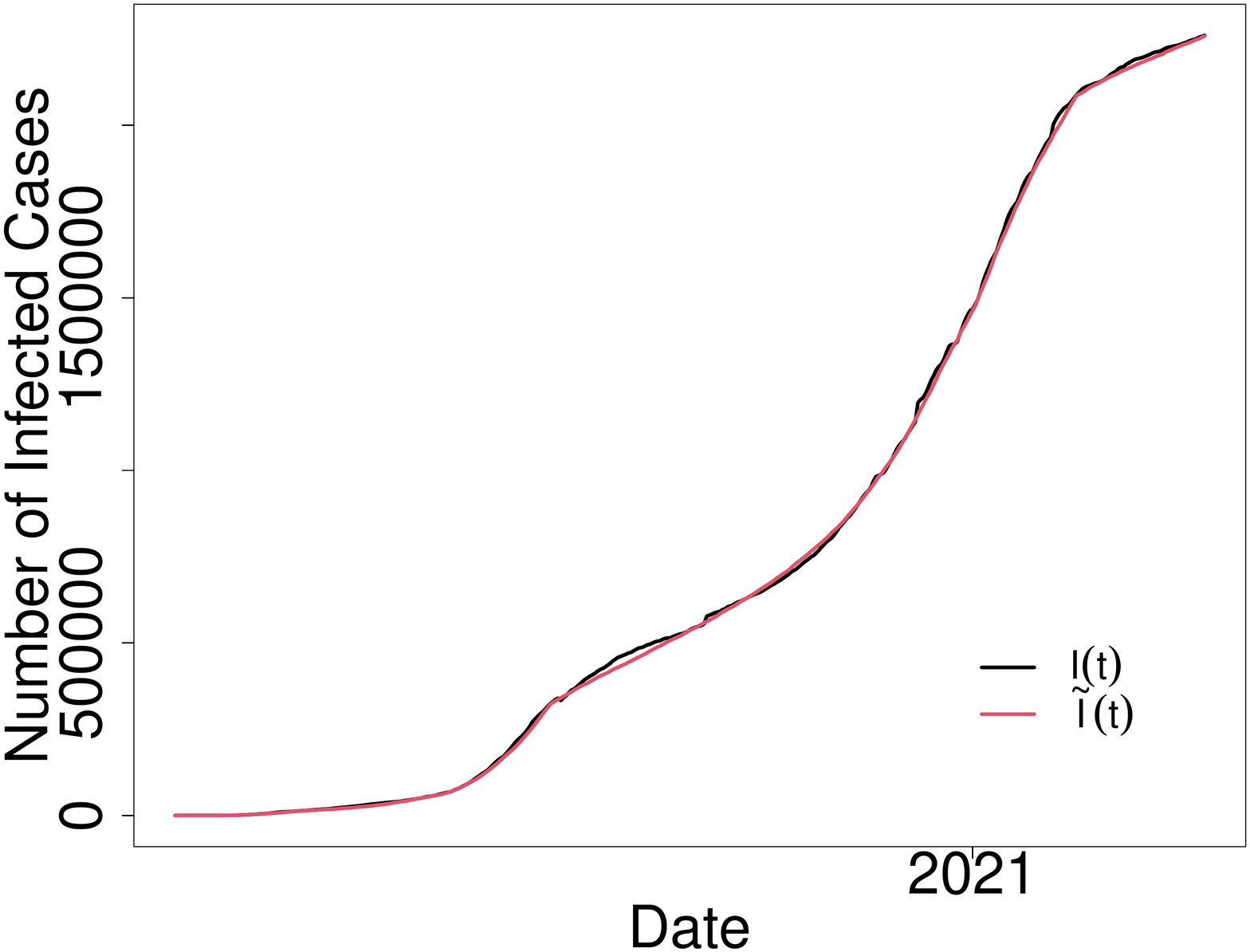}
         \subcaption{TX (Model 2.3)}
     \end{subfigure}}
     
     \resizebox{0.91\textwidth}{!}{
     \begin{subfigure}[b]{0.19\textwidth}
         \centering
         \includegraphics[width=\textwidth]{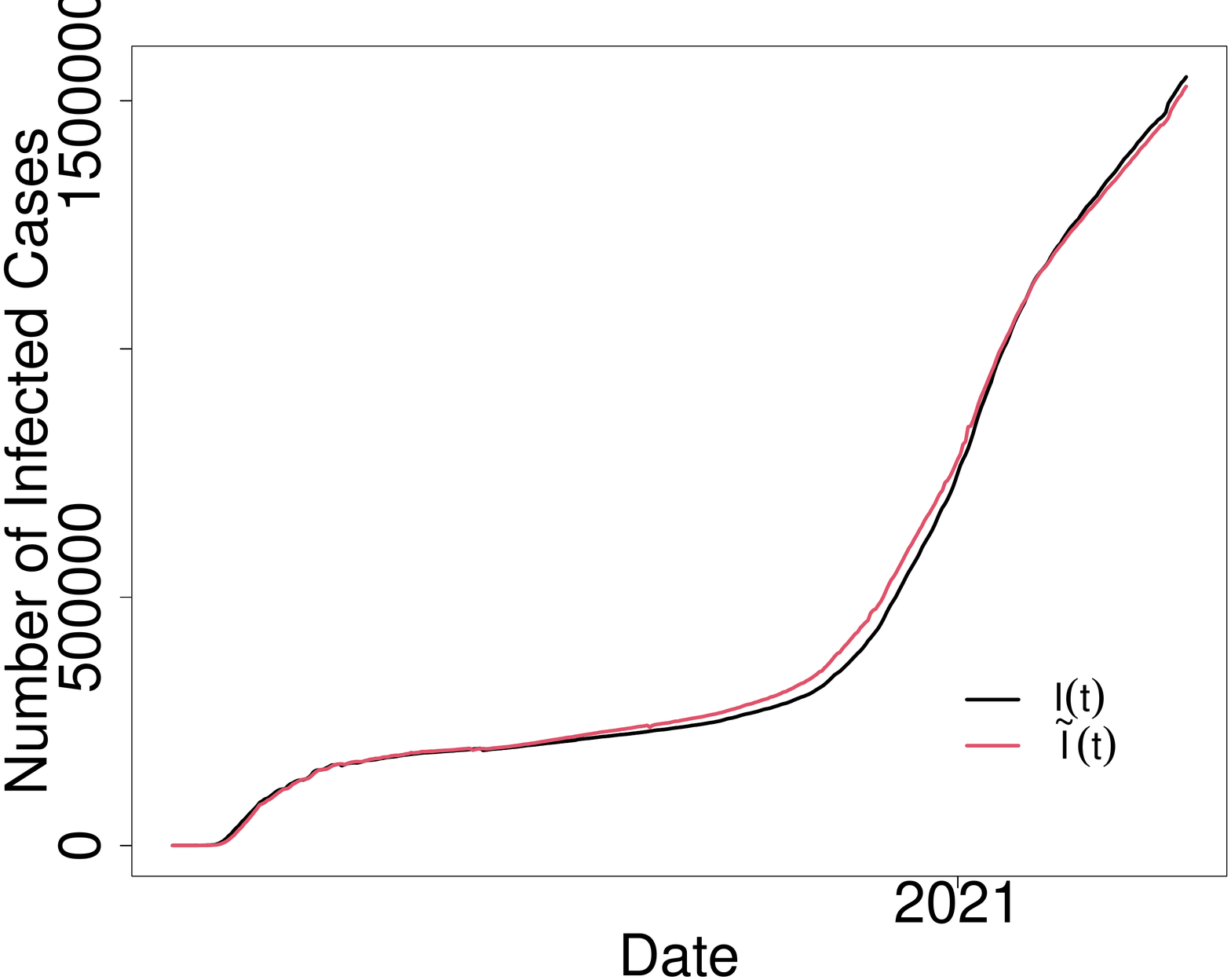}
         \subcaption{NY (Model 3)}
     \end{subfigure}
      \begin{subfigure}[b]{0.19\textwidth}
         \centering
         \includegraphics[width=\textwidth]{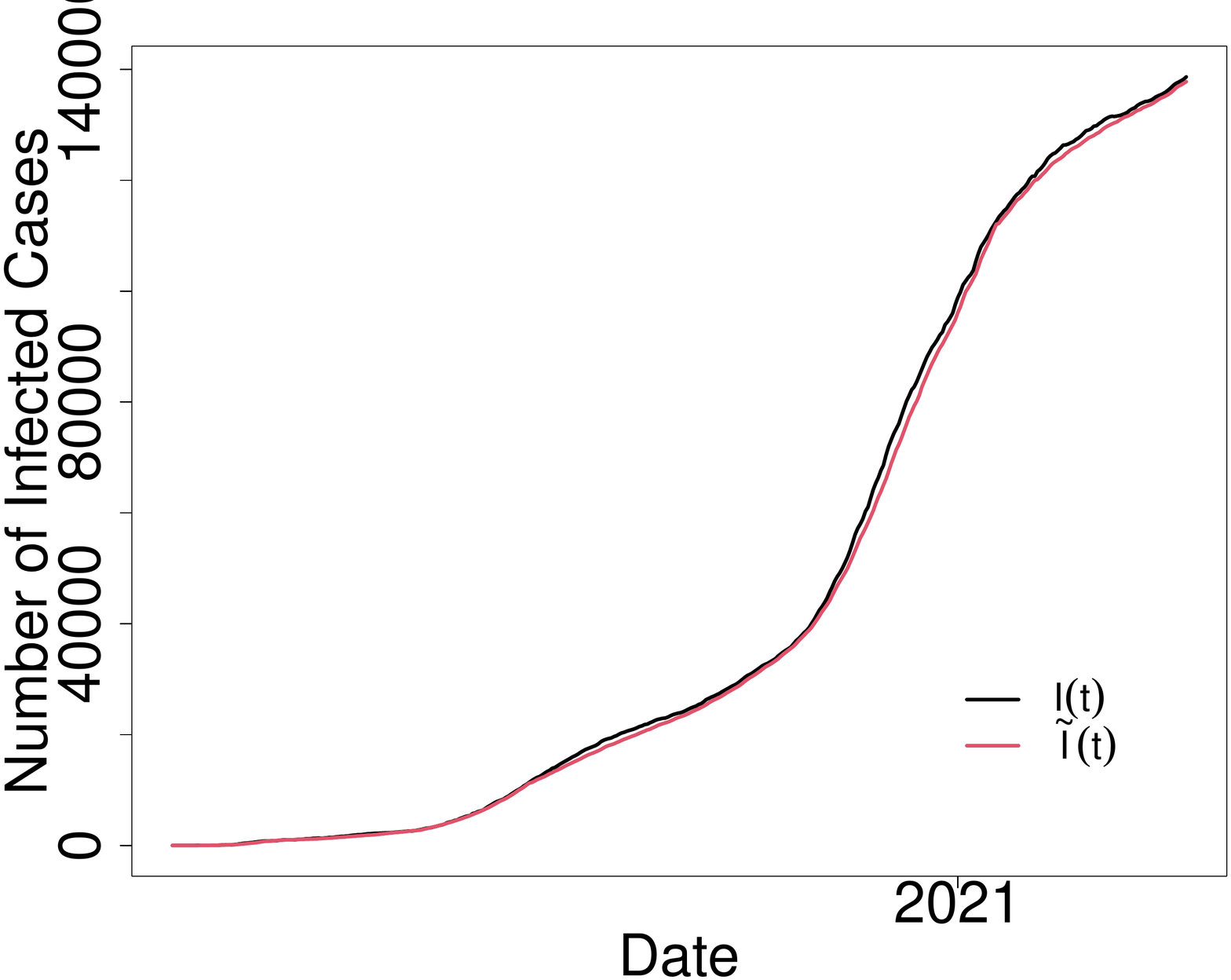}
         \subcaption{OR (Model 3)}
     \end{subfigure}
     \begin{subfigure}[b]{0.19\textwidth}
         \centering
         \includegraphics[width=\textwidth]{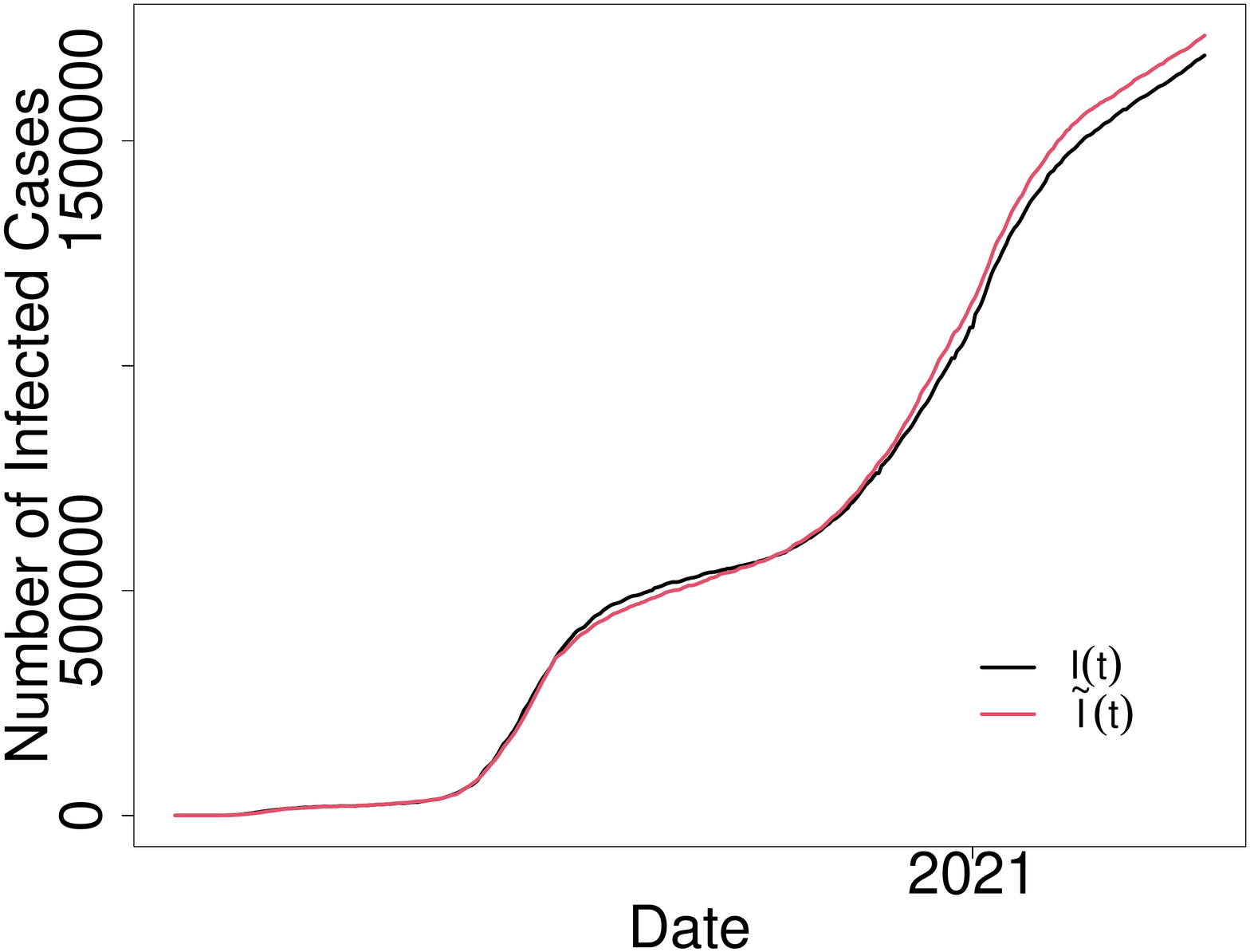}
          \subcaption{FL (Model 3)}
     \end{subfigure}
     \begin{subfigure}[b]{0.19\textwidth}
         \centering
         \includegraphics[width=\textwidth]{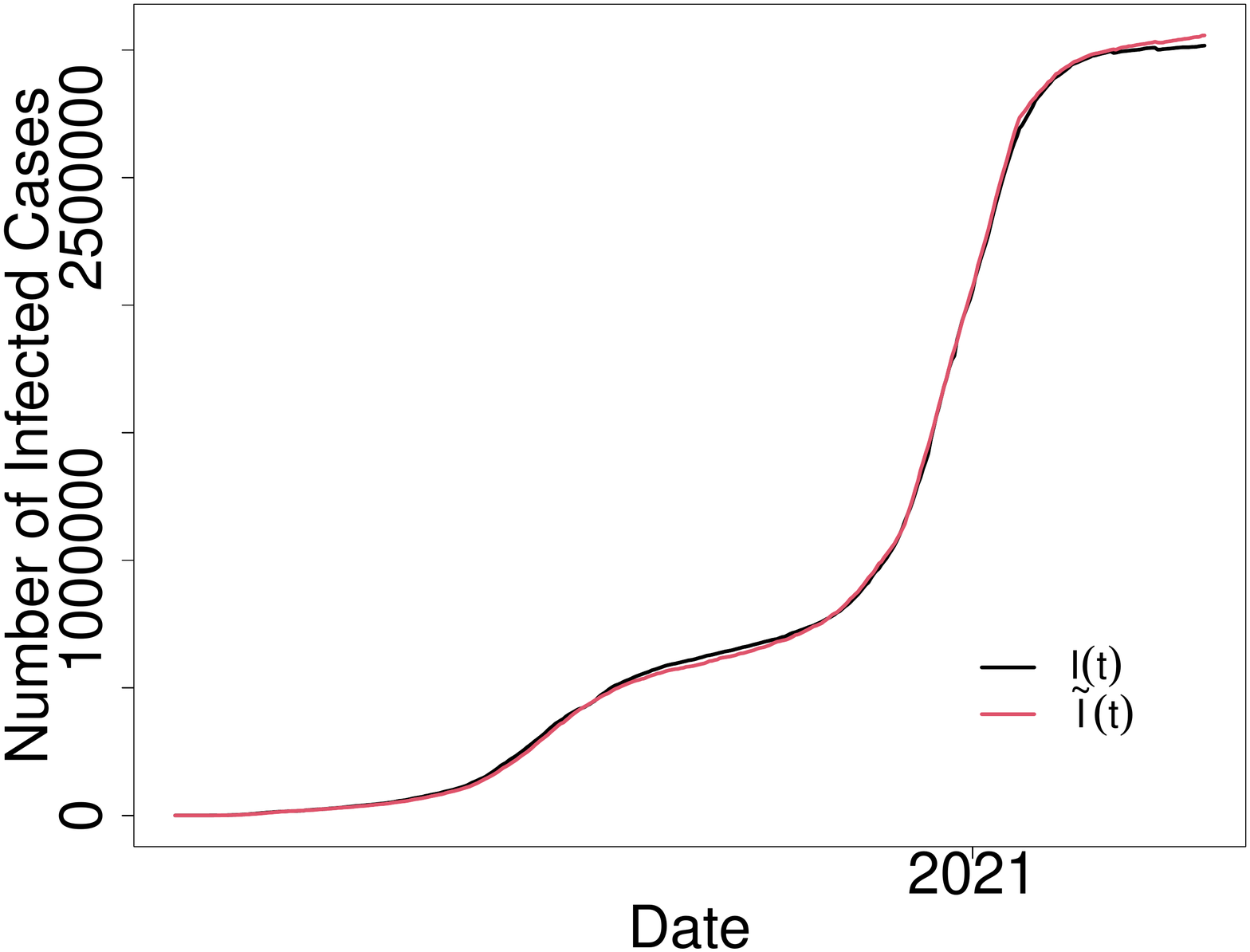}
         \subcaption{CA (Model 3)}
     \end{subfigure}
     \begin{subfigure}[b]{0.19\textwidth}
         \centering
         \includegraphics[width=\textwidth]{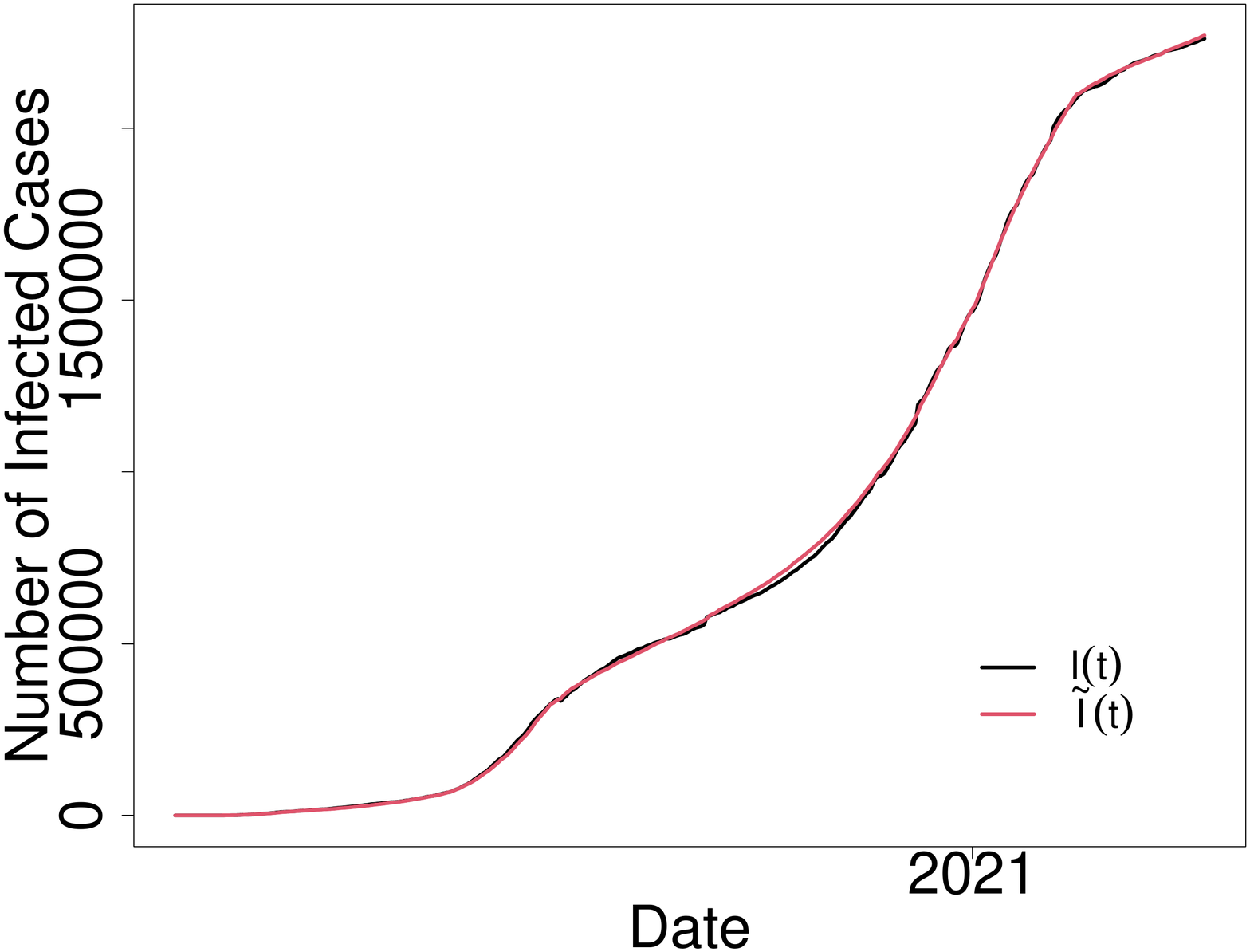}
         \subcaption{TX (Model 3)}
     \end{subfigure}}
        \caption{Observed (black) and fitted (red) number of infected cases estimated by three models in selected five states. 
        The pre-specified change points  are two fixed change points derived from the statewide lockdown date and reopening date.
       }
        \label{fig:number of infected}
\end{figure*}

To determine the significance level of the spatial effect in Model 2, we provide the estimate, p-value, and 95\% confidence intervals for the parameter $\alpha$ in Table \ref{table_alpha} in the Supplement. We find that the influence of the infected or recovered cases in adjacent states is statistically significant ($\mbox{p-value} < 0.05$) for all states.

Next, we assess the prediction performance for the three proposed models. The out-of-sample MRPE is used as the performance measurement, given by \eqref{eq:MRPE}.
We set the last two weeks of our observation period as the testing period and use the remaining time points for training the model. Note that  predicted number of infected and recovered cases are defined by \eqref{IR_est}.

The results of out-of-sample MRPE of $I(t)$ and $R(t)$  in the selected regions are reported in
\ref{table_MPE_all_out_refit}. The calculated MRPEs of $I(t)$ show that Model 3 which includes spatial effects and VAR temporal component outperforms the other models. Spatial smoothing itself (Model 2) reduces the prediction error significantly in some states. For example, in New York, the spatial smoothing reduced the MRPE (I) by 64\% when using Model 2.3 (similarity-based weight).
Finally, the reduction in MRPEs using Model~3 justifies the presence of the VAR component in the modeling framework. 
Additional results related to the VAR component (including the estimated auto-regressive parameters) are reported in Section \ref{sec:var-results} in the Supplement.

We also compare the developed model and associated methodology with the extended SIR, the ANN and the LSTM based models that were proposed in \cite{wangping2020extended}, \cite{moftakhar2020exponentially}, and \cite{chimmula2020time}, respectively.
The extended SIR model was trained using the ``eSIR" package in the \texttt{R} programming language.
The transmission rate modifier $\pi(t)$ was specified according to actual interventions at different times and regions as described in \cite{wangping2020extended}.
The ANN was trained using the ``nnfor" package in the \texttt{R} programming language. It comprises of 3 hidden layers with 10 nodes in each and a linear output activation function, which is the exact architecture in \cite{moftakhar2020exponentially}. 
The number of repetitions for this algorithm was set to be 20 for $I(t)$ and 10 for $R(t)$.
The LSTM architecture was implemented in \texttt{PyTorch}.
The \cite{chimmula2020time} did not specify the network architecture setting (number of layers, number of neurons). Thus, we performed a grid search over number of layers ranging from 1 to 3 with number of neurons as 10, 50, and 100. The best architecture in terms of minimizing the prediction error in the validation data is selected as the optimal LSTM architecture. Note that the prediction results with different network architecture in the LSTM were very similar. The selected number of layers and number of neurons based on grid search and corresponding prediction error for each region are also provided in Tables  \ref{table_MPE_I_out_refit_comparison} and  \ref{table_MPE_R_out_refit_comparison}.


The results of out-of-sample MRPE for $I(t)$ and $R(t)$ in the five states under consideration together with their sample standard deviations in the bracket are reported in
Tables \ref{table_MPE_I_out_refit_comparison} and  \ref{table_MPE_R_out_refit_comparison}, respectively. The proposed method clearly outperforms the extended SIR (eSIR) model across all five states for both $I(t)$ and $R(t)$. Further, it broadly matches the performance of ANN and LSTM for most states. Further, note that the proposed model is easy to interpret since its key parameters (infection and recovery rates) are routinely used by policy makers (see also the discussion in Section~\ref{subsec:counties}).

\begin{table}[ht!]
\caption{\label{table_MPE_I_out_refit_comparison}
Out-of-sample mean relative prediction error (MRPE) of $I(t)$. The number in between the brackets stands for the standard deviation of the relative prediction error .}
\centering
{
{ \begin{tabular}{lccccc} 
  \hline
  \hline
   & NY & OR & FL & CA & TX\\
   &MRPE(I)  &MRPE(I) &MRPE(I) &MRPE(I) &MRPE(I)\\
  \hline
  Model 1 & 0.0011(7e-04) & 9e-04(7e-04) & 0.0016(0.001) & 9e-04(2e-04) & 7e-04(4e-04) \\ 
  Model 2.1  & 0.0011(7e-04) & 0.001(8e-04) & 0.0018(7e-04) & 5e-04(2e-04) & 8e-04(6e-04) \\ 
  Model 2.2  & 0.0011(7e-04) & 0.001(7e-04) & 0.0017(7e-04) & 6e-04(2e-04) & 9e-04(8e-04) \\ 
  Model 2.3  & \textbf{4e-04}(5e-04) & \textbf{7e-04}(6e-04) & 0.002(8e-04) & \textbf{4e-04}(3e-04)& \textbf{6e-04}(4e-04) \\ 
  Model 2.4  & 0.0016(8e-04) & 0.001(8e-04) & 0.001(7e-04) & 7e-04(4e-04) & 8e-04(5e-04) \\ 
  Model 3  & \textbf{4e-04}(4e-04) & \textbf{7e-04}(6e-04) & 0.001(8e-04)& \textbf{4e-04}(3e-04) & \textbf{6e-04}(4e-04) \\ 
  eSIR  & 0.0062(9e-04)&0.0076(0.0013) & 0.0066(0.001)& 0.0106(6e-04)&0.0095(0.001)\\
    ANN  &5e-04(4e-04)& 0.0013(0.0013)& 8e-04(8e-04)&\textbf{4e-04}(3e-04)&8e-04(4e-04)\\ 
LSTM & 7e-04(4e-04)& 0.0012(0.001) & \textbf{7e-04}(8e-04)
& 6e-04(8e-04)& \textbf{6e-04}(5e-04)  \\
  \hline
  LSTM architecture (layer,neurons) & $(1,10)$ & $(3,10)$ & $(1,50)$& $(1,100)$&$(3,50)$\\
\hline
\end{tabular}} }
\end{table}

\begin{table}[!ht]
\caption{\label{table_MPE_R_out_refit_comparison}
Out-of-sample mean relative prediction error (MRPE)   of $R(t)$. The number in between the brackets stands for the standard deviation of the relative prediction error .}
\centering
{
{\begin{tabular}{lccccc} 
  \hline
  \hline
& NY & OR & FL & CA & TX\\
    &MRPE(R)  &MRPE(R) &MRPE(R) &MRPE(R) &MRPE(R)\\
  \hline
  Model 1  & 0.0019(3e-04) & 0.0028(0.0011) & 0.0028(7e-04)& 0.0042(8e-04) & 0.002(7e-04) \\ 
  Model 2.1  & 4e-04(3e-04) & \textbf{0.0017}(0.002) & 0.001(6e-04) & 0.0029(7e-04) & 0.0025(0.0029) \\ 
  Model 2.2  & 4e-04(3e-04) & 0.0018(0.0019) & 0.001(6e-04) & 0.0029(7e-04) & 0.0028(0.0042) \\ 
  Model 2.3  & 4e-04(2e-04) & 0.0022(0.0019) & 0.0012(8e-04) & 0.0022(7e-04) & 0.001(7e-04) \\ 
  Model 2.4  & 5e-04(8e-04) & 0.002(0.0024) & 9e-04(0.001) & 0.0029(0.001) & 0.0015(7e-04) \\ 
  Model 3  & 4e-04(2e-04) & 0.0022(0.0018) & 0.001(6e-04) & 0.0022(7e-04) & 0.001(7e-04) \\
    eSIR  &0.013(0.0011) &0.014(0.0031) & 0.0103(0.0012)& 0.0117(0.0012)&0.0113(0.001)\\
    ANN  &\textbf{3e-04}(2e-04)& \textbf{0.0017}(0.0021)& 8e-04(4e-04) & \textbf{0.0012}(0.001)& 8e-04(5e-04)\\ 
LSTM   &\textbf{3e-04}(3e-04)  &0.0025(0.0021) & \textbf{6e-04}(4e-04)&  0.0027(0.0036)& \textbf{6e-04}(6e-04) \\ 
  \hline
   LSTM architecture (layer,neurons) & $(1,10)$ & $(3,10)$ & $(1,50)$&$(1,100)$& $(3,50)$\\
\hline
\end{tabular}}}
\end{table}





\subsection{Results for Selected U.S. Counties}\label{subsec:counties}
We worked on nine counties/cities. Due to limited space, results for two counties in the state of California (Riverside and Santa Barbara) are presented here while rest are described in Section~\ref{sec:counties} in the supplementary materials. For determining their neighbors, a threshold of 100 miles is used. Model 2.3 uses the top five counties in the corresponding state with the smallest similarity score, while Model 2.4 uses all counties in the given state. The resulting neighbors for Model 2.3 are displayed in Table \ref{table_similar_all}.
The statewide and countywide policy start dates and the detected change points are shown in Tables~\ref{table_plan_all}.

\begin{table}[!ht]
\caption{\label{table_MPE_all_out_refit}
Out-of-sample mean relative prediction error (MRPE) of $I(t)$ and $R(t)$ for selected regions. }
\centering
{
\begin{tabular}{lcccccccccc} 
  \hline
  \hline
 & \multicolumn{2}{c}{New York}& \multicolumn{2}{c}{California} & \multicolumn{2}{c}{Riverside} & \multicolumn{2}{c}{Santa Barbara}\\
    &I  & R  &I  & R &I  & R &I  & R \\
  \hline
  Model 1 & 0.0011 & 0.0019 & 9e-04 & 0.0042 & 0.0056 & 0.0036 & 0.0037 & 0.0071 \\ 
  Model 2.1 & 0.0011 & 4e-04 & 5e-04 & 0.0029 & 0.0043 & 0.0028 & 0.0016 & 0.0039 \\ 
  Model 2.2 & 0.0011 & 4e-04 & 6e-04 & 0.0029 & 0.0045 & 0.0027 & 0.0014 & 0.0032 \\ 
  Model 2.3 & 4e-04 & 4e-04 & 4e-04 & 0.0022 & 0.0014 & 0.0025 & 0.0014 & 0.0028 \\ 
  Model 2.4 & 0.0016 & 5e-04 & 7e-04 & 0.0029 & 0.0013 & 0.0027 & 0.0017 & 0.0033 \\ 
  Model 3 & 4e-04 & 4e-04 & 4e-04 & 0.0022 & 0.0012 & 0.0016 & 0.0014 & 0.0028 \\ 
  \hline
\end{tabular}
}
\end{table}


In December, Southern California was experiencing a fast and sustained outbreak, believed to be driven by a new strain designated as CAL.20C \footnote{\url{https://www.newswise.com/coronavirus/local-covid-19-strain-found-in-over-one-third-of-los-angeles-patients2/}}. 
To that end, we analyzed the daily count of cases in Riverside and Santa Barbara counties in CA. As seen from Figure \ref{fig:rates_county_smooth} in the supplementary material, three of the detected change points occurred on April 17 2020, July 20 2020, and July 26 2020 in Riverside County. The first one can be related to the decreased transmission rate in April mainly caused by the statewide lockdown, while the July ones could be due to the pause of reopening to halt the spread of COVID-19.
Similarly, five change points are detected in Santa Barbara County. The first four change points can also be related to the lockdown, reopening and pause of reopening.
An additional change point in Riverside County and Santa Barbara County is detected on November 28 and December 27, respectively, which may be driven by the new CAL.20C variant. 
 


We also provide the out-of-sample MRPE of $I(t)$ and $R(t)$ of selected counties  in Table~  \ref{table_MPE_all_out_refit}. 
The MRPE of $I(t)$ results show that adding the spatial effect can significantly improve the MRPE of $I(t)$ in both the Riverside County and Santa Barbara County. Adding the VAR($p$) (Model 3) performs the best in  Riverside County.
In Riverside County, the spatial smoothing reduced the MRPE (I) by 76\% when using Model 2.4 (similarity-based weight) while in Santa Barbara County, the spatial smoothing reduced the MRPE (I) by approximately 62\% when using Model 2.3 (similarity-based weight).
 The MRPE of $R(t)$ results show that the piecewise constant model with spatial effect (Model 2) performs the best in  Santa Barbara County while Adding the VAR($p$) (Model 3) performs the best in Riverside County.  

{
\begin{figure*}[ht!]
     \centering
          \captionsetup[sub]{font=scriptsize, labelfont={bf,sf}}
        \begin{subfigure}[b]{0.3\textwidth}
         \centering
         \includegraphics[width=\textwidth]{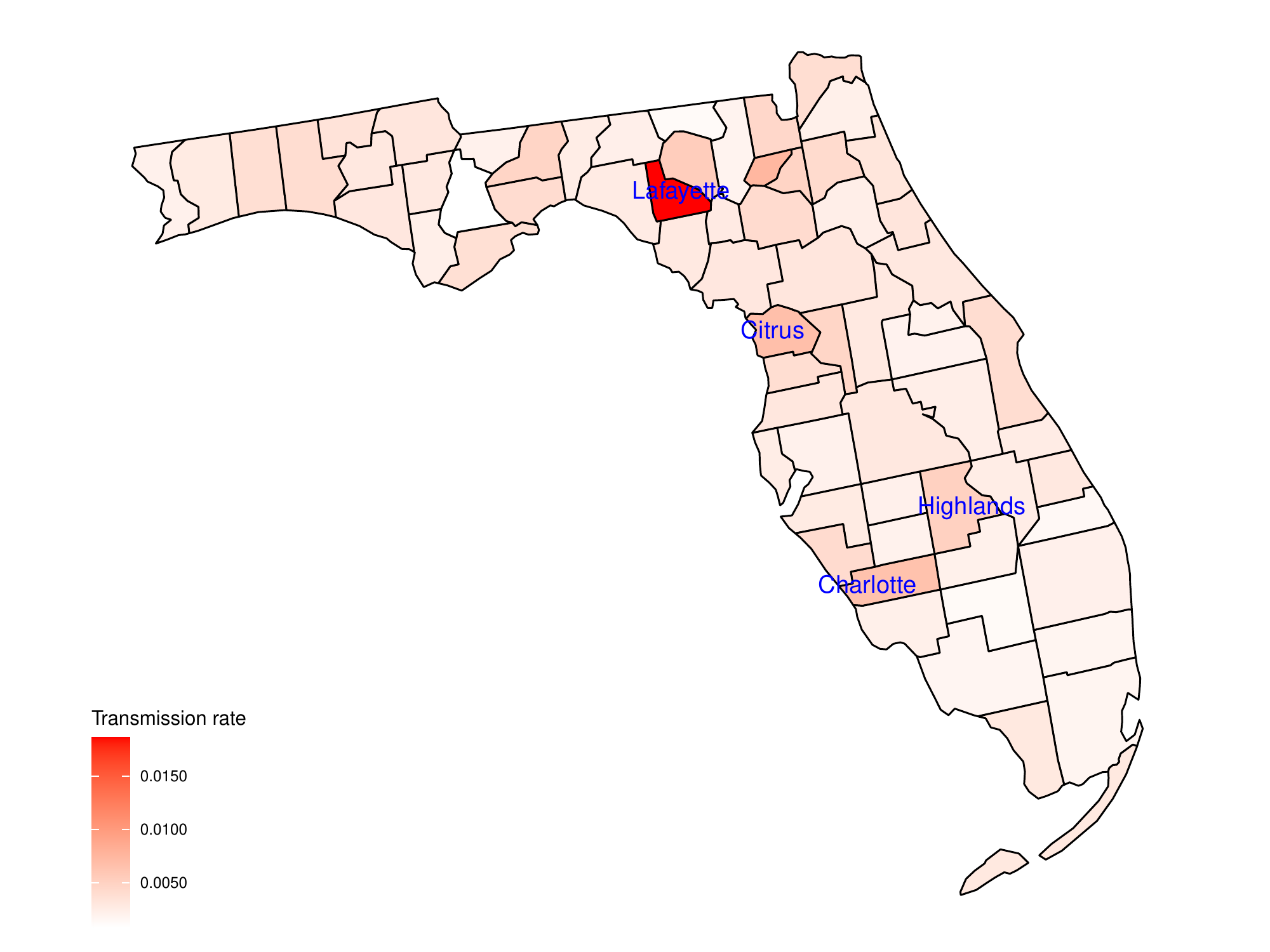}
         \subcaption{Averaged empirical transmission rate $\widetilde{\beta}$}
     \end{subfigure}      
     \begin{subfigure}[b]{0.3\textwidth}
         \centering
         \includegraphics[width=\textwidth]{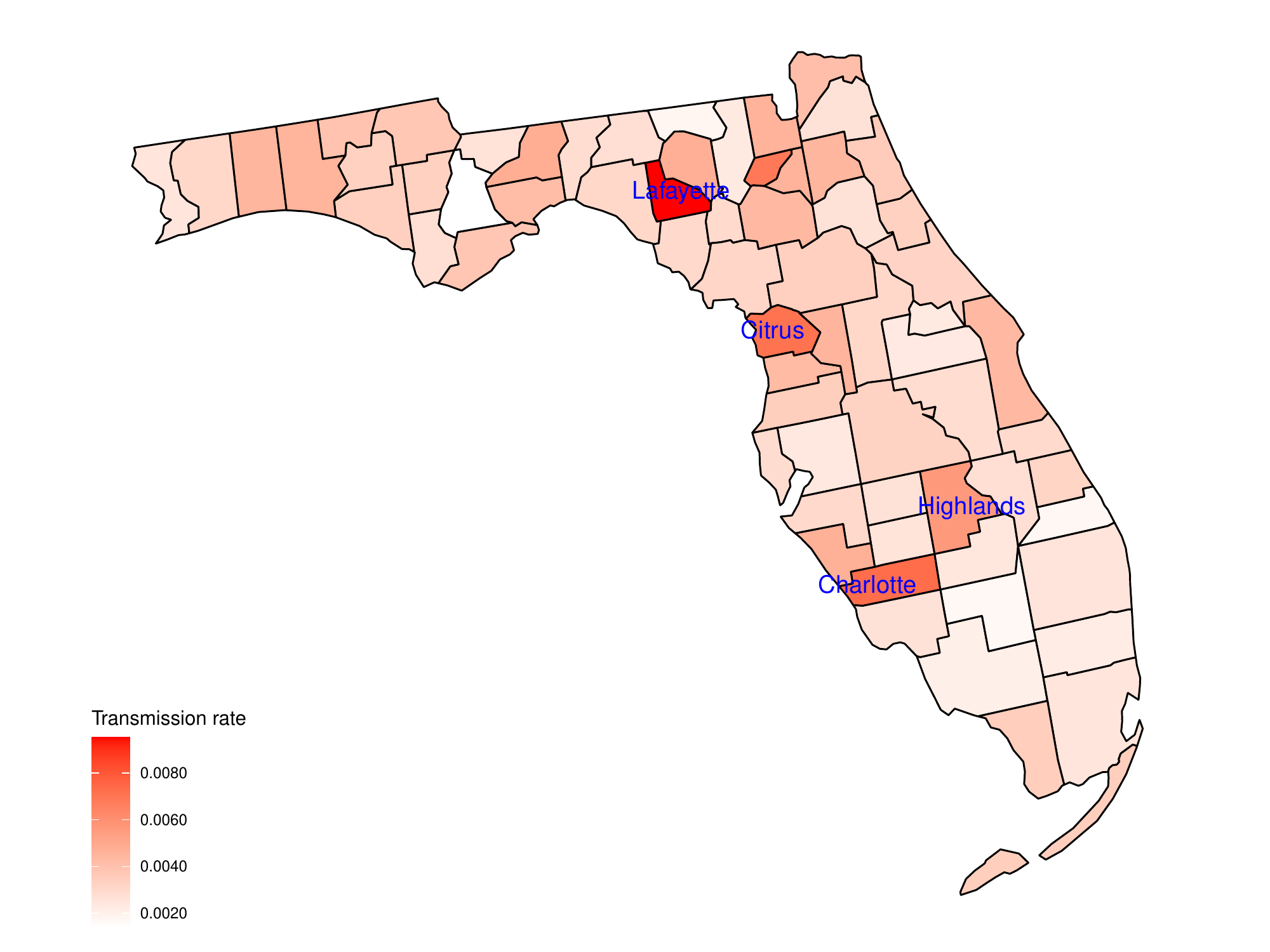}
         \subcaption{Estimated transmission rate $\widehat{\beta}$}
     \end{subfigure}
     \begin{subfigure}[b]{0.3\textwidth}
         \centering
         \includegraphics[width=\textwidth]{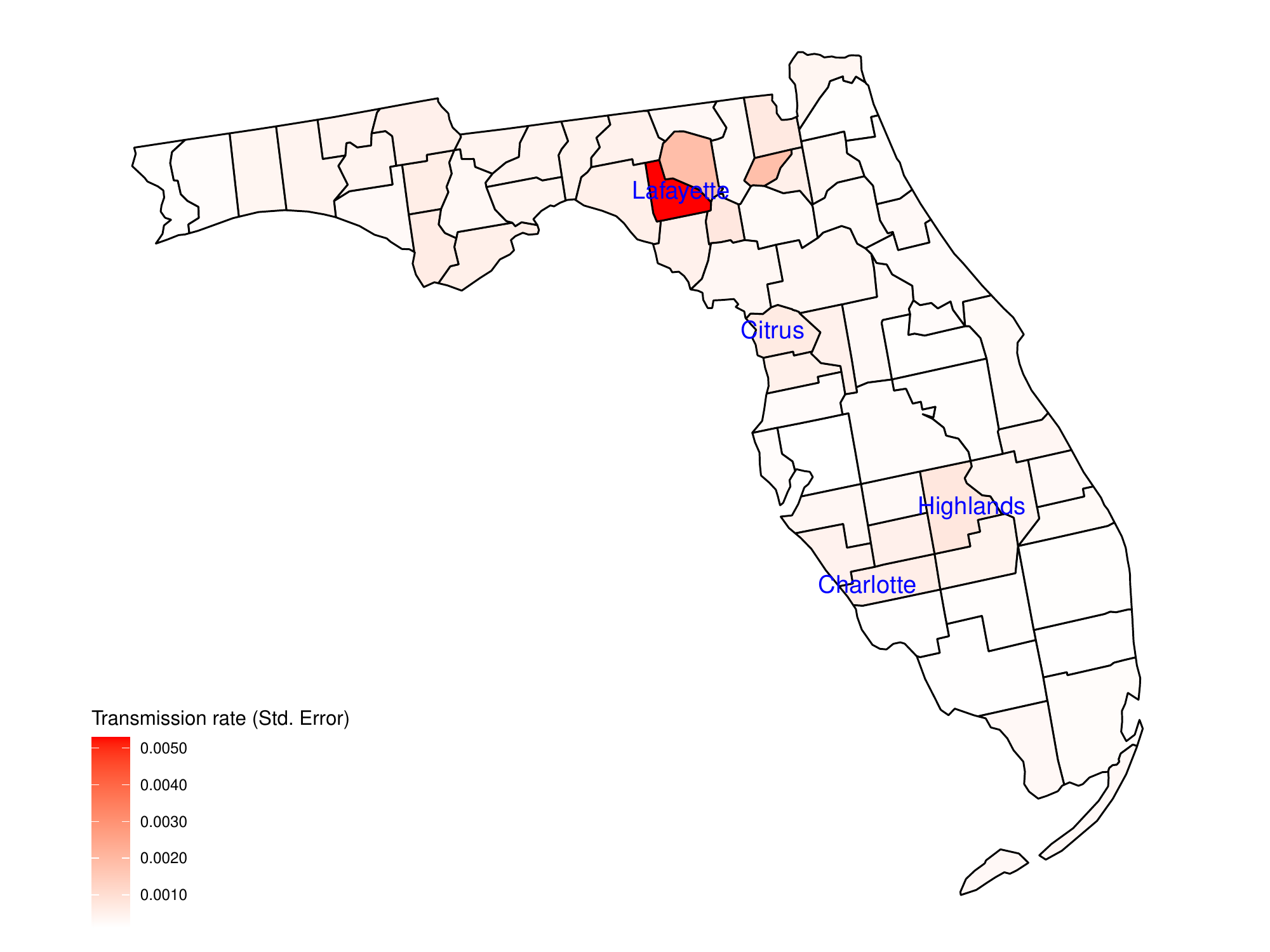}
         \subcaption{Standard deviation of estimated transmission rate $\widehat{\beta}$}
     \end{subfigure}
     
     \begin{subfigure}[b]{0.3\textwidth}
         \centering
         \includegraphics[width=\textwidth]{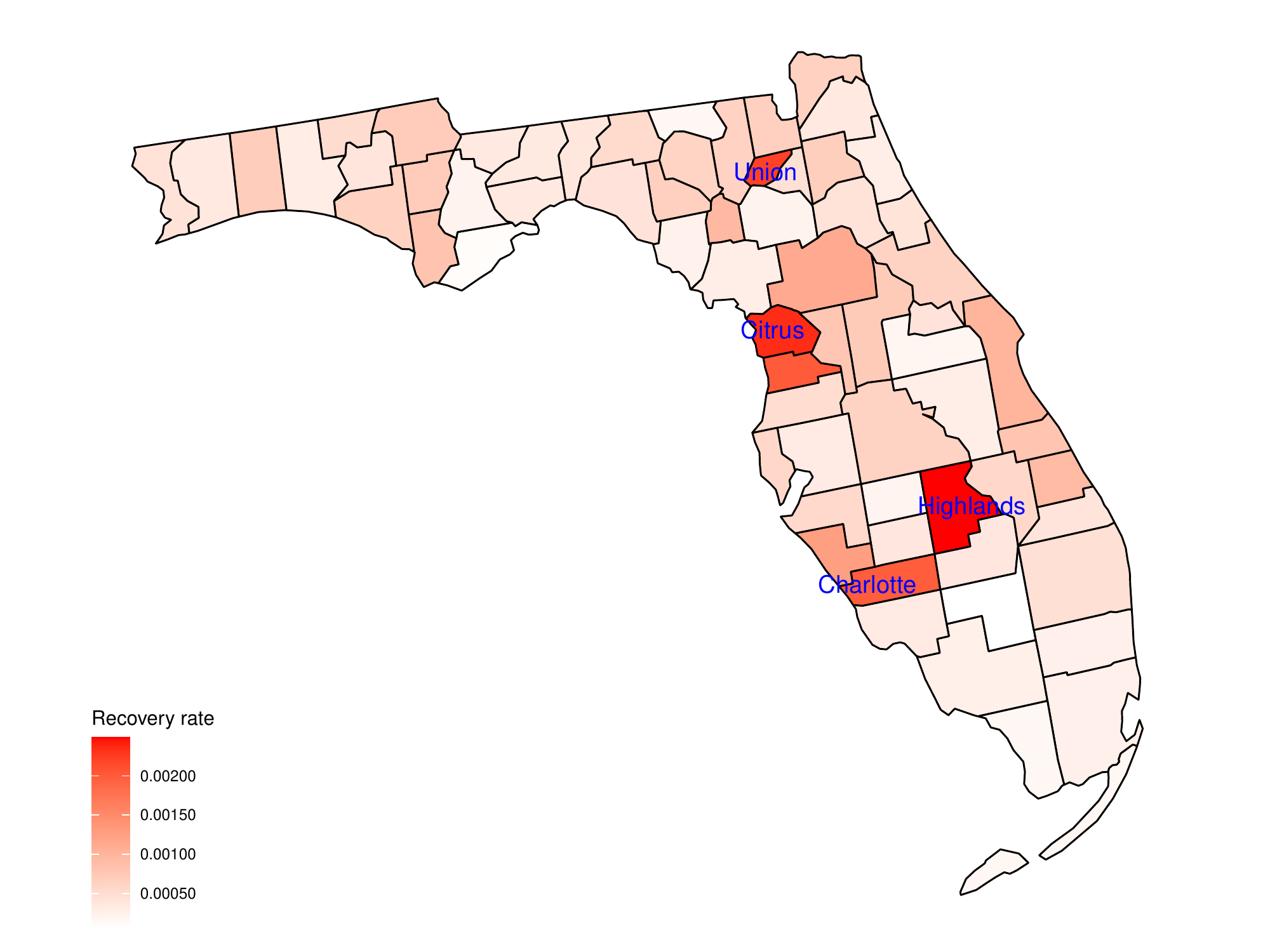}
         \subcaption{Averaged empirical recovery rate $\widetilde{\gamma}$}
     \end{subfigure} 
     \begin{subfigure}[b]{0.3\textwidth}
         \centering
         \includegraphics[width=\textwidth]{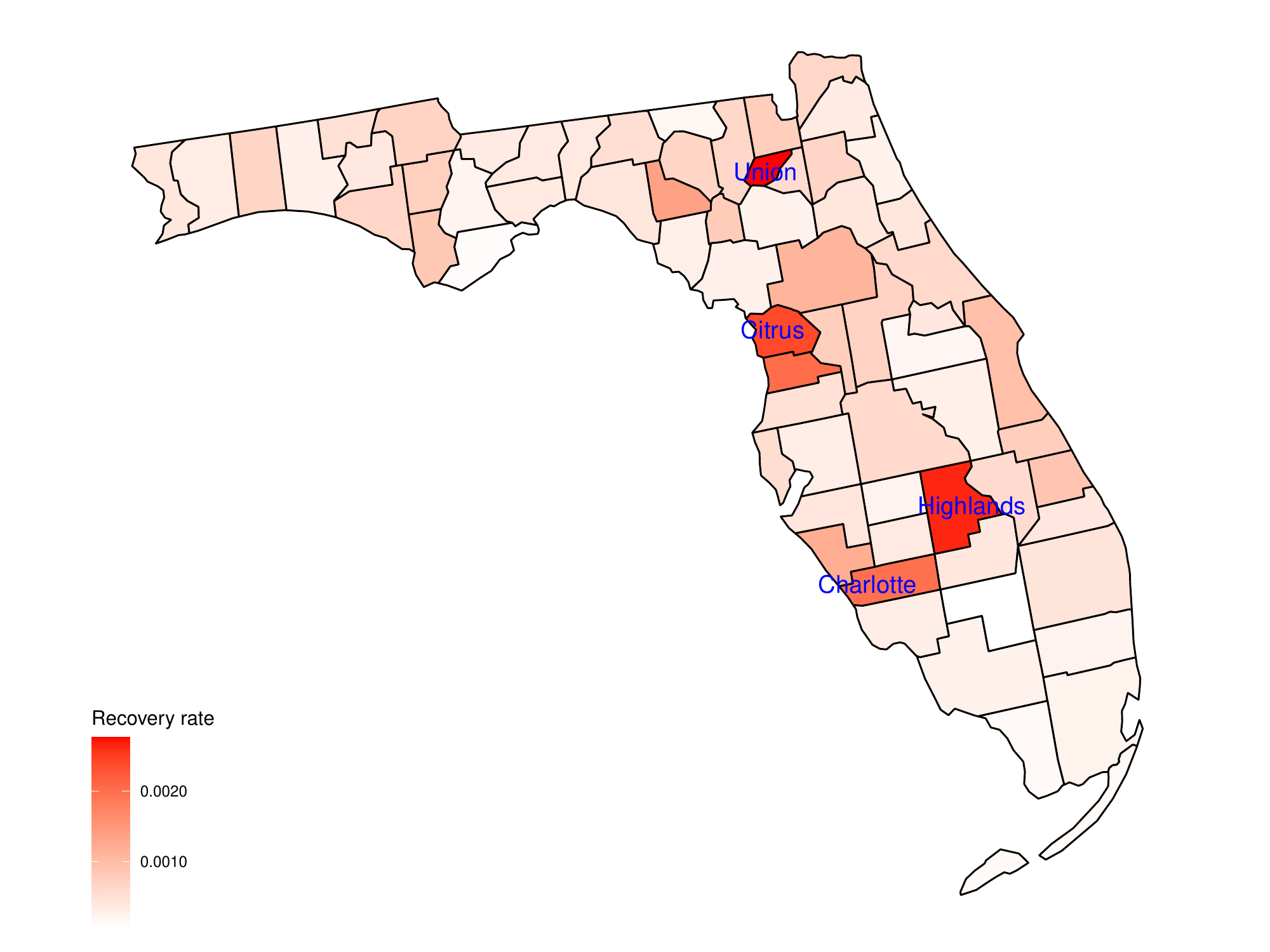}
         \subcaption{Estimated Recovery rate $\widehat{\gamma}$}
     \end{subfigure}
     \begin{subfigure}[b]{0.3\textwidth}
         \centering
         \includegraphics[width=\textwidth]{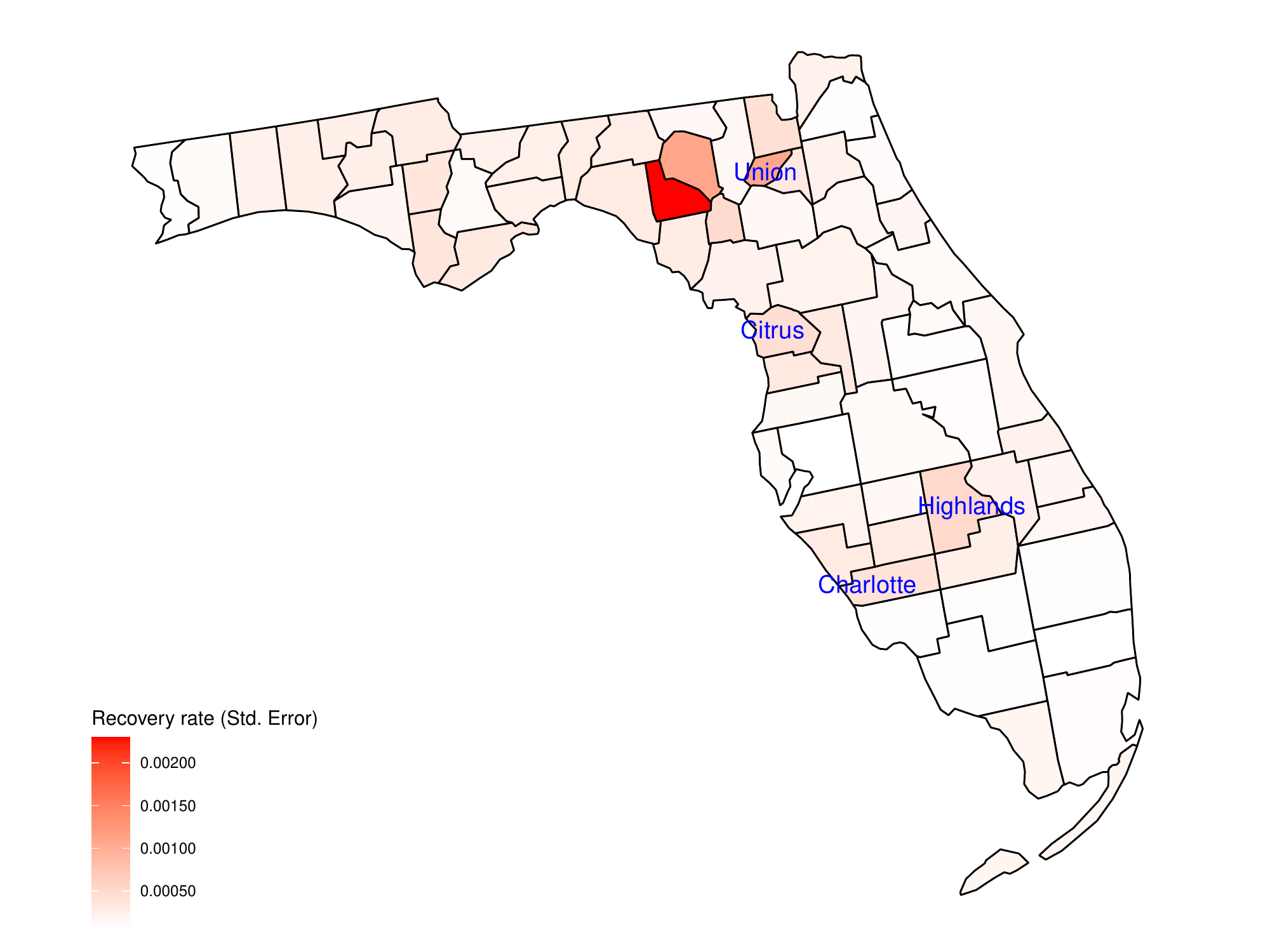}
         \subcaption{Standard deviation of estimated recovery rate $\widehat{\gamma}$}
     \end{subfigure}
        \caption{Heatmap of transmission rate and recovery rate in Florida (county-level).
        }
        \label{fig:FL_counties}
\end{figure*}

In Figure \ref{fig:FL_counties}, we provide heatmaps of the transmission and recovery rates in 67 counties in Florida, based on data from August 1st 2020 to December 1st 2020, a stable period in that state with no change point detected. The top left plot depicts the averaged empirical transmission rate  $\widetilde{\beta}$, wherein it is calculated as the average $\widetilde{\beta}(t)$ within the given time interval, i.e., 
$\widetilde{\beta}(t) =\left( \frac{\Delta I(t)}{(1- u(t+1))}+{\Delta R(t)} \right) /\left(\frac{S(t)}{N} I_f(t) \right)$; the top middle plot, the heatmeap of the estimated transmission rate  $\widehat{\beta}$ and finally the top right plot the heatmeap of the standard error of $\widehat{\beta}$ using linear regression analysis
(for more details, see Section~\ref{sec:heatmap} in the supplementary material).
The bottom plots are the corresponding heatmaps of the recovery rate, wherein the
bottom left plot depicts the averaged empirical recovery rate  $\widetilde{\gamma}$, which is the average $\widetilde{\gamma}(t)$ within the given time interval, and $\widetilde{\gamma}(t) = {\Delta R(t)} / I_f(t)$.
As shown in Figure \ref{fig:FL_counties}, the estimated rates are very close to the averaged empirical transmission and recovery rates. This point confirms the interpretability of the proposed hybrid modeling framework. During this time period, Lafayette County has the highest transmission rate,
while Union County has the highest recovery rate. 
Moreover, Citrus, Charlotte and Highlands Counties have both (significantly) high transmission and recovery rates based on standard errors of the estimated parameters.
}

\section{Concluding Remarks}\label{sec:conclusion}
COVID-19 has posed a number of challenges for modellers, both due to the lack of adequate data (especially early on in the course of the pandemic) and its characteristics (relative long period before emergence of symptoms compared to SARS and other respiratory viruses). A plethora of models -a number of them briefly summarized in the introductory section- were developed, most aiming to provide short and long term predictions of the spread of COVID-19. This work contributes to that goal by developing a hybrid model that enhances a piecewise stationary (mechanistic) SIR model with neighboring effects and temporal dependence to model the spread of COVID-19 at both state-level and county-level in the United States.
The reasonable forecasts of Model 3 (including spatial effects and the VAR component) confirm the existence of spatial and temporal dependence among new daily cases which can not be accounted by the homogeneous deterministic SIR model. Further, the detection of change points in  neighboring counties can provide insights into how the spread of COVID-19 impacted different communities at different points in time and also that of mitigation policies adopted by county (state) health administrators.


%

\vspace{-0.5cm}

{\appendices

\section{Algorithm Details}\label{sec:algo}


The key steps in our proposed detection strategy are summarized next while a summary is outlined in Algorithm~\ref{alg:piecewise constant} in the supplementary material. First, few notations are defined. \\


\noindent
\textit{Notation.} 
Denote the indicator function of a subset A as $\mathbbm{1}_{A}$.
$\mathbb{R}$ denotes the set of real numbers. 
For any vector $v \in \mathbb{R}^p$, 
we use $\Vert v \Vert_{1}$, $\Vert v \Vert_{2}$  and $\Vert v \Vert_{\infty}$  to denote $\sum_{i = 1}^p\vert v_i\vert$, $\sqrt{\sum_{i = 1}^p\vert v_i\vert^2}$ and $\max_{1 \leq i\leq p}\vert v_i\vert$, respectively.
The transpose of a matrix $\bm{A}$ is denoted by $\bm{A}^\prime$.\\

\vspace{-0.3cm}

\noindent
\textit{Steps of the Proposed Algorithm.}

\noindent
\textit{Block Fused Lasso:} the objective is to partition the observations into blocks of size $b_n$, wherein the model parameter $B_t$ remains fixed within each block and select only those blocks for which the corresponding change in the parameter vector is much larger than the others. Specifically, let $n = T-1$ be the number of the times points for the response data $Y_t$ and define a sequence of time points $1 = r_0 < r_1 < \dots < r_{k_n} = n+1$ for block segmentation, such that $r_{i} - r_{i-1} = b_n $ for $i =1 ,\dots, k_n-1$, $b_n \leq r_{k_n} - r_{k_n-1} < 2b_n $ , where $k_n = \lfloor \frac{n}{b_n} \rfloor$ is the total number of blocks. For ease of presentation, it is further assumed that $n$ is divisible by $b_n$ such that $r_{i} - r_{i-1} = b_n$ for all $i =1 ,\dots, k_n$. By partitioning
the observations into blocks of size $b_n$ and fixing the model parameters within each block, we set $\theta_1 = B^{(1)}$ and 
\begin{equation*}
  \theta_i = \begin{cases}
      B^{(j+1)} -B^{(j)} , & \text{when}\ t_j \in [r_{i-1}, r_{i})\  \text{for some}\ j  \\
      \bm{0}, & \text{otherwise,}
    \end{cases}  
\end{equation*}
for $i = 2,3,\dots, k_n$.
Note that $\theta_i \neq \bm{0}$ for $i \geq 2$ implies that $\theta_i$ has at least one non-zero entry and hence a change in the parameters. Next, we formulate the following linear regression model in terms of $\Theta(k_n) = (\theta_1^\prime, \dots, \theta_{k_n}^\prime)^\prime $:
\begin{equation}
\underbrace{\left(\begin{array}{c}
\bm{Y}_{1}\\
\bm{Y}_{2}\\
\vdots \\
\vdots \\
\bm{Y}_{k_n}
\end{array}\right)}_{\mathcal{Y}}= 
\underbrace{\left(\begin{array}{cccccc}
\bm{X}_{1}  & \mathbf{0} &\dots&\dots & \mathbf{0} \\
\bm{X}_{2}  & \bm{X}_{2} &\dots&\dots & \mathbf{0} \\
\vdots & \vdots & \ddots &  &\vdots  \\
\vdots & \vdots &  & \ddots & \vdots \\
\bm{X}_{k_n}  & \bm{X}_{k_n}  &\dots&\dots & \bm{X}_{k_n} \\
\end{array}\right)}_{\mathcal{X}}
\underbrace{
\begin{pmatrix}
{\theta_{1}} \\
{\theta_{2}}\\
\vdots \\
\theta_{k_n}
\end{pmatrix}}_{\Theta(k_n)}
+\underbrace{
\left(\begin{array}{c}
\bm{\mathcal{E}}_1\\
\bm{\mathcal{E}}_2\\
\vdots \\
\vdots \\ 
\bm{\mathcal{E}}_{k_n}\\
\end{array}\right)}_{E},   \nonumber
\end{equation}
where $\bm{Y}_i = \left(Y_{r_{i-1}}, \dots, Y_{r_{i} - 1} \right)^\prime$, $\bm{X}_i = \left(X_{r_{i-1}}, \dots, X_{r_{i} - 1} \right)^\prime$,
$\bm{\mathcal{E}}_i = \left({\varepsilon}_{r_{i-1}}, \dots, \varepsilon_{r_{i} - 1} \right)^\prime$,
$i = 1, \dots, k_n$.
 $\mathcal{Y} \in \R^{2n} $, $\mathcal{X} \in \R^{2n \times 2k_n } $, $\Theta(k_n) \in \R^{2k_n}$ and $E \in \R^{2n } $.

 A simple estimate of parameters $\Theta(k_n)$ can be obtained by using an $\ell_1$-penalized least squares regression of the form
\begin{equation}\label{eq:estimation_block}
\widehat{\Theta}(k_n) = \argmin_{\Theta\in \R^{2k_n}} \left\{ \frac{1}{2n} \| \mathcal{Y} -\mathcal{X}\Theta \|_2^2 + \lambda_{n} \|\Theta \|_1  \right \} ,
\end{equation}
which uses a fused lasso penalty to control the number of change points 
in the model. This penalty term  encourages the parameters across consecutive time blocks to be similar or even identical; hence, only large changes are registered, thus aiding in identifying the change points.
Further, a hard-thresholding procedure is added to cluster the jumps into two sets: large and small ones,  so that those redundant change points with small changes in the estimated parameters can be removed.  We only declare that there is a change point at the end point of a block, when associated with large jump of the model parameters.

\noindent
\textit{Hard Thresholding:} is based on a data-driven procedure for selecting the threshold $\eta$. 
The idea is to combine the $K$-means clustering method \cite{hartigan1979algorithm} with the BIC criterion \cite{schwarz1978estimating} to cluster the changes in the parameter matrix into two subgroups. The main steps are:
\begin{itemize}
    \item Step 1 (initial state): Denote the jumps for each block by $v_k = \left\|\widehat{\theta}_k \right\|_2^2 $, $k = 2, \cdots, k_n$ and let $v_1 = 0 $. Denote the set of selected blocks with large jumps as $J$ (initially, this is an empty set) and set $\text{BIC}^{old} = \infty $.  
     \item Step 2 (recursion state): 
     Apply $K$-means clustering to the jump vector $V= (v_1, v_2, \cdots, v_{k_n})$ with two centers. Denote the sub-vector with a smaller center as the small subgroup, $V_S$ , and the other sub-vector as the large subgroup, $V_L$.
     Add the corresponding blocks in the large subgroup into $J$. Compute the BIC by using the estimated parameters $\widehat{\Theta}$ after setting $\widehat{\theta}_{i} = \bm{0}$ for each block $i \notin J$ and denote it by $\text{BIC}^{new}$. Compute the difference $\text{BIC}^{\text{diff}}= \text{BIC}^{new} - \text{BIC}^{old}$and update $\text{BIC}^{old}= \text{BIC}^{new}$. Repeat this step until $\text{BIC}^{\text{diff}} \geq 0$. 
\end{itemize}

\noindent
\textit{Block clustering:} the Gap statistic \cite{tibshirani2001estimating} is applied to determine the number of clusters of the candidate change points. The basic idea is to run a clustering method (here, $K$-means is selected) over a grid of possible number of clusters, and to pick the optimal one by comparing the changes in within-cluster dispersion with that expected under an appropriate reference null distribution (for more details, see Section 3 in \cite{tibshirani2001estimating}).

\noindent
\textit{Exhaustive search:}
Define $l_i = (\min(C_i)-b_n) \mathbbm{1}_{\{\vert C_i\vert = 1\}} + \min (C_i) \mathbbm{1}_{\{\vert C_i\vert > 1\}}$ and $u_i = (\max(C_i)+b_n) \mathbbm{1}_{\{\vert C_i\vert = 1\}} + \max (C_i) \mathbbm{1}_{\{\vert C_i\vert > 1\}}$, where $C_i$'s are the subsets of candidates blocks by block clustering procedure. Denote the subset of corresponding block indices by $J_i$.
Define the following local coefficient parameter estimates:
\begin{equation}\label{eq:beta_est_local}
  {\widehat{\bm{B}}}_{i} = \sum_{k=1}^{\frac{1}{2}\left(\max(J_{i-1})+\min(J_{i})\right)} \widehat{{\theta}}_k, \text{ for } i = 1, \dots, \widetilde{m}^f + 1,
\end{equation}
where  $\widetilde{m}^f$ is the number of clusters obtained in the block clustering procedure,
 $J_0 = \{1\}$ and $J_{\widetilde{m}^f+1} = \{k_n\}$.

Now, given a subset $C_i$,  we apply the exhaustive search method for each time point $s$ in the interval $(l_i, u_i)$ to the data set truncated by the two end points in time, $ \min(C_i)-b_n $ and $ \max(C_i)+b_n$, i.e. only consider the data within the interval $ [\min(C_i)-b_n , \max(C_i)+b_n) $. Specifically, define the final estimated change point $ \widetilde{t}_i^f $ as
\begin{align}\label{eq:Exhaust:final}
\widetilde{t}_i^f = \argmin_{s \in (l_i, u_i)} \Bigg \{ \sum_{t= \min(C_i)-b_n}^{s-1} \left\| Y_{t}  - X_{t} \widehat{\bm{B}}_{i}\right\|_2^2  \nonumber\\
+  \sum_{t= s}^{\max (C_i)+b_n-1} \left\| Y_{t}  - X_{t}\widehat{\bm{B}}_{i+1} \right\|_2^2 \Bigg \},
\end{align}
for $ i = 1, \ldots, \widetilde{m}^f $, where $\widehat{\bm{B}}_{i}$'s  are the local coefficient parameter estimates based on the first step by block fused lasso.
Denote the set of final estimated change points from \eqref{eq:Exhaust:final} by $ \widetilde{\mathcal{A}}_n^f = \left\lbrace \widetilde{t}_1^f, \ldots, \widetilde{t}_{\widetilde{m}^f}^f  \right\rbrace $.

\noindent
{\em Remark:} An alternative approach for detection of break points is to run a full exhaustive search procedure for both single and multiple change point problems. Such procedures are computationally expensive, and not scalable for large data sets. Simple fused lasso is (block of size one) another method, which although computationally fast, it leads to over-estimating the number of break points; hence, it requires additional ``screening'' steps to remove redundant break points found using the fused lasso algorithm \cite{harchaoui2010multiple,safikhani2020joint}. Such screening steps usually include tuning several hyper-parameters. This task not only slows down the detection method, but is also not robust. The proposed approach (block fused lasso coupled with hard-thresholding) aims to solve this issue by choosing appropriate block sizes, while it only needs a single tuning parameter (the threshold) to be estimated.

\noindent
{\em Estimation of Infection and Recovery rates:} Once the locations of break points are obtained, one can estimate the model parameters by running a separate regression for each identified stationary segment of the time series data. The work of \cite{safikhani2020joint} shows that this strategy yields consistent model parameter estimates.

\noindent
{\em Grid search of parameter in the under-reporting function:} 
We use grid search to estimate the parameter $a$ in the under-reporting function $u(t)$.  
Given a parameter grid of $a$, we transform the observed infected data $I(t)$ by $I_f(t) = \Delta I(t)/ (1-u(t) )  +I_f(t-1) $, then apply the transformed data to the above method and compute the in-sample mean relative prediction error (MRPE) of $\Delta I_f(t)$. Choose the value $a$ that minimizes the in-sample MRPE of $\Delta I_f(t)$. 

\section{Theoretical Properties}\label{sec:main_proof}

In this section, we establish the prediction consistency of the estimator from \eqref{eq:estimation_block}.
To  establish  prediction consistency  of  the  procedure,  the  following    assumptions  are needed:

\begin{itemize}
\item [(A1.)] (Deviation bound)
 There exist constants $c_i > 0$ such that with probability at least $1  - c_1 \text{exp} ( - c_2 (\log 2n ))$, we have
 \begin{equation}\label{eq:db}
     \left\Vert  \frac{\mathcal{X}^\prime E}{2n} \right\Vert_{\infty} \leq c_3 \sqrt{\frac{\log 2n}{2n}}.
 \end{equation}
 \item[(A2.)] There exists a positive constant $M_{B}  > 0$ such that
    \[\text{max}_{1 \leq j \leq m_0 +1}\lVert B^{(j)} \rVert_{\infty} \leq M_{B} .\]
\end{itemize}

\begin{theorem}{1}\label{thm:pred_consistency}
Suppose A1-A2 hold. Choose 
$\lambda_{n} = 2C_1 \sqrt{\frac{ \log 2n  }{2n}}$
for some large constant $C_1> 0$, and assume $m_0 \leq m_n$ with $m_n = o(\lambda_{n}^{-1})$. Then with high probability approaching to 1 and $n \rightarrow + \infty$
, the following holds:
\begin{equation}
   \frac{1}{2n} \left\| \mathcal{X}(\widehat{\Theta} - \Theta) \right\|_2^2   \leq 8 M_{B}\lambda_{n} m_0.
\end{equation}
\end{theorem}

\begin{proof}[Proof of Theorem 1]
By the definition of $\widehat{\Theta}$ in \eqref{eq:estimation_block}, the value of the function in \eqref{eq:estimation_block} is minimized at $\widehat{\Theta}$. Therefore, we have
\begin{equation} 
\frac{1}{2n}\left \| \mathcal{Y} -\mathcal{X}\widehat{\Theta} \right\|_2^2 + \lambda_{n} \left \|\widehat{\Theta} \right\|_1 \leq \frac{1}{2n} \left\| \mathcal{Y} -\mathcal{X}\Theta \right\|_2^2 + \lambda_{n} \|\Theta \|_1.
\label{eq:obj_1}
\end{equation}
Denoting $\mathcal{A}= \{t_1, t_2, \dots, t_{m_0}\}$ as the set of true change points, we have
{\small
\begin{align}
&\frac{1}{2n} \left\| \mathcal{X}\left(\widehat{\Theta} - \Theta \right) \right\|_2^2 \nonumber\\
\leq& \frac{1}{n} \left(\widehat{\Theta} - \Theta \right)^\prime \mathcal{X}^\prime E +
\lambda_{n} \left (  \left \|\Theta \right \|_1  - \left\|\widehat{\Theta} \right\|_1\right )
\nonumber\\
\leq &\frac{1}{n} \left(\widehat{\Theta} - \Theta \right)^\prime \mathcal{X}^\prime E +
\lambda_{n} \left ( \sum_{i=1}^{k_n}\|\theta_i\|_1 - \sum_{i=1}^{k_n}\|\widehat{\theta}_i \|_1  \right )
\nonumber\\
\leq& 2 \sum_{i=1}^{k_n} \left\Vert\widehat{\theta}_i - \theta_i \right\Vert_1 \left\Vert \frac{\mathcal{X}^\prime E}{2n} \right\Vert_{\infty} 
+
\lambda_{n} \sum_{i\in \mathcal{A}} \left (  \|\theta_i\|_1  - \left \|\widehat{\theta}_i \right\|_1 \right ) \nonumber\\
& \quad \quad \quad - \lambda_{n}  \sum_{i\notin \mathcal{A}}\|\widehat{\theta}_i \|_1 
\nonumber\\
\leq & \lambda_{n}  \sum_{i\in \mathcal{A}} \left \|\widehat{\theta}_i -\theta_i \right \|_1  +
\lambda_{n} \sum_{i\in \mathcal{A}} \left (  \left\|\theta_i\right\|_1  - \left \|\widehat{\theta}_i \right\|_1 \right ) 
\nonumber\\
\leq  &
2\lambda_{n} \sum_{i\in \mathcal{A}}  \|\theta_i\|_1  
\nonumber\\
\leq & 
2\lambda_{n} m_0 \max_{1 \leq j\leq m_0 } \left\|B^{(j+1)} - B ^{(j)}\right\|_1  
\nonumber\\
\leq  &
8 M_B \lambda_{n} m_0,
\nonumber
 \label{eq:ineq_1}
\end{align}}
with high probability approaching to deviation bound in \eqref{eq:db}.
This completes the proof.
\end{proof}

Theoretical properties of lasso have been have been studied by several authors \cite{ bickel2009simultaneous, van2011adaptive, loh2012, negahban2012unified}.  In controlling the statistical error, a suitable deviation conditions on $\mathcal{X}^\prime E/{2n}$ is needed. 
 The deviation bound conditions (e.g. the assumption A1) are known to hold with high probability under several mild conditions.  Under the condition that the error term $E \sim \mathcal{N}(0, \sigma^2 I_{2n})$, the deviation bound condition holds with high probability by Lemme 3.1 in \cite{van2011adaptive}.  Given that the $p$ (the number of time series components) is small and fixed, we have $n \gg \log p$, therefore, in the case where the $\mathcal{X}$ is a zero-mean sub-Gaussian matrix with parameters ($\Sigma_x$, $\sigma_x^2$), and the error term $E$ is a zero-mean sub-Gaussian matrix with parameters ($\Sigma_e$, $\sigma_e^2$), the deviation bound condition holds with high probability by Lemme 14 in \cite{loh2012}.

\textit{Detection Accuracy}. When the block size is large enough, such that $\log n / b_n$ remains small, if the selected change point $ \widehat{t}_j $ is close to a true change point , the estimated $ \widehat{{\Theta}}_j $ will be large (asymptotically similar to the true jump size in the model parameters); if the selected change point $ \widehat{t}_j $ is far away from all the true change points, the estimated $ \widehat{{\Theta}}_j $ will be quite small (converges to zero as sample size tend to infinity).  
Therefore, after the hard-thresholding, the candidate change points that are located far from any true change points will be eliminated. In other words, for any selected change point $ \widetilde{t}_j  \in \widetilde{\mathcal{A}}_n$, there would exist a true change point $t_{j'} \in \mathcal{A}_n$ close by, with the distance being at most $ b_n$. Thus, the number of clusters (by radius $ b_n$ ) seems to be a reasonable estimate for the true number of break points in the model.

On the other hand,  since the set of true change points $\mathcal{A}_n$ has cardinality less than or equal to the cardinality of the set of selected change points $\widetilde{\mathcal{A}}_n$, i.e., $m_0 \leq \widetilde{m}$, there may be more than one selected change points remaining in the set $\widetilde{\mathcal{A}}_n$ in $b_n$-neighborhoods of each true change point. For a set $A$, define cluster $(A, x)$ to be the minimal partition of $A$, where the diameter for each subset is at most $x$. Denote the subset in $ \mbox{cluster} \left( \widetilde{\mathcal{A}}_n, b_n \right)$ by $ \mbox{cluster} \left( \widetilde{\mathcal{A}}_n, b_n \right) = \left\lbrace C_1, C_2,  \ldots, C_{m_0} \right\rbrace $, where each subset $C_i$ has a diameter at most $ b_n $, i.e., $\max_{a,b \in C_i} \vert a - b\vert \leq b_n$.
Then with high probability converging to one, the number of subsets in $ \mbox{cluster} \left( \widetilde{\mathcal{A}}_n, b_n \right) $ is exactly $ m_0 $. All candidate change points in $\widetilde{A}_n$ are within a $b_n$-neighborhood of at least one true change point and therefore, with high probability converging to one, there is a true change point $t_i$ within the interval $(C_i-b_n , C_i+b_n)$. 
The distance between the estimated change point and the true change point will be less then $2b_n$. Therefore, by selecting $b_n = c \log n$ for a large enough constant $c>0$, one can conclude that the proposed detection algorithm locates the true break points with an error bounded by the order $ \mathcal{O} \left( \log n \right)$.

}





\ifCLASSOPTIONcaptionsoff
  \newpage
\fi



\bibliography{bibtex/bib/ref.bib}{}
\bibliographystyle{IEEEtran}
\newpage
\title{Supplementary Material for Hybrid Modeling of Regional COVID-19 Transmission Dynamics in the U.S.}
\author{Yue Bai, 
        Abolfazl Safikhani,
        and~George Michailidis
}
\maketitle

\setcounter{section}{0}
In this file, summary of main steps of the proposed detection algorithm are presented in Section~\ref{sec:algo_more} while additional simulation settings are reported in Section~\ref{sec:sim}.
In Section~\ref{sec:states}, additional COVID-19 data analysis results for U.S. states are reported while additional results for U.S. counties are reported in Section~\ref{sec:counties}. Finally, additional details of VAR(p) model results are summarized in Section~\ref{sec:var-results}.

\section{Algorithm Details}\label{sec:algo_more}
The key steps in our proposed change point detection strategy are outlined in Algorithm~\ref{alg:piecewise constant}.

\begin{algorithm}[!ht]
    \DontPrintSemicolon
    \KwInput{Time series data $\{I_u(t), R(t)\}, t = 1, 2, \cdots, T$; regularization parameter $\lambda_n$; block size $b_n$ (value of $n$ as a function of $T$ specified in the sequel).
    }
    
    \textbf{Block Fused Lasso}: Partition the time points into blocks of size $b_n$ and fix the coefficient parameters within each block. Then, obtain candidate change points $\widehat{A}_n$ together with estimates of parameters  for each  block $\widehat{\Theta}$ by solving \eqref{eq:estimation_block}. 
    
    \textbf{Hard-thresholding}: For $i= 2, \cdots, k_n$, if $ \|\widehat{\theta}_i \|_2^2 \leq \eta$, declare that there are no change-points within the $i$-th block; otherwise, select the first time point $r_{i-1}$ in the $i$-th block as a candidate change point. Denote the set of candidate change points by $\widetilde{A}_n$.
    
    \textbf{Block clustering}:  Partition the $\widetilde{m}$ candidate change points $\widetilde{\mathcal{A}}_T = \left\{\widetilde{t}_1 ,\cdots \widetilde{t}_{\widetilde{m} } \right\}$ into $m_0 (\leq \widetilde{m})$ subsets $C = \{C_1, C_2, \cdots , C_{m_0}\}$, so as to minimize the within-cluster distance. 
    
    \textbf{Exhaustive Search}:  Set $l_i = (C_i-b_n)$ and $u_i = (C_i+b_n)$. 
        Apply an exhaustive search method for each time point $s$ in the search domain $(l_i, u_i)$ and estimate the change point $\widetilde{t}_i^f$ by minimizing the mean square error (MSE) of observations within the interval $[C_i - b_n, C_i + b_n )$ .

    \KwOutput{The estimated parameter matrix $\widehat{\theta}_j$ in each segment, $j = 1, 2, \ldots, \widehat{m} $ using the block fused lasso. The final estimated change points $ \widetilde{\mathcal{A}}_n^f = \left\lbrace \widetilde{t}_1^f, \ldots, \widetilde{t}_{m_0}^f  \right\rbrace $.  }
    \caption{\textbf{ Change Point Detection for the Piecewise Stationary SIR model}}
    \label{alg:piecewise constant}
\end{algorithm}

\section{Additional Simulation Results}\label{sec:sim}

We consider {five} additional simulation scenarios.
In the main paper, we only report the case where the  block size $b_n = 8$. {In this section,  we consider three different block size settings in the first three scenarios, where the block sizes are selected to be $b_n = 4, 8, 12$. In the last two  scenarios, the block sizes are selected to be $b_n = 8$.}

\begin{figure*}[ht!]
     \centering
     \begin{subfigure}[b]{0.32\textwidth}
         \centering
         \includegraphics[width=\textwidth]{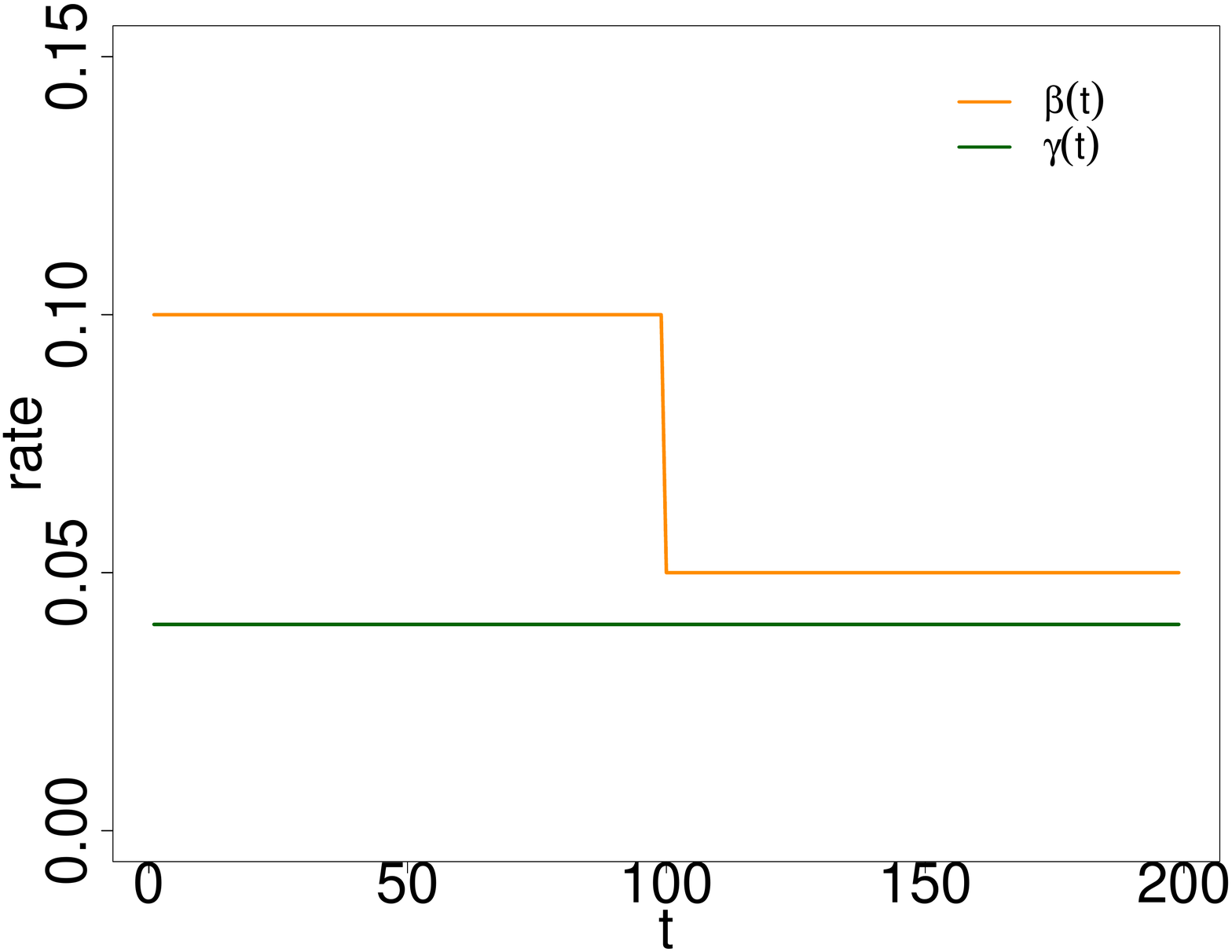}
         \subcaption{Scenario A, C, E (SIR  model rate)}
     \end{subfigure}
      \begin{subfigure}[b]{0.32\textwidth}
         \centering
         \includegraphics[width=\textwidth]{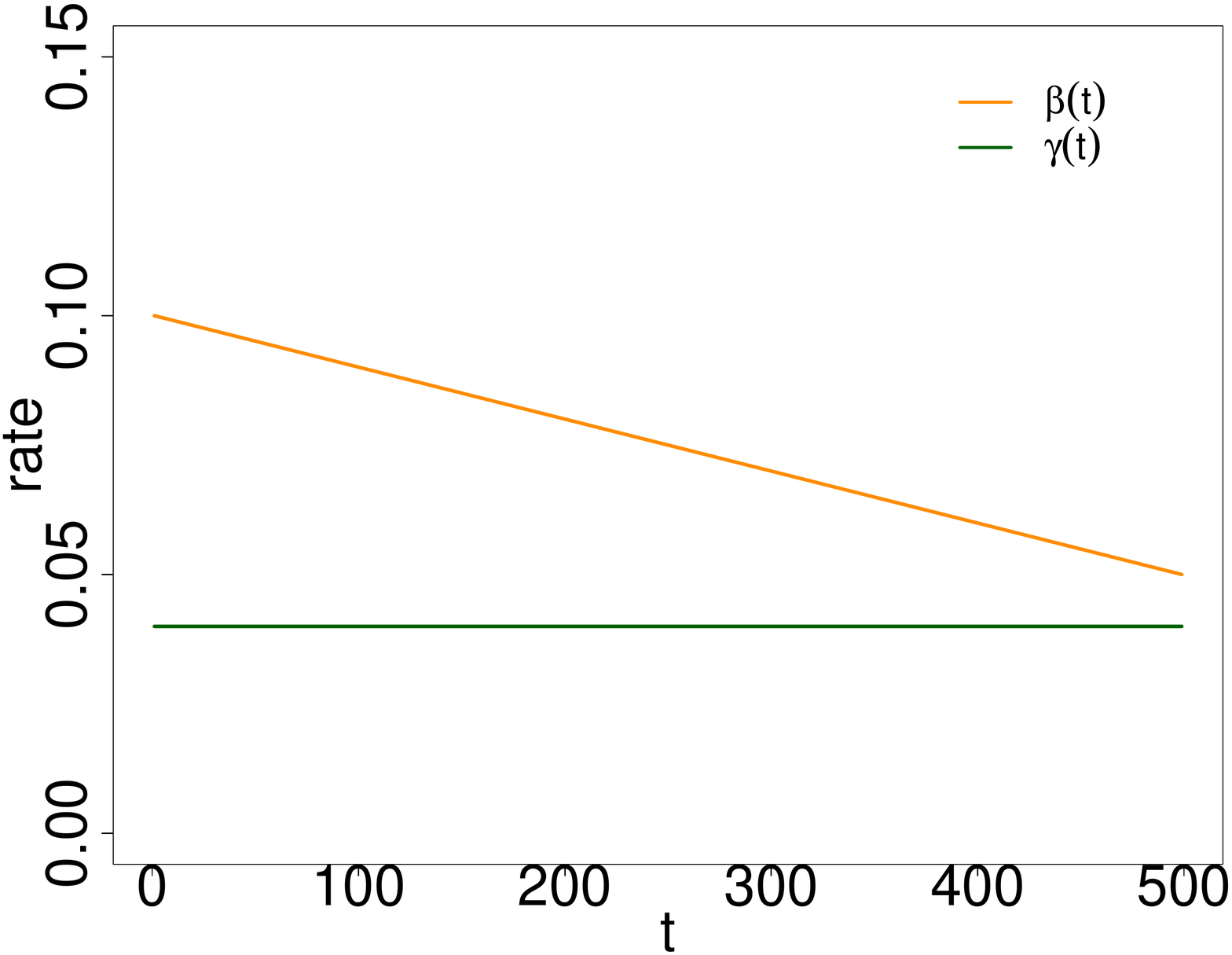}
         \subcaption{Scenario A, C, E (spatial SIR rate)}
     \end{subfigure}

    \begin{subfigure}[b]{0.32\textwidth}
         \centering
         \includegraphics[width=\textwidth]{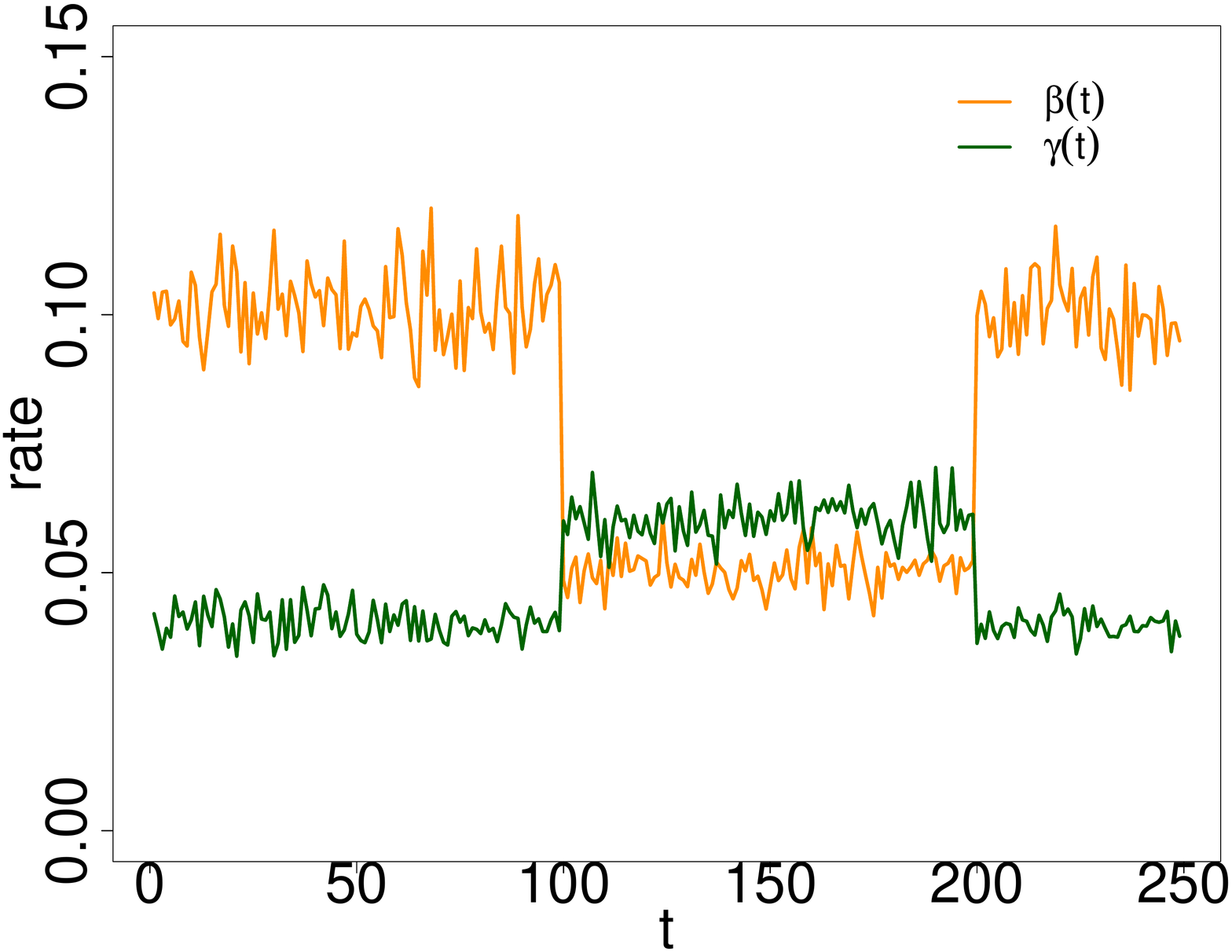}
         \subcaption{Scenario B, F (SIR  model rate)}
         \label{fig:sim5_rate}
     \end{subfigure}
     \begin{subfigure}[b]{0.32\textwidth}
         \centering
         \includegraphics[width=\textwidth]{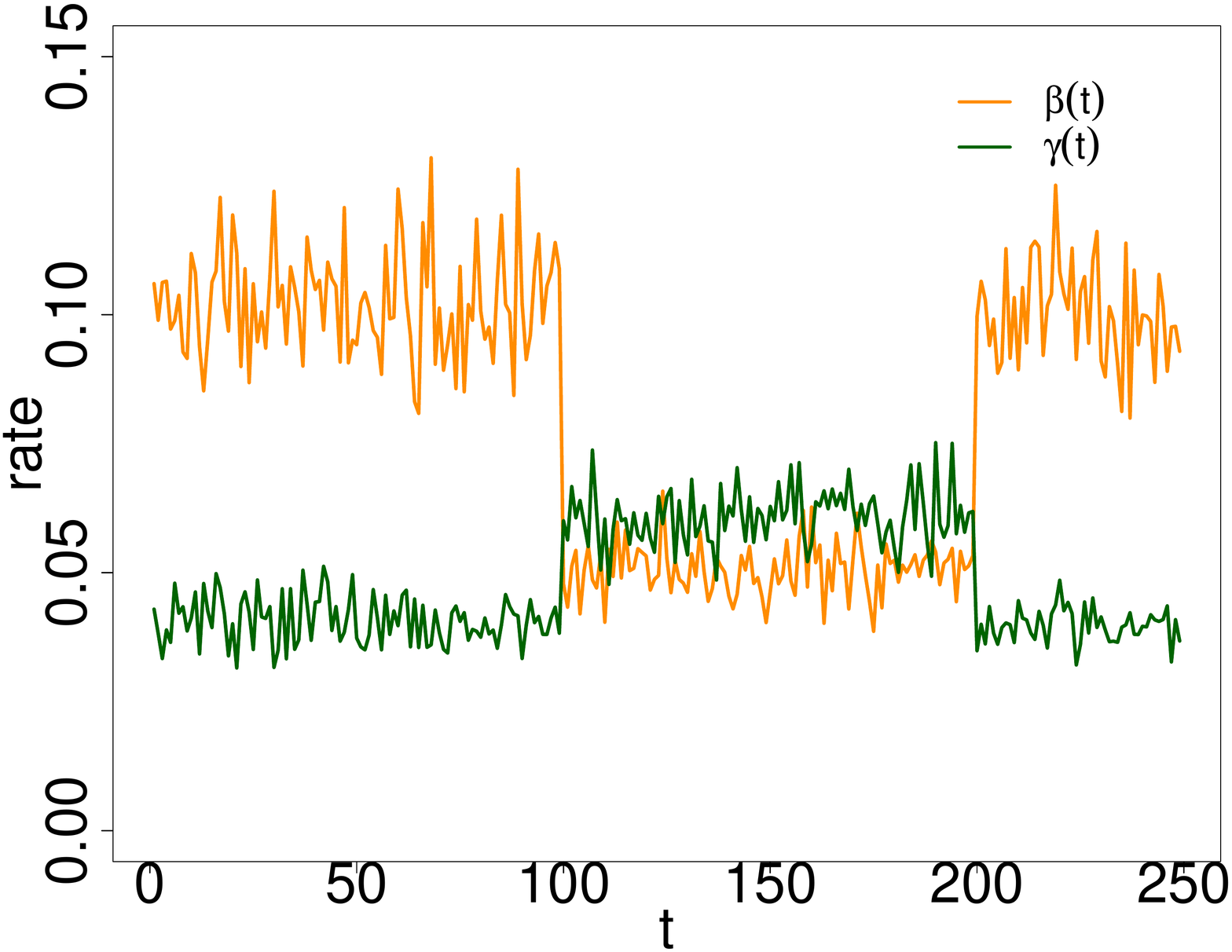}
         \subcaption{Scenario D (SIR model rate)}
         \label{fig:sim1}
     \end{subfigure}
     
     \begin{subfigure}[b]{0.32\textwidth}
         \centering
         \includegraphics[width=\textwidth]{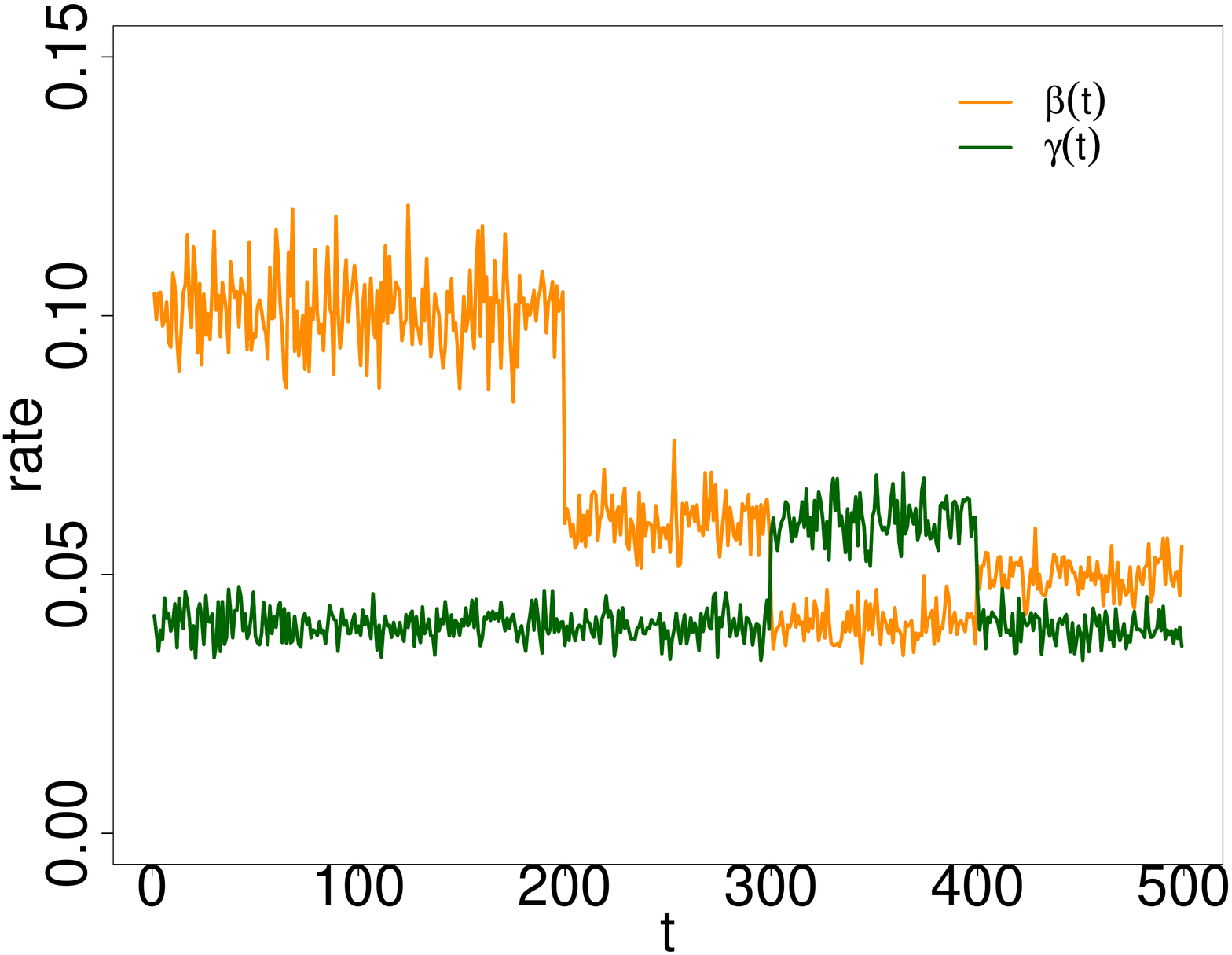}
         \subcaption{Scenario G (SIR  model rate)}
     \end{subfigure}
     \begin{subfigure}[b]{0.32\textwidth}
         \centering
         \includegraphics[width=\textwidth]{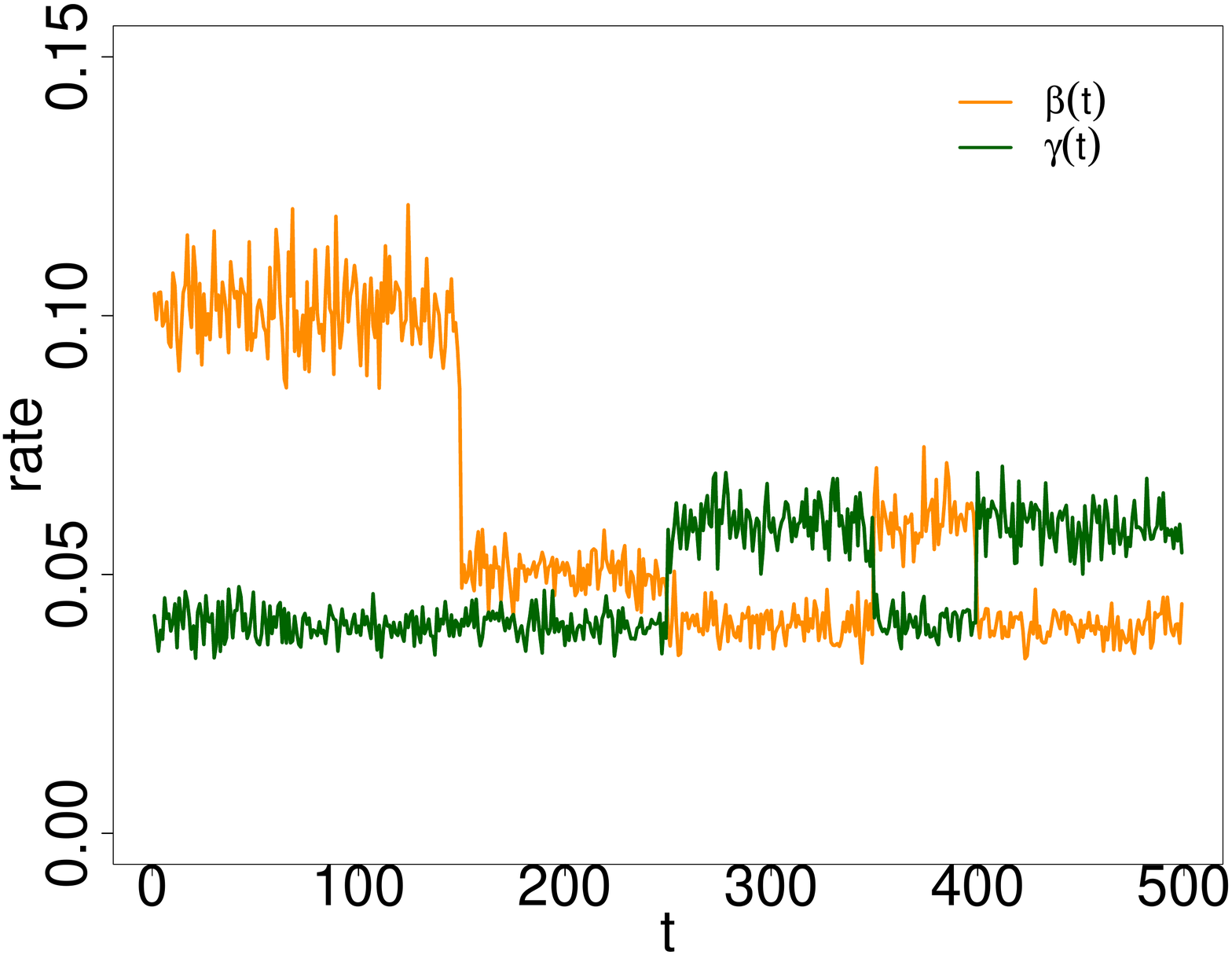}
         \subcaption{Scenario H (SIR  model rate)}
     \end{subfigure}
     
     \begin{subfigure}[b]{0.32\textwidth}
         \centering
         \includegraphics[width=\textwidth]{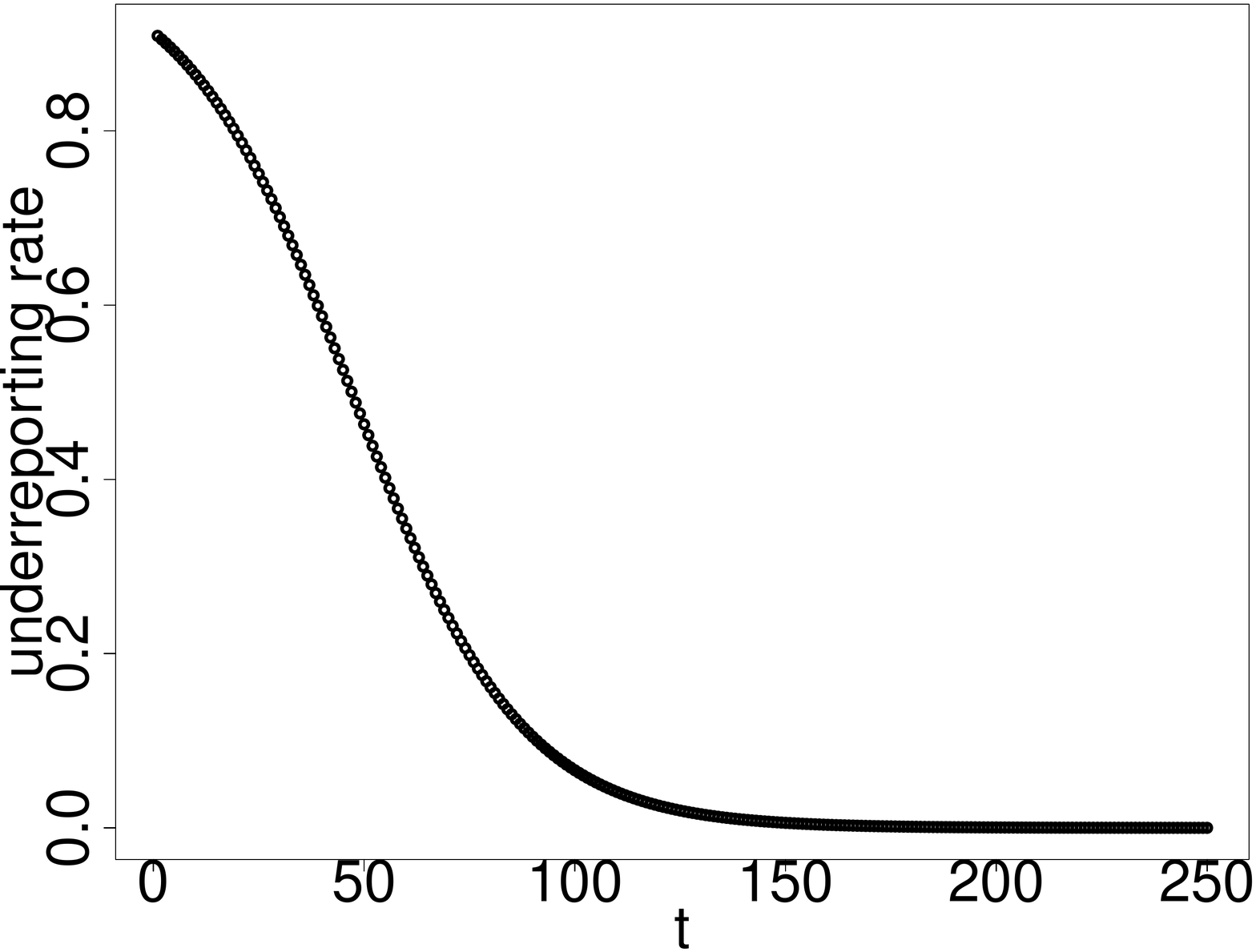}
         \subcaption{Scenario B (under-reporting rate)}
         \label{fig:sim7_1}
     \end{subfigure}
     \begin{subfigure}[b]{0.32\textwidth}
         \centering
         \includegraphics[width=\textwidth]{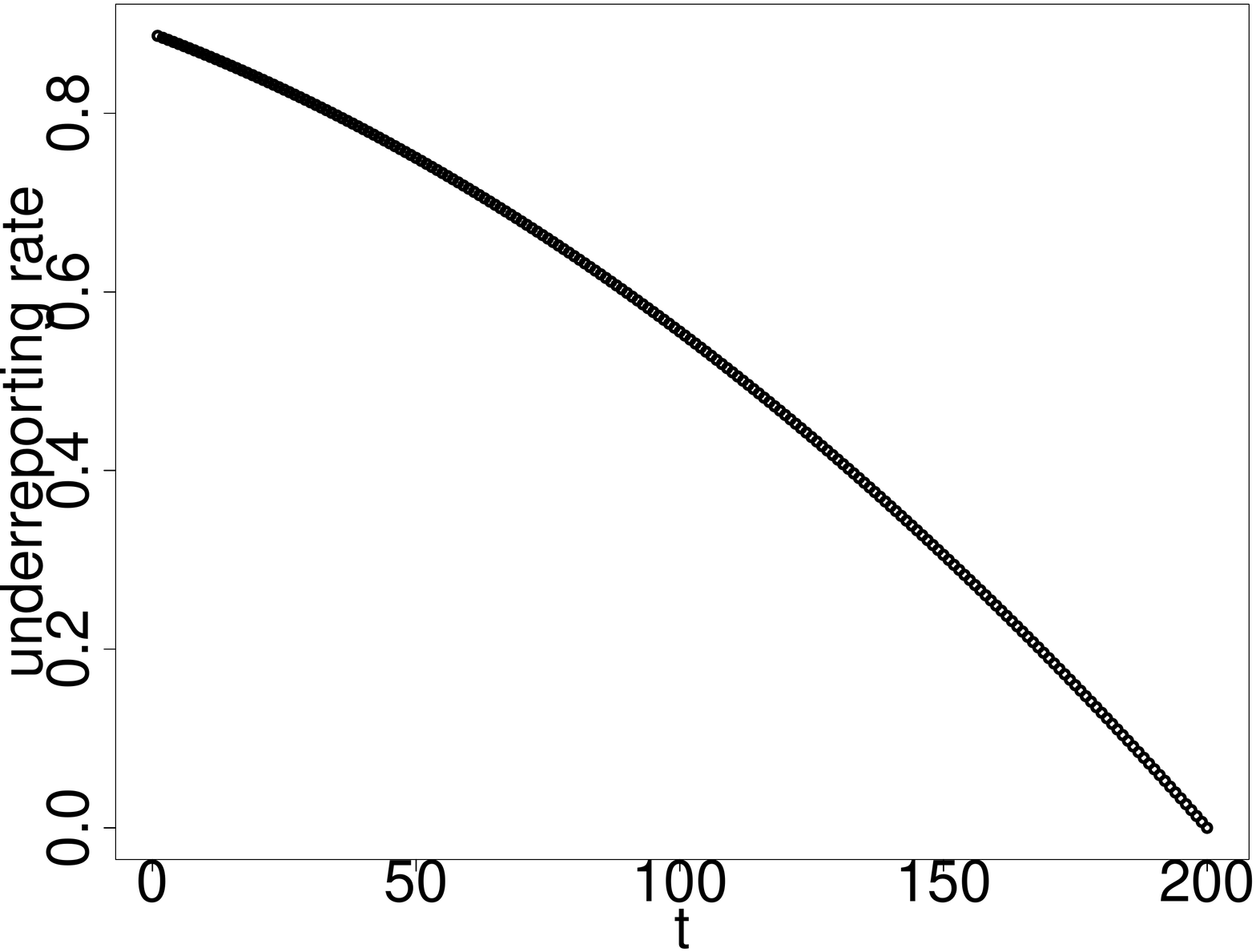}
         \subcaption{Scenario C (under-reporting rate)}
         \label{fig:sim8_1}
     \end{subfigure}
     \begin{subfigure}[b]{0.32\textwidth}
         \centering
         \includegraphics[width=\textwidth]{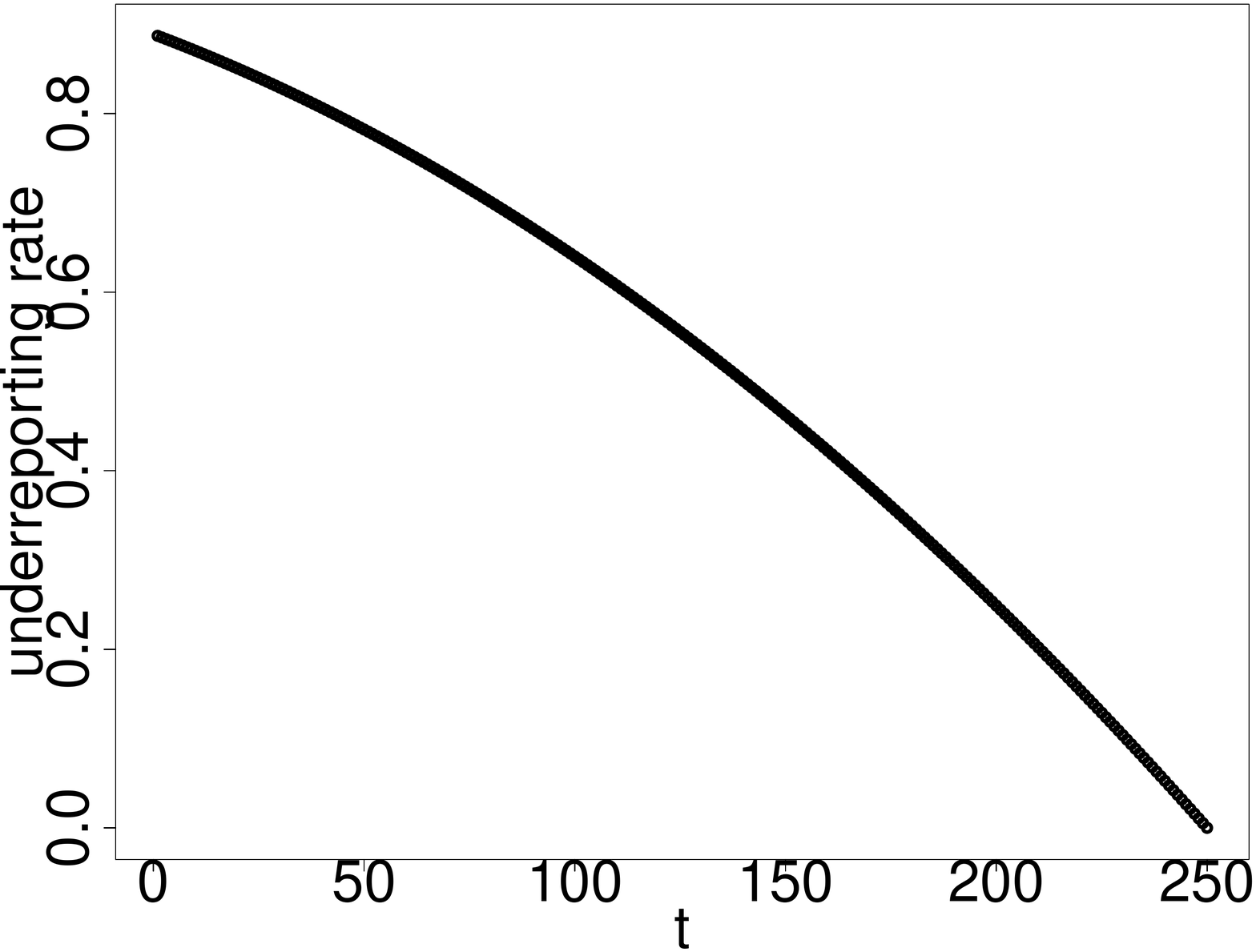}
         \subcaption{Scenario F (under-reporting rate)}
         \label{fig:sim6_1}
     \end{subfigure}
        \caption{True transmission rate, recovery rate and under-reporting rate in Scenario A-H.
        The top left plot provides the true transmission rate $\beta(t)$ and recovery rate $\gamma(t)$ in the SIR model in Scenario A, C, E; the top right plot provides the true transmission rate $\beta_s(t)$ and recovery rate $\gamma_s(t)$ for generating the spatial component in Scenario A, C, E. 
        the middle four plots provide the true transmission rate $\beta(t)$ and recovery rate $\gamma(t)$ in the SIR model in Scenario B, D, F, G, H; 
        The rest three plots are under-reporting rates in Scenario B, C, F, respectively.
        }
        \label{fig:sim_rate}
\end{figure*}

\noindent
\textit{Simulation Scenario D (Model 1):} the transmission and recovery rates are chosen to be piecewise constant. In this scenario, 
we set the number of time points $T=250$, $m_0 =2$, the change point $t_1 =100 $ and $t_2 = 200$. 
We choose $\beta^{(1)} = 0.10$,  $\beta^{(2)} = 0.05$,  $\beta^{(3)} = 0.10$,$\gamma^{(1)} = 0.04$, $\gamma^{(2)} = 0.06$, $\gamma^{(3)} = 0.04$.
Results are based on data generated from the SIR model in \eqref{eq:model_var} in main paper with 
$\beta(t) \sim  \text{Lognormal}(\sum_{j=1}^{m_0+1}\beta^{(j)}\mathbbm{1}_{\{t_{j-1}\leq t <t_j \}} , 0.01)$ and $\gamma(t) \sim  \text{Lognormal}(\sum_{j=1}^{m_0+1}\gamma^{(j)}\mathbbm{1}_{\{t_{j-1}\leq t <t_j \}} , 0.01)$.
The coefficients for the SIR model are depicted in Figure~\ref{fig:sim_rate}. We assume no under-reporting issue in this scenario, i.e., $u(t) = 1$, $\Delta I(t) = \Delta I_u(t)$, $t = 1, \dots, T-1$.

\noindent
\textit{Simulation Scenario E (Model 2):} the transmission  and recovery rates are chosen to be piecewise constant. In this scenario, we set the number of time points $T=200$, $m_0 =1$, the change point $t_1 = \lfloor \frac{T}{2}\rfloor =  100  $. We choose  $\alpha = 1$, $\beta^{(1)} = 0.10$,  $\beta^{(2)} = 0.05$, $\gamma^{(1)} = 0.04$, $\gamma^{(2)} = 0.04$,
$\beta_s(t)= 0.10  - \frac{0.05t}{T-1}$,
$\gamma_s(t)= 0.04$,
$t = 1, \dots, T-1$.
We generate  the spatial effect data from SIR model in \eqref{eq:model_var} in main paper. By plugging in the spatial effect data, we generate the response variable $Y_t$ with additional white noise error term from $\mathcal{N} (0, I_2)$.
The SIR model's coefficients for the response variable and the spatial effect variable are depicted in Figure~\ref{fig:sim_rate}.
We assume no under-reporting issue in this scenario, i.e., $u(t) = 1$, $\Delta I(t) = \Delta I_u(t)$, $t = 1, \dots, T-1$.

\textit{Simulation Scenario F (Model 1 with decreasing under-reporting rate, quadratic function):} the transmission and recovery rates are chosen to be piecewise constant. In this scenario, 
we set the number of time points $T= 250$, $m_0 =2$, the change point $t_1 =100 $ and $t_2 = 200$. 
We choose $\beta^{(1)} = 0.10$,  $\beta^{(2)} = 0.05$,  $\beta^{(3)} = 0.10$,
$\gamma^{(1)} = 0.04$, $\gamma^{(2)} = 0.06$, $\gamma^{(3)} = 0.04$.
Results are based on data generated from the SIR model in \eqref{two_eq_1} with 
$\beta(t) \sim  \text{Lognormal}(\sum_{j=1}^{m_0+1}\beta^{(j)}\mathbbm{1}_{\{t_{j-1}\leq t <t_j \}} , 0.005)$ and $\gamma(t) \sim  \text{Lognormal}(\sum_{j=1}^{m_0+1}\gamma^{(j)}\mathbbm{1}_{\{t_{j-1}\leq t <t_j \}} , 0.005)$.
The true coefficients for the SIR model are depicted in Figure~\ref{fig:sim5_rate}. The under-reporting rate is chosen to change over time, as shown in Figure~\ref{fig:sim6_1}. In this  scenario, the under-reporting rate changes over time. 
Specifically, we set under-reporting rate $u(t) =  1- \left(\frac{t + aT}{(1+ a)T}\right)^2$, $t = 1, \dots, T$, $a = 0.5$. The grid search values for $a$ are $0.1, 0.25, 0.5, 0.75, 1$. 

\noindent
\textit{Simulation Scenario G (Model 1):} the transmission and recovery rates are chosen to be piecewise constant. In this scenario, 
we set the number of time points $T=500$, $m_0 =3$, the change point $t_1 =200 $, $t_2 =300 $ and $t_3 = 400$. 
We choose $\beta^{(1)} = 0.10$,  $\beta^{(2)} = 0.06$,  $\beta^{(3)} = 0.04$,  $\beta^{(4)} = 0.05$, $\gamma^{(1)} = 0.04$,  $\gamma^{(2)} = 0.04$,  $\gamma^{(3)} = 0.06$, $\gamma^{(4)} = 0.04$.
Results are based on data generated from the SIR model in \eqref{eq:model_var} in main paper with 
$\beta(t) \sim  \text{Lognormal}(\sum_{j=1}^{m_0+1}\beta^{(j)}\mathbbm{1}_{\{t_{j-1}\leq t <t_j \}} , 0.005)$ and $\gamma(t) \sim  \text{Lognormal}(\sum_{j=1}^{m_0+1}\gamma^{(j)}\mathbbm{1}_{\{t_{j-1}\leq t <t_j \}} , 0.005)$.
The coefficients for the SIR model are depicted in Figure~\ref{fig:sim_rate}. We assume no under-reporting issue in this scenario, i.e., $u(t) = 1$, $\Delta I(t) = \Delta I_u(t)$, $t = 1, \dots, T-1$.

\noindent
\textit{Simulation Scenario H (Model 1):} the transmission and recovery rates are chosen to be piecewise constant. In this scenario, 
we set the number of time points $T=500$, $m_0 =4$, the change point $t_1 =150$, $t_2 =250 $,$t_3 = 350$ and $t_3 = 400$. 
We choose $\beta^{(1)} = 0.10$,  $\beta^{(2)} = 0.05$,  $\beta^{(3)} = 0.04$,  $\beta^{(4)} = 0.06$, $\beta^{(5)} = 0.04$, 
$\gamma^{(1)} = 0.04$,  $\gamma^{(2)} = 0.04$,  $\gamma^{(3)} = 0.06$, $\gamma^{(4)} = 0.04$, $\gamma^{(5)} = 0.06$.
Results are based on data generated from the SIR model in \eqref{eq:model_var} in main paper with 
$\beta(t) \sim  \text{Lognormal}(\sum_{j=1}^{m_0+1}\beta^{(j)}\mathbbm{1}_{\{t_{j-1}\leq t <t_j \}} , 0.005)$ and $\gamma(t) \sim  \text{Lognormal}(\sum_{j=1}^{m_0+1}\gamma^{(j)}\mathbbm{1}_{\{t_{j-1}\leq t <t_j \}} , 0.005)$.
The coefficients for the SIR model are depicted in Figure~\ref{fig:sim_rate}. We assume no under-reporting issue in this scenario, i.e., $u(t) = 1$, $\Delta I(t) = \Delta I_u(t)$, $t = 1, \dots, T-1$.


\begin{table*}[!htbp]
\caption{\label{table_sim_selection} Results of the mean and standard deviation of estimated change point location and the selection rate under the piecewise constant setting.}
\centering
\footnotesize
\begin{tabular}{ccccccccccccc} 
  \hline
  \hline
 & Model & change point & truth   & mean  & std & selection rate \\
   \hline
Scenario D &\\
&Model 1 ($b_n = 4$) & 1 & 0.4 & 0.4012 & 0.0096 & 0.98 \\ 
  &Model 1 ($b_n = 4$) & 2 & 0.8 & 0.8 & 0 & 0.97 \\ 
  &Model 1 ($b_n = 8$) & 1 & 0.4 & 0.4 & 4e-04 & 0.99 \\ 
  &Model 1 ($b_n = 8$) & 2 & 0.8 & 0.8003 & 0.0032 & 0.99 \\ 
  &Model 1 ($b_n = 12$) & 1 & 0.4 & 0.4 & 0 & 1 \\ 
  &Model 1 ($b_n = 12$) & 2 & 0.8 & 0.7997 & 0.0032 & 1 \\ 
  Scenario E &\\
&Model 1 ($b_n = 4$) & 1 & 0.5 & 0.4988 & 0.0039 & 1 \\ 
  &Model 1 ($b_n = 8$) & 1 & 0.5 & 0.5 & 0 & 1 \\ 
  &Model 1 ($b_n = 12$) & 1 & 0.5 & 0.5 & 0 & 1 \\
 Scenario F &\\    
  &Model 1 ($b_n = 4$) & 1 & 0.4 & 0.4 & 0 & 1 \\ 
  &Model 1 ($b_n = 4$) & 2 & 0.8 & 0.8005 & 0.0048 & 1 \\ 
  &Model 1 ($b_n = 8$) & 1 & 0.4 & 0.4 & 0 & 0.96 \\ 
  &Model 1 ($b_n = 8$) & 2 & 0.8 & 0.8011 & 0.0056 & 1 \\ 
  &Model 1 ($b_n = 12$) & 1 & 0.4 & 0.4 & 0 & 1 \\ 
  &Model 1 ($b_n = 12$) & 2 & 0.8 & 0.7994 & 0.0045 & 1 \\
 Scenario G &  \\ 
  &Model 1 ($b_n = 8$) & 1 & 0.4 & 0.4 & 6e-04 & 0.99 \\ 
  &Model 1 ($b_n = 8$) & 2 & 0.6 & 0.6 & 0 & 0.97 \\ 
  &Model 1 ($b_n = 8$) & 3 & 0.8 & 0.8 & 0 & 1 \\ 
  Scenario H &\\  
  &Model 1 ($b_n = 8$) & 1 & 0.3 & 0.3 & 0 & 1 \\ 
  &Model 1 ($b_n = 8$) & 2 & 0.5 & 0.5 & 0 & 0.98 \\ 
  &Model 1 ($b_n = 8$) & 3 & 0.7 & 0.7 & 0 & 0.97 \\ 
  &Model 1 ($b_n = 8$) & 4 & 0.8 & 0.7999 & 0.0014 & 0.98 \\
   \hline
\end{tabular}
\end{table*}

\begin{table}[!htbp]
\caption{\label{table_sim_1} Results of the mean and standard deviation of estimated parameters $\widehat{\beta}$, $\widehat{\gamma}$ in simulation scenario D. }
\centering
\footnotesize
\begin{tabular}{cccccccccc} 
  \hline
  \hline
  Model  & parameter  &true value & mean & std  \\
   \hline
 \multirow{ 6}{*}{Model 1 ($b_n = 4$)} & $\beta_1$ & 0.1 & 0.1005 & 0.0026 \\ 
   & $\beta_2$ & 0.05 & 0.0513 & 0.0068 \\ 
   & $\beta_3$ & 0.1 & 0.0915 & 0.0194 \\ 
   & $\gamma_1$ & 0.04 & 0.0401 & 9e-04 \\ 
   & $\gamma_2$ & 0.06 & 0.0599 & 0.0028 \\ 
   & $\gamma_3$ & 0.04 & 0.0438 & 0.0077 \\ 
  \multirow{ 6}{*}{Model 1 ($b_n = 8$)} & $\beta_1$ & 0.1 & 0.1002 & 0.005 \\ 
   & $\beta_2$ & 0.05 & 0.0513 & 0.0072 \\ 
   & $\beta_3$ & 0.1 & 0.0994 & 0.0076 \\ 
   & $\gamma_1$ & 0.04 & 0.0403 & 0.0019 \\ 
   & $\gamma_2$ & 0.06 & 0.06 & 0.0028 \\ 
   & $\gamma_3$ & 0.04 & 0.0407 & 0.003 \\ 
  \multirow{ 6}{*}{Model 1 ($b_n = 12$)} & $\beta_1$ & 0.1 & 0.1005 & 0.0027 \\ 
   & $\beta_2$ & 0.05 & 0.0503 & 7e-04 \\ 
   & $\beta_3$ & 0.1 & 0.1004 & 0.0024 \\ 
   & $\gamma_1$ & 0.04 & 0.0401 & 9e-04 \\ 
   & $\gamma_2$ & 0.06 & 0.0604 & 7e-04 \\ 
   & $\gamma_3$ & 0.04 & 0.0403 & 8e-04 \\  
   \hline
\end{tabular}
\end{table}

 \begin{table}[!htbp]
\caption{\label{table_sim_alpha} Results of the mean and standard deviation of estimated parameters $\widehat{\beta}$, $\widehat{\gamma}$, $\widehat{\alpha}$ in simulation scenario E. }
\centering
\footnotesize
\begin{tabular}{cccccccccc} 
  \hline
  \hline
  Model  & parameter  &true value & mean & std  \\
   \hline
 \multirow{ 5}{*}{Model 2 ($b_n = 4$)} & $\beta_1$ & 0.1& 0.1 & 2e-04 \\  
     & $\beta_2$ & 0.05& 0.0501 & 3e-04 \\ 
     & $\gamma_1$& 0.04 & 0.04 & 1e-04  \\ 
    & $\gamma_2$ & 0.04& 0.04 & 1e-04  \\ 
     & $\alpha$ &1 & 0.9881 & 0.0487  \\  
    \multirow{ 5}{*}{Model 2 ($b_n = 8$)} & $\beta_1$& 0.1 & 0.1 & 1e-04  \\  
     & $\beta_2$& 0.05 & 0.05 & 1e-04 \\  
     & $\gamma_1$ & 0.04& 0.04 & 1e-04  \\ 
     & $\gamma_2$ & 0.04& 0.04 & 1e-04   \\  
     & $\alpha$ &1 & 1.0012 & 0.0236  \\  
    \multirow{ 5}{*}{Model 2 ($b_n = 12$)} & $\beta_1$ & 0.1& 0.1 & 1e-04  \\  
     & $\beta_2$ & 0.05& 0.05 & 1e-04 \\ 
     & $\gamma_1$& 0.04& 0.04 & 1e-04  \\ 
     & $\gamma_2$ & 0.04& 0.04 & 1e-04 \\  
     & $\alpha$ &1 & 1.0012 & 0.0239  \\ 
   \hline
\end{tabular}
\end{table}

\begin{table}[!htbp]
\caption{\label{table_sim_underreport} Results of the mean and standard deviation of estimated parameters $\widehat{\beta}$, $\widehat{\gamma}$, $\widehat{a}$ in simulation scenario F. }
\centering
\footnotesize
\begin{tabular}{cccccccccc} 
  \hline
  \hline
  Model  & parameter  &true value & mean & std  \\
   \hline
 \multirow{ 7}{*}{Model 1 ($b_n = 4$)} & $\beta_1$ & 0.1 & 0.1004 & 0.0109  \\ 
   & $\beta_2$ & 0.05 & 0.0517 & 0.0127  \\ 
   & $\beta_3$ & 0.1 & 0.0985 & 0.0145  \\ 
   & $\gamma_1$ & 0.04 & 0.0405 & 0.0082  \\ 
   & $\gamma_2$ & 0.06 & 0.061 & 0.0132 \\ 
   & $\gamma_3$ & 0.04 & 0.0405 & 0.0073  \\ 
   & $a$ & 0.5 & 0.5235 & 0.1767  \\ 
  \multirow{ 7}{*}{Model 1 ($b_n = 8$)} & $\beta_1$ & 0.1 & 0.0998 & 0.0103  \\ 
   & $\beta_2$ & 0.05 & 0.0532 & 0.0155 \\ 
   & $\beta_3$ & 0.1 & 0.0997 & 0.0059 \\ 
   & $\gamma_1$ & 0.04 & 0.0422 & 0.0117  \\ 
   & $\gamma_2$ & 0.06 & 0.0602 & 0.0124  \\ 
   & $\gamma_3$ & 0.04 & 0.0399 & 0.0024 \\ 
   & $a$ & 0.5 & 0.5375 & 0.1873  \\ 
  \multirow{ 7}{*}{Model 1 ($b_n = 12$)} & $\beta_1$ & 0.1 & 0.1011 & 0.0101 \\ 
   & $\beta_2$ & 0.05 & 0.0518 & 0.0114 \\ 
   & $\beta_3$ & 0.1 & 0.1001 & 0.0055 \\ 
   & $\gamma_1$ & 0.04 & 0.041 & 0.0078  \\ 
   & $\gamma_2$ & 0.06 & 0.062 & 0.0126\\ 
   & $\gamma_3$ & 0.04 & 0.0401 & 0.0023 \\ 
   & $a$ & 0.5 & 0.534 & 0.1784  \\ 
   \hline
\end{tabular}
\end{table}

\begin{table}[!htbp]
\caption{\label{table_sim_G} Results of the mean and standard deviation of estimated parameters $\widehat{\beta}$, $\widehat{\gamma}$ in simulation scenario G. }
\centering
\footnotesize
{
\begin{tabular}{cccccccccc} 
  \hline
  \hline
  Model  & parameter  &true value & mean & std  \\
   \hline
\multirow{ 6}{*}{Model 1 ($b_n = 8$)} & $\beta_1$ & 0.1 & 0.0997 & 0.0046 \\ 
   & $\beta_2$ & 0.06 & 0.0596 & 0.0027 \\ 
   & $\beta_3$ & 0.04 & 0.0405 & 0.002 \\ 
   & $\beta_4$ & 0.05 & 0.0501 & 0.0011 \\ 
   & $\gamma_1$ & 0.04 & 0.0401 & 7e-04 \\ 
   & $\gamma_2$ & 0.04 & 0.0406 & 0.0026 \\ 
   & $\gamma_3$ & 0.06 & 0.0595 & 0.004 \\ 
   & $\gamma_4$ & 0.04 & 0.0403 & 0.0022 \\ 
   \hline
\end{tabular}}
\end{table}

\begin{table}[!htbp]
\caption{\label{table_sim_H} Results of the mean and standard deviation of estimated parameters $\widehat{\beta}$, $\widehat{\gamma}$ in simulation scenario H. }
\centering
\footnotesize
{
\begin{tabular}{cccccccccc} 
  \hline
  \hline
  Model  & parameter  &true value & mean & std  \\
   \hline
\multirow{ 8}{*}{Model 1 ($b_n = 8$)} & $\beta_1$ & 0.1 & 0.1002 & 0.002 \\ 
   & $\beta_2$ & 0.05 & 0.05 & 6e-04 \\ 
   & $\beta_3$ & 0.04 & 0.0403 & 0.0013 \\ 
   & $\beta_4$ & 0.06 & 0.06 & 0.0027 \\ 
   & $\beta_5$ & 0.04 & 0.0402 & 0.0016 \\ 
   & $\gamma_1$ & 0.04 & 0.0399 & 7e-04 \\ 
   & $\gamma_2$ & 0.04 & 0.0403 & 0.0011 \\ 
   & $\gamma_3$ & 0.06 & 0.0597 & 0.0022 \\ 
   & $\gamma_4$ & 0.04 & 0.0404 & 0.0026 \\ 
   & $\gamma_5$ & 0.06 & 0.0599 & 0.002 \\ 
   \hline
\end{tabular}}
\end{table}

The mean and standard deviation of the  location of the selected change point, relative to the the number of time points $T$ -- i.e., $\widetilde{t}^f_1/T$ -- for all simulation scenarios are summarized in Table~\ref{table_sim_selection}.  
The results clearly indicate that, in the piecewise constant setting, our procedure accurately detects the location of change points. 
It also indicates  the proposed algorithm's robustness with respect to block size selection. 
The results of the estimated transmission rate $\widehat{\beta}$, recovery rate $\widehat{\gamma}$ and spatial effect $\widehat{\alpha}$ in Table~\ref{table_sim_1} and Table~\ref{table_sim_alpha} also suggest that our procedure produces accurate estimates of the parameters, under the various settings considered.   
The results of the estimated transmission rate $\widehat{\beta}$, recovery rate $\widehat{\gamma}$ and  parameter of the  under-reporting rate function $\widehat{a}$ in Table~\ref{table_sim_underreport} also suggest that our procedure produces accurate estimates of the parameters, under the various under-reporting function settings.

The estimated transition matrices of the VAR component are close to the true values, regardless of block size selection. For example, for $b_n=4$, the entries of estimated autoregressive coefficients matrix which are averaged among 100 replicates are 0.7577, -0.0057, 0.2361, 0.6663, with standard deviations 0.1125, 0.0461, 0.3582, 0.1147, respectively, while the entries of true autoregressive coefficients matrices are 0.8, 0, 0.2, 0.7 from top left to bottom right. Moreover, for $b_n=8$, the entries of estimated autoregressive coefficients matrix are 0.7652 -0.0057, 0.2088,  0.6767, with standard deviations 0.0533, 0.0464, 0.0459, 0.0573, respectively. These results confirm the good performance of the proposed algorithm in finite sample while also indicates its robustness with respect to block size selection.

\section{Additional Results for U.S. States}\label{sec:states}

The block size is set to $b_n = 7$, i.e. partition the observations into blocks of 7 days for all regions. 
Finally, the tuning parameter $\lambda_n$ is selected via a cross-validated grid search \cite{hastie2009elements}.

\begin{figure*}[!ht]
     \centering
     \begin{subfigure}[b]{0.19\textwidth}
         \centering
         \includegraphics[width=\textwidth]{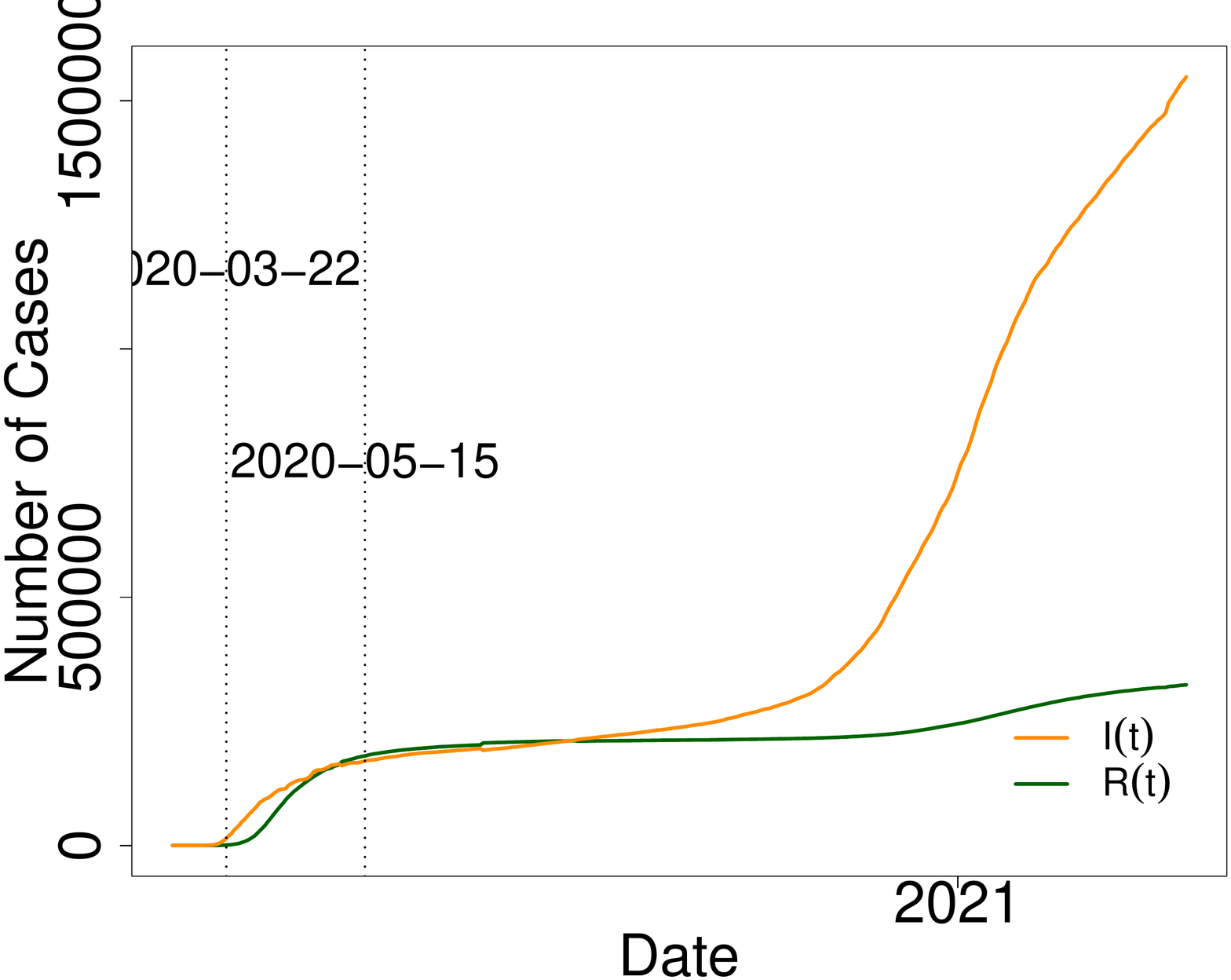}
         \subcaption{New York}
     \end{subfigure}
     \begin{subfigure}[b]{0.19\textwidth}
         \centering
         \includegraphics[width=\textwidth]{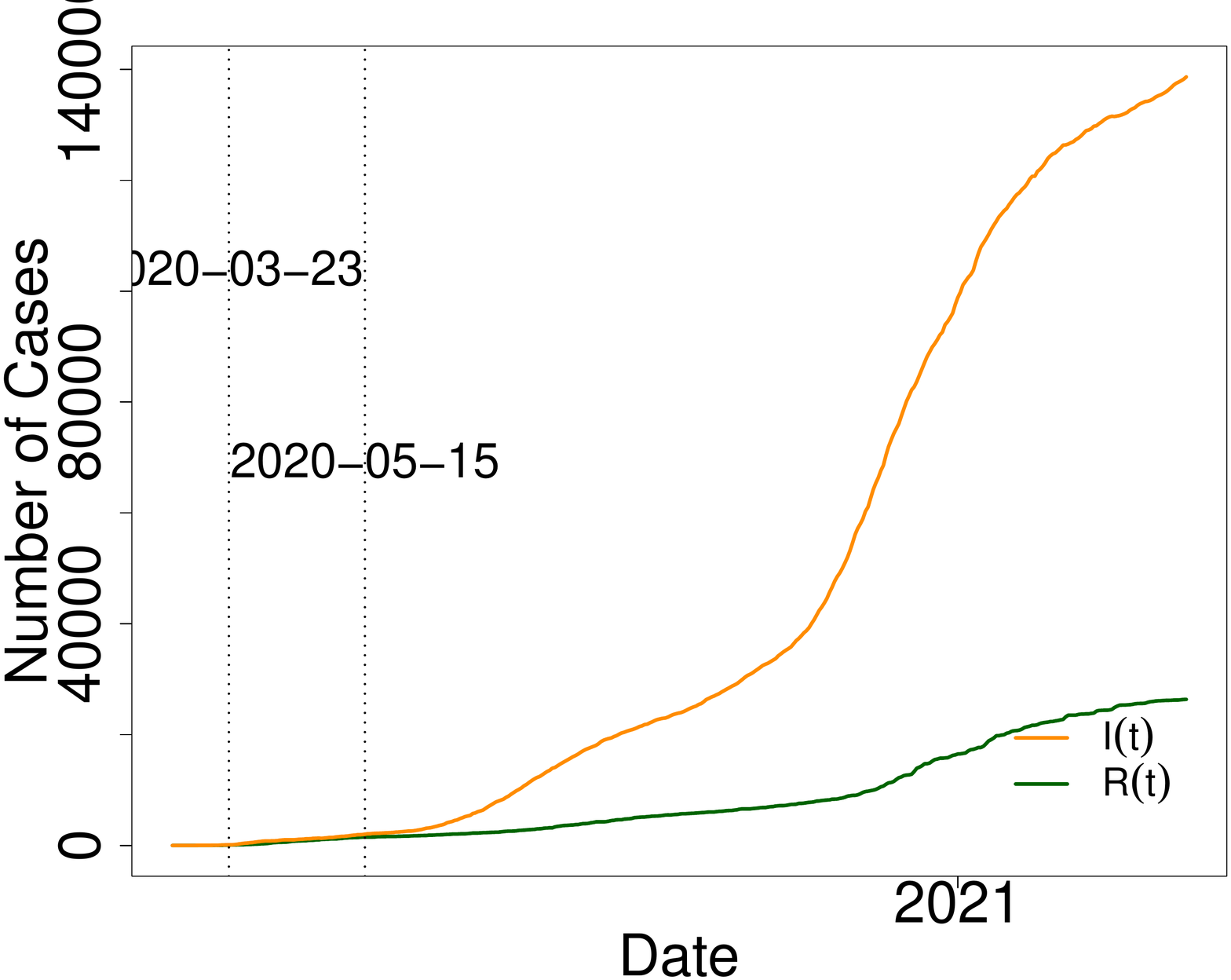}
         \subcaption{Oregon}
     \end{subfigure}
     \begin{subfigure}[b]{0.19\textwidth}
         \centering
         \includegraphics[width=\textwidth]{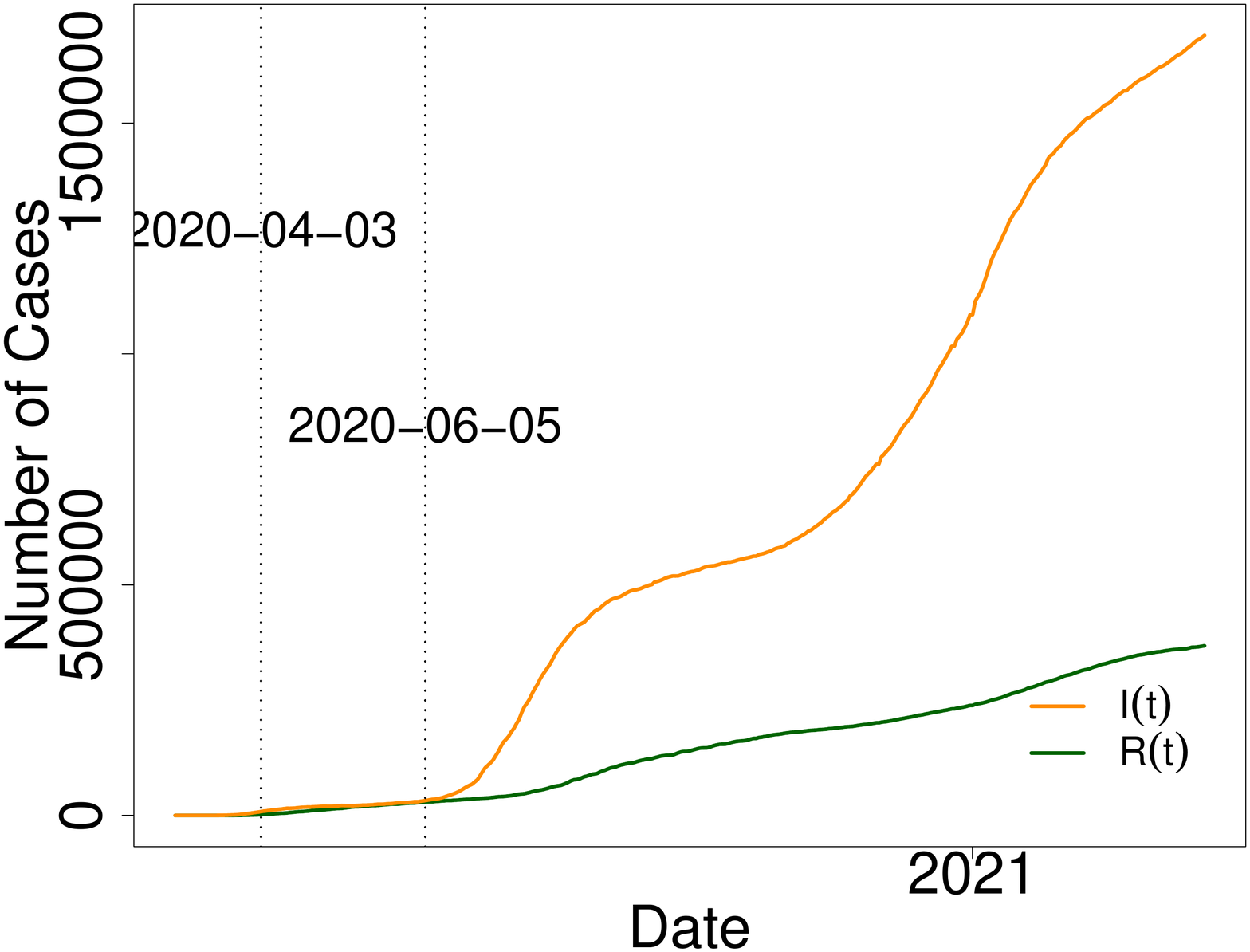}
         \subcaption{Florida}
     \end{subfigure}
     \begin{subfigure}[b]{0.19\textwidth}
         \centering
         \includegraphics[width=\textwidth]{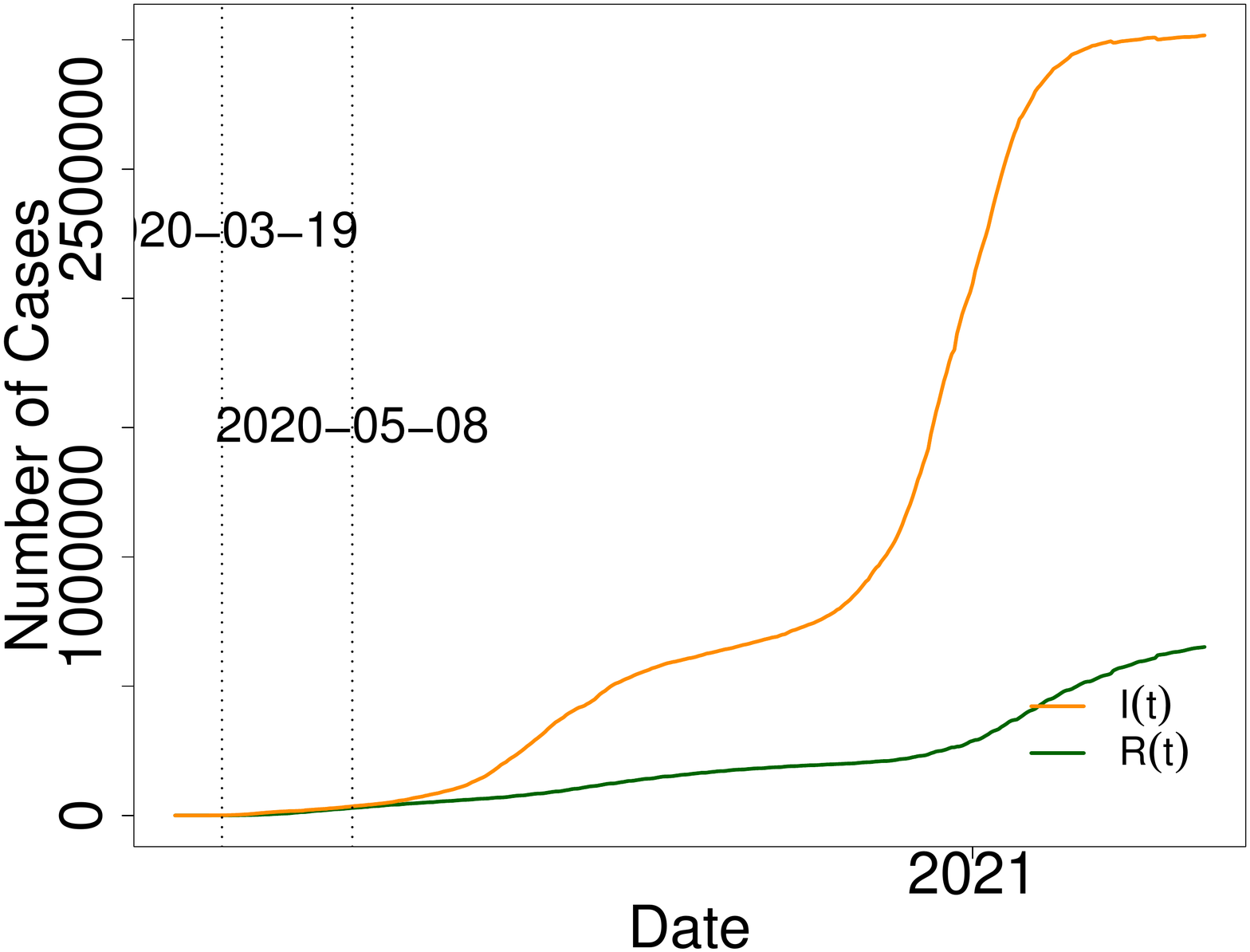}
         \subcaption{California}
     \end{subfigure}
     \begin{subfigure}[b]{0.19\textwidth}
         \centering
         \includegraphics[width=\textwidth]{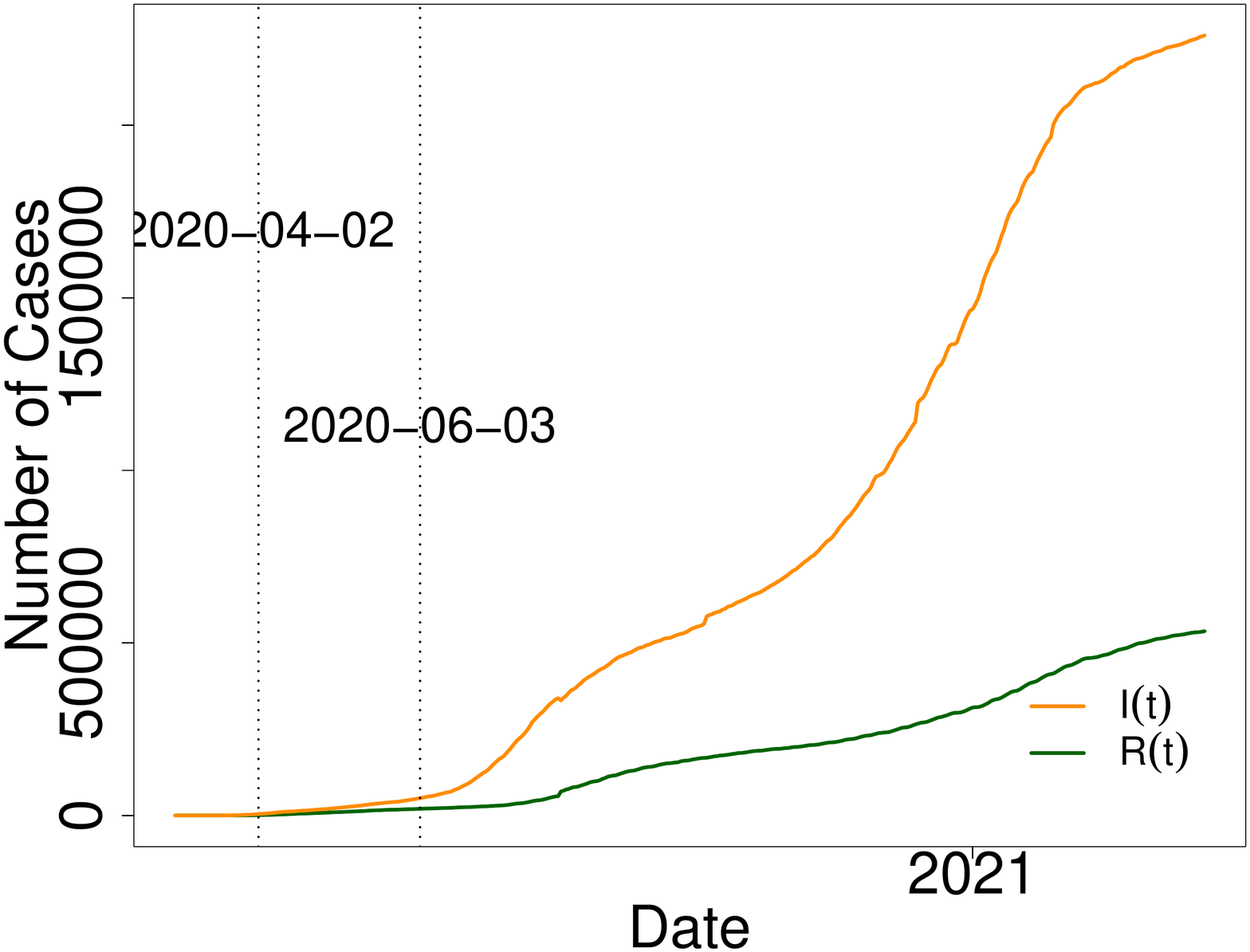}
         \subcaption{Texas}
     \end{subfigure}
        \caption{
        Number of infected cases (dark orange) and recovered cases (dark green) in several states. The first vertical black dotted line indicates the ``stay-at-home'' order start date  for each state, while the second vertical black dotted line indicates the ``reopening'' start date.}
        \label{fig:numbers_states}
\end{figure*}

\begin{table}[!ht]
\caption{\label{table_adj}Neighboring states by distance (for models 2.1 and 2.2). }
\centering
{
\begin{tabular}{l l}
  \hline
  \hline
Region &  Neighboring Regions  (states: within 500 miles)   \\
  \hline
  New York  &  Connecticut, New Hampshire, New Jersey, Pennsylvania, Vermont\\
   Oregon & California, Idaho, Nevada, Washington\\
   Florida & Alabama, Georgia, South Carolina \\
   California   &  Arizona, Nevada, Oregon, Utah\\
   Texas &  Arkansas, Kansas, Louisiana, New Mexico, Oklahoma\\
 \hline
\end{tabular}}
\end{table}

\begin{figure*}[ht!]
     \centering
      \captionsetup[sub]{font=small, labelfont={bf,sf}}
     \begin{subfigure}[b]{0.19\textwidth}
         \centering
         \includegraphics[width=\textwidth]{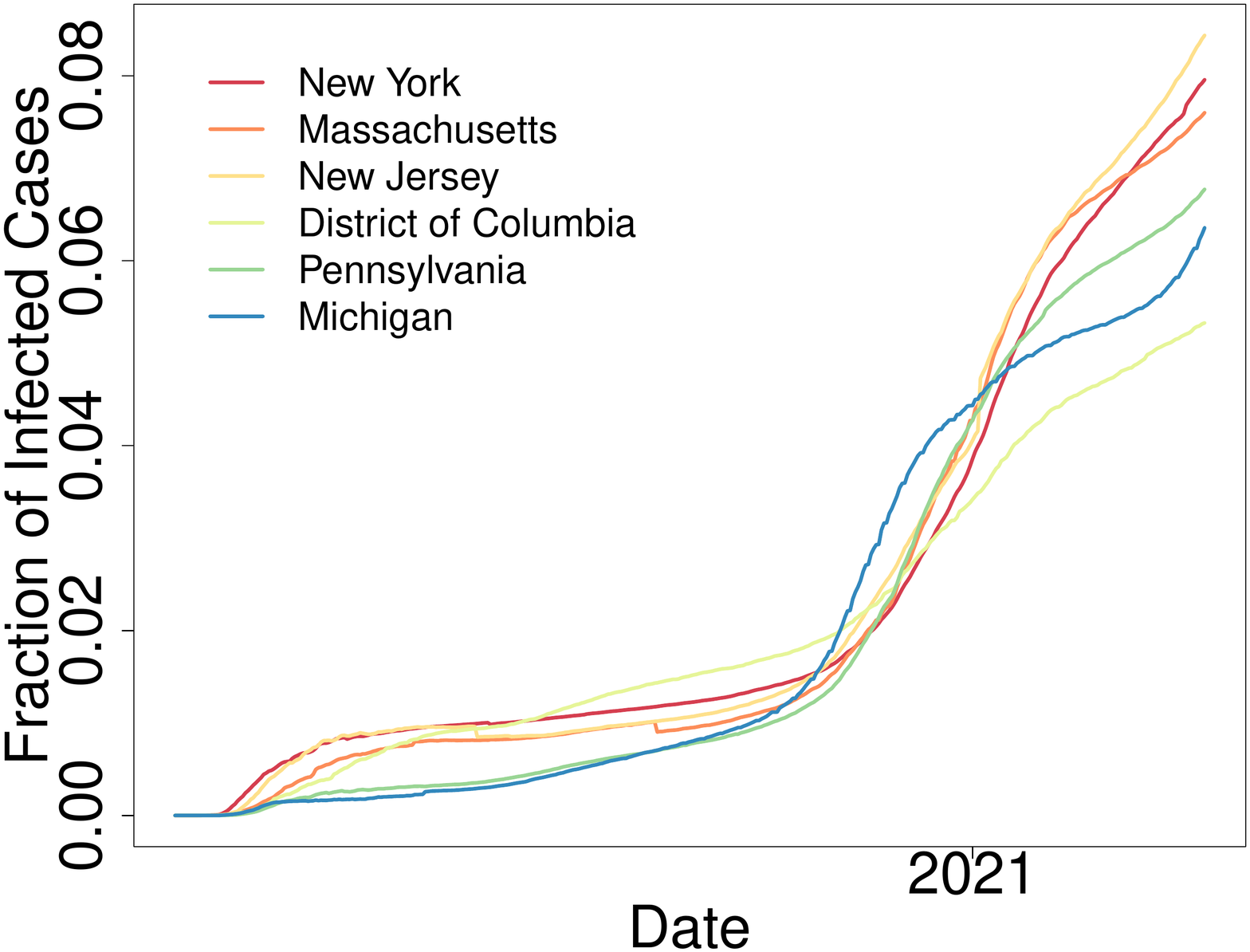}
         \subcaption{New York}
     \end{subfigure}
    \begin{subfigure}[b]{0.19\textwidth}
         \centering
         \includegraphics[width=\textwidth]{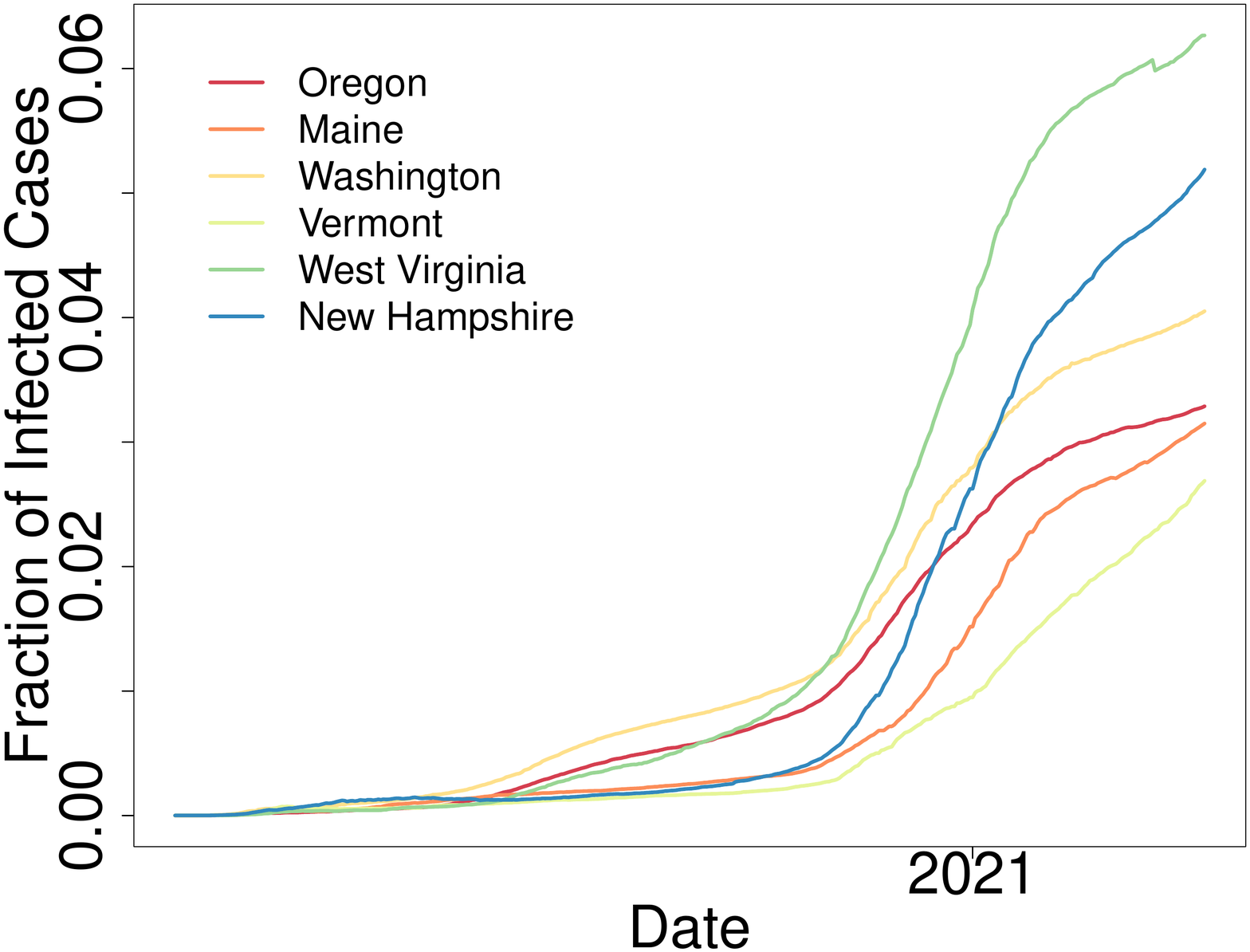}
         \subcaption{Oregon}
     \end{subfigure}
    \begin{subfigure}[b]{0.19\textwidth}
         \centering
         \includegraphics[width=\textwidth]{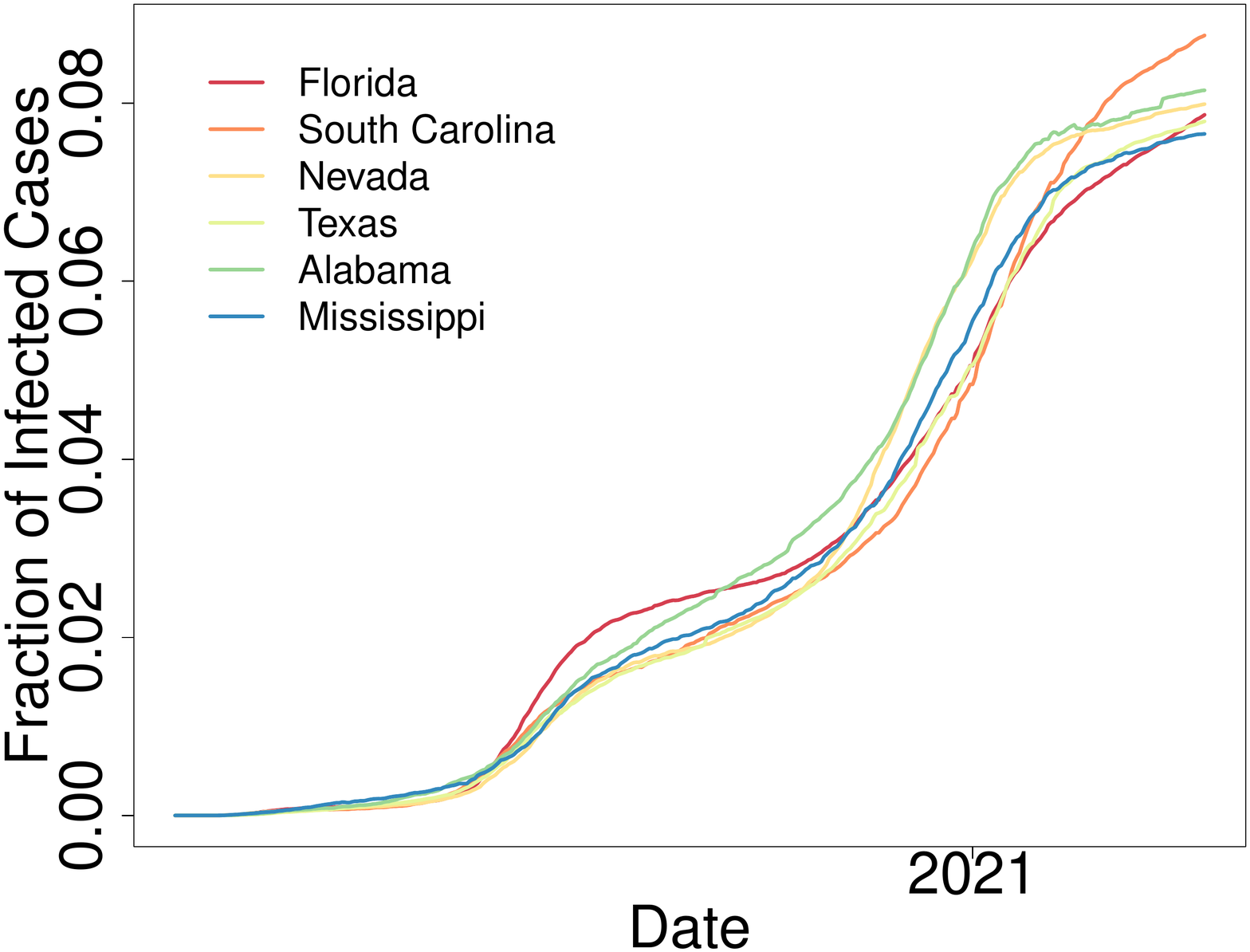}
         \subcaption{Florida}
     \end{subfigure}
     \begin{subfigure}[b]{0.19\textwidth}
         \centering
         \includegraphics[width=\textwidth]{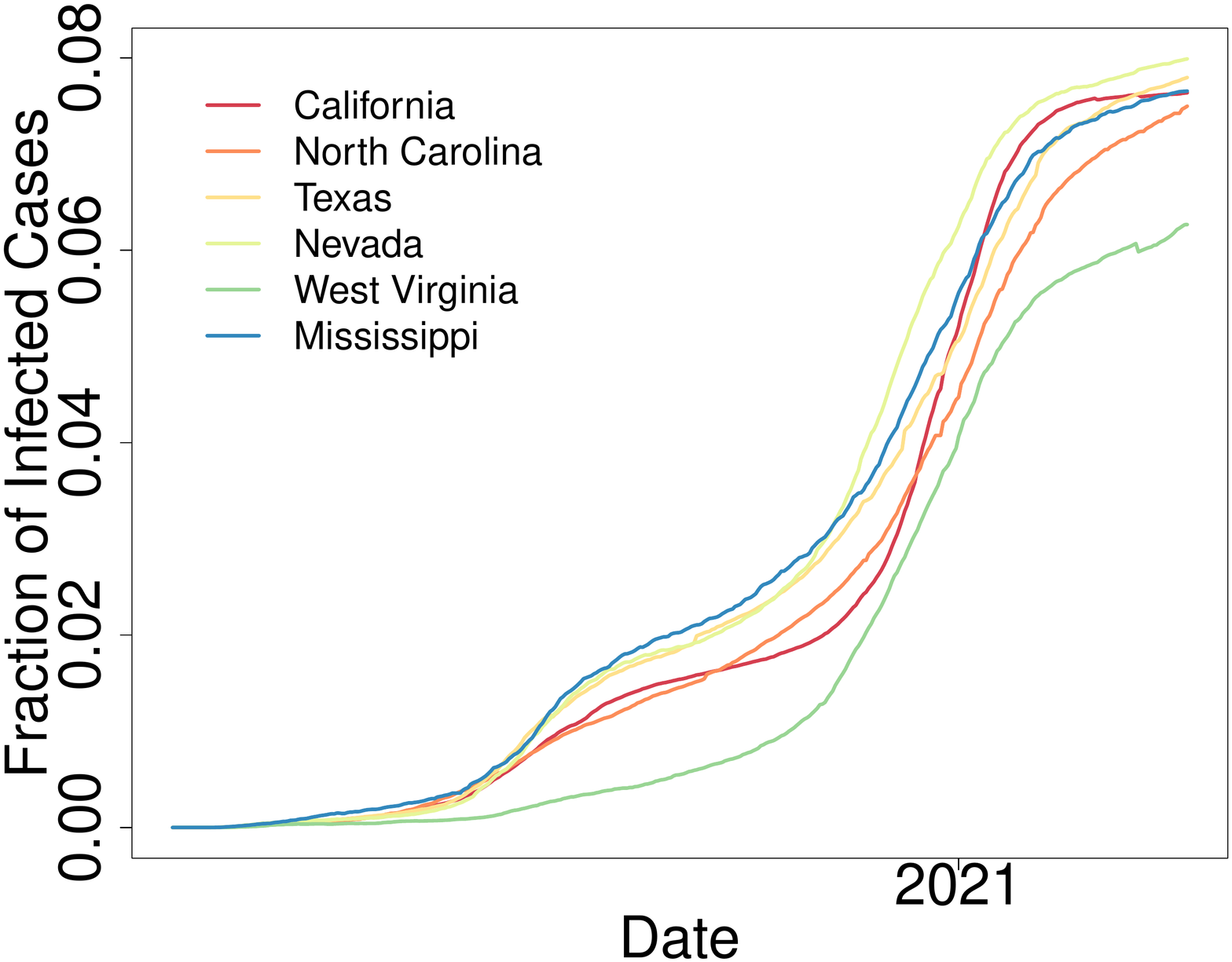}
         \subcaption{California}
     \end{subfigure}
     \begin{subfigure}[b]{0.19\textwidth}
         \centering
         \includegraphics[width=\textwidth]{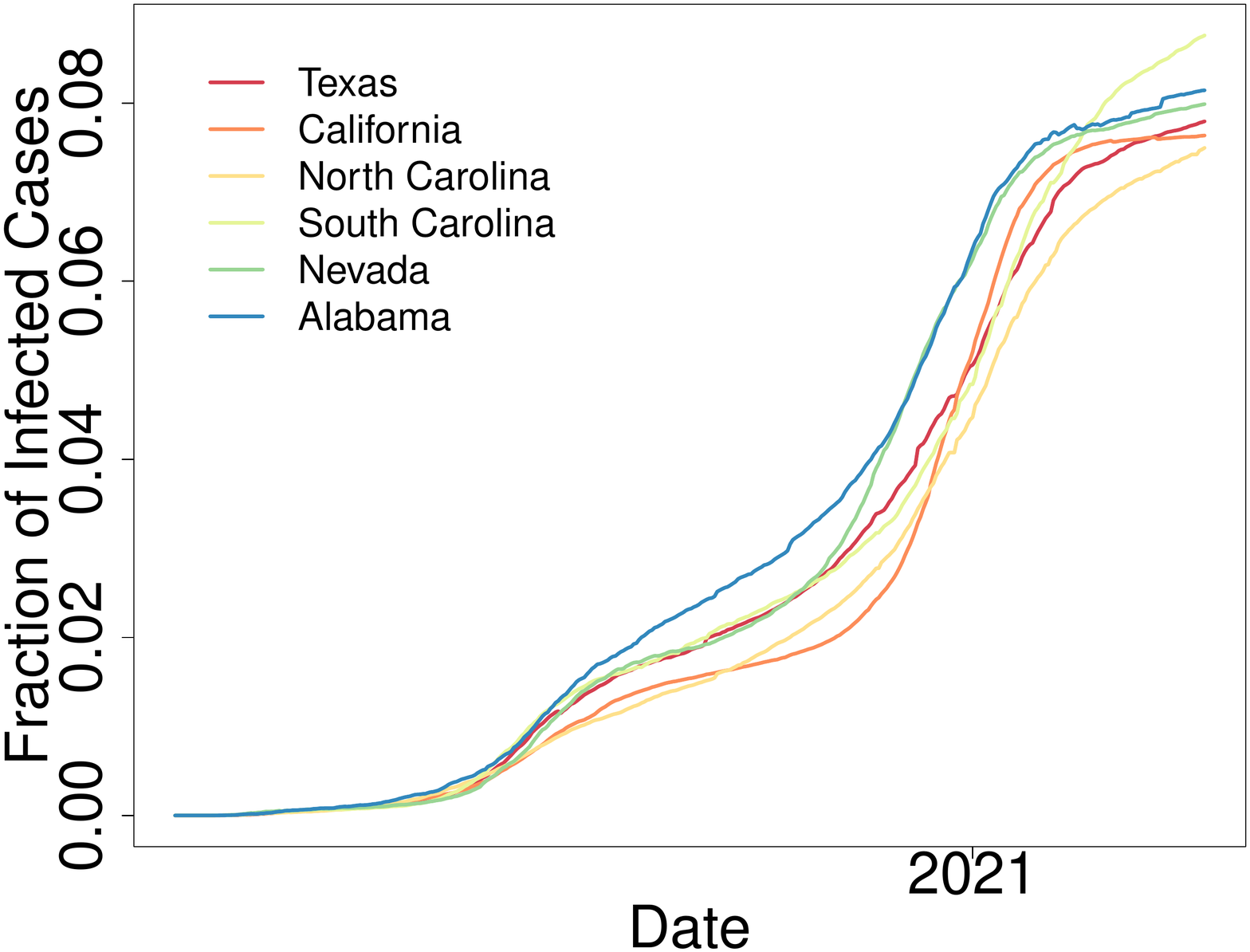}
         \subcaption{Texas}
     \end{subfigure}
     
           \begin{subfigure}[b]{0.19\textwidth}
         \centering
         \includegraphics[width=\textwidth]{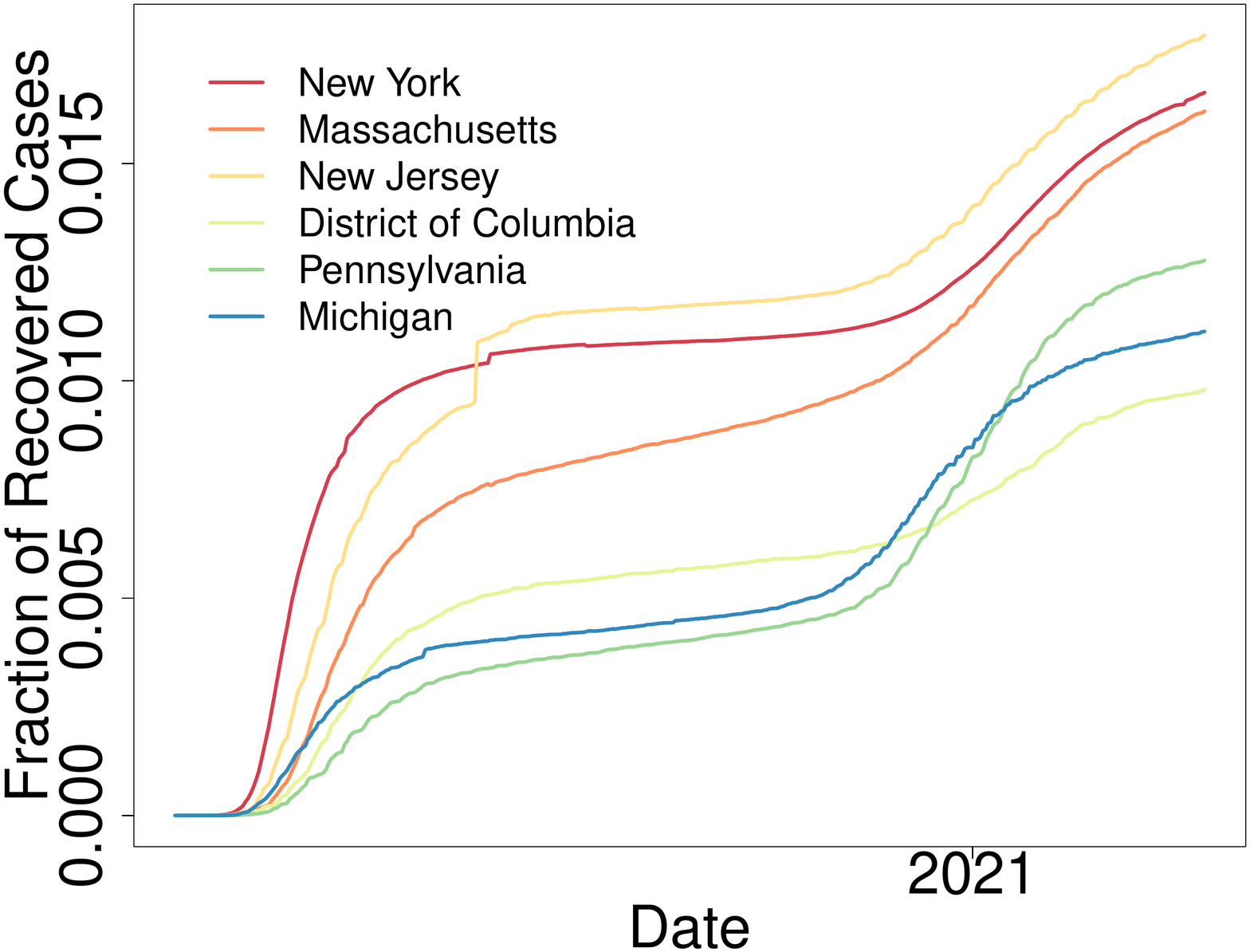}
         \subcaption{New York}
     \end{subfigure}
    \begin{subfigure}[b]{0.19\textwidth}
         \centering
         \includegraphics[width=\textwidth]{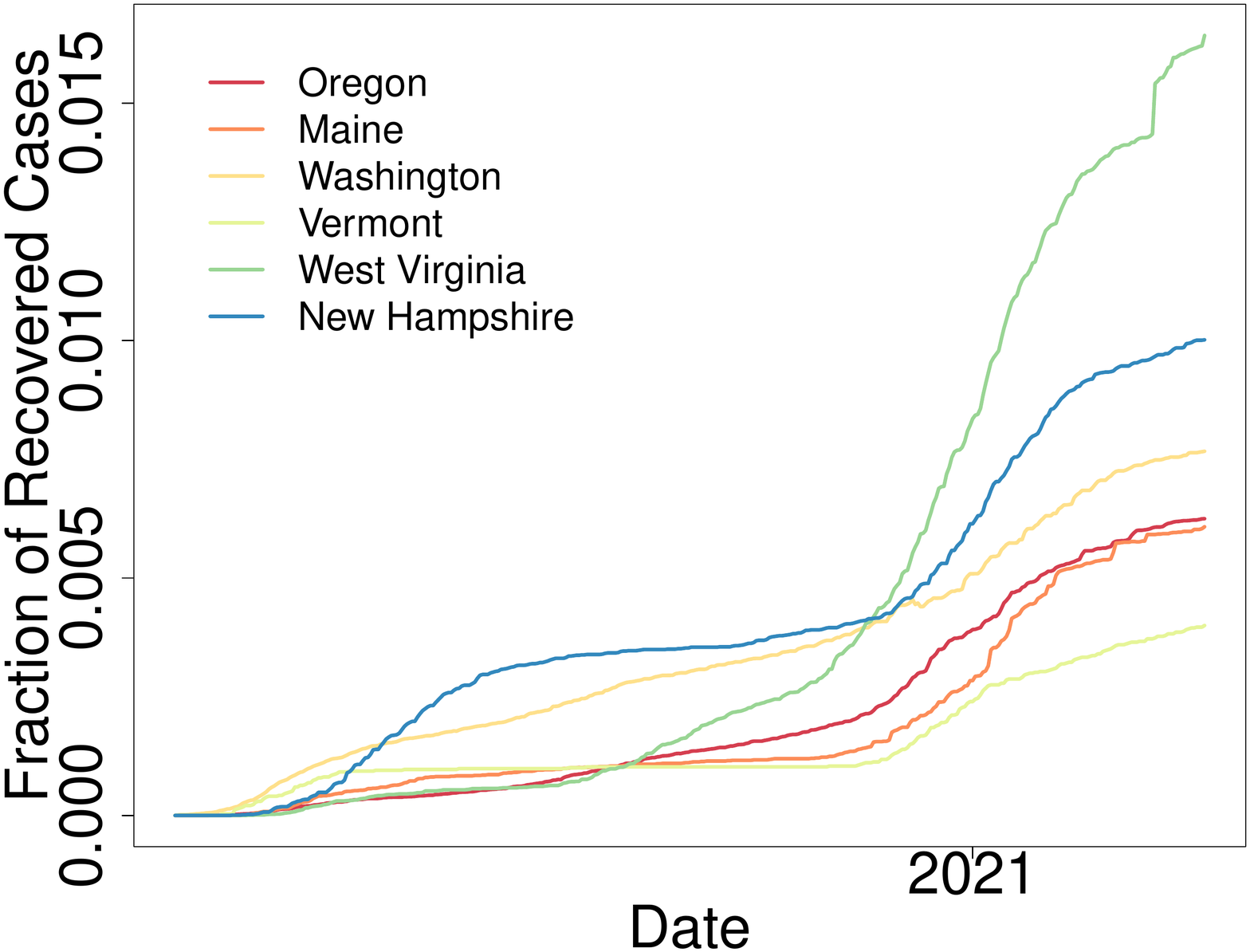}
         \subcaption{Oregon}
     \end{subfigure}
     \begin{subfigure}[b]{0.19\textwidth}
         \centering
         \includegraphics[width=\textwidth]{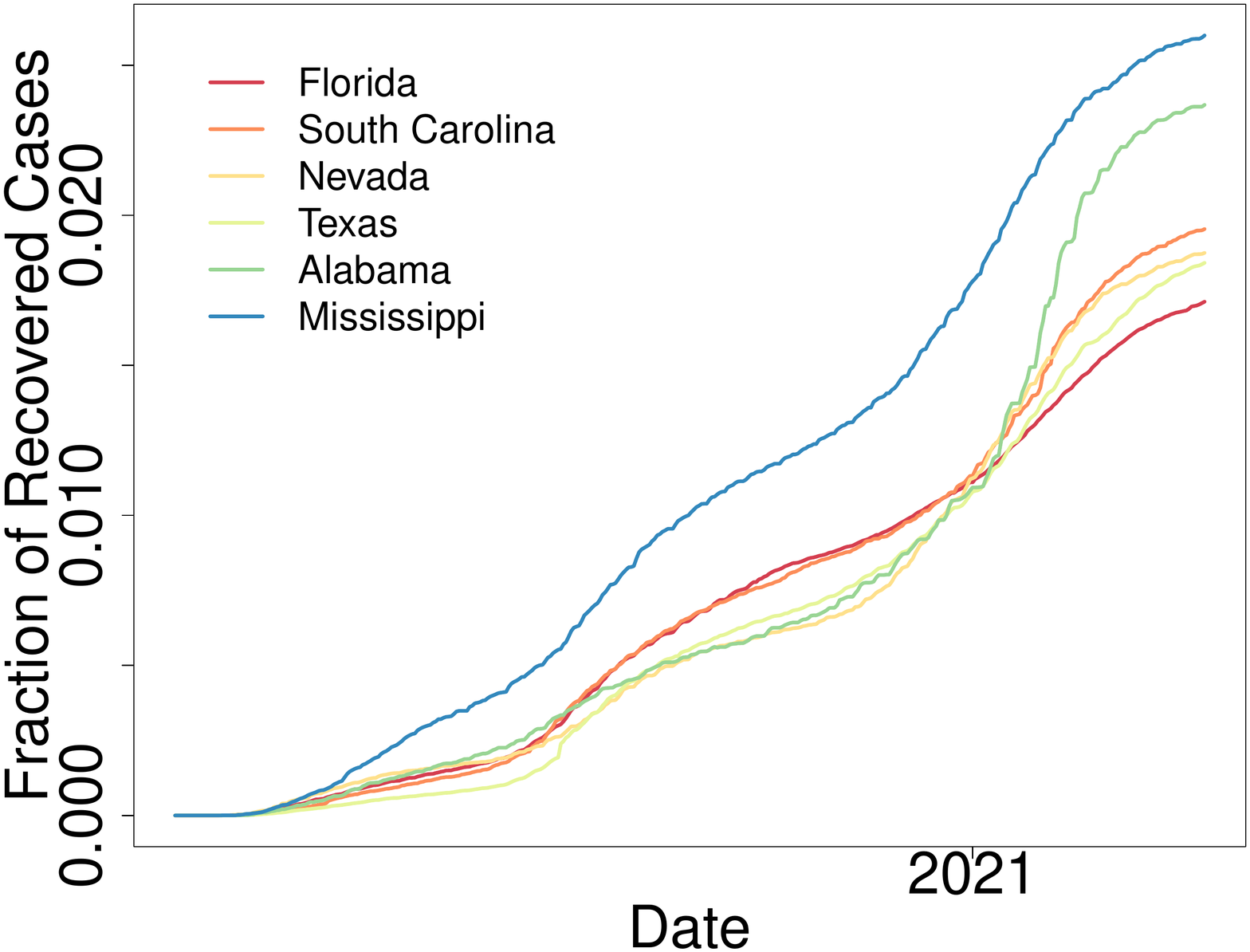}
         \subcaption{Florida}
     \end{subfigure}
     \begin{subfigure}[b]{0.19\textwidth}
         \centering
         \includegraphics[width=\textwidth]{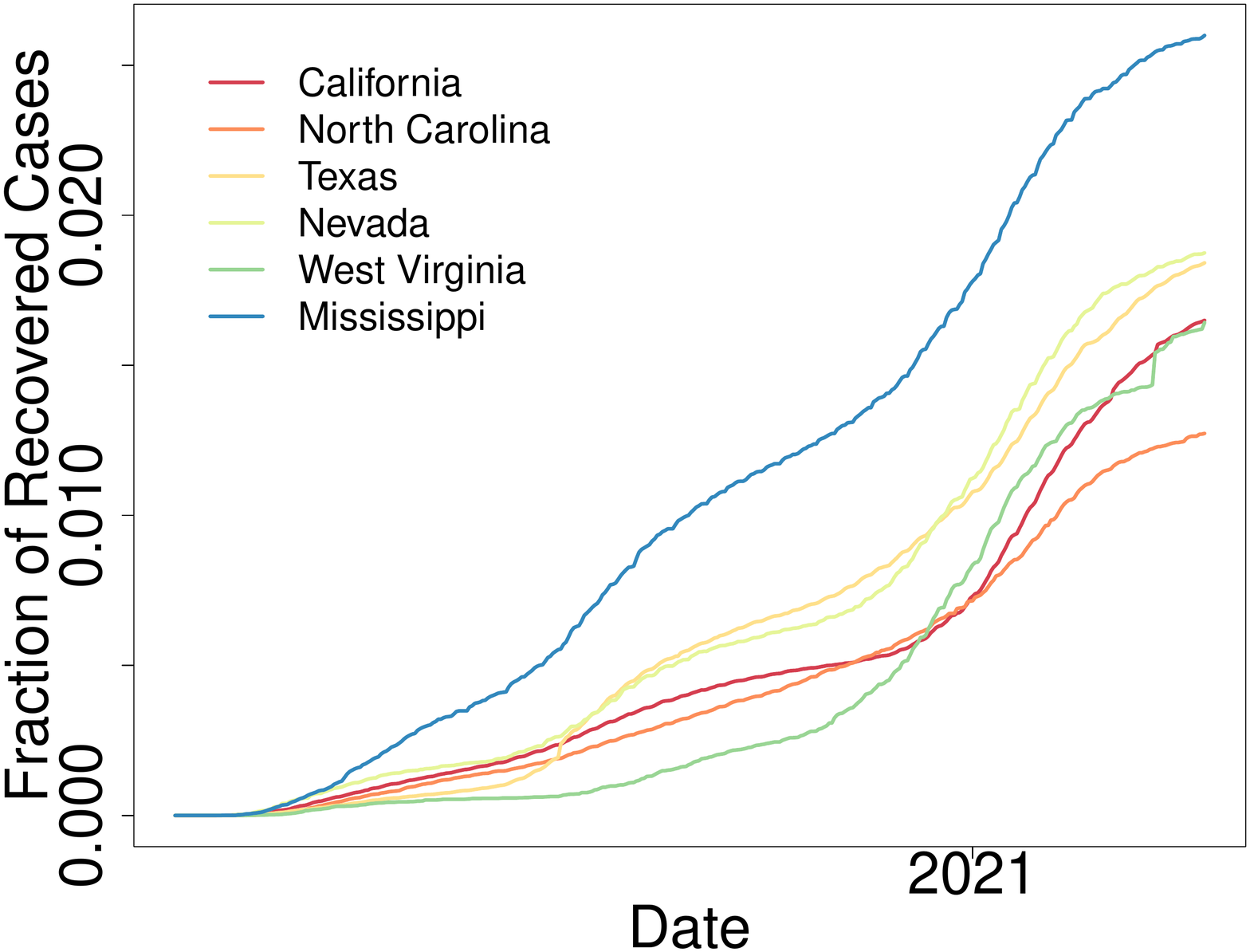}
         \subcaption{California}
     \end{subfigure}
     \begin{subfigure}[b]{0.19\textwidth}
         \centering
         \includegraphics[width=\textwidth]{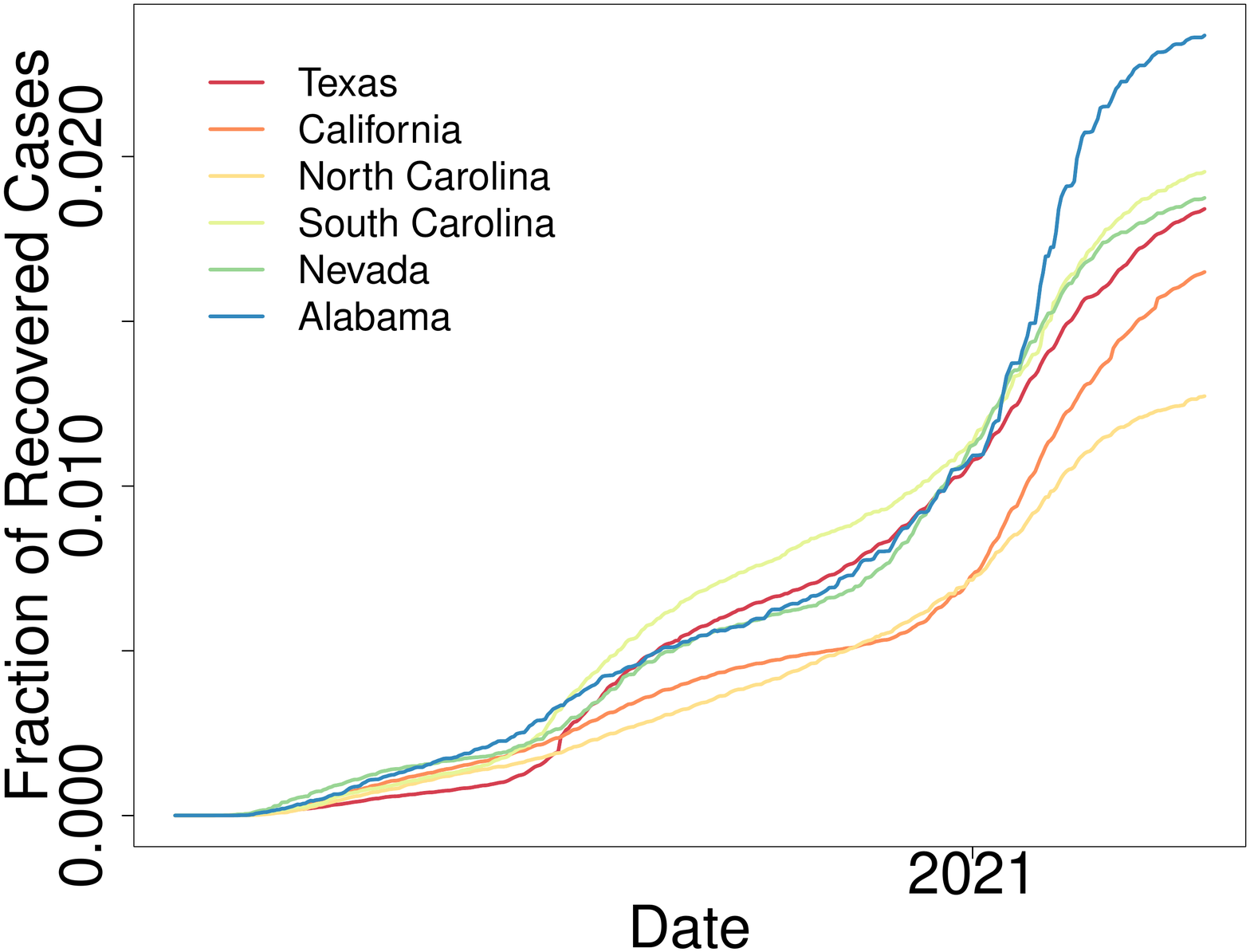}
         \subcaption{Texas}
     \end{subfigure}
    \caption{Fractions of infected (top) and recovered (bottom) in several states and their neighboring states chosen by similarity score. }
        \label{fig:fraction_adj}
\end{figure*}

\begin{figure*}[ht!]
     \centering
     \captionsetup[sub]{font=small, labelfont={bf,sf}}
     \begin{subfigure}[b]{0.24\textwidth}
         \centering
         \includegraphics[width=\textwidth]{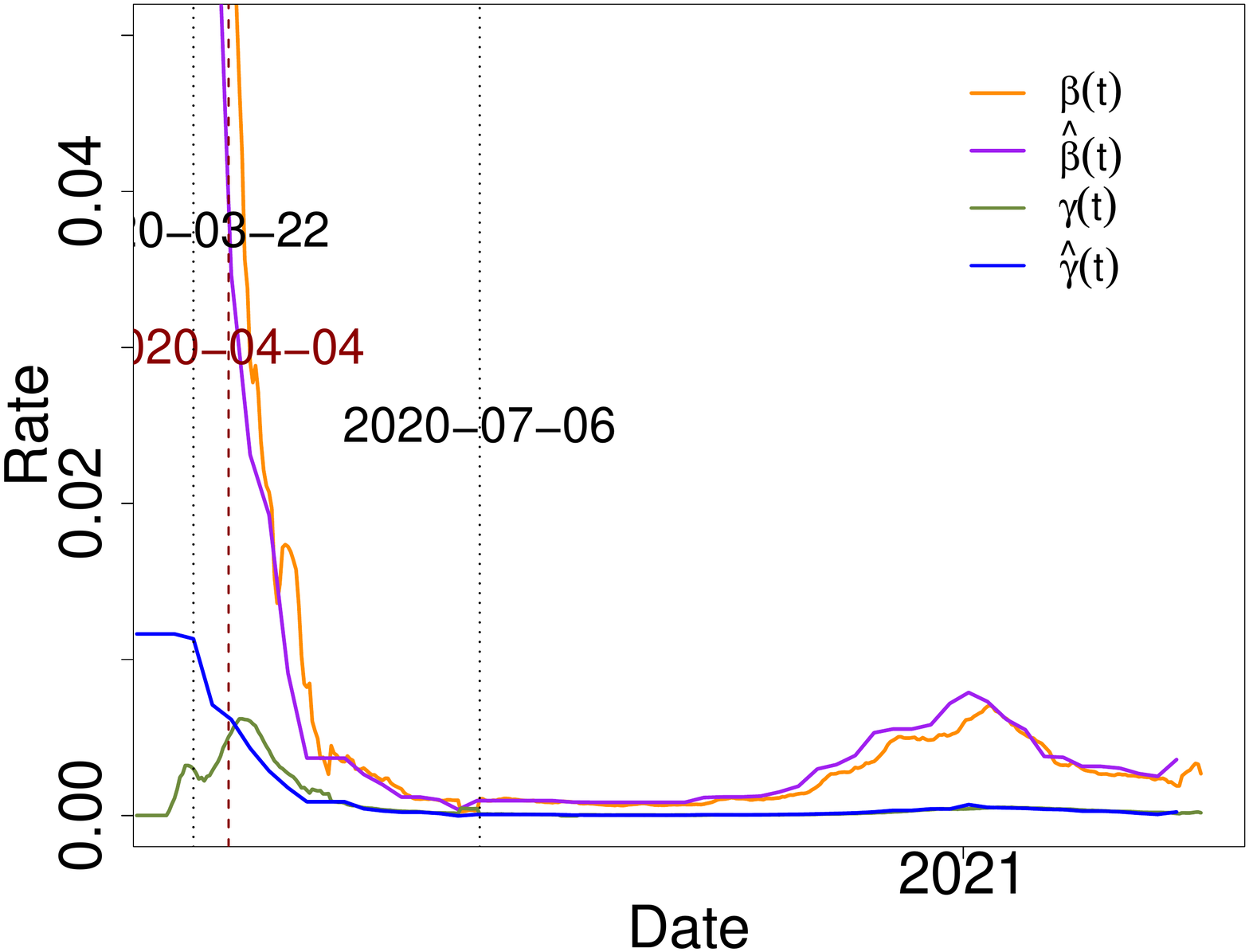}
         \subcaption{NY (Model 1)}
     \end{subfigure}
     \begin{subfigure}[b]{0.24\textwidth}
         \centering
         \includegraphics[width=\textwidth]{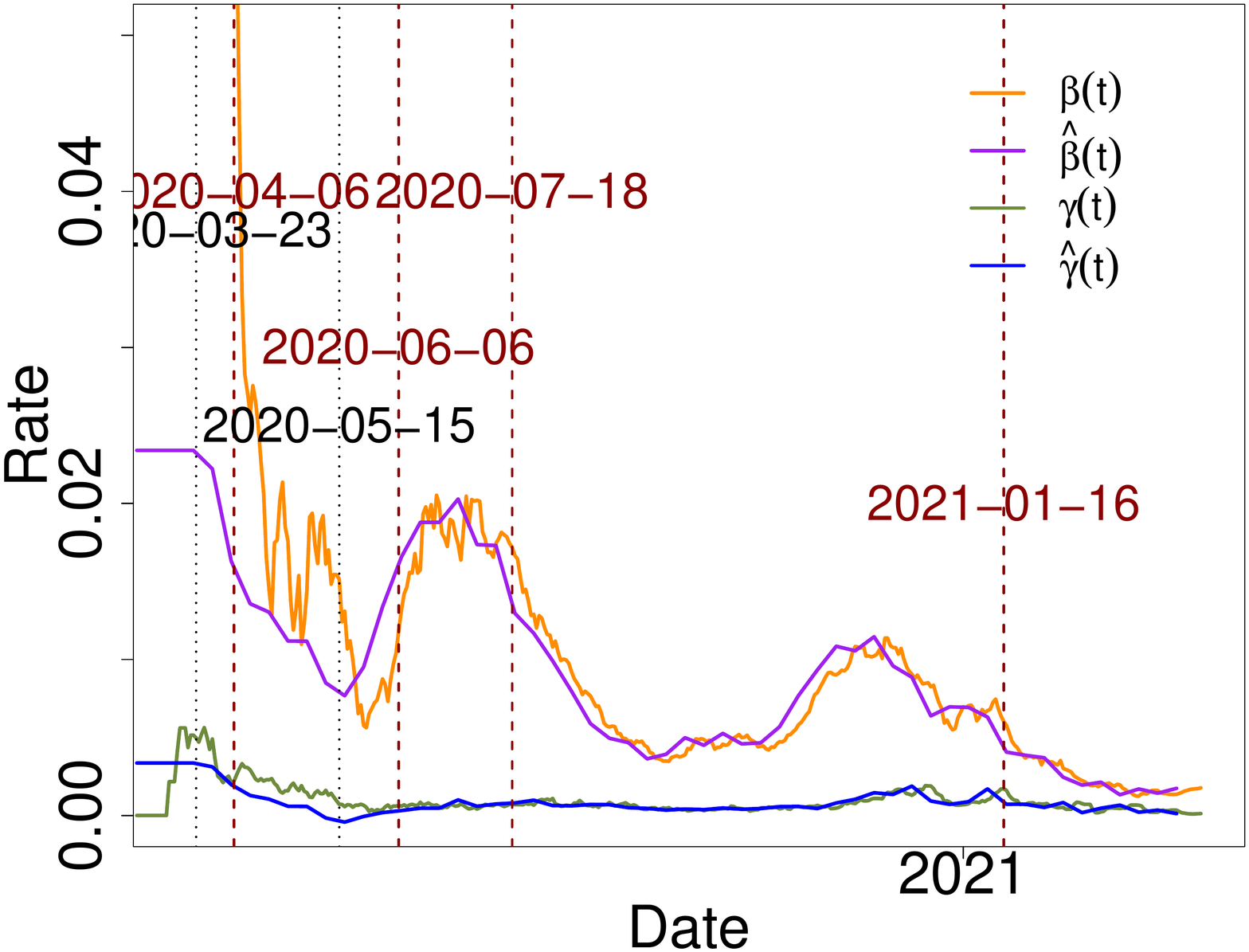}
         \subcaption{OR (Model 1)}
     \end{subfigure}
     \begin{subfigure}[b]{0.24\textwidth}
         \centering
         \includegraphics[width=\textwidth]{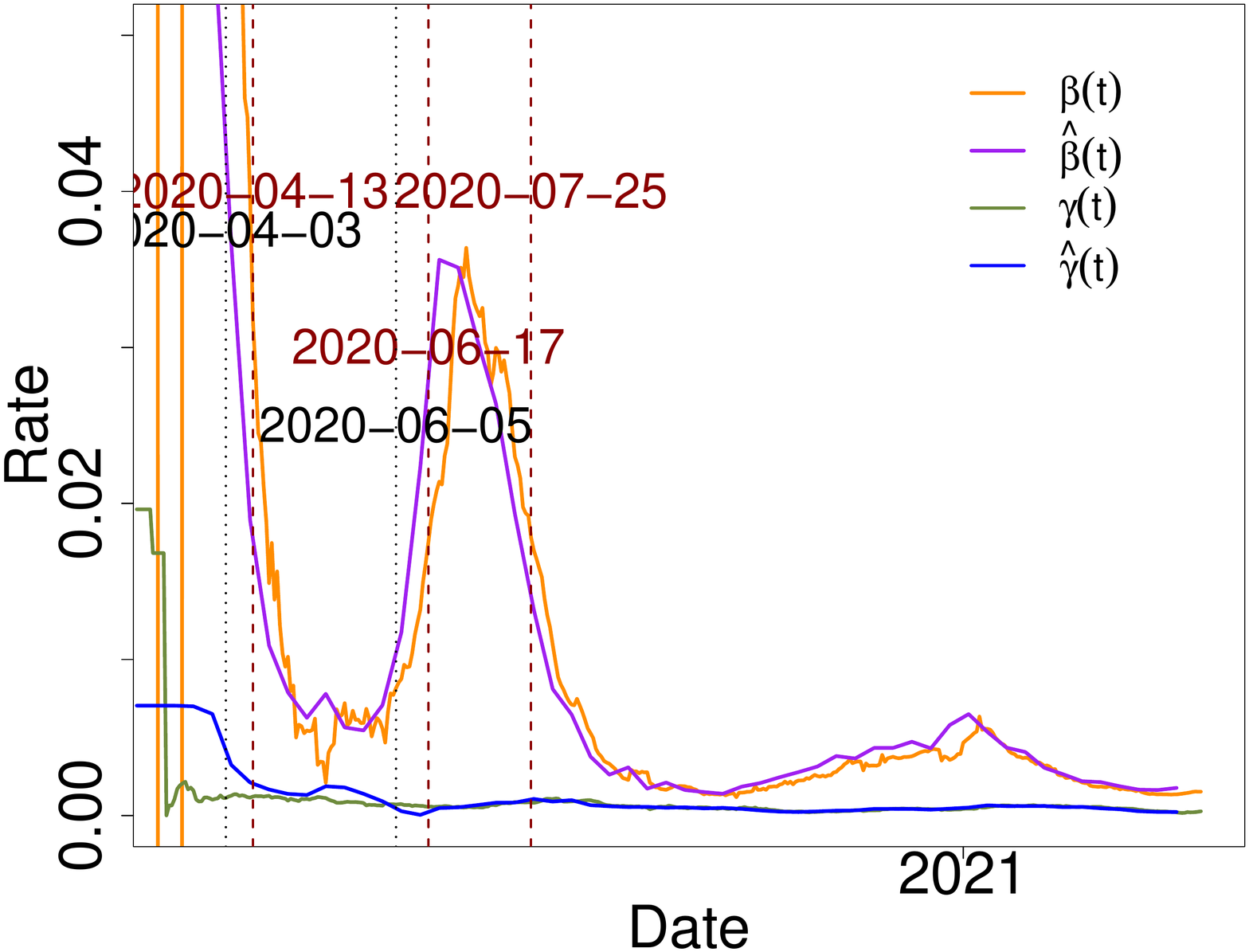}
         \subcaption{FL (Model 1)}
     \end{subfigure}
     
     \begin{subfigure}[b]{0.24\textwidth}
         \centering
         \includegraphics[width=\textwidth]{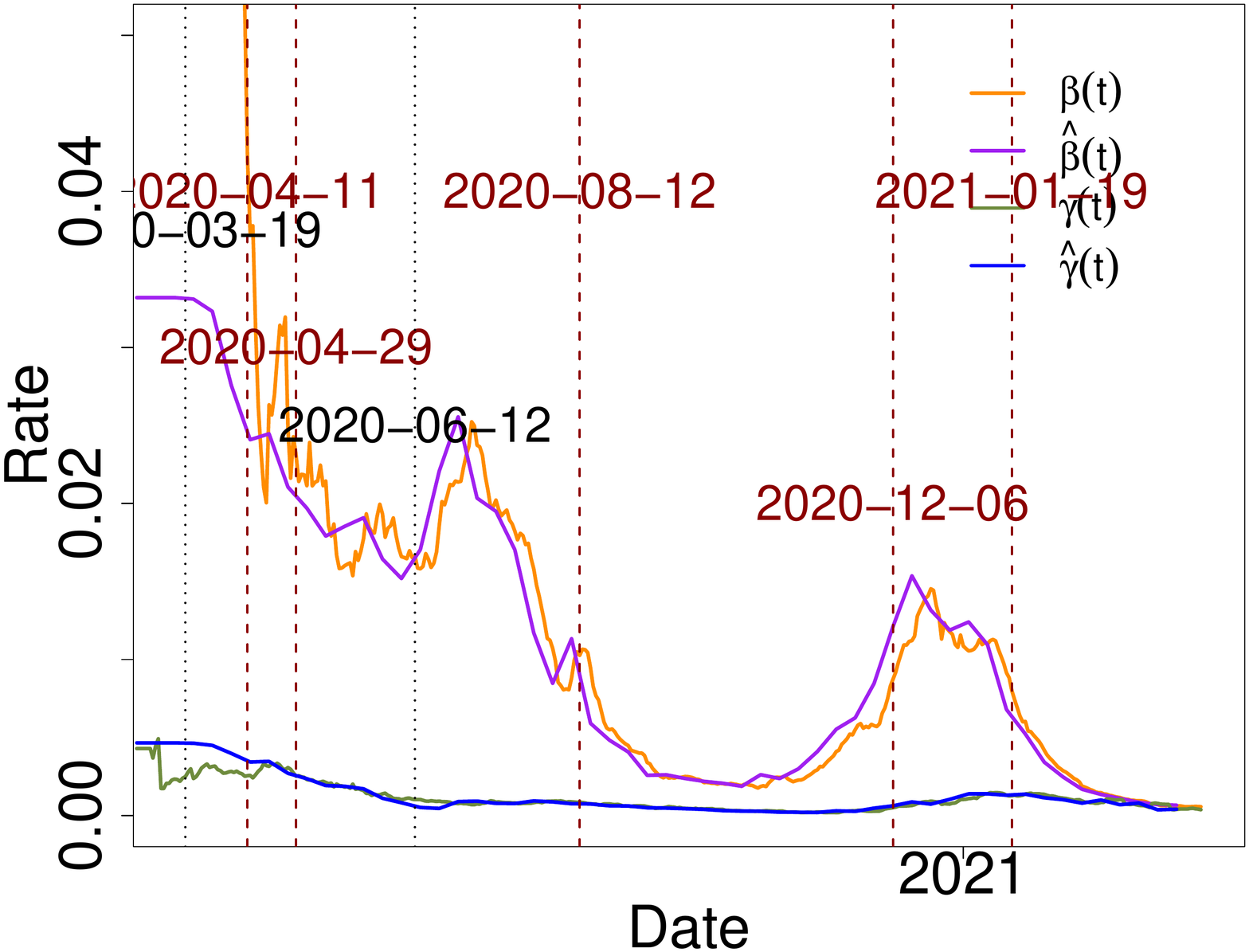}
         \subcaption{CA (Model 1)}
     \end{subfigure}
     \begin{subfigure}[b]{0.24\textwidth}
         \centering
         \includegraphics[width=\textwidth]{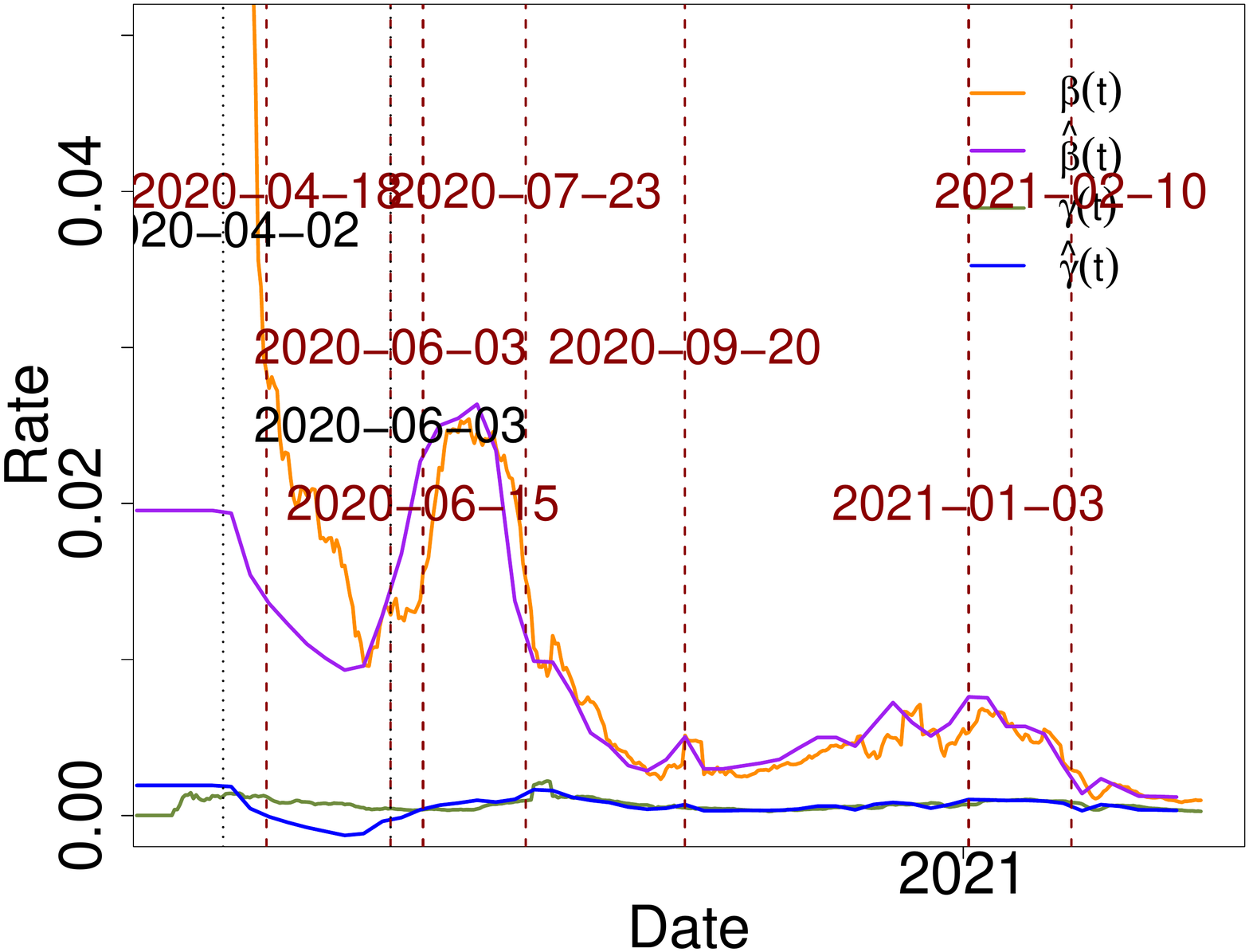}
         \subcaption{TX (Model 1)}
     \end{subfigure}
     \begin{subfigure}[b]{0.24\textwidth}
         \centering
         \includegraphics[width=\textwidth]{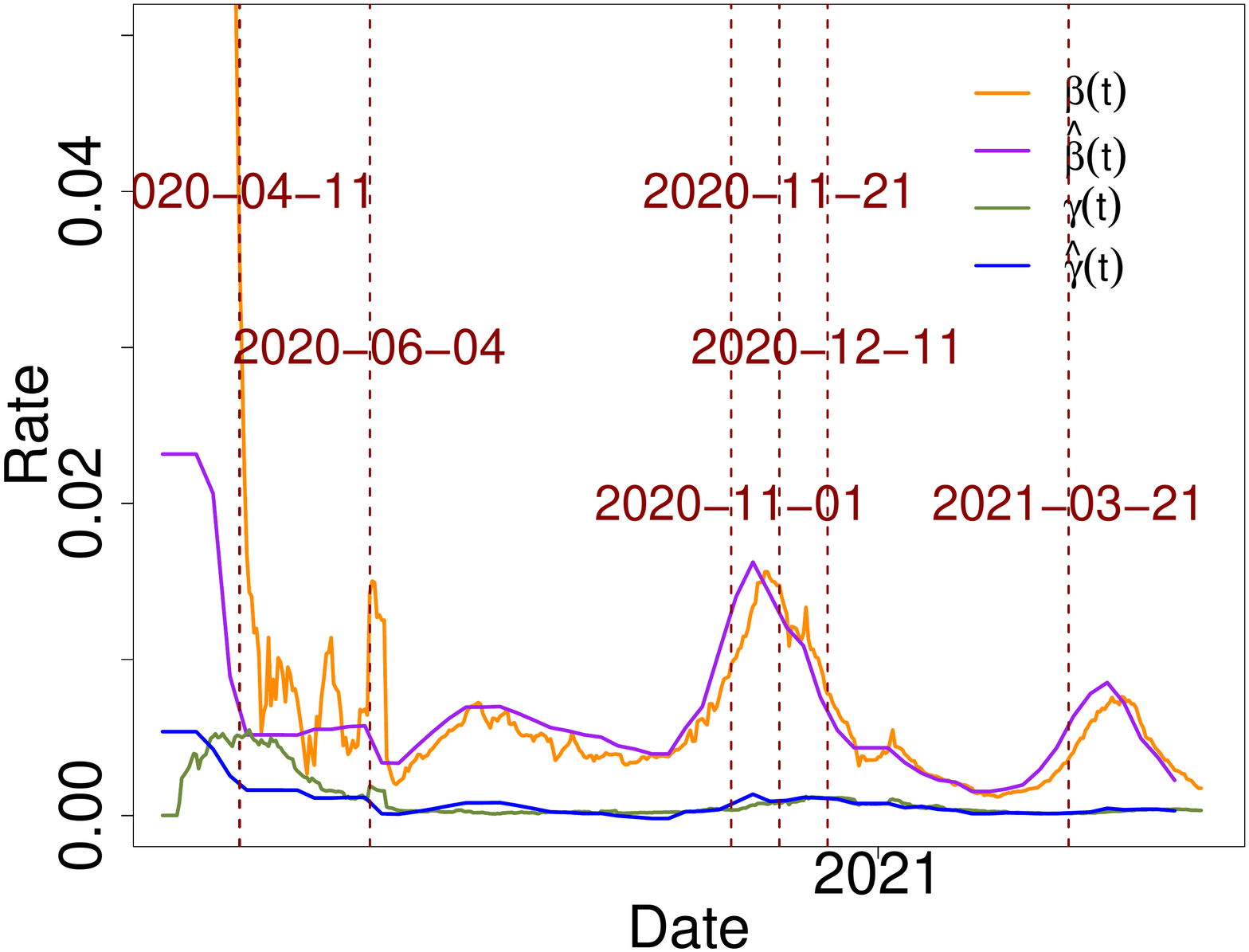}
         \subcaption{MI (Model 1)}
     \end{subfigure}
        \caption{The 7-day moving average of observed transmission rate $\beta(t)$  (orange) and recovery rate  $\gamma(t)$(green) along with the estimated transmission rate  $\widehat{\beta}(t)$ (purple) and recovery rate $\widehat{\gamma}(t)$ (blue)  in several states. The vertical black dotted line indicates the statewide ``stay-at-home'' order begin date or reopening begin date.
        The vertical dark red dashed line indicates the estimated change time point for each state.
        }
        \label{fig:rates_states_smooth}
\end{figure*}

Few additional plots and tables related to the results of applying our method to some U.S. states are provided in this section. 
Let $I(t)$ and $R(t)$ denote the number of infected and recovered individuals (cases) on day $t$. Figure \ref{fig:numbers_states} depicts the actual case numbers $I(t)$ and $R(t)$ in the five states considered.
In Figure \ref{fig:fraction_adj}, we provide fractions of infected and recovered cases in given states and their neighboring states selected by similarity score in the Model 2.3. 

In Figure \ref{fig:rates_states_smooth}, we show the 7-day moving average of measured transmission rate $\beta(t)$ and recovery rate $\gamma(t)$
along with the estimated $\widehat{\beta}(t)$ and $\widehat{\gamma}(t)$ from Model 1.
To further examine the spread of the COVID-19, one could consider
\[R_0(t) = \frac{\beta(t)}{\gamma(t)},\]
where $R_0(t)$ is the basic reproduction number of a newly infected person at time $t$. It mean that an infected person can further infect on average  $R_0(t)$ persons. 
Generally, when $R_0 > 1$, the infection will be able to start spreading in a population with an exponential rate, and it will be hard to control the epidemic.

Further, the observed and fitted number of infected cases and recovered cases are displayed in Figure~\ref{fig:number of infected} and Figure~\ref{fig:number of recovered}, respectively. The fitted number of infected cases and recovered cases are defined as 
\begin{align}
    \widetilde{I}(t) = I(1) + \sum_{k=1}^{t-1}\widehat{\Delta I}(k), \quad \widetilde{R}(t) = R(1) + \sum_{k=1}^{t-1}\widehat{\Delta R}(k),
\end{align}
for all $t = 2,\dots, T$.

\begin{figure*}[ht!]
     \centering
     \begin{subfigure}[b]{0.19\textwidth}
         \centering
         \includegraphics[width=\textwidth]{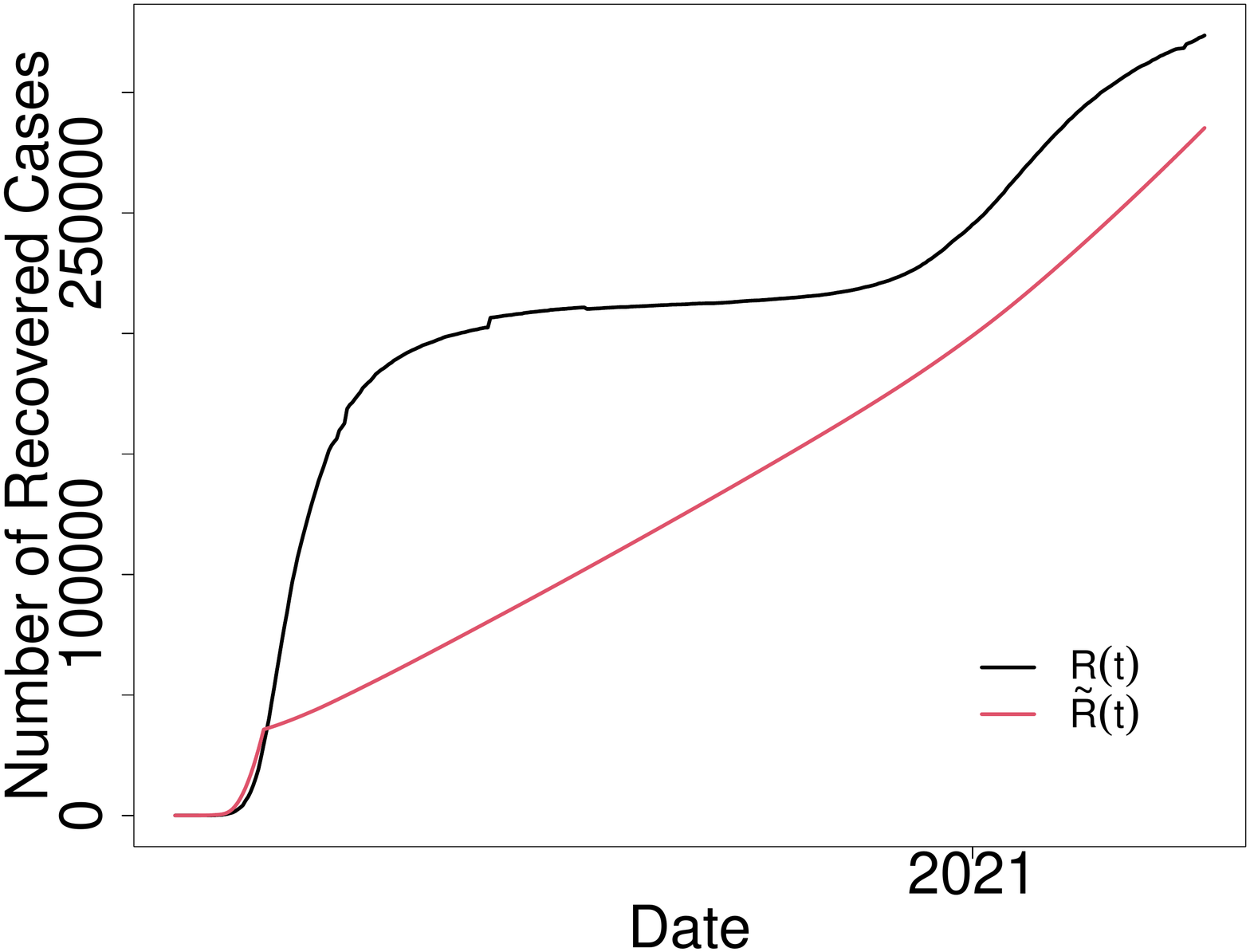}
         \subcaption{NY (Model 1)}
     \end{subfigure}
    \begin{subfigure}[b]{0.19\textwidth}
         \centering
         \includegraphics[width=\textwidth]{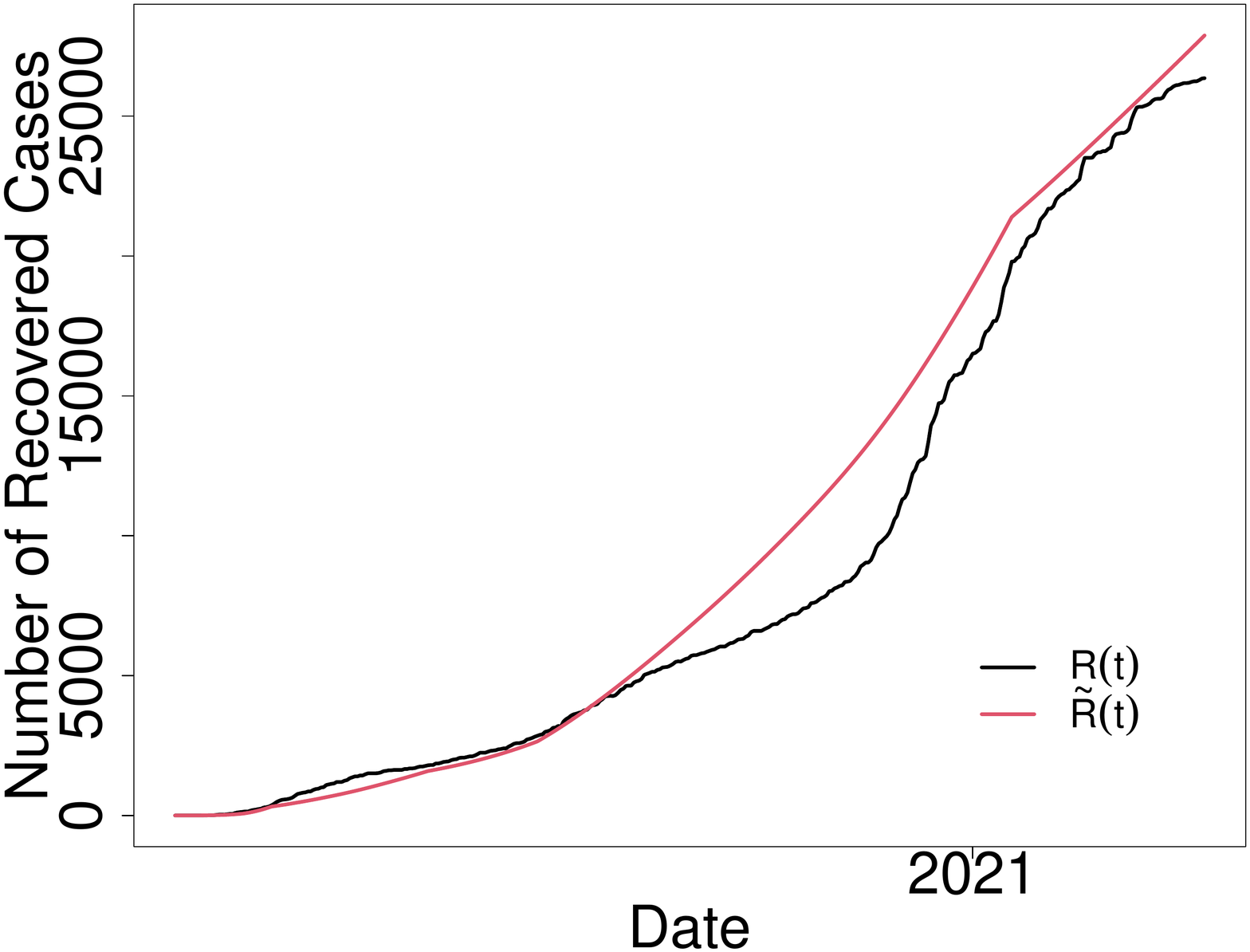}
         \subcaption{OR (Model 1)}
     \end{subfigure}
     \begin{subfigure}[b]{0.19\textwidth}
         \centering
         \includegraphics[width=\textwidth]{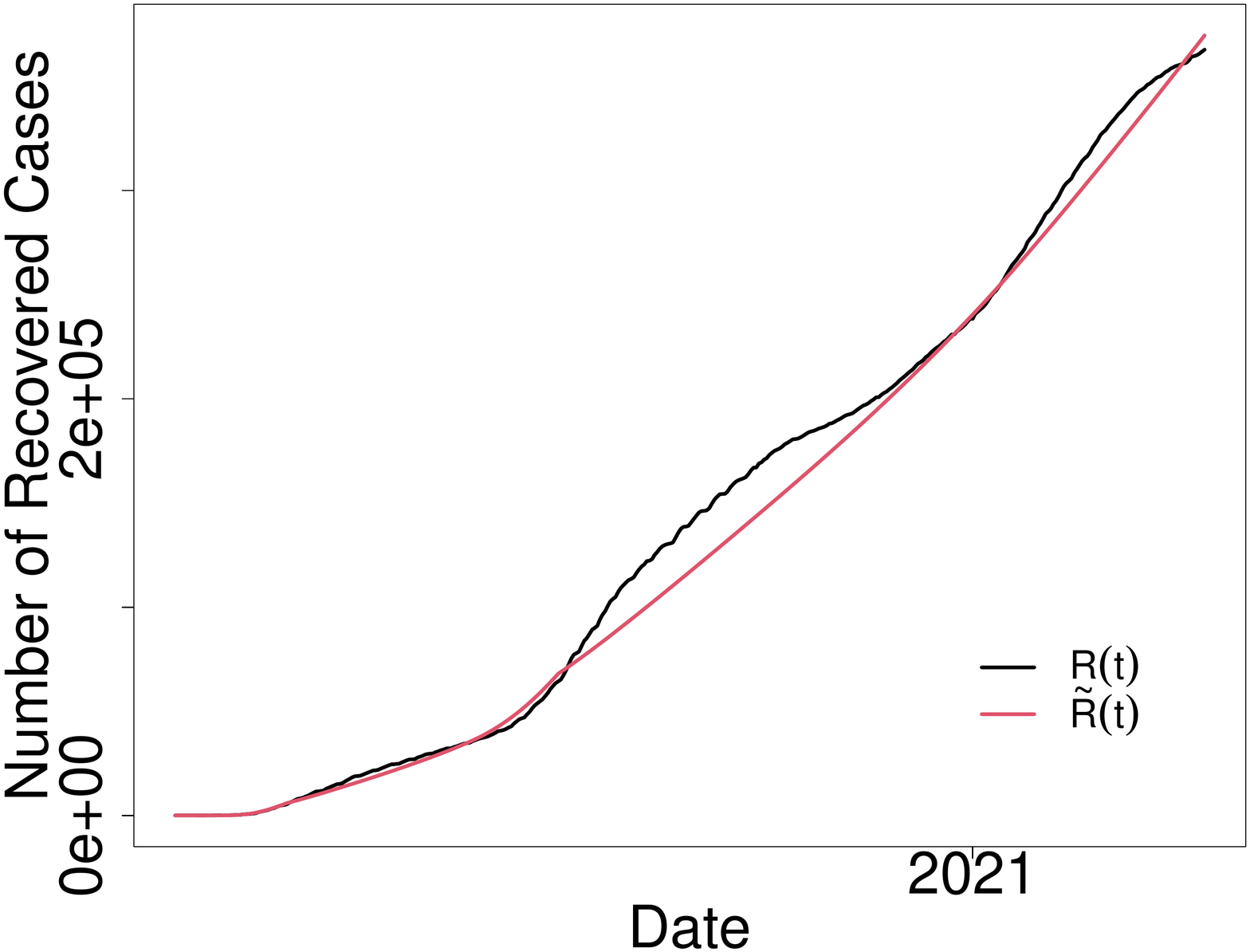}
         \subcaption{FL (Model 1)}
     \end{subfigure}
     \begin{subfigure}[b]{0.19\textwidth}
         \centering
         \includegraphics[width=\textwidth]{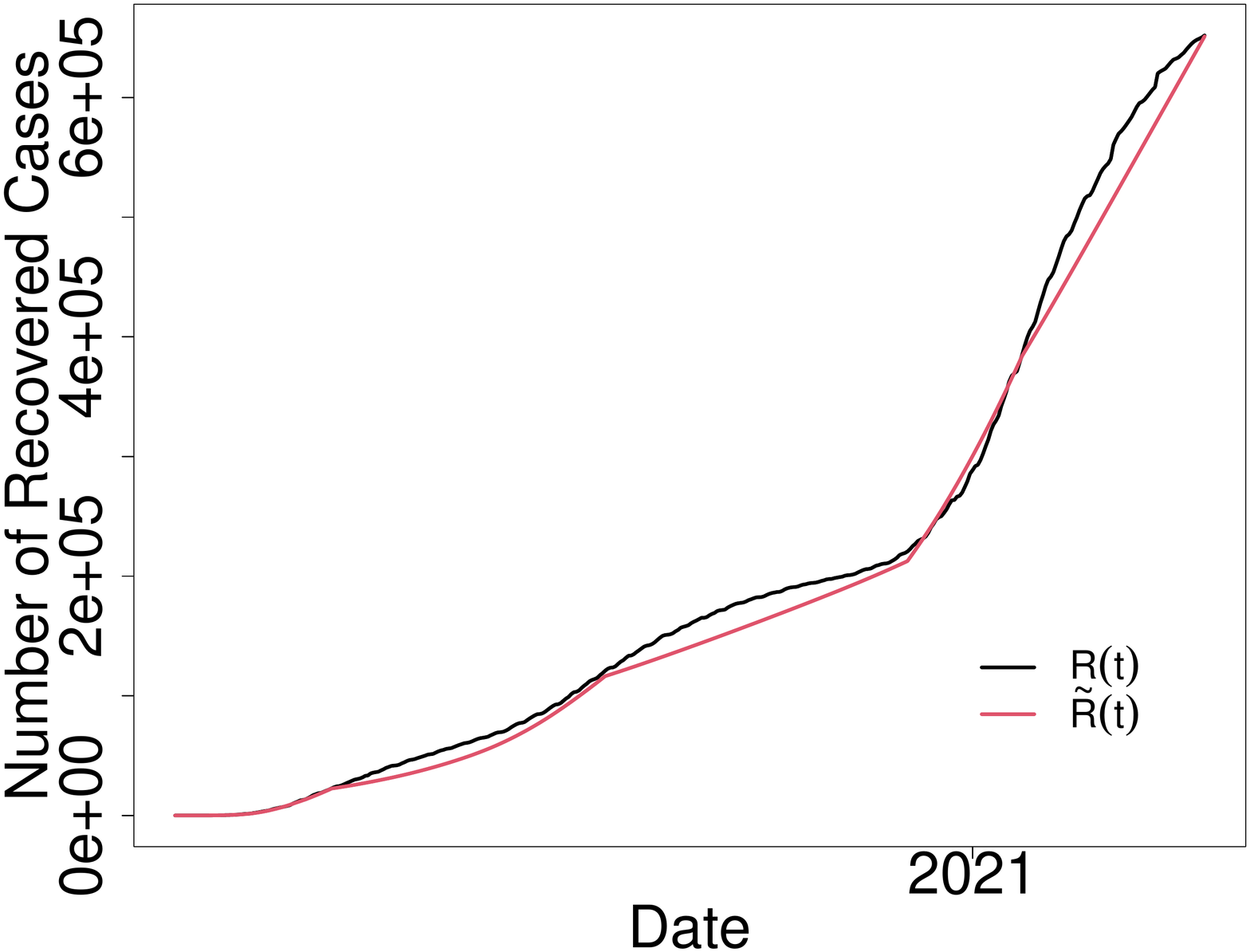}
         \subcaption{CA (Model 1)}
     \end{subfigure}
     \begin{subfigure}[b]{0.19\textwidth}
         \centering
         \includegraphics[width=\textwidth]{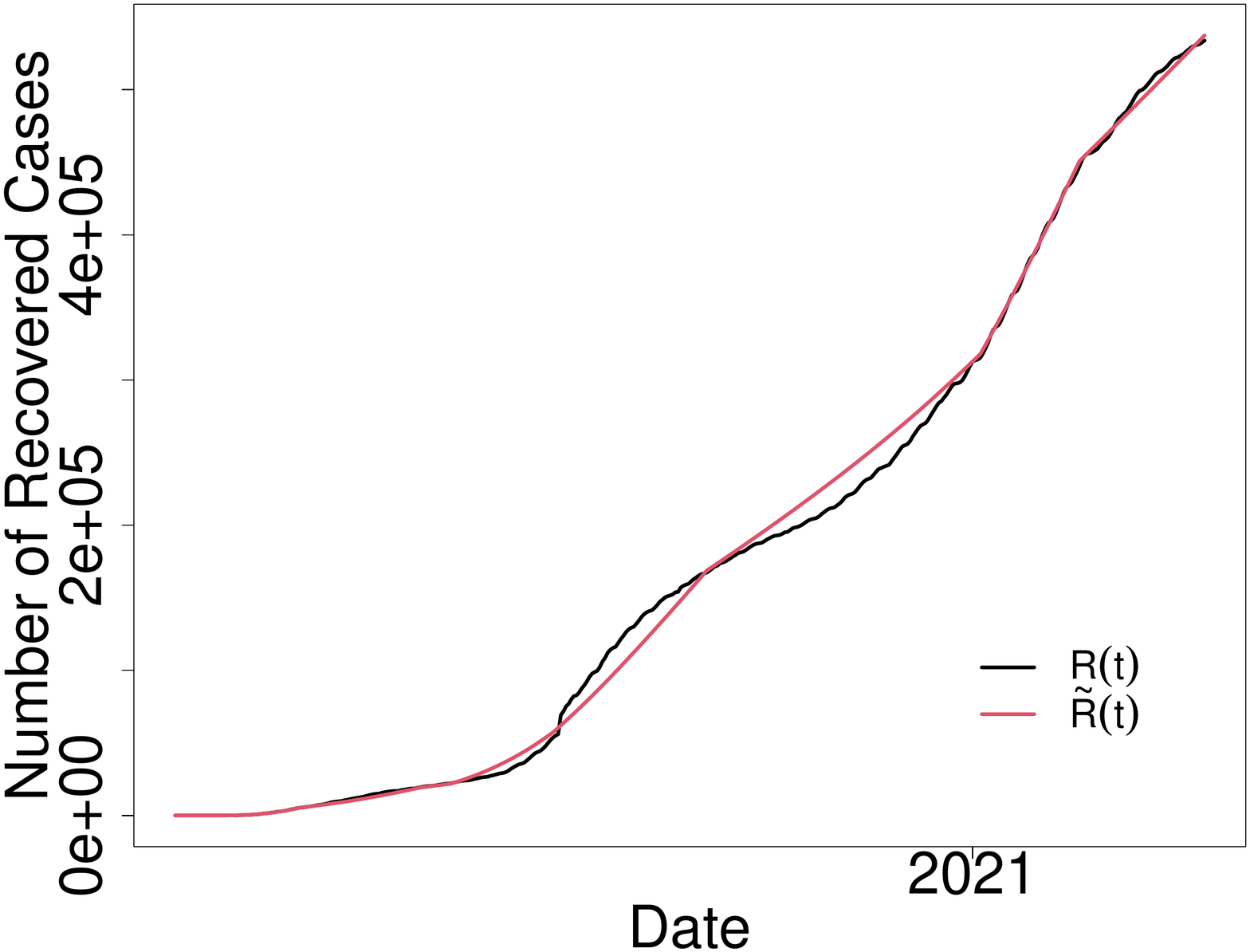}
         \subcaption{TX (Model 1)}
     \end{subfigure}
     
     \begin{subfigure}[b]{0.19\textwidth}
         \centering
         \includegraphics[width=\textwidth]{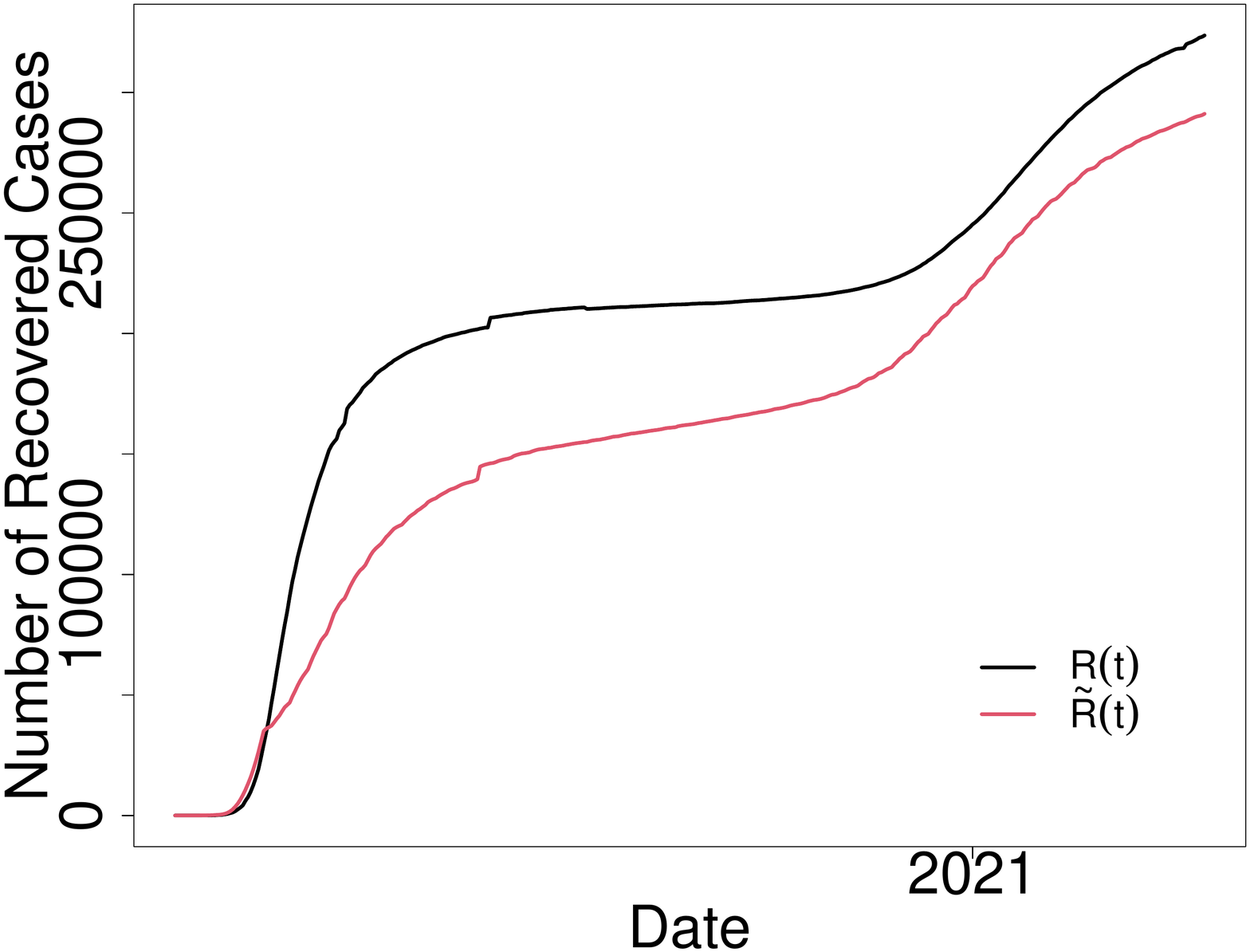}
         \subcaption{NY (Model 2.3)}
     \end{subfigure}
    \begin{subfigure}[b]{0.19\textwidth}
         \centering
         \includegraphics[width=\textwidth]{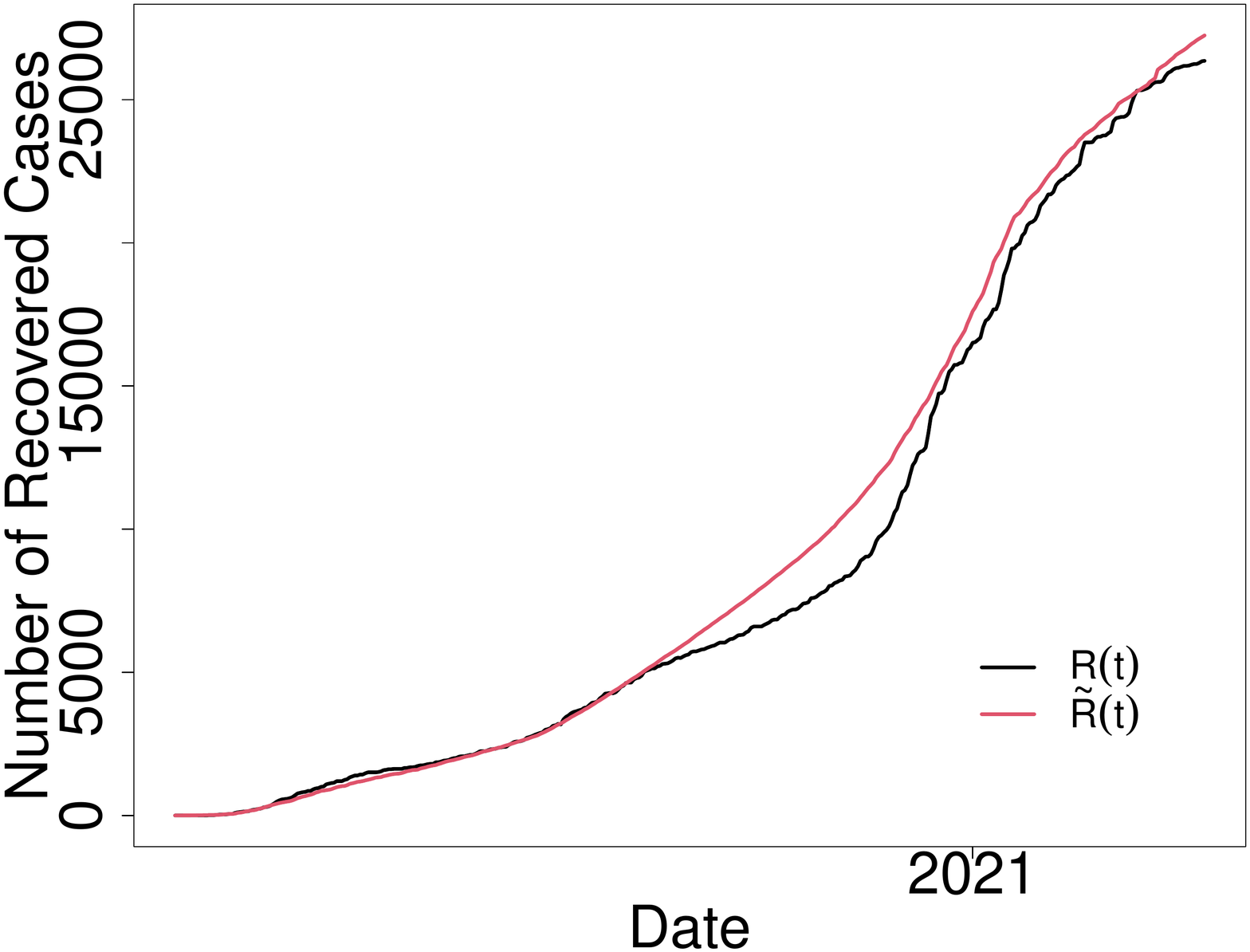}
         \subcaption{OR (Model 2.3)}
     \end{subfigure}
     \begin{subfigure}[b]{0.19\textwidth}
         \centering
         \includegraphics[width=\textwidth]{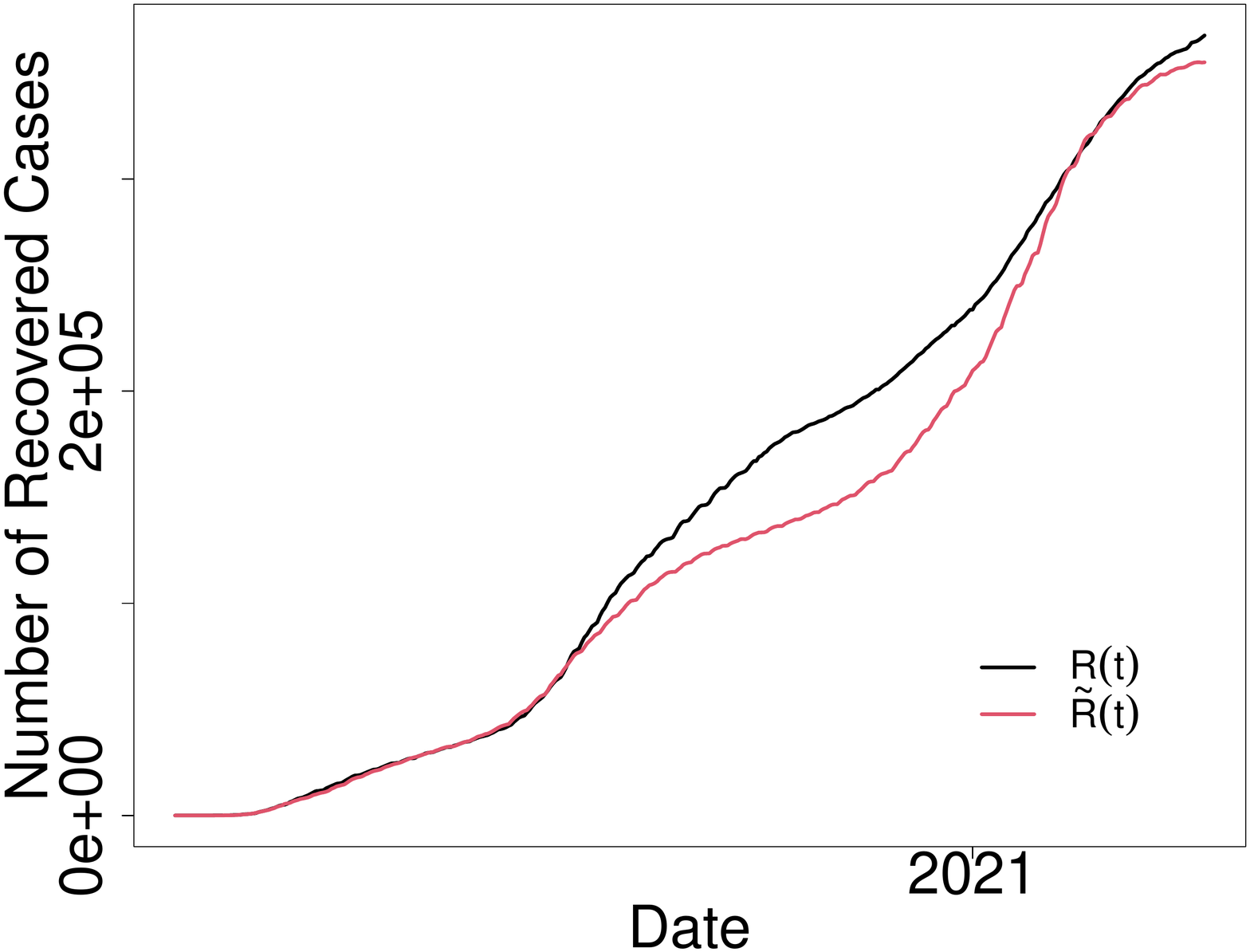}
         \subcaption{FL (Model 2.3)}
     \end{subfigure}
     \begin{subfigure}[b]{0.19\textwidth}
         \centering
         \includegraphics[width=\textwidth]{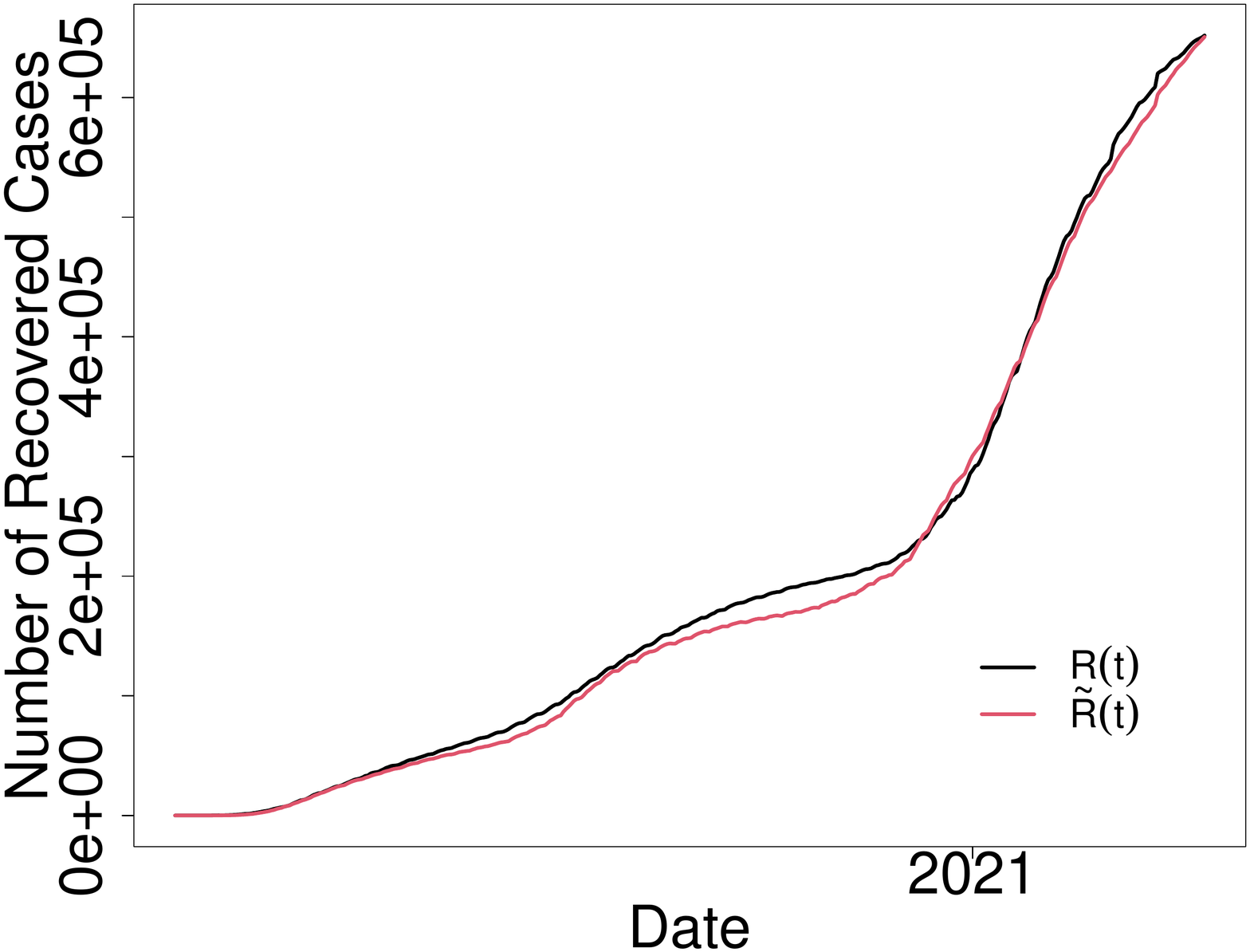}
         \subcaption{CA (Model 2.3)}
     \end{subfigure}
     \begin{subfigure}[b]{0.19\textwidth}
         \centering
         \includegraphics[width=\textwidth]{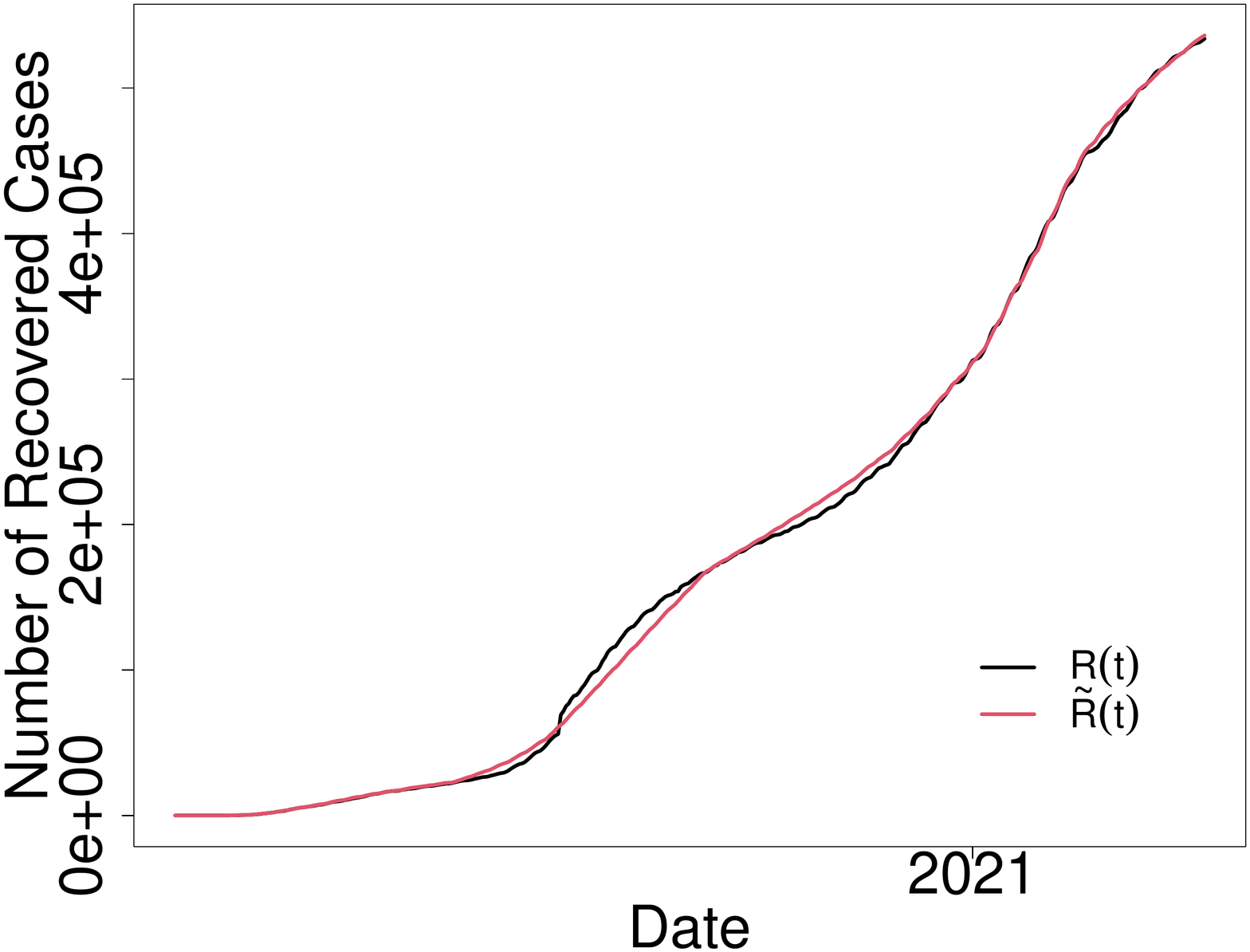}
         \subcaption{TX (Model 2.3)}
     \end{subfigure}
     
      \begin{subfigure}[b]{0.19\textwidth}
         \centering
         \includegraphics[width=\textwidth]{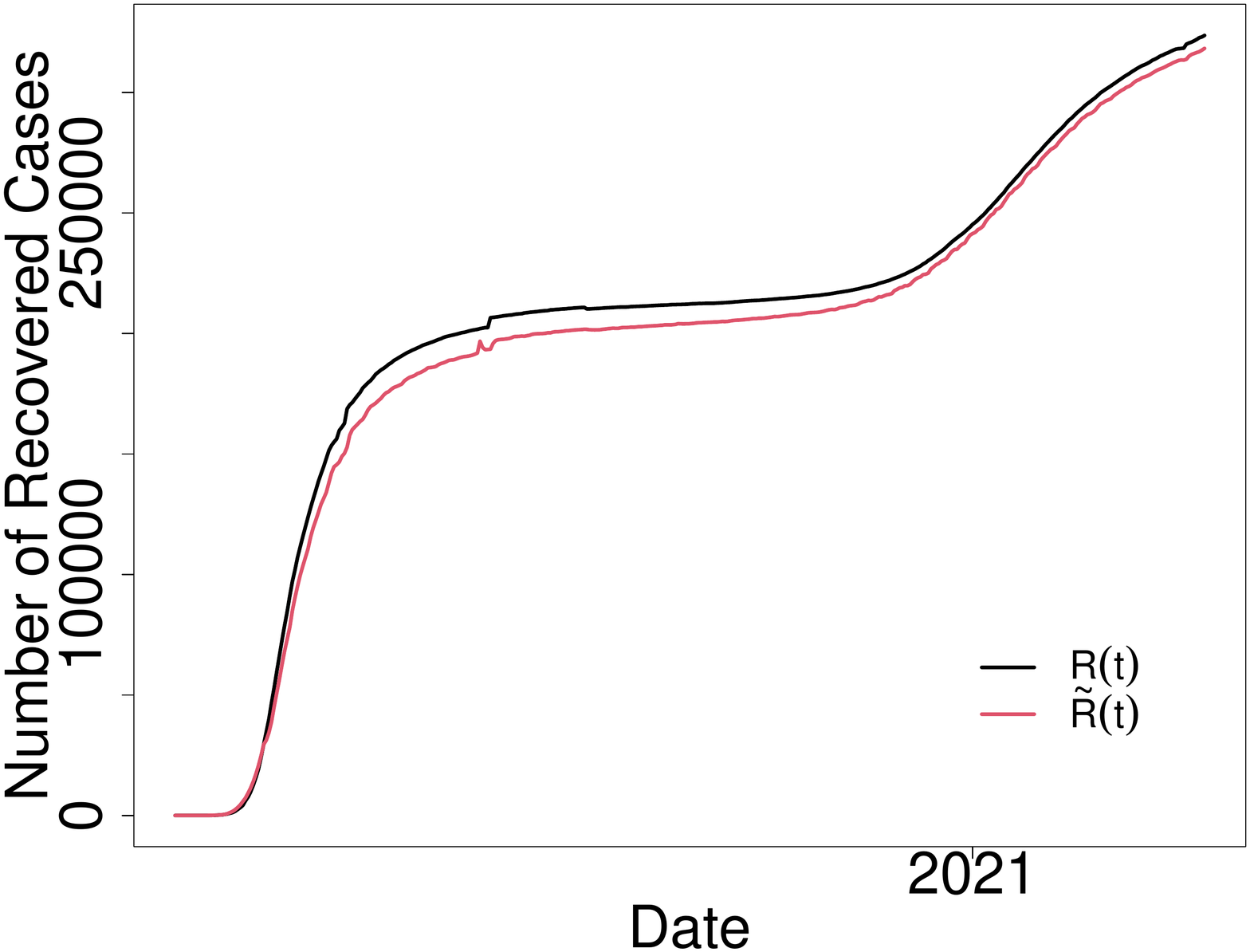}
         \subcaption{NY (Model 3)}
     \end{subfigure}
     \begin{subfigure}[b]{0.19\textwidth}
         \centering
         \includegraphics[width=\textwidth]{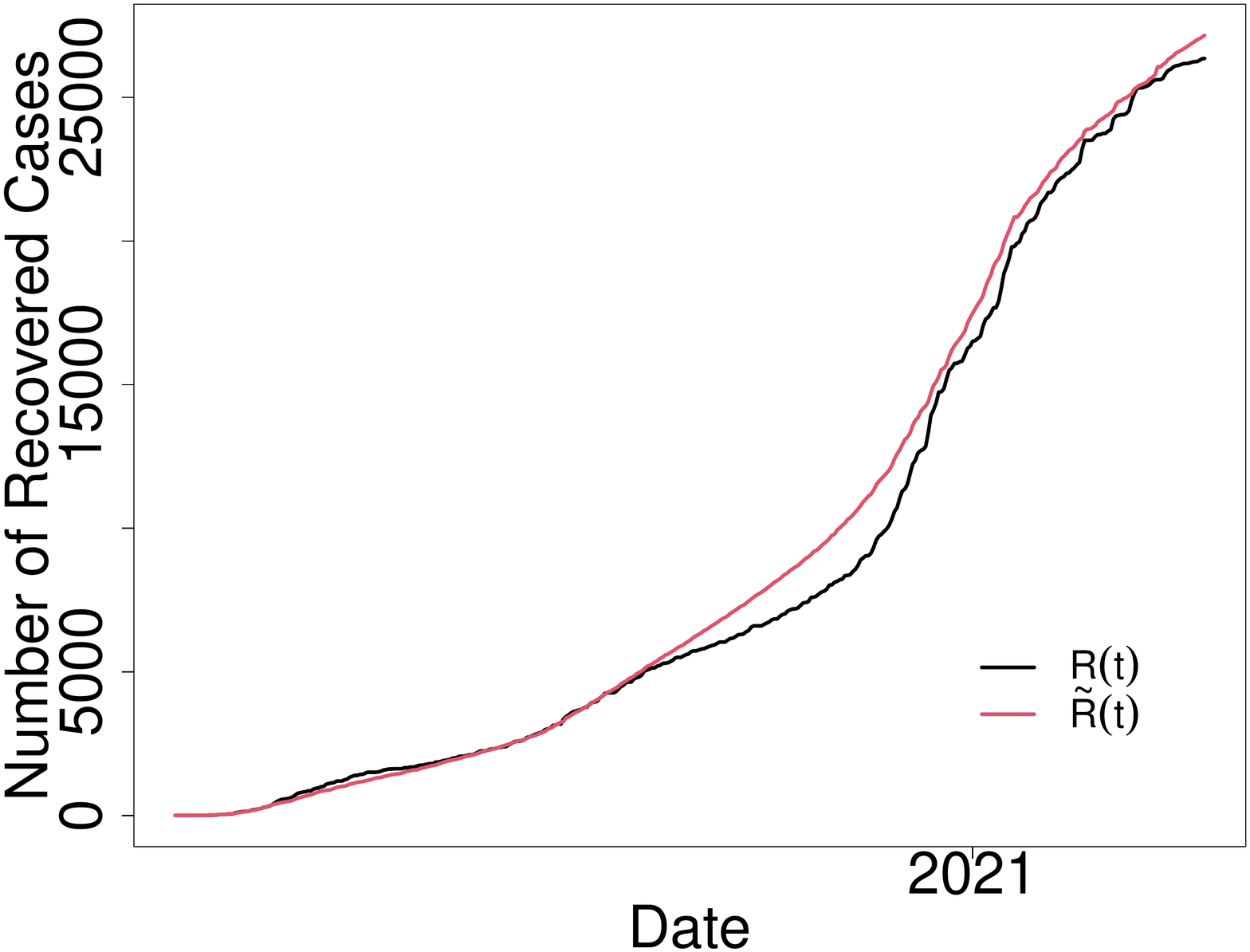}
         \subcaption{OR (Model 3)}
     \end{subfigure}
     \begin{subfigure}[b]{0.19\textwidth}
         \centering
         \includegraphics[width=\textwidth]{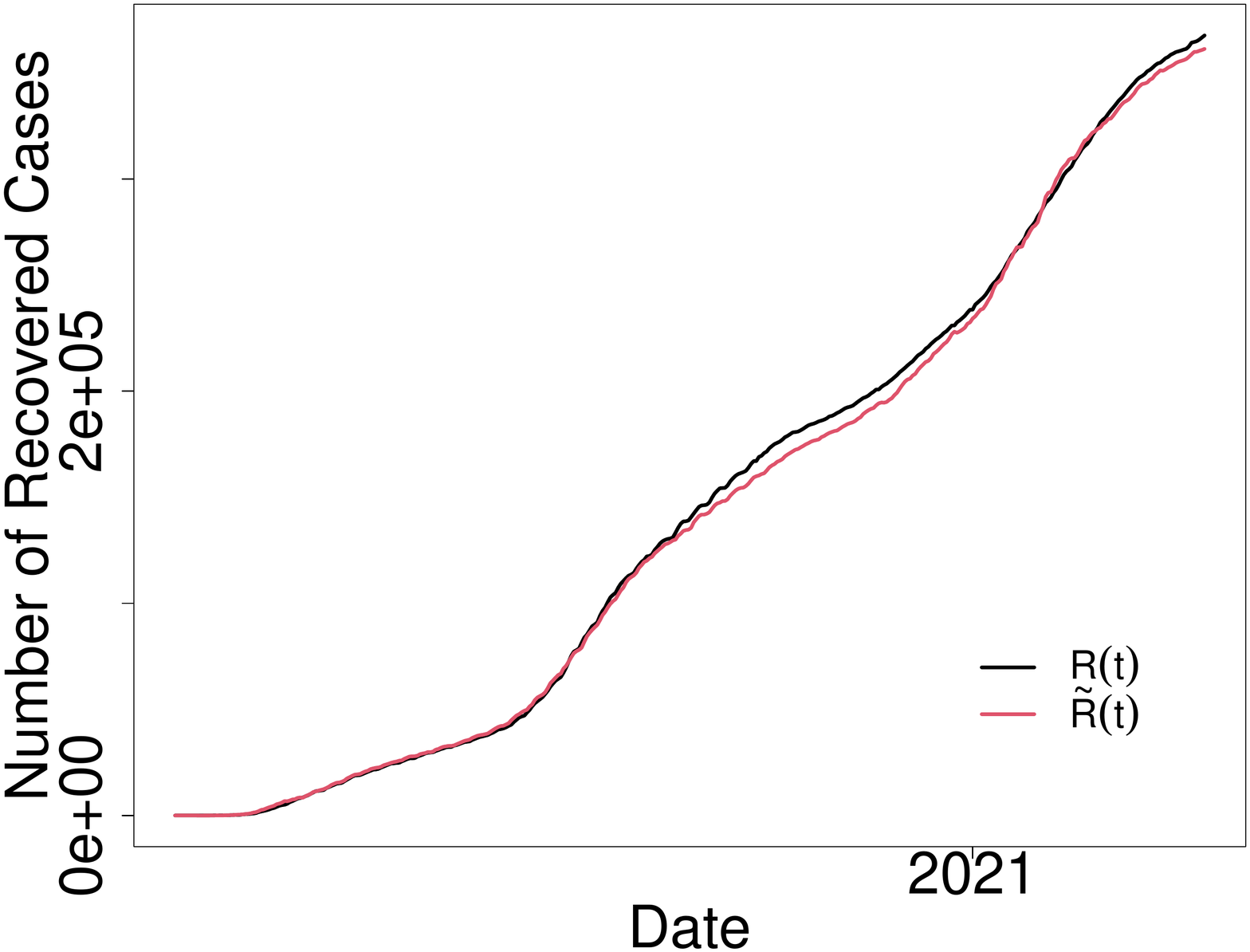}
         \subcaption{FL (Model 3)}
     \end{subfigure}
     \begin{subfigure}[b]{0.19\textwidth}
         \centering
         \includegraphics[width=\textwidth]{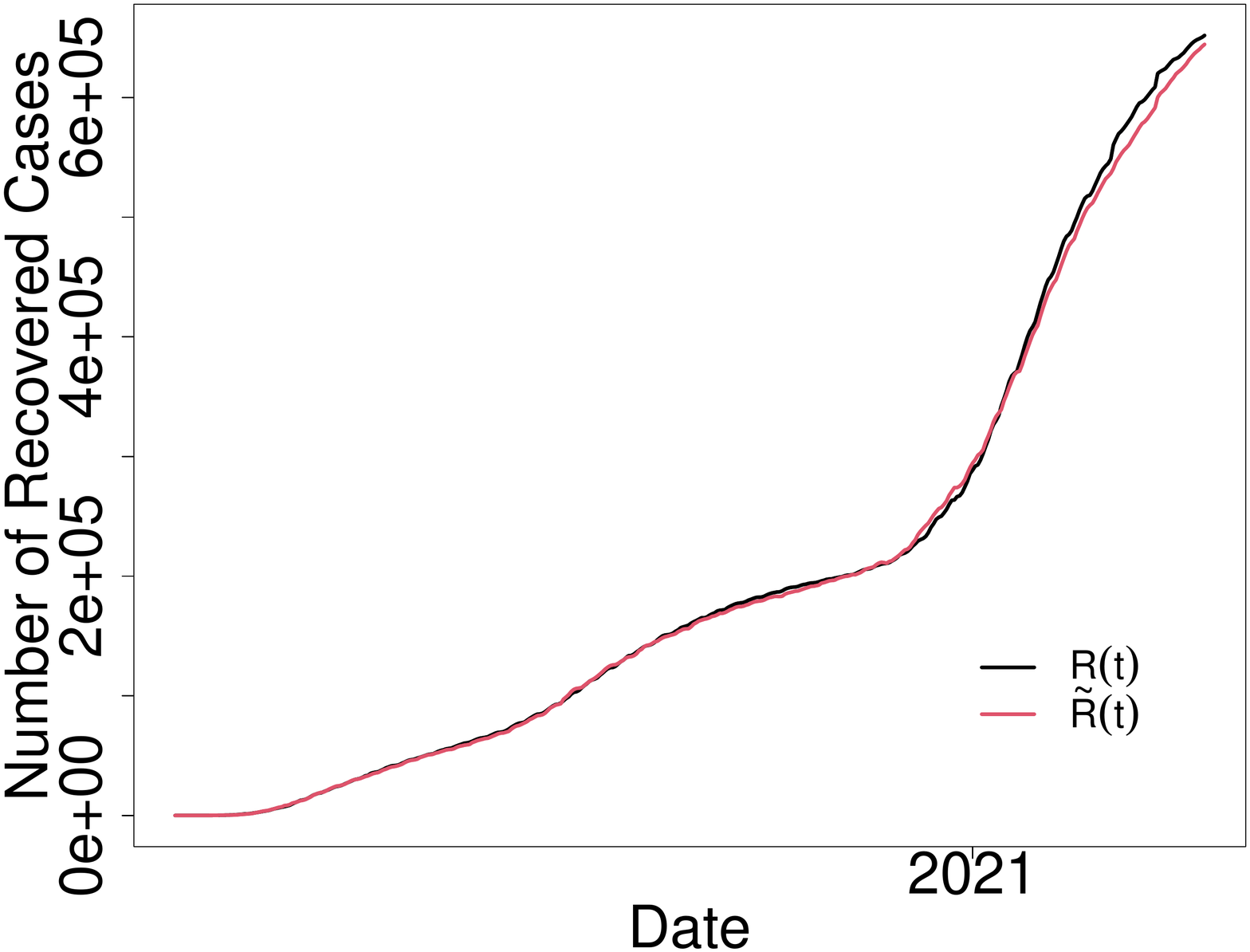}
         \subcaption{CA (Model 3)}
     \end{subfigure}
     \begin{subfigure}[b]{0.19\textwidth}
         \centering
         \includegraphics[width=\textwidth]{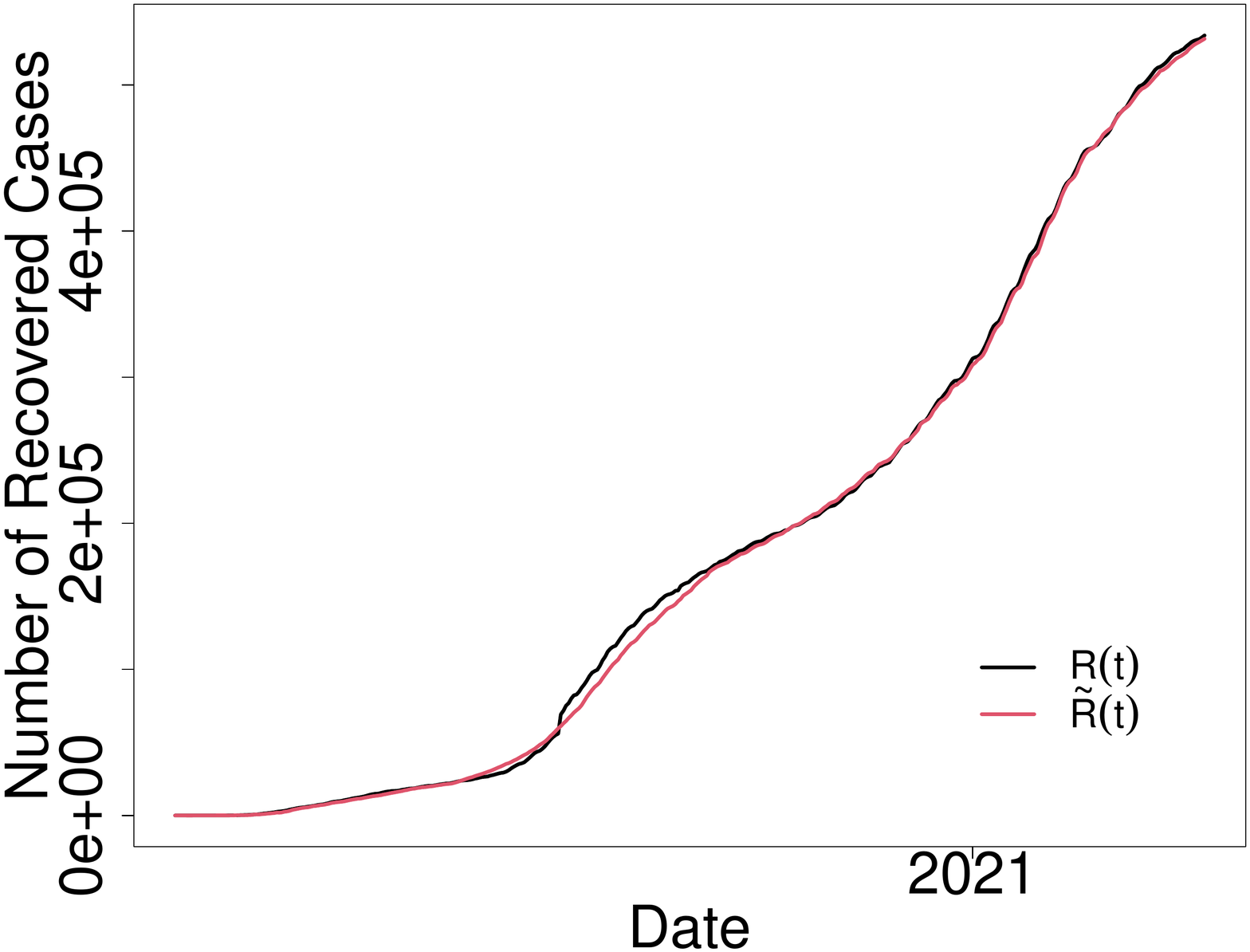}
         \subcaption{TX (Model 3)}
     \end{subfigure}
        \caption{Observed (black) and fitted (red) number of recovered cases estimated by three models in five states. }
        \label{fig:number of recovered}
\end{figure*}

\begin{table}[!ht]
\caption{\label{table_alpha} 
Spatial effect estimation in Model 2, including the estimate, p-value, and 95\% confidence intervals for the parameter $\alpha$ in five states.}
\scriptsize
\centering
\begin{tabular}{lcccc} 
  \hline
  \hline
 & Model  & Estimate & P-value & Confidence interval\\ 
   \hline
\multirow{ 4}{*}{NY} & Model 2.1 & 0.9227 & 2.3705e-50 & ( 0.8098 , 1.0355 ) \\ 
   & Model 2.2 & 0.9323 & 7.7901e-50 & ( 0.8176 , 1.047 ) \\ 
   & Model 2.3 & 0.8725 & 7.3138e-71 & ( 0.7859 , 0.9592 ) \\ 
   & Model 2.4 & 0.8355 & 8.6397e-21 & ( 0.6651 , 1.006 ) \\ 
  \multirow{ 4}{*}{OR} & Model 2.1 & 0.2576 & 1.6331e-57 & ( 0.2286 , 0.2867 ) \\ 
   & Model 2.2 & 0.2632 & 1.4148e-56 & ( 0.2332 , 0.2932 ) \\ 
   & Model 2.3 & 0.3118 & 8.9138e-31 & ( 0.261 , 0.3626 ) \\ 
   & Model 2.4 & 0.3891 & 9.2506e-85 & ( 0.3546 , 0.4235 ) \\ 
  \multirow{ 4}{*}{FL} & Model 2.1 & 0.7748 & 2.9034e-92 & ( 0.7098 , 0.8397 ) \\ 
   & Model 2.2 & 0.7525 & 8.5616e-90 & ( 0.6883 , 0.8167 ) \\ 
   & Model 2.3 & 0.8723 & 5.1120e-104 & ( 0.8048 , 0.9398 ) \\ 
   & Model 2.4 & 0.793 & 1.0635e-62 & ( 0.7081 , 0.8779 ) \\ 
  \multirow{ 4}{*}{CA} & Model 2.1 & 0.4665 & 1.2187e-61 & ( 0.4161 , 0.517 ) \\ 
   & Model 2.2 & 0.4739 & 1.0232-62 & ( 0.4232 , 0.5247 ) \\ 
   & Model 2.3 & 0.7472 & 4.7916e-88 & ( 0.6826 , 0.8117 ) \\ 
   & Model 2.4 & 0.6639 & 7.1997e-57 & ( 0.5885 , 0.7393 ) \\ 
  \multirow{ 4}{*}{TX} & Model 2.1 & 0.2127 & 6.7714e-17 & ( 0.1639 , 0.2616 ) \\ 
   & Model 2.2 & 0.2231 & 3.3484e-17 & ( 0.1724 , 0.2738 ) \\ 
   & Model 2.3 & 0.4885 & 4.6776e-36 & ( 0.4159 , 0.561 ) \\ 
   & Model 2.4 & 0.4808 & 3.3429e-29 & ( 0.4001 , 0.5615 ) \\ 
   \hline
\end{tabular}
\end{table}

In Table \ref{table_alpha},  we provide the estimate, p-value, and 95\% confidence intervals for the parameter $\alpha$ in Model 2 for all five states. 

For a stationary time series $X_t$, the auto-covariance function (ACVF) at lag $h$ is defined as
\begin{equation}
    \gamma(h) = \text{Cov}(X_t, X_{t+h}),
\end{equation}
and the auto-correlation function (ACF) at lag $h$ is defined as
\begin{equation}
    \rho(h) =\gamma(h)/ \gamma(0).
\end{equation}

Let $\widehat{\epsilon}$ and $\widetilde{\epsilon}$ be the residuals from models 2 and 3, respectively. We use auto-correlation function (ACF)  to measure the degree of dependence among the residual processes at different times. Figure \ref{fig:acf}  shows the ACF of these residual time series. As can be seen from these plots,  in all states, there is a statistically significant auto-correlation in the process $\widehat{\epsilon}$ at different times while this temporal dependence has been reduced significantly in the residuals of Model~3 which confirms the applicability of VAR modeling for the considered data sets.

\begin{figure*}[ht!]
     \centering
     \begin{subfigure}[b]{0.19\textwidth}
         \centering
         \includegraphics[width=\textwidth]{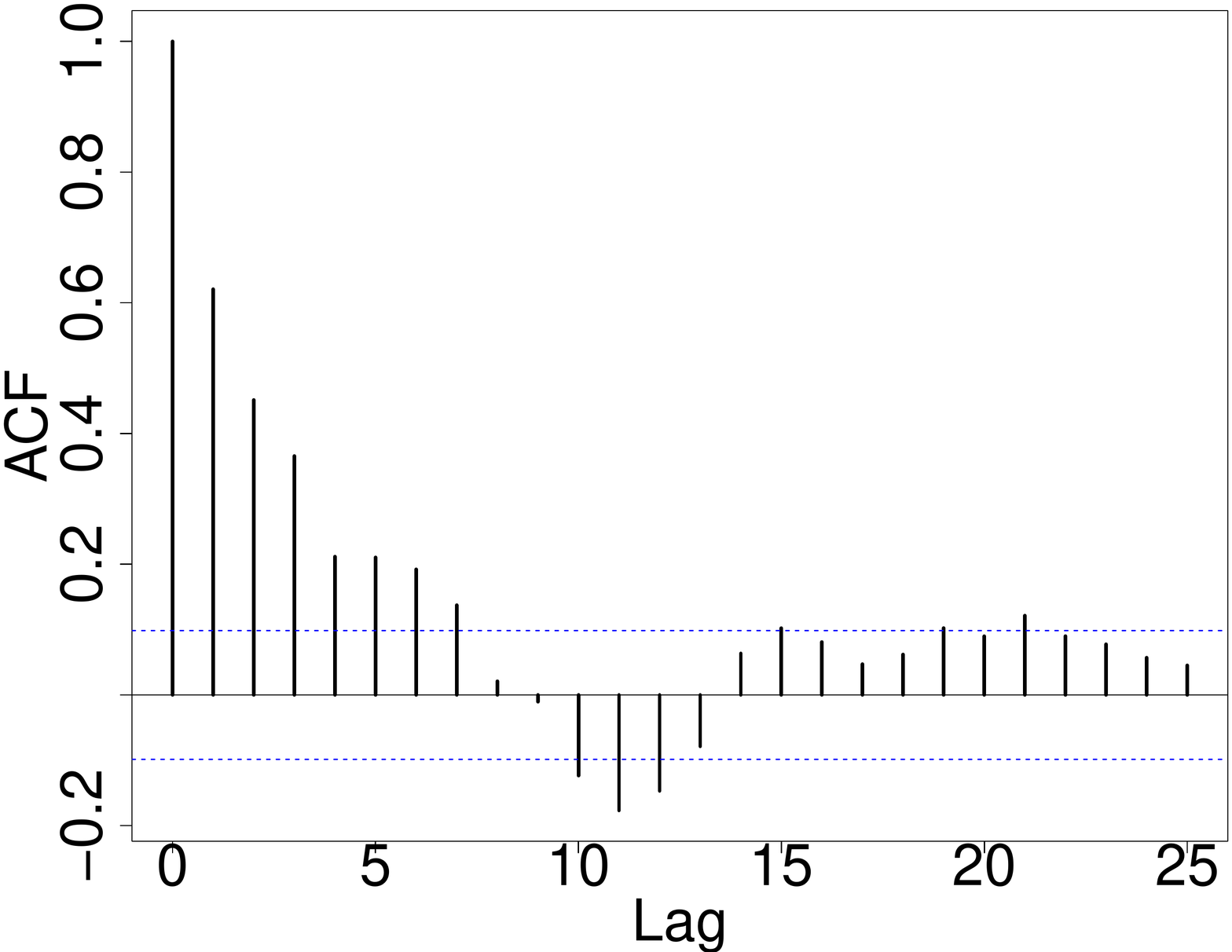}
         \subcaption{NY $\widehat{\epsilon}(t) (\Delta I)$}
     \end{subfigure}
     \begin{subfigure}[b]{0.19\textwidth}
         \centering
         \includegraphics[width=\textwidth]{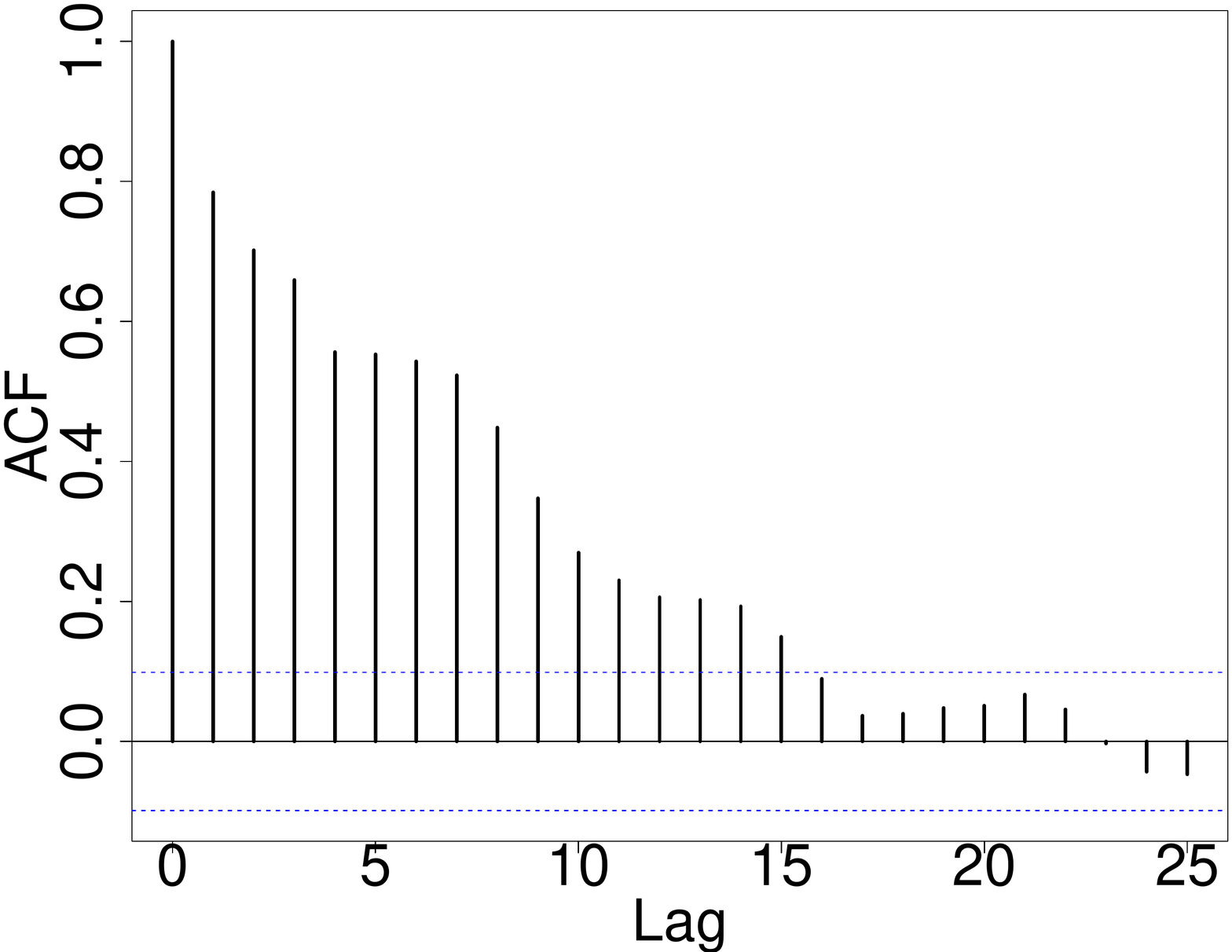}
         \subcaption{NY $\widehat{\epsilon}(t) (\Delta R)$}
     \end{subfigure}
     \begin{subfigure}[b]{0.19\textwidth}
         \centering
         \includegraphics[width=\textwidth]{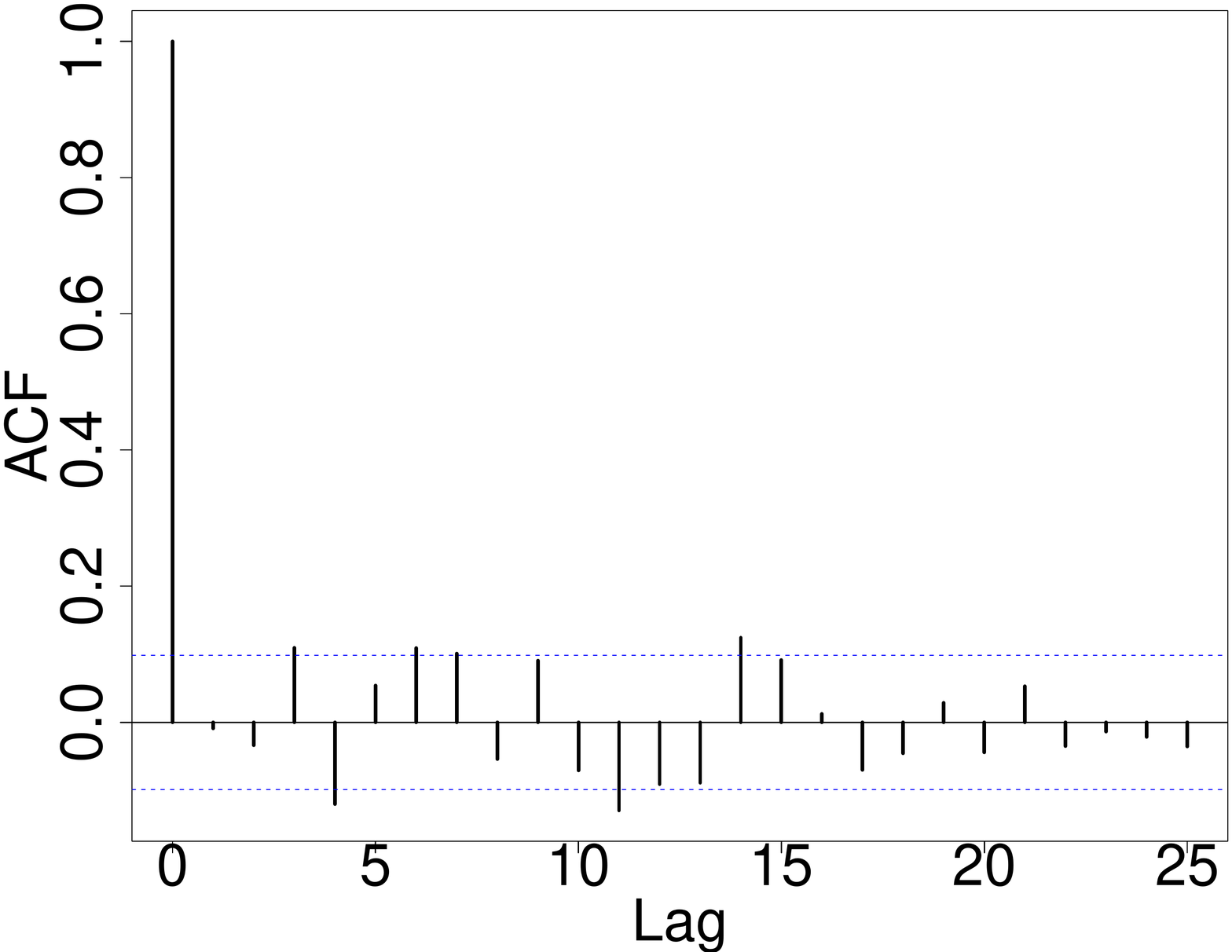}
         \subcaption{NY $\widetilde{\epsilon}(t) (\Delta I)$}
     \end{subfigure}
     \begin{subfigure}[b]{0.19\textwidth}
         \centering
         \includegraphics[width=\textwidth]{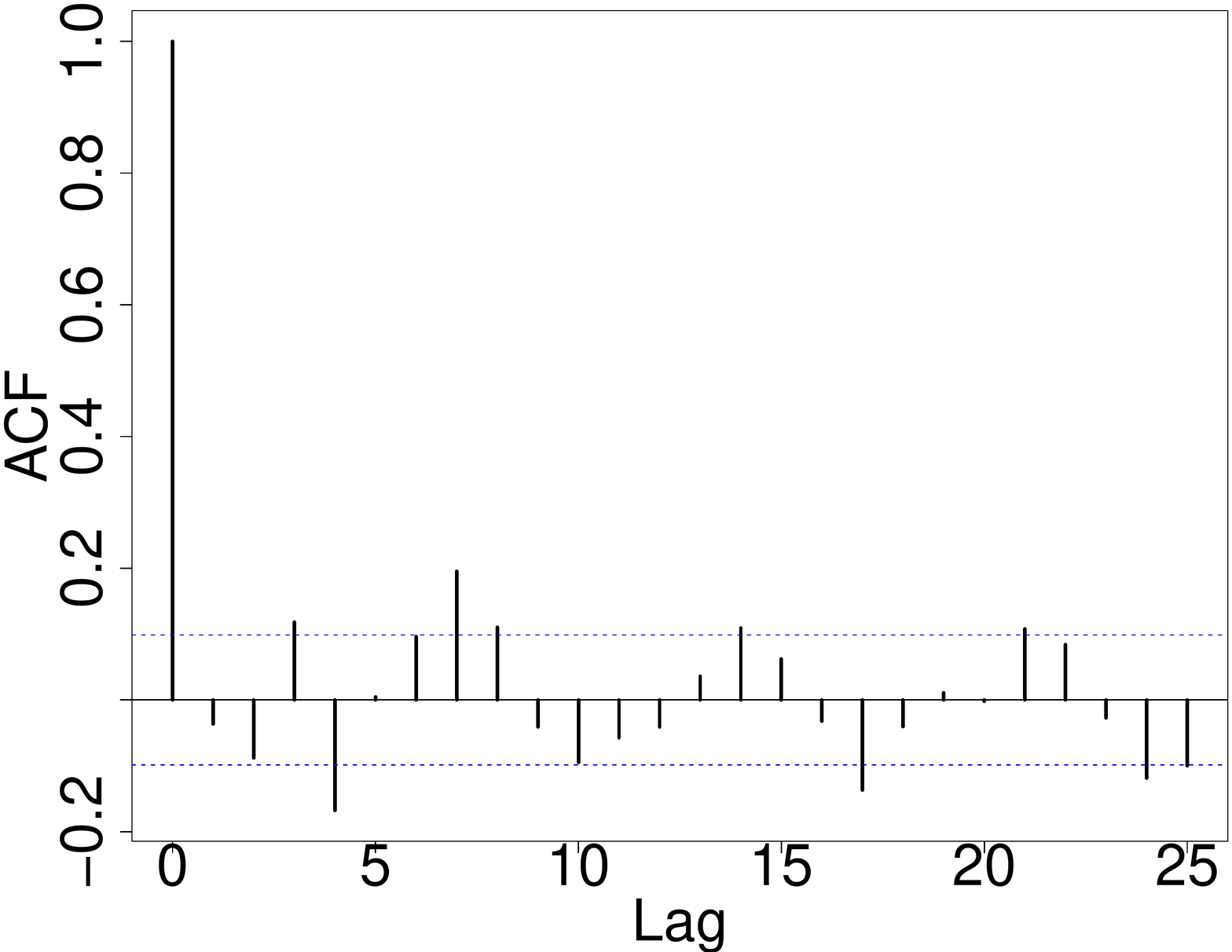}
         \subcaption{NY $\widetilde{\epsilon}(t) (\Delta R)$}
     \end{subfigure}
     
     \begin{subfigure}[b]{0.19\textwidth}
         \centering
         \includegraphics[width=\textwidth]{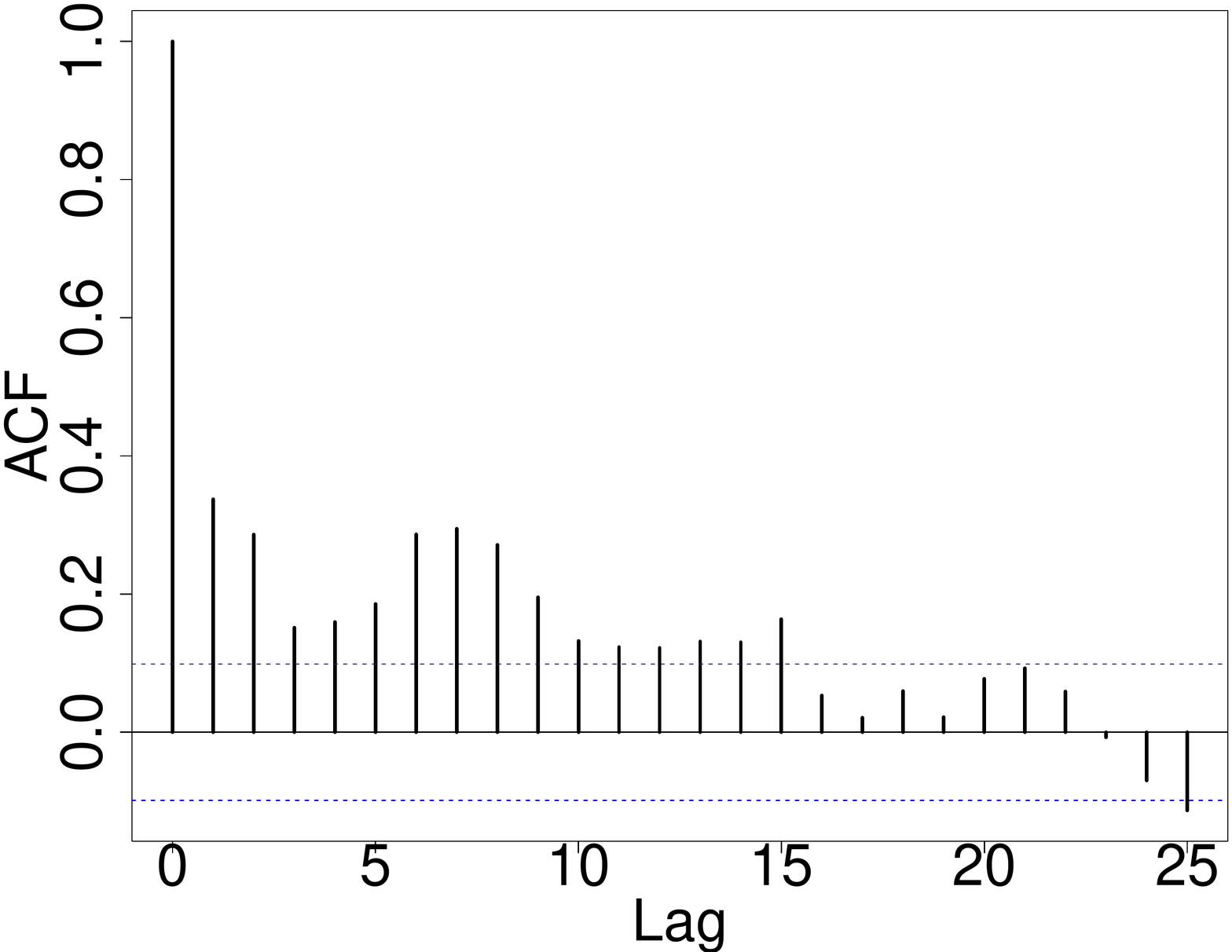}
         \subcaption{OR $\widehat{\epsilon}(t) (\Delta I)$}
     \end{subfigure}
     \begin{subfigure}[b]{0.19\textwidth}
         \centering
         \includegraphics[width=\textwidth]{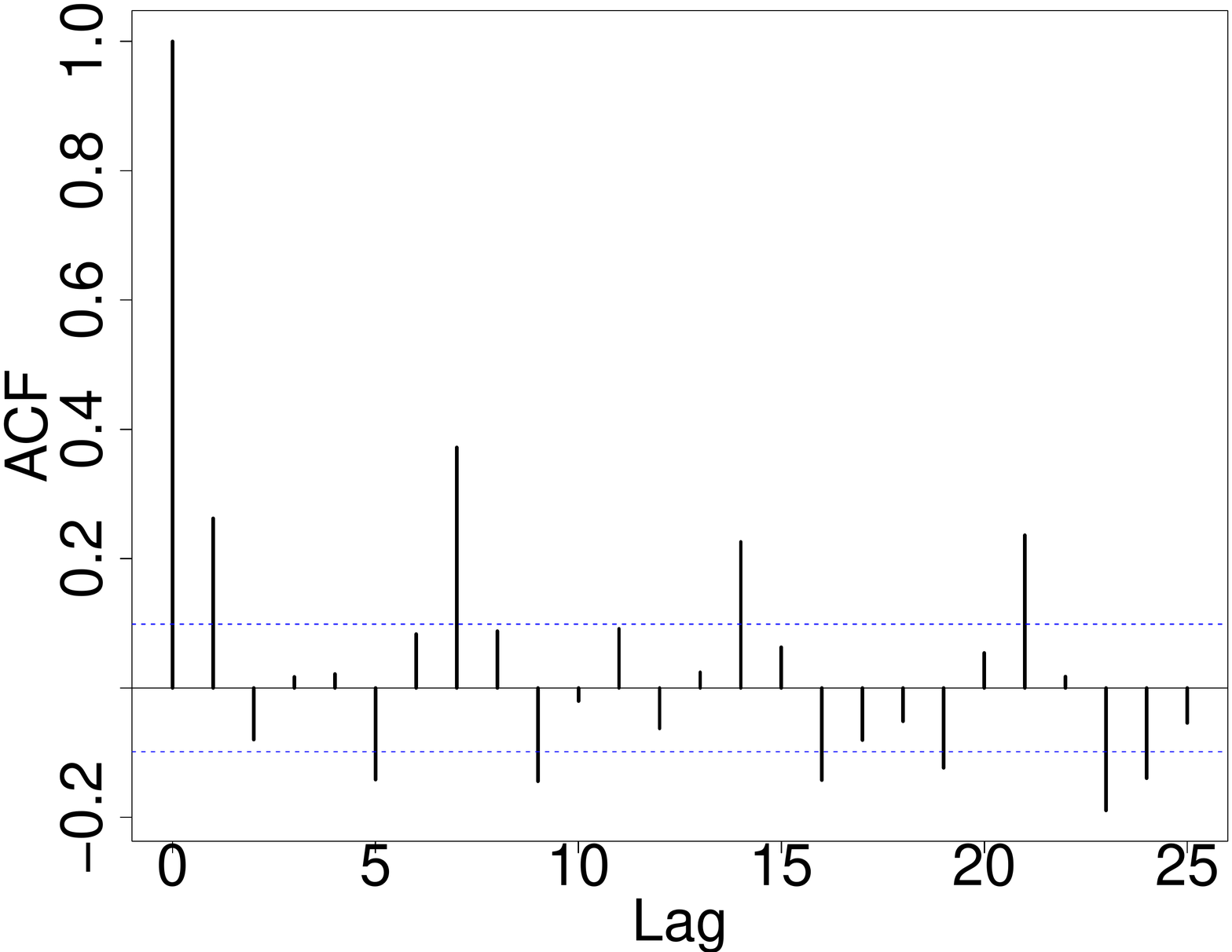}
         \subcaption{OR $\widehat{\epsilon}(t) (\Delta R)$}
     \end{subfigure}
     \begin{subfigure}[b]{0.19\textwidth}
         \centering
         \includegraphics[width=\textwidth]{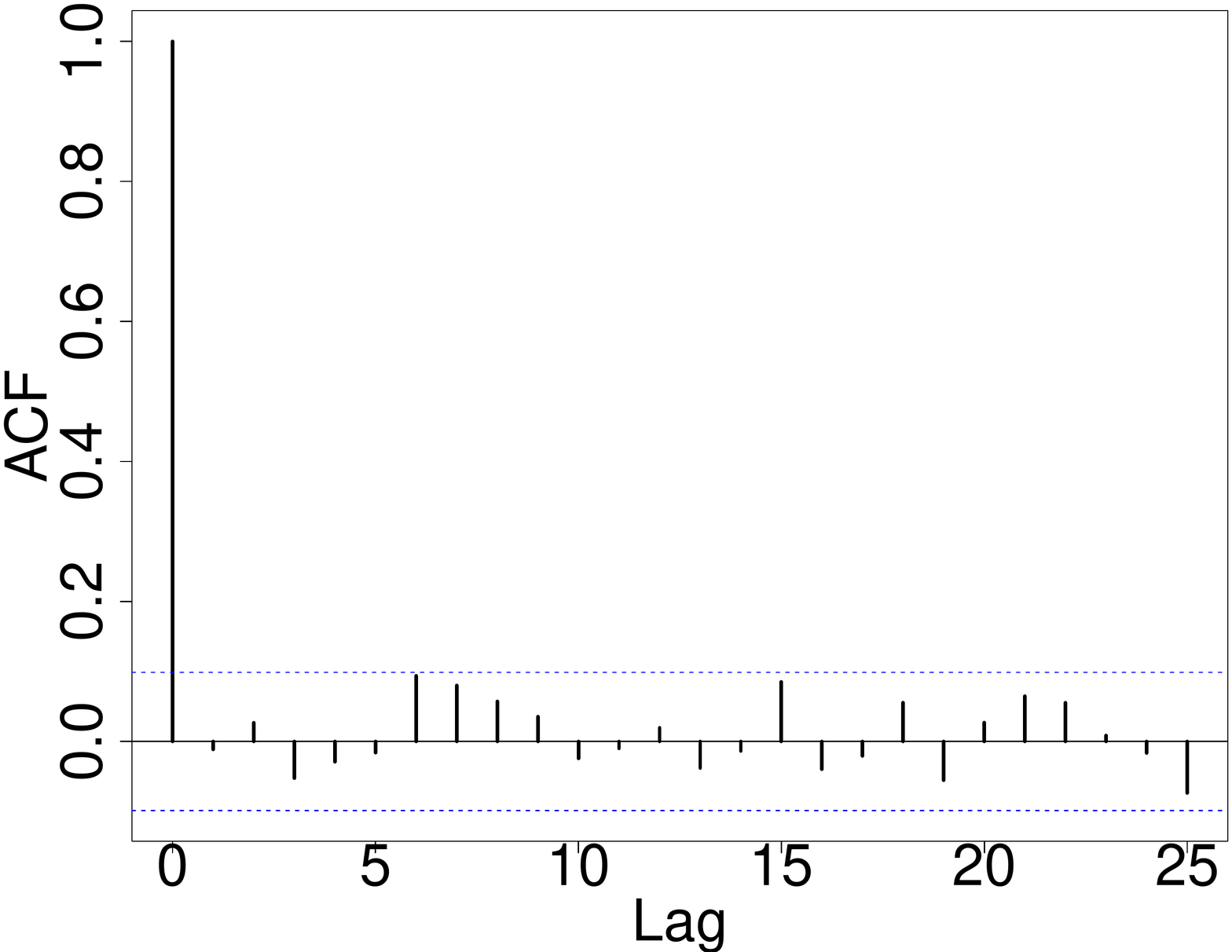}
         \subcaption{OR $\widetilde{\epsilon}(t) (\Delta I)$}
     \end{subfigure}
     \begin{subfigure}[b]{0.19\textwidth}
         \centering
         \includegraphics[width=\textwidth]{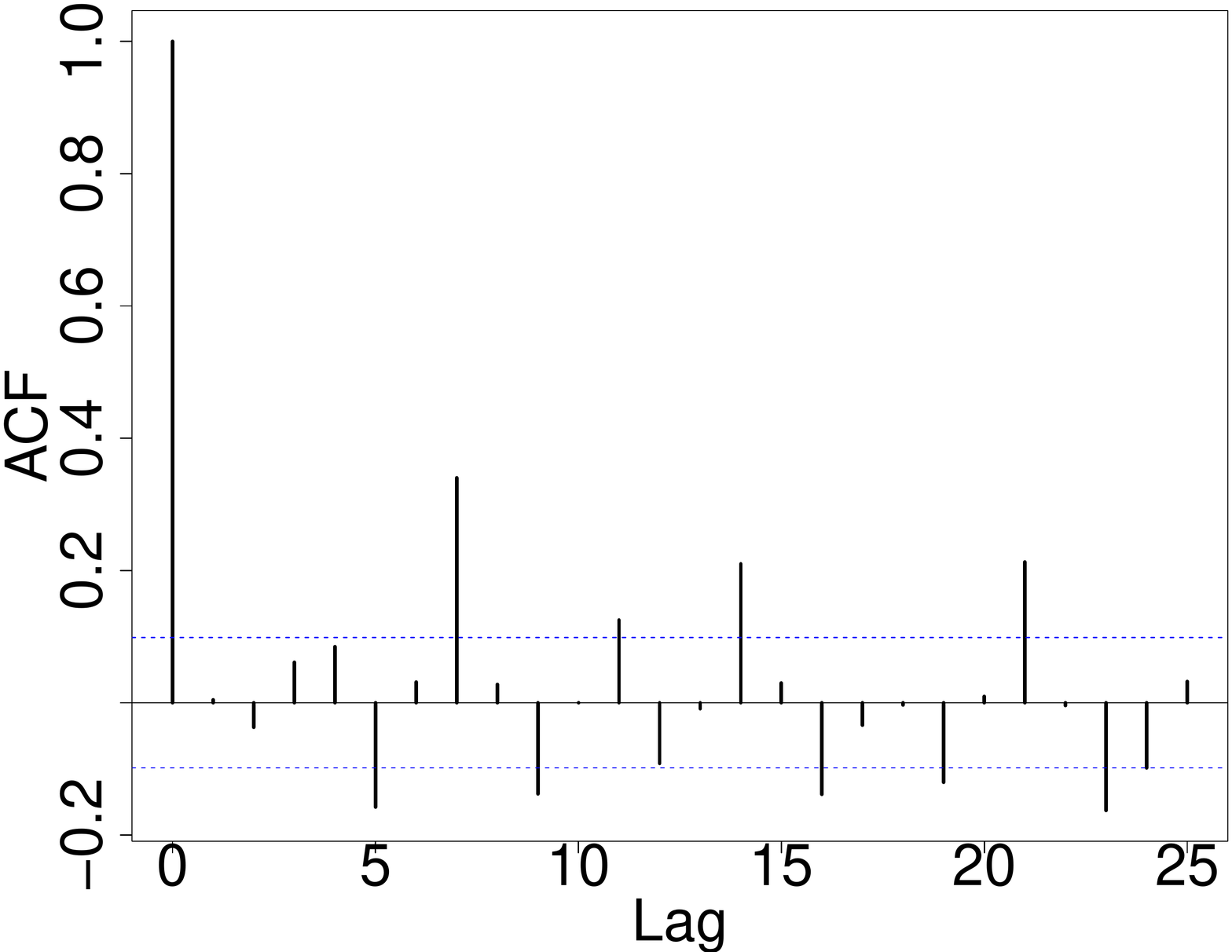}
         \subcaption{OR $\widetilde{\epsilon}(t) (\Delta R)$}
     \end{subfigure}
     
     \begin{subfigure}[b]{0.19\textwidth}
         \centering
         \includegraphics[width=\textwidth]{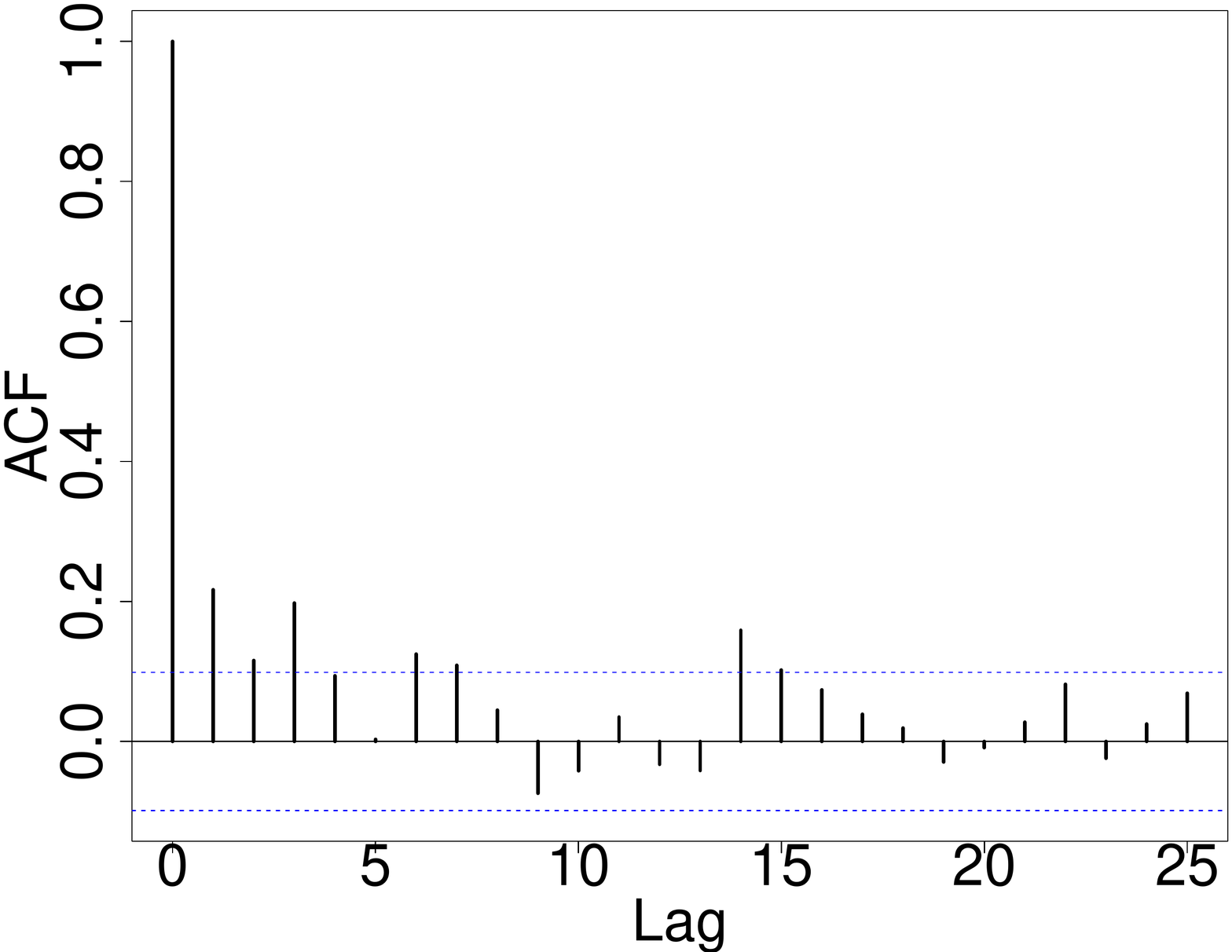}
         \subcaption{FL $\widehat{\epsilon}(t) (\Delta I)$}
     \end{subfigure}
     \begin{subfigure}[b]{0.19\textwidth}
         \centering
         \includegraphics[width=\textwidth]{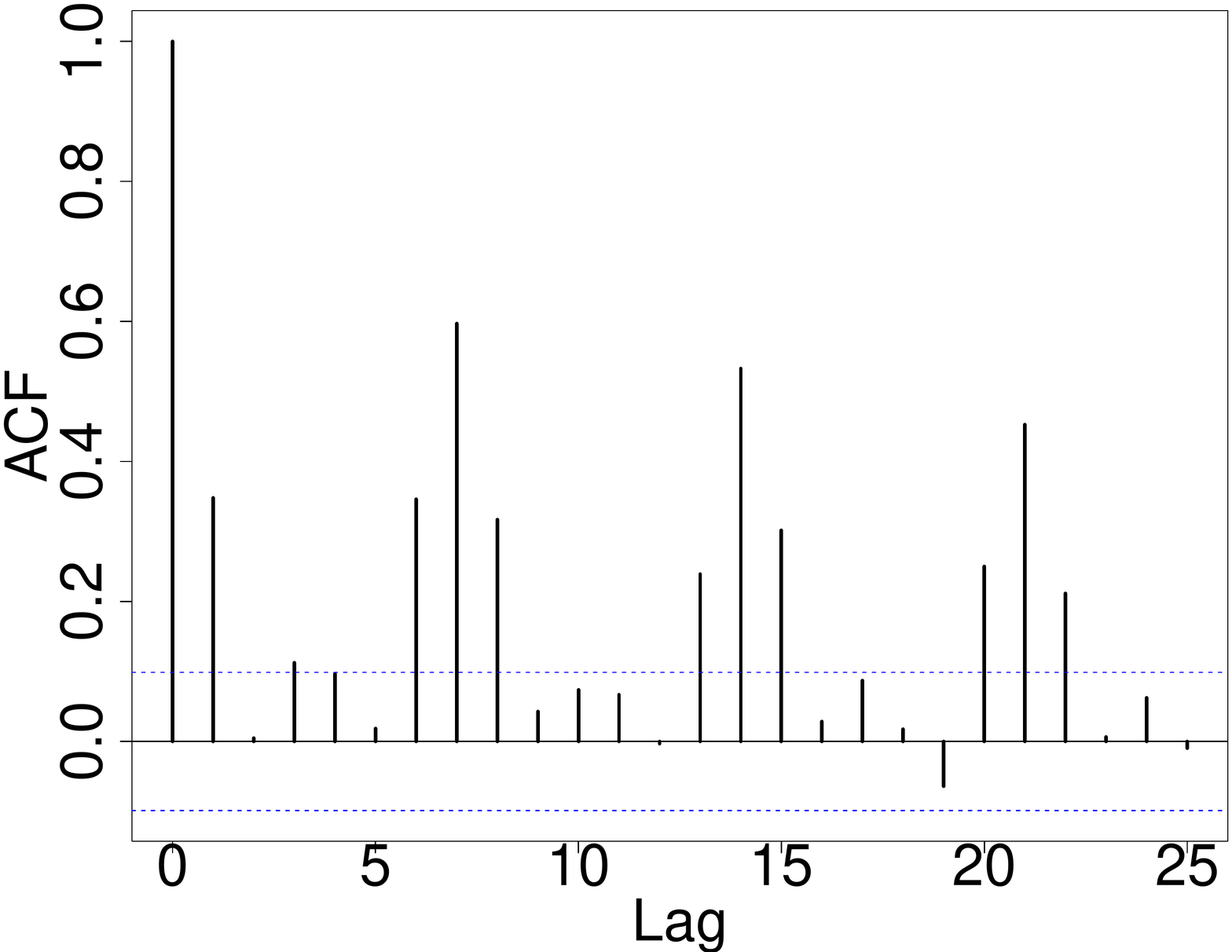}
         \subcaption{FL $\widehat{\epsilon}(t) (\Delta R)$}
     \end{subfigure}
     \begin{subfigure}[b]{0.19\textwidth}
         \centering
         \includegraphics[width=\textwidth]{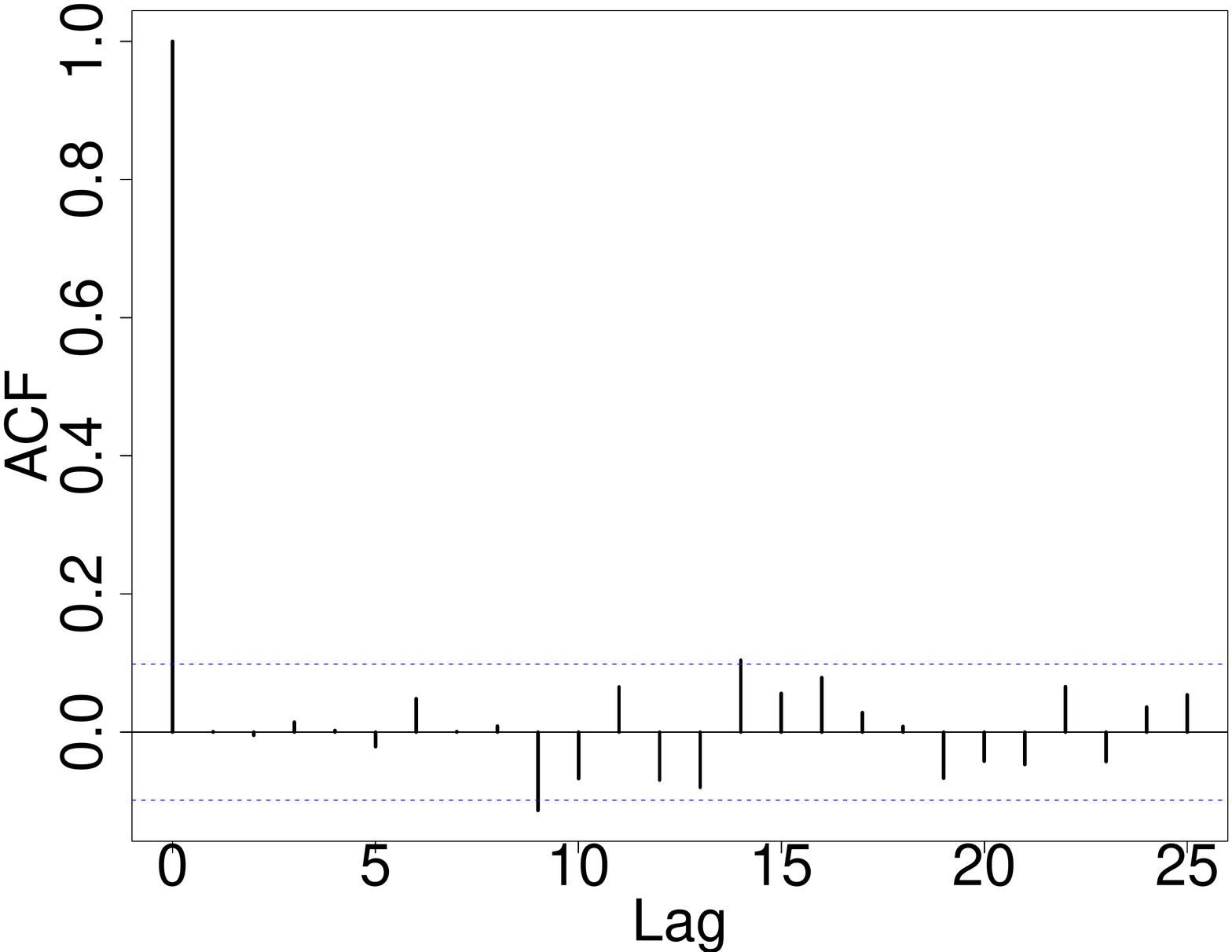}
         \subcaption{FL $\widetilde{\epsilon}(t) (\Delta I)$}
     \end{subfigure}
     \begin{subfigure}[b]{0.19\textwidth}
         \centering
         \includegraphics[width=\textwidth]{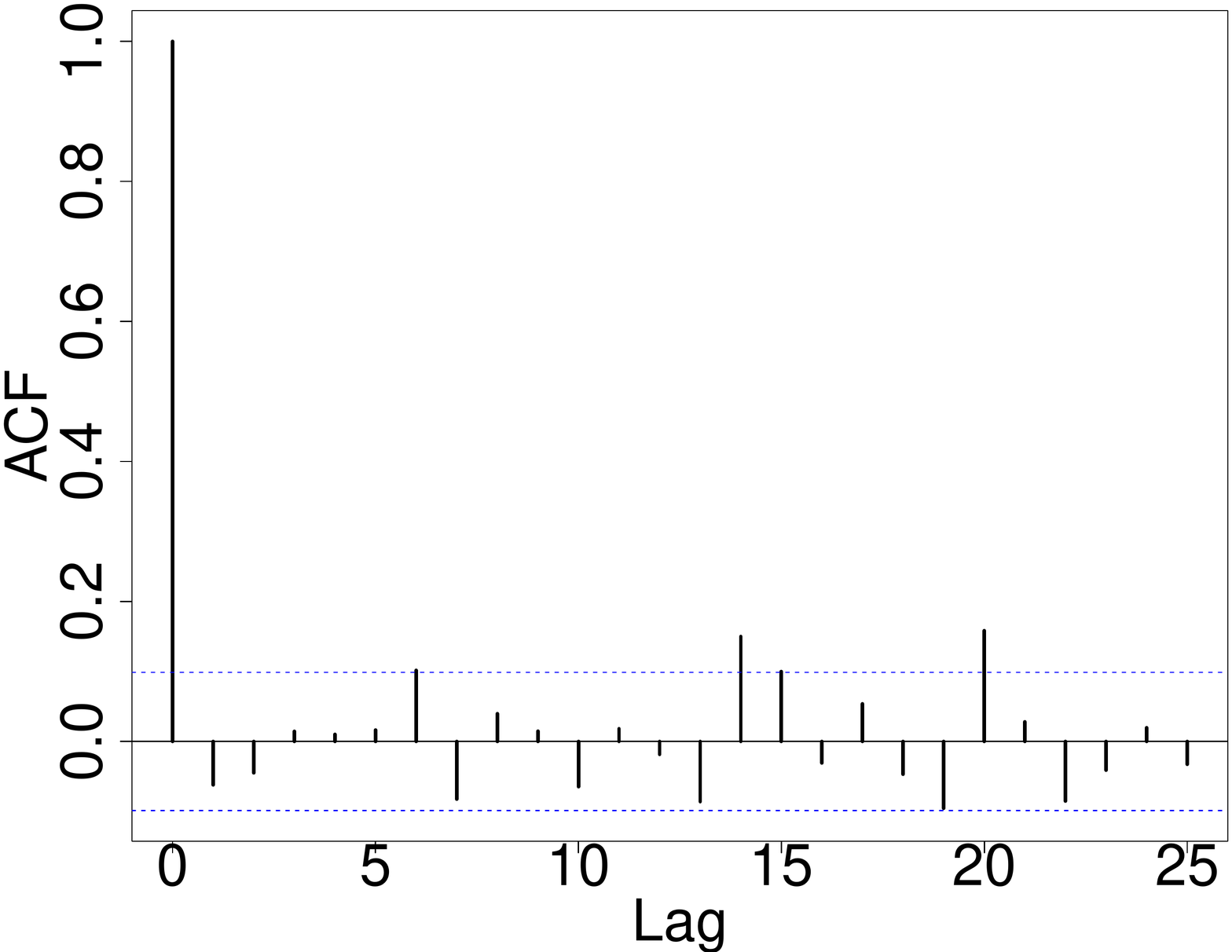}
         \subcaption{FL $\widetilde{\epsilon}(t) (\Delta R)$}
     \end{subfigure}
     
     \begin{subfigure}[b]{0.19\textwidth}
         \centering
         \includegraphics[width=\textwidth]{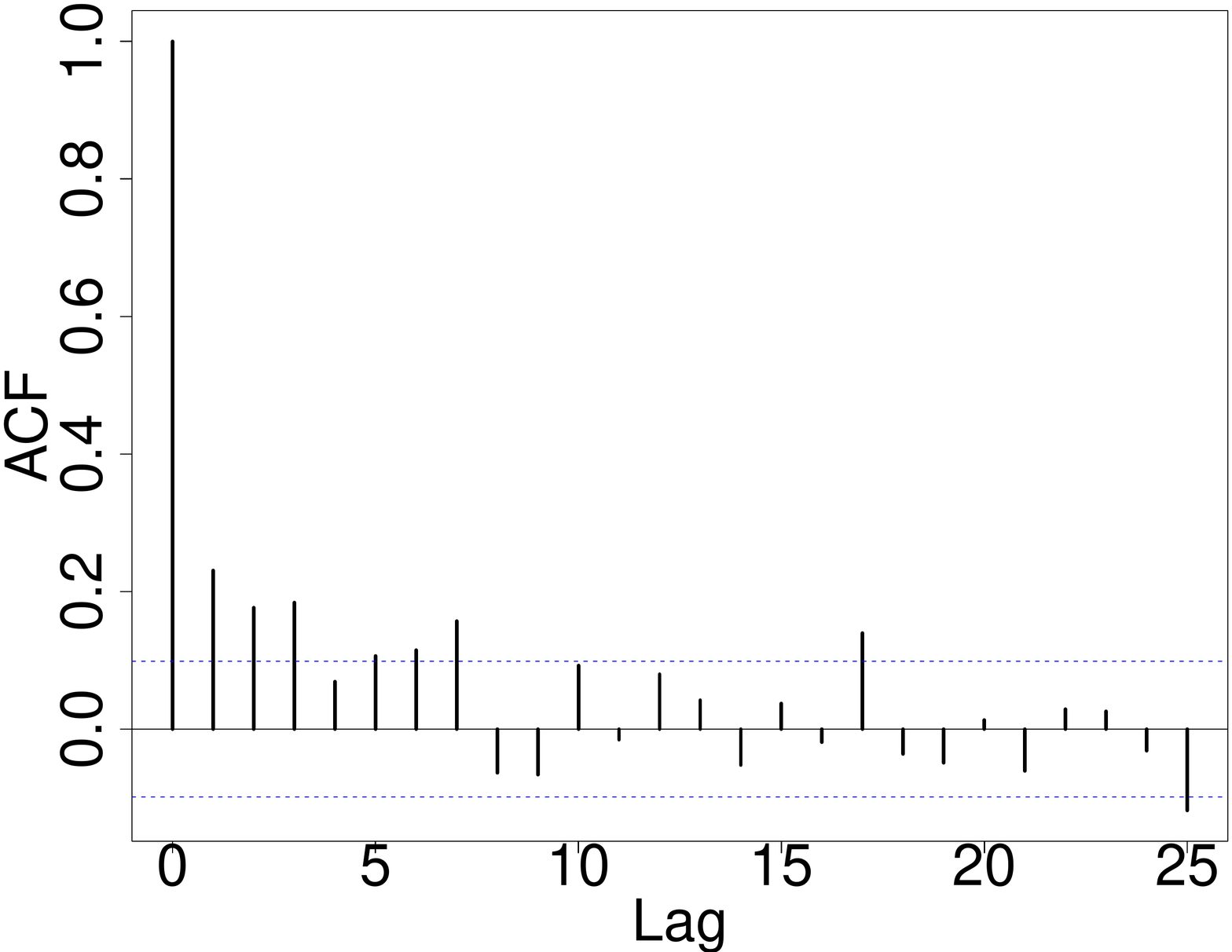}
         \subcaption{CA $\widehat{\epsilon}(t) (\Delta I)$}
     \end{subfigure}
     \begin{subfigure}[b]{0.19\textwidth}
         \centering
         \includegraphics[width=\textwidth]{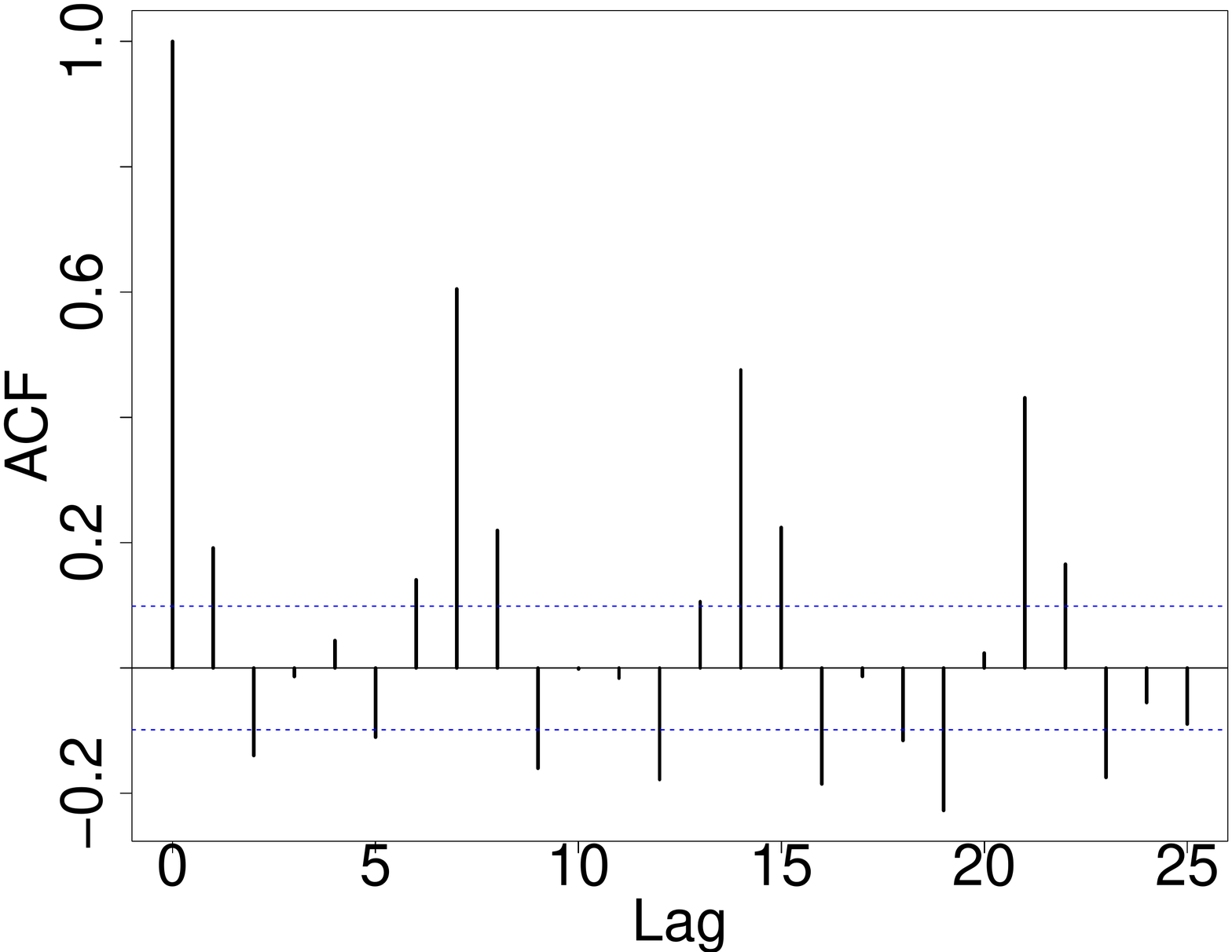}
         \subcaption{CA $\widehat{\epsilon}(t) (\Delta R)$}
     \end{subfigure}
     \begin{subfigure}[b]{0.19\textwidth}
         \centering
         \includegraphics[width=\textwidth]{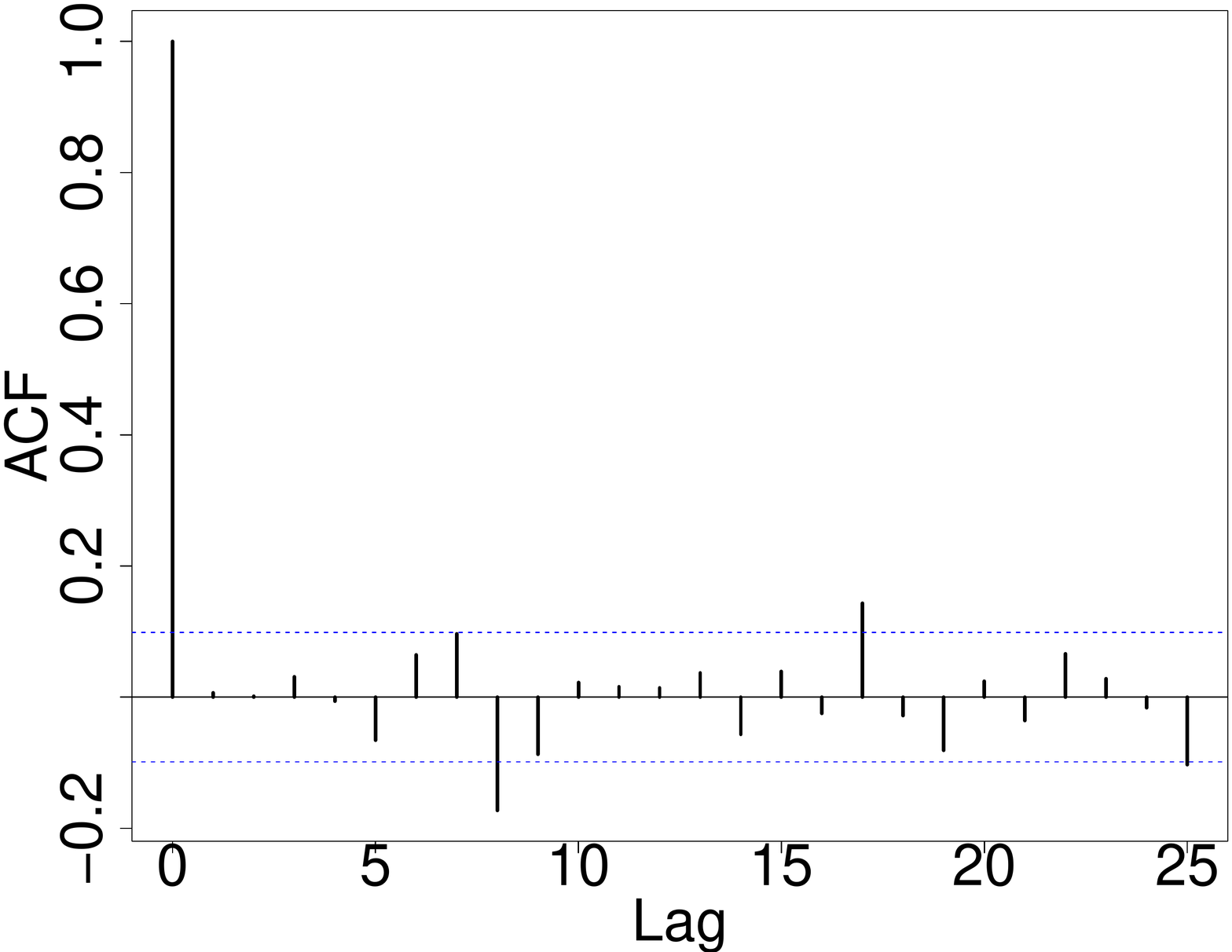}
         \subcaption{CA $\widetilde{\epsilon}(t) (\Delta I)$}
     \end{subfigure}
     \begin{subfigure}[b]{0.19\textwidth}
         \centering
         \includegraphics[width=\textwidth]{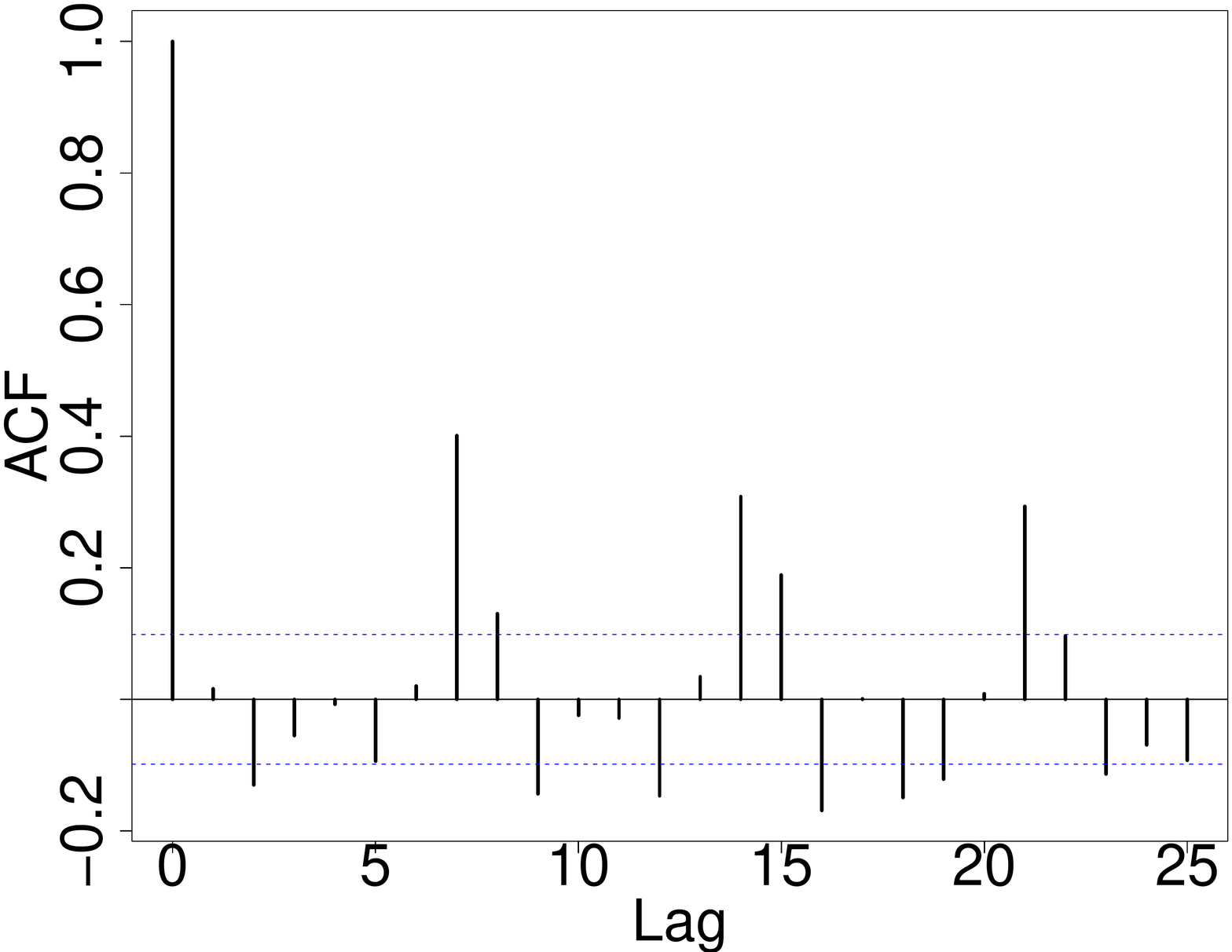}
         \subcaption{CA $\widetilde{\epsilon}(t) (\Delta R)$}
     \end{subfigure}
     
     \begin{subfigure}[b]{0.19\textwidth}
         \centering
         \includegraphics[width=\textwidth]{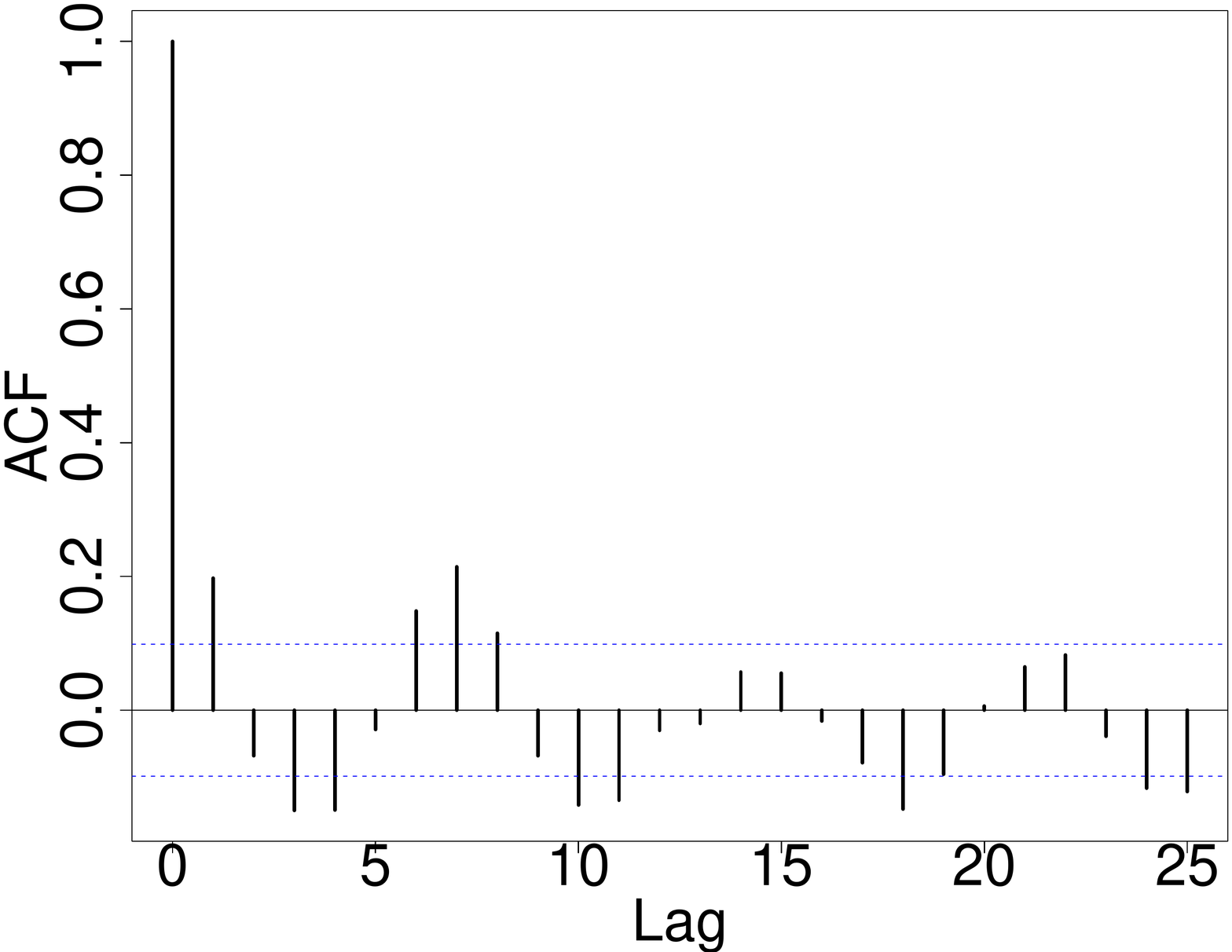}
         \subcaption{TX $\widehat{\epsilon}(t) (\Delta I)$}
     \end{subfigure}
     \begin{subfigure}[b]{0.19\textwidth}
         \centering
         \includegraphics[width=\textwidth]{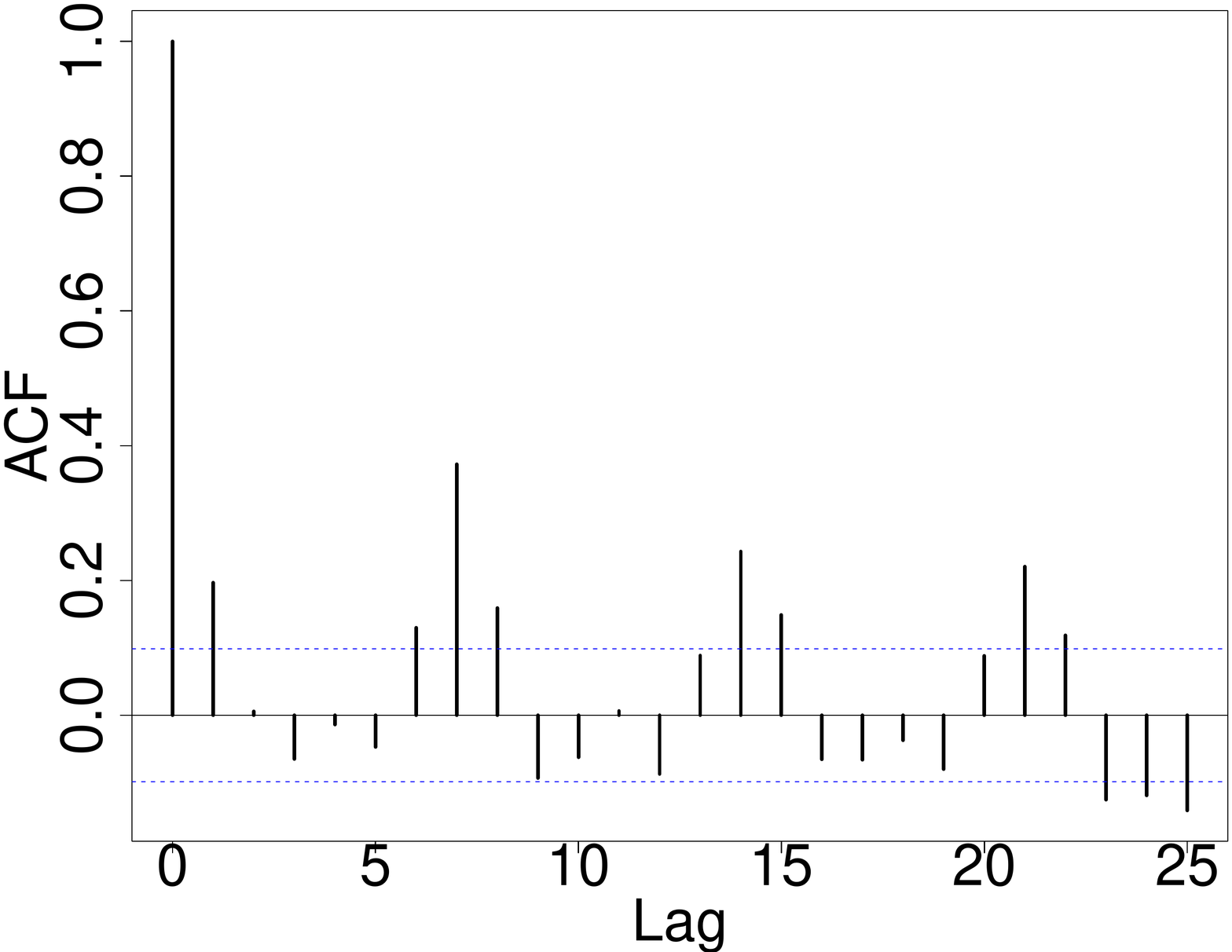}
         \subcaption{TX $\widehat{\epsilon}(t) (\Delta R)$}
     \end{subfigure}
     \begin{subfigure}[b]{0.19\textwidth}
         \centering
         \includegraphics[width=\textwidth]{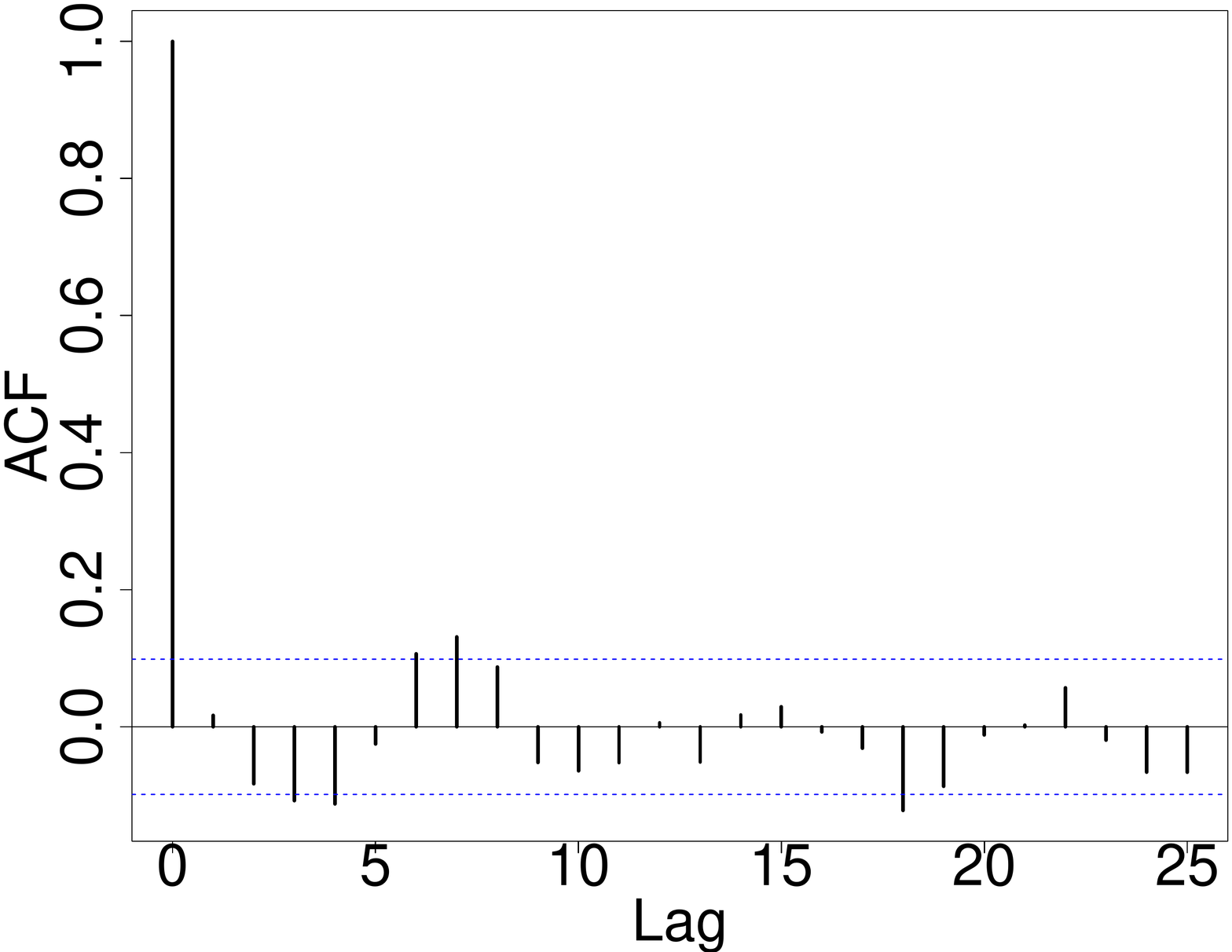}
         \subcaption{TX $\widetilde{\epsilon}(t) (\Delta I)$}
     \end{subfigure}
     \begin{subfigure}[b]{0.19\textwidth}
         \centering
         \includegraphics[width=\textwidth]{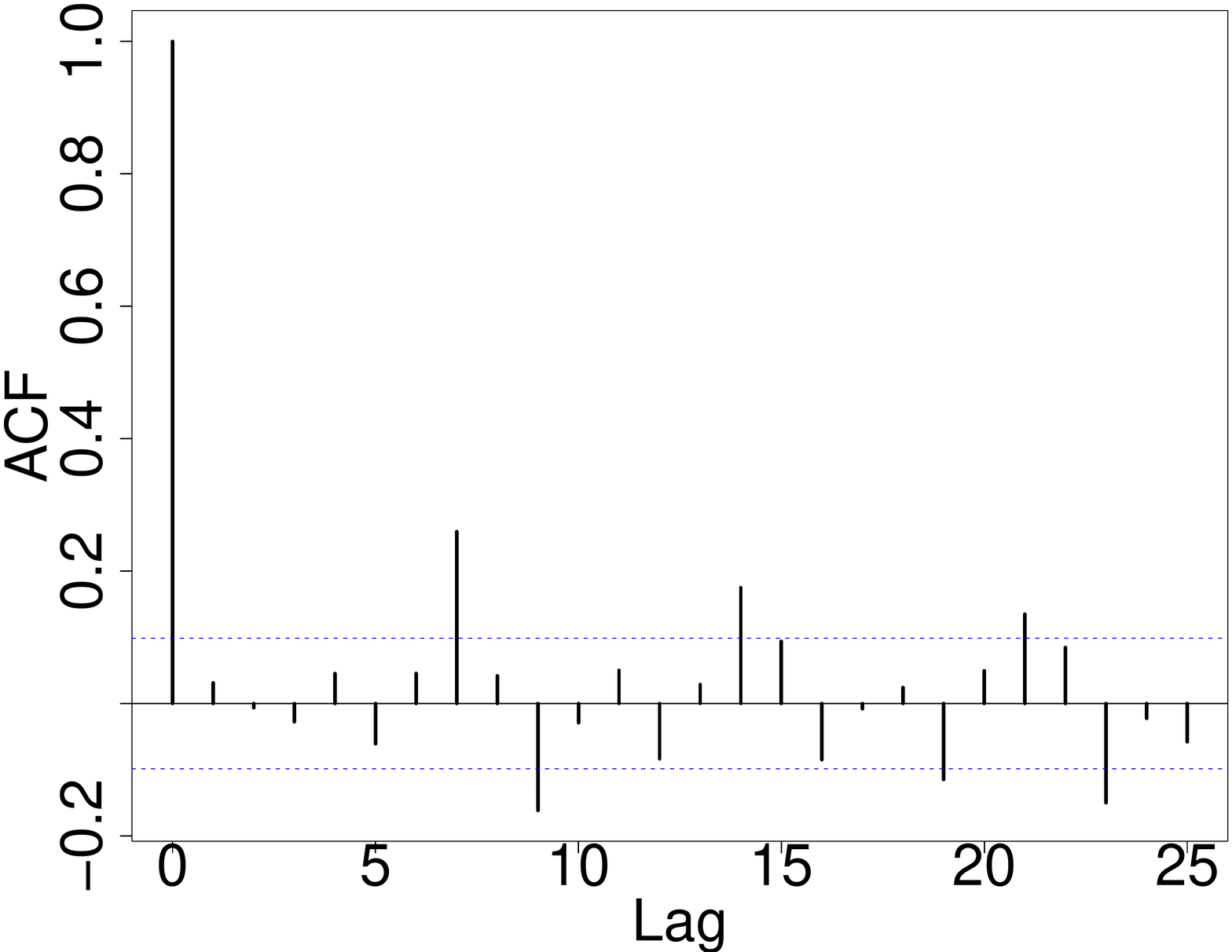}
         \subcaption{TX $\widetilde{\epsilon}(t) (\Delta R)$}
     \end{subfigure}
        \caption{Auto-correlation plot of residuals by Model 2.3 (left: piecewise constant SIR model + spatial effect ) and Model 3 (right: piecewise constant SIR model + spatial effect + VAR($p$) in five states. }
        \label{fig:acf}
\end{figure*}

The results of out-of-sample MRPE of $I(t)$ and $R(t)$  in the five states are reported in
Table \ref{table_MPE_I_out_refit} and Table \ref{table_MPE_R_out_refit}. The calculated MRPEs of $I(t)$ show that the Model 3 which includes spatial effects and VAR temporal component outperforms the other models. Spatial smoothing itself (Model 2) reduces the prediction error significantly in some states.

\begin{table}[ht!]
\caption{\label{table_MPE_I_out_refit}
Out-of-sample mean relative prediction error (MRPE)   of $I(t)$.  }
\scriptsize
\centering
{
\begin{tabular}{lccccc} 
  \hline
  \hline
   & NY & OR & FL & CA & TX\\
   &MRPE(I)  &MRPE(I) &MRPE(I) &MRPE(I) &MRPE(I)\\
  \hline
  Model 1 & 0.0011 & 9e-04 & 0.0016 & 9e-04 & 7e-04 \\ 
  Model 2.1  & 0.0011 & 0.001 & 0.0018 & 5e-04 & 8e-04 \\ 
  Model 2.2  & 0.0011 & 0.001 & 0.0017 & 6e-04 & 9e-04 \\ 
  Model 2.3  & 4e-04 & 7e-04 & 0.002 & 4e-04 & 6e-04 \\ 
  Model 2.4  & 0.0016 & 0.001 & 0.001 & 7e-04 & 8e-04 \\ 
  Model 3  & 4e-04 & 7e-04 & 0.001 & 4e-04 & 6e-04 \\ 
  \hline
\end{tabular}}
\end{table}

\begin{table}[!ht]
\caption{\label{table_MPE_R_out_refit}
Out-of-sample mean relative prediction error (MRPE)   of $R(t)$. }
\centering
{
\begin{tabular}{lccccc} 
  \hline
  \hline
& NY & OR & FL & CA & TX\\
    &MRPE(R)  &MRPE(R) &MRPE(R) &MRPE(R) &MRPE(R)\\
  \hline
  Model 1  & 0.0019 & 0.0028 & 0.0028 & 0.0042 & 0.002 \\ 
  Model 2.1  & 4e-04 & 0.0017 & 0.001 & 0.0029 & 0.0025 \\ 
  Model 2.2  & 4e-04 & 0.0018 & 0.001 & 0.0029 & 0.0028 \\ 
  Model 2.3  & 4e-04 & 0.0022 & 0.0012 & 0.0022 & 0.001 \\ 
  Model 2.4  & 5e-04 & 0.002 & 9e-04 & 0.0029 & 0.0015 \\ 
  Model 3  & 4e-04 & 0.0022 & 0.001 & 0.0022 & 0.001 \\ 
  \hline
\end{tabular}}
\end{table}

\begin{figure*}[ht!]
     \centering
        \captionsetup[sub]{font=scriptsize, labelfont={bf,sf}}
     \begin{subfigure}[b]{0.16\textwidth}
         \centering
         \includegraphics[width=\textwidth]{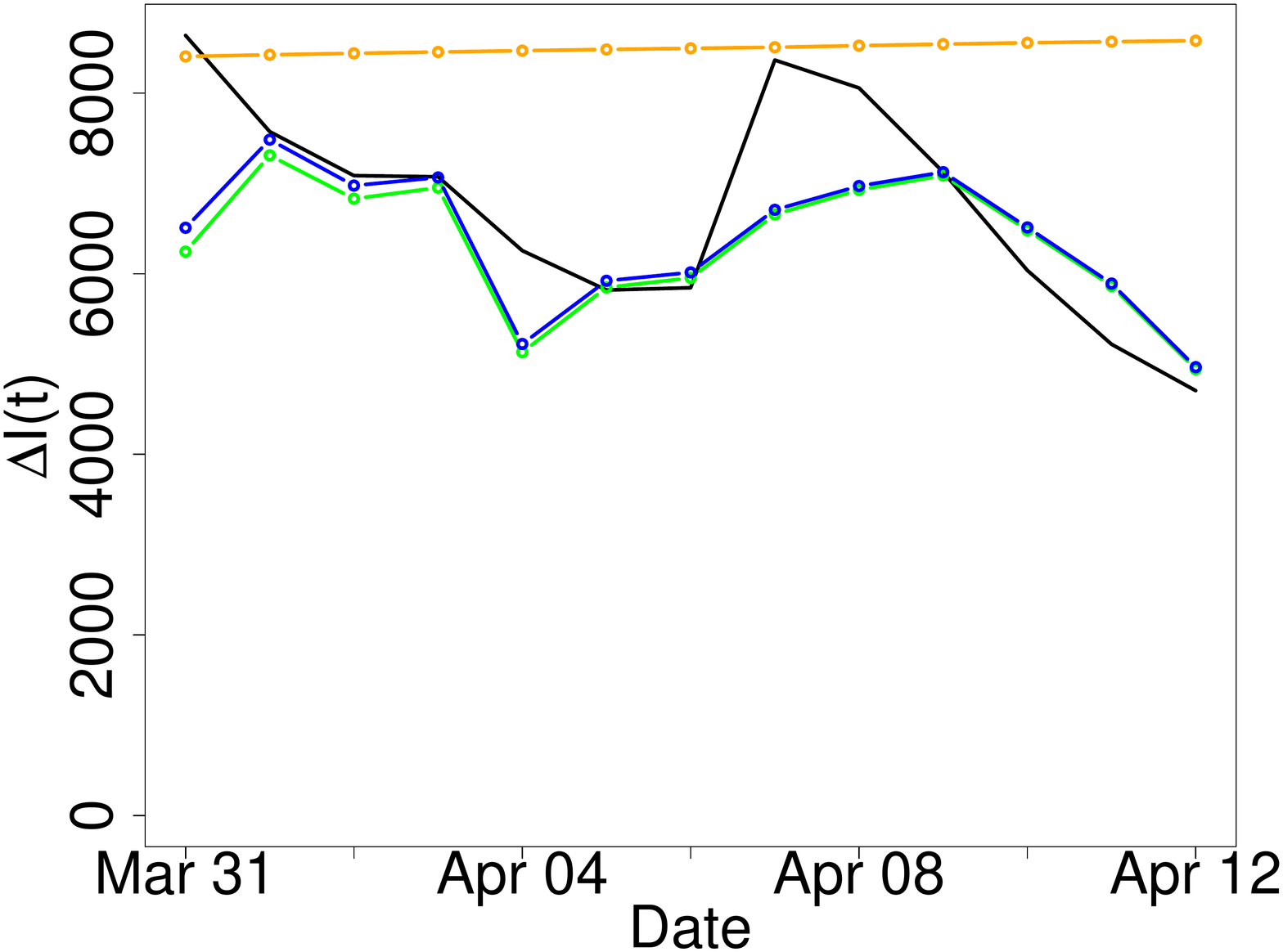}
         \subcaption{NY $\widehat{\Delta I}(t)$}
     \end{subfigure}
     \begin{subfigure}[b]{0.16\textwidth}
         \centering
         \includegraphics[width=\textwidth]{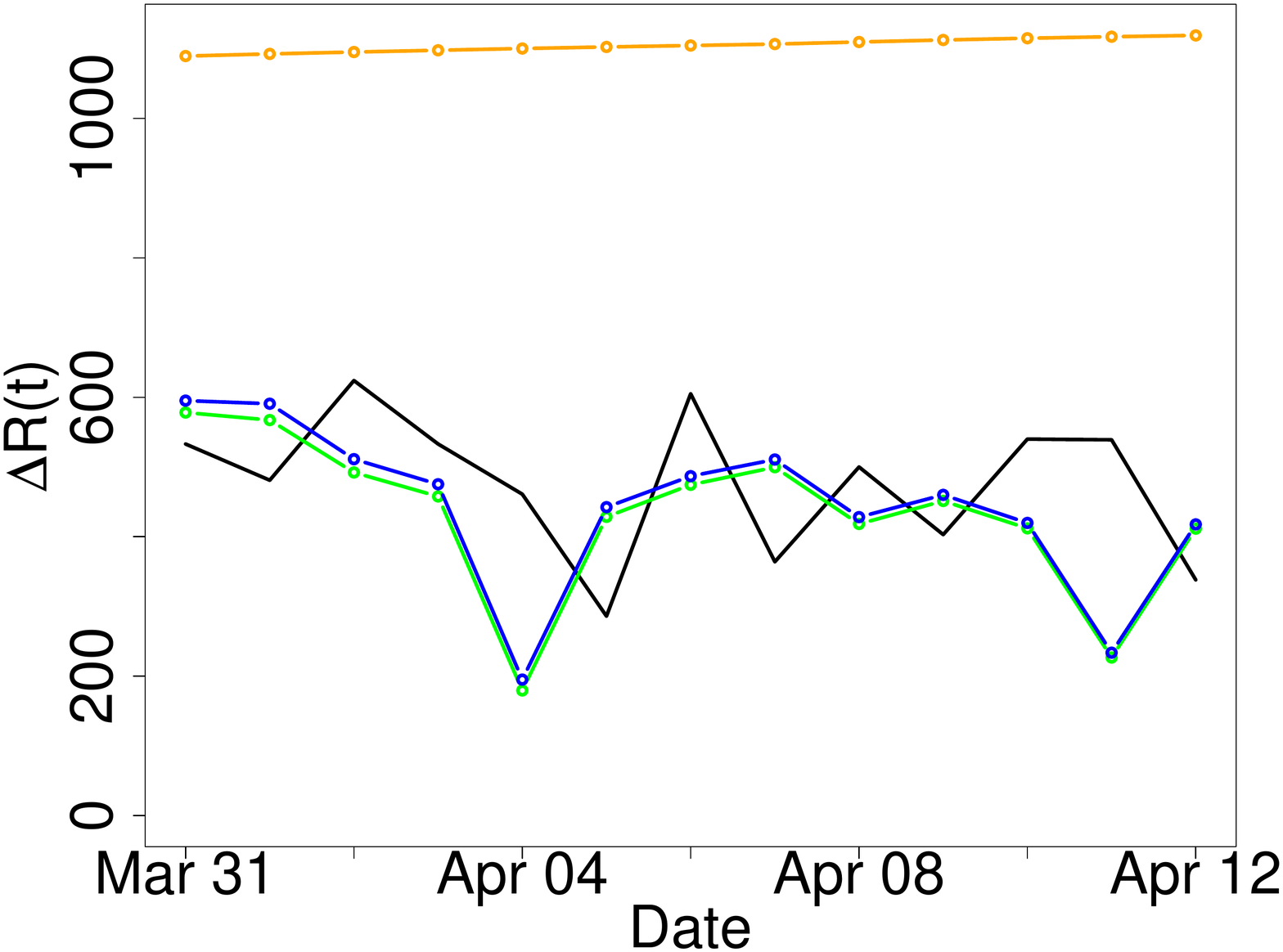}
         \subcaption{NY $\widehat{\Delta R}(t)$}
     \end{subfigure}
     \begin{subfigure}[b]{0.16\textwidth}
         \centering
         \includegraphics[width=\textwidth]{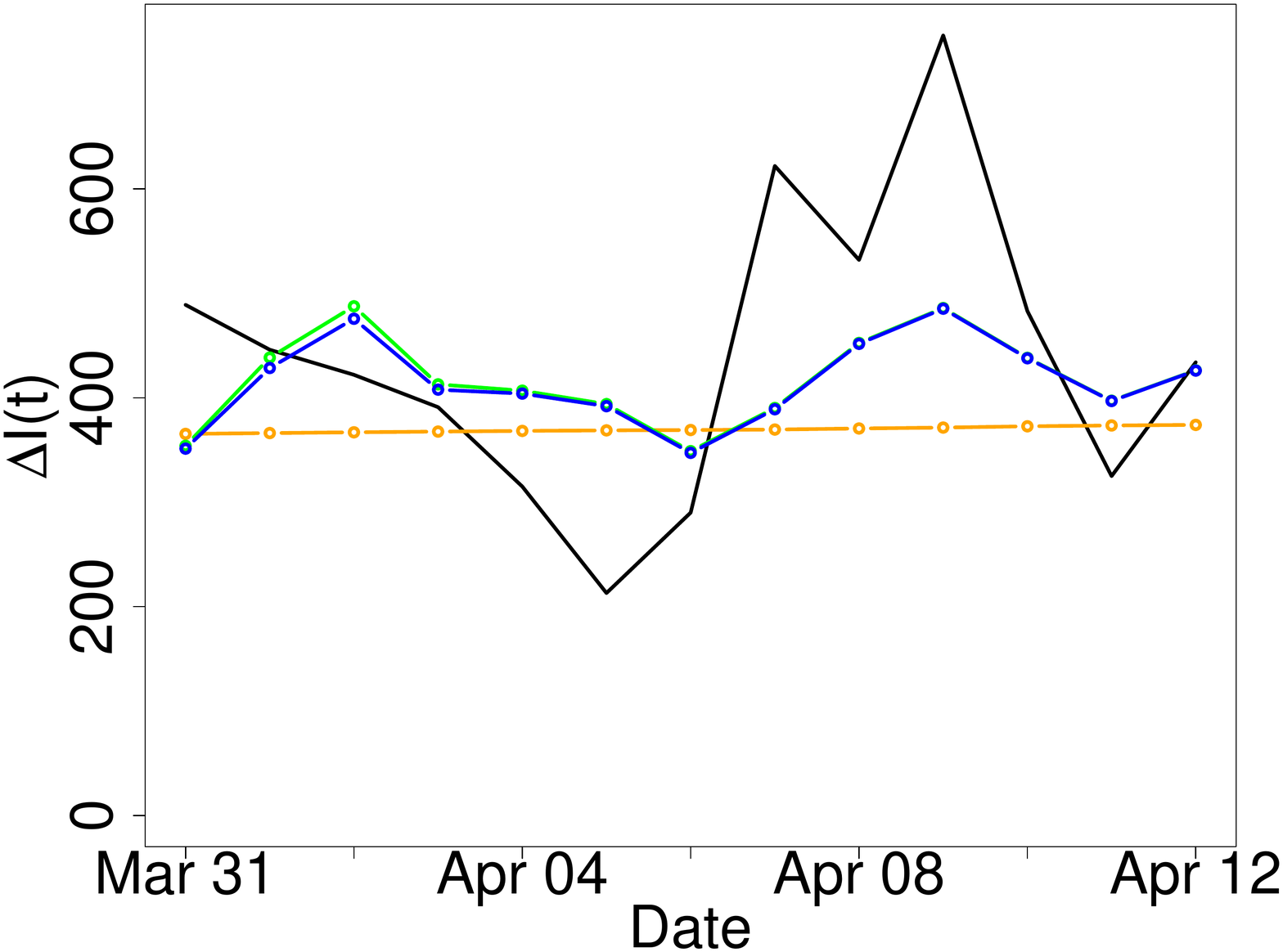}
         \subcaption{OR $\widehat{\Delta I}(t)$}
     \end{subfigure}
     \begin{subfigure}[b]{0.16\textwidth}
         \centering
         \includegraphics[width=\textwidth]{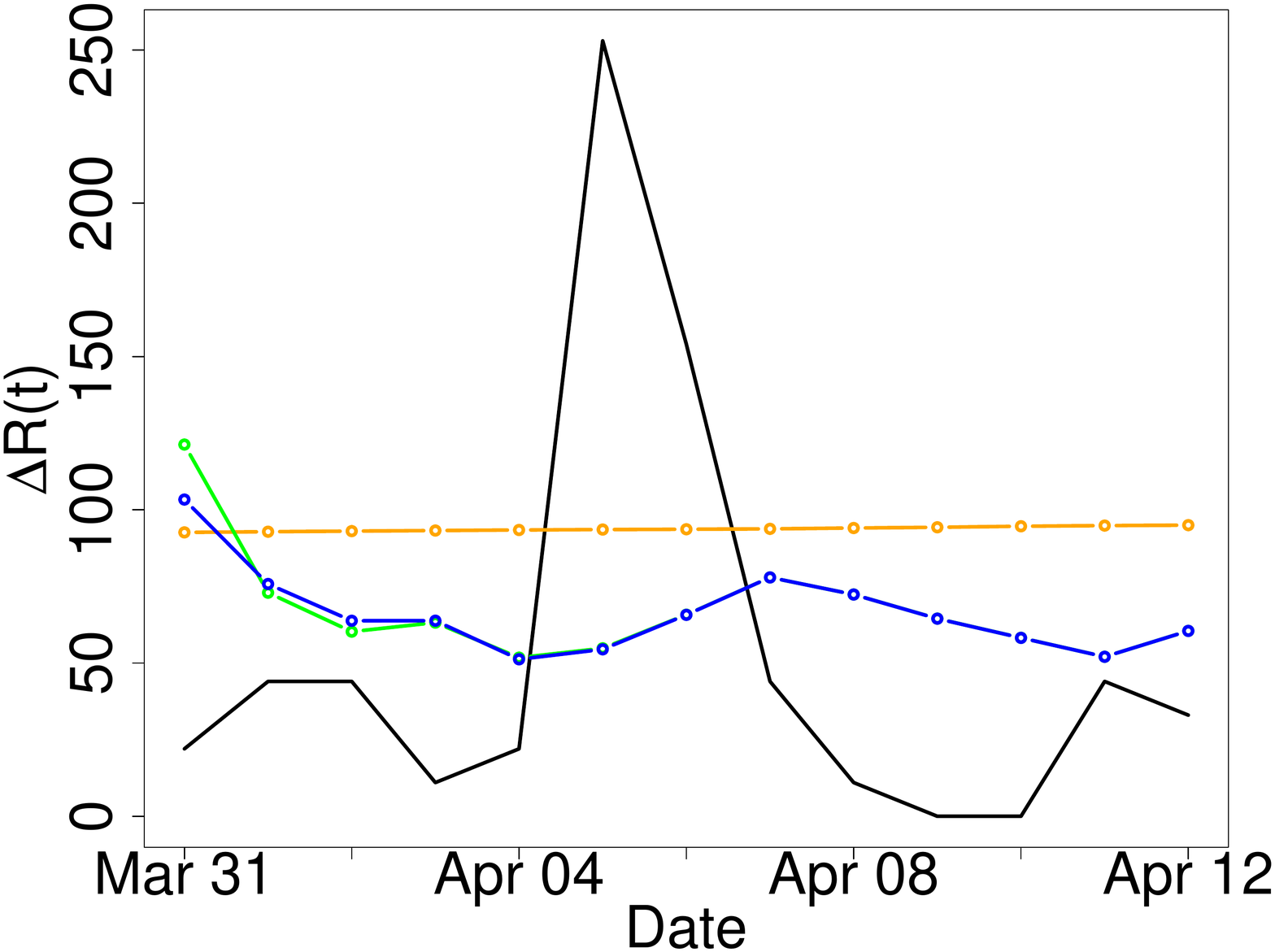}
         \subcaption{OR $\widehat{\Delta R}(t)$}
     \end{subfigure}
     \begin{subfigure}[b]{0.16\textwidth}
         \centering
         \includegraphics[width=\textwidth]{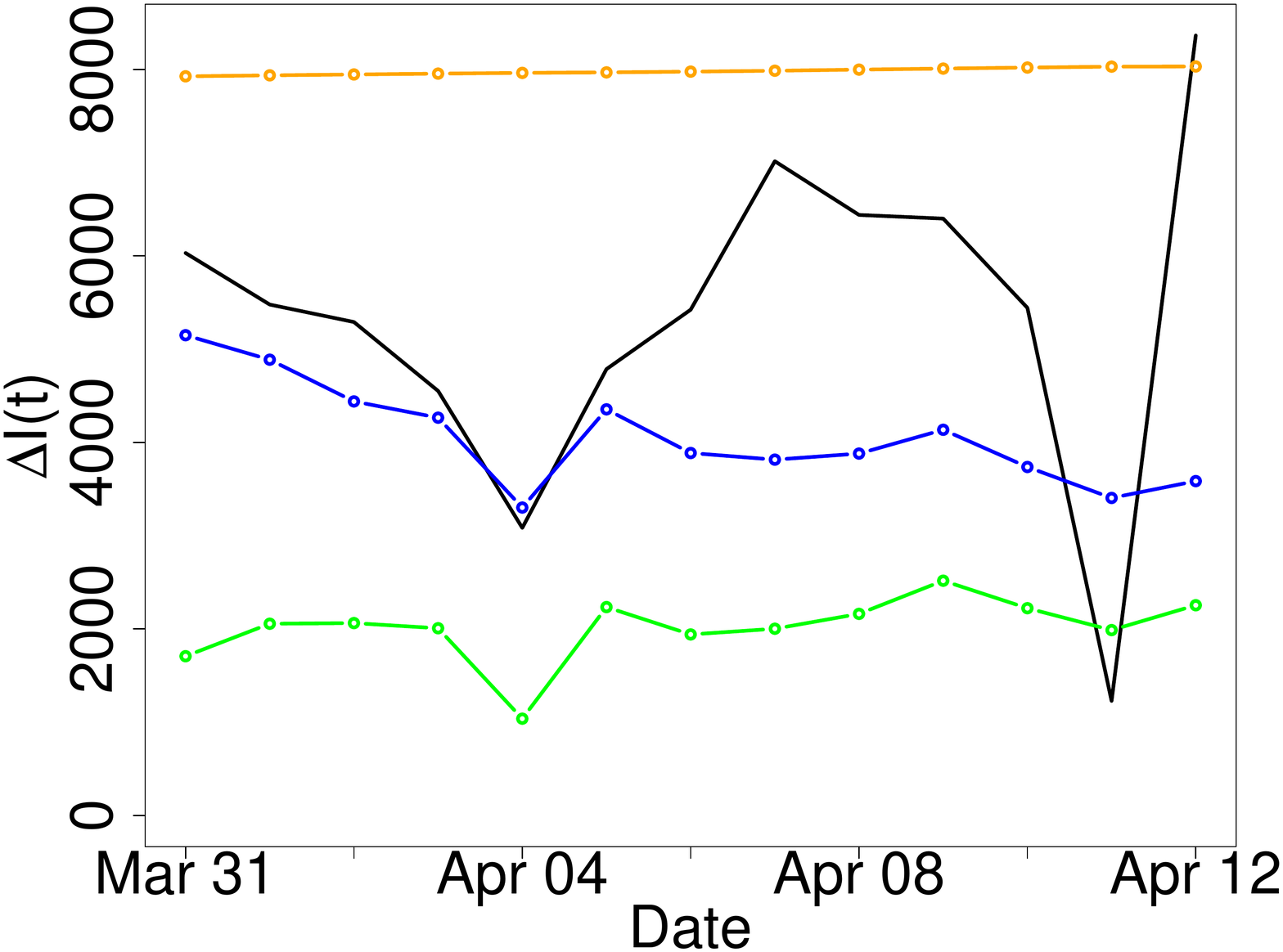}
         \subcaption{FL $\widehat{\Delta I}(t)$}
     \end{subfigure}
          \begin{subfigure}[b]{0.16\textwidth}
         \centering
         \includegraphics[width=\textwidth]{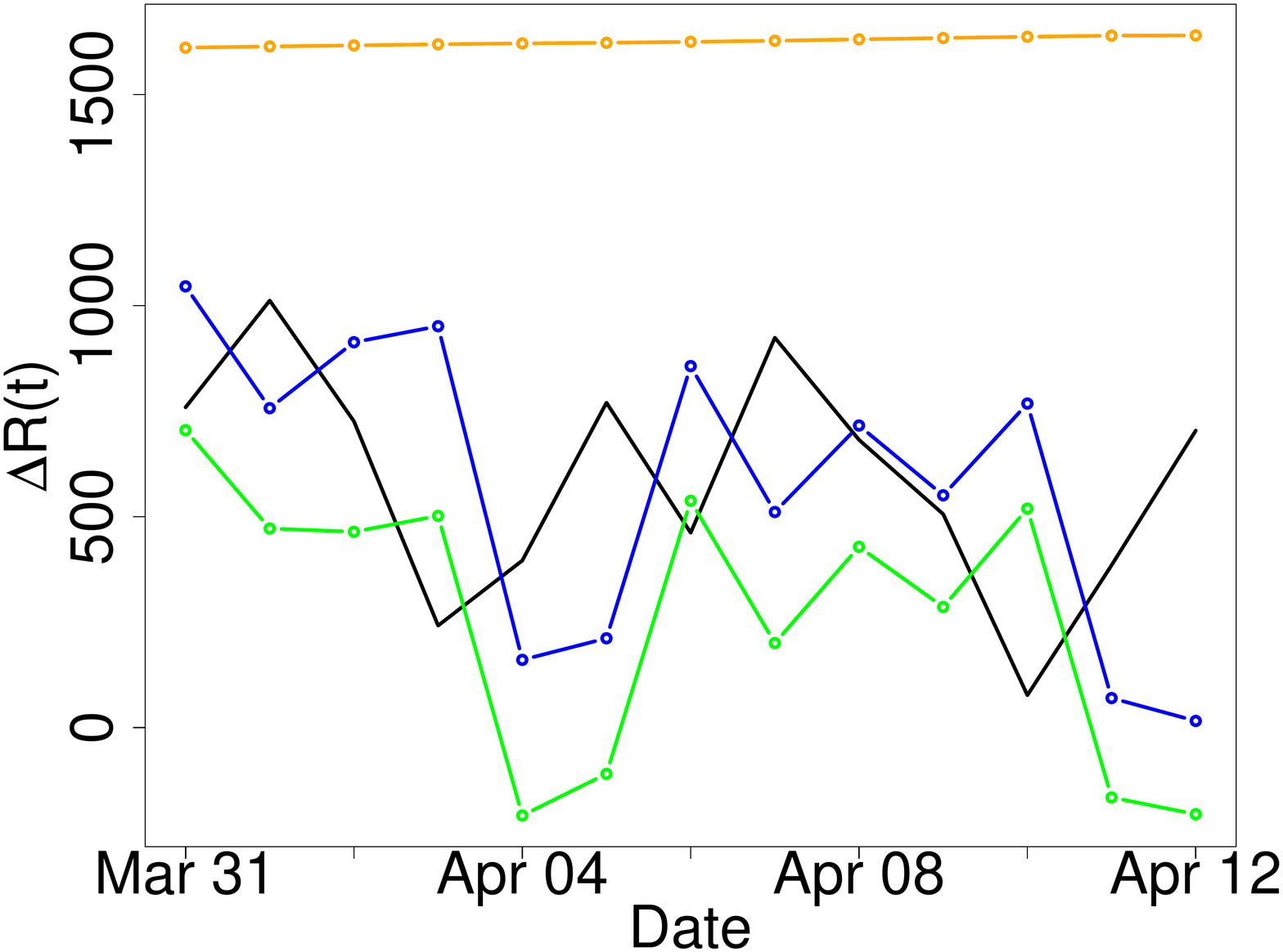}
         \subcaption{FL $\widehat{\Delta R}(t)$}
     \end{subfigure}
      \begin{subfigure}[b]{0.16\textwidth}
         \centering
         \includegraphics[width=\textwidth]{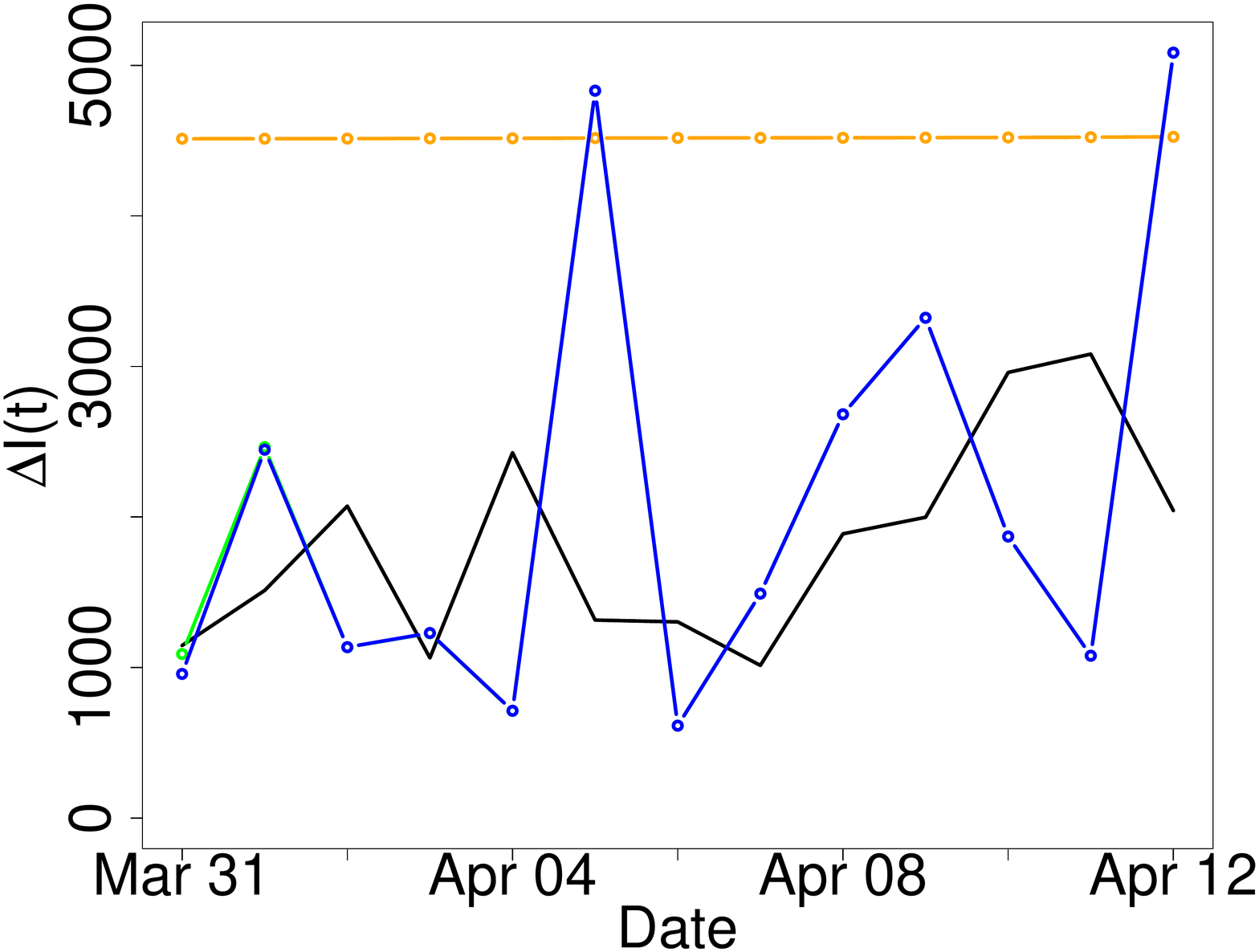}
         \subcaption{CA $\widehat{\Delta I}(t)$}
     \end{subfigure}
     \begin{subfigure}[b]{0.16\textwidth}
         \centering
         \includegraphics[width=\textwidth]{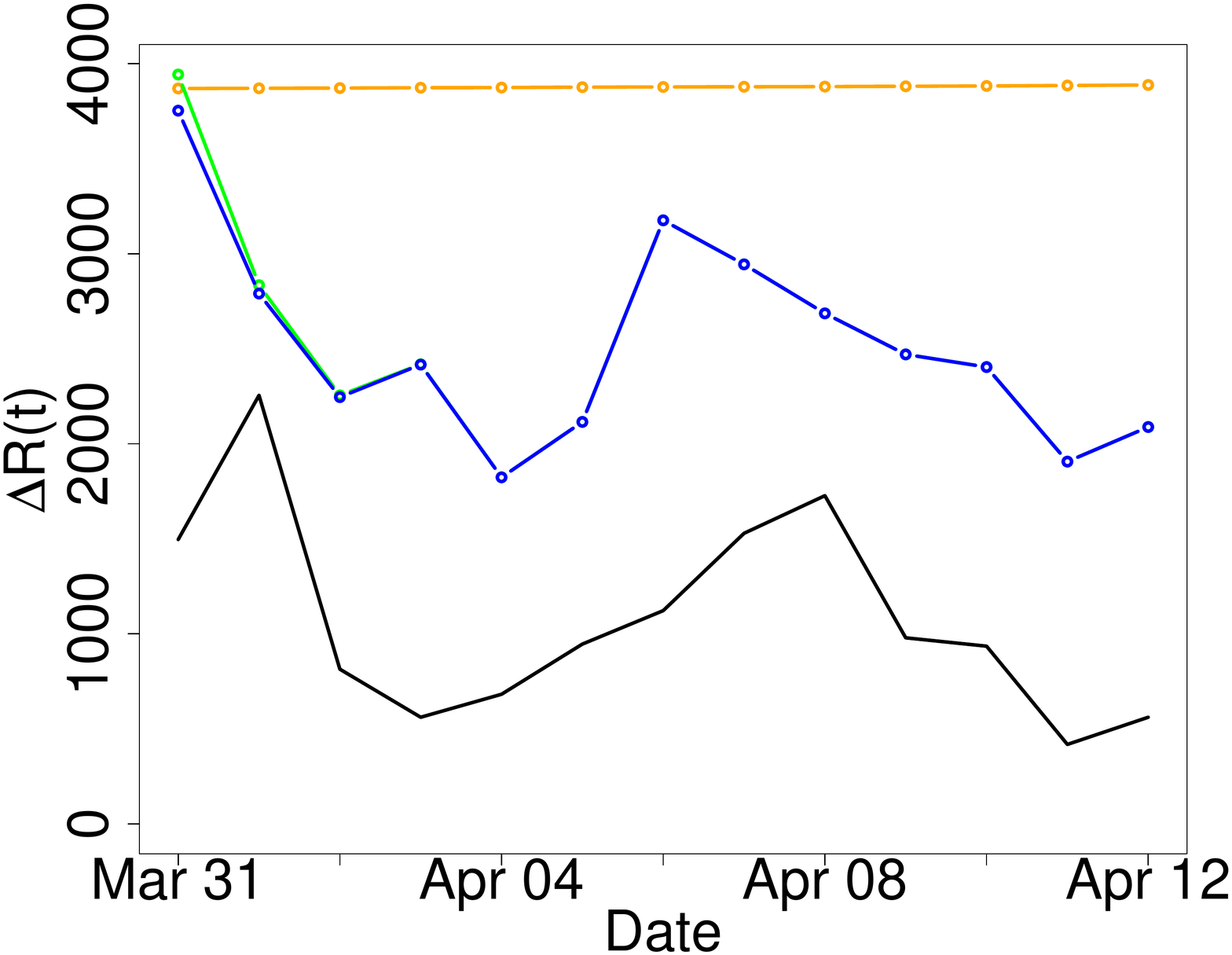}
         \subcaption{CA $\widehat{\Delta R}(t)$}
     \end{subfigure}
      \begin{subfigure}[b]{0.16\textwidth}
         \centering
         \includegraphics[width=\textwidth]{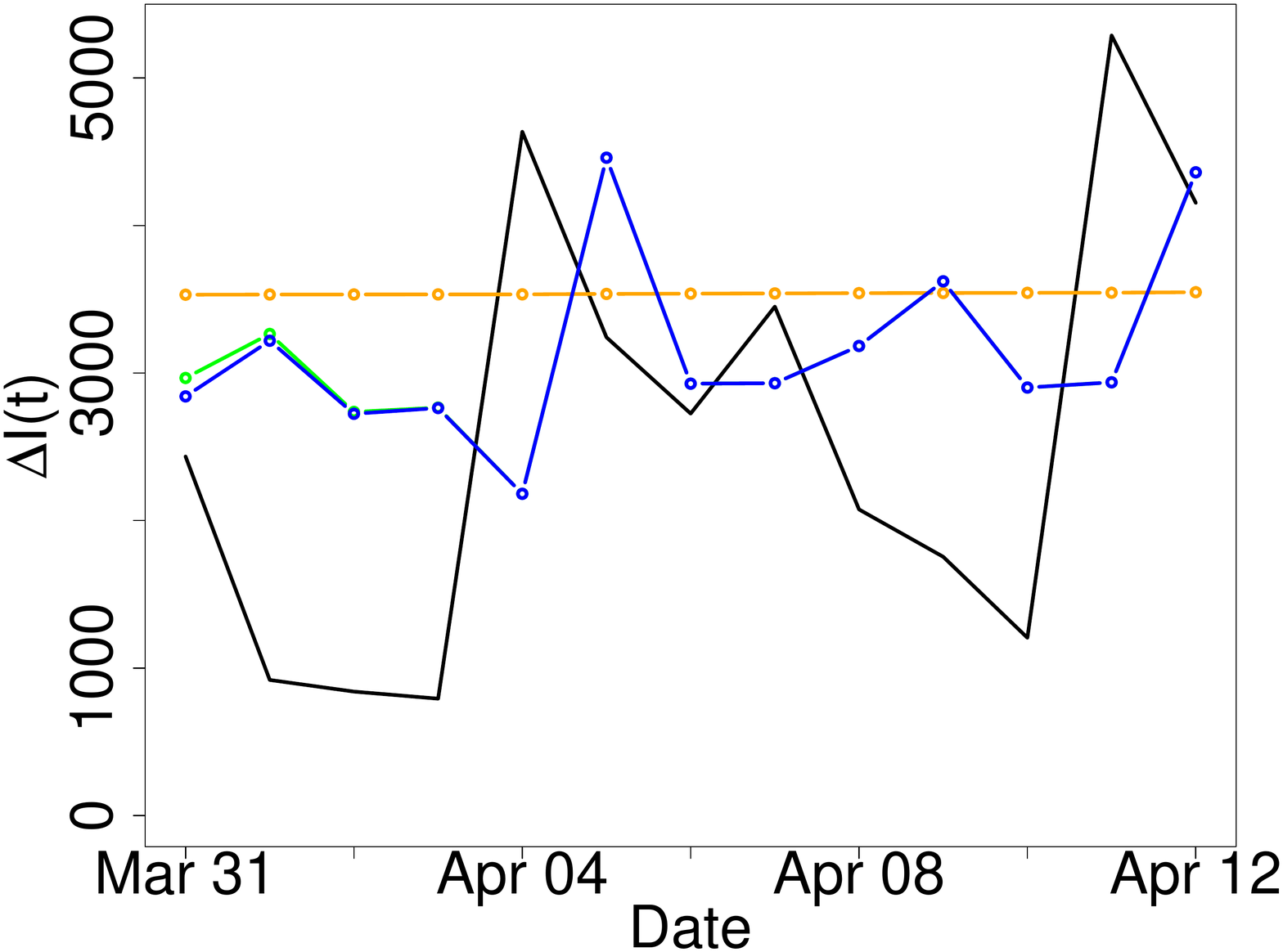}
         \subcaption{TX $\widehat{\Delta I}(t)$}
     \end{subfigure}
     \begin{subfigure}[b]{0.16\textwidth}
         \centering
         \includegraphics[width=\textwidth]{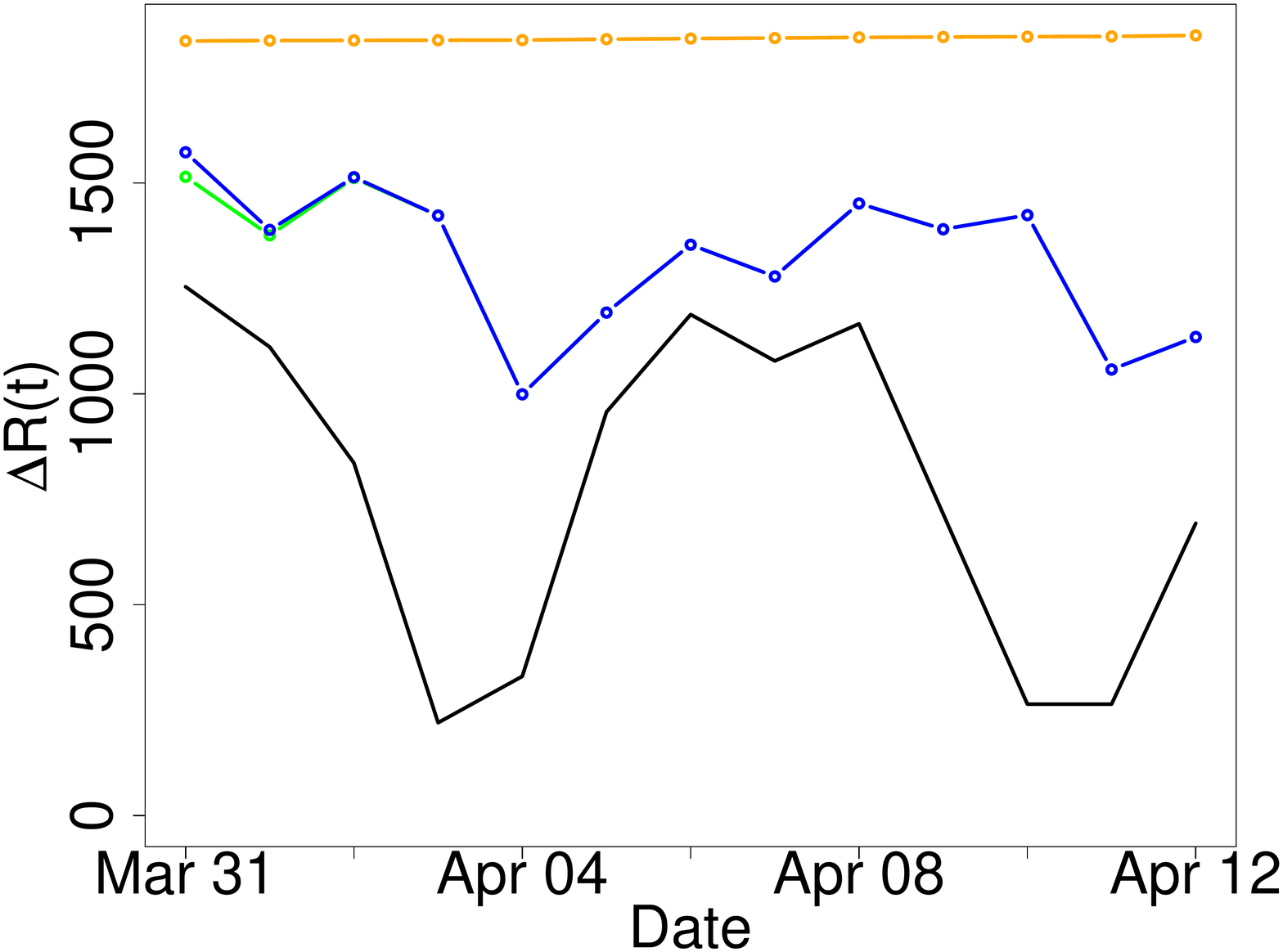}
         \subcaption{TX $\widehat{\Delta R}(t)$}
     \end{subfigure}
        \caption{Prediction of response variable $y_h$ in three models. Black line: true value;  orange: Model 1; green: Model 2; blue: Model 3.
        }
        \label{fig:prediction value}
\end{figure*}

In Figure \ref{fig:prediction value}, we provide the predicted values in the last 2 weeks from all three different models in five states: New York, Oregon, Florida, California and Texas.  
In both the daily infected cases $\Delta I(t)$ and the daily recovered cases $\Delta R(t)$, the predictions by Model 2 and Model 3 perform better than Model 1 in all  states.

{
The change point detection results using exponential function are presented in  Table~\ref{table_plan_all_exp}. 
As shown in Table~\ref{table_plan_all_exp}, using the exponential under-reporting function, we missed some meaningful change points. For example, in Oregon, Florida, California, and Texas, the change point related to the ``Stay-at-home" plan (March-April 2020) is missing.
}
\begin{table*}[ht!]
\caption{\label{table_plan_all_exp}
Statewide ``Stay-at-home'' plan and reopening plan begin dates, along with the detected change points (CPs) in states (using exponential under-reporting function). }
\centering
{
{\begin{tabular}{lccccc} 
  \hline
  \hline
Region  & ``Stay-at-home'' plan  &Reopening plan (statewide) & Detected change points \\
\hline
New York & March 22 & July 6 (Phase 3) &  April 06 2020  \\ 
Oregon  & March 23 & May 15 (Phase 1) &    Jan 10 2021 \\
Florida &   April 3 & June 5 (Phase 2) &   Aug 01 2020\\
California  & March 19 & June 12 (Phase 2)  &July 25 2020, Aug 17 2020, Jan 13 2021\\
Texas  & April 2  & June 3 (Phase 3) &   Aug 05 2020,  Feb 12 2021\\
 \hline
 \multicolumn{4}{l}{*NY reopening dates vary in counties. Here, we use the NYC reopening dates instead. }\\
\end{tabular}}}
\end{table*}

\vspace{-.5cm}

\subsection{Results for the State of Michigan}
Towards the end of March 2021, Michigan exhibited a large increase in cases.
To that end, we analyzed the daily count of cases for the period March 01, 2020 to May 15, 2021. As can be seen from Table \ref{table_plan_all},
five detected change points occurred in 2020, including April 11 2020, June 04 2020, November 01 2020, November 21 2020, and December 11 2020.
The first one in April 2020 is related to the decreased transmission rate due to the statewide lockdown, while the June one could be linked to the state reopening many activities on June 1, 2020. 
An additional change point was detected in early November as Michigan experienced a fall surge; the state recorded the highest number of cases on November 13. As a policy response, the state governor announced the closure of several businesses and public services, including high schools and universities, for three weeks, effective November 18, 2020. The latter led to sharp decrease of cases and hence to the two change points detected at the end of November and early December. Finally, a change point is detected on March 21, 2021, which is associated to variant B.1.1.7 becoming dominant in the state, while vaccination uptake was somewhat lagging.
\footnote{\url{https://www.deseret.com/coronavirus/2021/4/12/22379609/michigan-covid-19-coronavirus}}

\vspace{-0.5cm}

\section{Additional Results for U.S. Counties}\label{sec:counties}

We focus on the following counties in this analysis: New York City, King County in Washington, Miami-Dade County in Florida,  Charleston, Greenville, Richland, and Horry in South Caroline and Riverside, Santa Barbara in California. Note that the New York Times reports a single value for New York City which includes the cases in New York, Kings, Queens, Bronx and Richmond Counties.

The block size $b_n$ is set to 7 for all counties.
We can choose different block sizes if, in some regions, there are more than one break points of the transmission rate very close to each other which violates the assumptions of our model. To mitigate this issue, we suggest to select a smaller block size to detect most of the change points. Another situation is that the transmission rate in some regions exhibit many small noise disturbances. In this case, we may use a larger block size to smooth out the parameter estimates and detect the change point more accurately.
Further, for determining their neighbors displayed in Table \ref{table_adj_2}, a threshold of 100 miles is used.
Model 2.3 uses the top five counties in the corresponding state with the smallest similarity score, while Model 2.4 uses all counties in the given state.
The additional resulting neighbors for Model 2.3 are displayed in Table \ref{table_similar_2}.

Since the parameters obtained from \eqref{eq:model_var} vary greatly, the 7-day moving average of transmission  $\beta(t)$ and recovery $\gamma(t)$ rates is provided in Figure \ref{fig:rates_county_smooth}. As can be seen from those figures, our estimated piecewise parameters are very close to the  real parameters in the 7-day moving average.

\begin{table*}[ht!]
\caption{\label{table_plan_counties} 
Statewide ``Stay-at-home'' plan and reopening plan begin dates, along with the detected change points (CPs) in counties/cities.}
\scriptsize
\centering
 \resizebox{0.98\columnwidth}{!}{
\begin{tabular}{lccccc} 
  \hline
  \hline
Region  & ``Stay-at-home'' plan  &Reopening plan (statewide) &  Mask mandate  (county) & Detected CP \\
\hline
New York City (NY)  & March 22 & May 15 & -  &  April 05 2020 \\
King County (WA) & March 23 & May 30 & - &April 11 2020\\
 Miami-Dade County (FL) & April 03 & June 05 & - &  April 13 2020, June 23 2020, July 29 2020\\
Charleston  & April 7 & April 20 & July 01  &  April 09 2020, June 12 2020, July 06 2020 \\
Greenville & April 7 & April 20 & June 23 & May 06 2020, May 30 2020, Jun 20 2020\\
Richland & April 7 & April 20 & July 06 &  April 13 2020, July 22 2020\\
Horry & April 7 & April 20 & July 03 & June 06 2020, July 06 2020 \\
 \hline
\end{tabular}}
\end{table*}

\begin{table}[!ht]
\caption{\label{table_adj_2}Neighboring cities/counties by distance (for Model 2.1 and Model 2.2). }
\centering
{
\begin{tabular}{ll}
  \hline
  \hline
Region &  Neighboring Regions  (counties: within 100 miles)   \\
  \hline
  New York City (NY) & Hudson, Essex, Union, Bergen, Nassau      \\
  King County (WA) &  Pierce, Island, Kitsap, Kittitas,  Snohomish \\
    Miami-Dade County (FL)   & Broward, Collier, Hendry, Palm Beach, Monroe   \\ 
  Charleston (SC)&  Dorchester, Beaufort, Berkeley, Colleton, Georgetown \\
  Greenville (SC)&  Polk, Henderson, Anderson, Pickens, Spartanburg\\
  Richland (SC)& Calhoun,  Lexington, Sumter, Fairfield, Kershaw \\
  Horry  (SC)&    Columbus, Dillon, Marion, Florence, Georgetown\\
  Riverside (CA) & San Bernardino, San Diego, Imperial \\ 
  Santa Barbara (CA) & San Luis Obispo, Kern, Ventura \\ 
 \hline
\end{tabular}}
\end{table}

\begin{table}[!ht]
\caption{\label{table_similar_2}  Neighboring cities/counties by similarity score (for Model 2.3). }
\centering
{
\begin{tabular}{l l}
  \hline
  \hline
Region &  Neighboring Regions  (counties: in the same state)   \\
  \hline
  New York City (NY)& Dutchess, Erie, Onondaga, Monroe, Tompkins \\ 
  King County (WA)& Kitsap, Thurston, Clallam, Clark, Pierce \\ 
    Miami-Dade County (FL) &  Broward, Osceola, Duval, Escambia, Orange \\ 
    Charleston (SC)&  Horry, Dorchester, Berkeley, Beaufort, Greenville \\ 
  Greenville (SC)&  Lexington, Pickens, York, Spartanburg, Horry \\ 
  Richland (SC)& Lexington, York, Greenville, Berkeley, Beaufort \\ 
  Horry  (SC)&  Berkeley, Dorchester, York, Greenville, Pickens \\ 
  Riverside (CA)& Orange, Los Angeles, Ventura, San Diego, Monterey \\ 
  Santa Barbara (CA) & San Diego, Ventura, Orange, San Francisco, Contra Costa \\ 
 \hline
\end{tabular}}
\end{table}

\begin{figure*}[!ht]
     \centering
     \captionsetup[sub]{font=small,labelfont={bf,sf}}
     \begin{subfigure}[b]{0.19\textwidth}
         \centering
         \includegraphics[width=\textwidth]{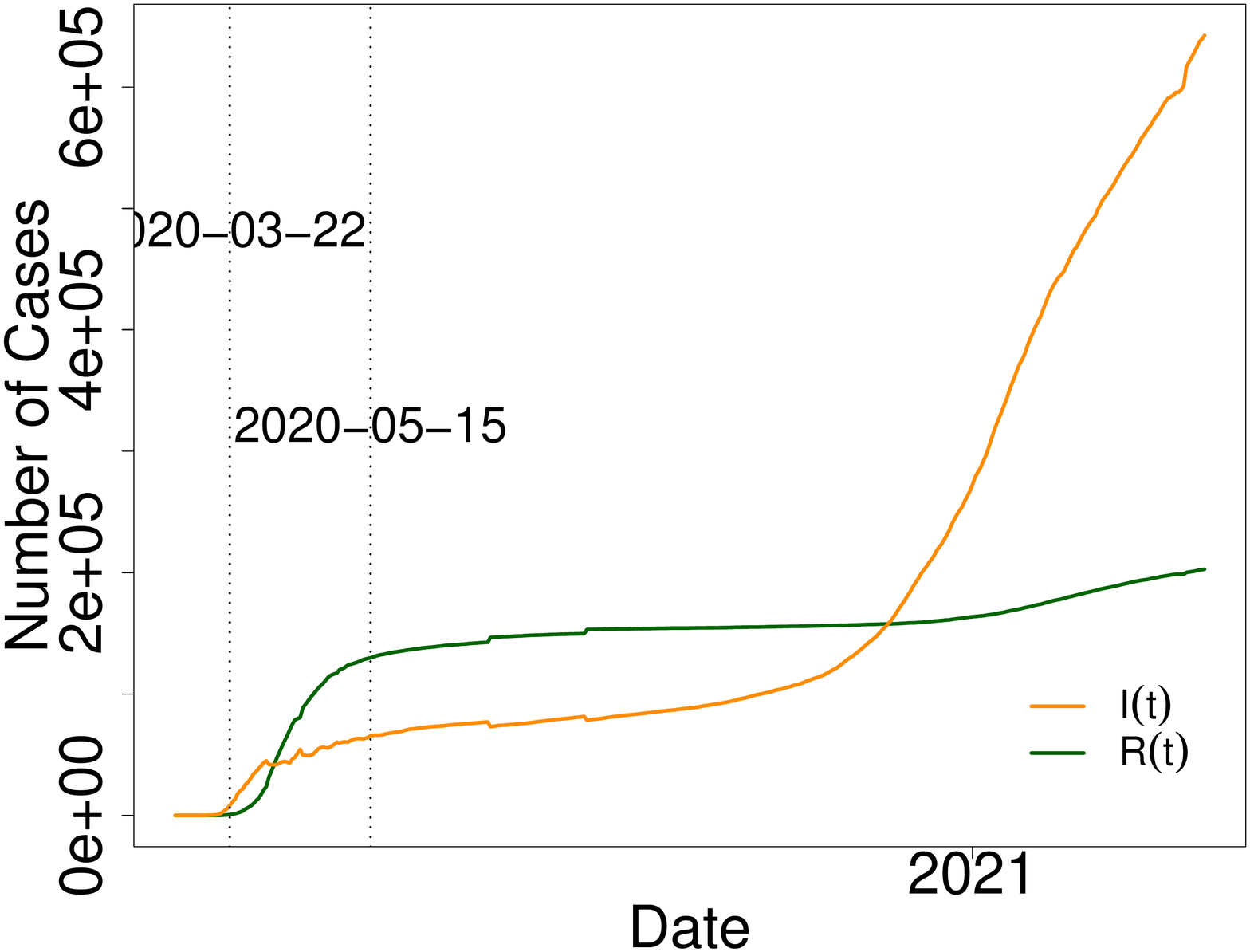}
         \subcaption{New York City}
     \end{subfigure}
     \begin{subfigure}[b]{0.19\textwidth}
         \centering
         \includegraphics[width=\textwidth]{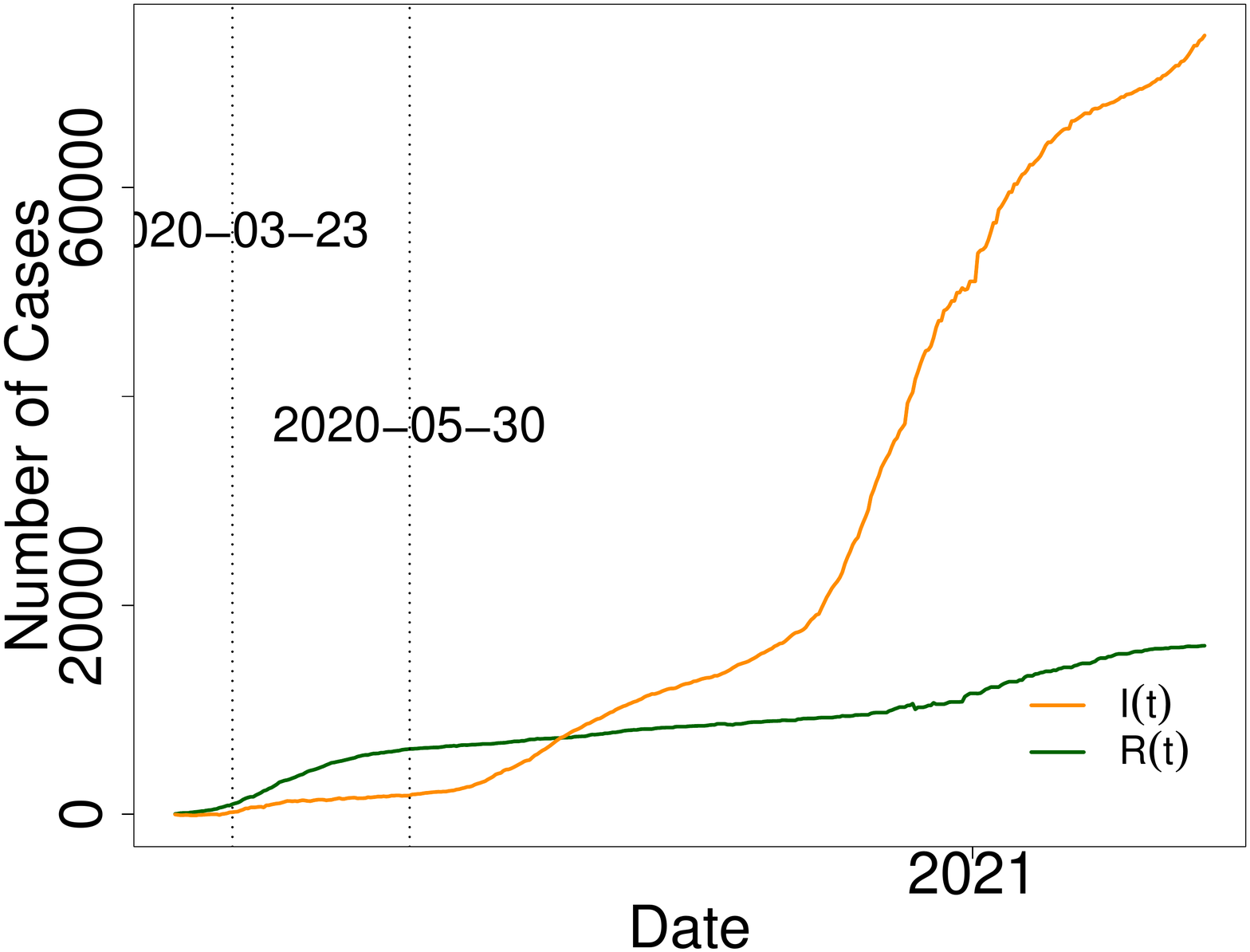}
         \subcaption{King County}
     \end{subfigure}
     \begin{subfigure}[b]{0.19\textwidth}
         \centering
         \includegraphics[width=\textwidth]{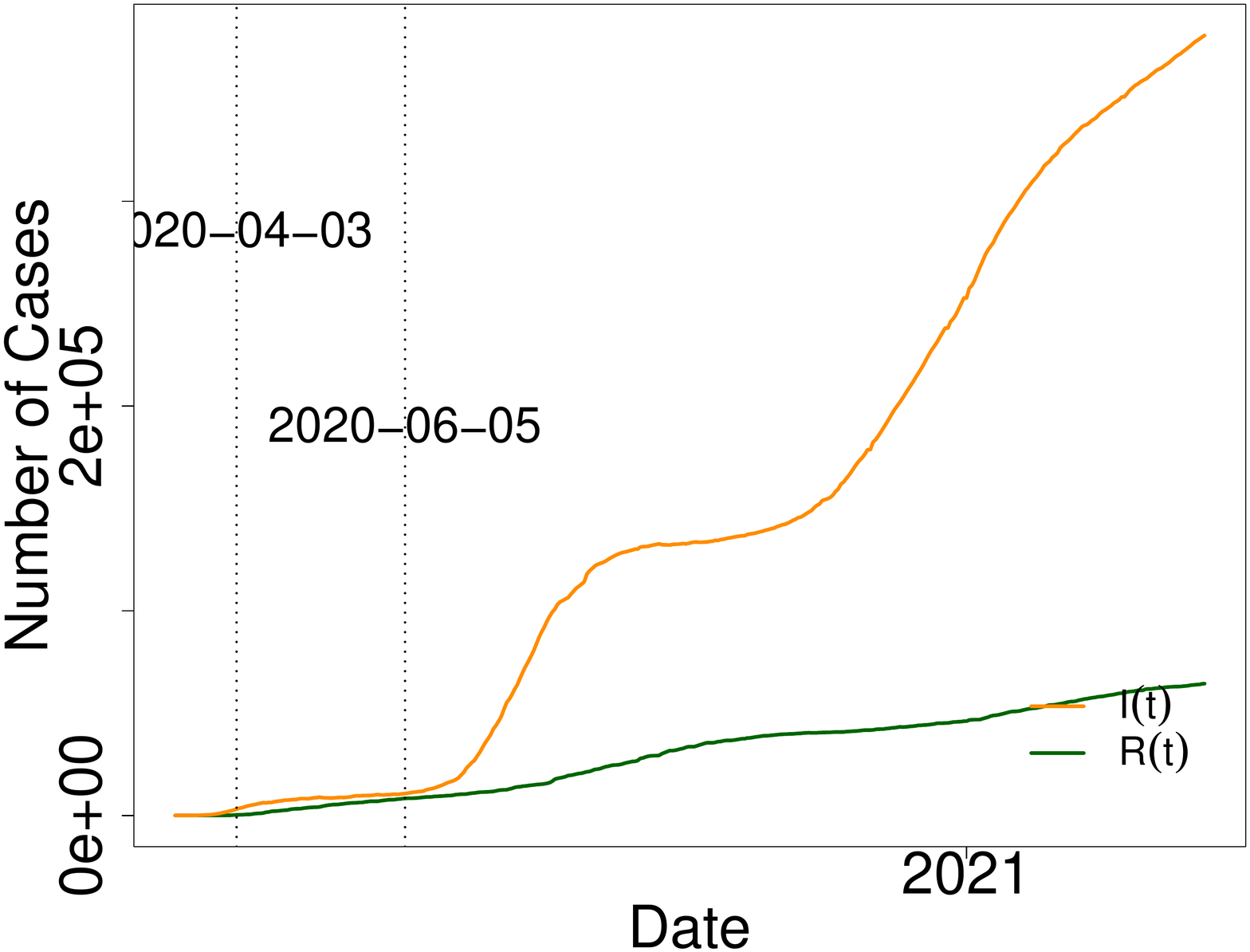}
         \subcaption{Miami-Dade }
     \end{subfigure}
     \begin{subfigure}[b]{0.19\textwidth}
         \centering
         \includegraphics[width=\textwidth]{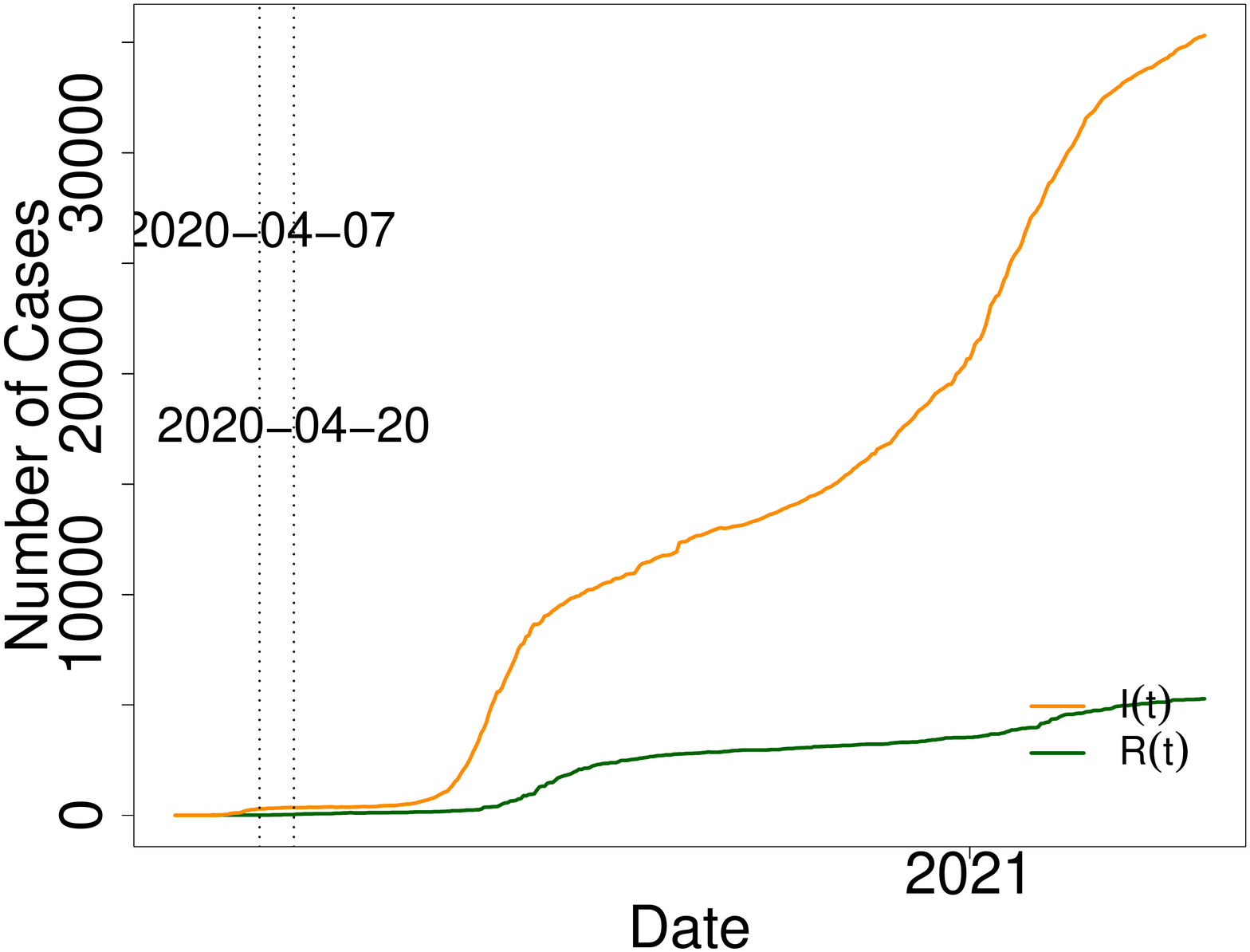}
         \subcaption{Charleston}
     \end{subfigure}
     \begin{subfigure}[b]{0.19\textwidth}
         \centering
         \includegraphics[width=\textwidth]{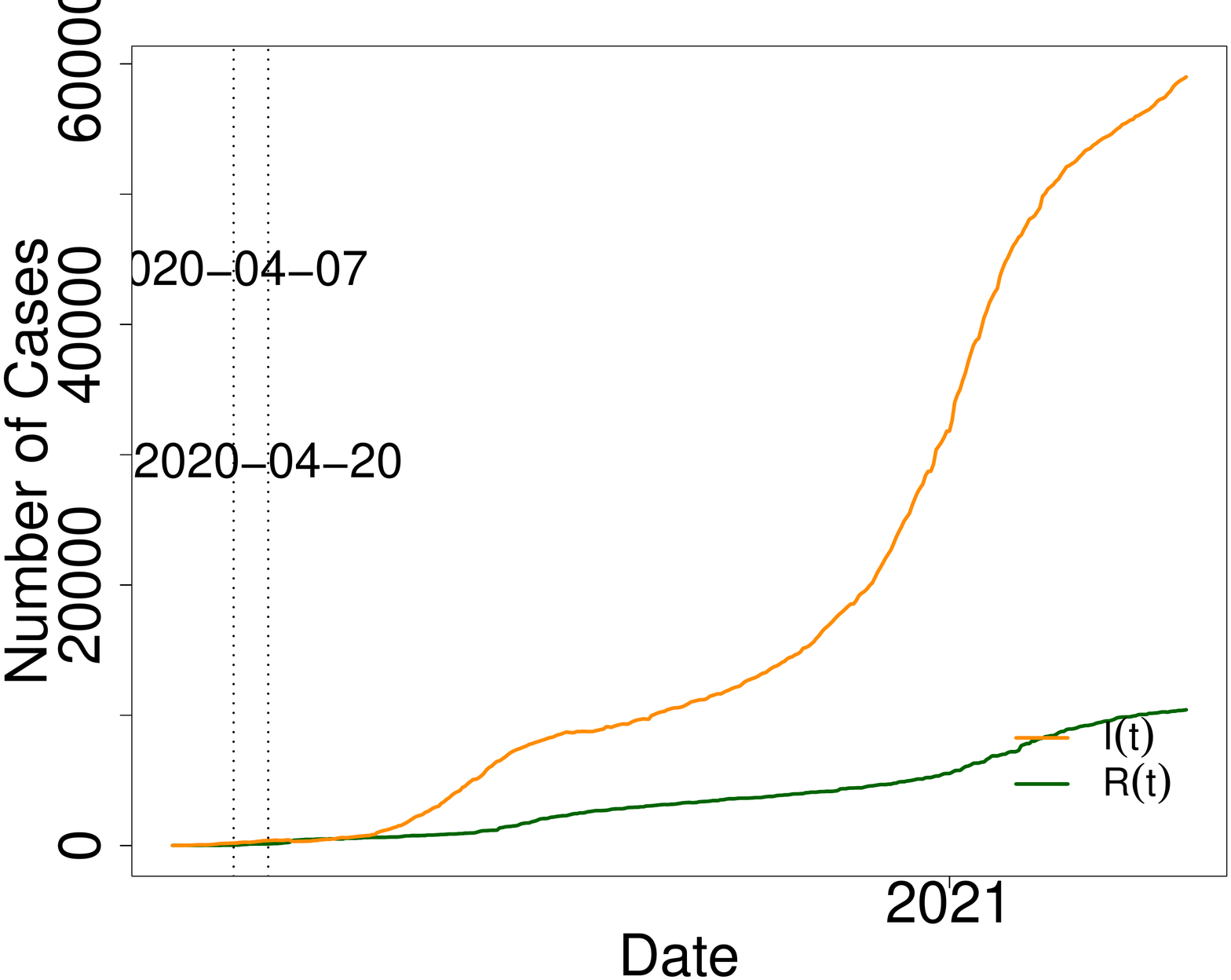}
         \subcaption{Greenville}
     \end{subfigure}
     \begin{subfigure}[b]{0.19\textwidth}
         \centering
         \includegraphics[width=\textwidth]{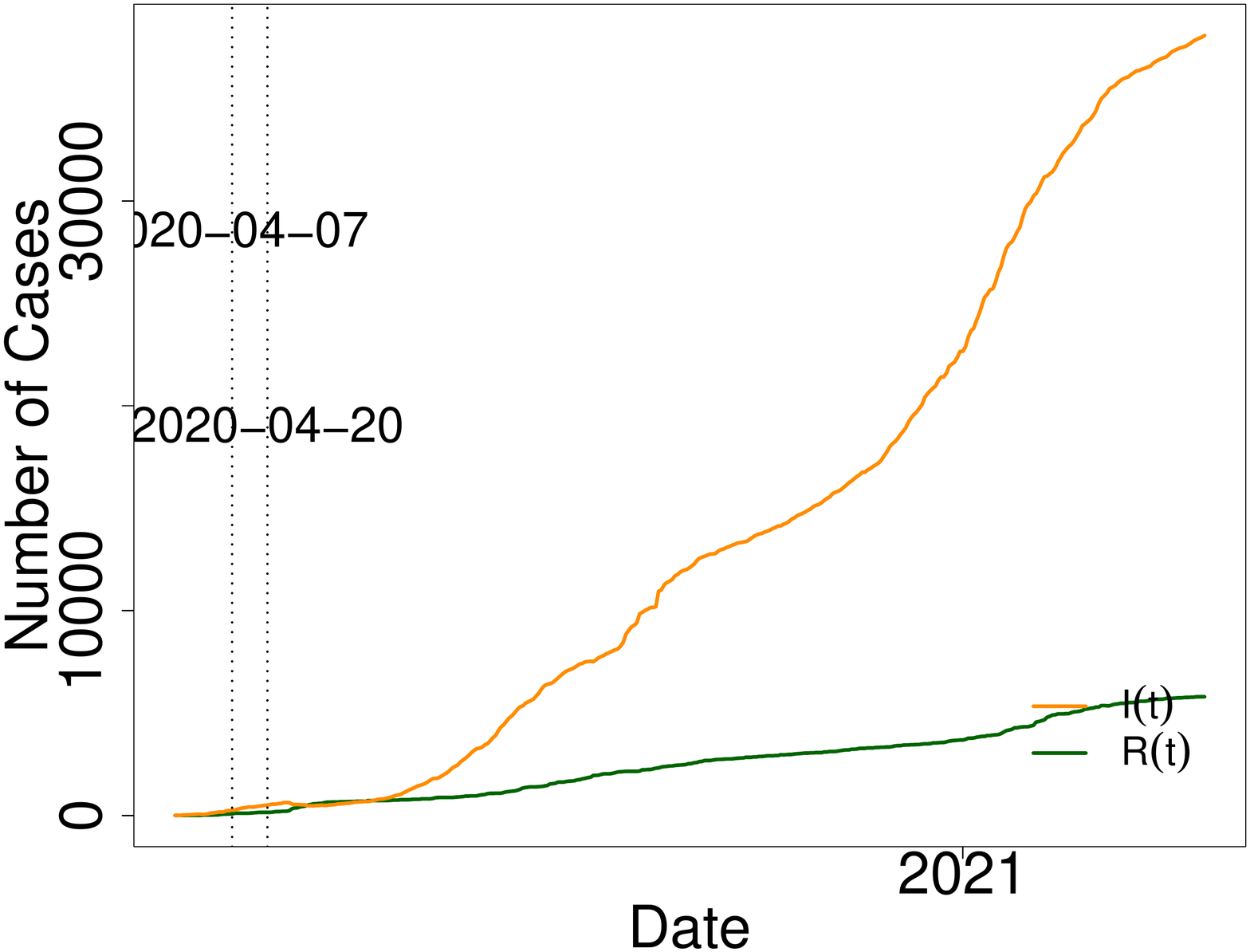}
         \subcaption{Richland}
     \end{subfigure}
     \begin{subfigure}[b]{0.19\textwidth}
         \centering
         \includegraphics[width=\textwidth]{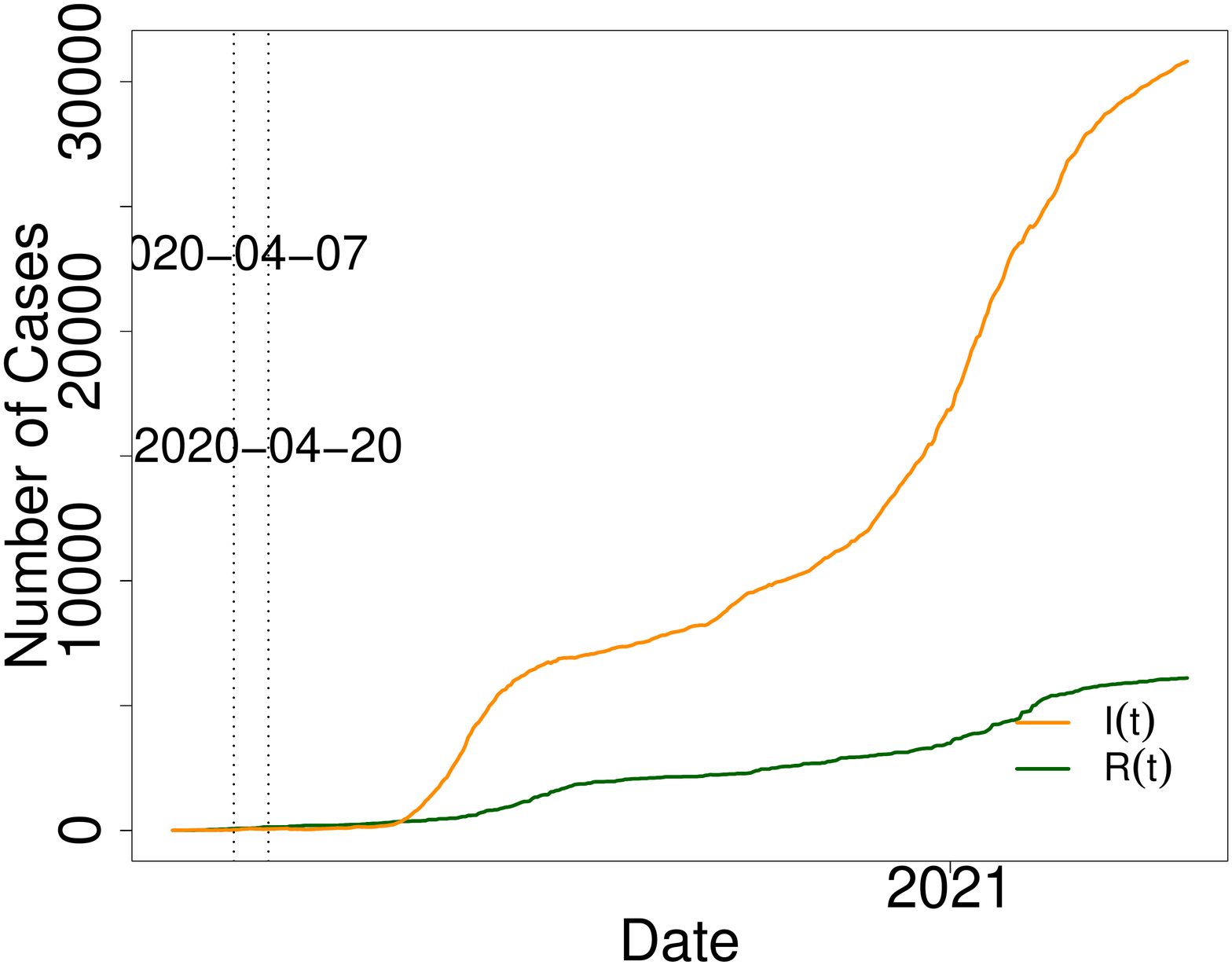}
         \subcaption{Horry}
     \end{subfigure}
     \begin{subfigure}[b]{0.19\textwidth}
         \centering
         \includegraphics[width=\textwidth]{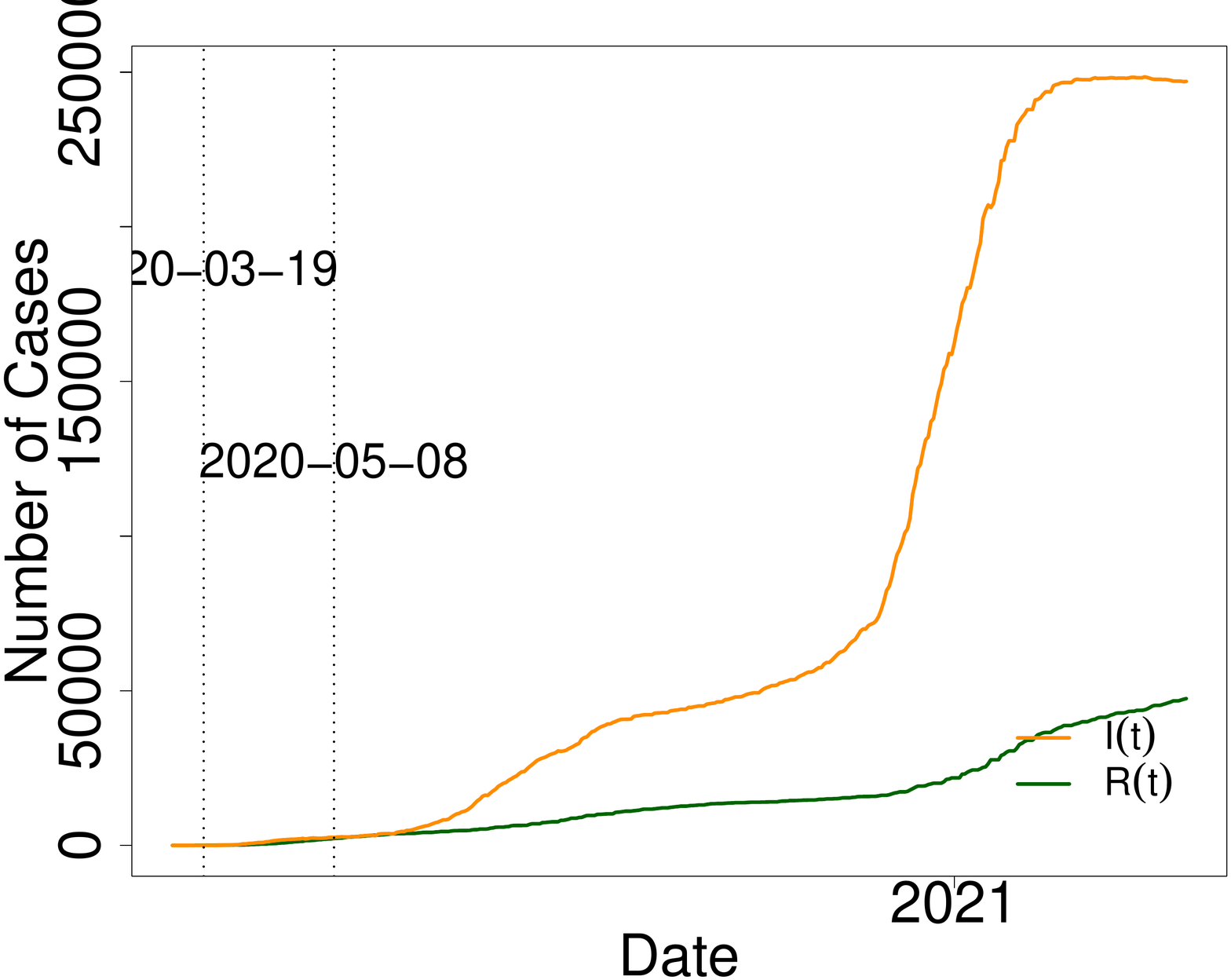}
         \subcaption{Riverside}
     \end{subfigure}
     \begin{subfigure}[b]{0.19\textwidth}
         \centering
         \includegraphics[width=\textwidth]{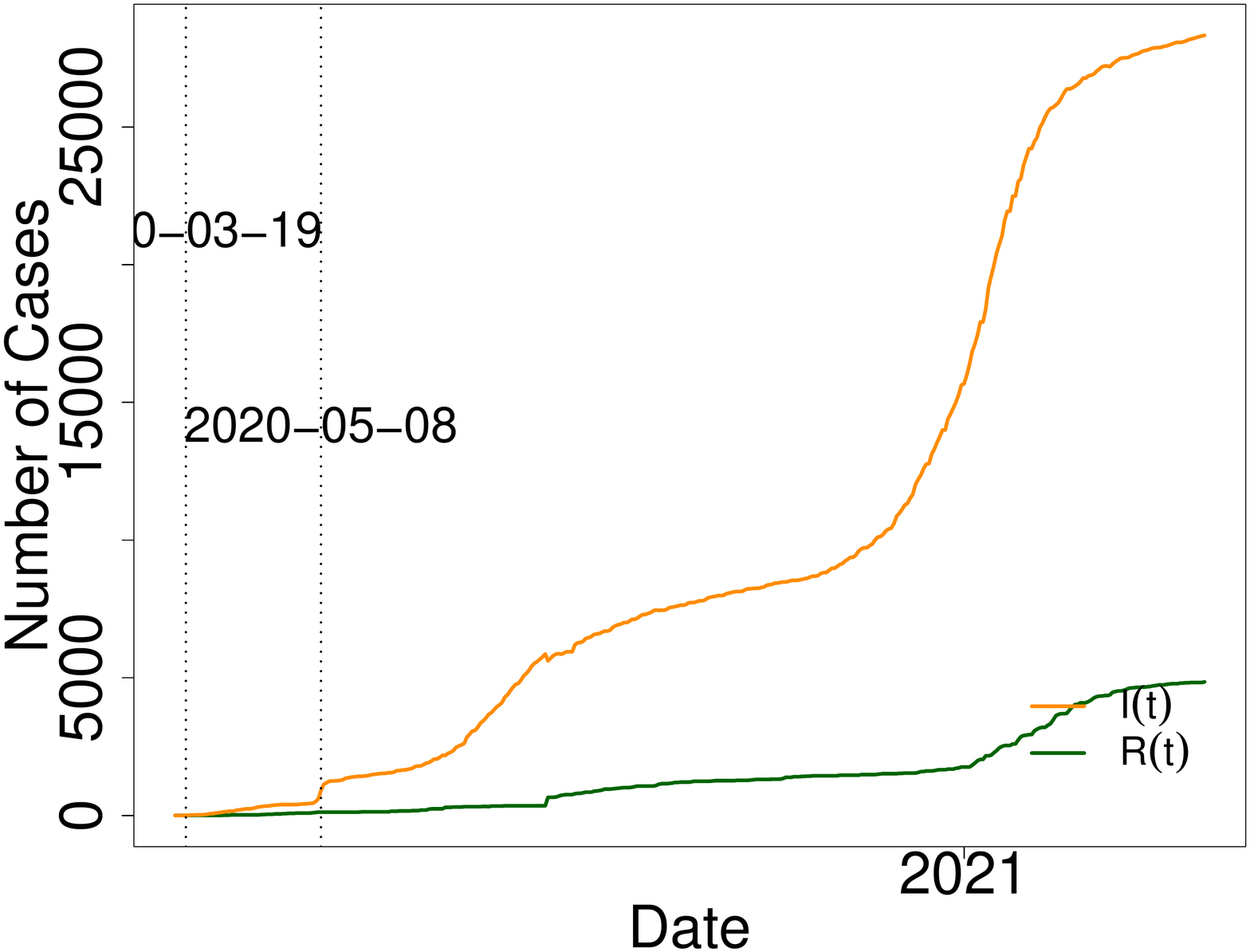}
         \subcaption{Santa Barbara}
     \end{subfigure}
        \caption{
        Number of infected cases (dark orange) and recovered cases (dark green) in counties/cities. The first vertical black dotted line indicates the ``stay-at-home'' order begin date for each corresponding state, while the second vertical black dotted line indicates the reopening begin date.}
        \label{fig:numbers_2}
\end{figure*}

\begin{figure*}[ht!]
     \centering
     \captionsetup[sub]{font=small, labelfont={bf,sf}}
     \begin{subfigure}[b]{0.19\textwidth}
         \centering
         \includegraphics[width=\textwidth]{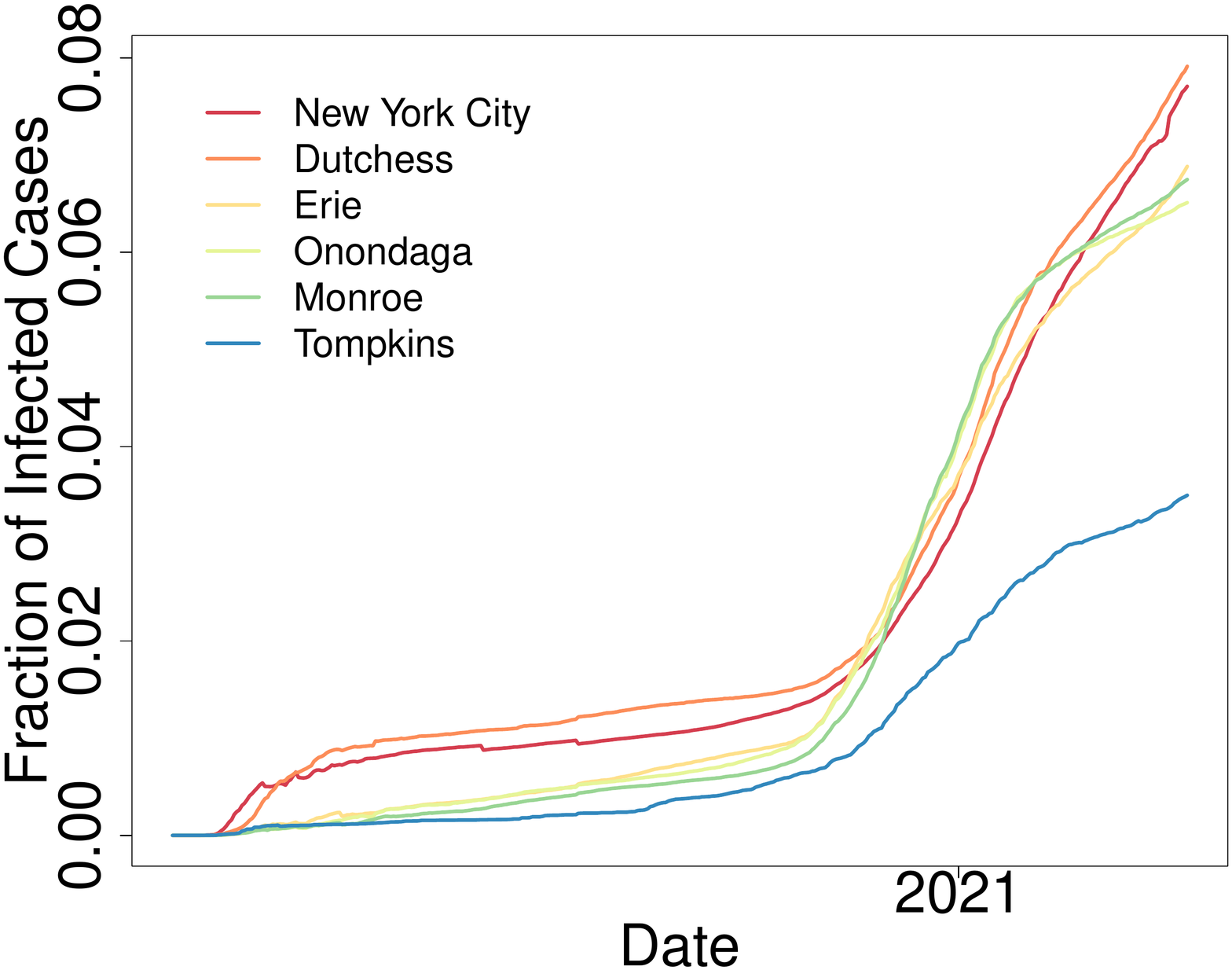}
         \subcaption{New York City}
     \end{subfigure}
     \begin{subfigure}[b]{0.19\textwidth}
         \centering
         \includegraphics[width=\textwidth]{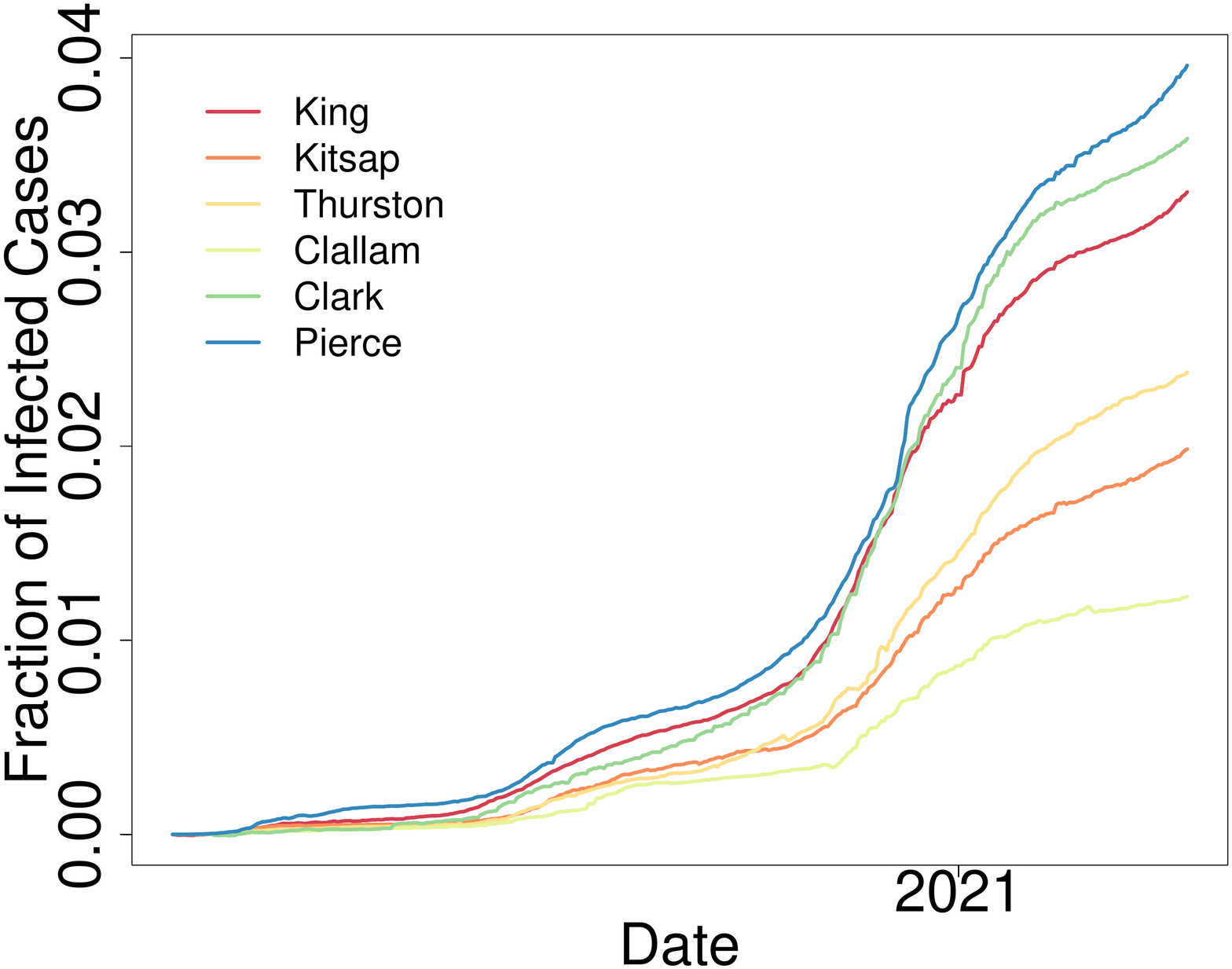}
         \subcaption{King County}
     \end{subfigure}
     \begin{subfigure}[b]{0.19\textwidth}
         \centering
         \includegraphics[width=\textwidth]{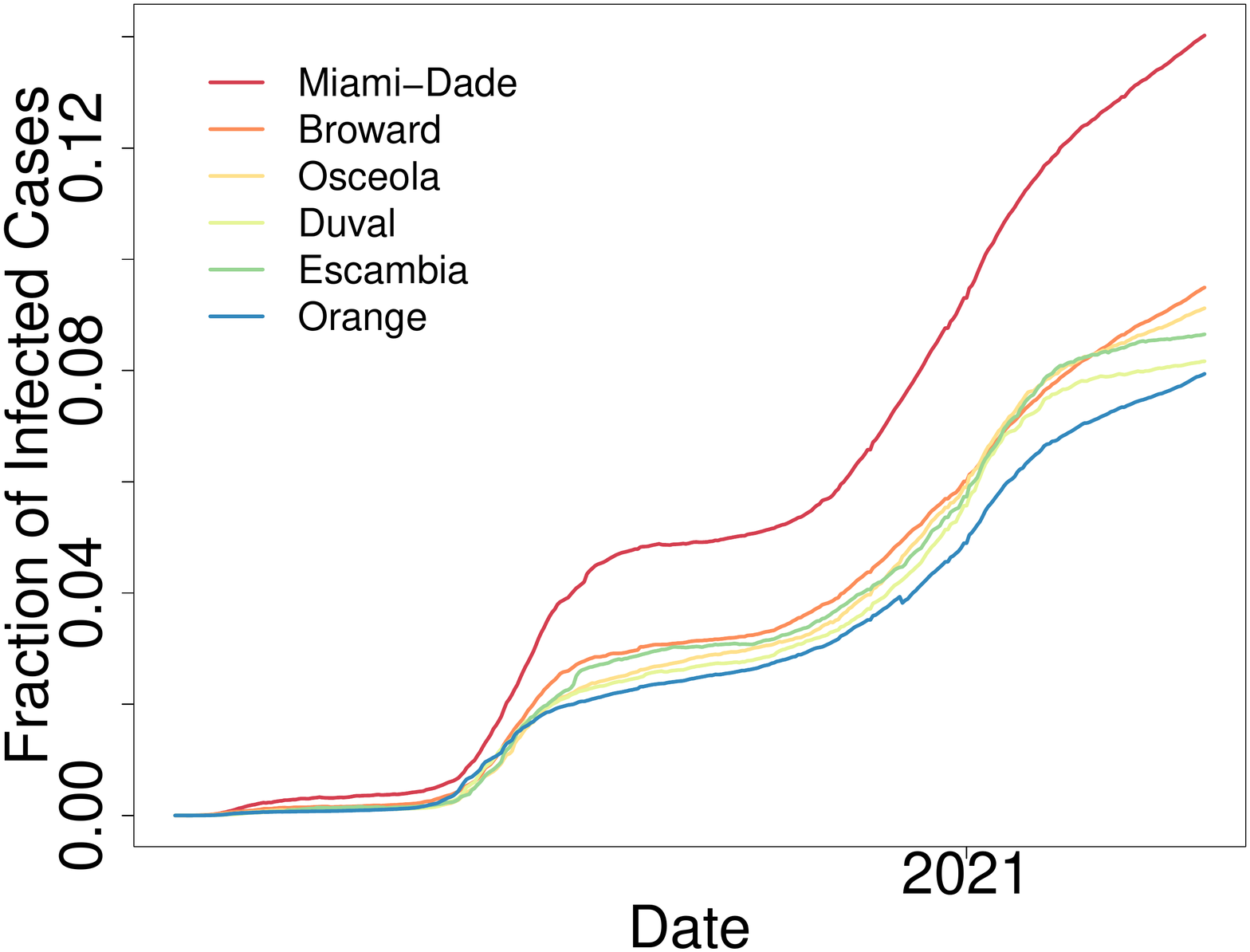}
         \subcaption{Miami-Dade}
     \end{subfigure}
     \begin{subfigure}[b]{0.19\textwidth}
         \centering
         \includegraphics[width=\textwidth]{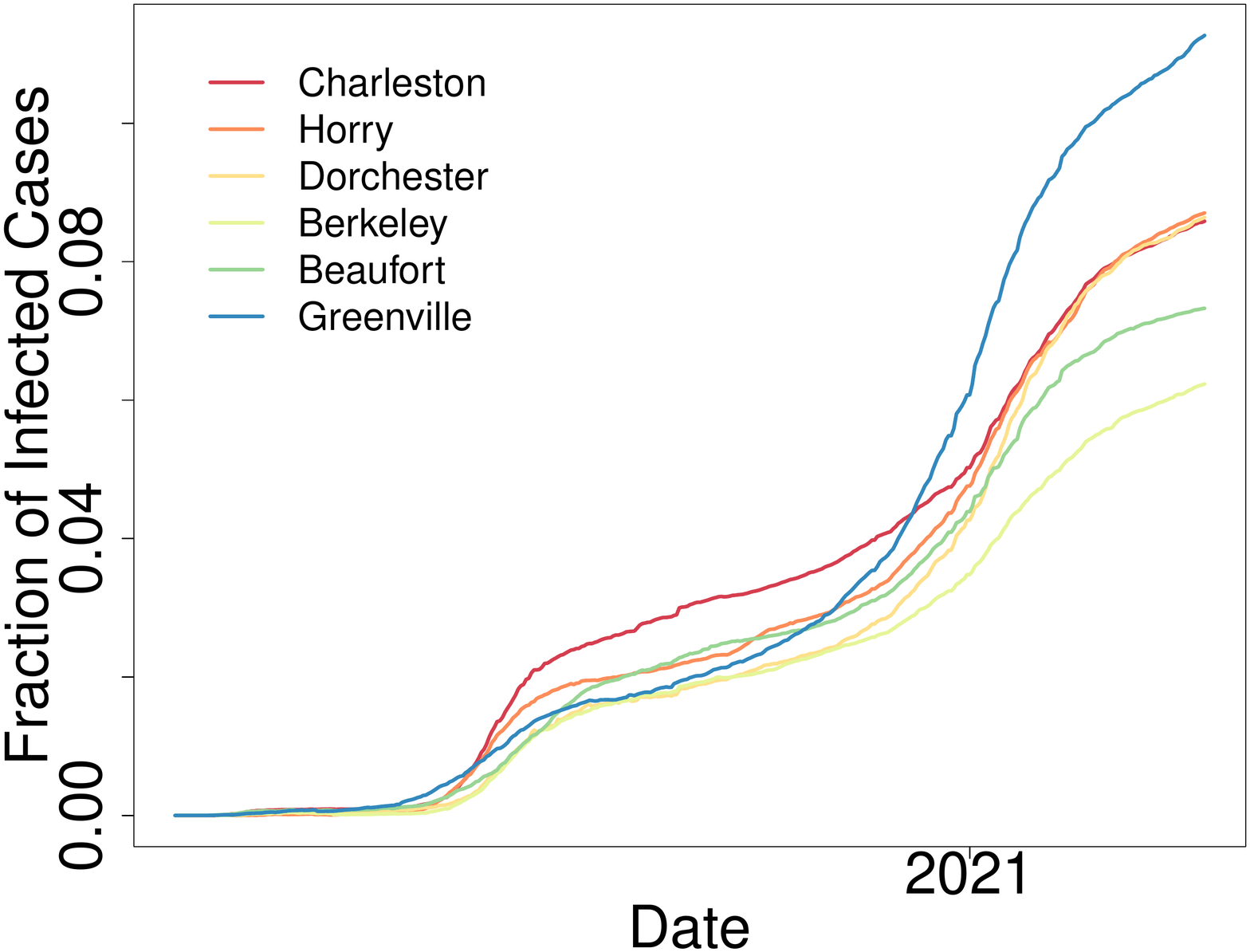}
         \subcaption{Charleston}
     \end{subfigure}
     \begin{subfigure}[b]{0.19\textwidth}
         \centering
         \includegraphics[width=\textwidth]{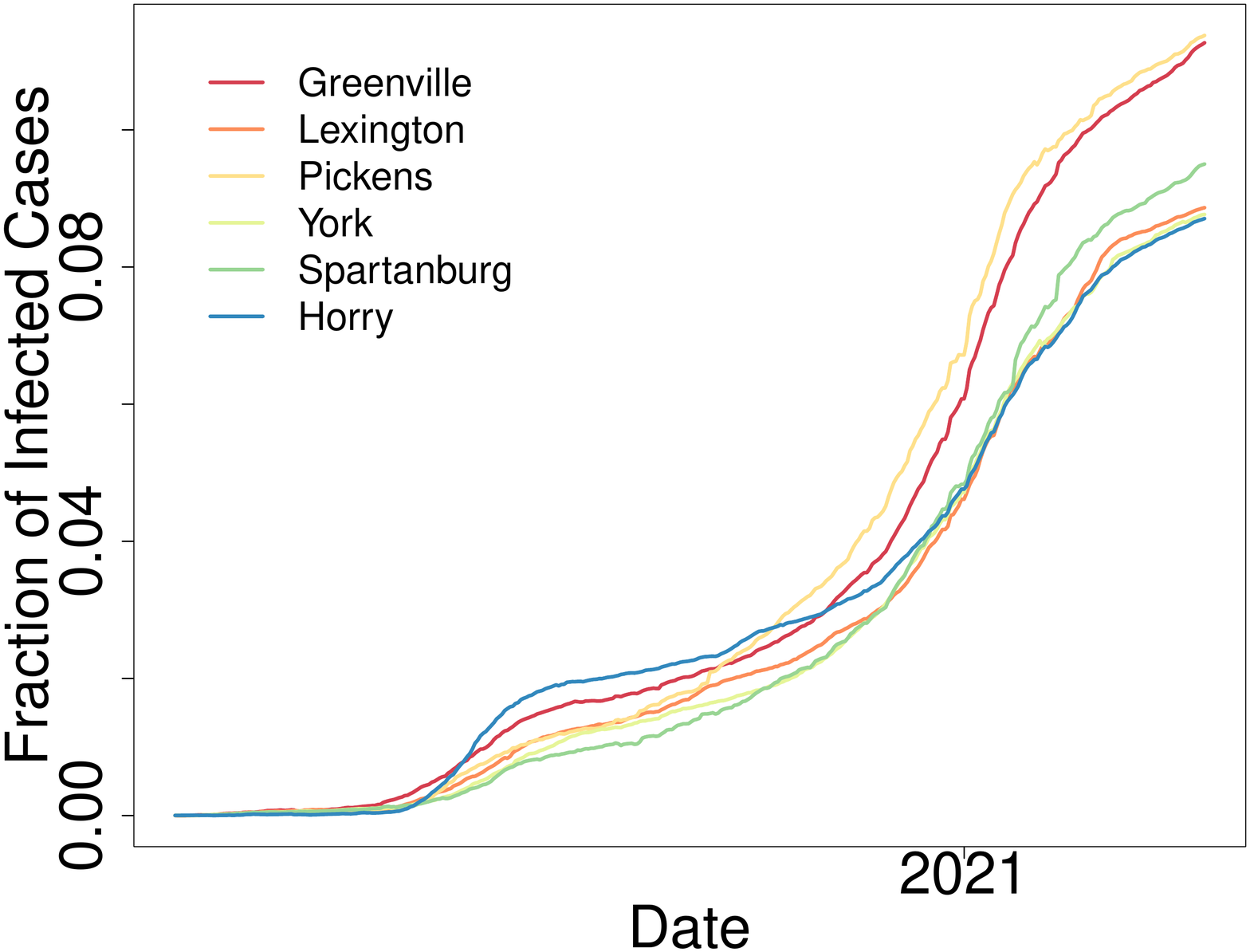}
         \subcaption{Greenville}
     \end{subfigure}
     \begin{subfigure}[b]{0.19\textwidth}
         \centering
         \includegraphics[width=\textwidth]{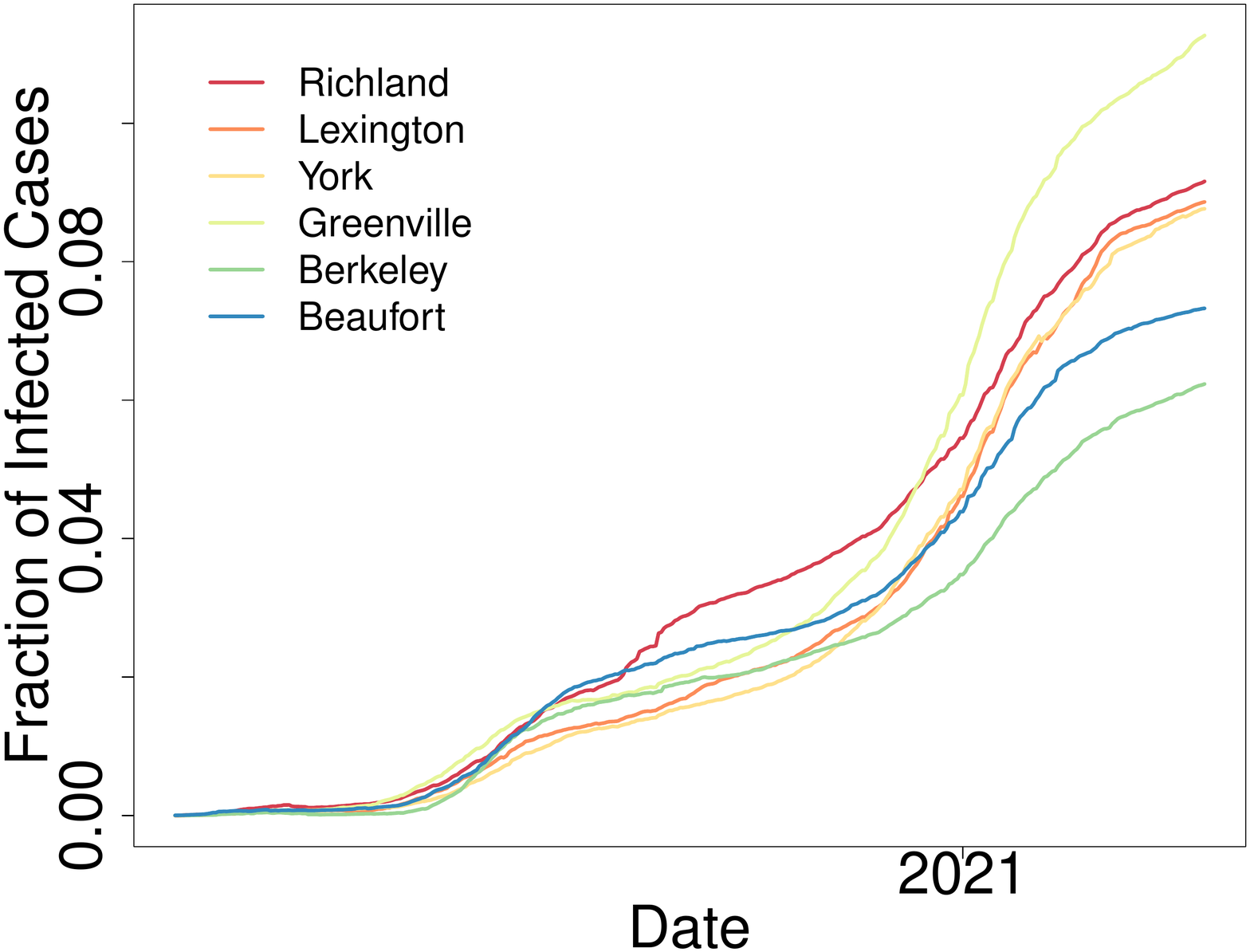}
         \subcaption{Richland}
     \end{subfigure}
     \begin{subfigure}[b]{0.19\textwidth}
         \centering
         \includegraphics[width=\textwidth]{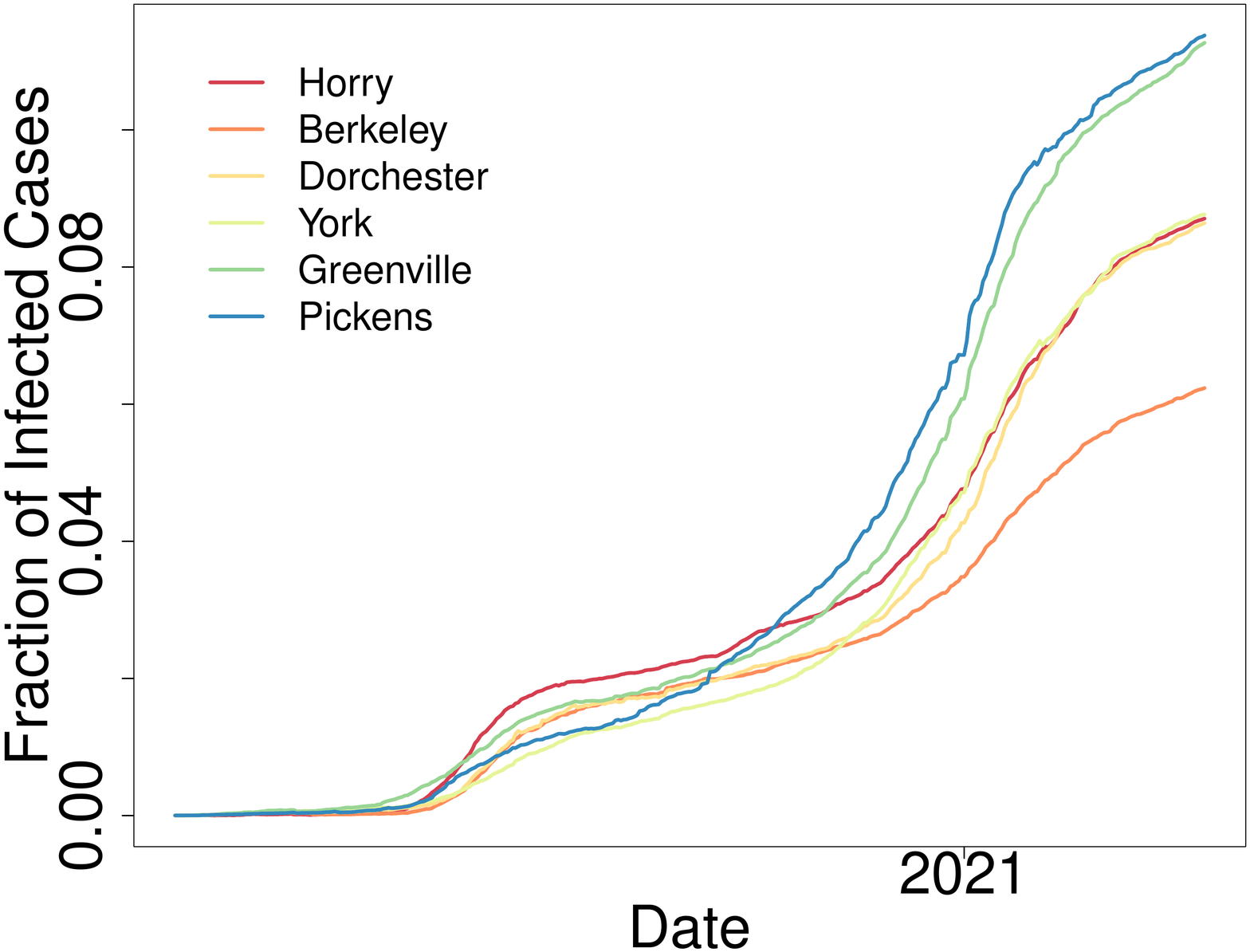}
         \subcaption{Horry}
     \end{subfigure}
     \begin{subfigure}[b]{0.19\textwidth}
         \centering
         \includegraphics[width=\textwidth]{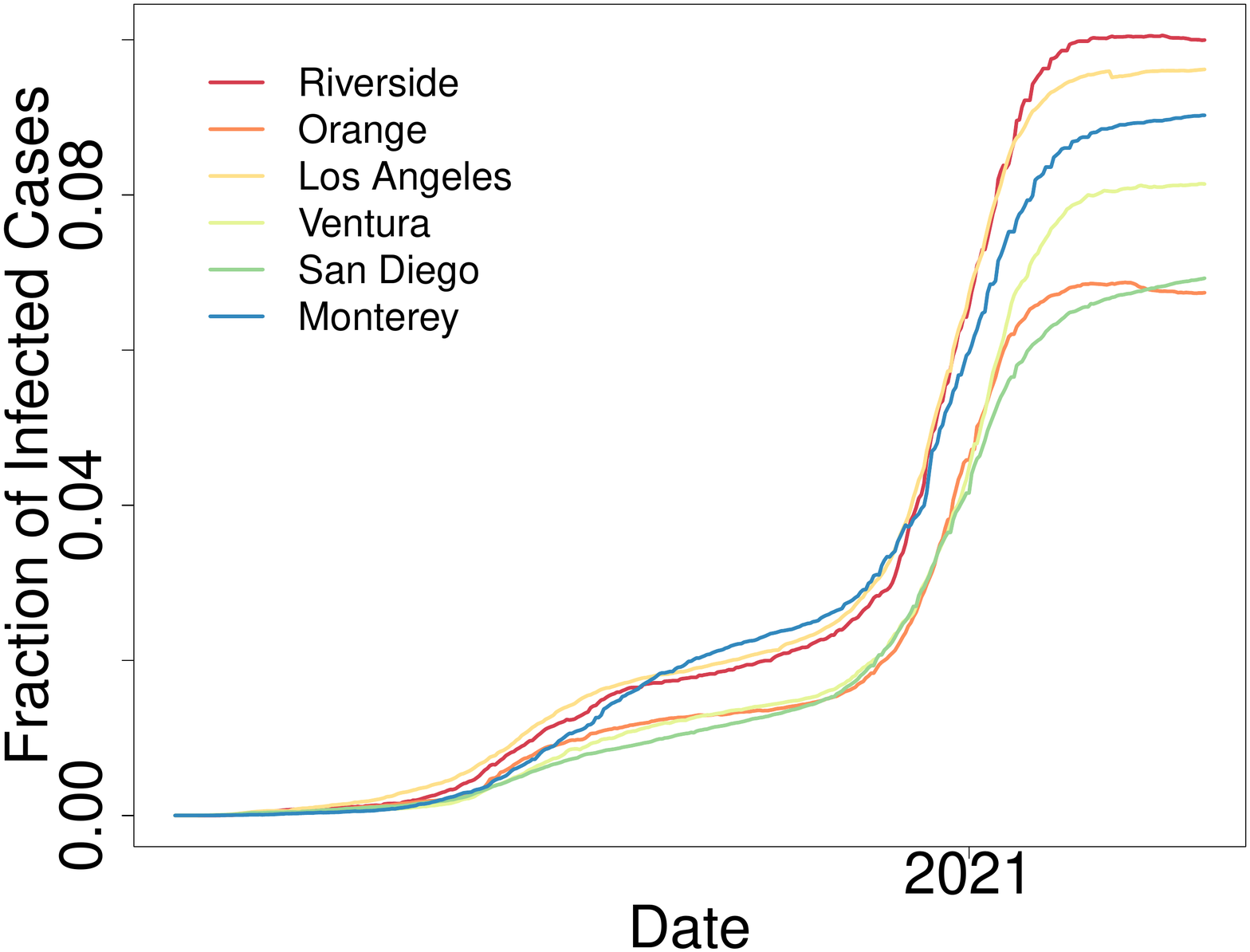}
         \subcaption{Riverside}
     \end{subfigure}
     \begin{subfigure}[b]{0.19\textwidth}
         \centering
         \includegraphics[width=\textwidth]{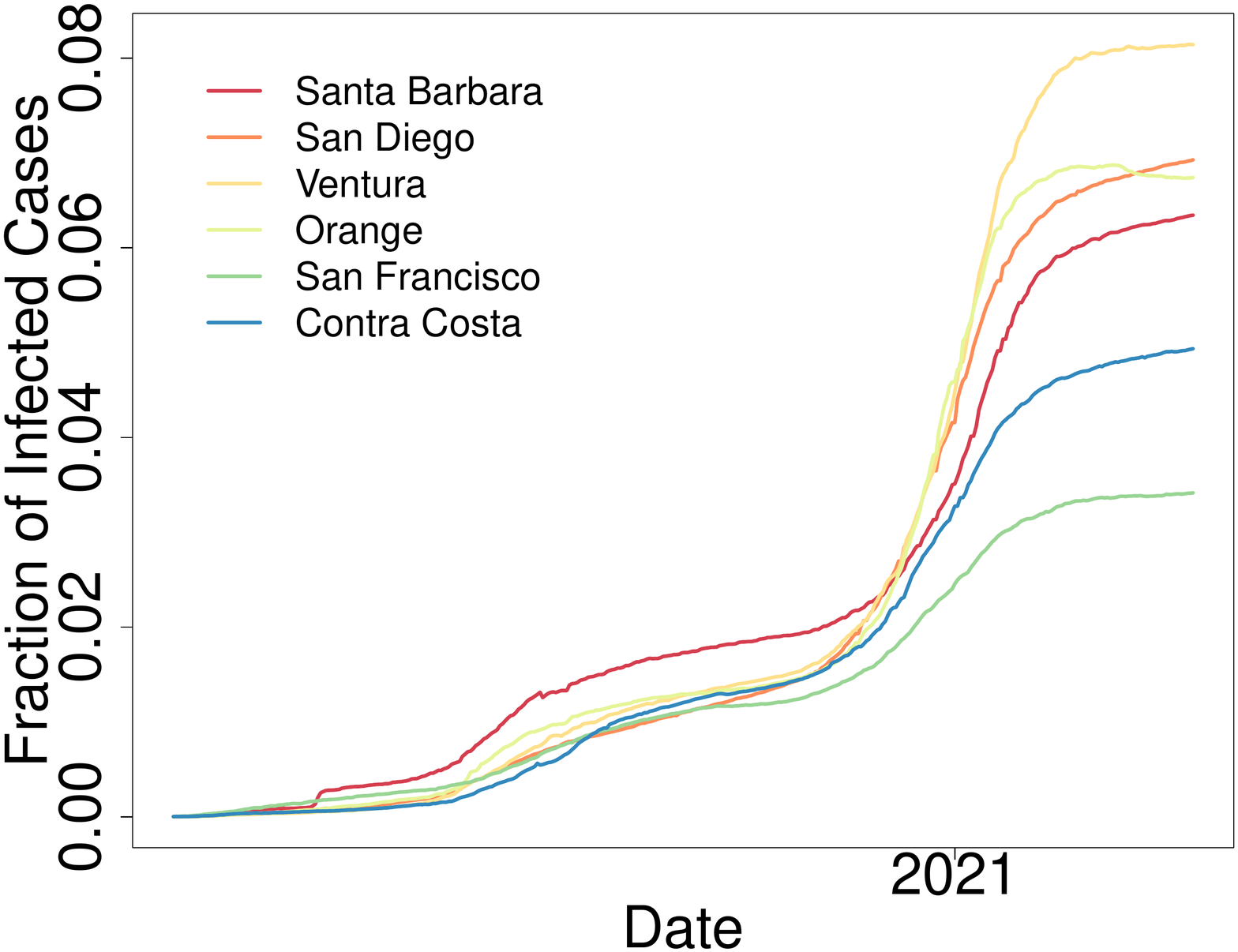}
         \subcaption{Santa Barbara}
     \end{subfigure}
    \caption{Fractions of infected in selected regions and their neighboring regions chosen by similarity score. }
        \label{fig:Infected_fraction_adj_2}
\end{figure*}

\begin{figure*}[ht!]
     \centering
      \captionsetup[sub]{font=small, labelfont={bf,sf}}
     \begin{subfigure}[b]{0.19\textwidth}
         \centering
         \includegraphics[width=\textwidth]{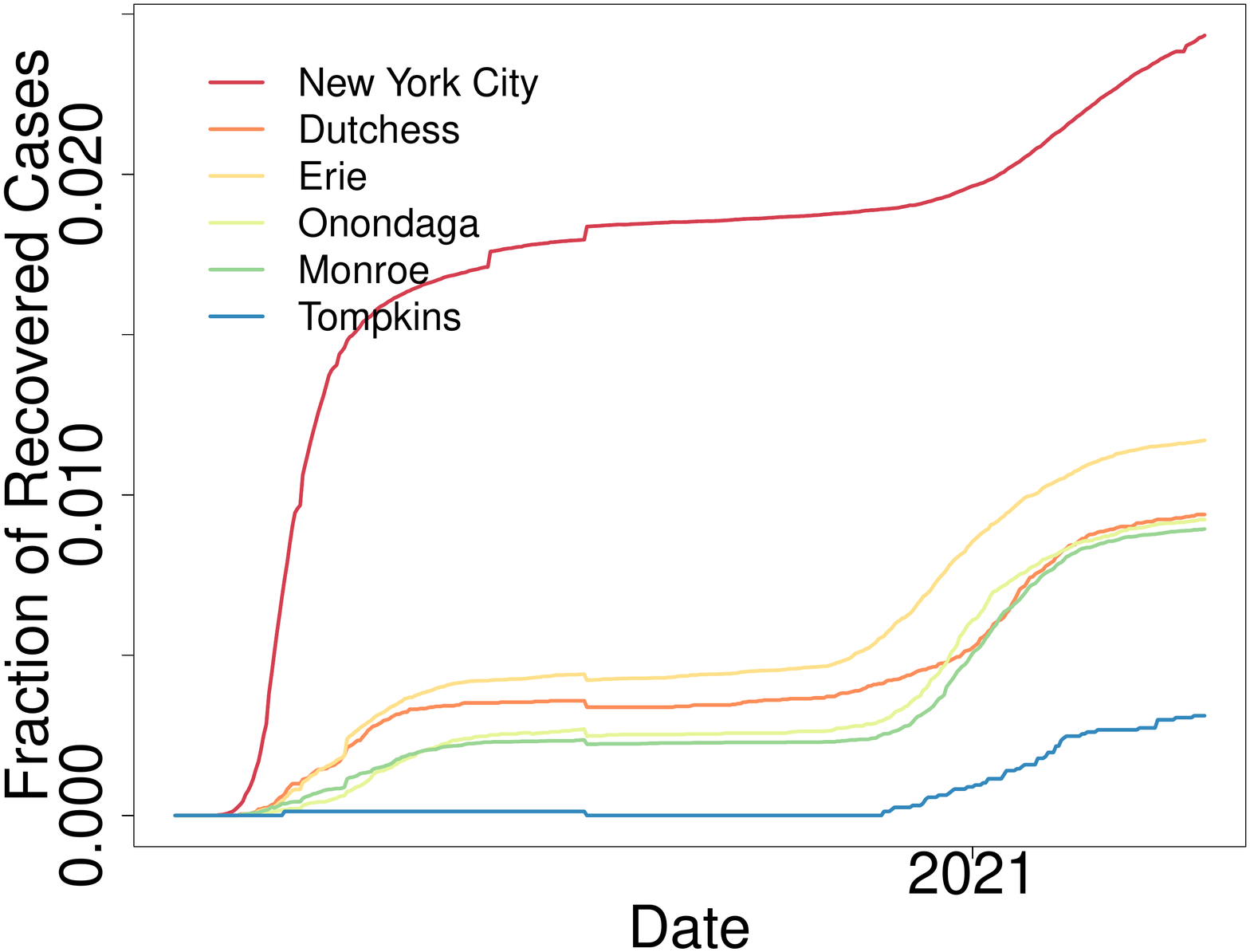}
         \subcaption{New York City}
     \end{subfigure}
     \begin{subfigure}[b]{0.19\textwidth}
         \centering
         \includegraphics[width=\textwidth]{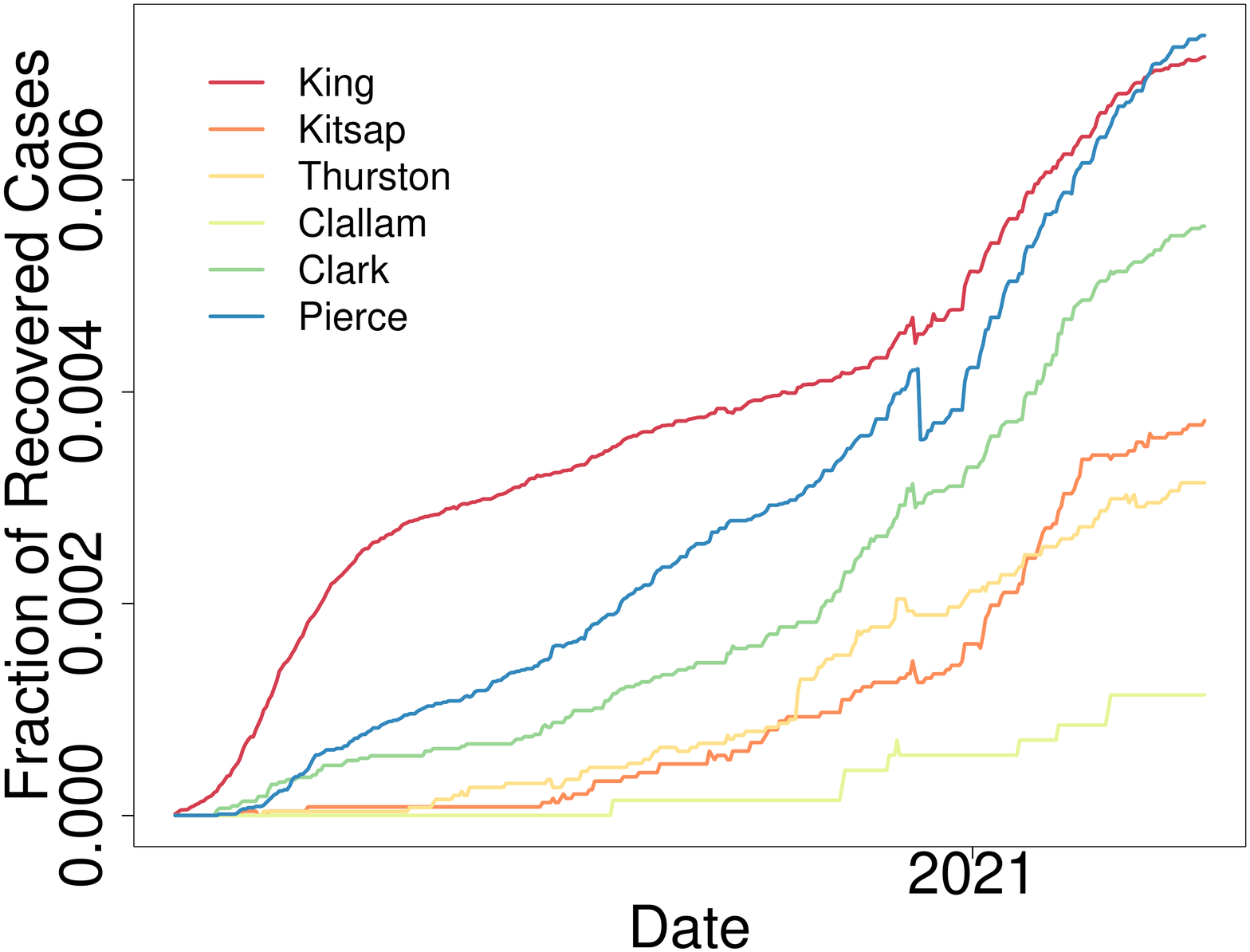}
         \subcaption{King County}
     \end{subfigure}
     \begin{subfigure}[b]{0.19\textwidth}
         \centering
         \includegraphics[width=\textwidth]{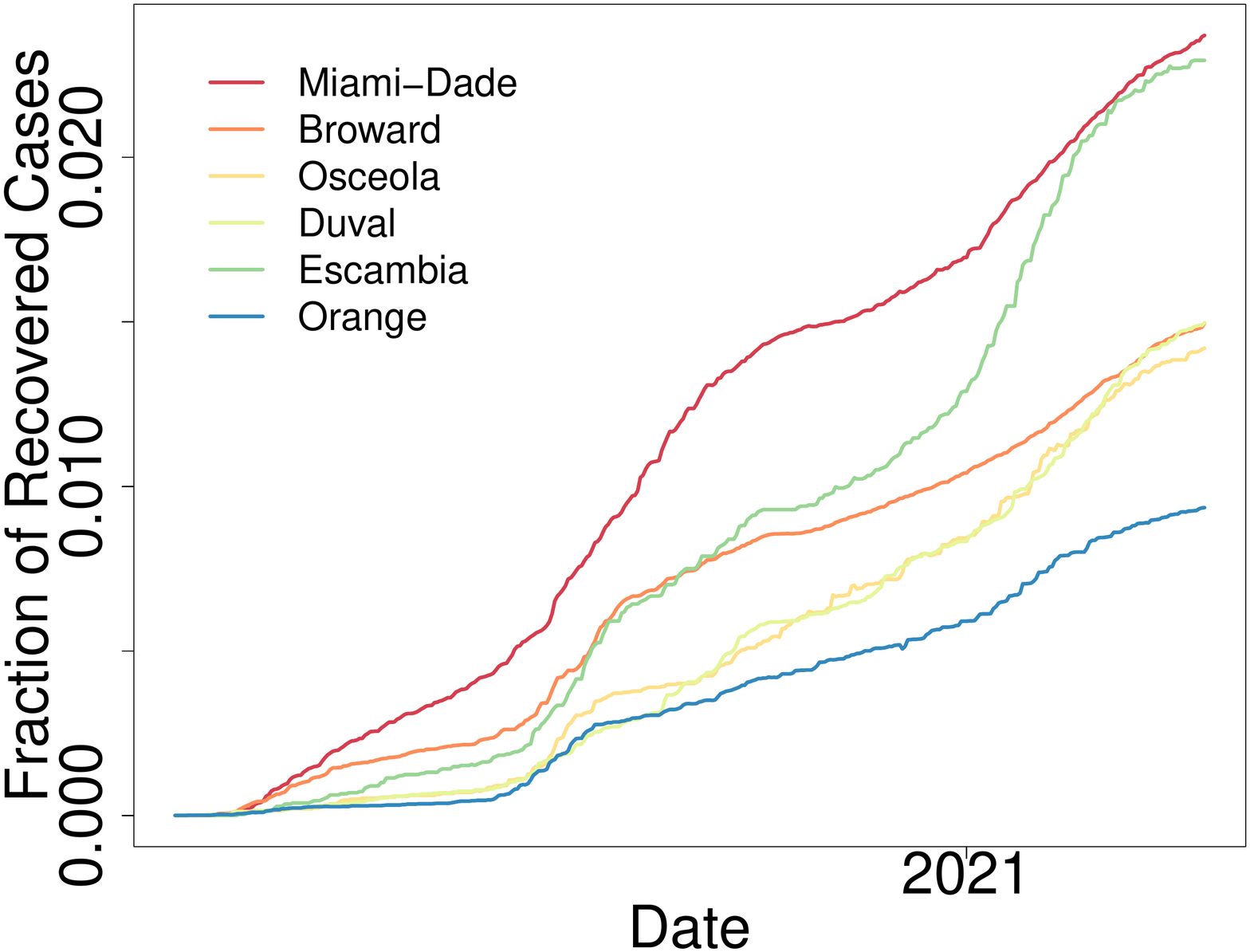}
         \subcaption{Miami-Dade}
     \end{subfigure}
     \begin{subfigure}[b]{0.19\textwidth}
         \centering
         \includegraphics[width=\textwidth]{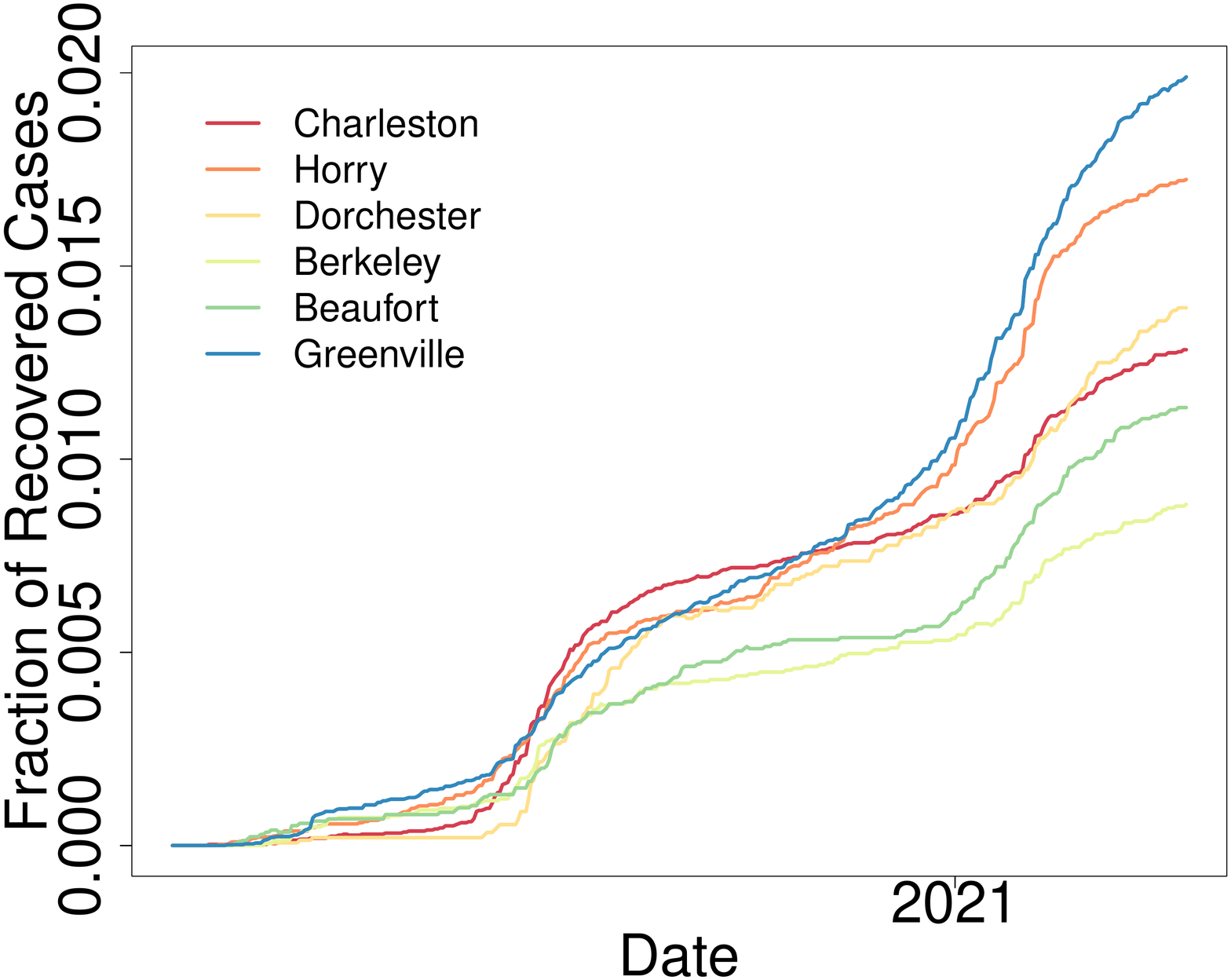}
         \subcaption{Charleston}
     \end{subfigure}
     \begin{subfigure}[b]{0.19\textwidth}
         \centering
         \includegraphics[width=\textwidth]{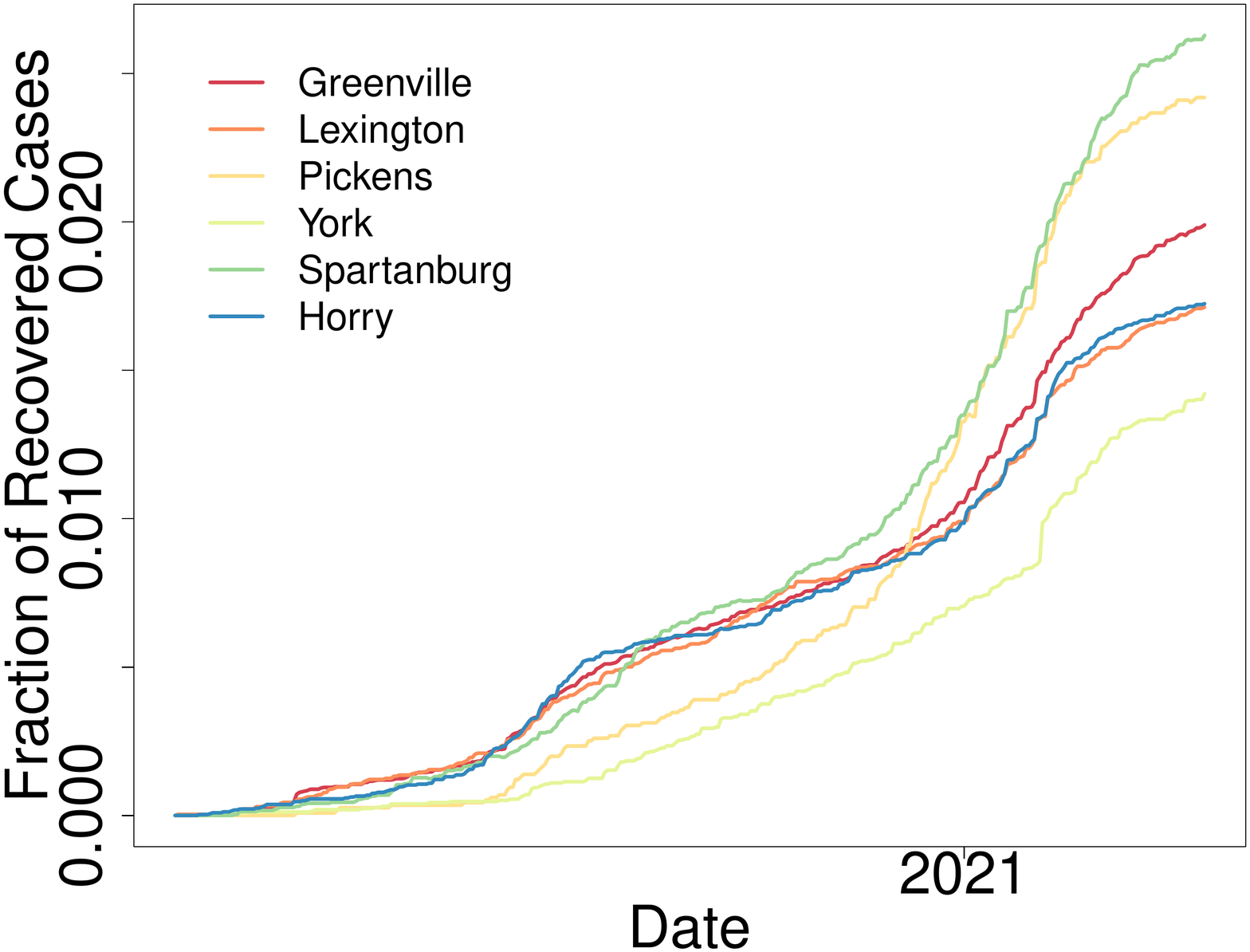}
         \subcaption{Greenville}
     \end{subfigure}
     \begin{subfigure}[b]{0.19\textwidth}
         \centering
         \includegraphics[width=\textwidth]{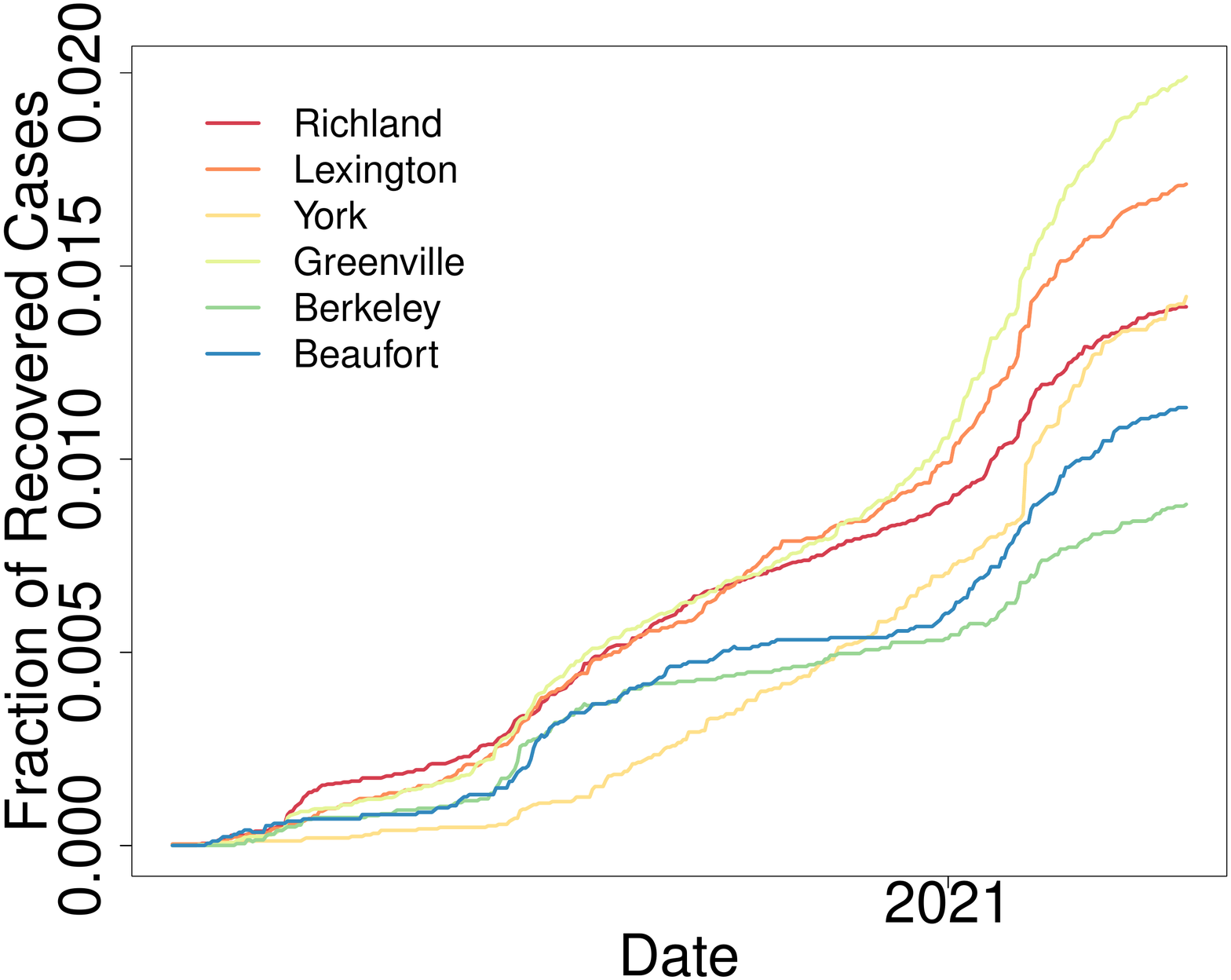}
         \subcaption{Richland}
     \end{subfigure}
     \begin{subfigure}[b]{0.19\textwidth}
         \centering
         \includegraphics[width=\textwidth]{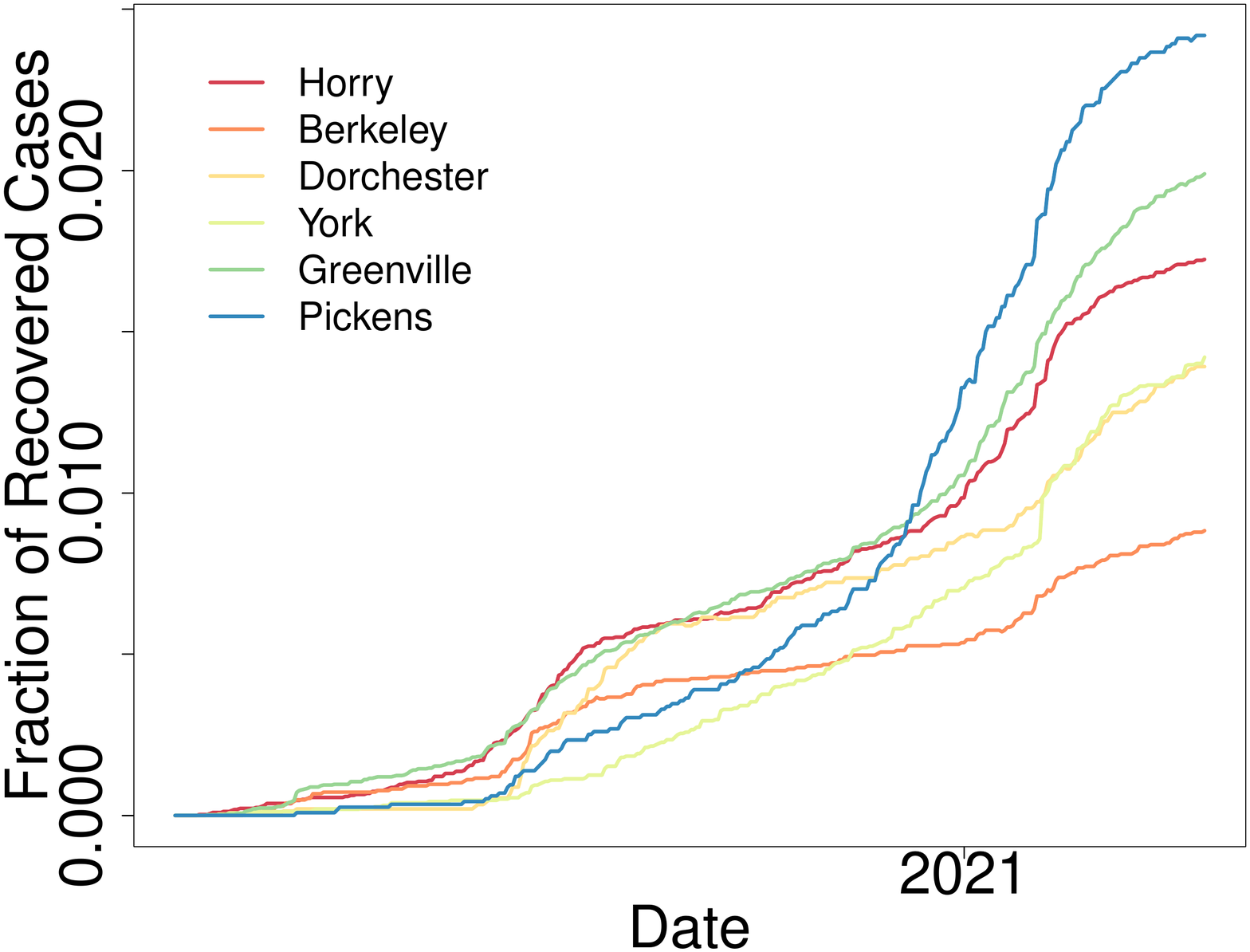}
         \subcaption{Horry}
     \end{subfigure}
     \begin{subfigure}[b]{0.19\textwidth}
         \centering
         \includegraphics[width=\textwidth]{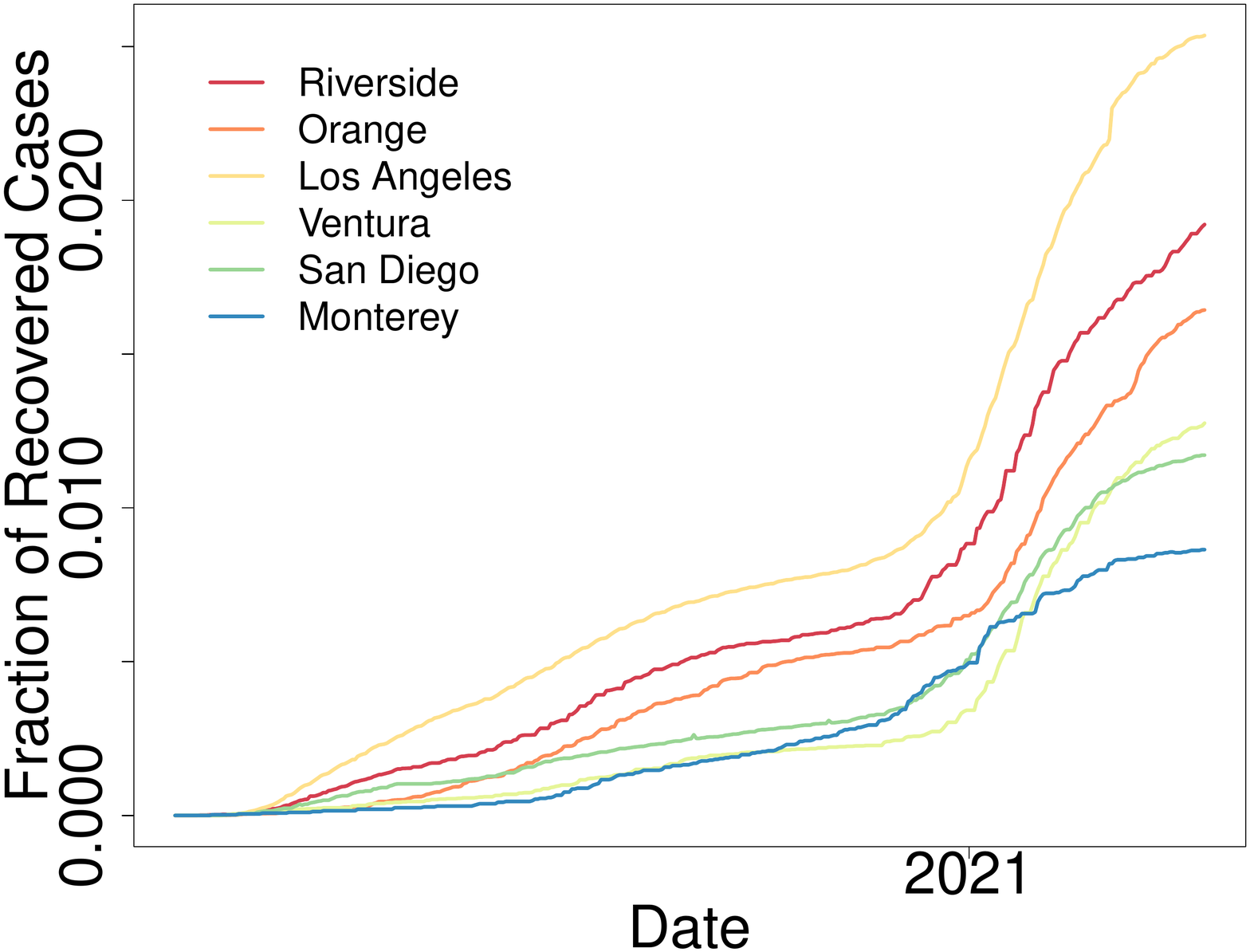}
         \subcaption{Riverside}
     \end{subfigure}
     \begin{subfigure}[b]{0.19\textwidth}
         \centering
         \includegraphics[width=\textwidth]{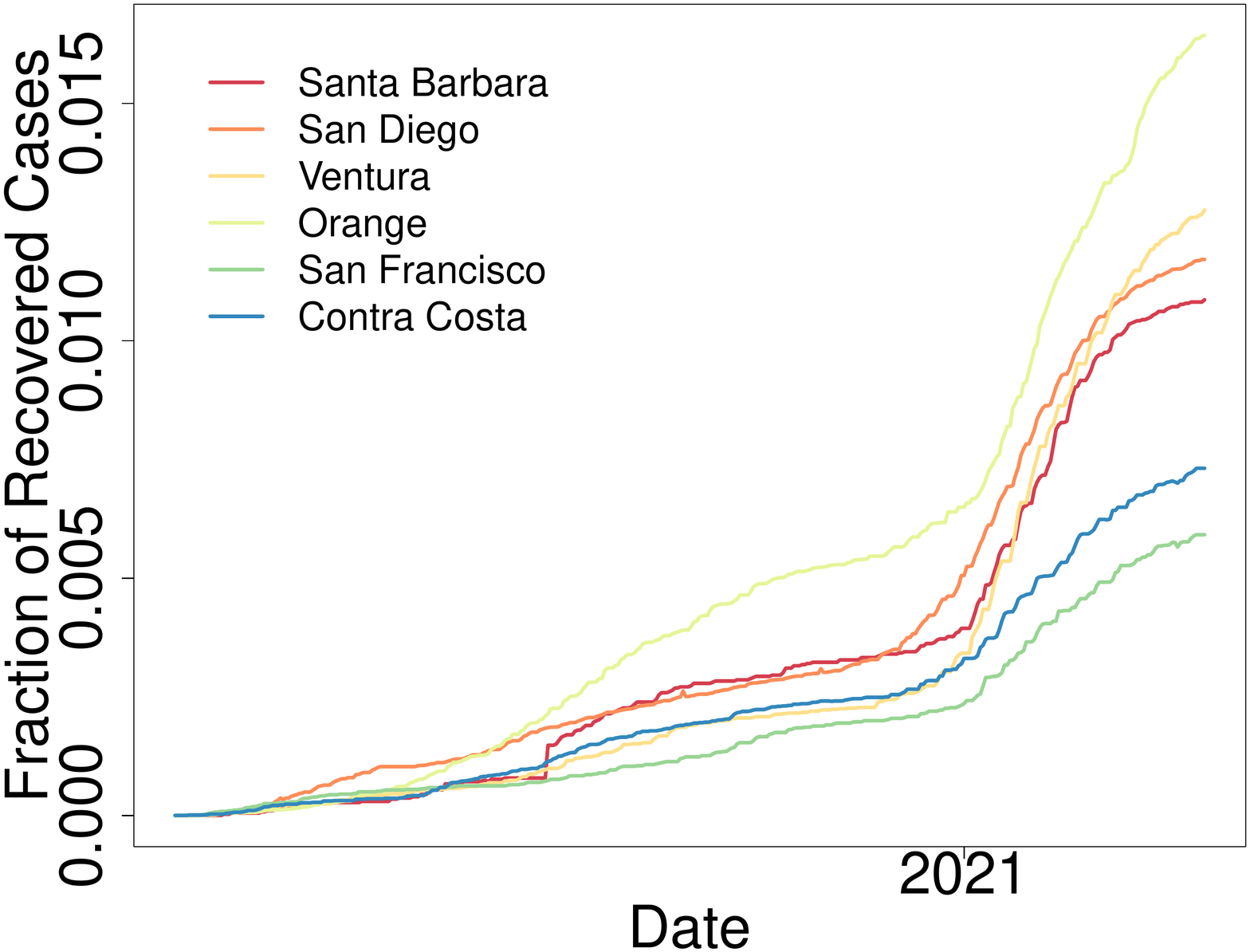}
         \subcaption{Santa Barbara}
     \end{subfigure}
        \caption{Fractions of  recovered in selected regions and their neighboring regions chosen by similarity score. }
        \label{fig:Recovered_fraction_adj_2}
\end{figure*}

Few additional plots and tables related to the results of applying our method to some U.S. counties are provided in this section. Figure \ref{fig:numbers_2} depicts the observed infected case numbers $I(t)$ and  recovered case numbers $R(t)$ in the nine counties/cities. In Figures \ref{fig:Infected_fraction_adj_2} and \ref{fig:Recovered_fraction_adj_2}, we provide fractions of infected and recovered cases in given counties and their neighboring regions selected by similarity score in the Model 2.3.

\begin{figure*}[ht!]
     \centering
     \captionsetup[sub]{font=small, labelfont={bf,sf}}
      \begin{subfigure}[b]{0.24\textwidth}
         \centering
         \includegraphics[width=\textwidth]{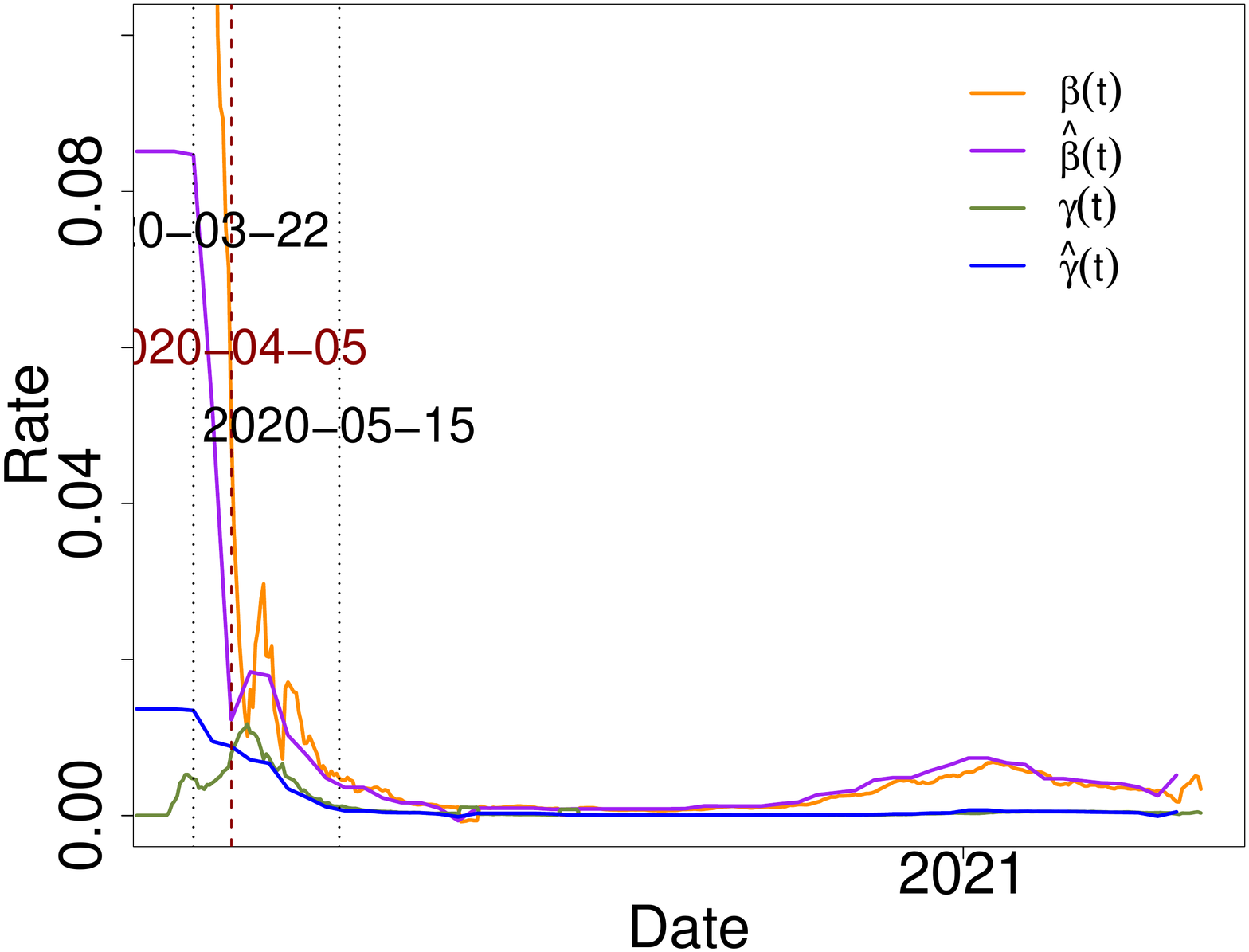}
         \subcaption{NYC  (Model 1) }
     \end{subfigure}
     \begin{subfigure}[b]{0.24\textwidth}
         \centering
         \includegraphics[width=\textwidth]{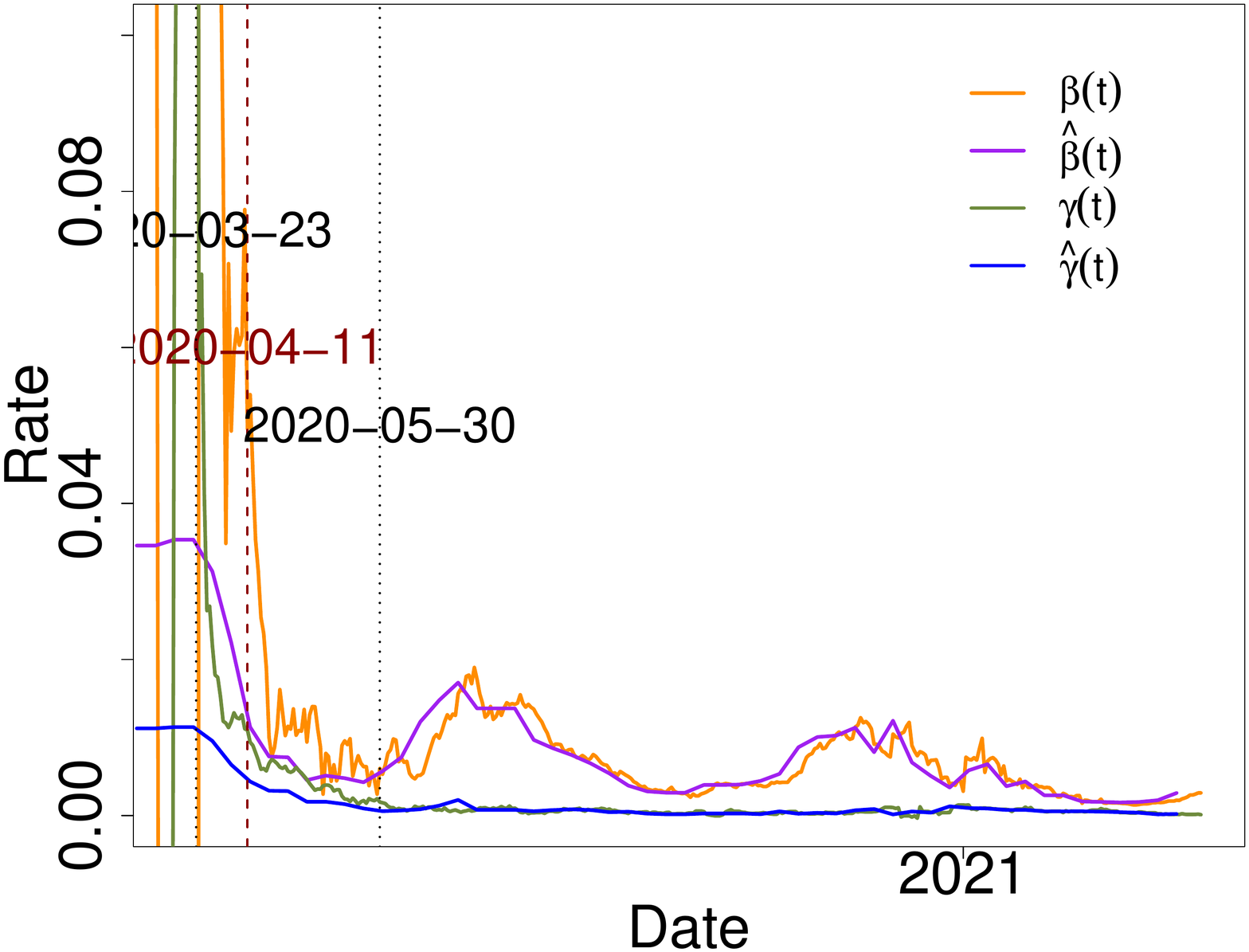}
         \subcaption{King  (Model 1) }
     \end{subfigure}
     \begin{subfigure}[b]{0.24\textwidth}
         \centering
         \includegraphics[width=\textwidth]{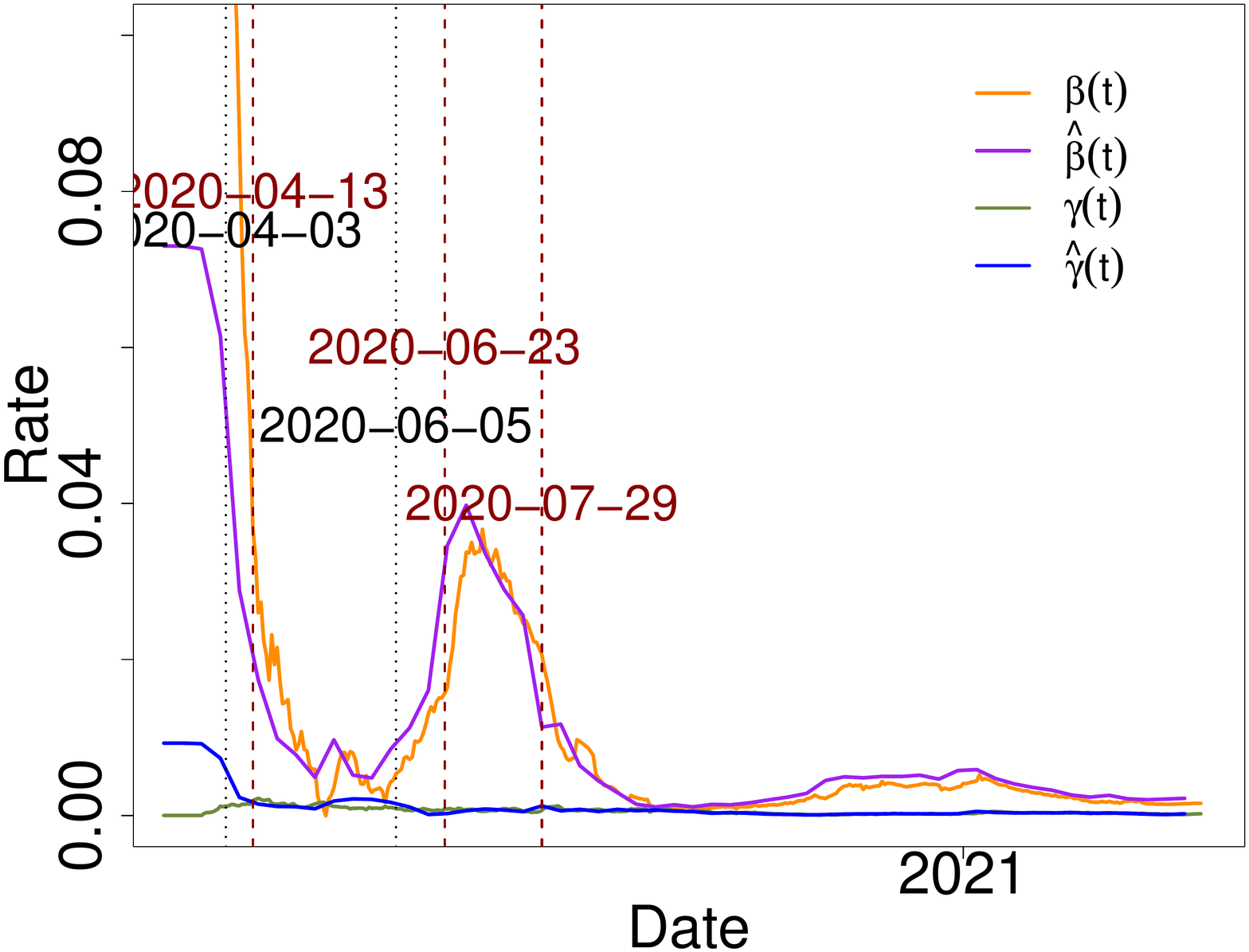}
         \subcaption{Miami-Dade (Model 1) }
     \end{subfigure}
     
     \begin{subfigure}[b]{0.24\textwidth}
         \centering
         \includegraphics[width=\textwidth]{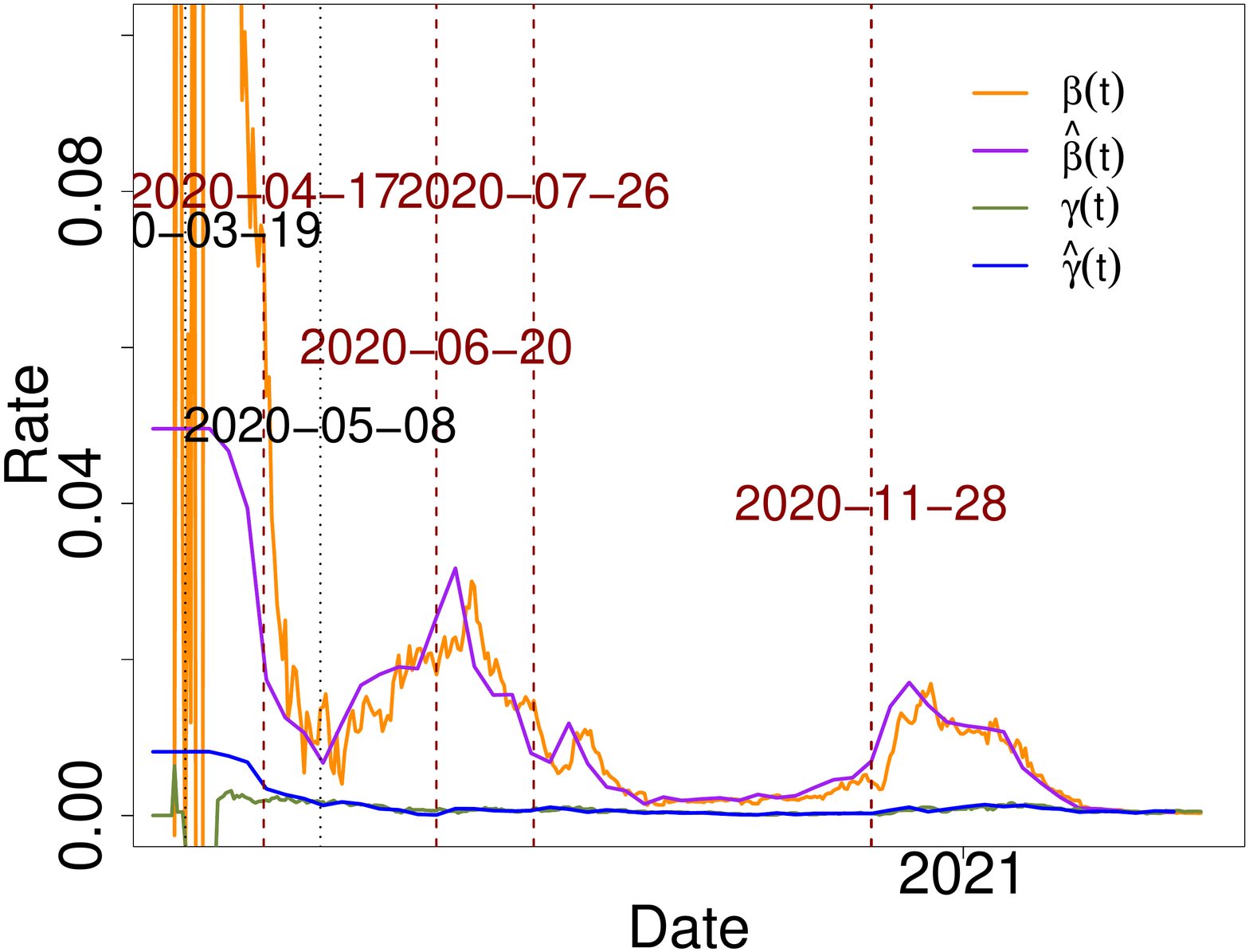}
         \subcaption{Riverside (Model 1) }
     \end{subfigure}
     \begin{subfigure}[b]{0.24\textwidth}
         \centering
         \includegraphics[width=\textwidth]{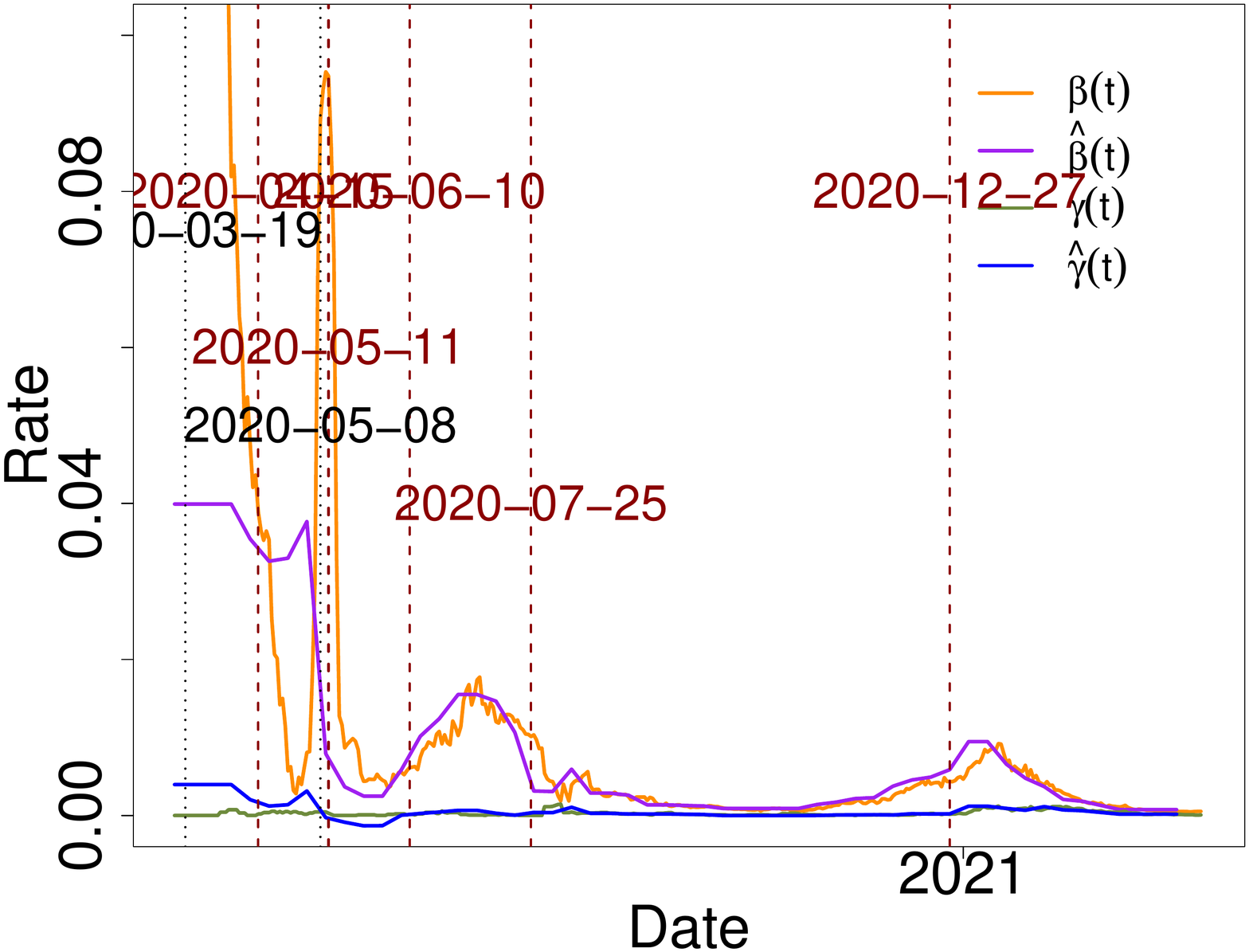}
         \subcaption{Santa Barbara (Model 1) }
     \end{subfigure}
     \begin{subfigure}[b]{0.24\textwidth}
         \centering
         \includegraphics[width=\textwidth]{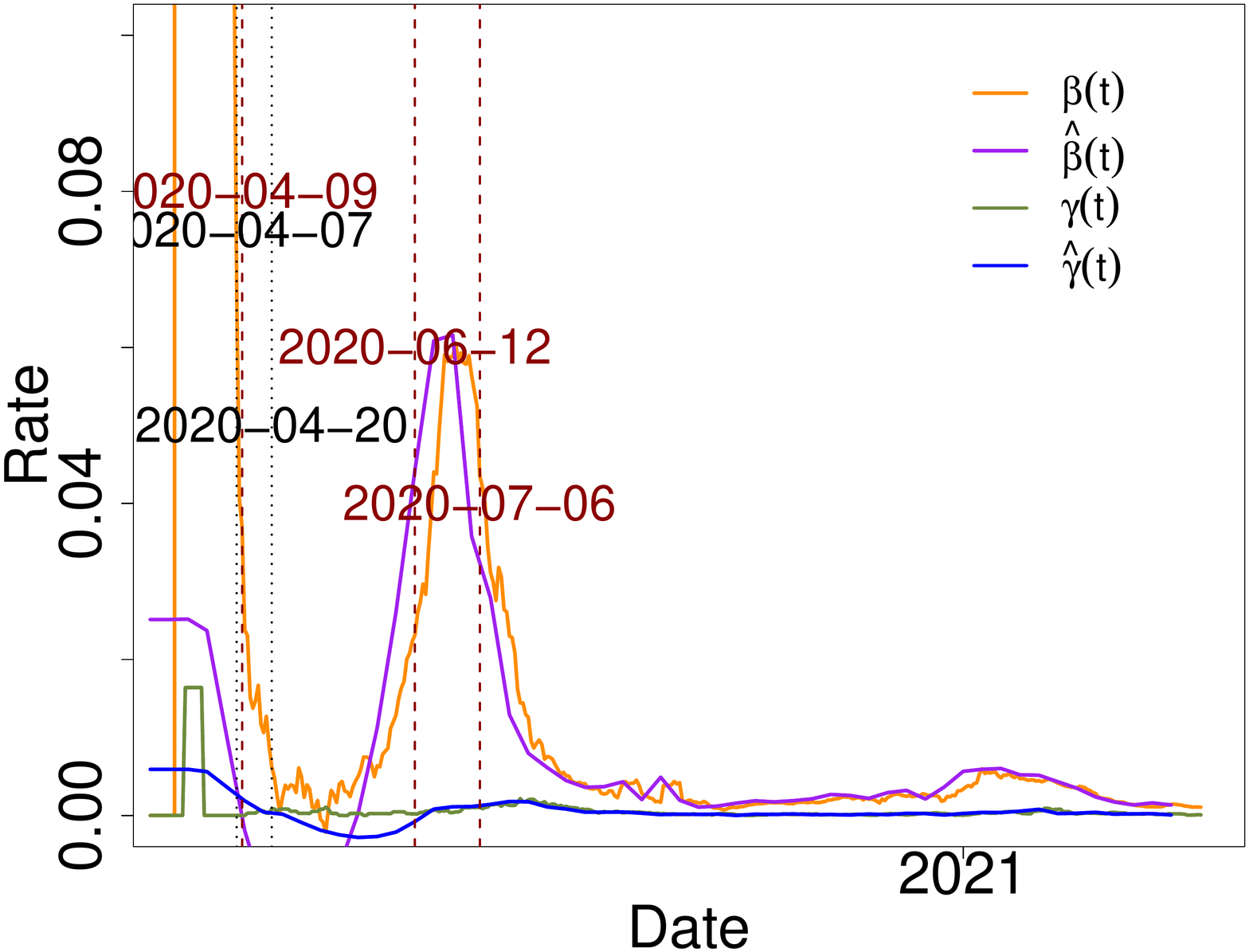}
         \subcaption{Charleston(Model 1) }
     \end{subfigure}
     \begin{subfigure}[b]{0.24\textwidth}
         \centering
         \includegraphics[width=\textwidth]{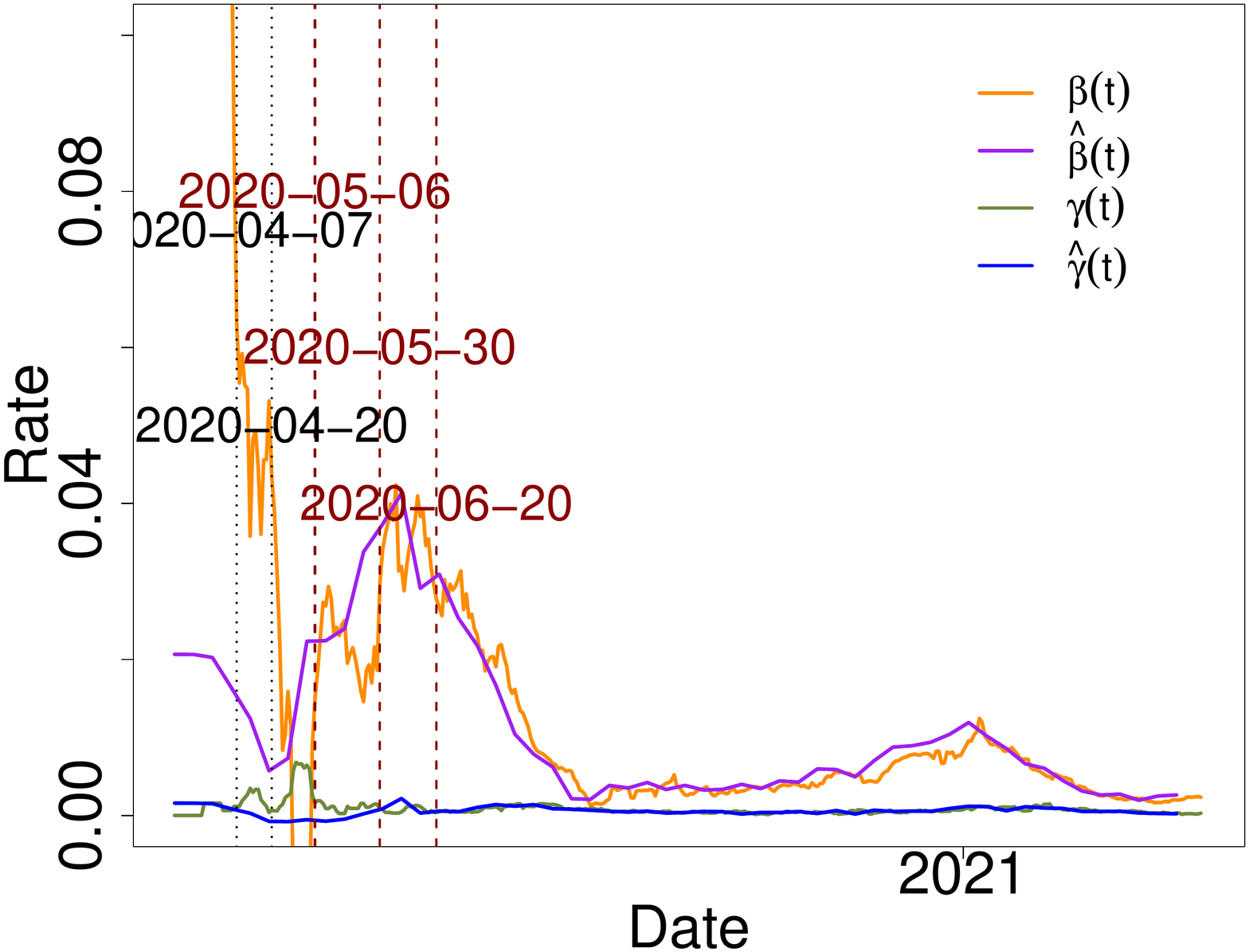}
         \subcaption{Greenville (Model 1) }
     \end{subfigure}
     \begin{subfigure}[b]{0.24\textwidth}
         \centering
         \includegraphics[width=\textwidth]{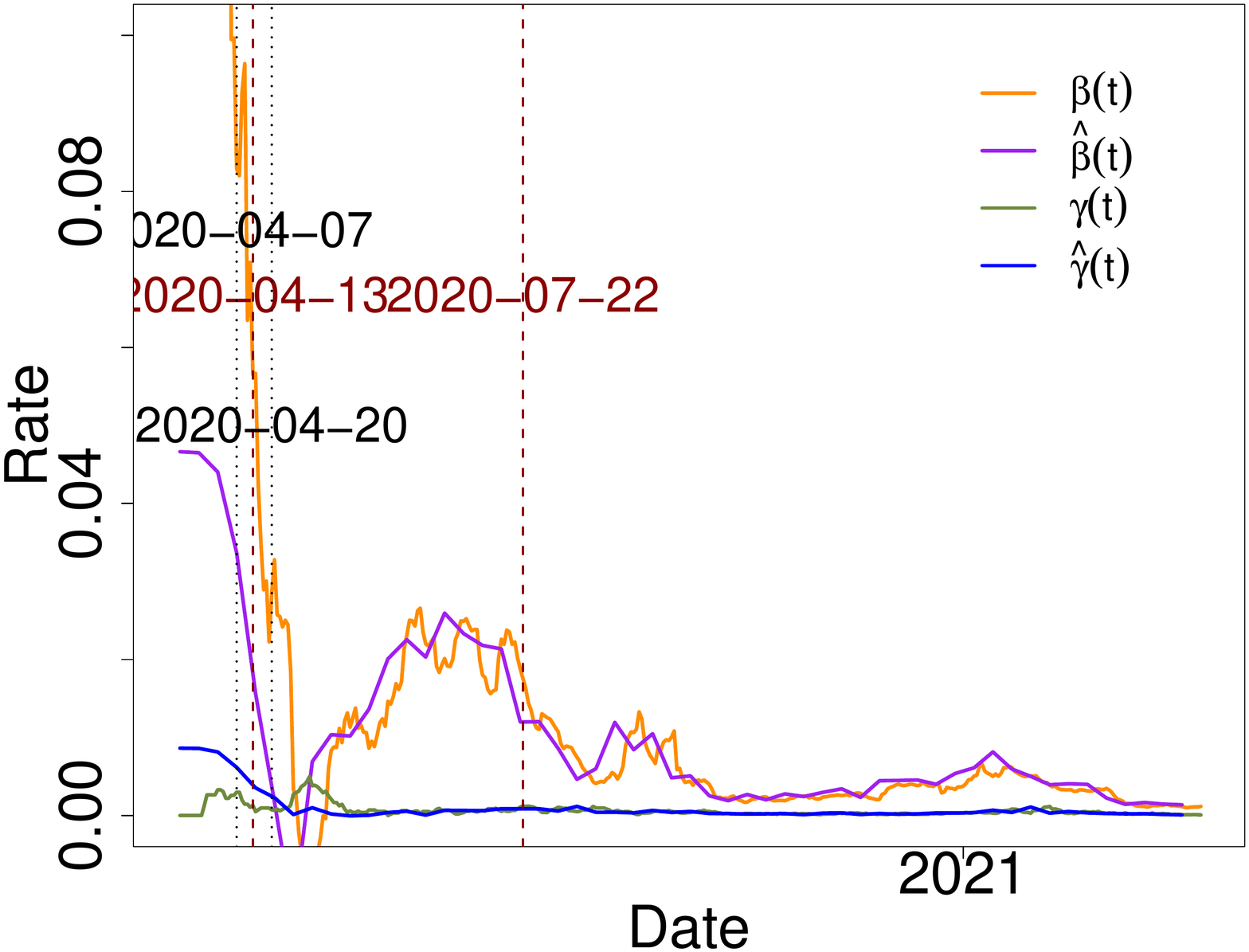}
         \subcaption{Richland  (Model 1) }
     \end{subfigure}
      \begin{subfigure}[b]{0.24\textwidth}
         \centering
         \includegraphics[width=\textwidth]{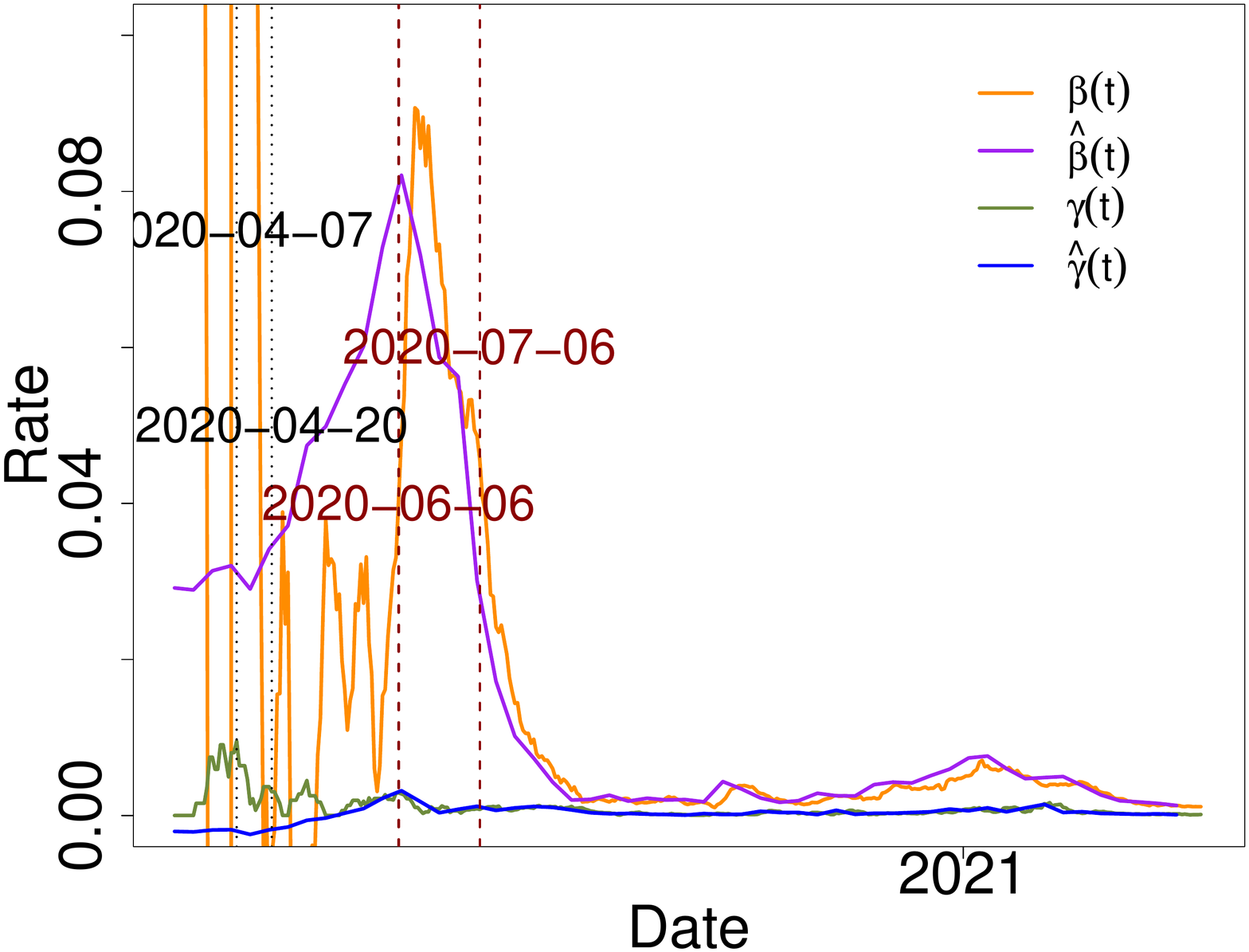}
         \subcaption{Horry  (Model 1) }
     \end{subfigure}
      \caption{The 7-day moving average of observed (orange and green) and estimated (purple and blue) transmission rate and recovery rate in counties/cities. The vertical black dotted line indicates the statewide ``stay-at-home'' order begin date or reopening  begin date in the corresponding state.
        The vertical dark red dashed line indicates the estimated change time point in the county/city. 
        }
        \label{fig:rates_county_smooth}
\end{figure*}


\begin{figure*}[ht!]
     \centering
      \captionsetup[sub]{font=scriptsize,labelfont={bf,sf}}
    \begin{subfigure}[b]{0.19\textwidth}
         \centering
         \includegraphics[width=\textwidth]{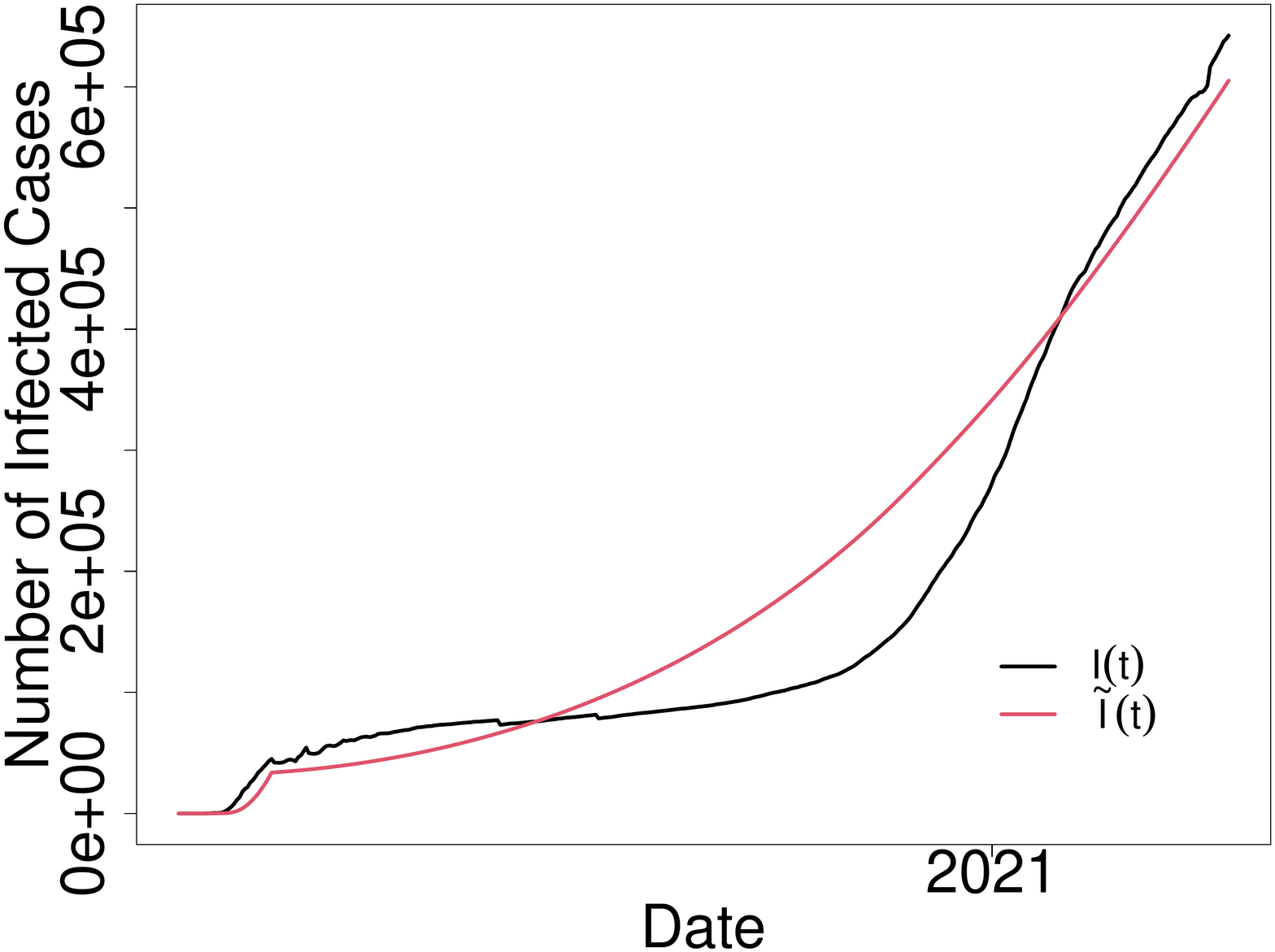}
          \subcaption{NYC (Model 1)}
     \end{subfigure} 
  \begin{subfigure}[b]{0.19\textwidth}
         \centering
         \includegraphics[width=\textwidth]{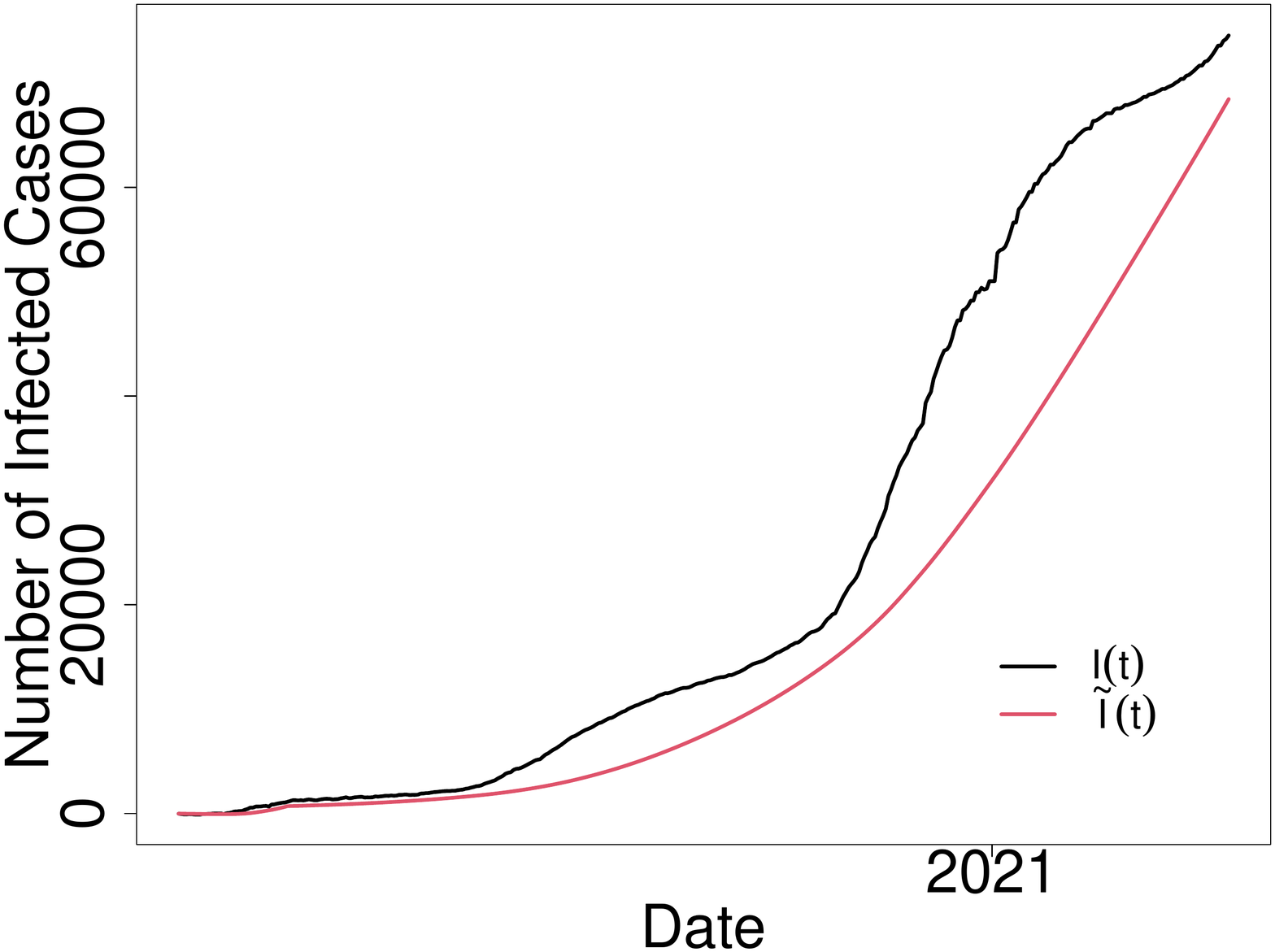}
         \subcaption{King (Model 1)}
     \end{subfigure}  
     \begin{subfigure}[b]{0.19\textwidth}
         \centering
         \includegraphics[width=\textwidth]{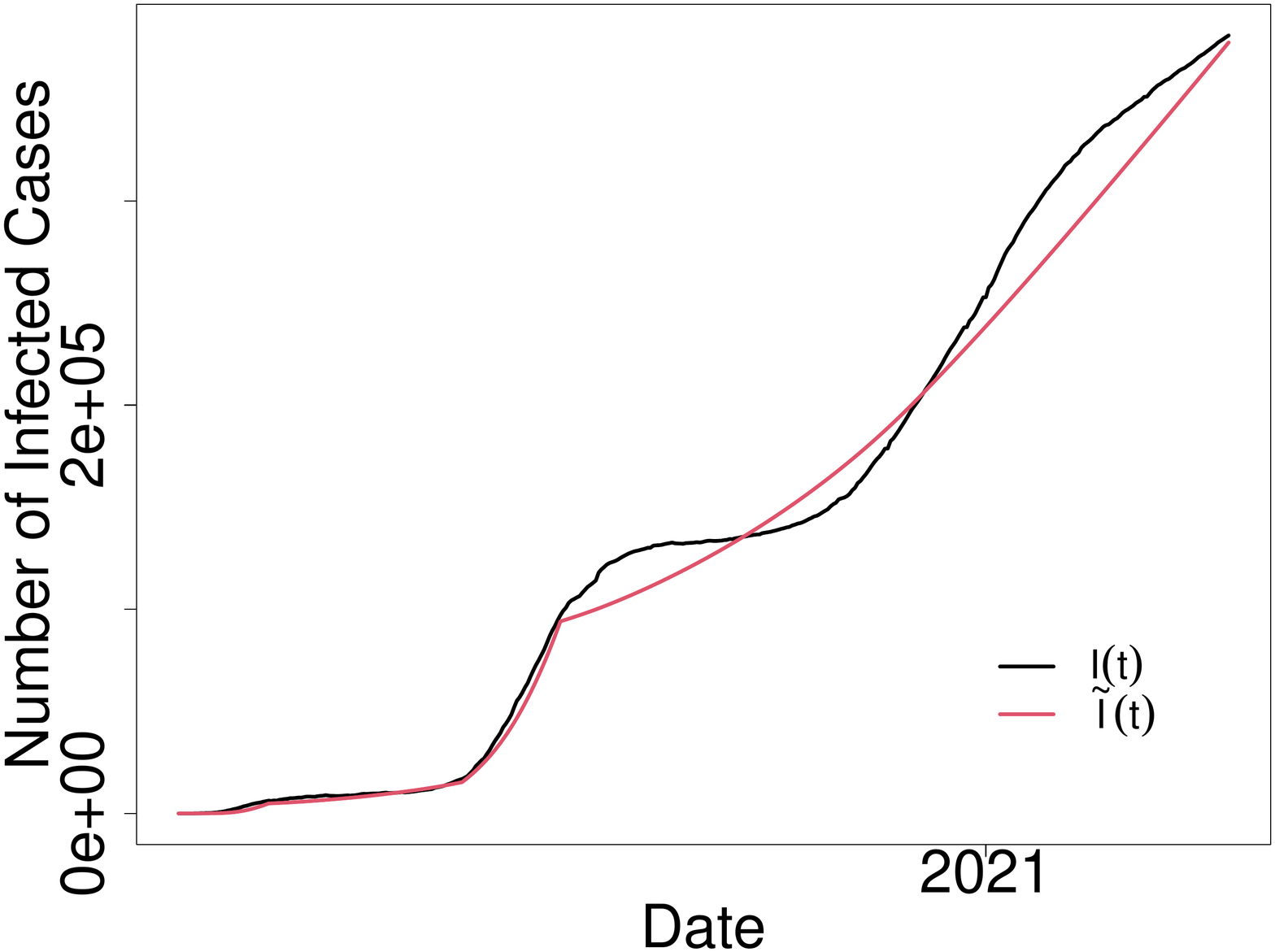}
         \subcaption{Miami-Dade (Model 1)}
     \end{subfigure}  
\begin{subfigure}[b]{0.19\textwidth}
         \centering
         \includegraphics[width=\textwidth]{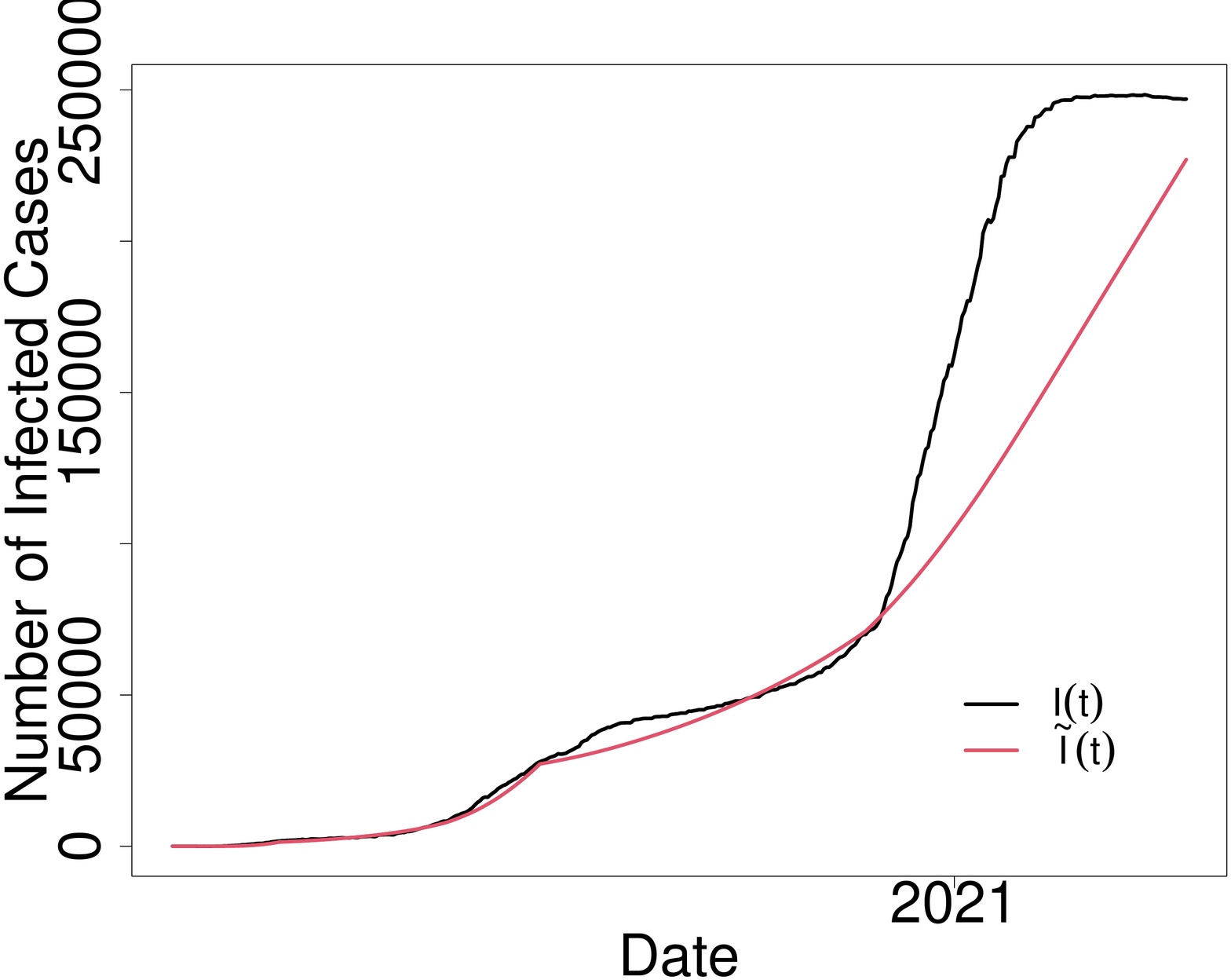}
         \subcaption{Riverside (Model 1)}
     \end{subfigure}
     \begin{subfigure}[b]{0.19\textwidth}
         \centering
         \includegraphics[width=\textwidth]{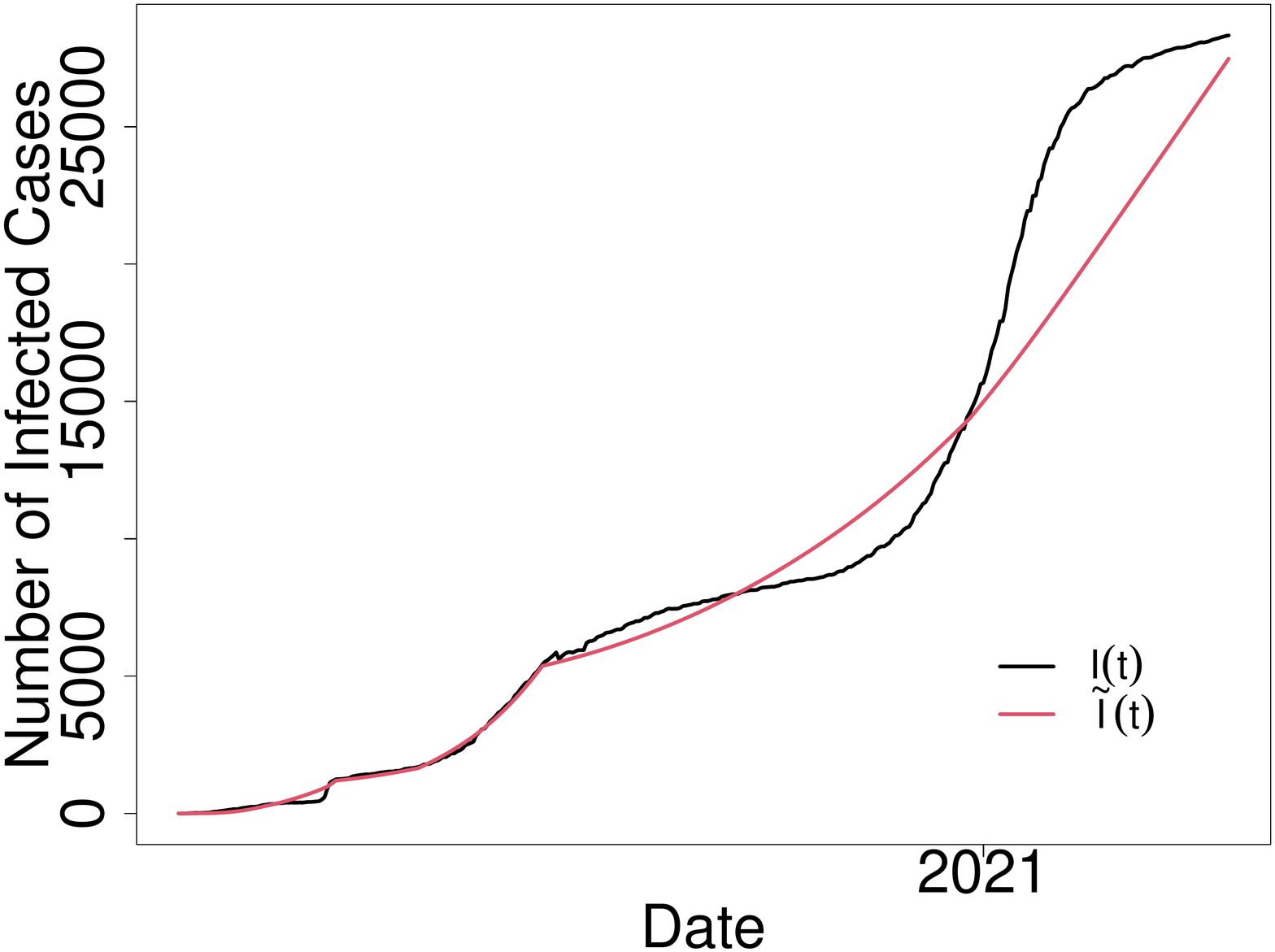}
         \subcaption{Santa Barbara (Model 1)}
     \end{subfigure}
    
      \begin{subfigure}[b]{0.19\textwidth}
         \centering
         \includegraphics[width=\textwidth]{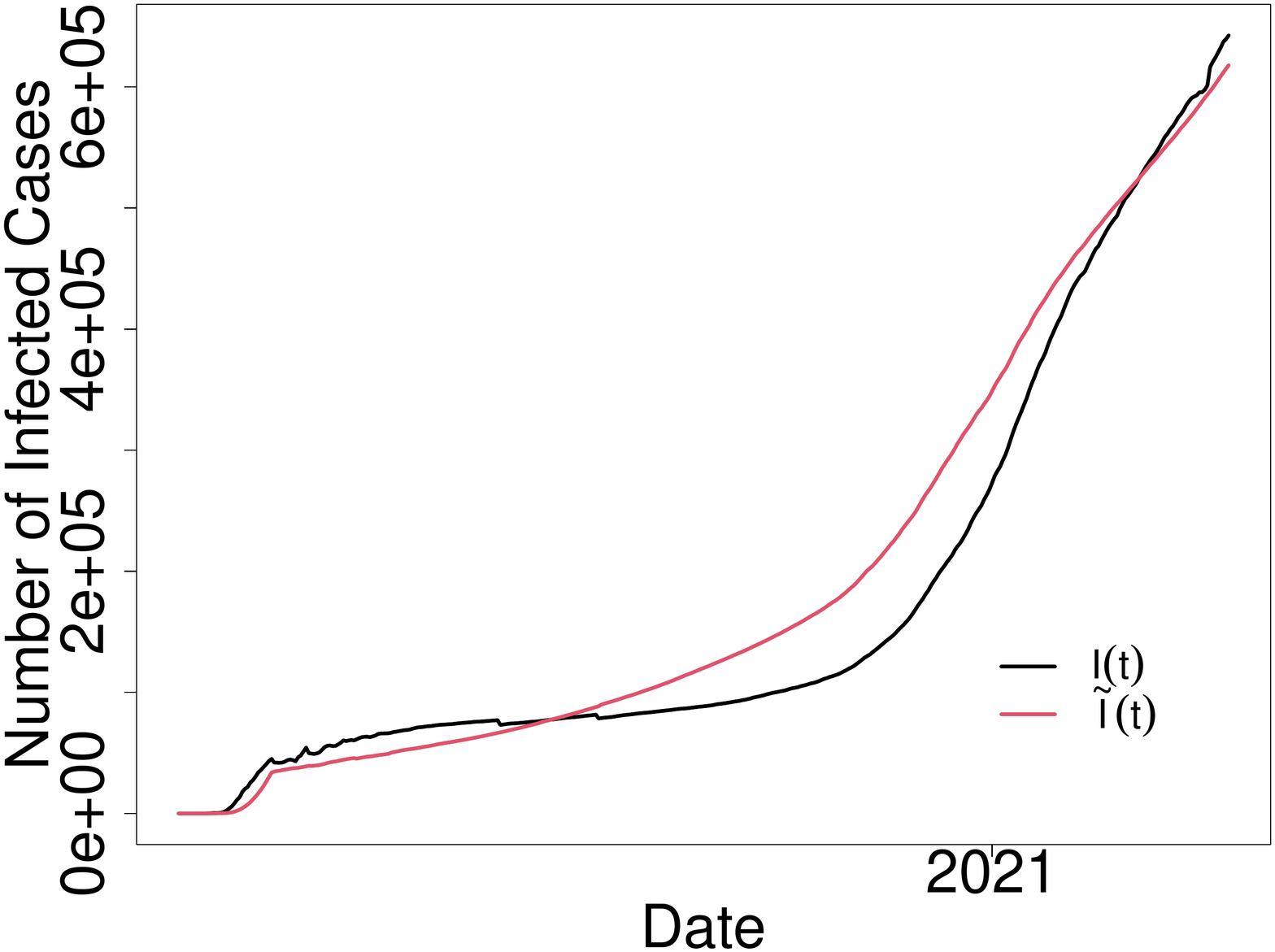}
          \subcaption{NYC (Model 2.3)}
     \end{subfigure}
     \begin{subfigure}[b]{0.19\textwidth}
         \centering
         \includegraphics[width=\textwidth]{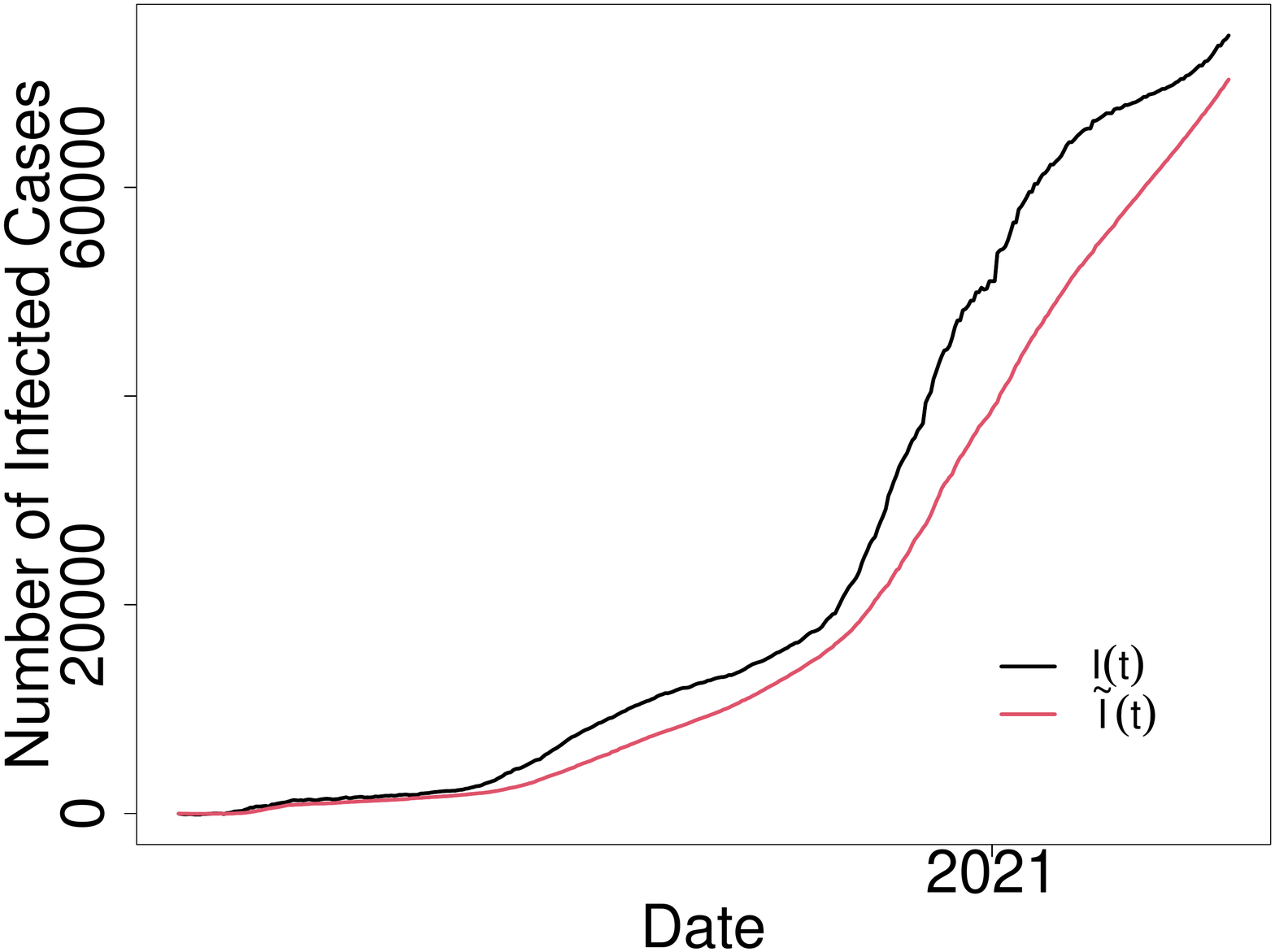}
         \subcaption{King (Model 2.3)}
     \end{subfigure}
     \begin{subfigure}[b]{0.19\textwidth}
         \centering
         \includegraphics[width=\textwidth]{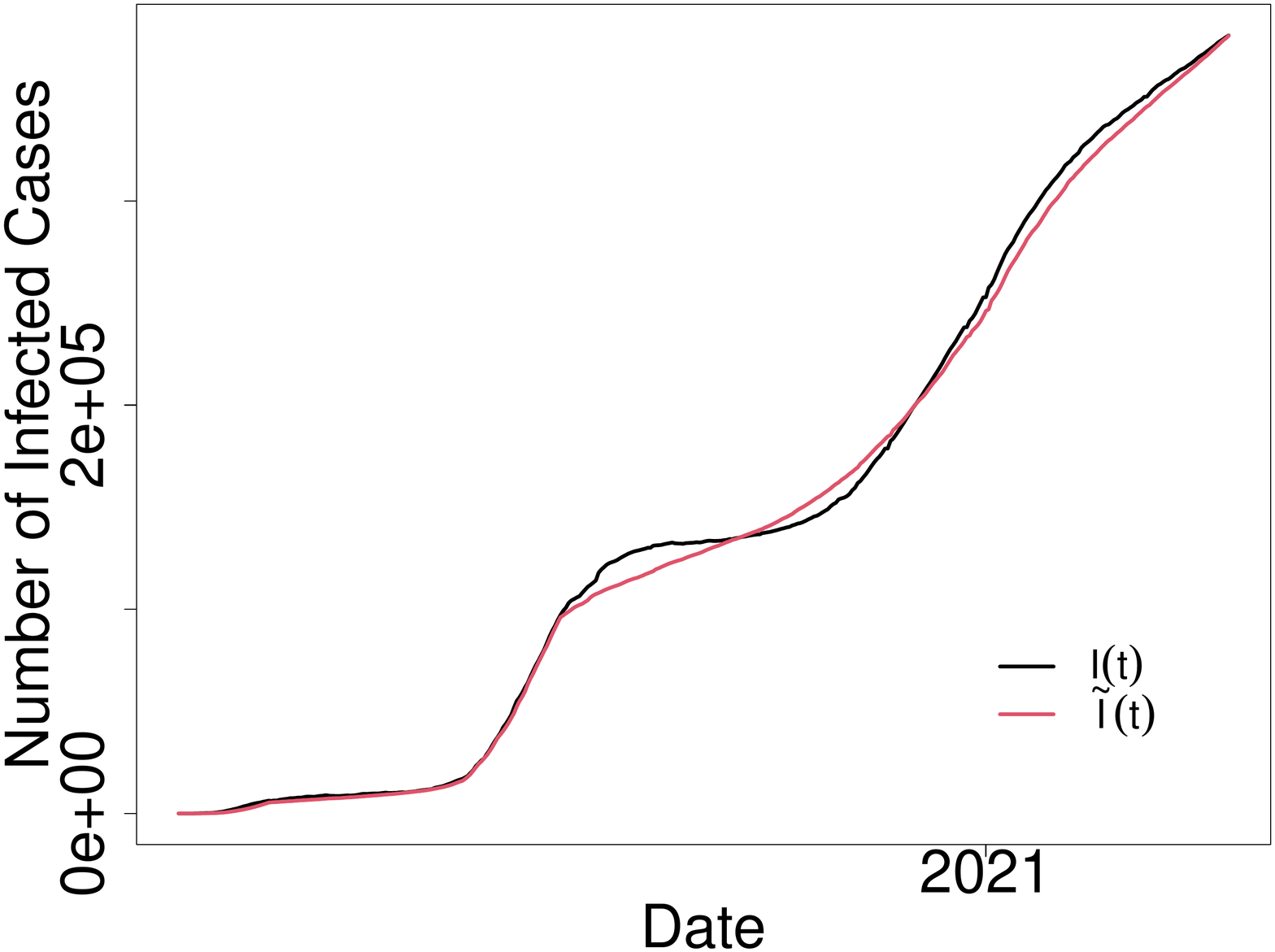}
         \subcaption{Miami-Dade (Model 2.3)}
     \end{subfigure} 
     \begin{subfigure}[b]{0.19\textwidth}
         \centering
         \includegraphics[width=\textwidth]{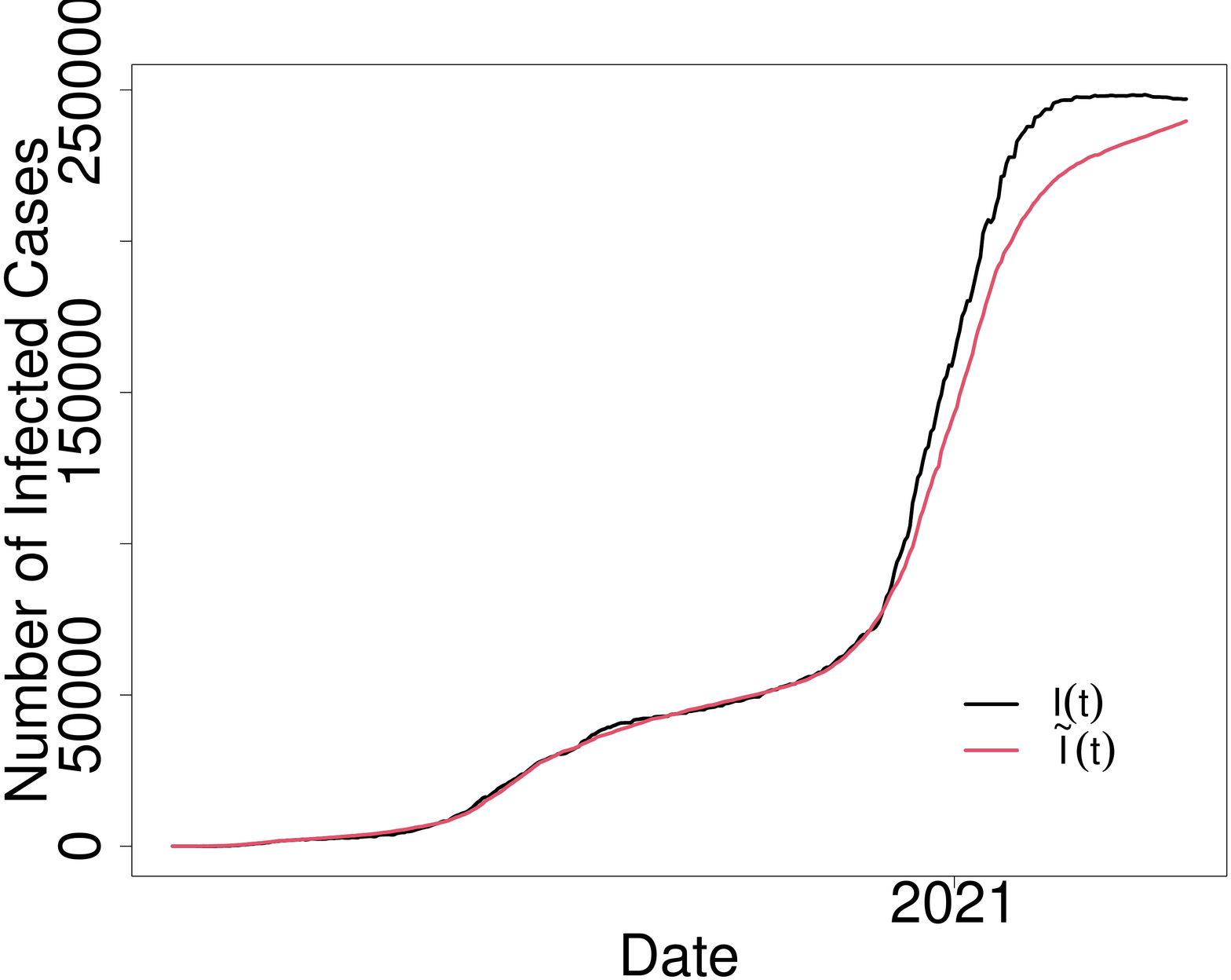}
         \subcaption{Riverside (Model 2.3)}
     \end{subfigure}
     \begin{subfigure}[b]{0.19\textwidth}
         \centering
         \includegraphics[width=\textwidth]{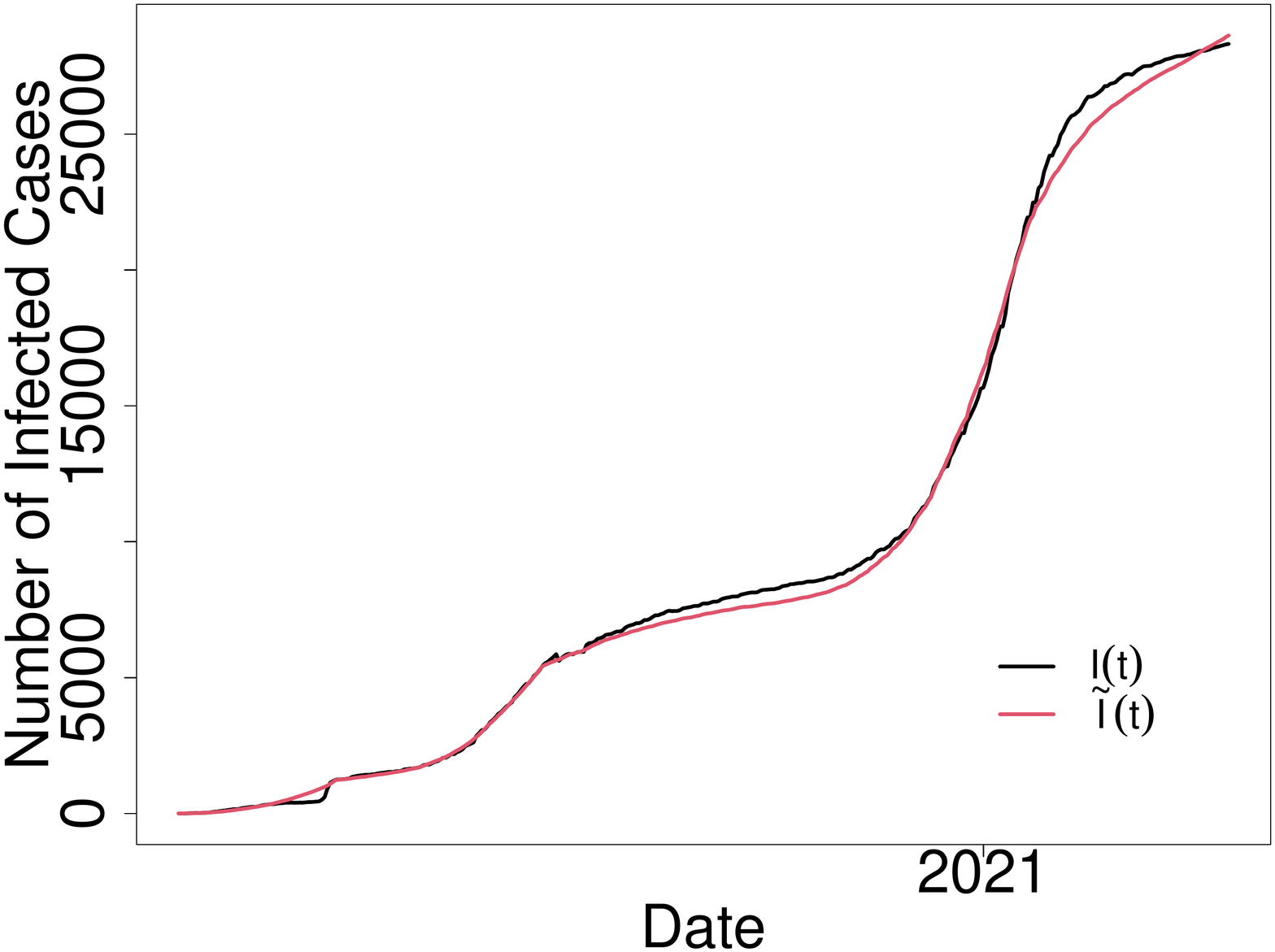}
         \subcaption{Santa Barbara (Model 2.3)}
     \end{subfigure}
    
      \begin{subfigure}[b]{0.19\textwidth}
         \centering
         \includegraphics[width=\textwidth]{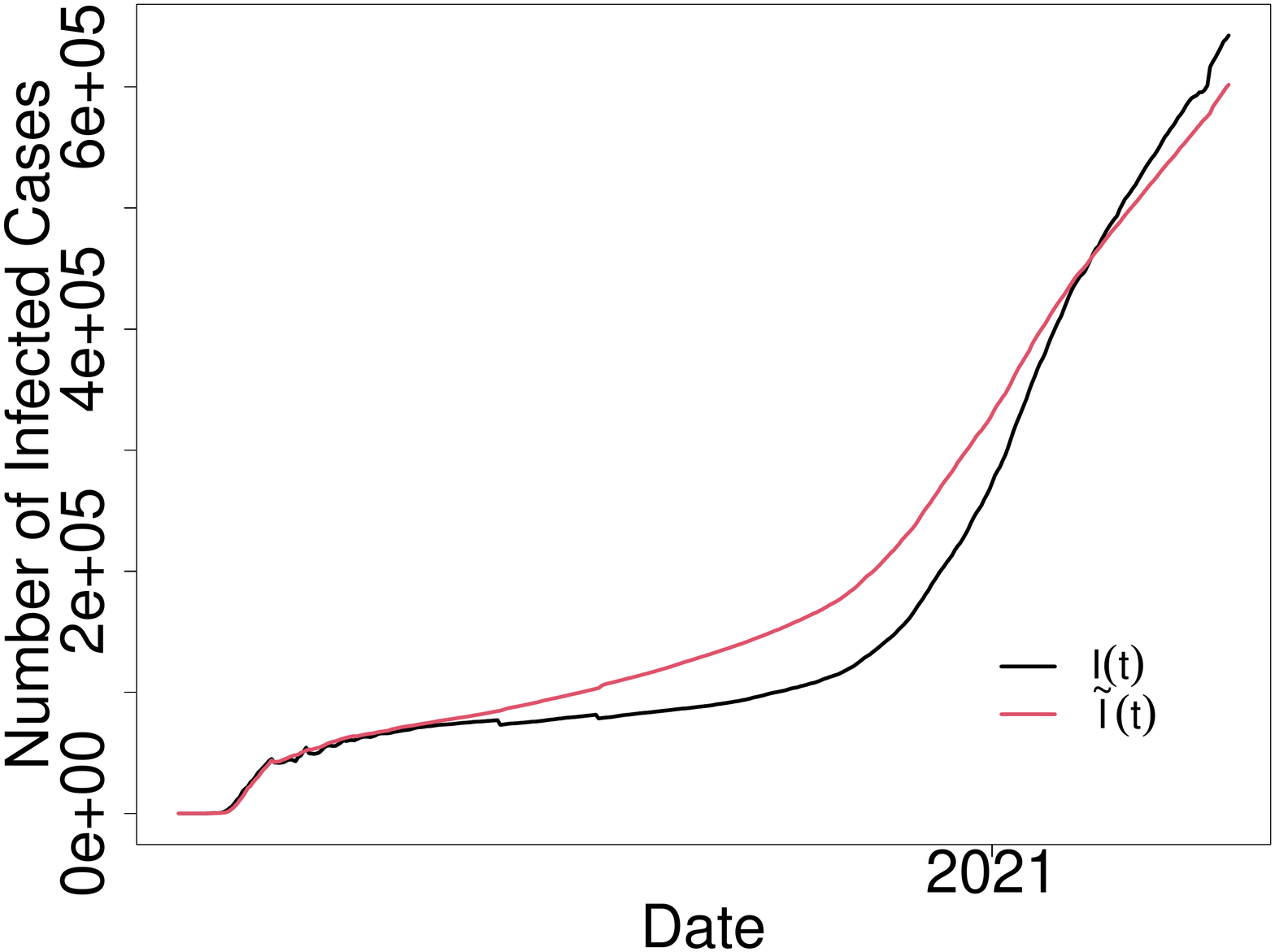}
          \subcaption{NYC (Model 3)}
     \end{subfigure}
     \begin{subfigure}[b]{0.19\textwidth}
         \centering
         \includegraphics[width=\textwidth]{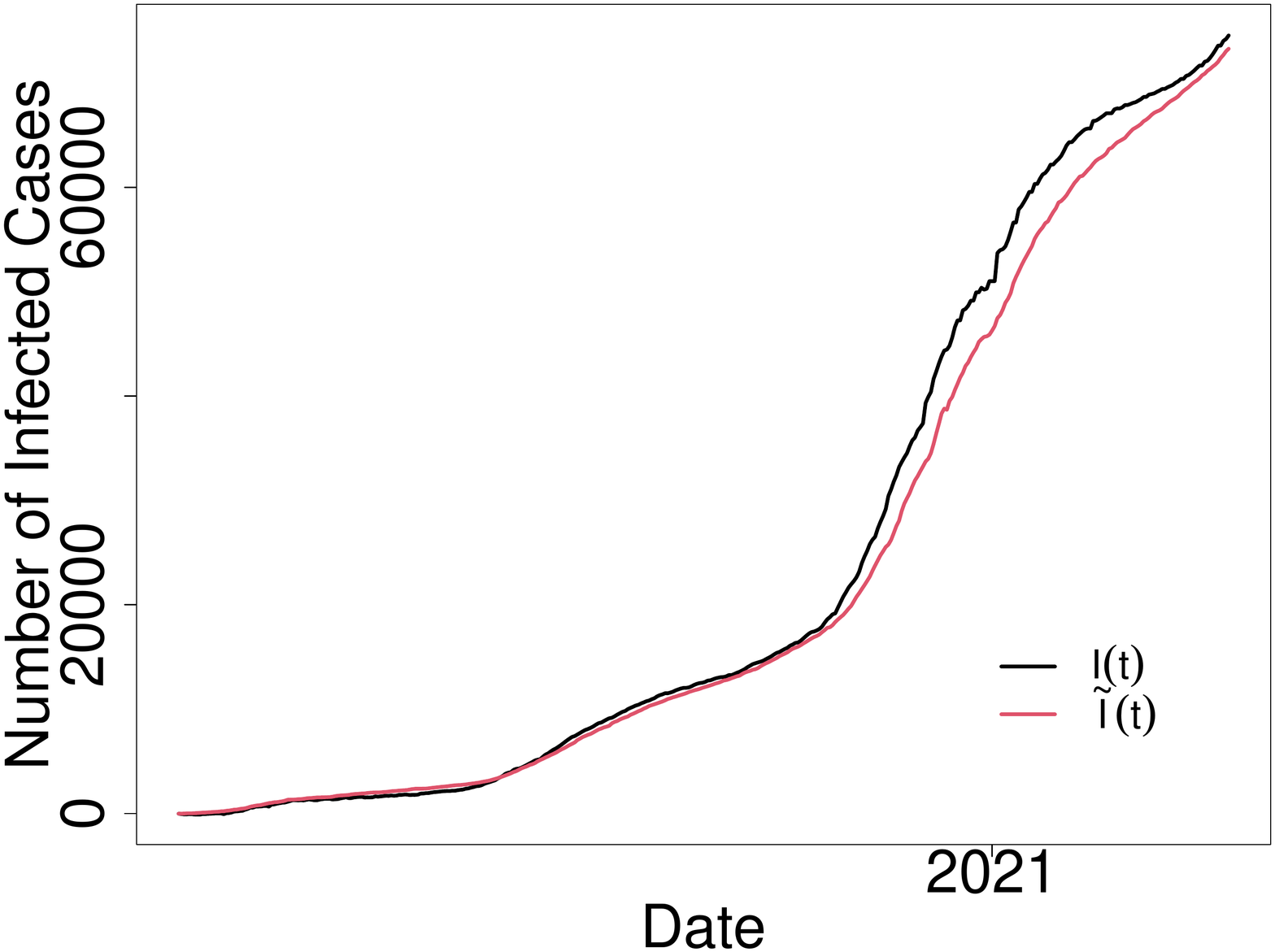}
         \subcaption{King (Model 3)}
     \end{subfigure}
     \begin{subfigure}[b]{0.19\textwidth}
         \centering
         \includegraphics[width=\textwidth]{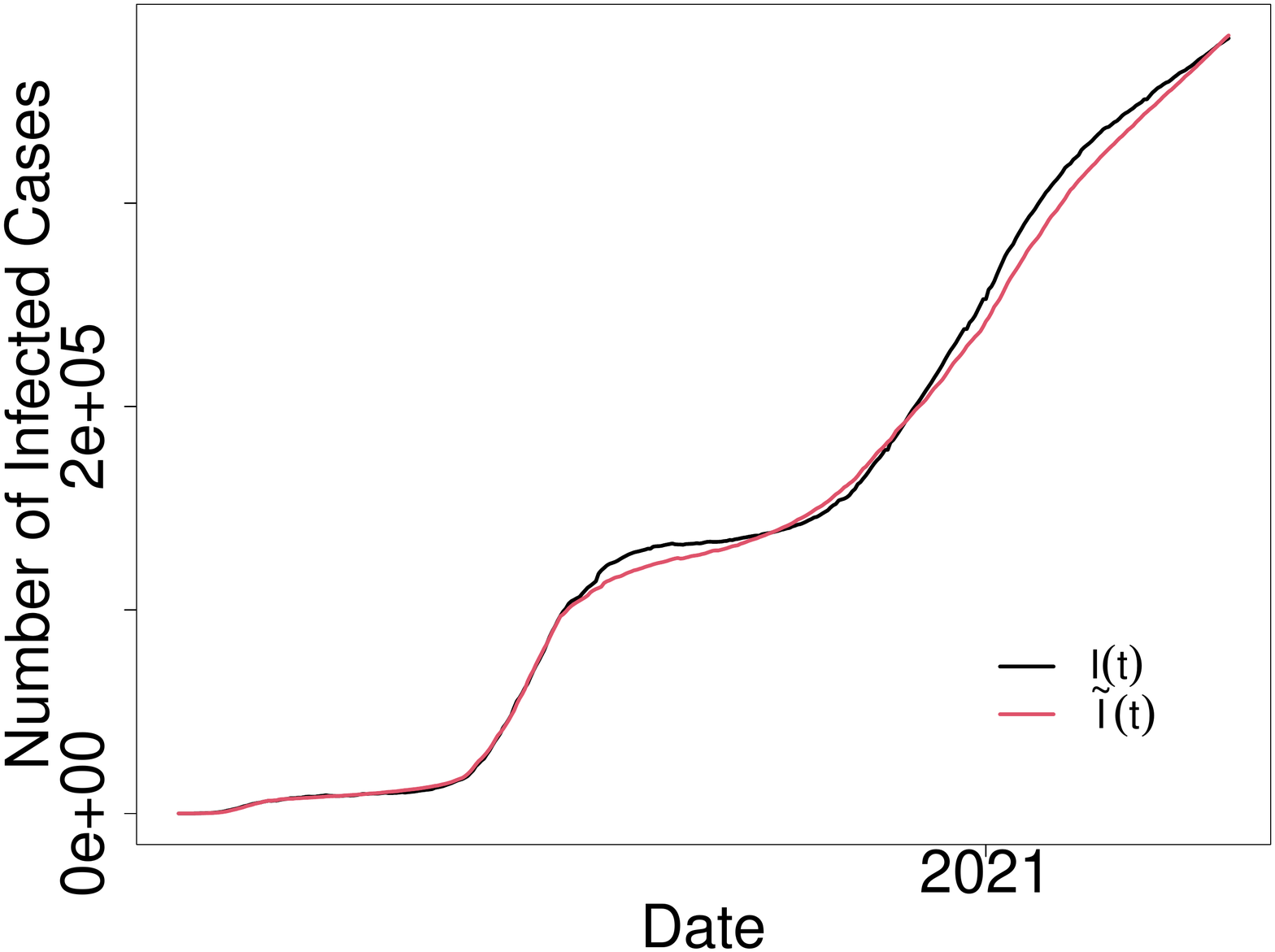}
         \subcaption{Miami-Dade (Model 3)}
     \end{subfigure} 
    \begin{subfigure}[b]{0.19\textwidth}
         \centering
         \includegraphics[width=\textwidth]{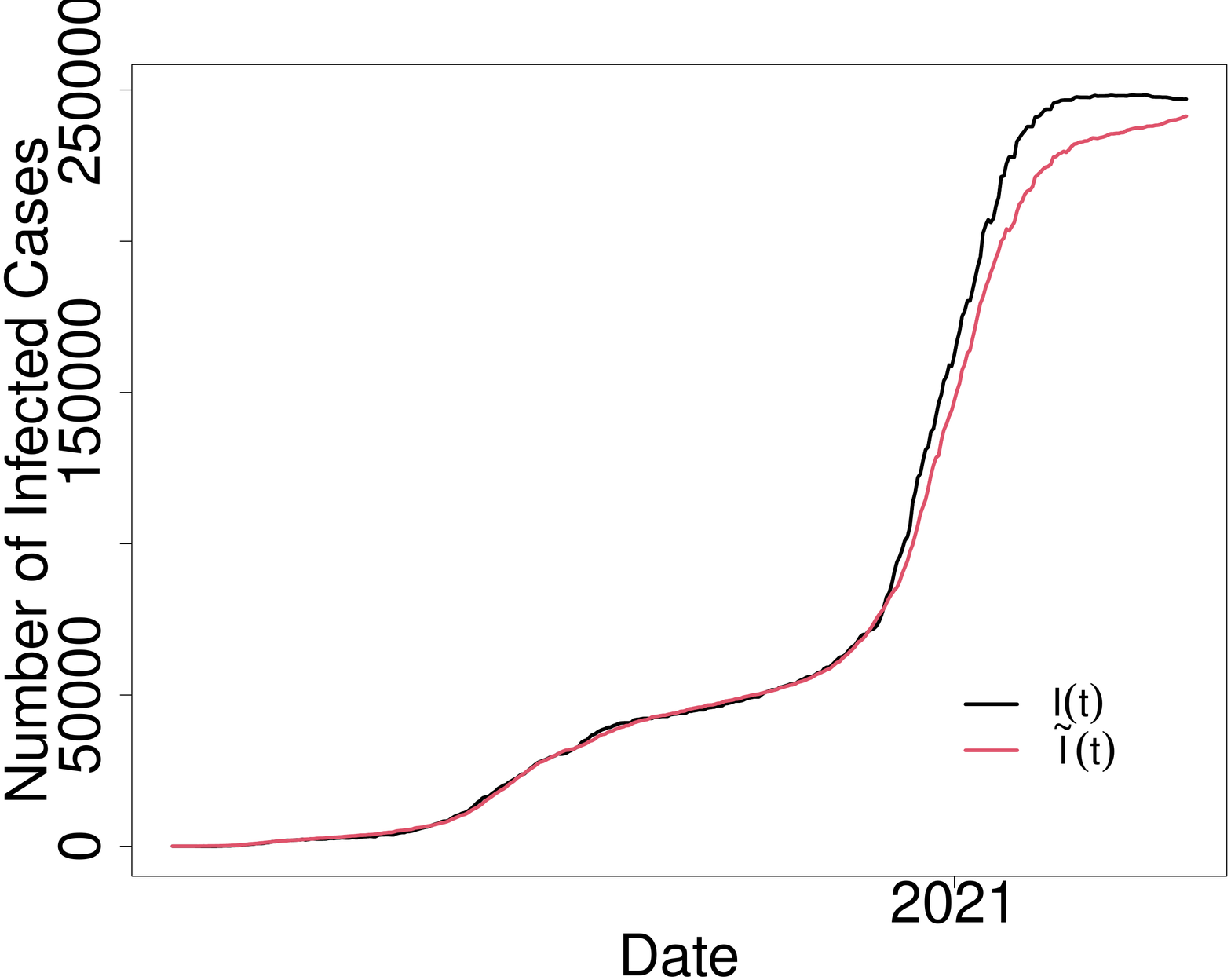}
         \subcaption{Riverside (Model 3)}
     \end{subfigure}
     \begin{subfigure}[b]{0.19\textwidth}
         \centering
         \includegraphics[width=\textwidth]{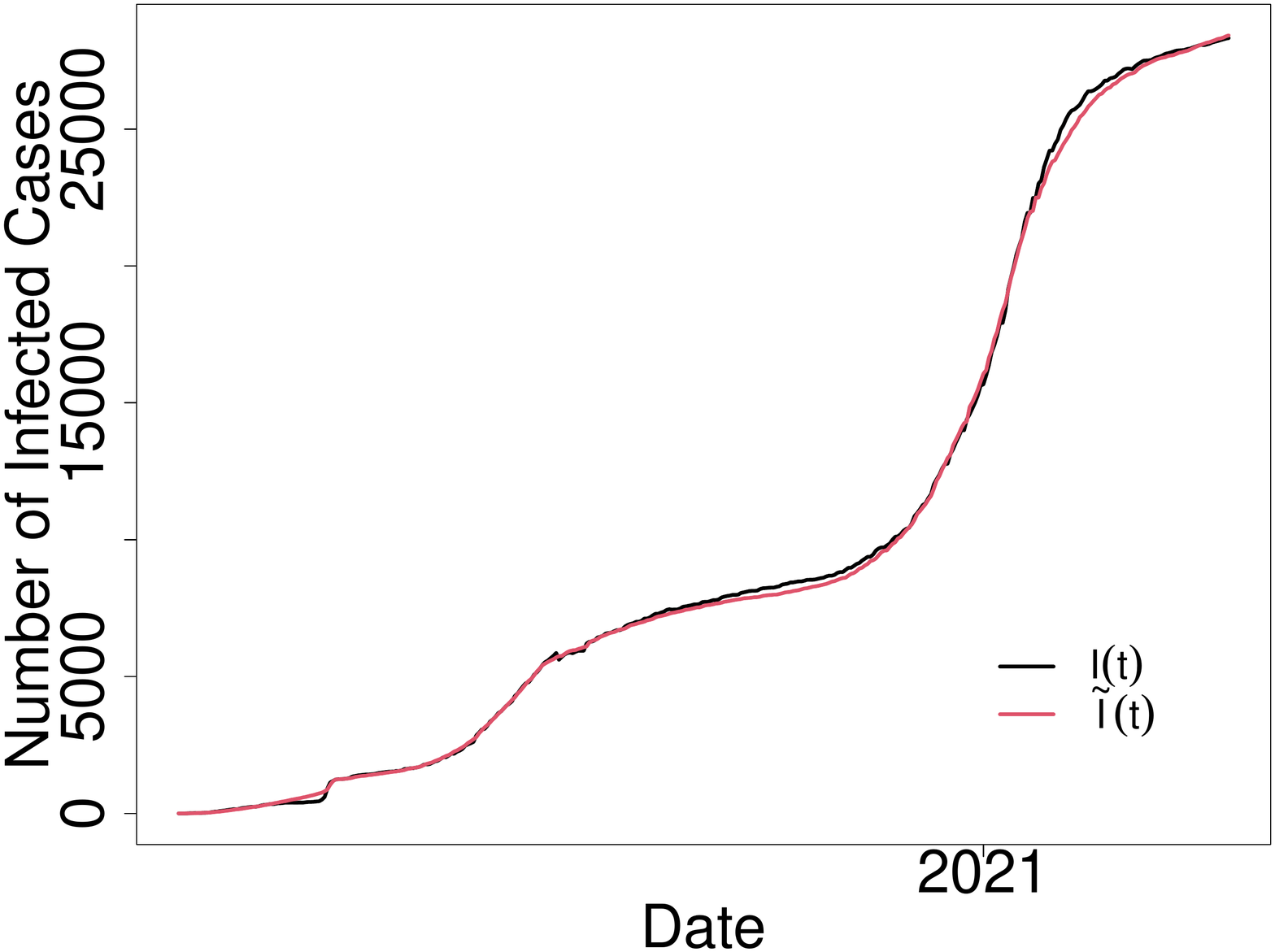}
         \subcaption{Santa Barbara (Model 3)}
     \end{subfigure}
        \caption{Observed (black) and fitted (red) number of infected cases estimated by three models in selected counties/cities. }
        \label{fig:number of infected_county}
\end{figure*}

\begin{figure*}[ht!]
     \centering
     \captionsetup[sub]{font=scriptsize,labelfont={bf,sf}}
  \begin{subfigure}[b]{0.19\textwidth}
         \centering
         \includegraphics[width=\textwidth]{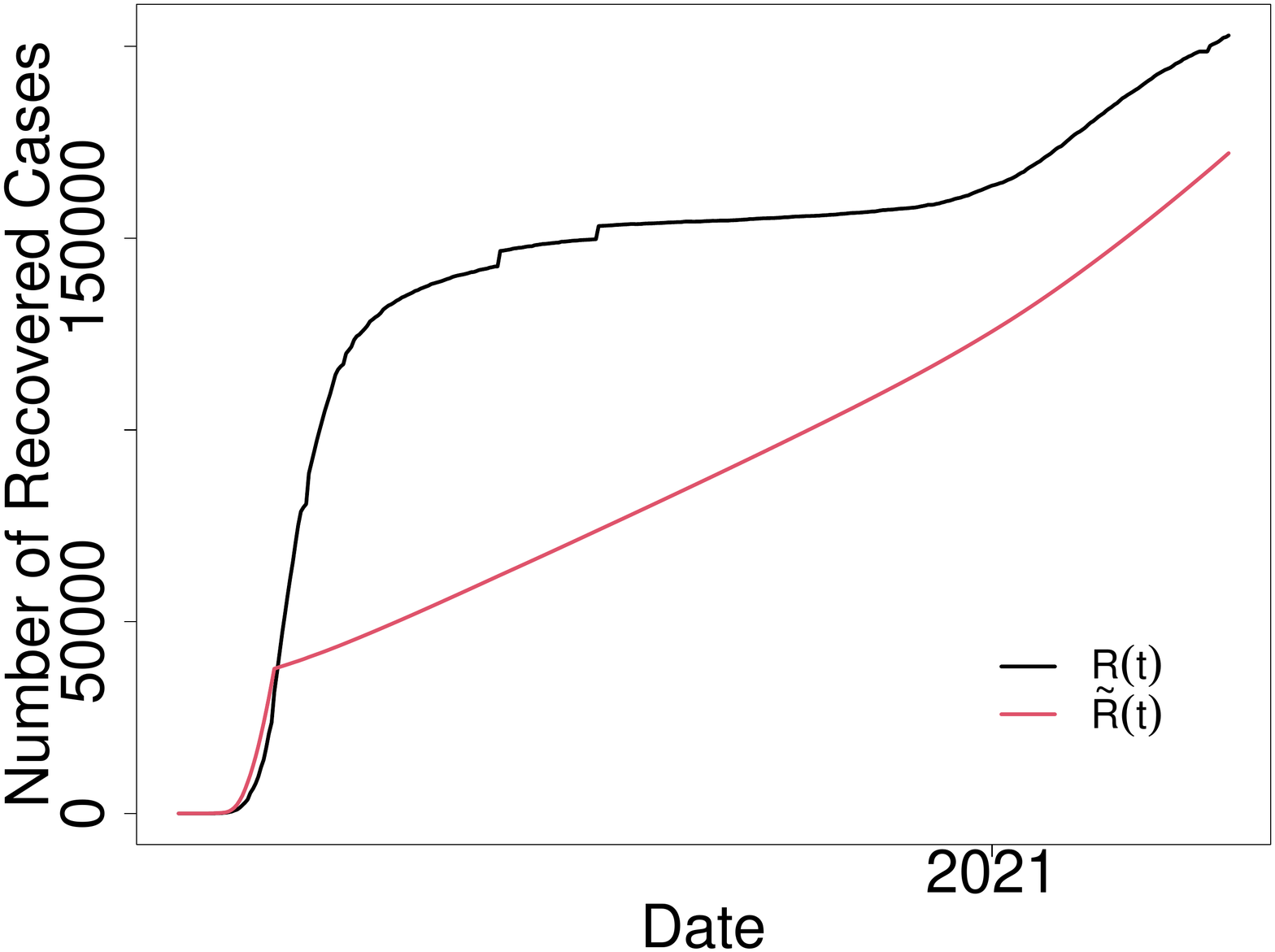}
         \subcaption{NYC (Model 1)}
     \end{subfigure}
     \begin{subfigure}[b]{0.19\textwidth}
         \centering
         \includegraphics[width=\textwidth]{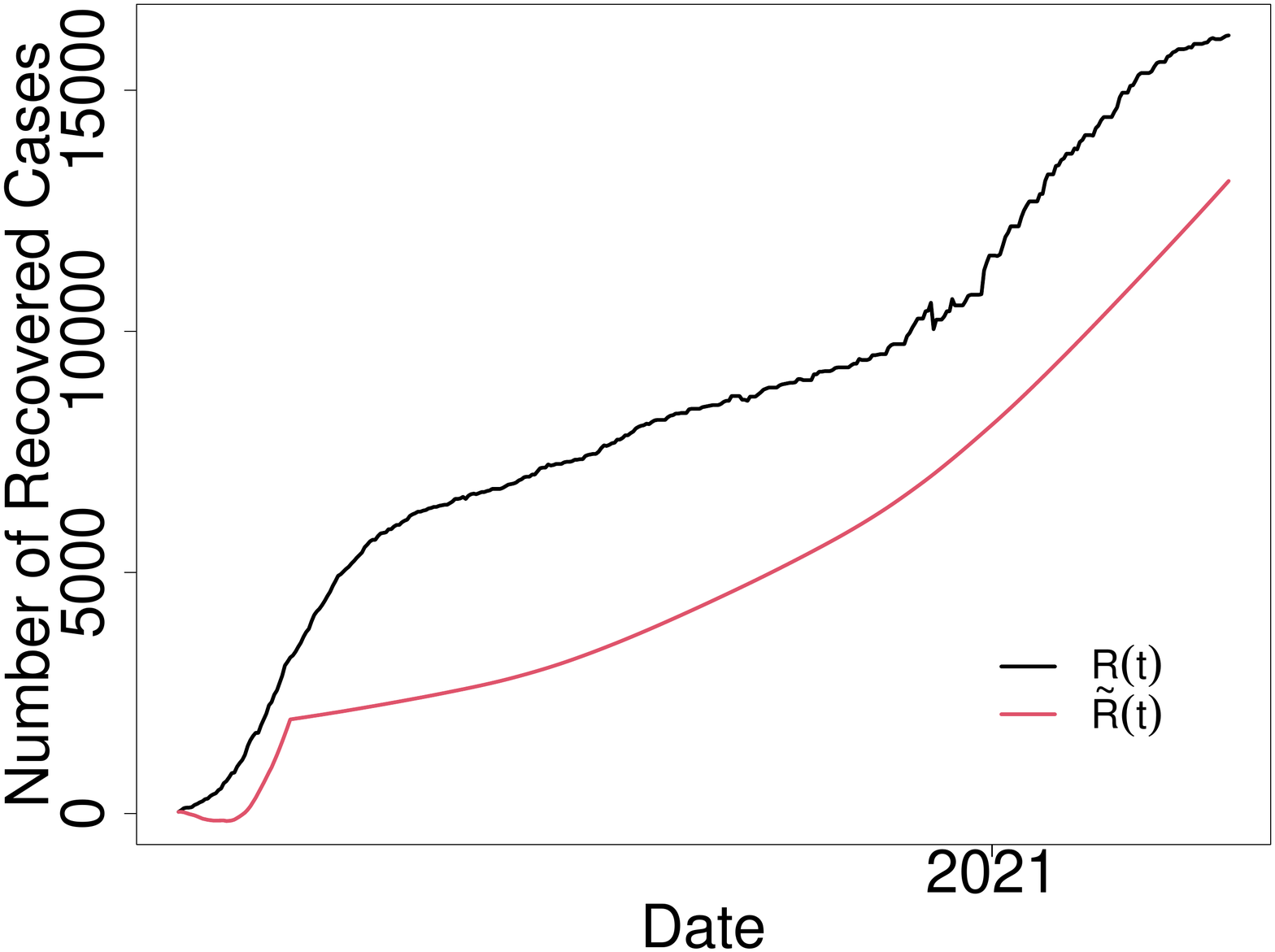}
         \subcaption{King (Model 1)}
     \end{subfigure}
    \begin{subfigure}[b]{0.19\textwidth}
         \centering
         \includegraphics[width=\textwidth]{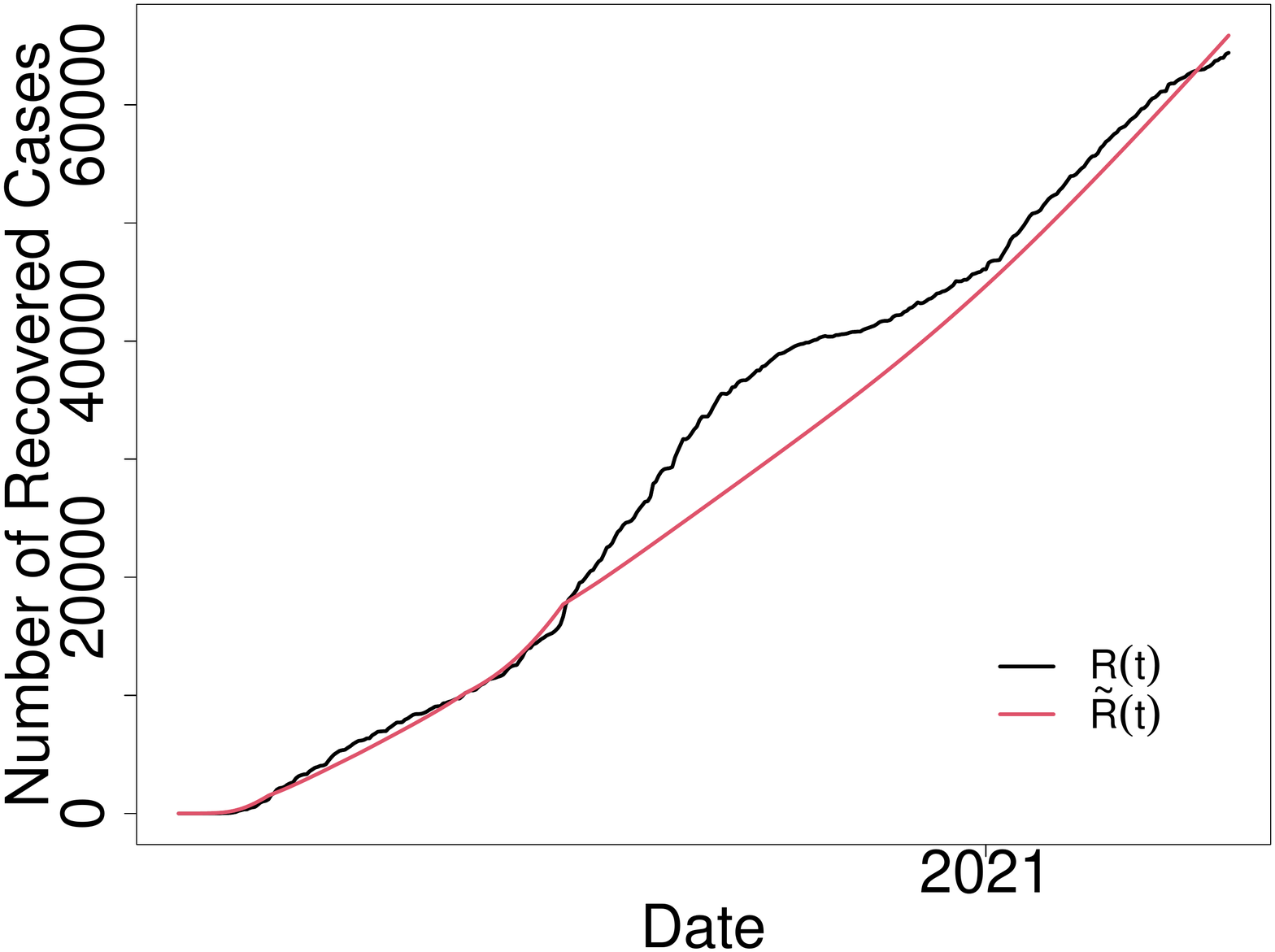}
         \subcaption{Miami (Model 1)}
     \end{subfigure} 
\begin{subfigure}[b]{0.19\textwidth}
         \centering
         \includegraphics[width=\textwidth]{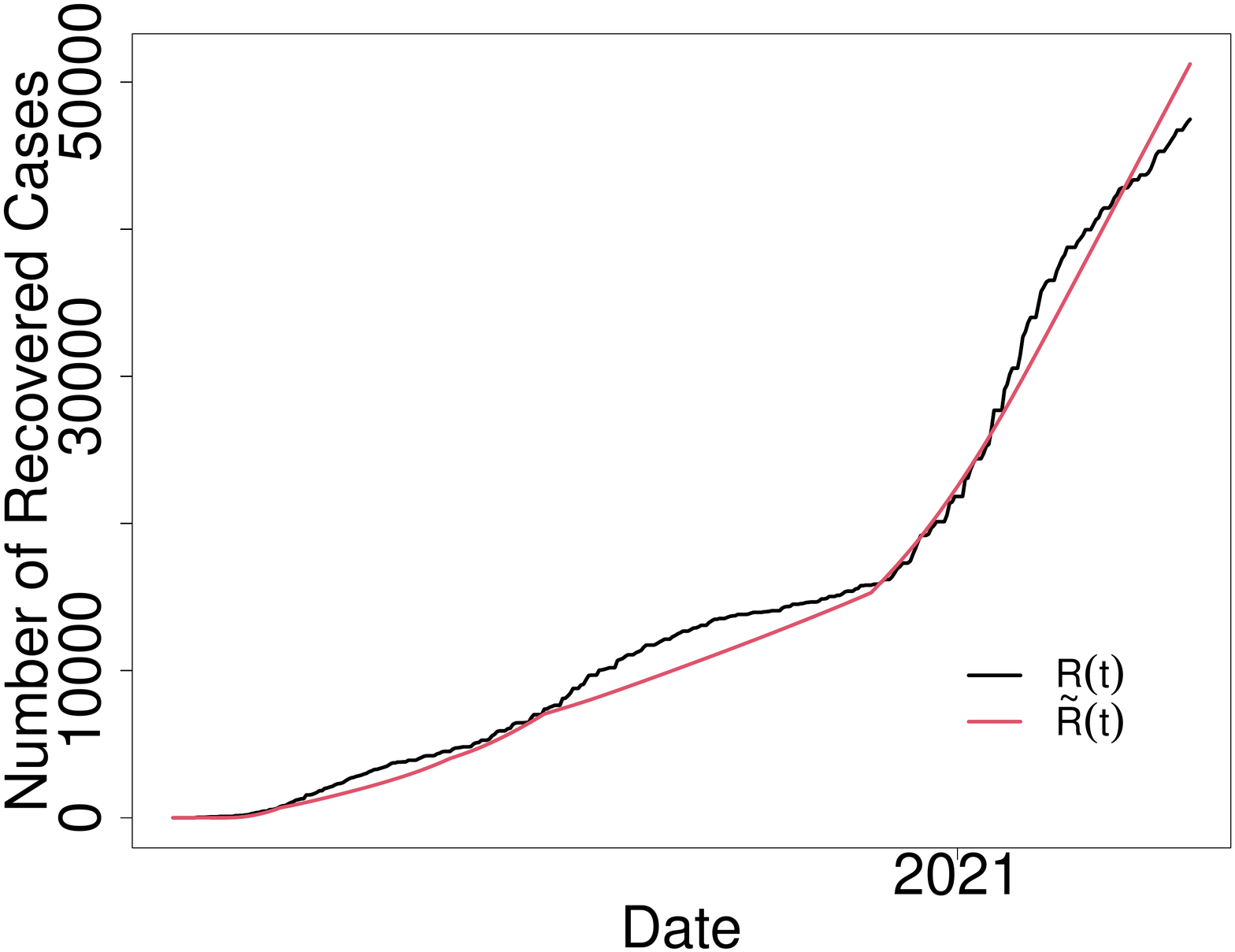}
         \subcaption{Riverside (Model 1)}
     \end{subfigure}
     \begin{subfigure}[b]{0.19\textwidth}
         \centering
         \includegraphics[width=\textwidth]{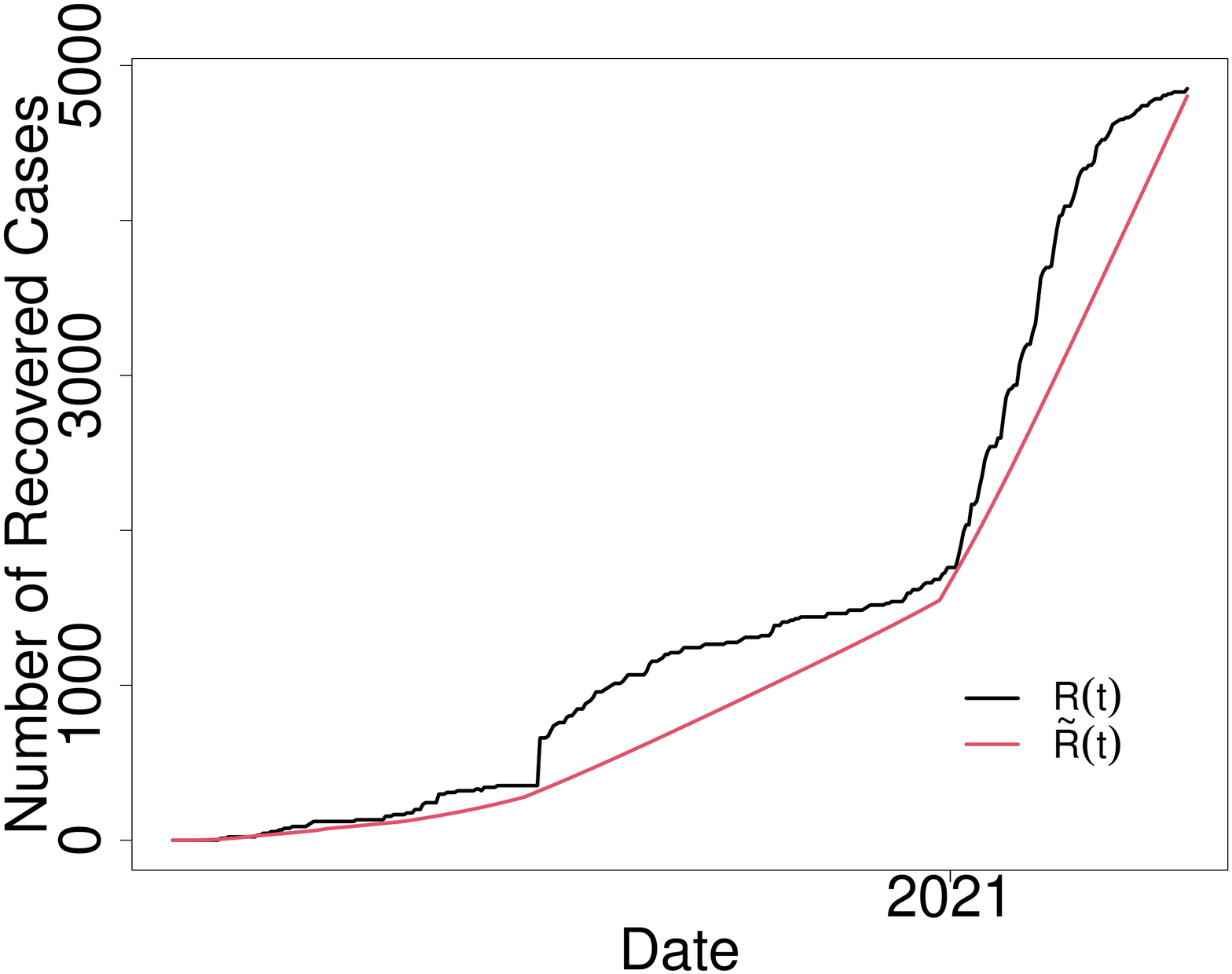}
         \subcaption{Santa Barbara (Model 1)}
     \end{subfigure}
    
     \begin{subfigure}[b]{0.19\textwidth}
         \centering
         \includegraphics[width=\textwidth]{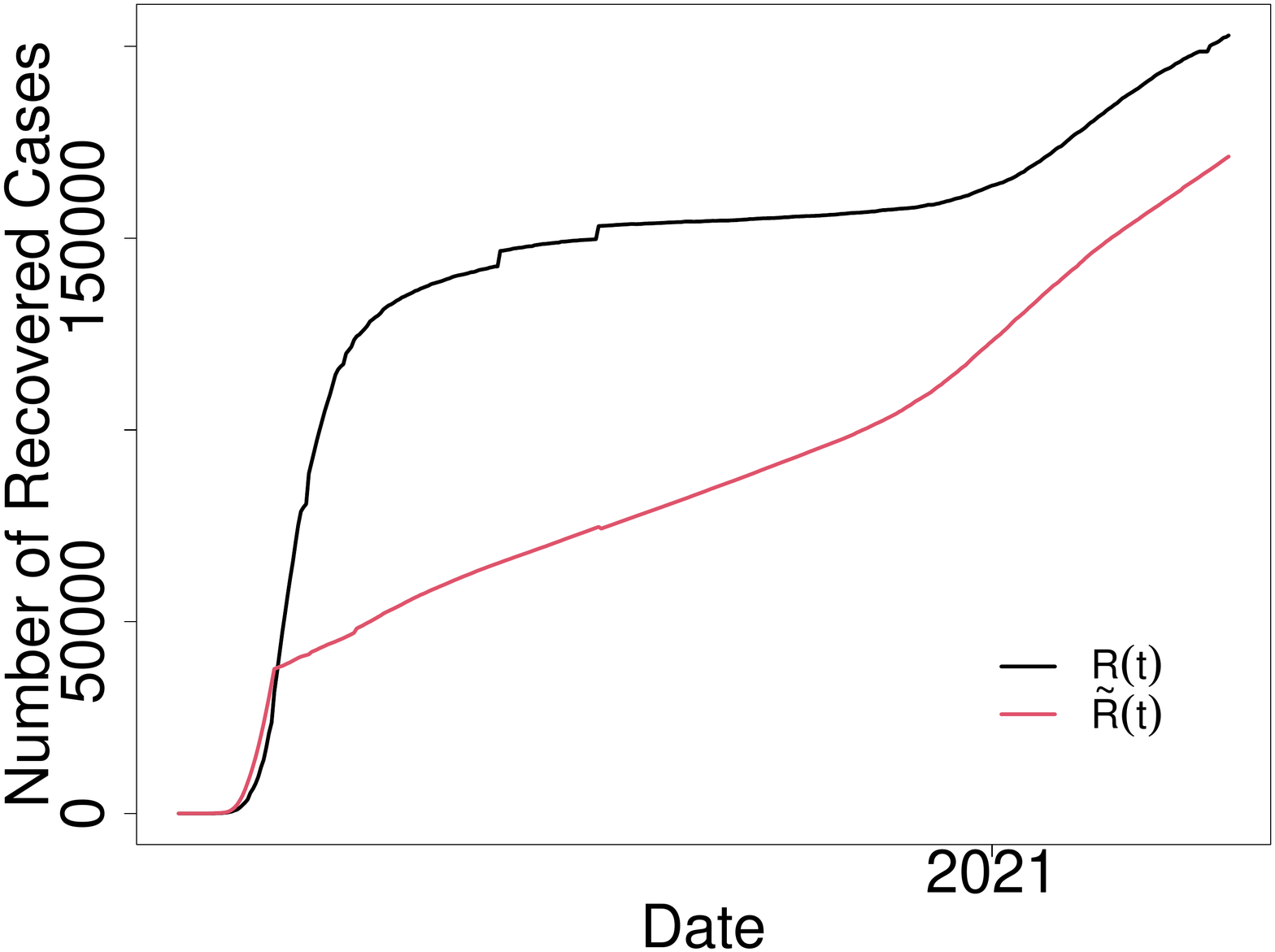}
         \subcaption{NYC (Model 2.3)}
     \end{subfigure}
     \begin{subfigure}[b]{0.19\textwidth}
         \centering
         \includegraphics[width=\textwidth]{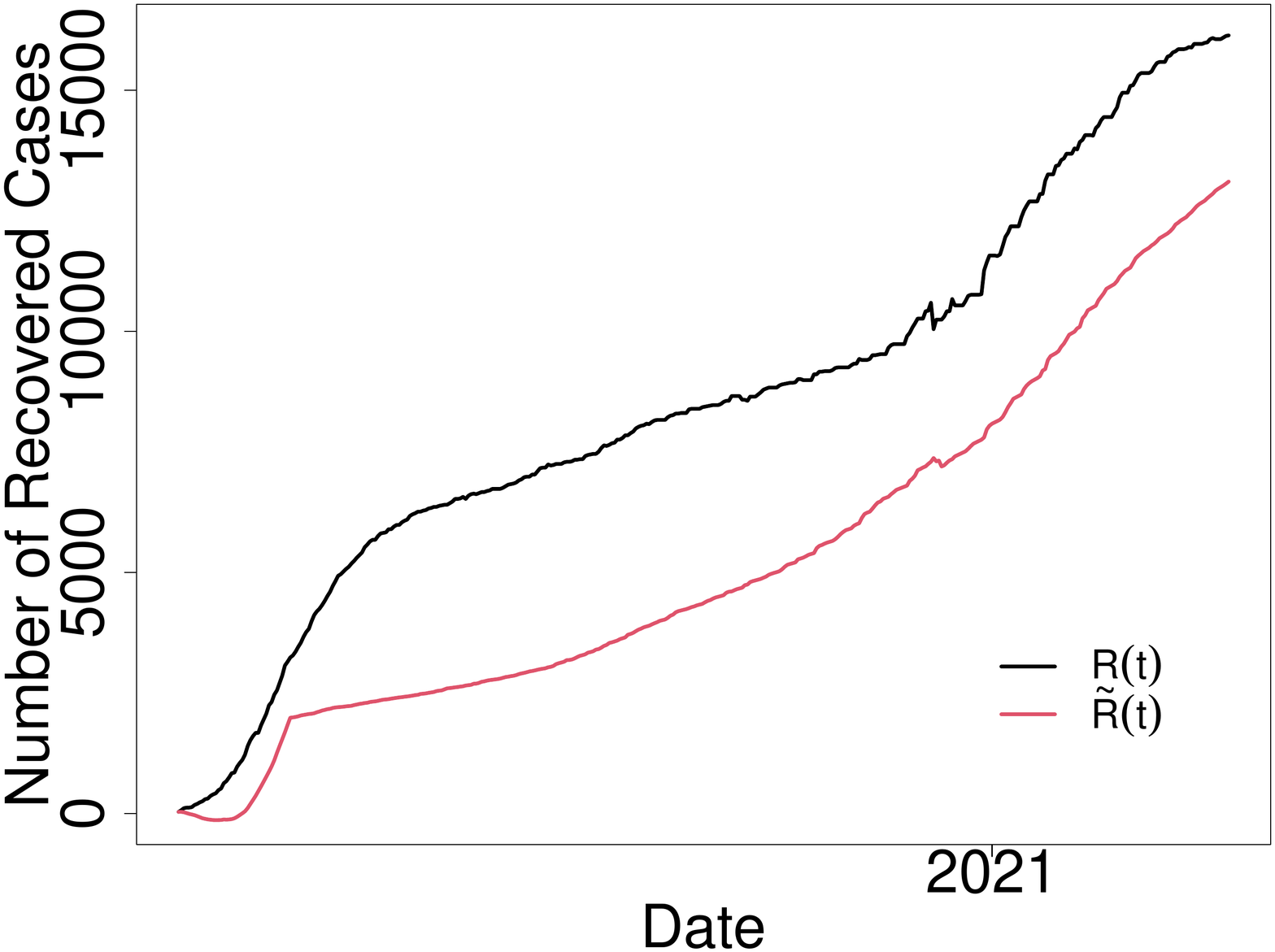}
         \subcaption{King (Model 2.3)}
     \end{subfigure}
     \begin{subfigure}[b]{0.19\textwidth}
         \centering
         \includegraphics[width=\textwidth]{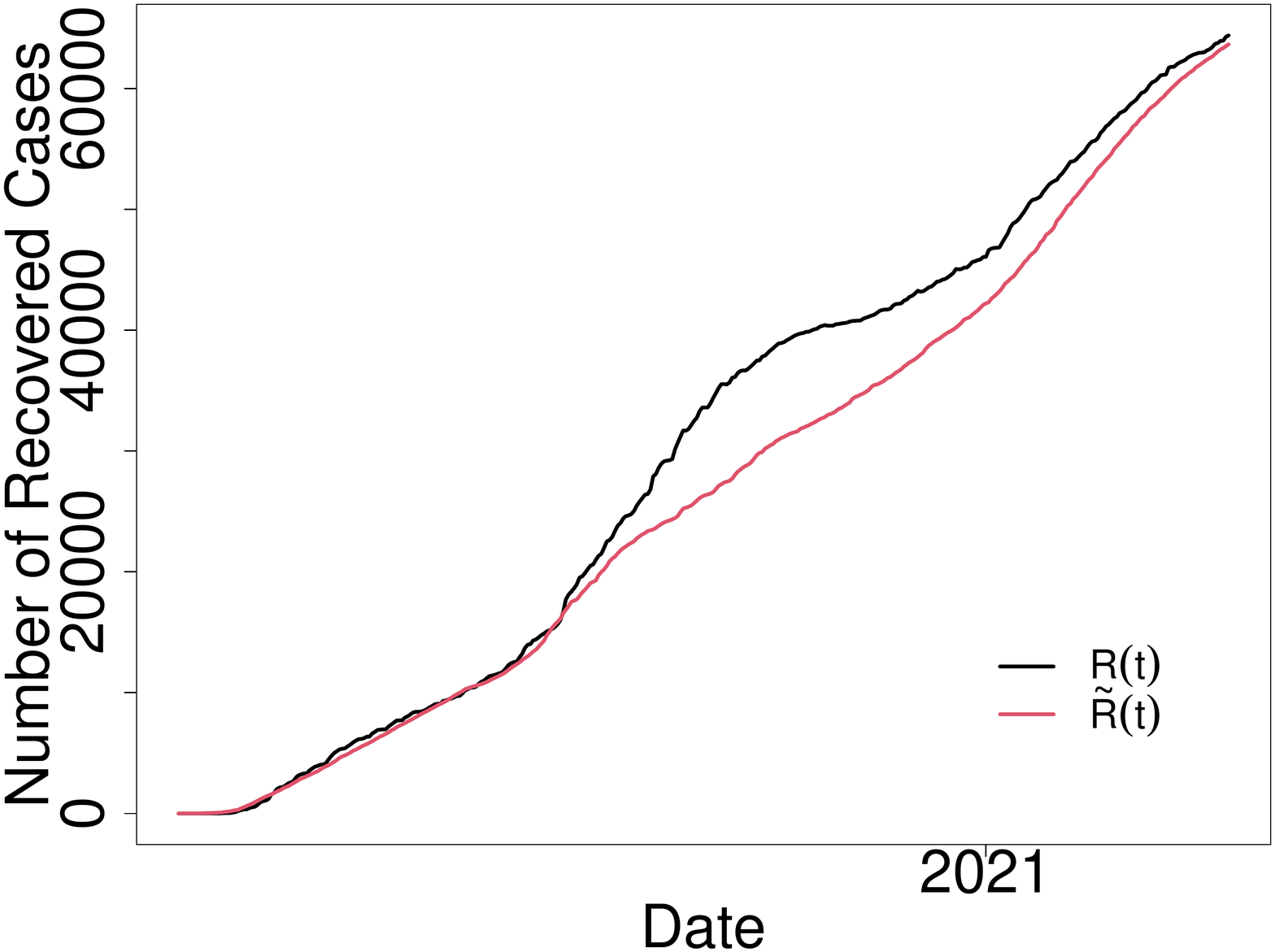}
         \subcaption{Miami (Model 2.3)}
     \end{subfigure}
   \begin{subfigure}[b]{0.19\textwidth}
         \centering
         \includegraphics[width=\textwidth]{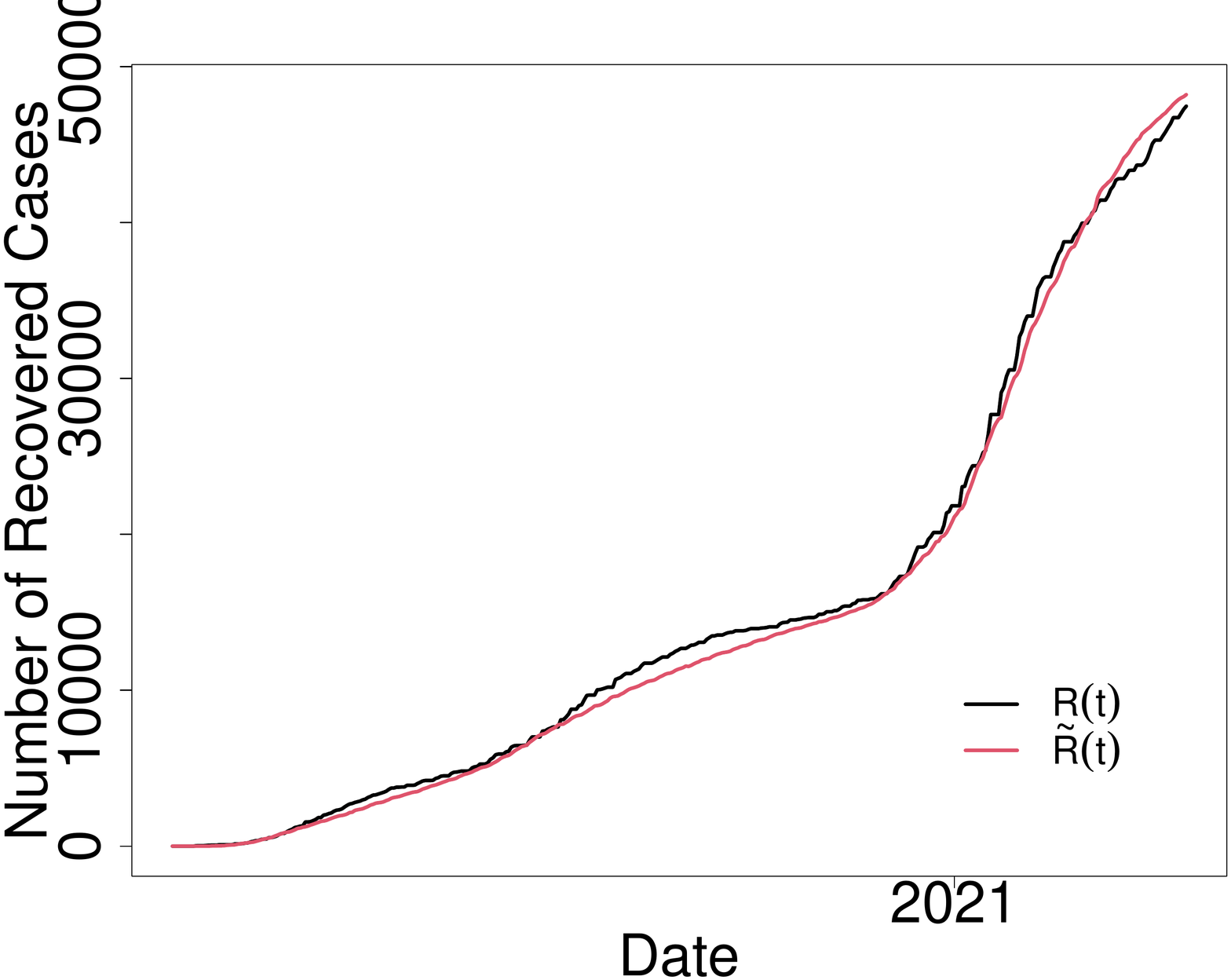}
         \subcaption{Riverside (Model 2.3)}
     \end{subfigure}
     \begin{subfigure}[b]{0.19\textwidth}
         \centering
         \includegraphics[width=\textwidth]{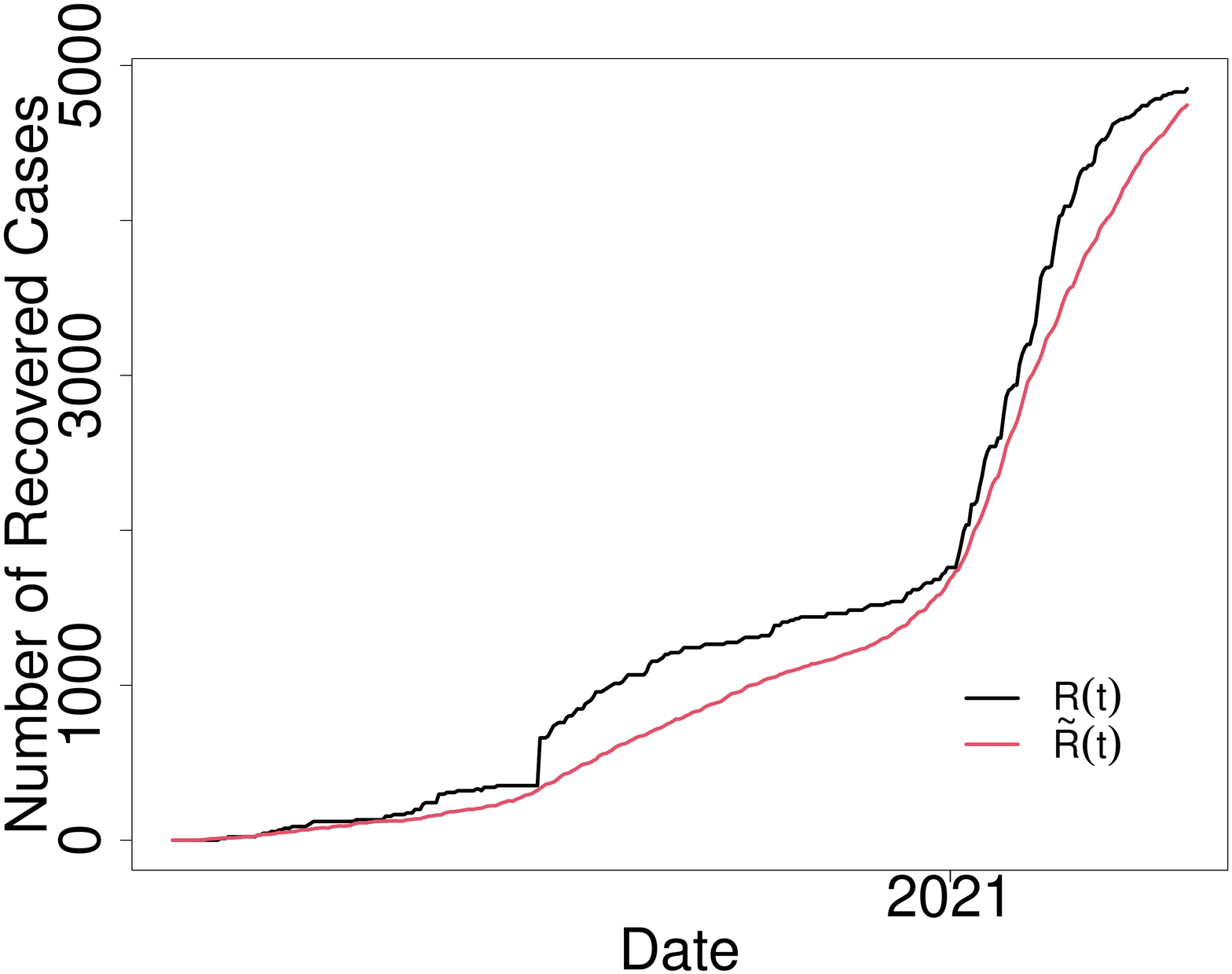}
         \subcaption{Santa Barbara (Model 2.3)}
     \end{subfigure}

     \begin{subfigure}[b]{0.19\textwidth}
         \centering
         \includegraphics[width=\textwidth]{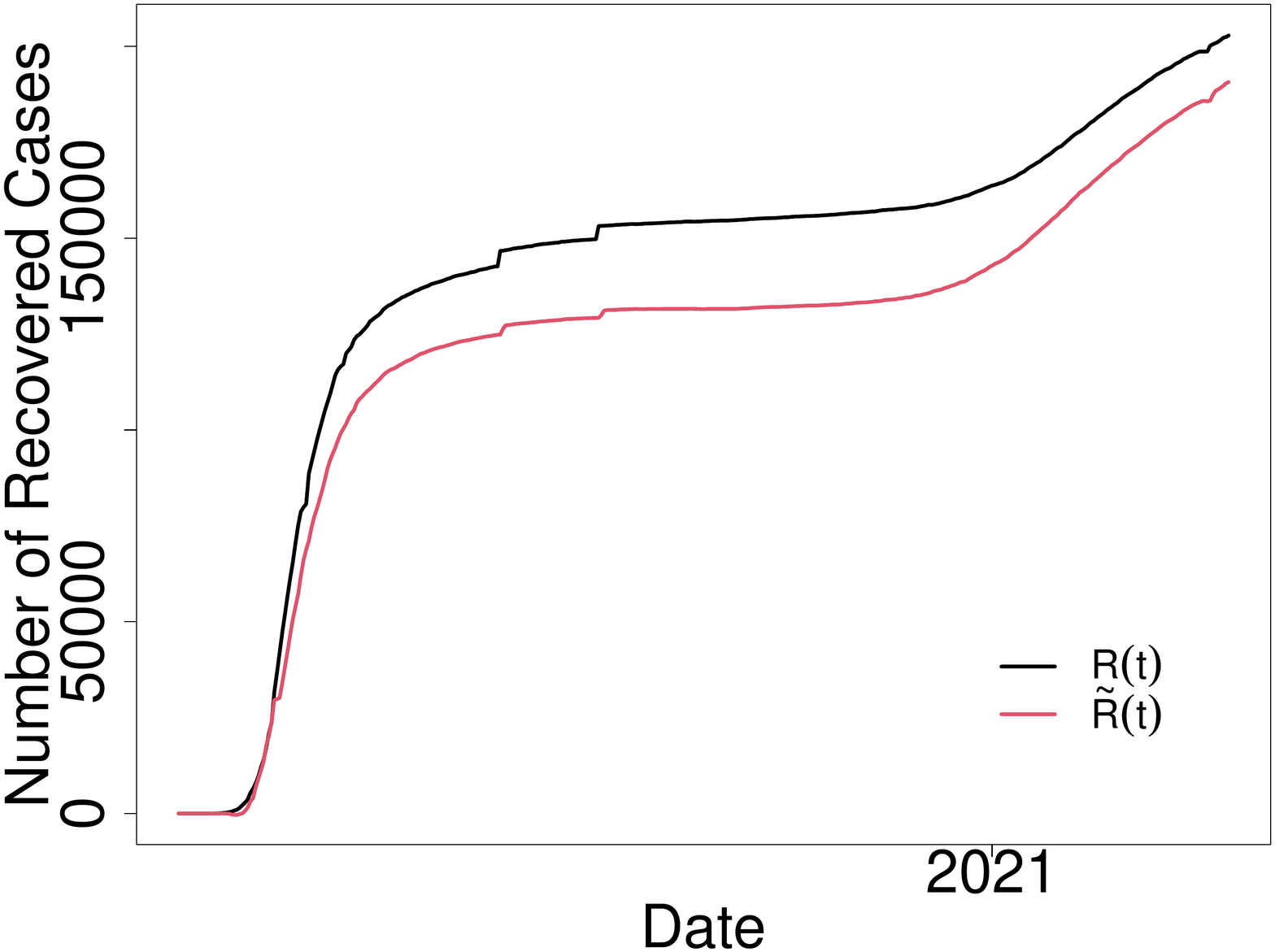}
         \subcaption{NYC (Model 3)}
     \end{subfigure}
     \begin{subfigure}[b]{0.19\textwidth}
         \centering
         \includegraphics[width=\textwidth]{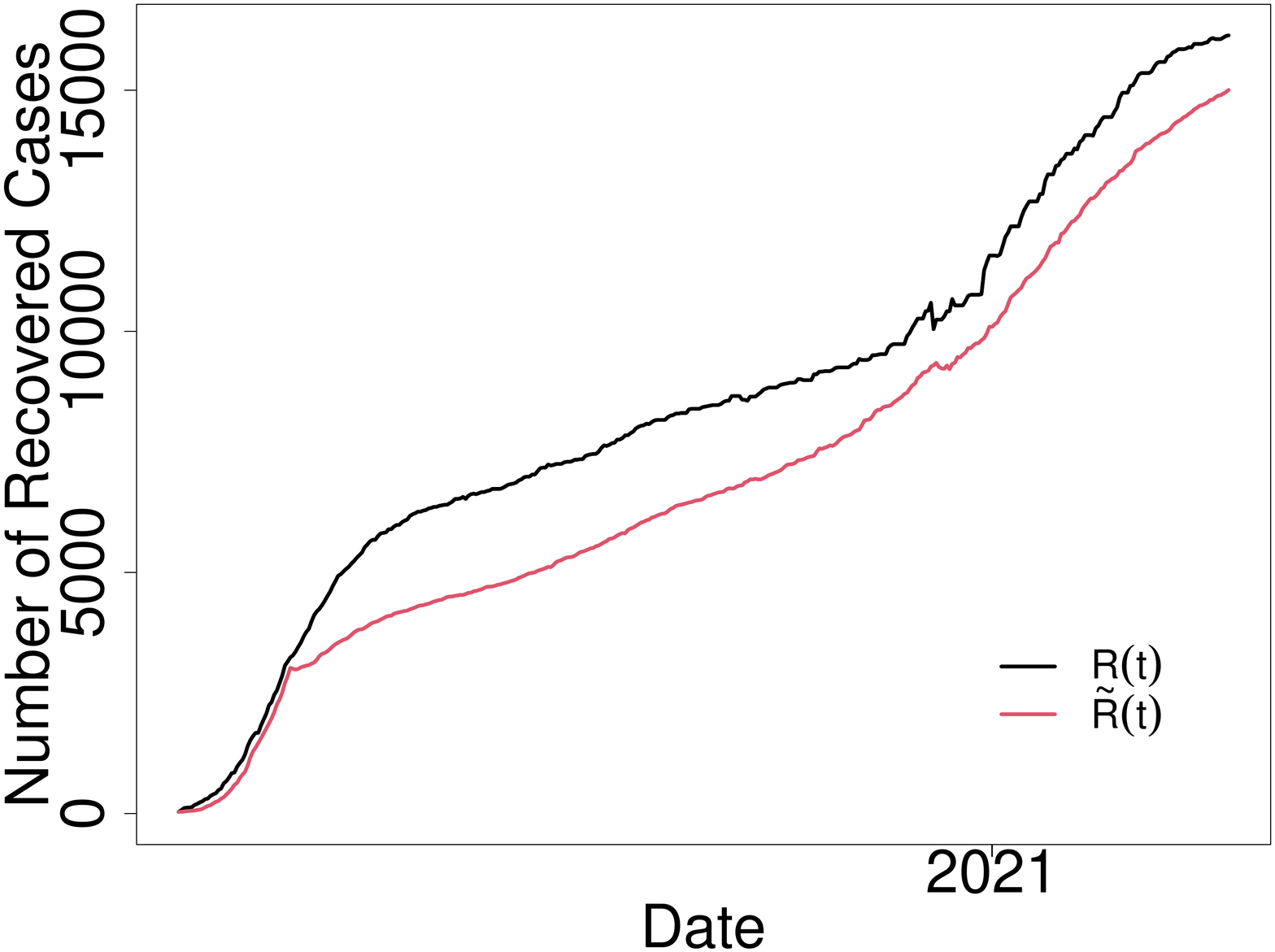}
         \subcaption{King (Model 3)}
     \end{subfigure}
     \begin{subfigure}[b]{0.19\textwidth}
         \centering
         \includegraphics[width=\textwidth]{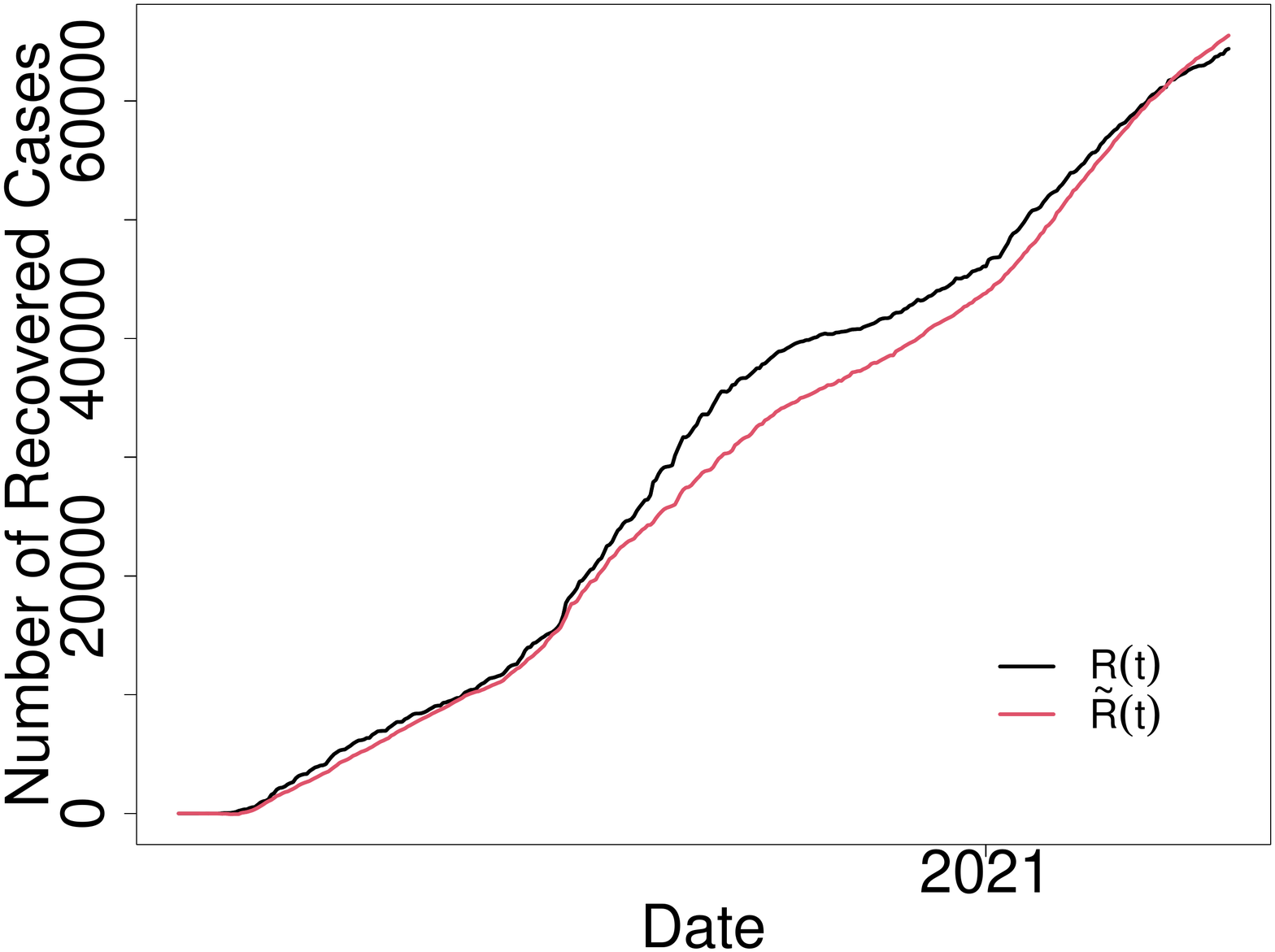}
         \subcaption{Miami (Model 3)}
     \end{subfigure}
     \begin{subfigure}[b]{0.19\textwidth}
         \centering
         \includegraphics[width=\textwidth]{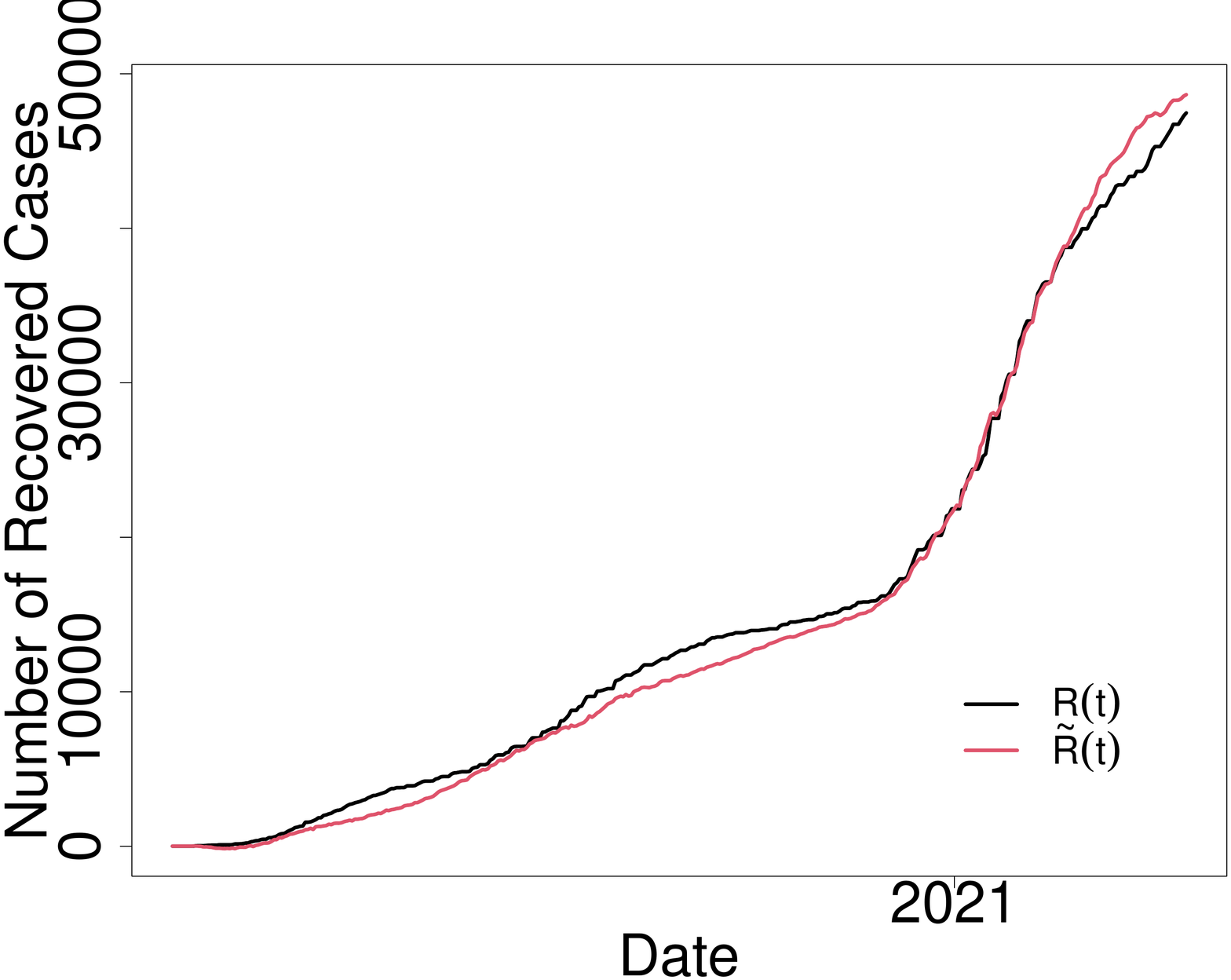}
         \subcaption{Riverside (Model 3)}
     \end{subfigure}
     \begin{subfigure}[b]{0.19\textwidth}
         \centering
         \includegraphics[width=\textwidth]{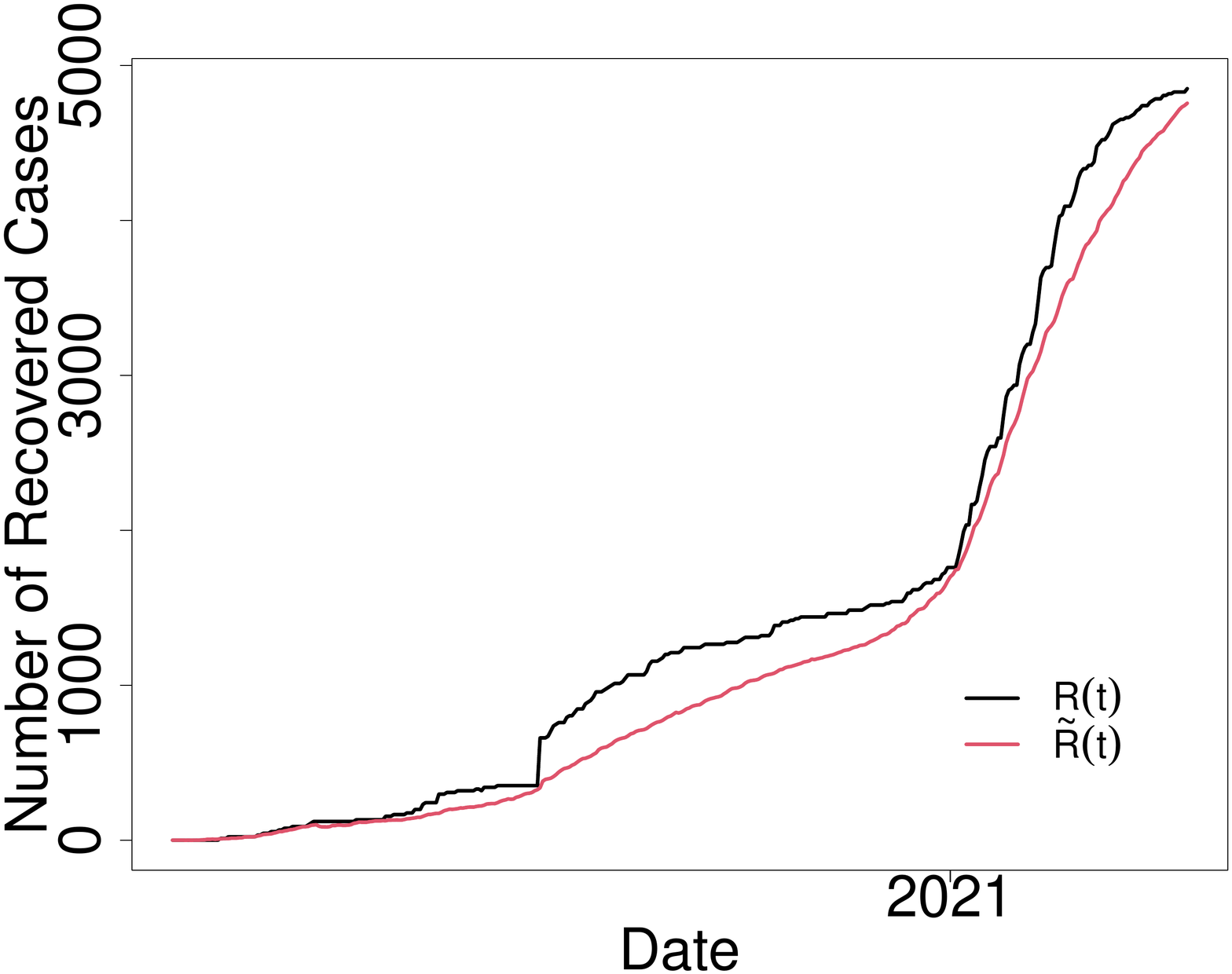}
         \subcaption{Santa Barbara (Model 3)}
     \end{subfigure}
        \caption{Observed (black) and fitted (red) number of recovered cases estimated by three models in selected  counties/cities.  }
        \label{fig:number of recovered_County}
\end{figure*}

The observed and fitted number of infected cases and recovered cases are displayed in Figure \ref{fig:number of infected_county} and Figure \ref{fig:number of recovered_County}, respectively. 
Note that we calculate the number of recovered cases by an approximated nationwide recovery-to-death ratio which probably leads to a less satisfactory fitted result of recovered cases at the county level.        
Nevertheless, it does not affect the change point detection result.

\begin{table}[!ht]
\caption{\label{table_alpha_county_2} Spatial effect estimation in Model 2, including the estimate, p-value, and 95\% confidence intervals for the parameter $\alpha$ in counties/cities.}
\scriptsize
\centering
\begin{tabular}{lcccc} 
  \hline
  \hline
    & Model  & Estimate & P-value & Confidence interval \\ 
  \hline
\multirow{ 4}{*}{NYC} & Model 2.1 & 0.0509 & 0.0798 & ( -0.0061 , 0.1078 ) \\ 
   & Model 2.2 & 0.0161 & 0.5714 & ( -0.0397 , 0.0719 ) \\ 
   & Model 2.3 & 0.3916 & 8.5486e-05 & ( 0.1969 , 0.5862 ) \\ 
   & Model 2.4 & 0.8174 & 1.4267e-15 & ( 0.6206 , 1.0143 ) \\ 
  \multirow{ 4}{*}{King} & Model 2.1 & 0.3247 & 4.2279e-20 & ( 0.2572 , 0.3922 ) \\ 
   & Model 2.2 & 0.3644 & 4.08359e-21 & ( 0.2907 , 0.4381 ) \\ 
   & Model 2.3 & 0.5357 & 1.6510e-11 & ( 0.3818 , 0.6895 ) \\ 
   & Model 2.4 & 0.4741 & 7.5644e-19 & ( 0.3718 , 0.5764 ) \\ 
  \multirow{ 4}{*}{Miami-Dade} & Model 2.1 & 0.7229 & 1.7405e-37 & ( 0.6179 , 0.8278 ) \\ 
   & Model 2.2 & 0.8381 & 5.3081e-46 & ( 0.7302 , 0.946 ) \\ 
   & Model 2.3 & 0.807 & 2.0541e-52 & ( 0.7109 , 0.9031 ) \\ 
   & Model 2.4 & 0.9815 & 1.2398e-46 & ( 0.8562 , 1.1068 ) \\ 
  \multirow{ 4}{*}{Charleston} & Model 2.1 & 0.6006 & 1.2029e-42 & ( 0.5197 , 0.6815 ) \\ 
   & Model 2.2 & 0.6502 & 2.0676e-47 & ( 0.5679 , 0.7325 ) \\ 
   & Model 2.3 & 0.8485 & 5.9081e-69 & ( 0.7629 , 0.934 ) \\ 
   & Model 2.4 & 0.7418 & 3.6122e-41 & ( 0.6398 , 0.8438 ) \\ 
  \multirow{ 4}{*}{Greenville} & Model 2.1 & 0.6386 & 4.3653e-42 & ( 0.552 , 0.7252 ) \\ 
   & Model 2.2 & 0.6544 & 2.4293e-44 & ( 0.5684 , 0.7404 ) \\ 
   & Model 2.3 & 0.7811 & 9.0873e-66 & ( 0.7002 , 0.862 ) \\ 
   & Model 2.4 & 0.7781 & 8.3114e-61 & ( 0.6935 , 0.8627 ) \\ 
  \multirow{ 4}{*}{Richland} & Model 2.1 & 0.1737 & 3.9669e-07 & ( 0.1071 , 0.2404 ) \\ 
   & Model 2.2 & 0.1749 & 4.7063e-07 & ( 0.1073 , 0.2424 ) \\ 
   & Model 2.3 & 0.5769 & 1.0419e-25 & ( 0.4729 , 0.681 ) \\ 
   & Model 2.4 & 0.551 & 5.3149e-22 & ( 0.4423 , 0.6597 ) \\ 
  \multirow{ 4}{*}{Horry} & Model 2.1 & 0.188 & 1.2290e-09 & ( 0.128 , 0.2479 ) \\ 
   & Model 2.2 & 0.1572 & 1.1276e-07 & ( 0.0996 , 0.2148 ) \\ 
   & Model 2.3 & 0.5857 & 1.1273e-46 & ( 0.5109 , 0.6604 ) \\ 
   & Model 2.4 & 0.5221 & 2.0689e-35 & ( 0.4438 , 0.6004 ) \\ 
  \multirow{ 4}{*}{Riverside} & Model 2.1 & 0.383 & 4.0323e-35 & ( 0.3252 , 0.4408 ) \\ 
   & Model 2.2 & 0.3339 & 4.3780e-30 & ( 0.2788 , 0.3889 ) \\ 
   & Model 2.3 & 1.0056 & 7.4856e-68 & ( 0.9032 , 1.1079 ) \\ 
   & Model 2.4 & 1.6968 & 7.4801e-90 & ( 1.5525 , 1.8412 ) \\ 
  \multirow{ 4}{*}{Santa Barbara} & Model 2.1 & 0.4413 & 3.8589e-16 & ( 0.3373 , 0.5454 ) \\ 
   & Model 2.2 & 0.513 & 2.0562e-18 & ( 0.401 , 0.6251 ) \\ 
   & Model 2.3 & 0.7747 & 3.3965e-20 & ( 0.6142 , 0.9351 ) \\ 
   & Model 2.4 & 0.8119 & 2.0459e-16 & ( 0.6224 , 1.0014 ) \\ 
  \hline
\end{tabular}
\end{table}

In Table \ref{table_alpha_county_2},  we provide the estimate, p-value, and 95\% confidence intervals for the parameter $\alpha$ in Model 2 for additional counties. 
The results in Table \ref{table_alpha_county_2} indicate that the neighboring counties are making an significant impact on estimating and forecasting the transmission dynamic. 
Finally, we use auto-correlation function (ACF) to measure the degree of dependence among residuals at different time points. Figure \ref{fig:acf_county}  shows the ACF of these residual time series in several counties/cities. After adding the VAR($p$) model component to the modeling framework, the residual series $\widetilde{\epsilon}$ in Model 3 has a  significantly reduced the auto-correlation.

\begin{figure*}[ht!]
     \centering
          \captionsetup[sub]{font=scriptsize, labelfont={bf,sf}}
     \begin{subfigure}[b]{0.19\textwidth}
         \centering
         \includegraphics[width=\textwidth]{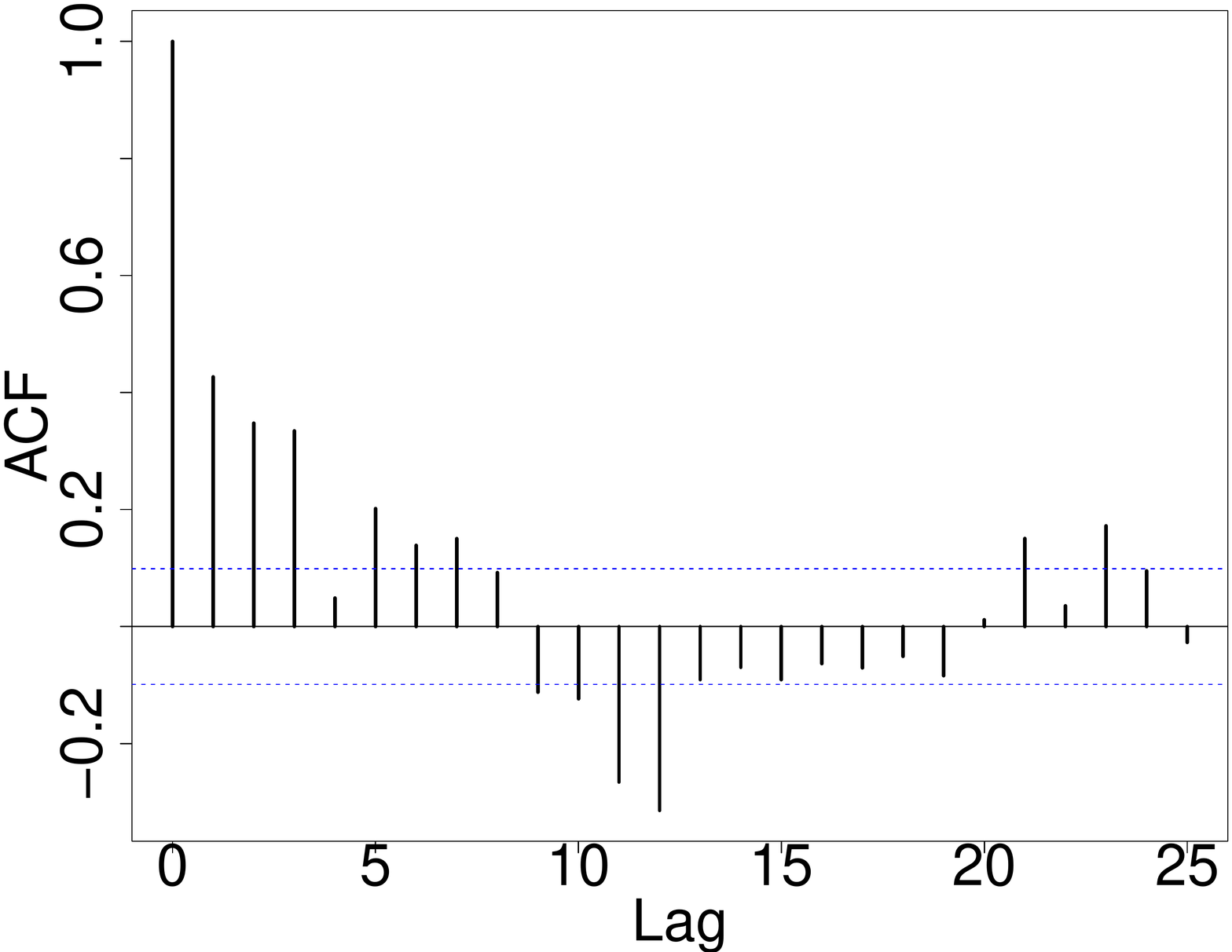}
         \subcaption{NYC $\widehat{\epsilon}(t) (\Delta I)$}
     \end{subfigure}
     \begin{subfigure}[b]{0.19\textwidth}
         \centering
         \includegraphics[width=\textwidth]{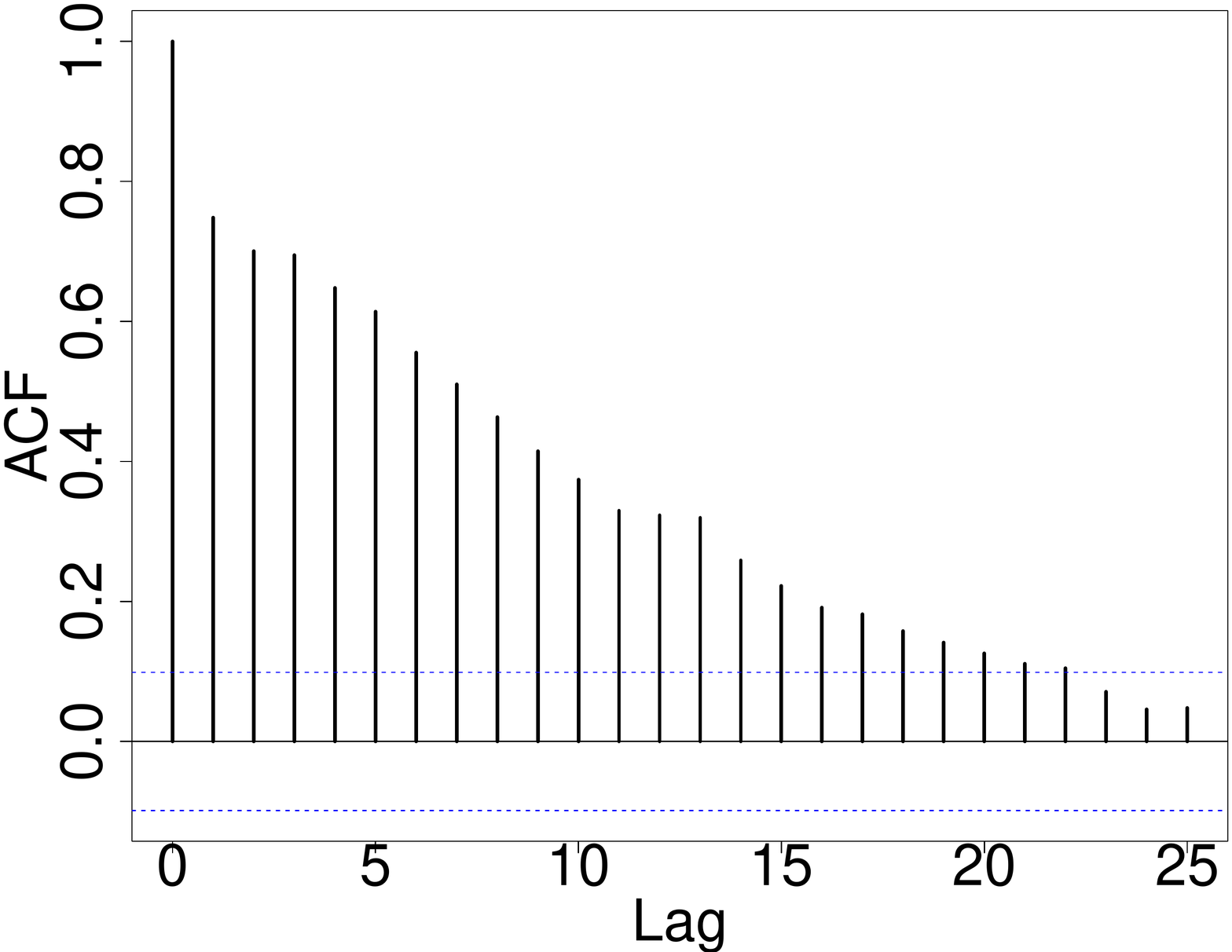}
         \subcaption{NYC $\widehat{\epsilon}(t) (\Delta R)$}
     \end{subfigure}
     \begin{subfigure}[b]{0.19\textwidth}
         \centering
         \includegraphics[width=\textwidth]{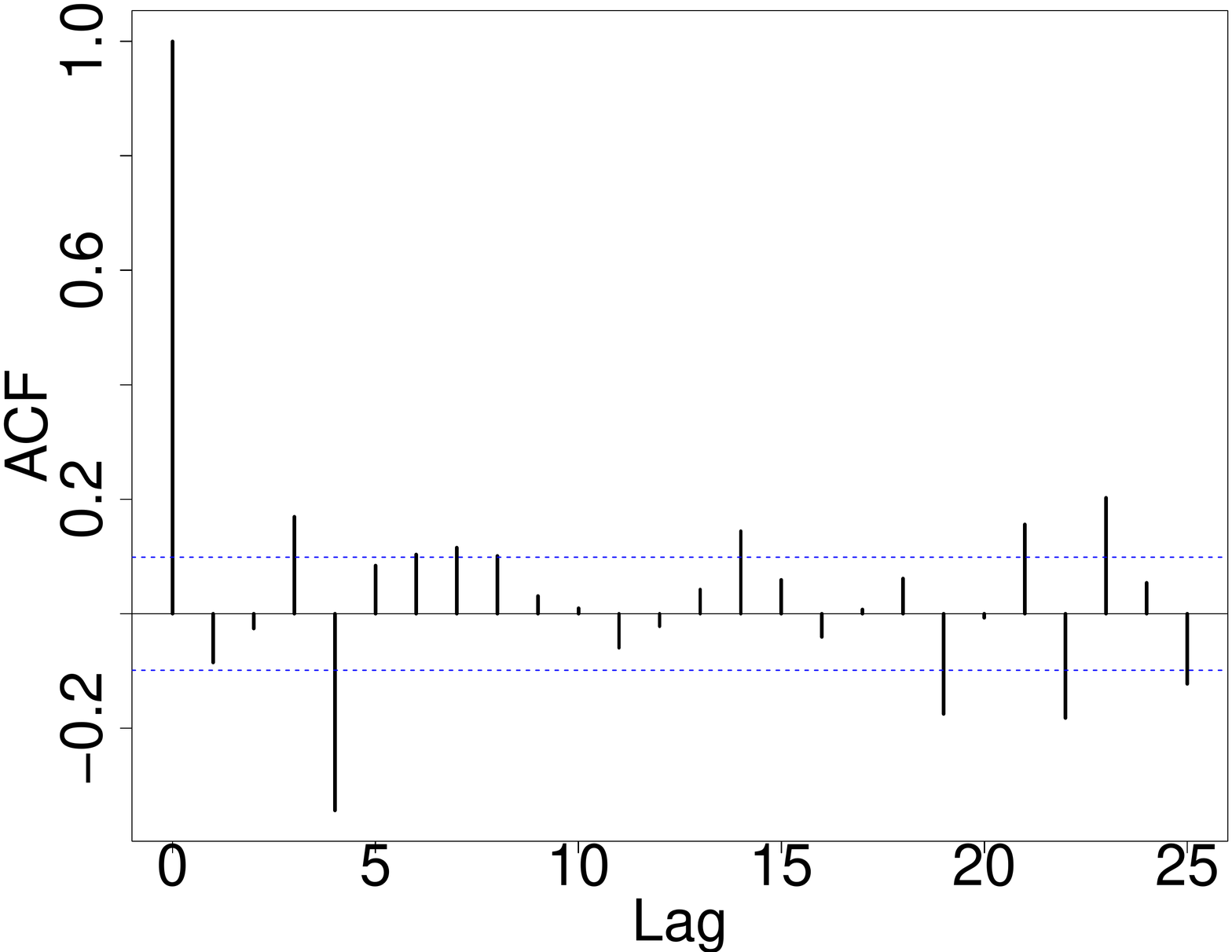}
         \subcaption{NYC $\widetilde{\epsilon}(t) (\Delta I)$}
     \end{subfigure}
     \begin{subfigure}[b]{0.19\textwidth}
         \centering
         \includegraphics[width=\textwidth]{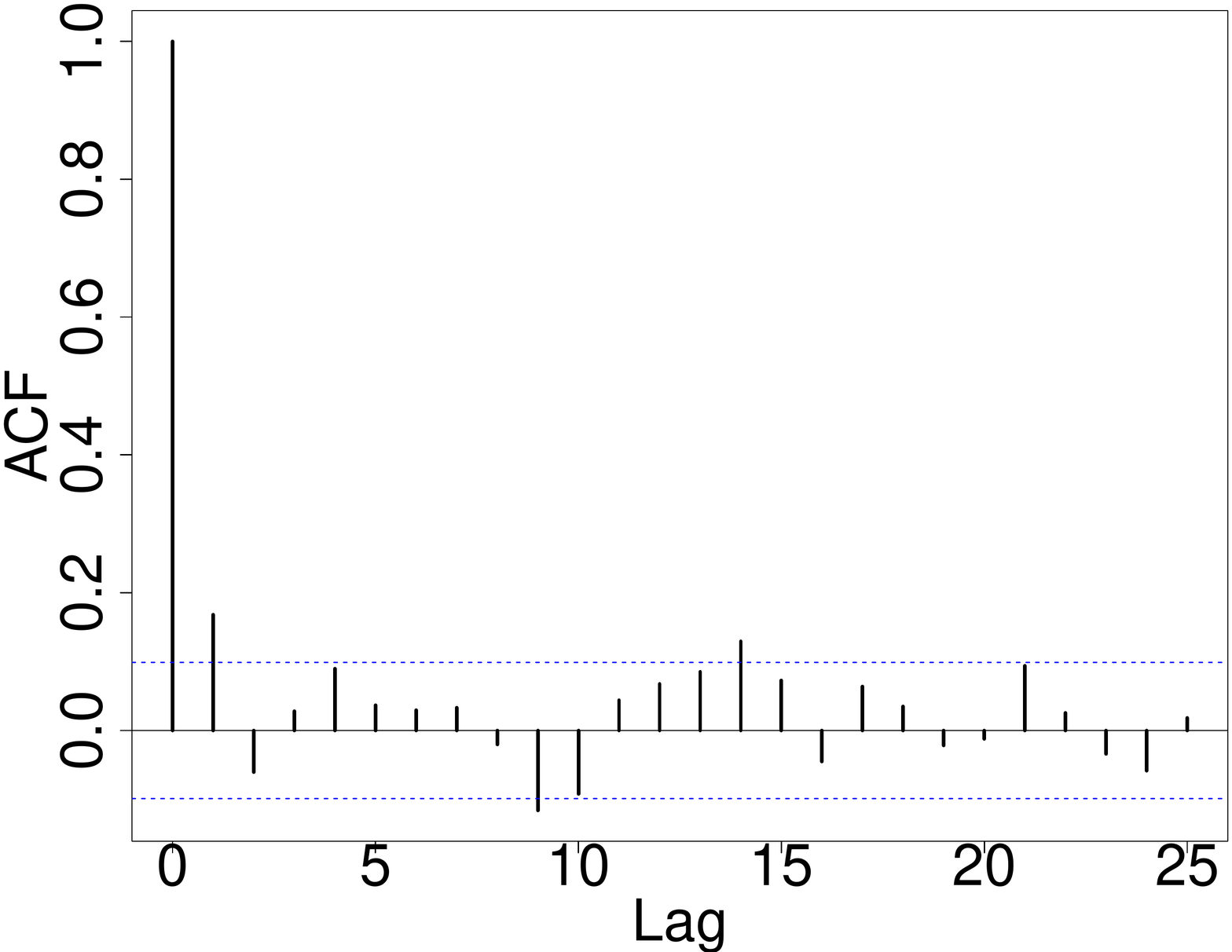}
         \subcaption{NYC $\widetilde{\epsilon}(t) (\Delta R)$}
     \end{subfigure}

     \begin{subfigure}[b]{0.19\textwidth}
         \centering
         \includegraphics[width=\textwidth]{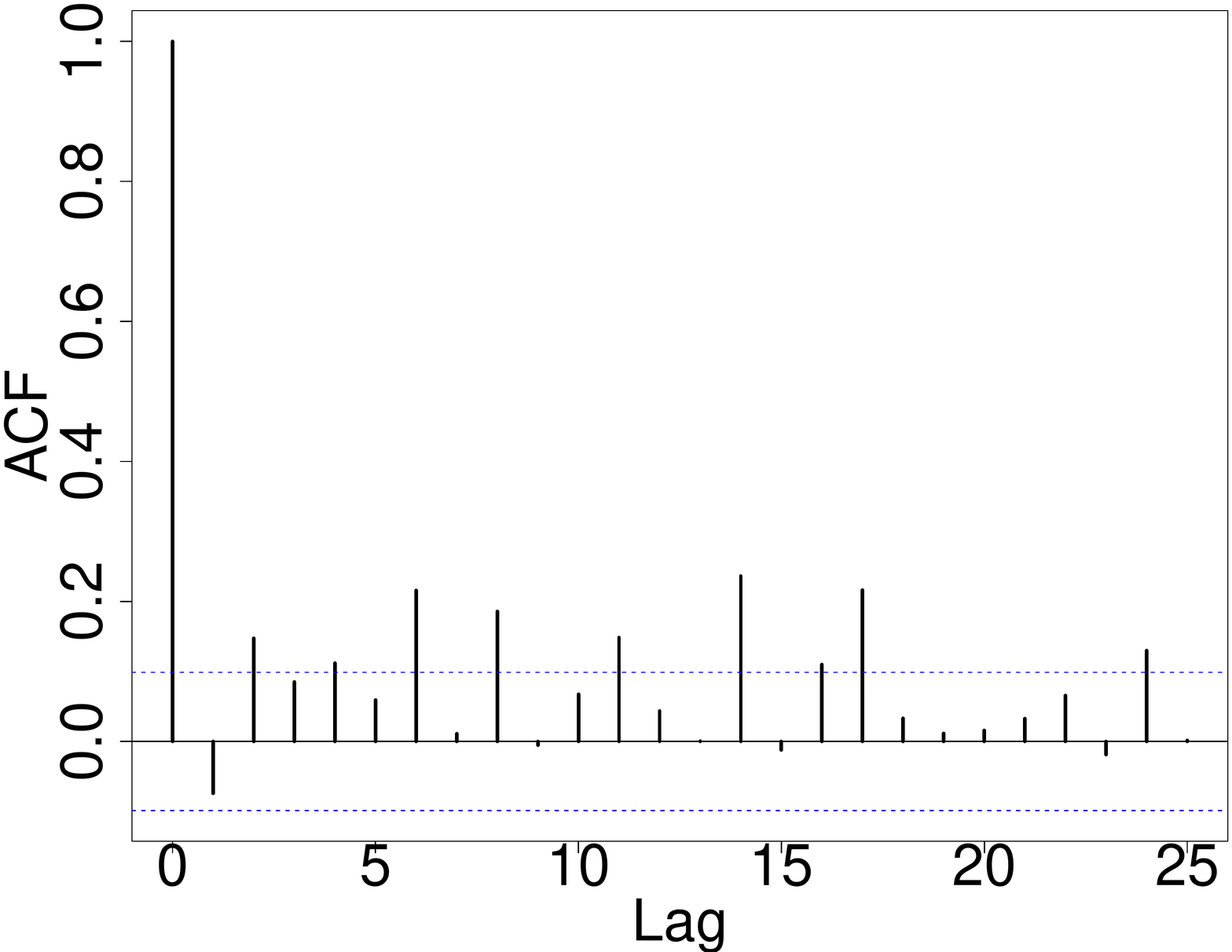}
         \subcaption{King $\widehat{\epsilon}(t) (\Delta I)$}
     \end{subfigure}
     \begin{subfigure}[b]{0.19\textwidth}
         \centering
         \includegraphics[width=\textwidth]{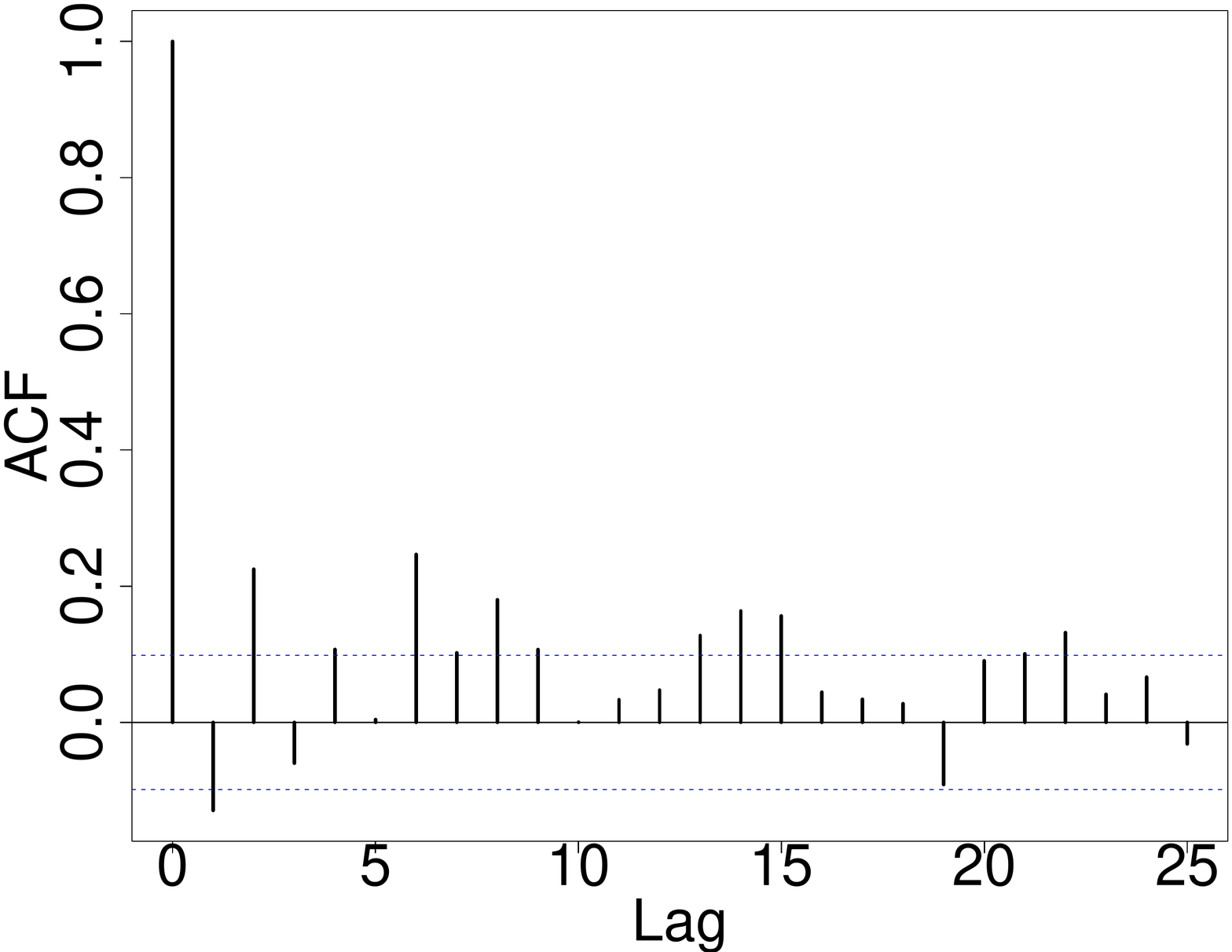}
         \subcaption{King $\widehat{\epsilon}(t) (\Delta R)$}
     \end{subfigure}
     \begin{subfigure}[b]{0.19\textwidth}
         \centering
         \includegraphics[width=\textwidth]{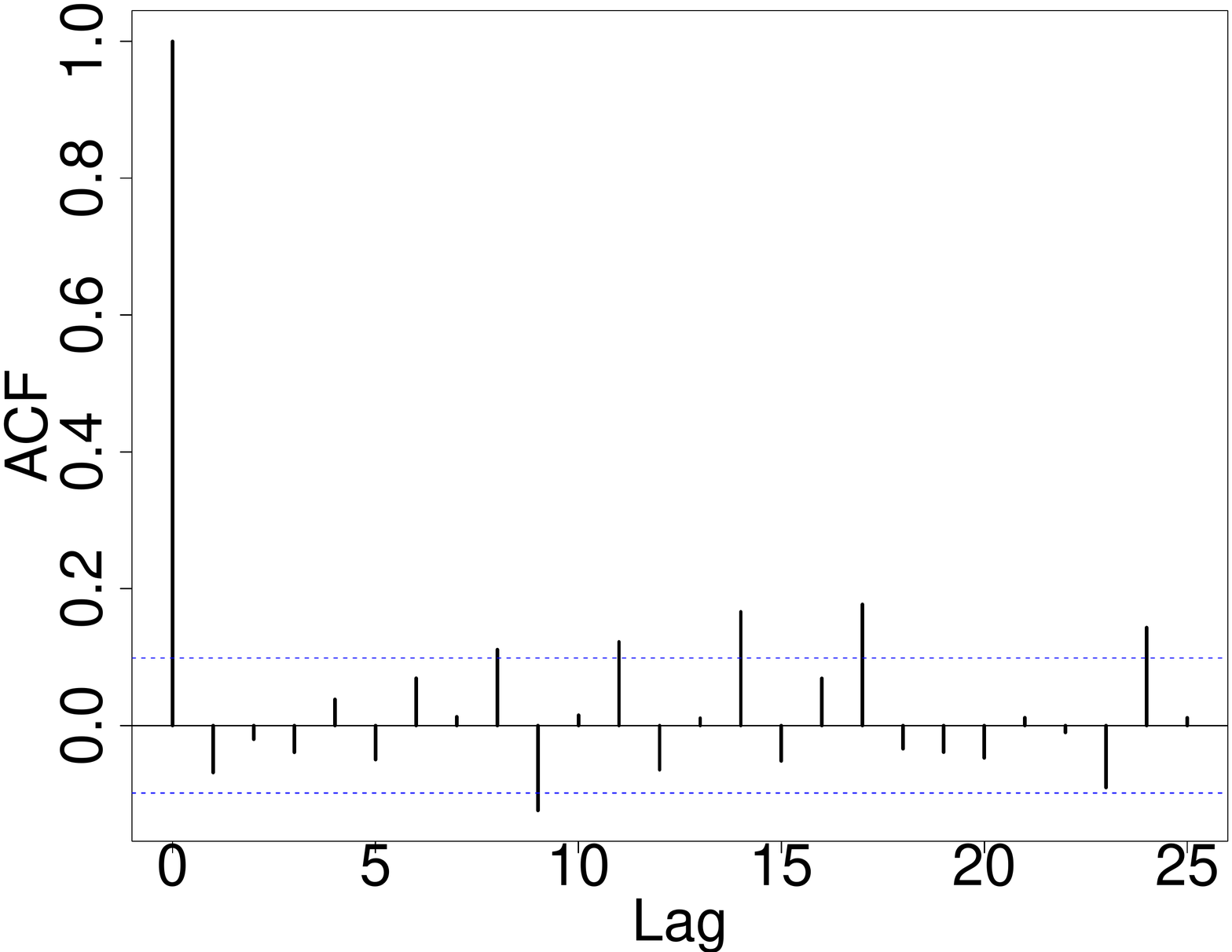}
         \subcaption{King $\widetilde{\epsilon}(t) (\Delta I)$}
     \end{subfigure}
     \begin{subfigure}[b]{0.19\textwidth}
         \centering
         \includegraphics[width=\textwidth]{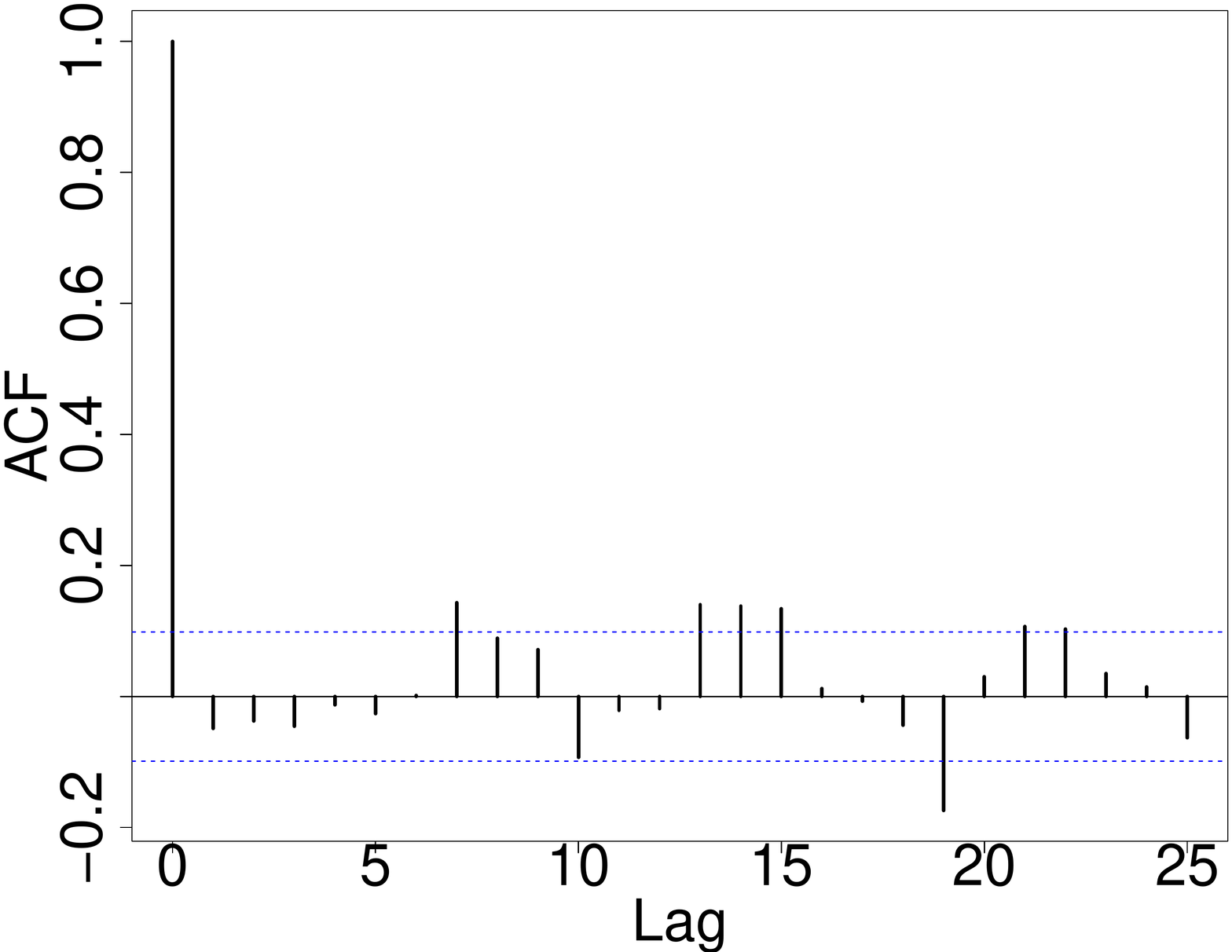}
         \subcaption{King $\widetilde{\epsilon}(t) (\Delta R)$}
     \end{subfigure}

     \begin{subfigure}[b]{0.19\textwidth}
         \centering
         \includegraphics[width=\textwidth]{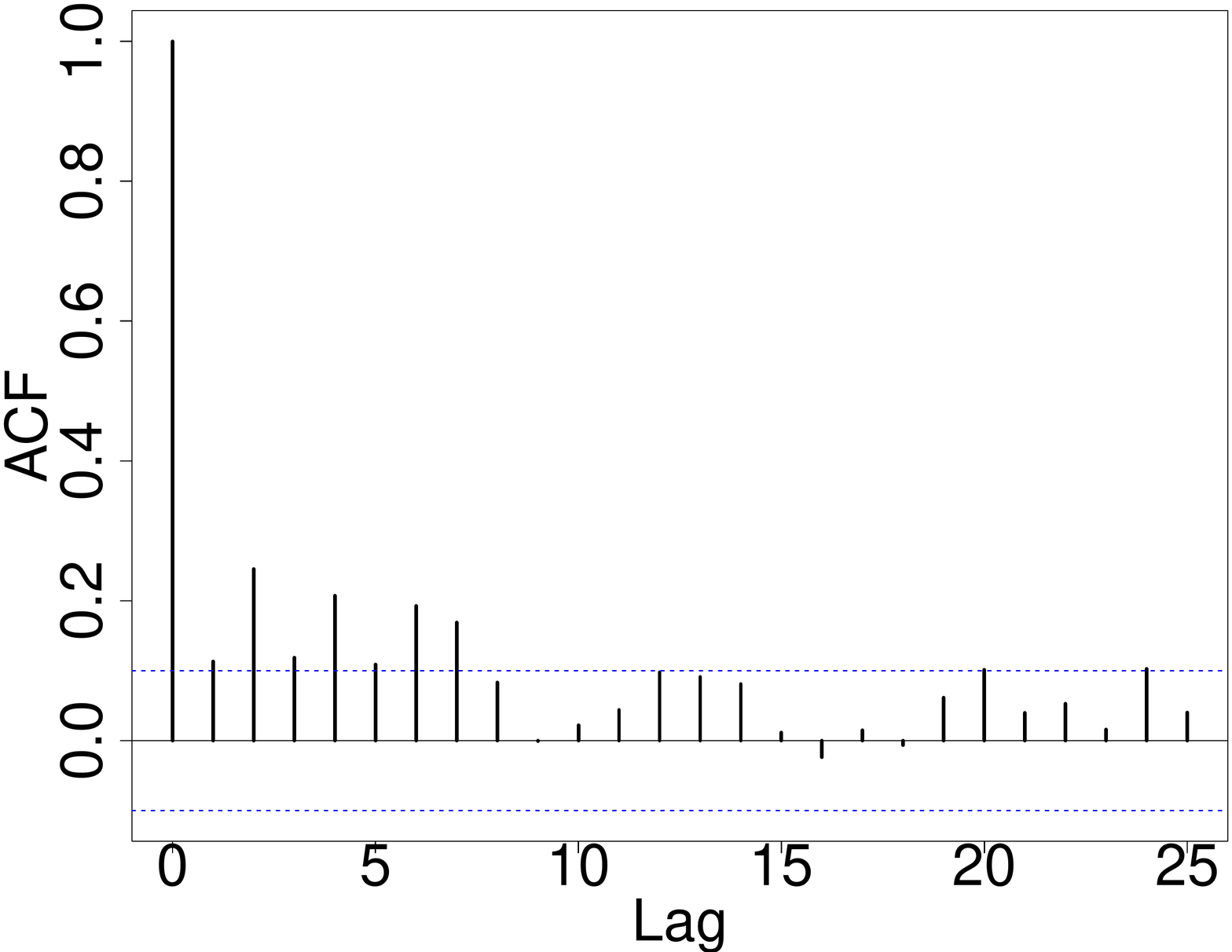}
         \subcaption{Miami  $ \widehat{\epsilon}(t)  (\Delta I)$}
     \end{subfigure}
     \begin{subfigure}[b]{0.19\textwidth}
         \centering
         \includegraphics[width=\textwidth]{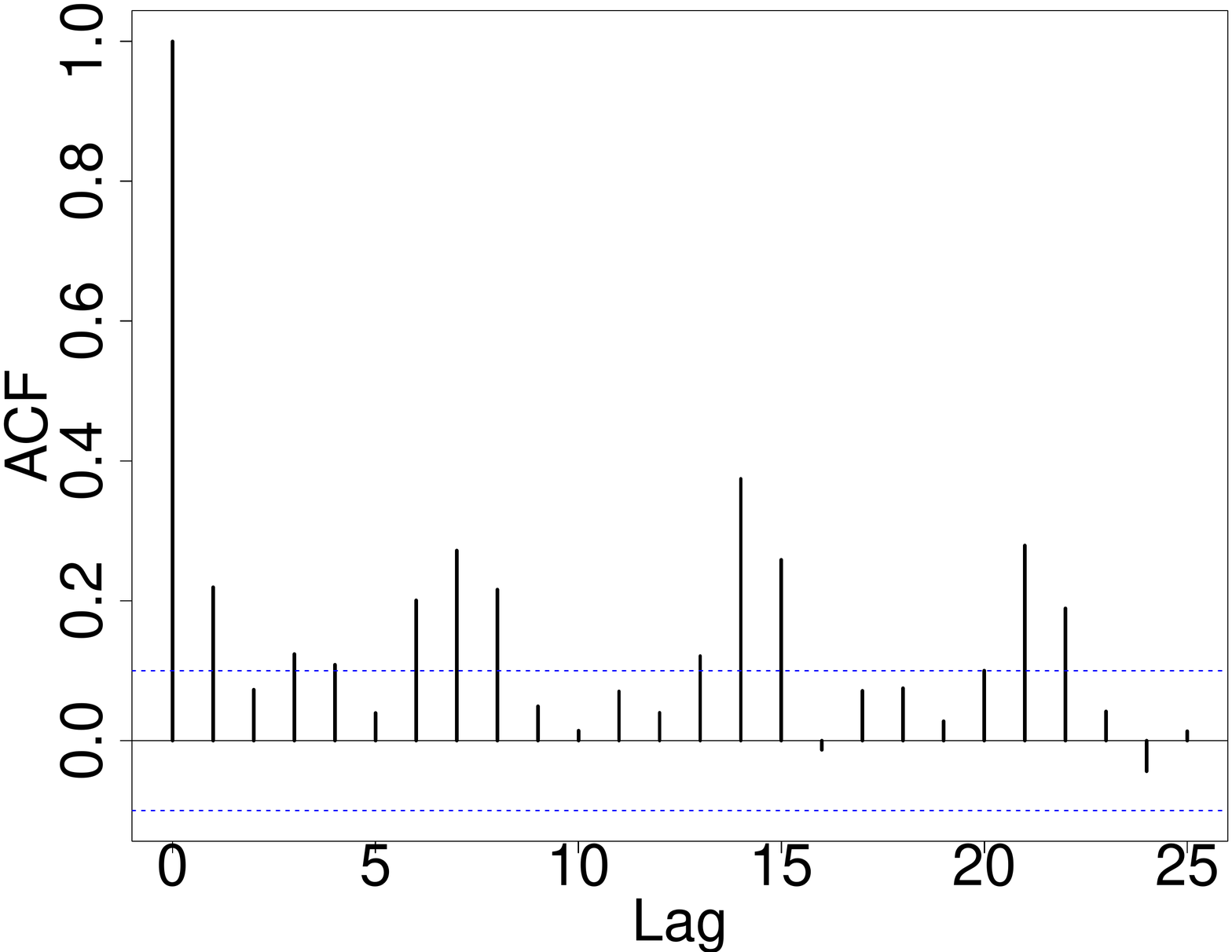}
         \subcaption{Miami $\widehat{\epsilon}(t) (\Delta R)$}
     \end{subfigure}
     \begin{subfigure}[b]{0.19\textwidth}
         \centering
         \includegraphics[width=\textwidth]{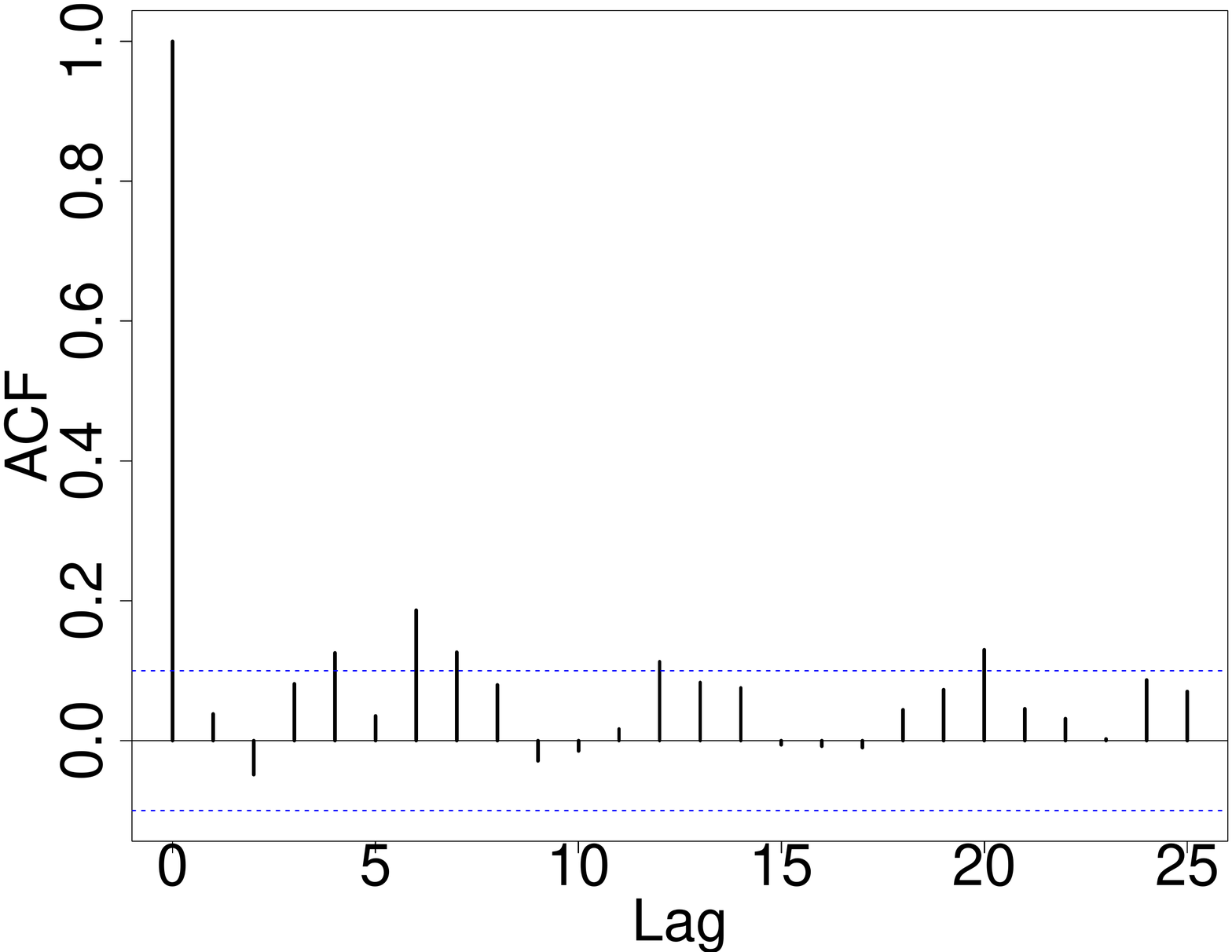}
         \subcaption{Miami $\widetilde{\epsilon}(t) (\Delta I)$}
     \end{subfigure}
     \begin{subfigure}[b]{0.19\textwidth}
         \centering
         \includegraphics[width=\textwidth]{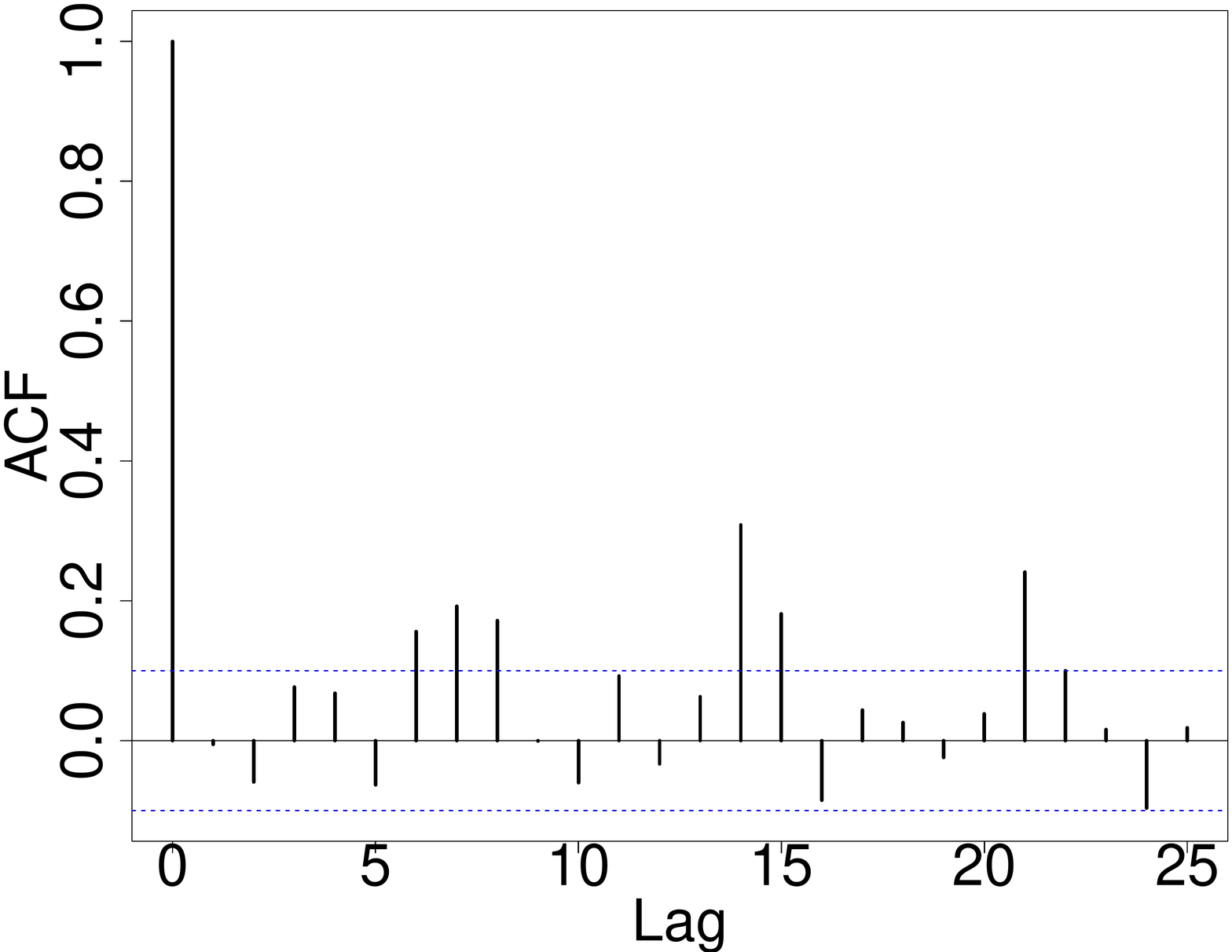}
         \subcaption{Miami $\widetilde{\epsilon}(t) (\Delta R)$}
     \end{subfigure}
     
     \begin{subfigure}[b]{0.19\textwidth}
         \centering
         \includegraphics[width=\textwidth]{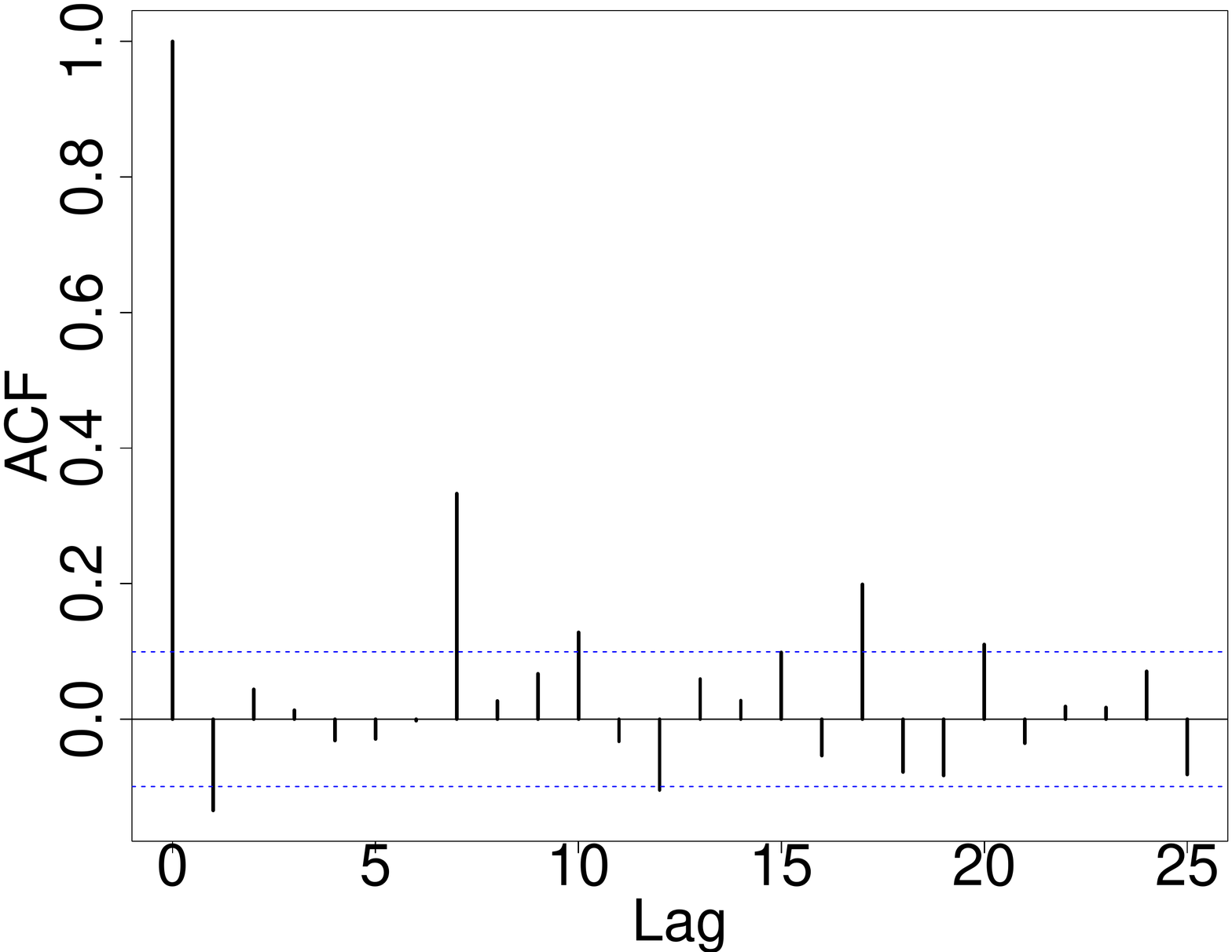}
         \subcaption{Riverside $\widehat{\epsilon}(t) (\Delta I)$}
     \end{subfigure}
     \begin{subfigure}[b]{0.19\textwidth}
         \centering
         \includegraphics[width=\textwidth]{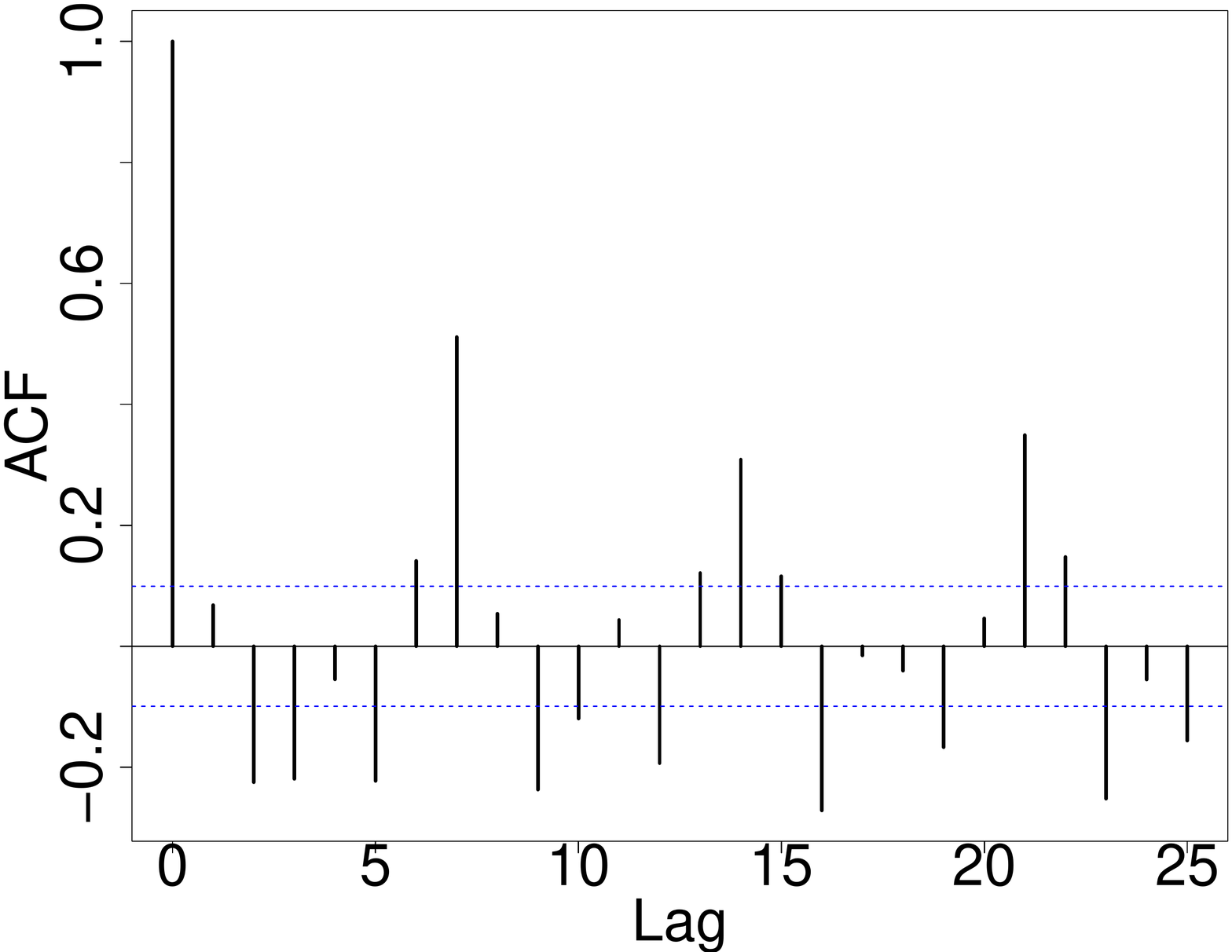}
         \subcaption{Riverside $\widehat{\epsilon}(t) (\Delta R)$}
     \end{subfigure}
     \begin{subfigure}[b]{0.19\textwidth}
         \centering
         \includegraphics[width=\textwidth]{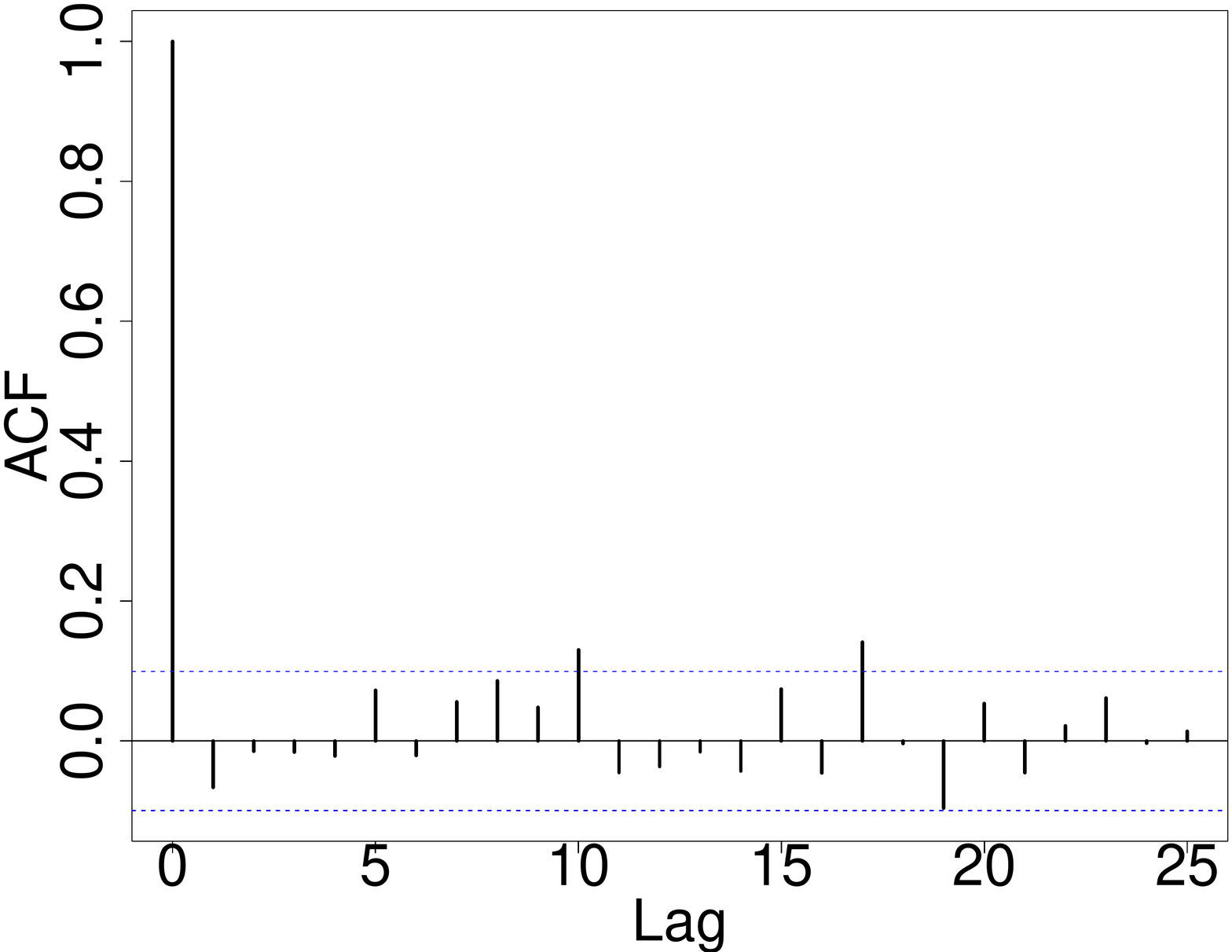}
         \subcaption{Riverside $\widetilde{\epsilon}(t) (\Delta I)$}
     \end{subfigure}
     \begin{subfigure}[b]{0.19\textwidth}
         \centering
         \includegraphics[width=\textwidth]{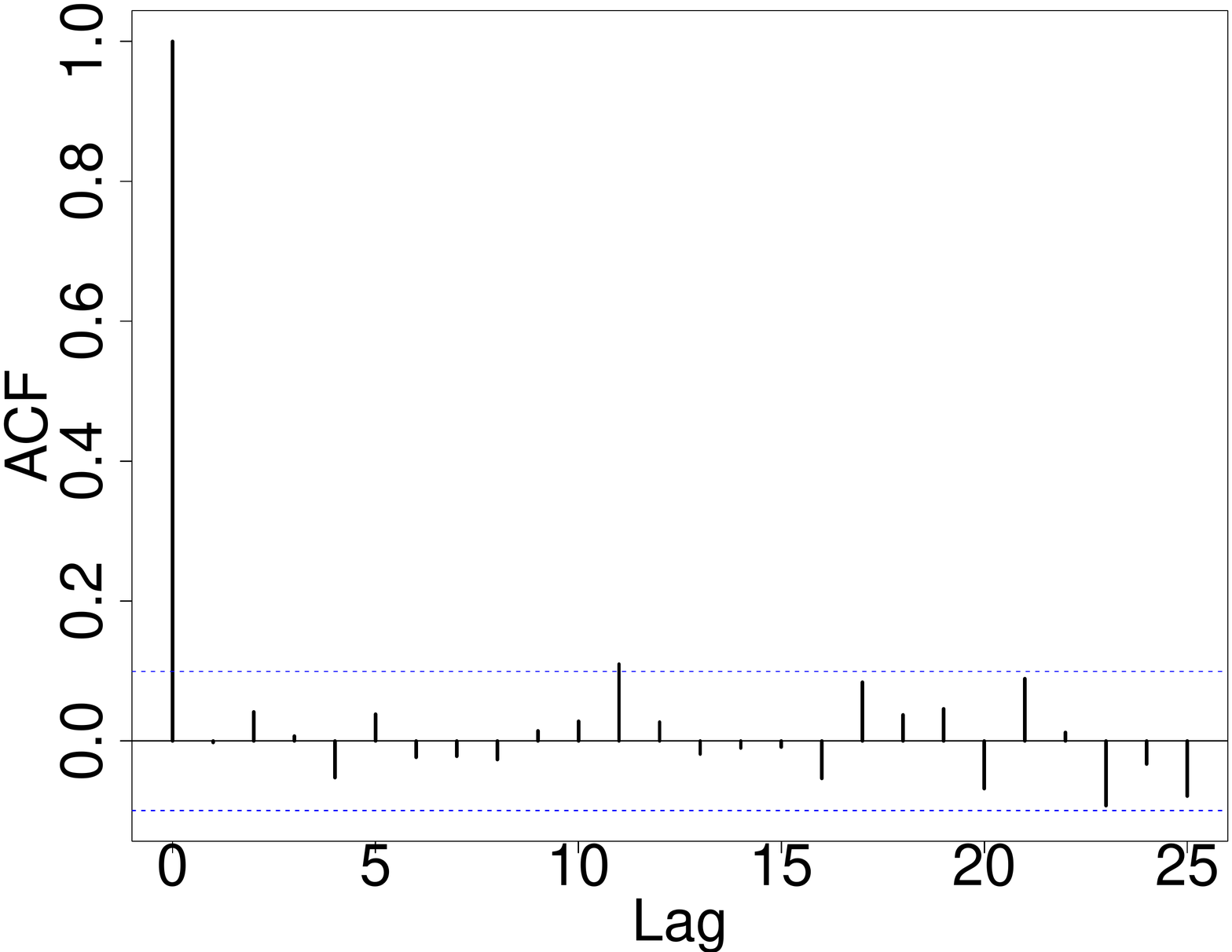}
         \subcaption{Riverside $\widetilde{\epsilon}(t) (\Delta R)$}
     \end{subfigure}
     
     \begin{subfigure}[b]{0.19\textwidth}
         \centering
         \includegraphics[width=\textwidth]{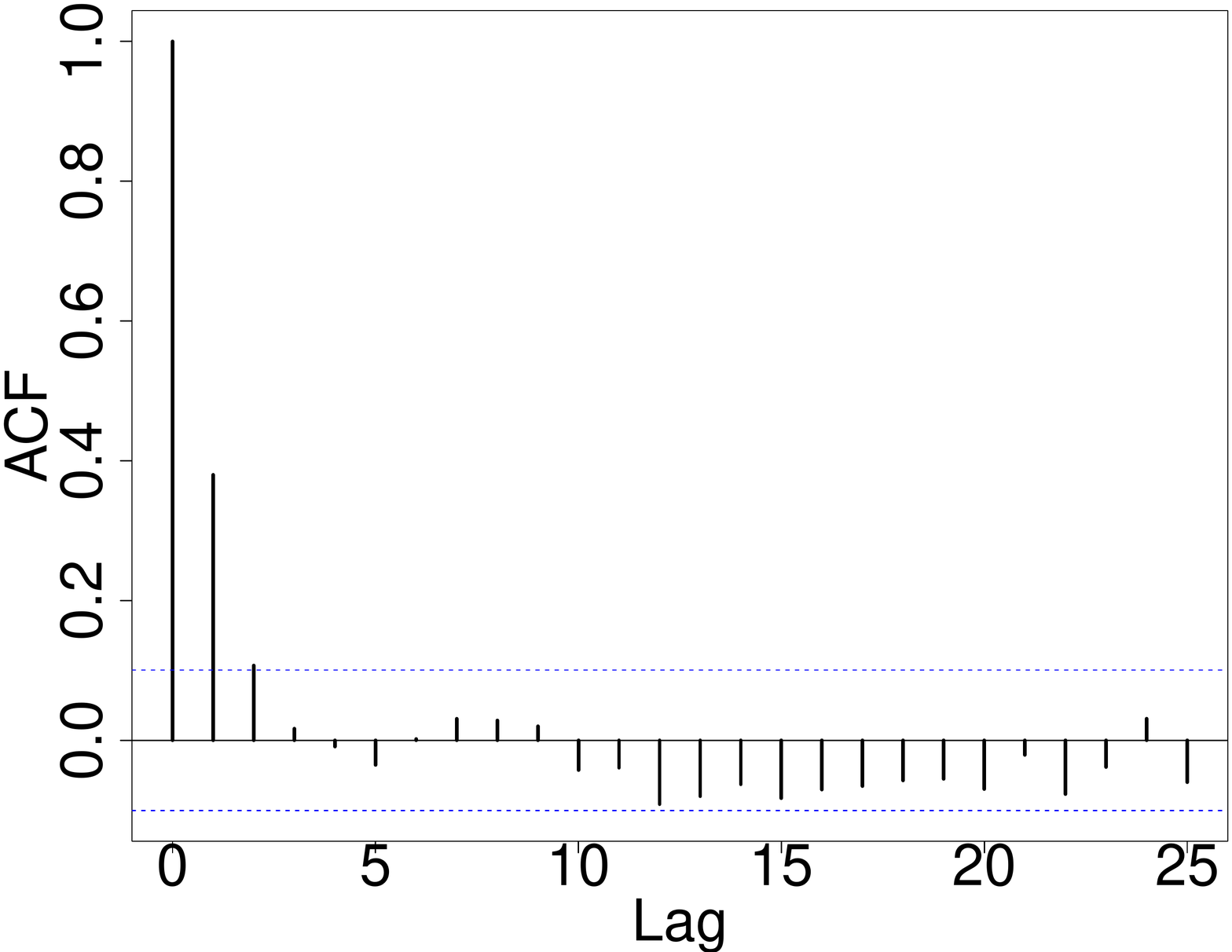}
         \subcaption{Santa Barbara  $ \widehat{\epsilon}(t)  (\Delta I)$}
     \end{subfigure}
     \begin{subfigure}[b]{0.19\textwidth}
         \centering
         \includegraphics[width=\textwidth]{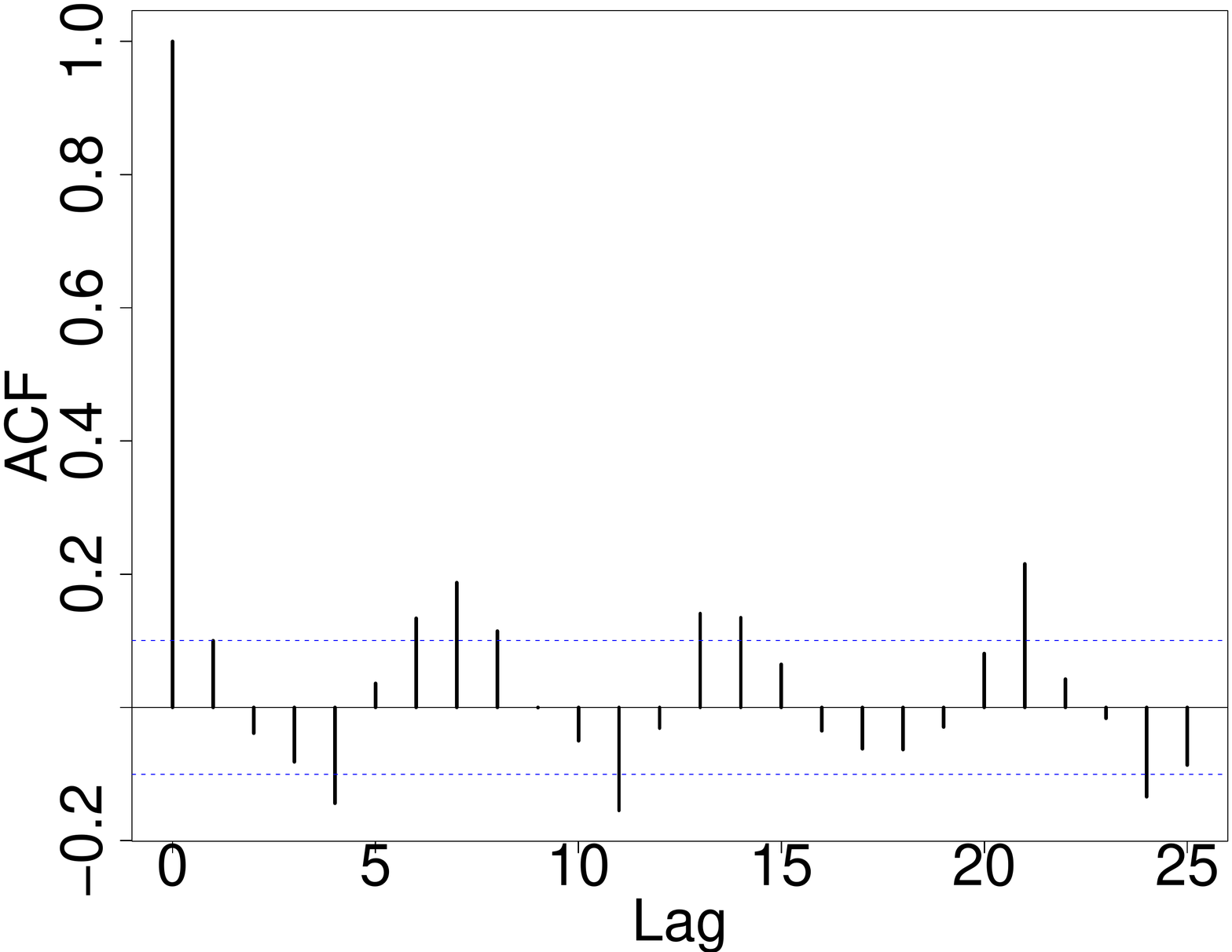}
         \subcaption{Santa Barbara $\widehat{\epsilon}(t) (\Delta R)$}
     \end{subfigure}
     \begin{subfigure}[b]{0.19\textwidth}
         \centering
         \includegraphics[width=\textwidth]{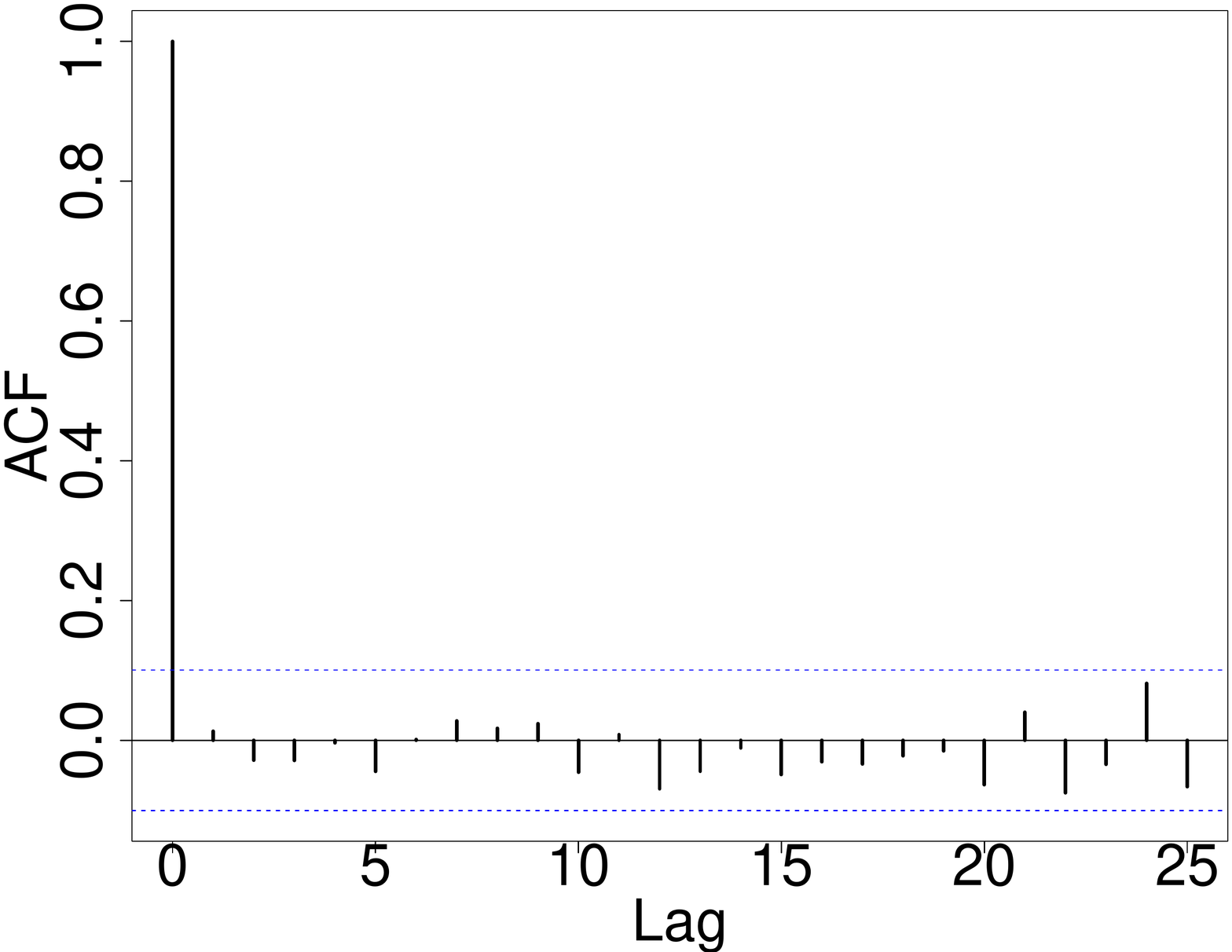}
         \subcaption{Santa Barbara $\widetilde{\epsilon}(t) (\Delta I)$}
     \end{subfigure}
     \begin{subfigure}[b]{0.19\textwidth}
         \centering
         \includegraphics[width=\textwidth]{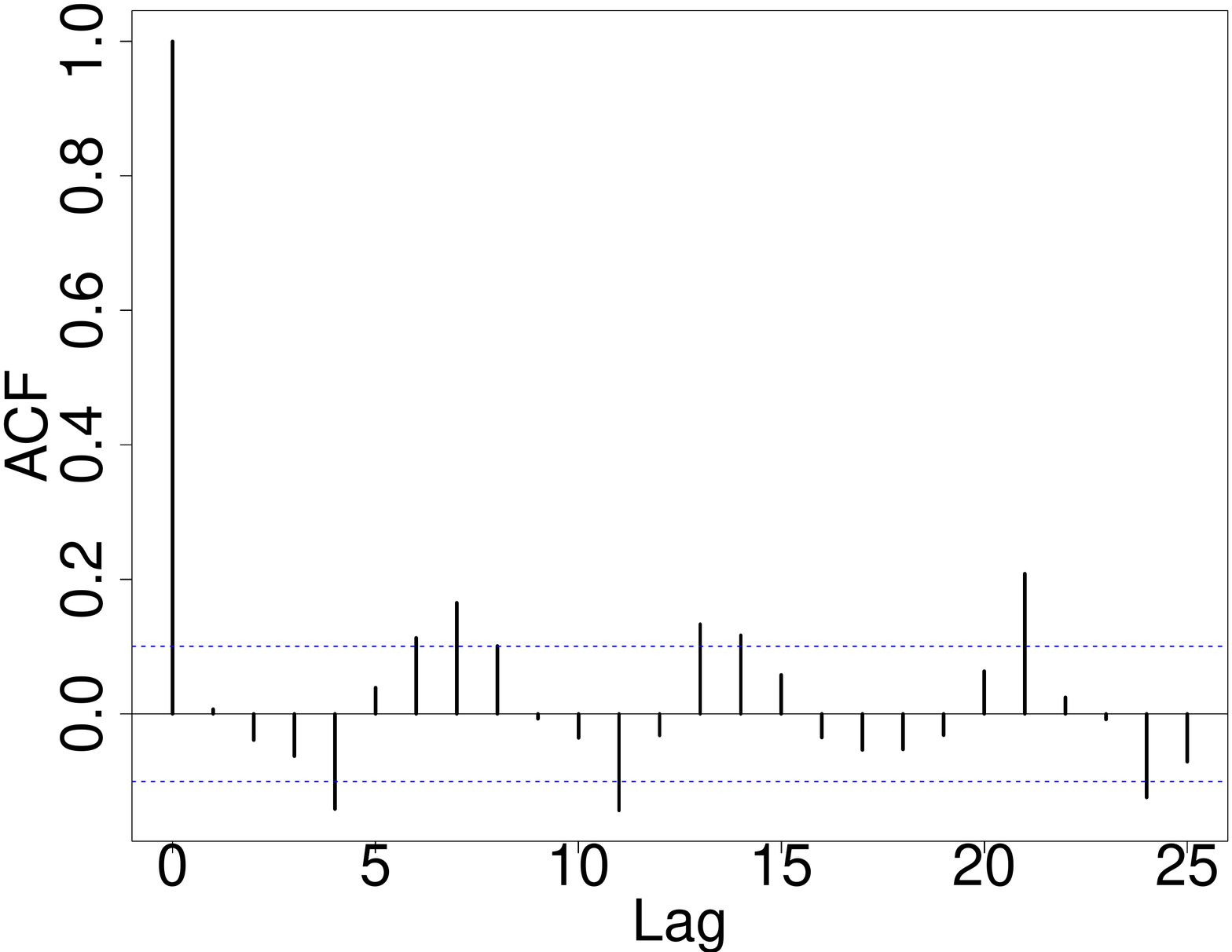}
         \subcaption{Santa Barbara $\widetilde{\epsilon}(t) (\Delta R)$}
     \end{subfigure}
        \caption{Auto-correlation plot of residuals in selected counties by  Model 2.3 (left: piecewise constant SIR model + spatial effect ) and Model 3 (right: piecewise constant SIR model + spatial effect + VAR($p$).}
        \label{fig:acf_county}
\end{figure*}

\begin{table}[!ht]
\caption{\label{table_MPE_county_out_refit}Out-of-sample mean relative prediction error (MRPE)   of $I(t)$ and $R(t)$.   }
\centering
{
\begin{tabular}{lcccccccccc} 
  \hline
  \hline
  & \multicolumn{2}{c}{NYC}& \multicolumn{2}{c}{King} & \multicolumn{2}{c}{Miami-Dade} & \multicolumn{2}{c}{Riverside} & \multicolumn{2}{c}{Santa Barbara} \\
    &I &R  &I &R & I & R& I & R & I & R  \\
  \hline
Model 1 & 8e-04 & 0.0013 & 0.0024 & 0.0026 & 0.0012 & 0.0022 & 0.0056 & 0.0036 & 0.0037 & 0.0071 \\ 
  Model 2.1 & 8e-04 & 0.0012 & 0.0027 & 0.0022 & 0.0011 & 0.0011 &  0.0043 & 0.0028 & 0.0016 & 0.0039 \\ 
  Model 2.2 & 8e-04 & 0.0013 & 0.0028 & 0.0021 & 0.0011 & 0.0011 & 0.0045 & 0.0027 & 0.0014 & 0.0032 \\ 
  Model 2.3 & 7e-04 & 8e-04 & 0.0027 & 0.0018 & 9e-04 & 8e-04 &  0.0014 & 0.0025 & 0.0014 & 0.0028 \\ 
  Model 2.4 & 7e-04 & 5e-04 & 0.0024 & 0.0011 & 0.001 & 0.001 &  0.0013 & 0.0027 & 0.0017 & 0.0033 \\ 
  Model 3 & 7e-04 & 5e-04 & 0.0027 & 0.0017 & 9e-04 & 8e-04 & 0.0012 & 0.0016 & 0.0014 & 0.0028 \\ 
  \hline
\end{tabular}}
\end{table}

\begin{table}[!ht]
\caption{\label{table_MPE_county_out_refit_2}Out-of-sample mean relative prediction error (MRPE)   of $I(t)$ and $R(t)$.   }
\centering
{
\begin{tabular}{llcccccccccccccc} 
  \hline
  \hline
  & \multicolumn{2}{c}{Charleston}& \multicolumn{2}{c}{Greenville} & \multicolumn{2}{c}{Richland} & \multicolumn{2}{c}{Horry} \\
    & I & R  & I & R  & I & R &I &R    \\
  \hline
 Model 1 & 0.0032 & 0.004 & 0.0034 & 0.0048 & 0.0039 & 0.0033 & 0.004 & 0.0037  \\ 
  Model 2.1 & 0.0013 & 0.0024 & 0.0015 & 0.0019 & 0.0033 & 0.0025 & 0.003 & 0.0027  \\ 
  Model 2.2 & 0.0012 & 0.0024 & 0.0015 & 0.0018 & 0.0033 & 0.0025 & 0.0032 & 0.0028  \\ 
  Model 2.3 & 8e-04 & 0.0023 & 0.0012 & 0.0018 & 0.0021 & 0.0021 & 0.002 & 0.0021 \\ 
  Model 2.4 & 8e-04 & 0.0024 & 0.0013 & 0.0017 & 0.0021 & 0.0018 & 0.0018 & 0.0018  \\ 
  Model 3 & 8e-04 & 0.0023 & 0.0012 & 0.0017 & 0.0021 & 0.0021 & 0.0019 & 0.0021 \\ 
  \hline
\end{tabular}}
\end{table}

We also provide the out-of-sample MRPE of $I(t)$ and $R(t)$  in  \Cref{table_MPE_county_out_refit,table_MPE_county_out_refit_2}. 
The MRPE of $I(t)$ results show that the piecewise constant model (Model 1) performs the best in King County and R, while the piecewise constant model with spatial effect (Model 2) performs the best in New York City, Charleston, Greenville, Richland and Horry County. Adding the VAR($p$) (Model 3) performs the best in  Miami-Dade County, Riverside County and Santa Barbara County.
 The MRPE of $R(t)$ results show that the piecewise constant model with spatial effect (Model 2) performs the best in New York City,  King County, Greenville, Richland and Horry County, while adding the VAR($p$) (Model 3) performs the best in Miami-Dade County, Riverside County and Santa Barbara County and Charleston.

\begin{figure*}[ht!]
     \centering
          \captionsetup[sub]{font=scriptsize, labelfont={bf,sf}}
     \begin{subfigure}[b]{0.16\textwidth}
         \centering
         \includegraphics[width=\textwidth]{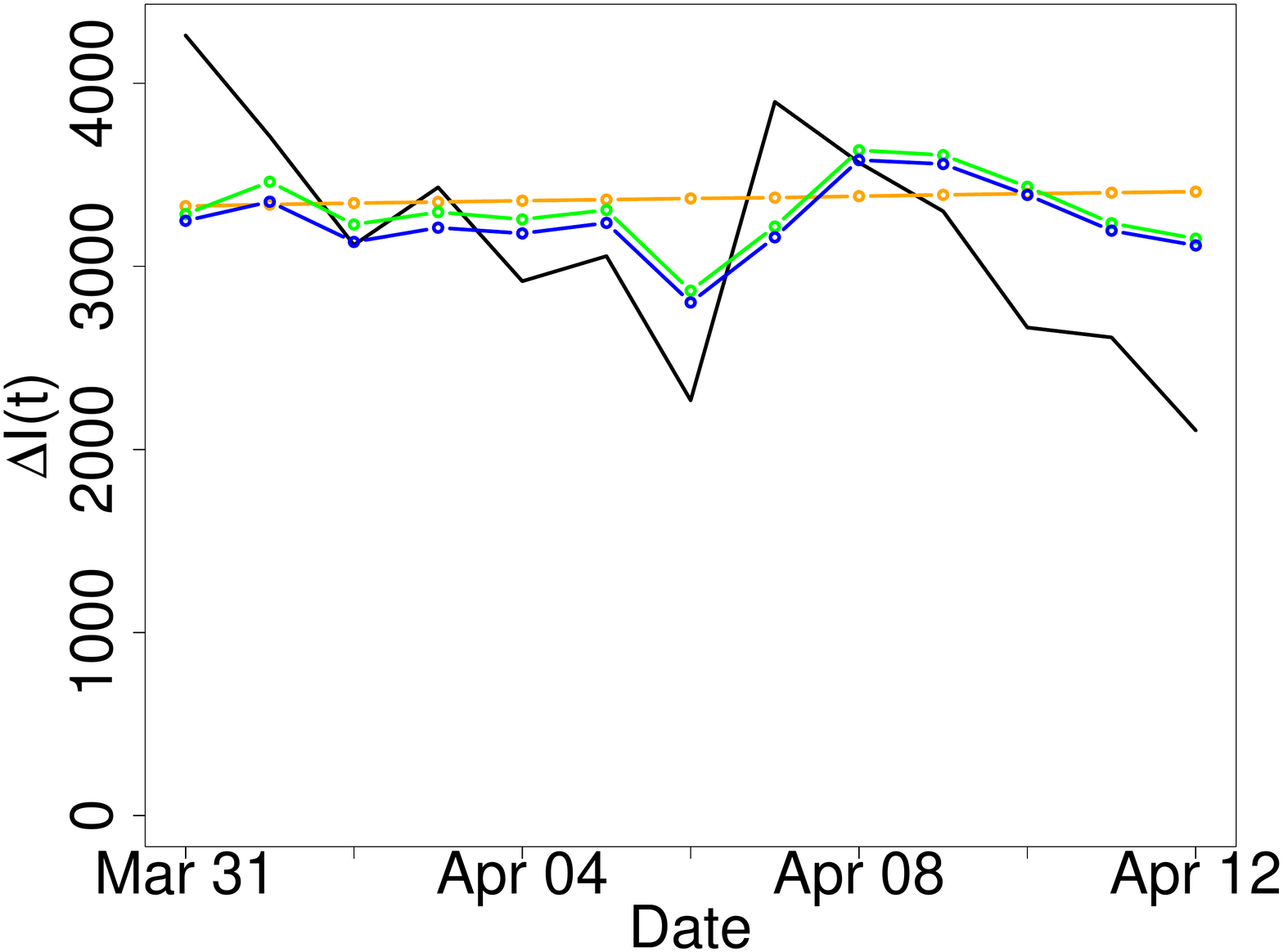}
         \subcaption{NYC $\widehat{\Delta I}(t)$}
     \end{subfigure}
     \begin{subfigure}[b]{0.16\textwidth}
         \centering
         \includegraphics[width=\textwidth]{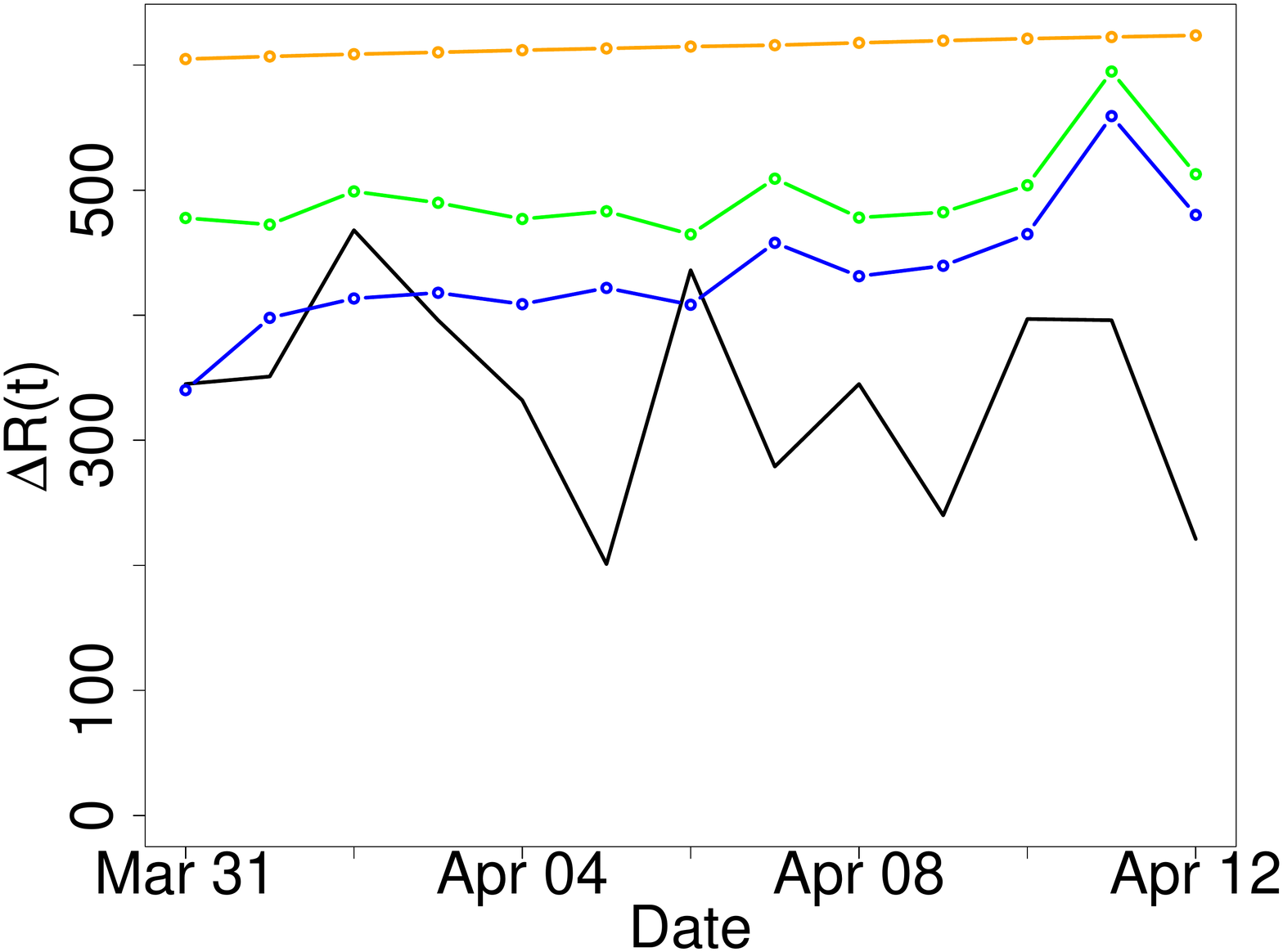}
         \subcaption{NYC $\widehat{\Delta R}(t)$}
     \end{subfigure}
     \begin{subfigure}[b]{0.16\textwidth}
         \centering
         \includegraphics[width=\textwidth]{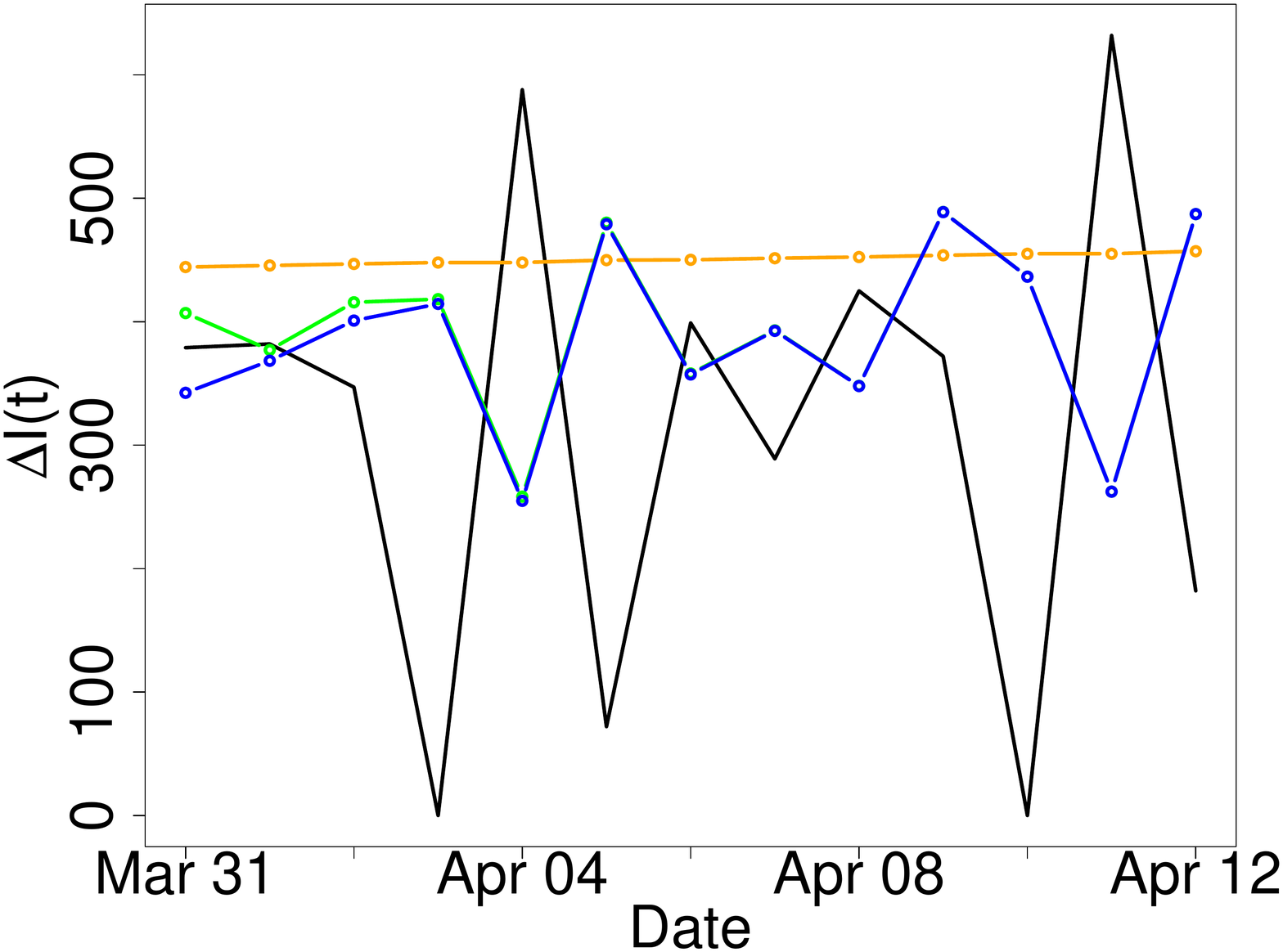}
         \subcaption{King $\widehat{\Delta I}(t)$}
     \end{subfigure}
     \begin{subfigure}[b]{0.16\textwidth}
         \centering
         \includegraphics[width=\textwidth]{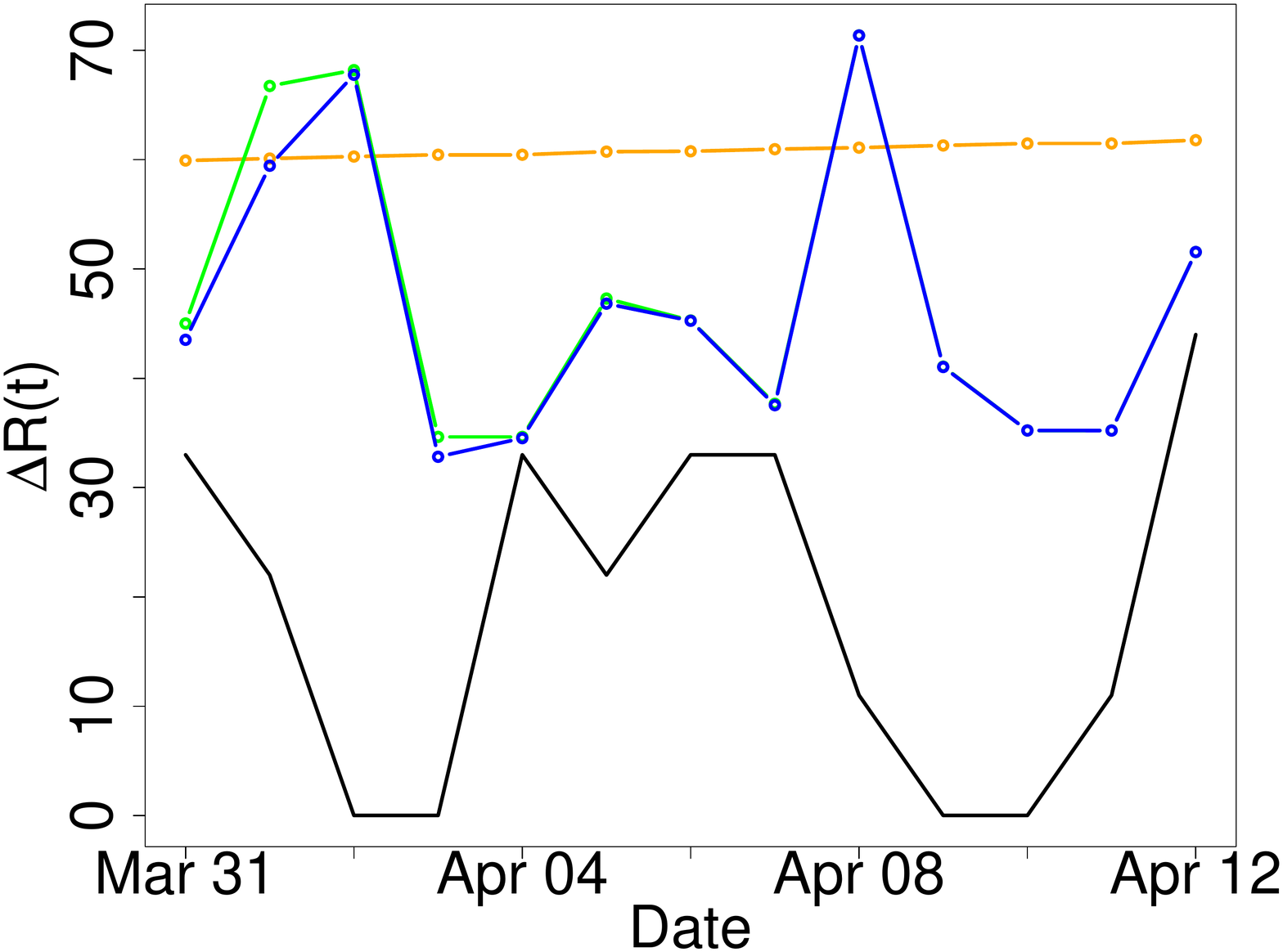}
         \subcaption{King $\widehat{\Delta R}(t)$}
     \end{subfigure}
     \begin{subfigure}[b]{0.16\textwidth}
         \centering
         \includegraphics[width=\textwidth]{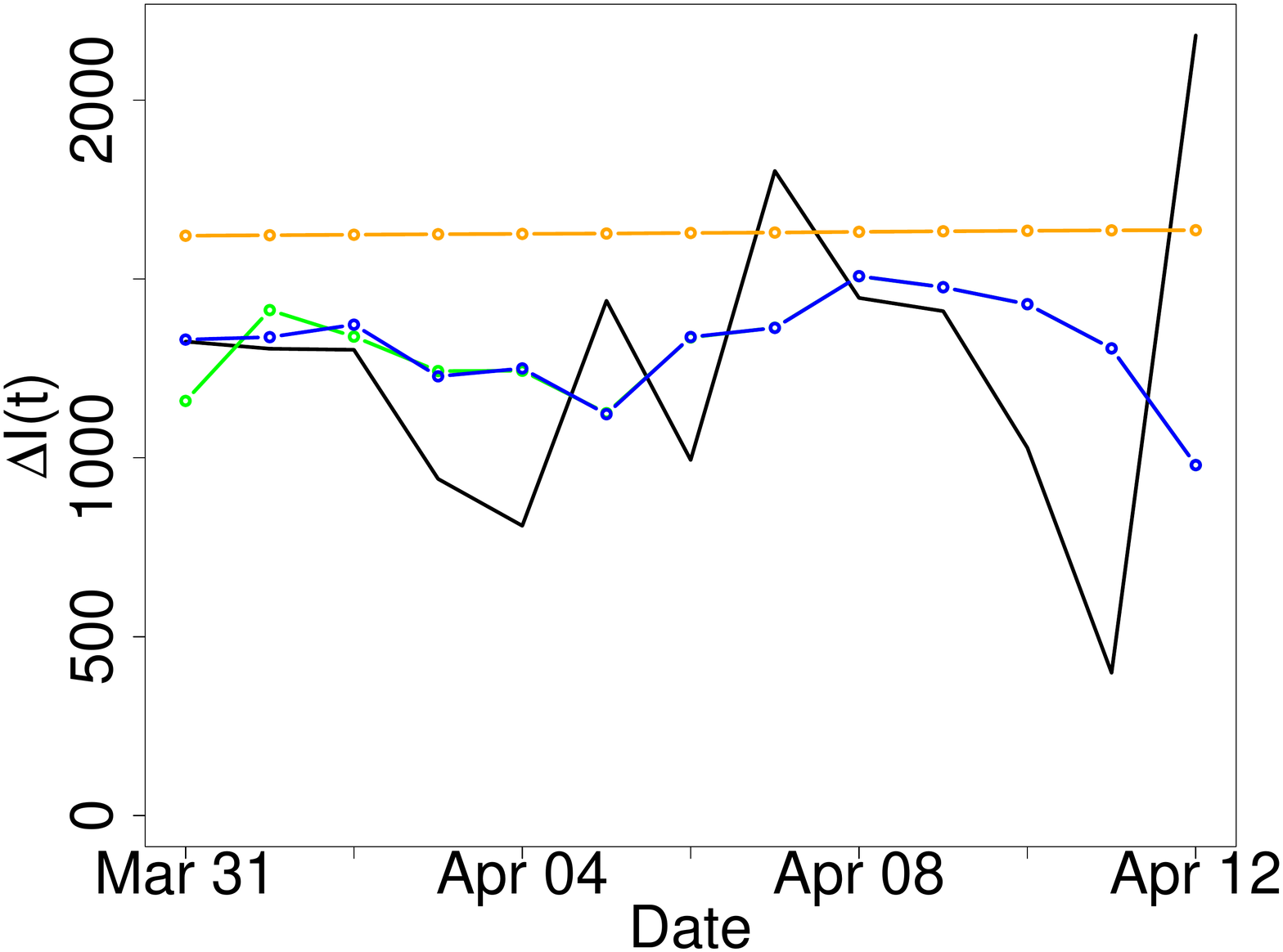}
         \subcaption{Miami $\widehat{\Delta I}(t)$}
     \end{subfigure}
     \begin{subfigure}[b]{0.16\textwidth}
         \centering
         \includegraphics[width=\textwidth]{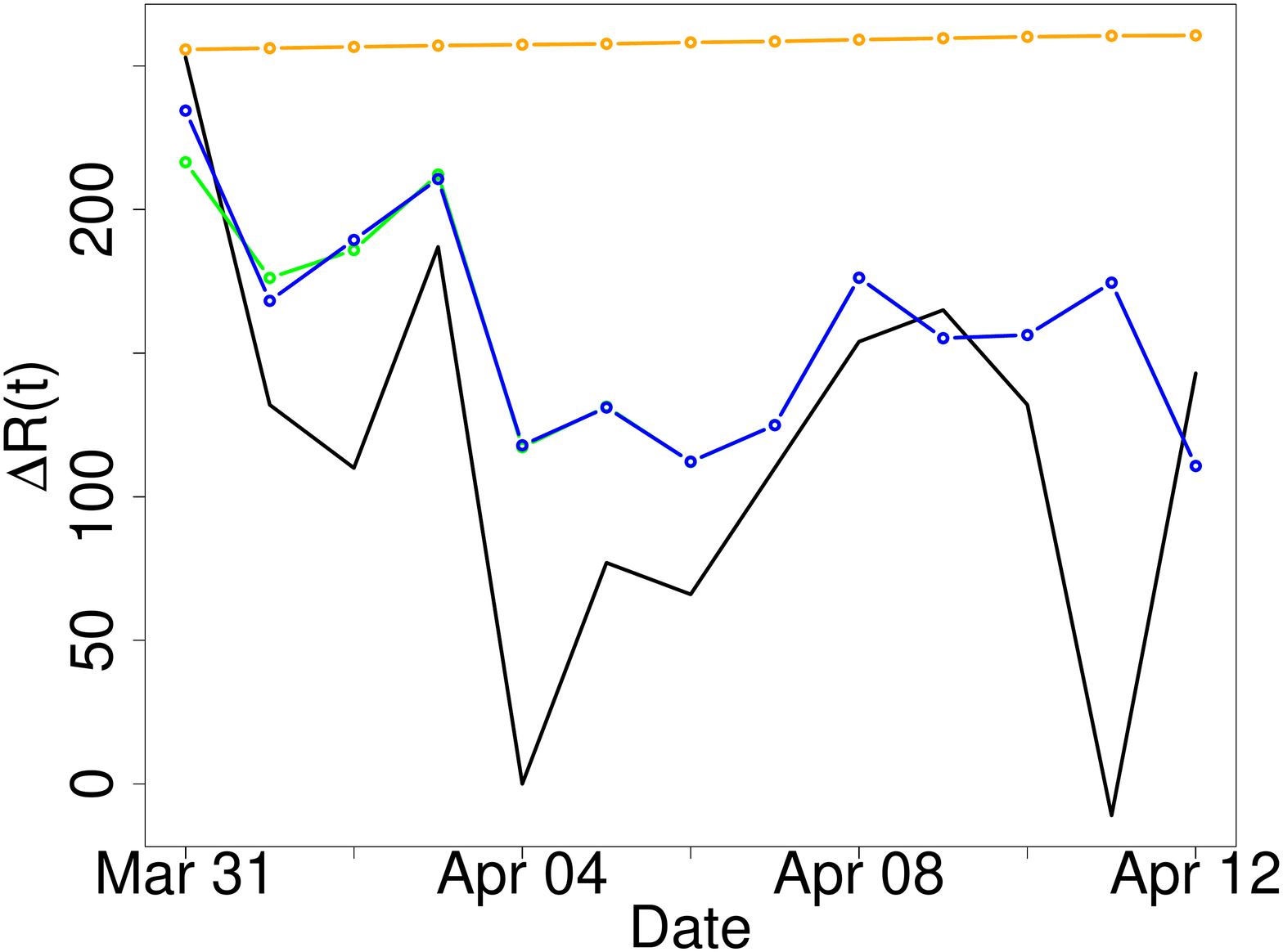}
         \subcaption{Miami $\widehat{\Delta R}(t)$}
     \end{subfigure}
     \begin{subfigure}[b]{0.16\textwidth}
         \centering
         \includegraphics[width=\textwidth]{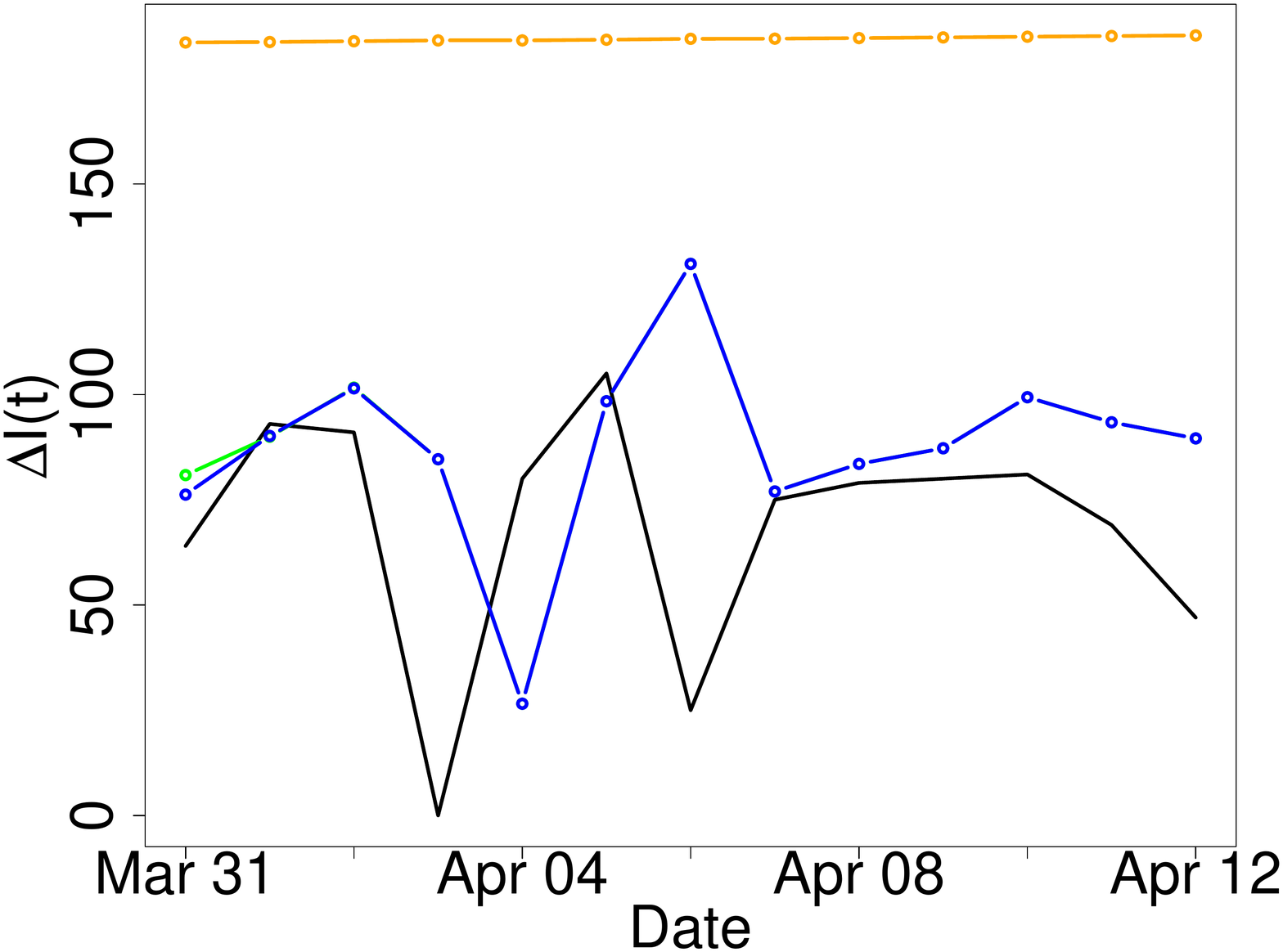}
         \subcaption{Charleston $\widehat{\Delta I}(t)$}
     \end{subfigure}
     \begin{subfigure}[b]{0.16\textwidth}
         \centering
         \includegraphics[width=\textwidth]{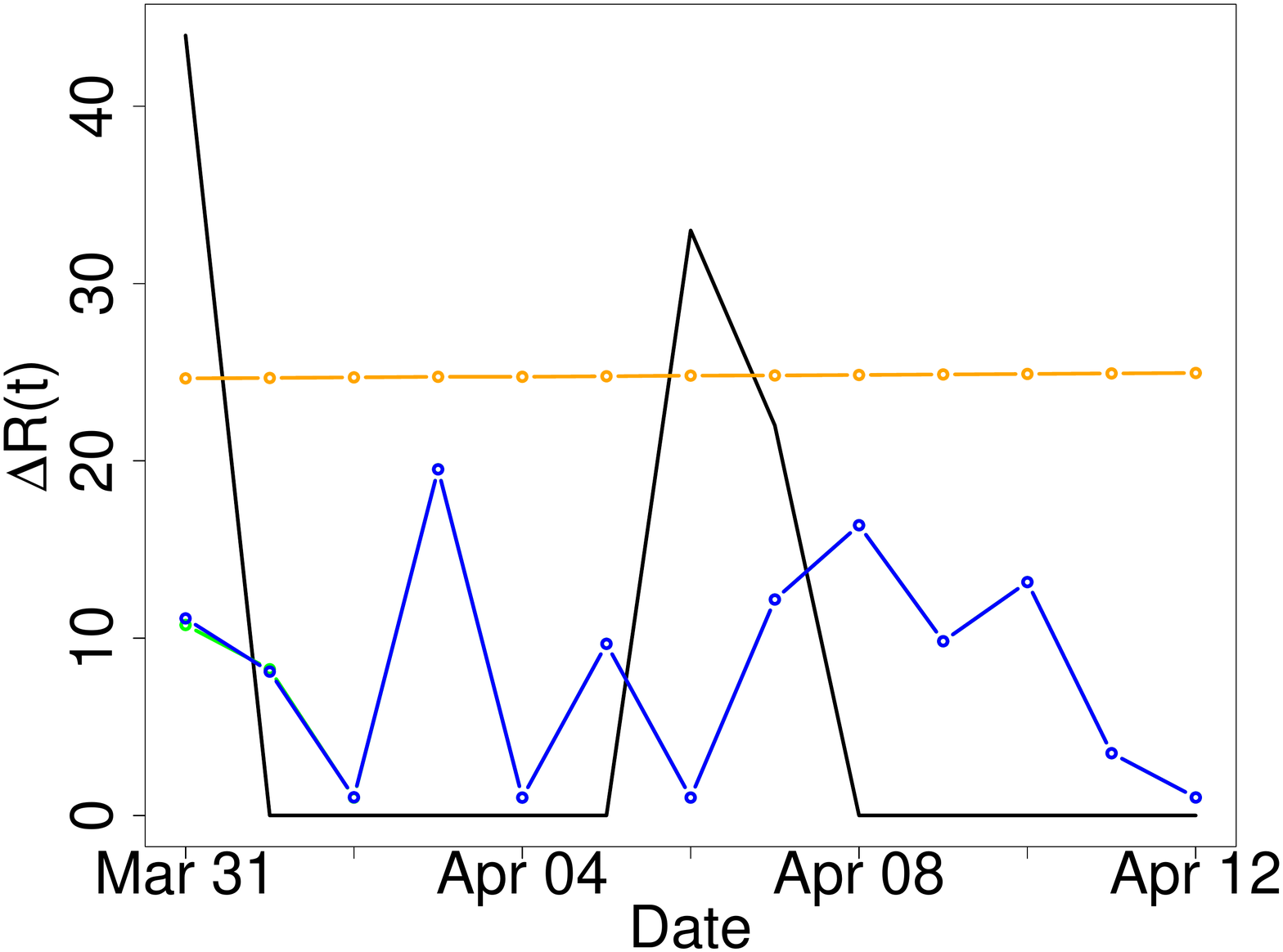}
         \subcaption{Charleston $\widehat{\Delta R}(t)$}
     \end{subfigure}
     \begin{subfigure}[b]{0.16\textwidth}
         \centering
         \includegraphics[width=\textwidth]{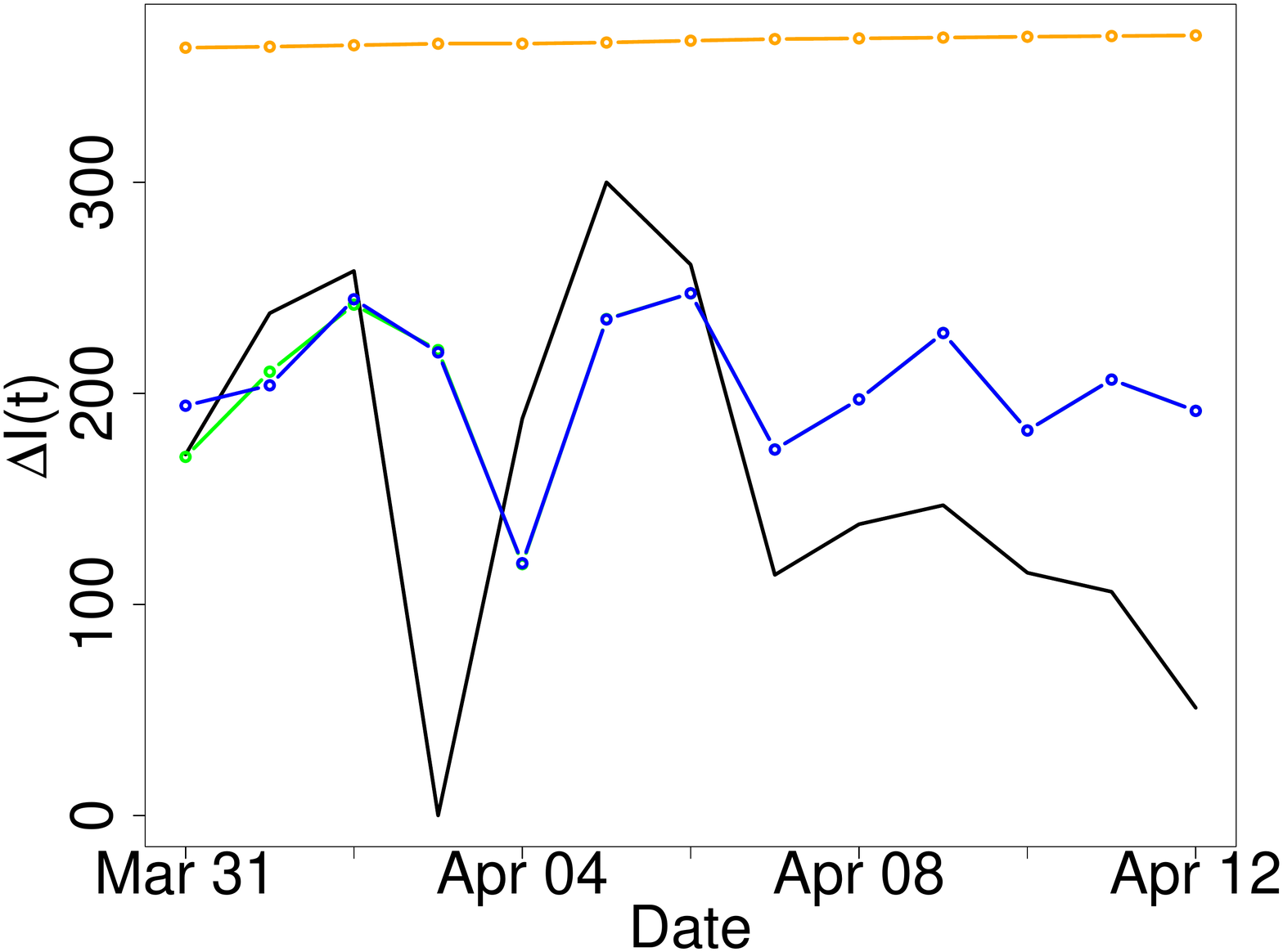}
         \subcaption{Greenville $\widehat{\Delta I}(t)$}
     \end{subfigure}
     \begin{subfigure}[b]{0.16\textwidth}
         \centering
         \includegraphics[width=\textwidth]{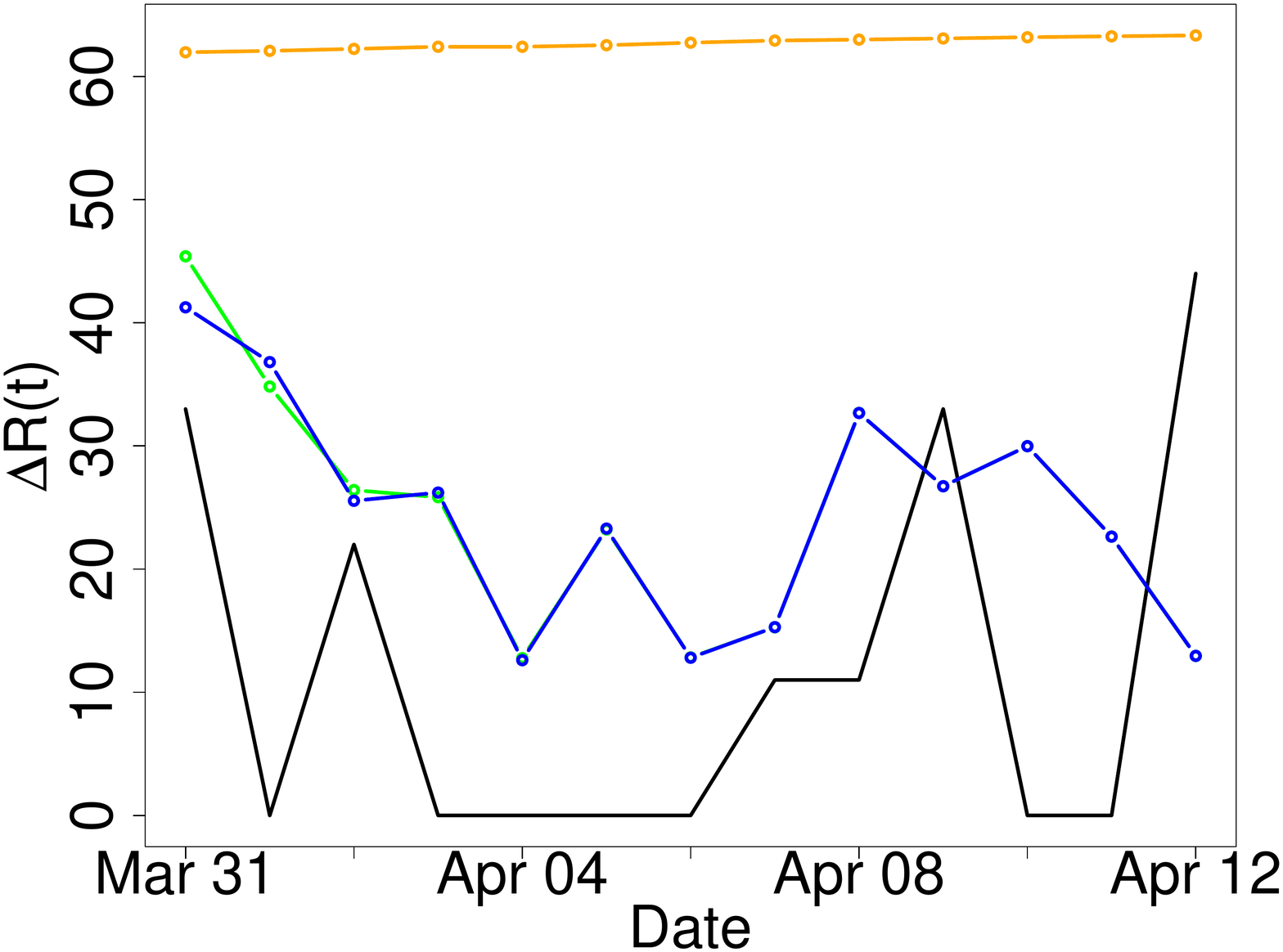}
         \subcaption{Greenville $\widehat{\Delta R}(t)$}
     \end{subfigure}
     \begin{subfigure}[b]{0.16\textwidth}
         \centering
         \includegraphics[width=\textwidth]{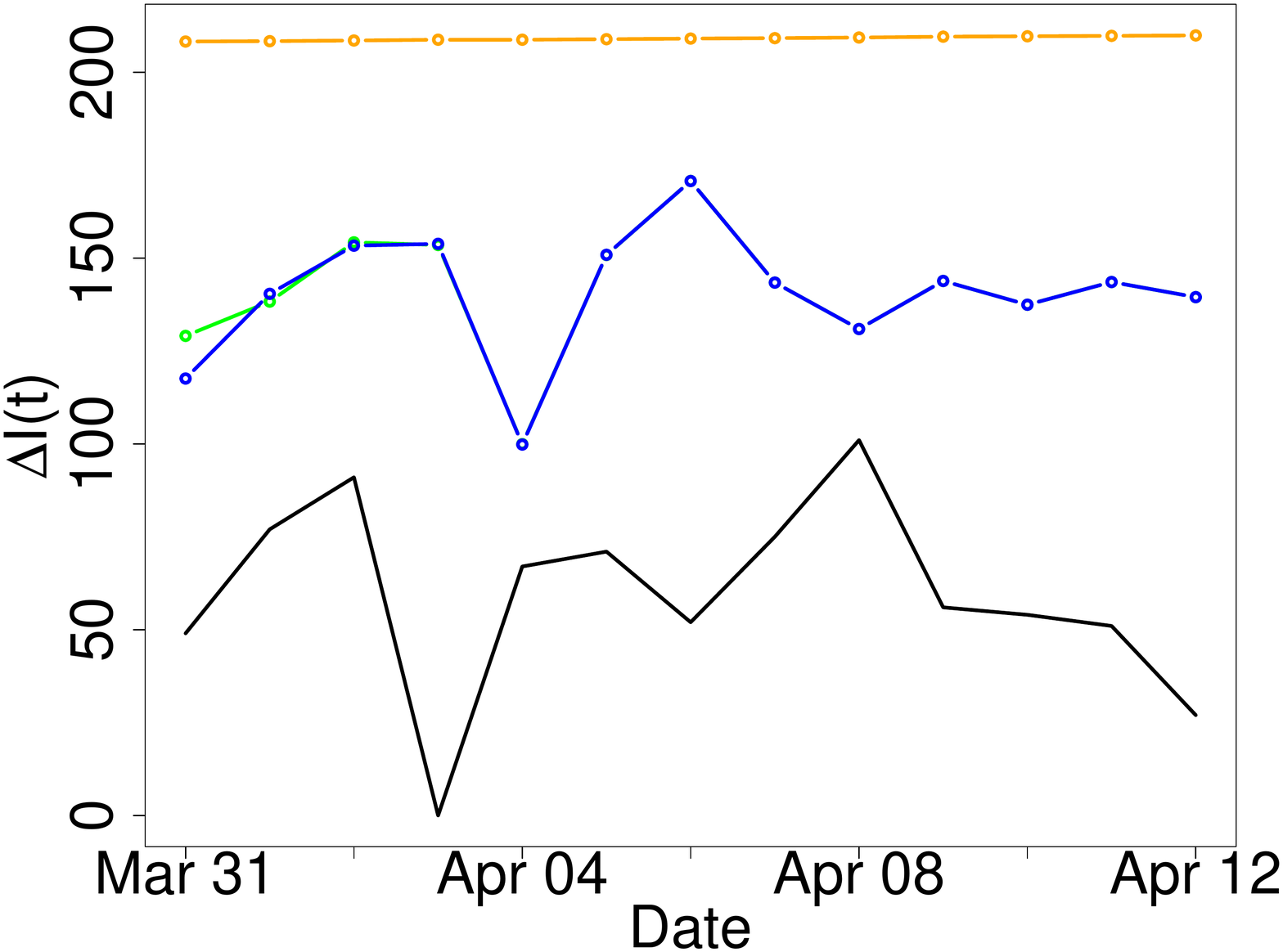}
         \subcaption{Richland $\widehat{\Delta I}(t)$}
     \end{subfigure}
     \begin{subfigure}[b]{0.16\textwidth}
         \centering
         \includegraphics[width=\textwidth]{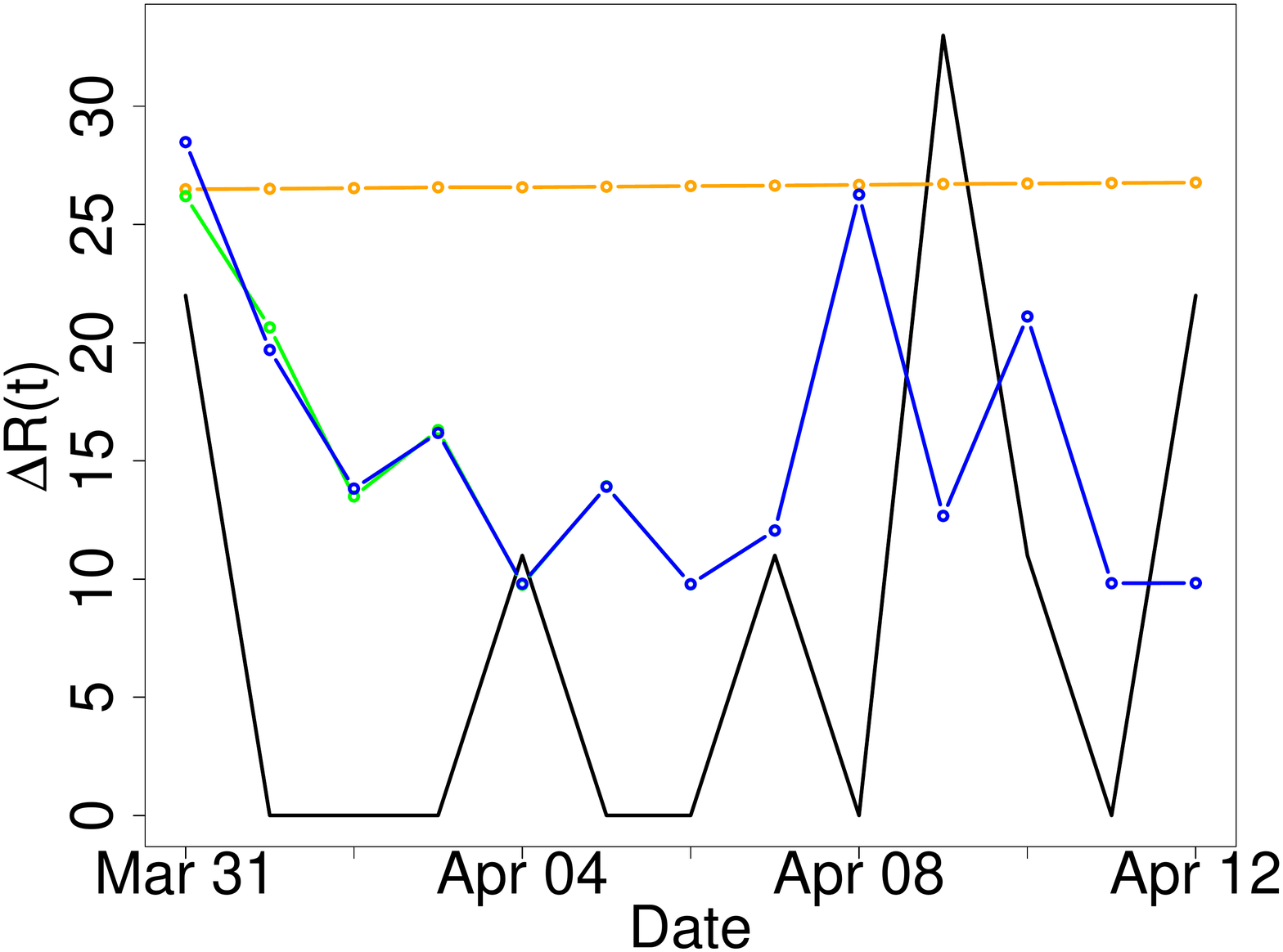}
         \subcaption{Richland $\widehat{\Delta R}(t)$}
     \end{subfigure}
     \begin{subfigure}[b]{0.16\textwidth}
         \centering
         \includegraphics[width=\textwidth]{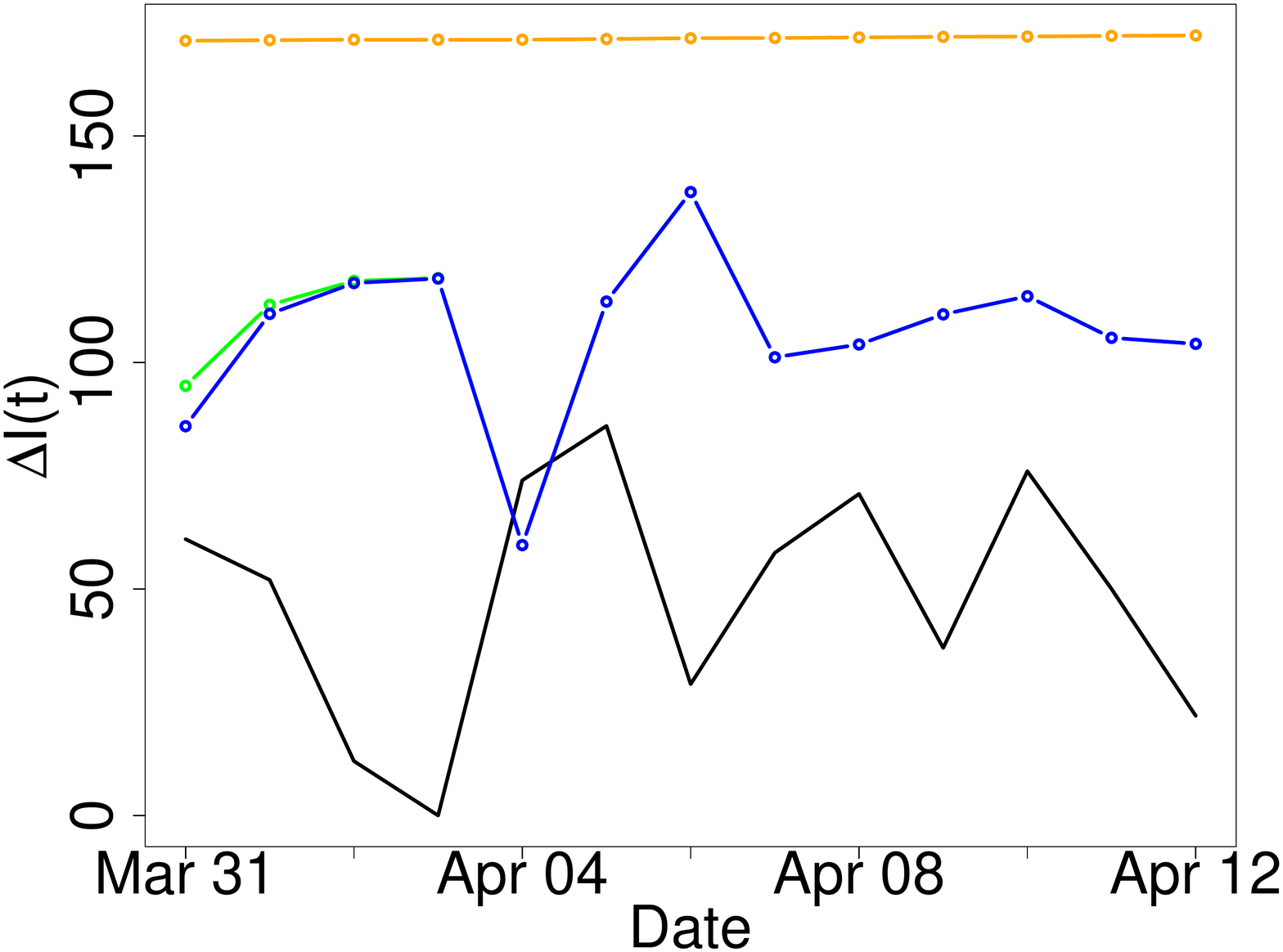}
         \subcaption{Horry $\widehat{\Delta I}(t)$}
     \end{subfigure}
     \begin{subfigure}[b]{0.16\textwidth}
         \centering
         \includegraphics[width=\textwidth]{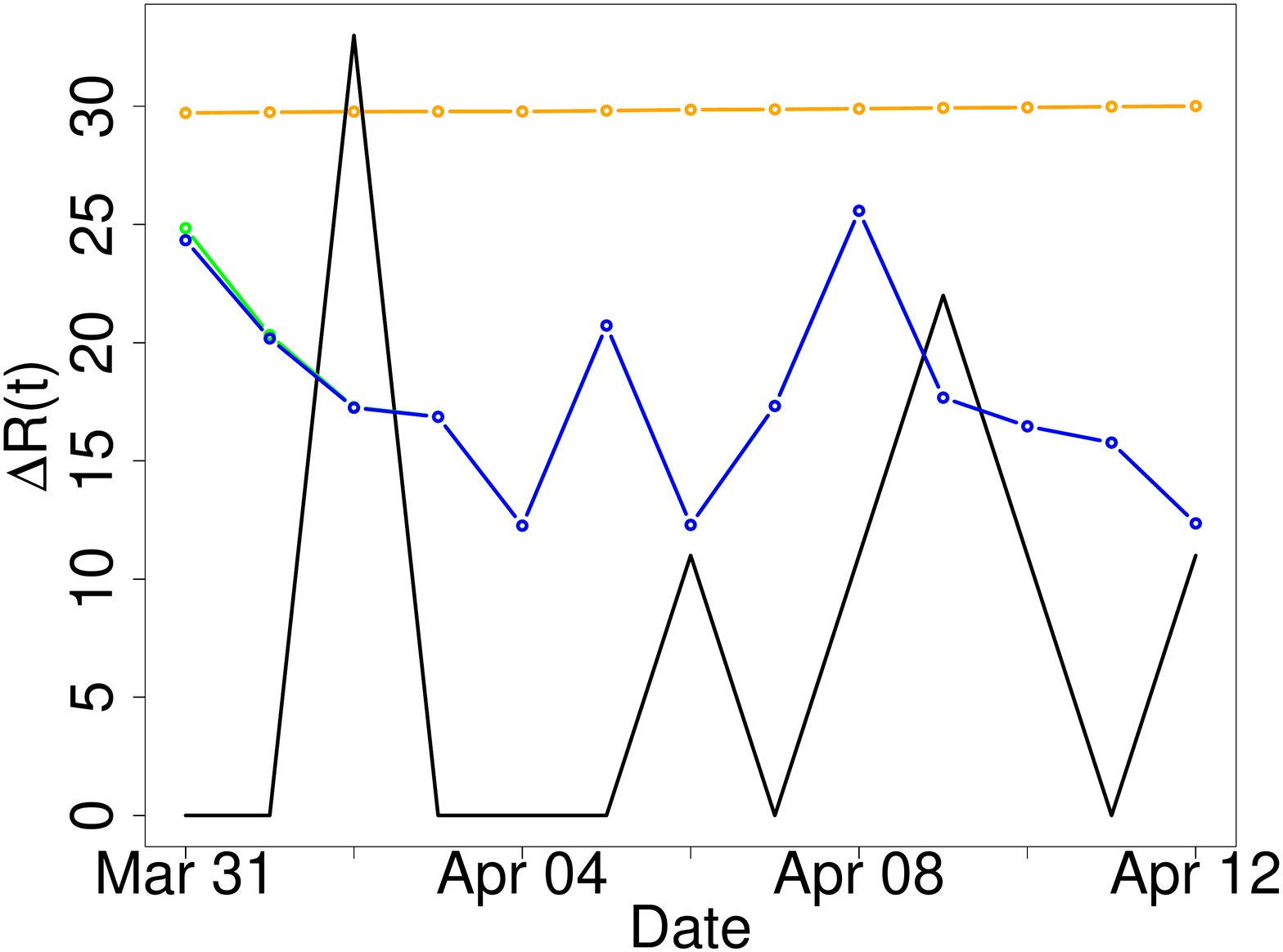}
         \subcaption{Horry $\widehat{\Delta R}(t)$}
     \end{subfigure}
     \begin{subfigure}[b]{0.16\textwidth}
         \centering
         \includegraphics[width=\textwidth]{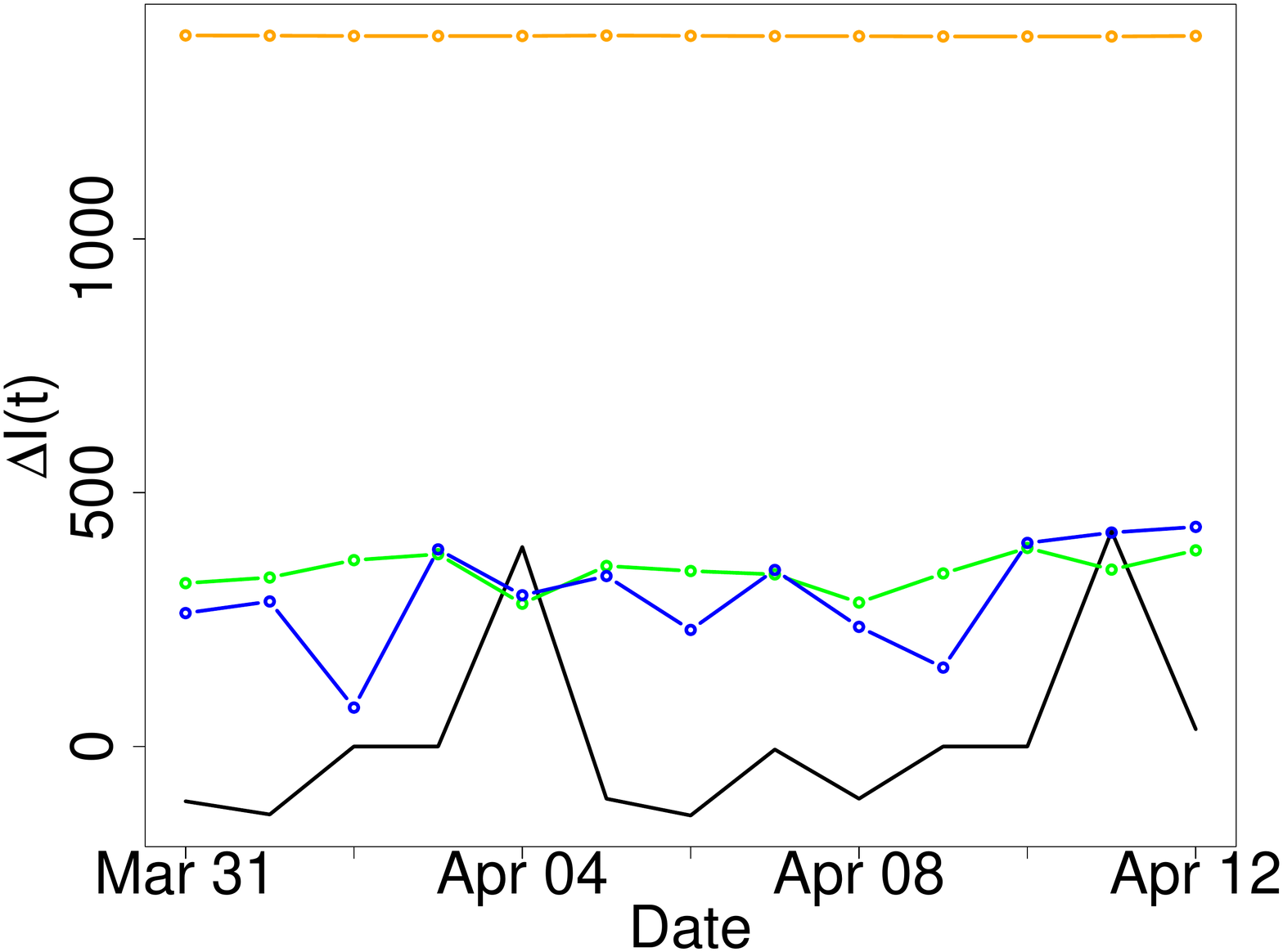}
         \subcaption{Riverside $\widehat{\Delta I}(t)$}
     \end{subfigure}
     \begin{subfigure}[b]{0.16\textwidth}
         \centering
         \includegraphics[width=\textwidth]{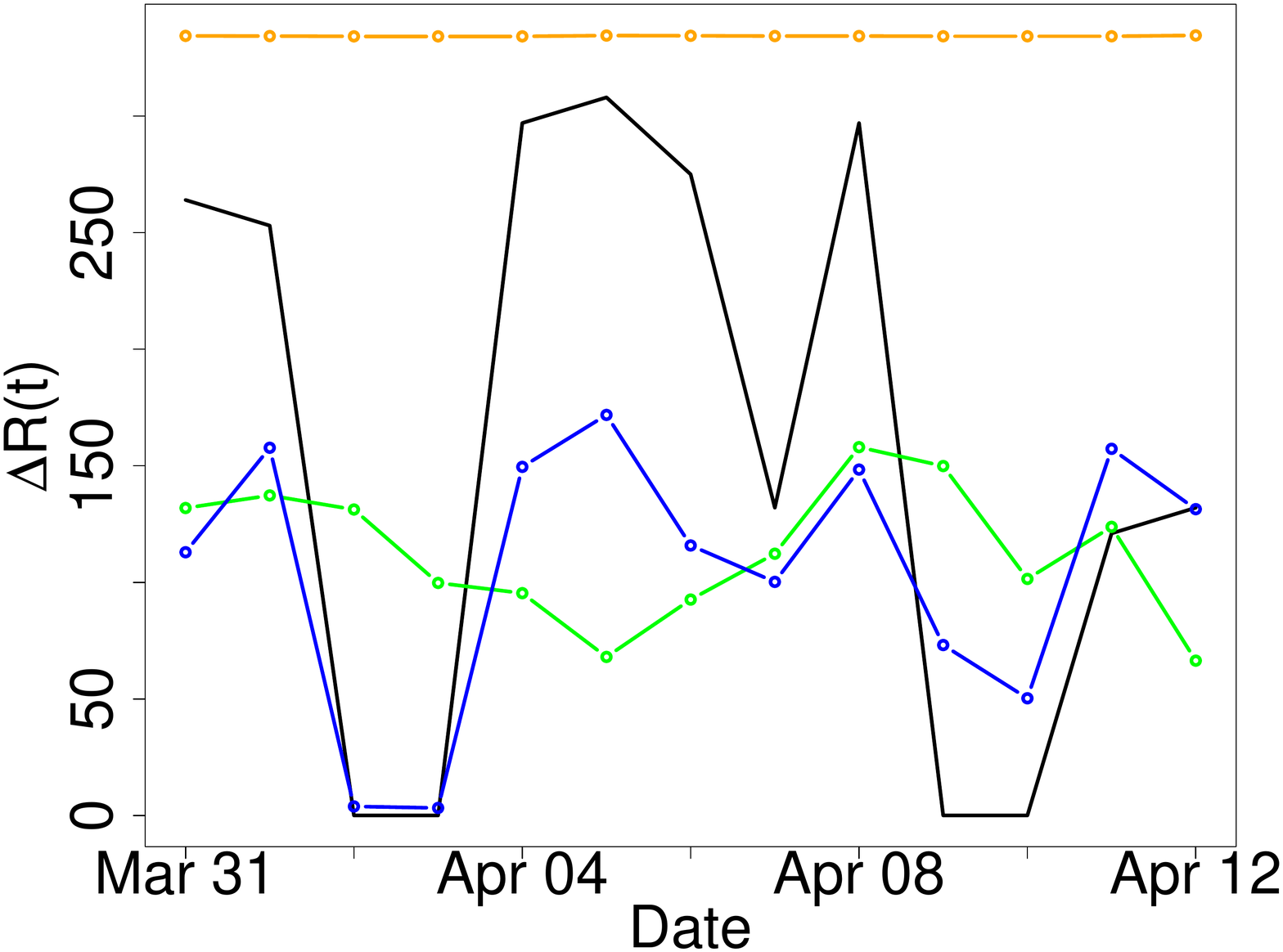}
         \subcaption{Riverside $\widehat{\Delta R}(t)$}
     \end{subfigure}
     \begin{subfigure}[b]{0.16\textwidth}
         \centering
         \includegraphics[width=\textwidth]{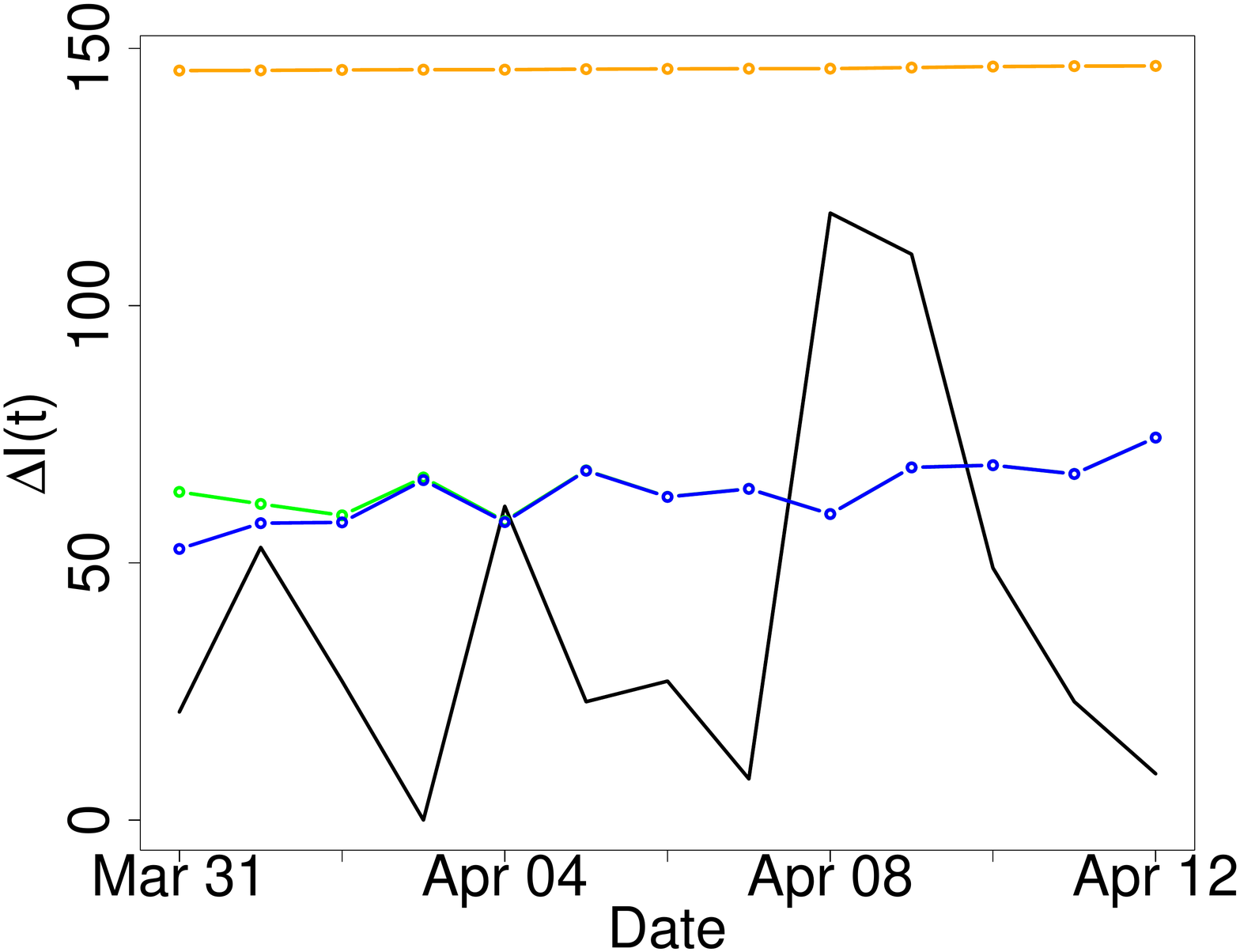}
         \subcaption{Santa Barbara $\widehat{\Delta I}(t)$}
     \end{subfigure}
     \begin{subfigure}[b]{0.16\textwidth}
         \centering
         \includegraphics[width=\textwidth]{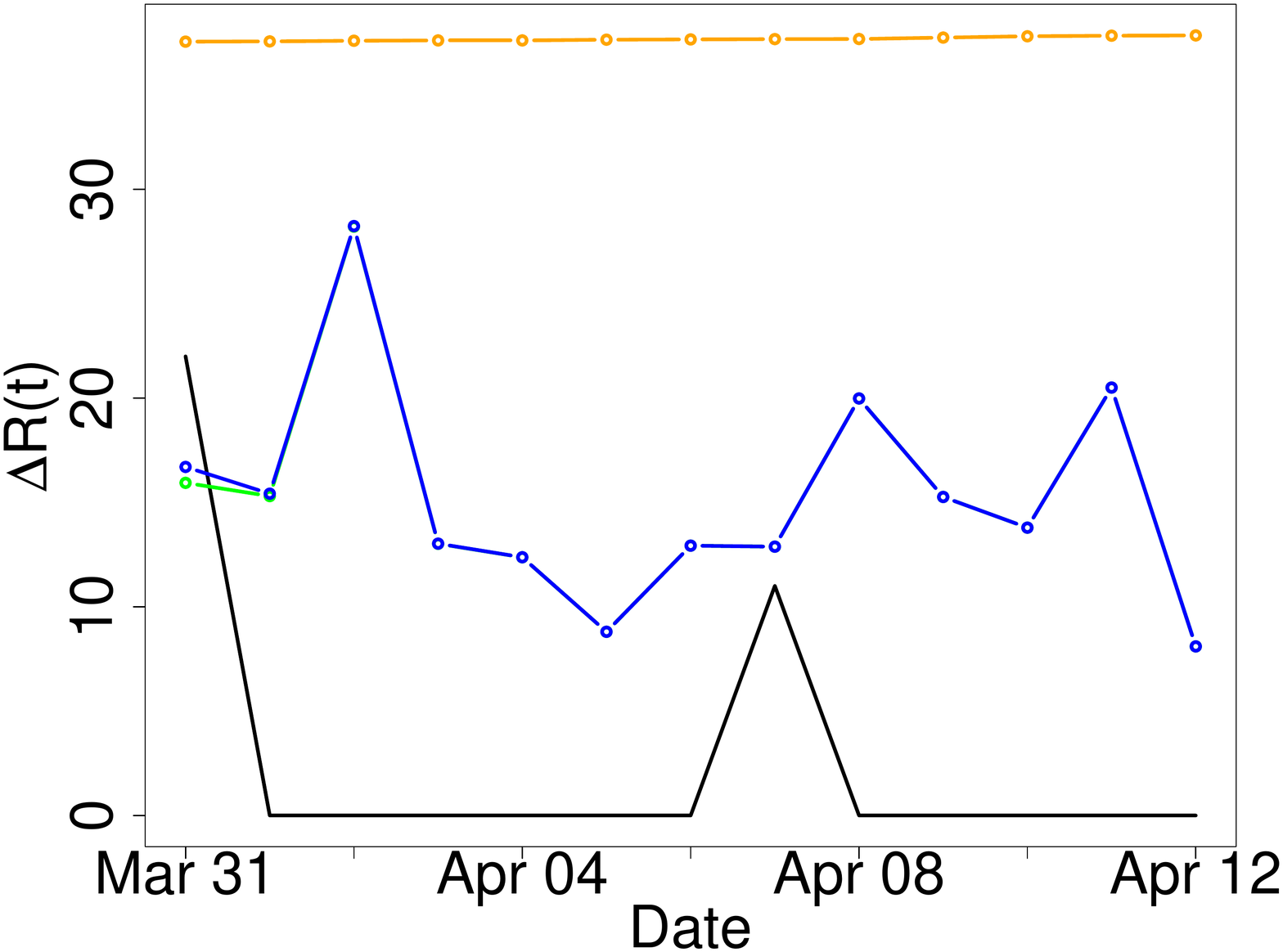}
         \subcaption{Santa Barbara $\widehat{\Delta R}(t)$}
     \end{subfigure}
        \caption{Prediction  for the response variable $y_h$ in three models in selected counties. Black line: true value; orange: Model 1; green: Model 2; blue: Model 3.
        }
        \label{fig:prediction_county}
\end{figure*}

In Figure \ref{fig:prediction_county}, we provide the predicted values in the last 2 weeks from all three different models in nine counties/cities.  
In both the daily infected cases $\Delta I(t)$ and the daily recovered cases $\Delta R(t)$, the predictions by Model 2 and Model 3 perform better than Model 1 in all regions.

\subsection{Details about Heatmap of Transmission Rate and Recovery Rate}\label{sec:heatmap}
In Figure \ref{fig:FL_counties} in the main paper, we provide heatmaps of the transmission and recovery rates in 67 counties in Florida. In this section, the details about the calculations of
 the estimated transmission rate  $\widehat{\beta}$, the estimated recovery rate  $\widehat{\gamma}$  and the corresponding
 standard errors of $\widehat{\beta}$ and $\widehat{\gamma}$ are presented.
 
Consider a linear model within the stationary segment, $Y = X B + \epsilon$, where $B =(\beta, \gamma)^\prime \in \mathbb{R}^{2}$, $Y = (Y_1^\prime, \dots, Y_n^\prime)^\prime \in \mathbb{R}^{2n}$, $X = (X_1^\prime, \dots, X_n^\prime)^\prime \in \mathbb{R}^{2n \times 2}$, $X_t  \in \mathbb{R}^{2 \times 2}$ and $Y_t  \in \mathbb{R}^{2}$ are  defined in \eqref{two_eq_1}. Using linear regression analysis
\cite{seber2012linear}, one can calculate the estimated coefficients and corresponding  standard errors as shown below:
 
 \begin{itemize}
     \item  The estimated transmission rate $\widehat{\beta}$  and recovery rate $\widehat{\gamma}$   are given by
     $ \widehat{B}   = (\widehat{\beta}, \widehat{\gamma})^\prime= (X^\prime X)^{-1}X^\prime Y$.
     \item The standard error of the estimated coefficient vector  $\widehat{B} $
      is given by
     $ SE(\widehat{B} )  = ((X^\prime X)^{-1}\widehat{\sigma}^2)^{\frac{1}{2}}$, 
     where $\widehat{\sigma}^2$ is the sample variance $\widehat{\sigma}^2 = (2n-2)^{-1}\sum_{t=1}^{n}((Y_{t,1} - X_{t,1}\widehat{B})^2 + (Y_{t,2} - X_{t,2}\widehat{B})^2 )$ where $Y_t = (Y_{t,1},Y_{t,2})^\prime$ and $X_{t,i}$ is the $i$-th row of $X$, $i=1,2$. 
     The standard errors of $\widehat{\beta}$ and  $\widehat{\gamma}$
     are the diagonal elements of the  matrix $SE(\widehat{B} )$.
 \end{itemize}

\section{Additional Results of VAR($p$) Model}\label{sec:var-results}

\begin{figure*}[!ht]
     \centering
     \begin{subfigure}[b]{0.19\textwidth}
         \centering
         \includegraphics[width=\textwidth]{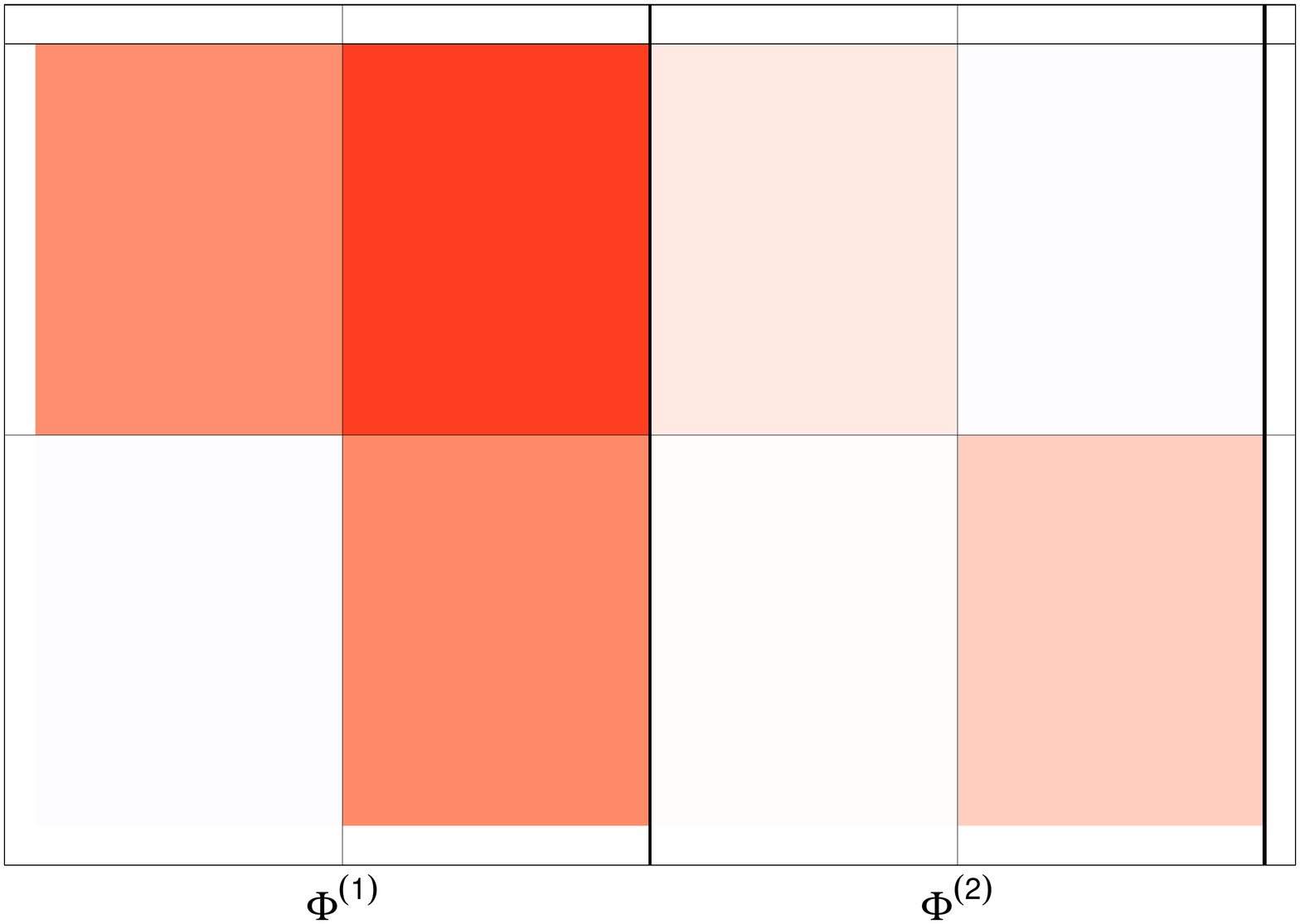}
         \subcaption{NY }
     \end{subfigure}
     \begin{subfigure}[b]{0.19\textwidth}
         \centering
         \includegraphics[width=\textwidth]{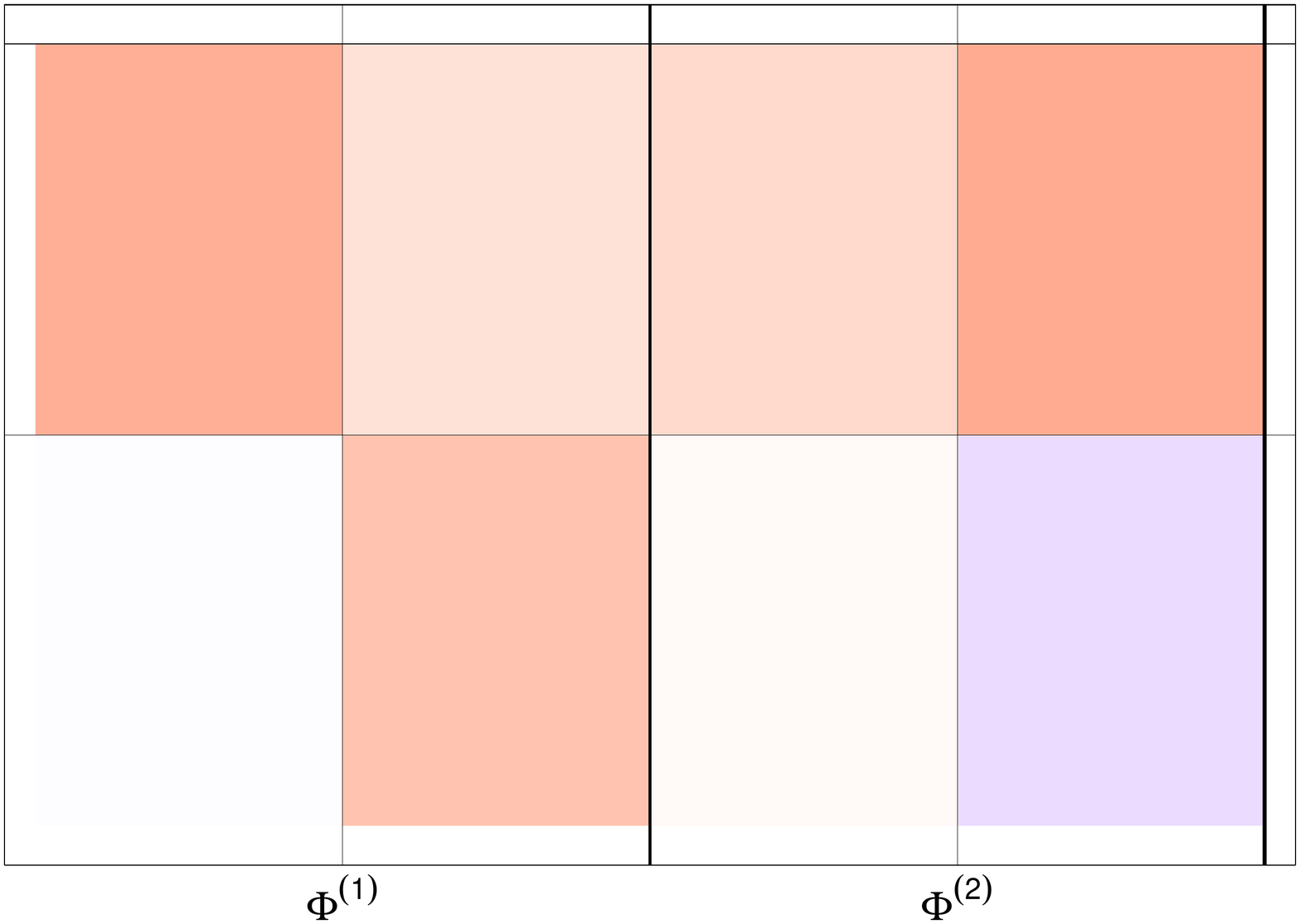}
         \subcaption{OR }
     \end{subfigure}
     \begin{subfigure}[b]{0.19\textwidth}
         \centering
         \includegraphics[width=\textwidth]{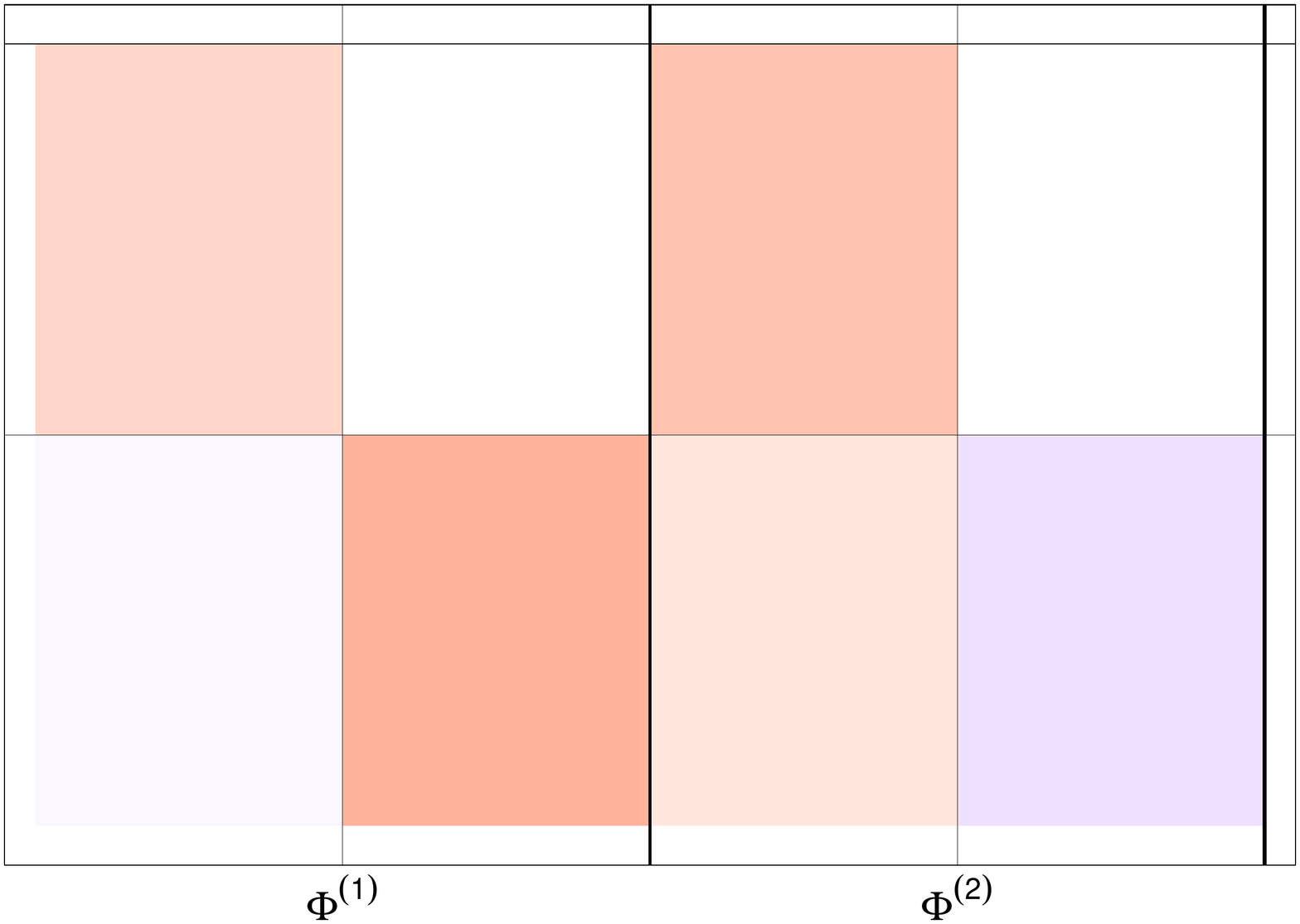}
         \subcaption{FL}
     \end{subfigure}
     \begin{subfigure}[b]{0.19\textwidth}
         \centering
         \includegraphics[width=\textwidth]{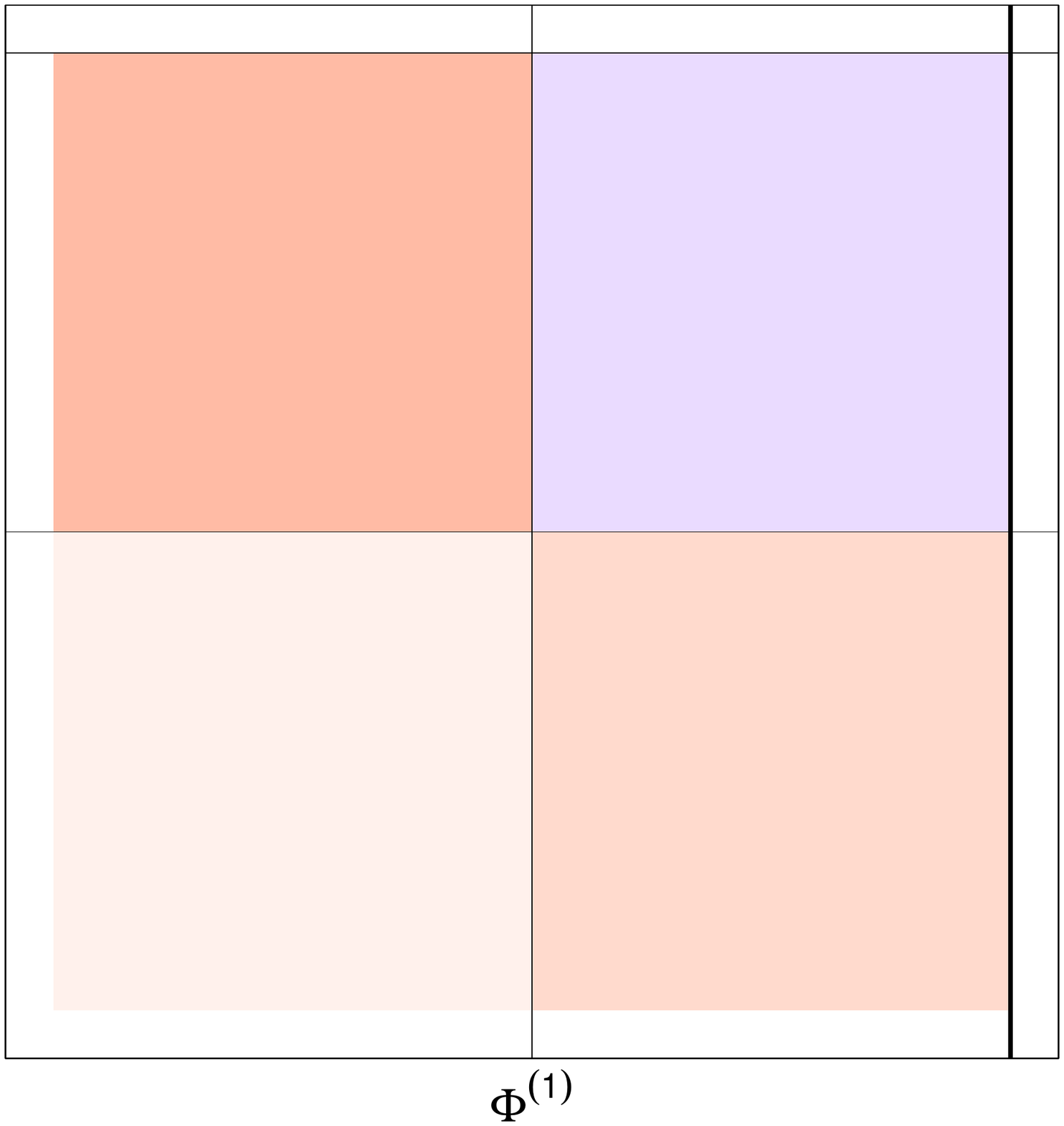}
         \subcaption{CA}
     \end{subfigure}
     \begin{subfigure}[b]{0.19\textwidth}
         \centering
         \includegraphics[width=\textwidth]{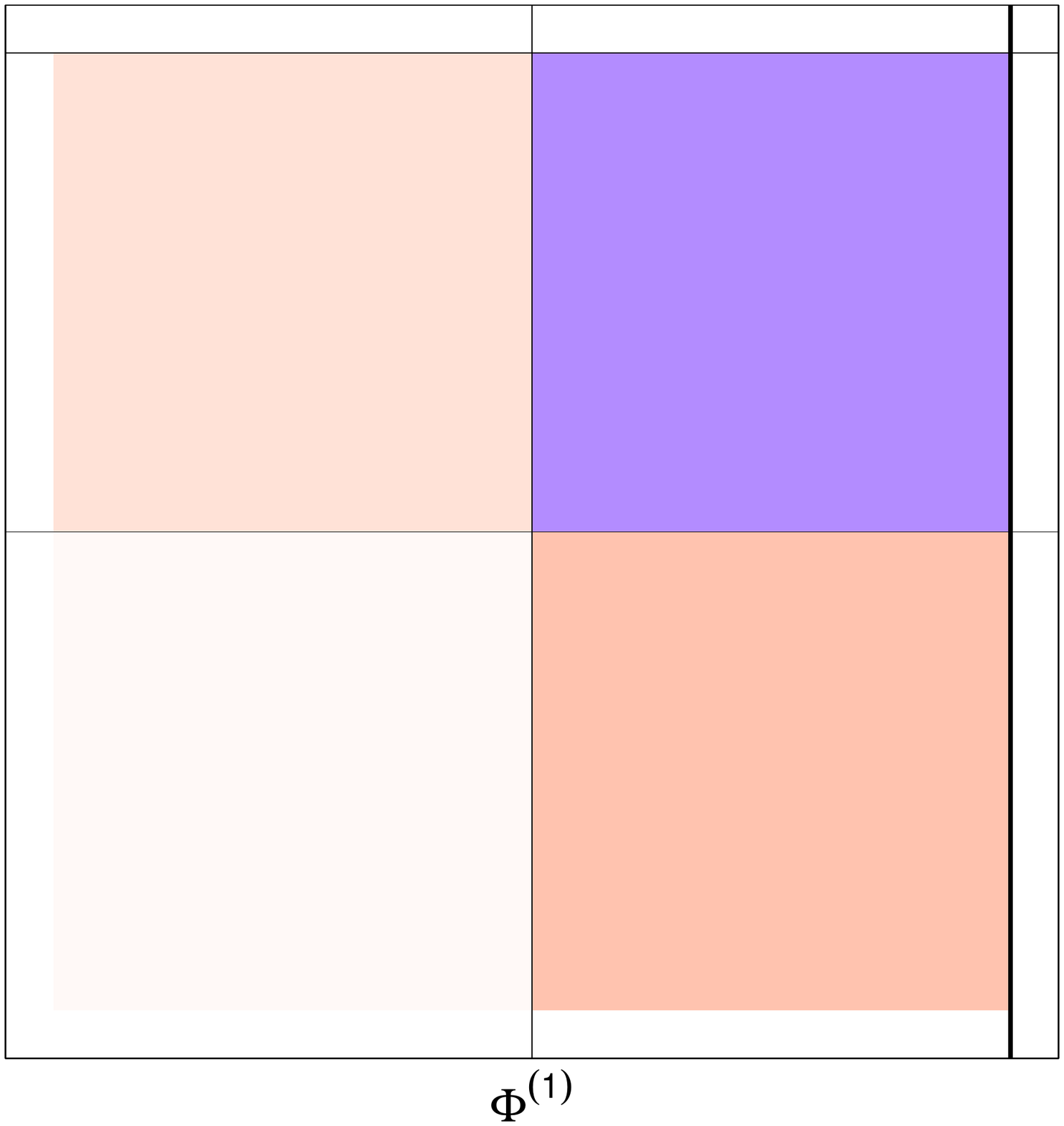}
         \subcaption{TX}
     \end{subfigure}
     
          \begin{subfigure}[b]{0.19\textwidth}
         \centering
         \includegraphics[width=\textwidth]{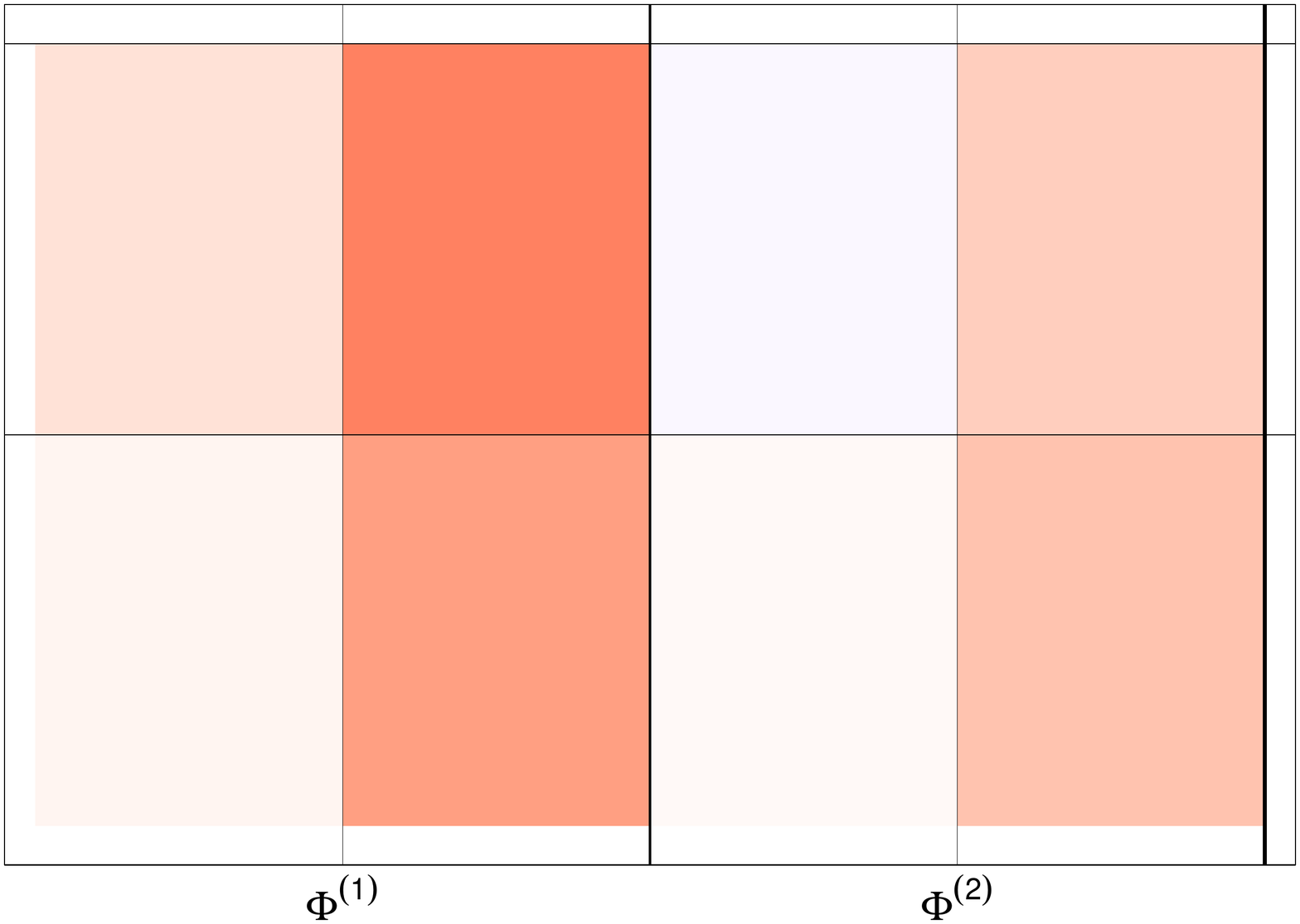}
         \subcaption{NYC}
     \end{subfigure}
     \begin{subfigure}[b]{0.19\textwidth}
         \centering
         \includegraphics[width=\textwidth]{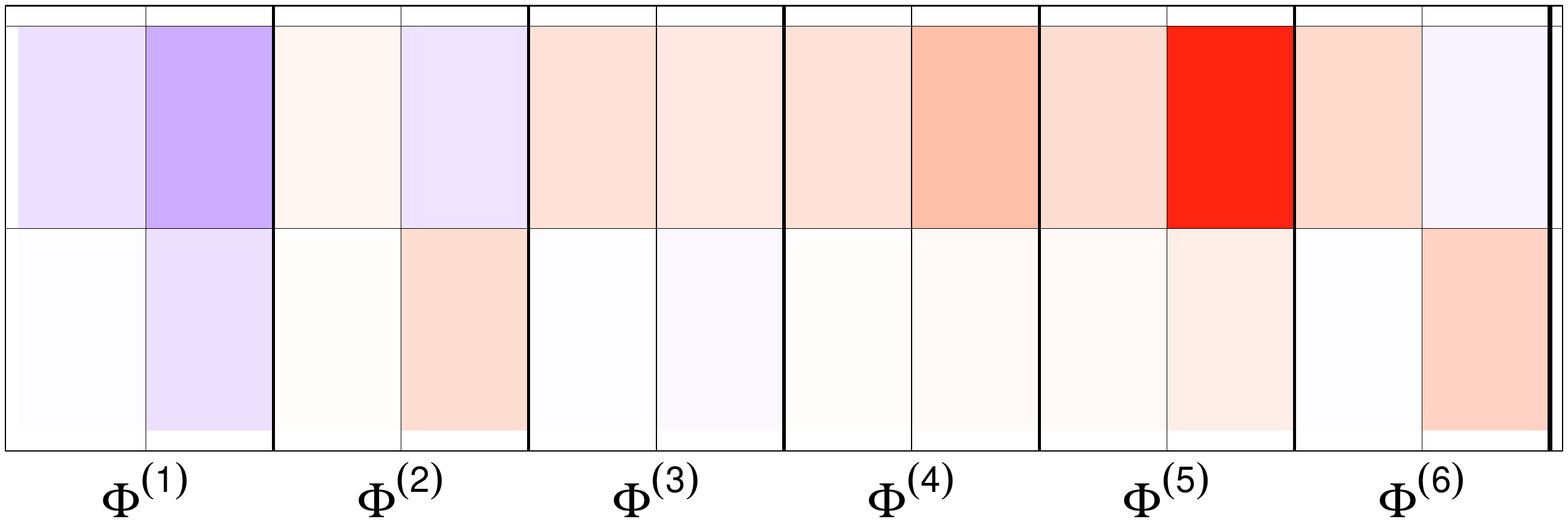}
         \subcaption{King}
     \end{subfigure}
     \begin{subfigure}[b]{0.19\textwidth}
         \centering
         \includegraphics[width=\textwidth]{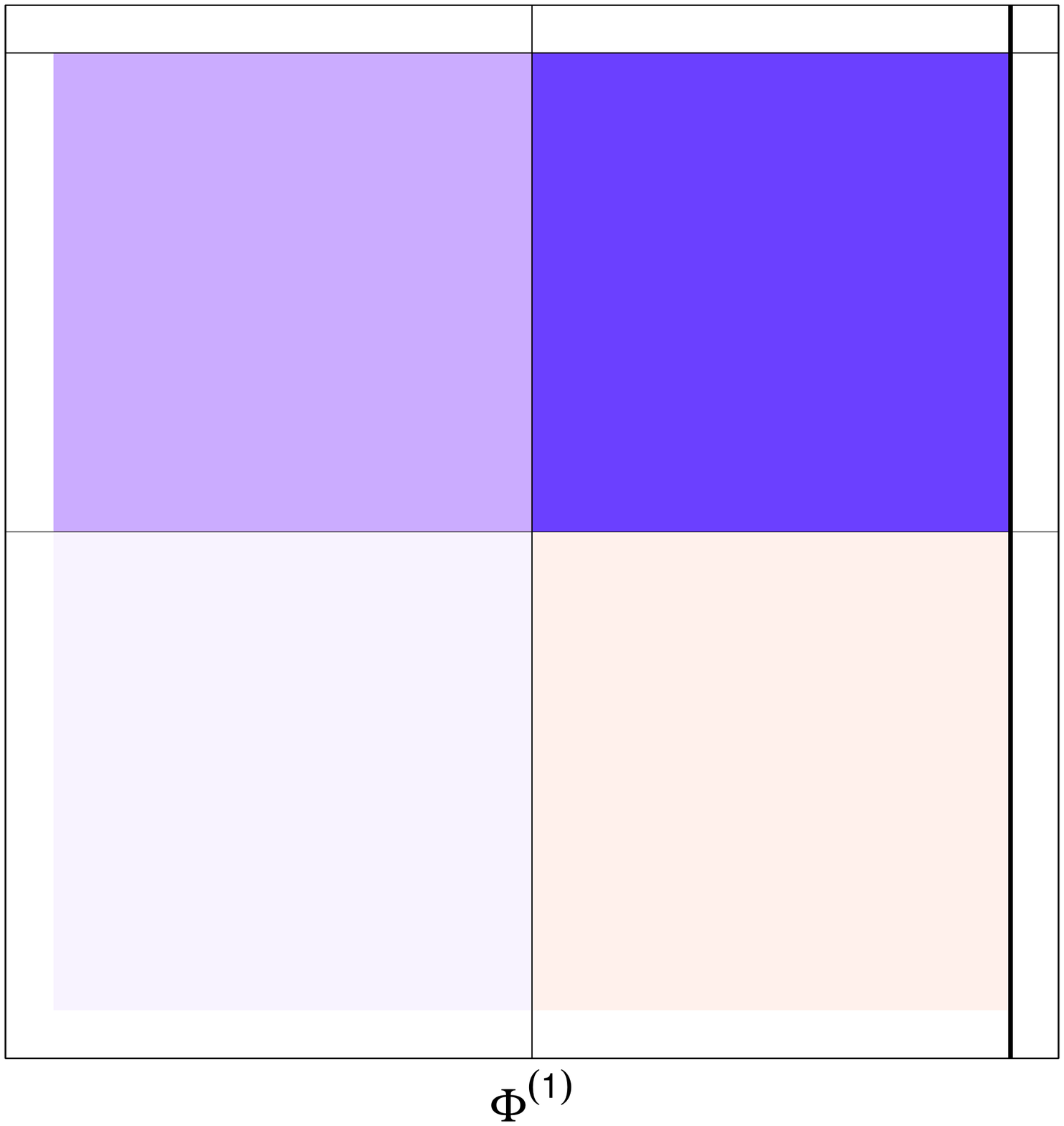}
         \subcaption{Miami-Dade}
     \end{subfigure}
     \begin{subfigure}[b]{0.19\textwidth}
         \centering
         \includegraphics[width=\textwidth]{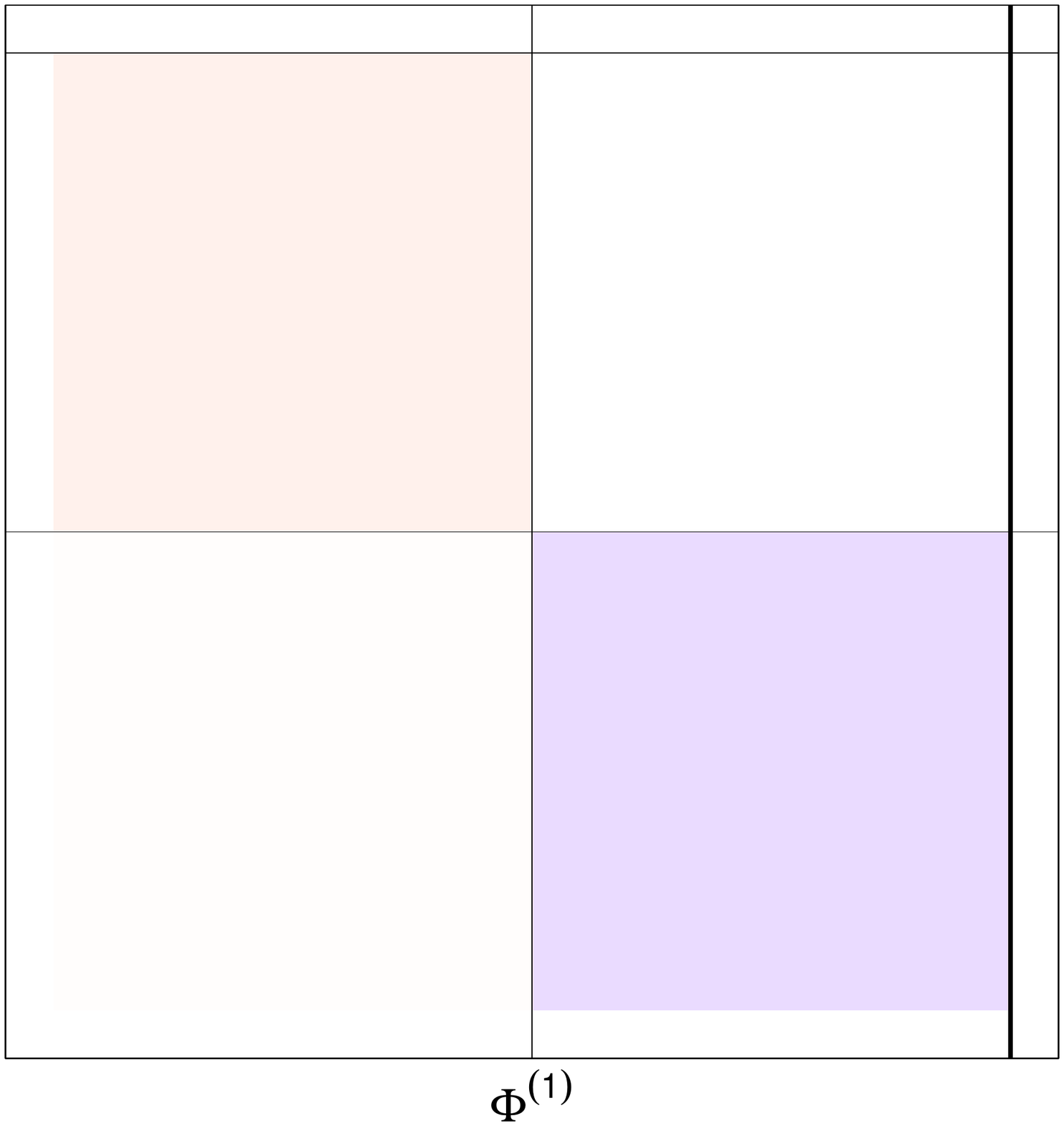}
         \subcaption{Charleston}
     \end{subfigure}
     \begin{subfigure}[b]{0.19\textwidth}
         \centering
         \includegraphics[width=\textwidth]{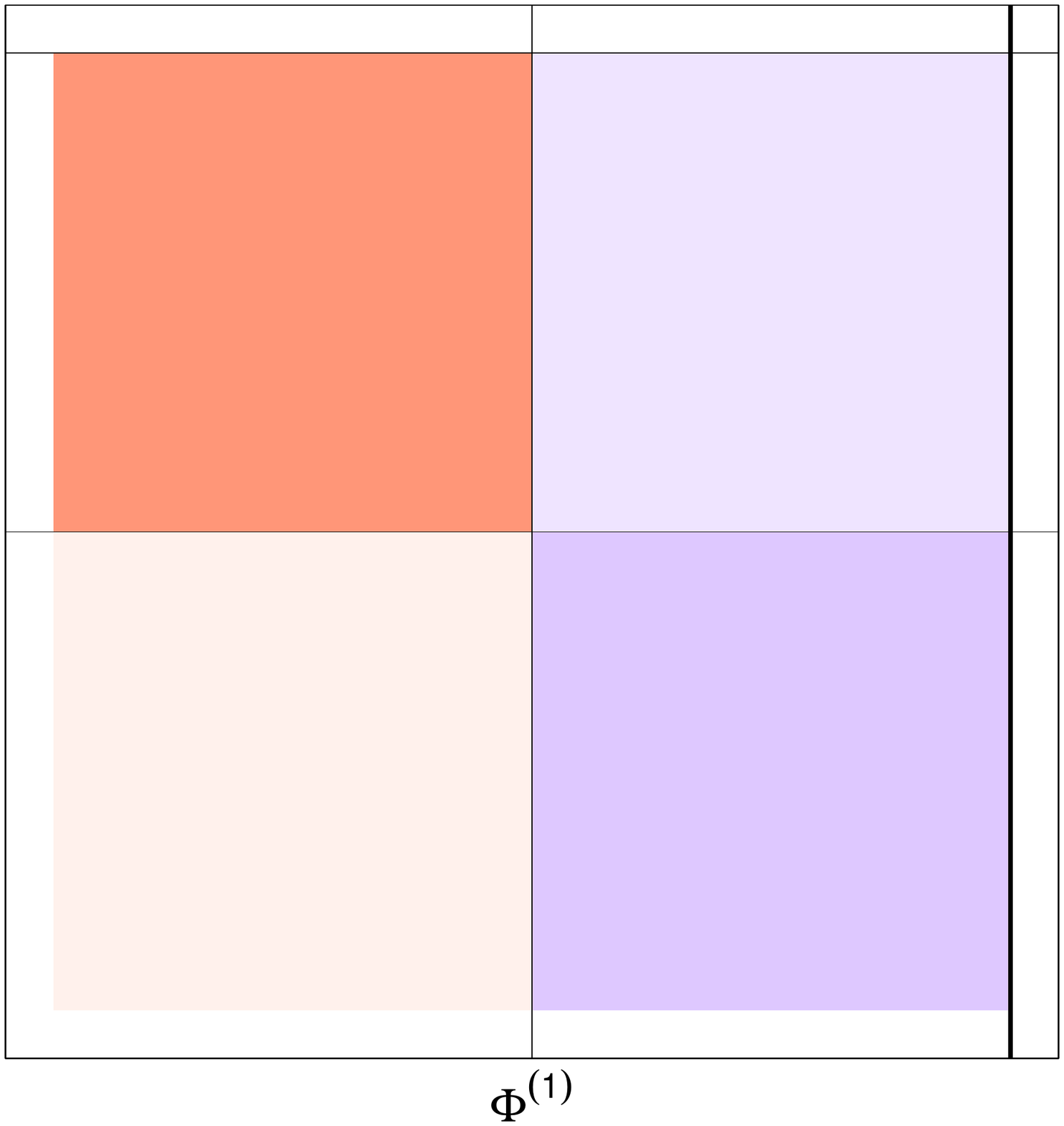}
         \subcaption{Greenville}
     \end{subfigure}
     \begin{subfigure}[b]{0.19\textwidth}
         \centering
         \includegraphics[width=\textwidth]{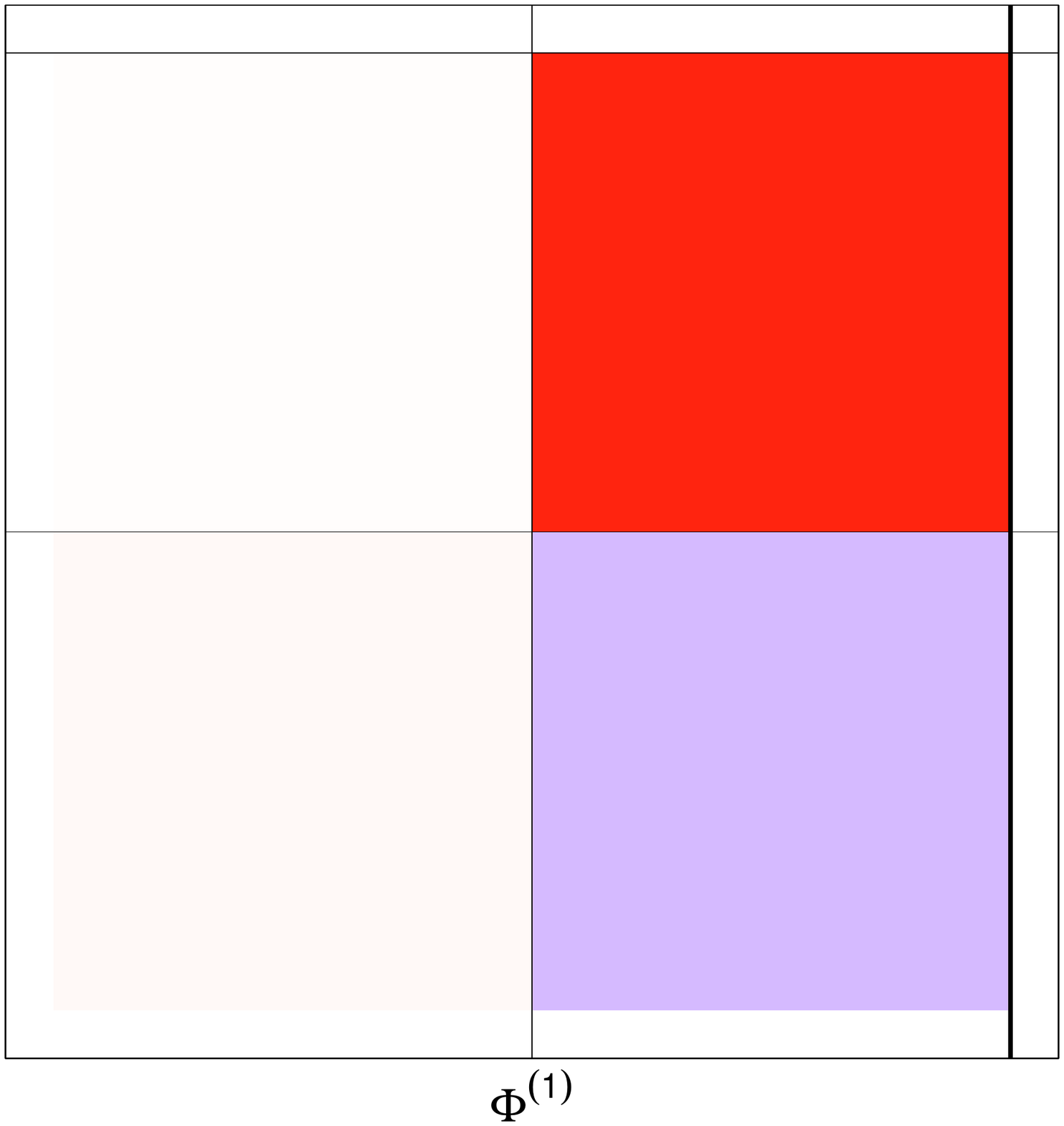}
         \subcaption{Richland}
     \end{subfigure}
     \begin{subfigure}[b]{0.19\textwidth}
         \centering
         \includegraphics[width=\textwidth]{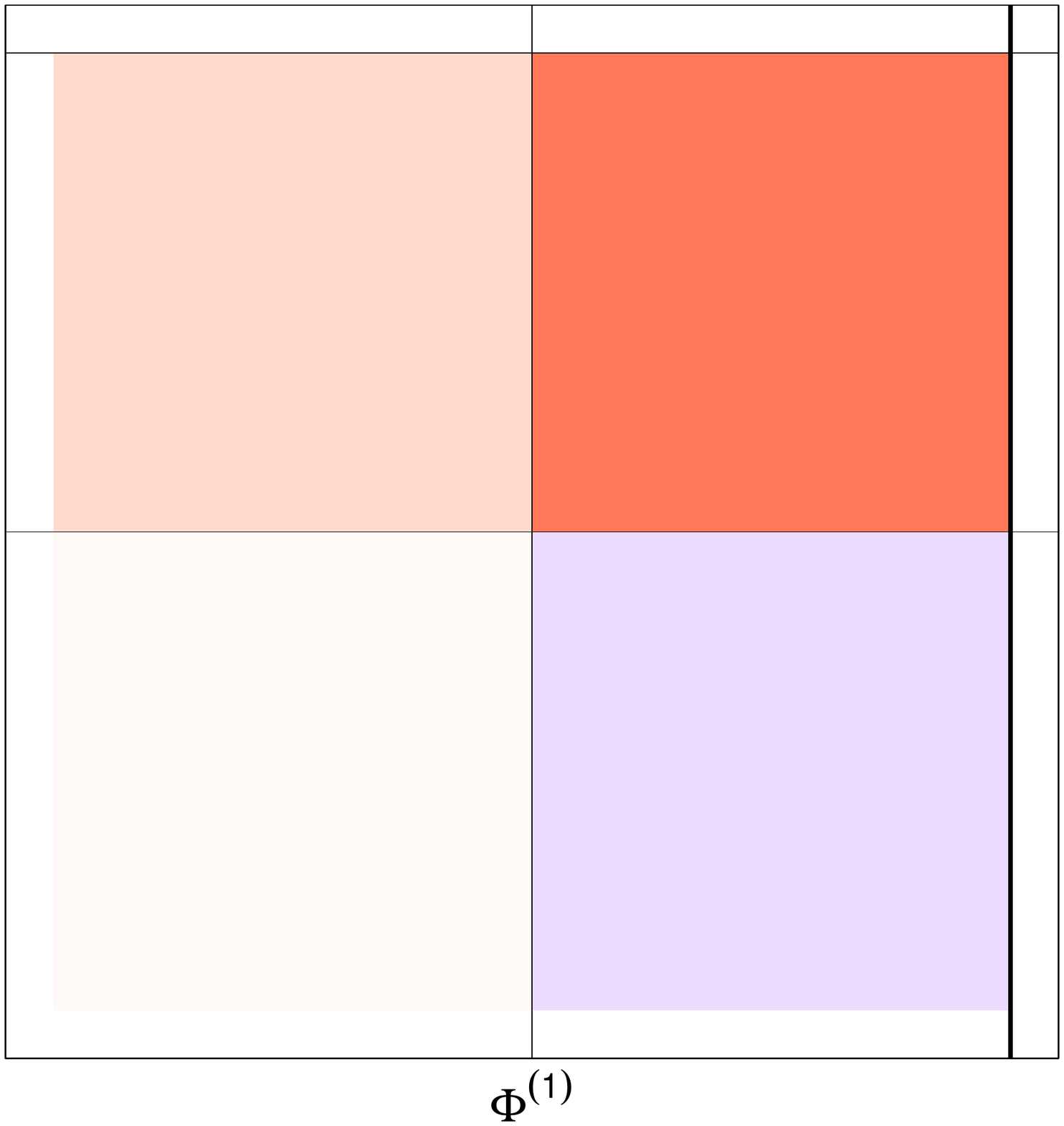}
         \subcaption{Horry}
     \end{subfigure}
     \begin{subfigure}[b]{0.19\textwidth}
         \centering
         \includegraphics[width=\textwidth]{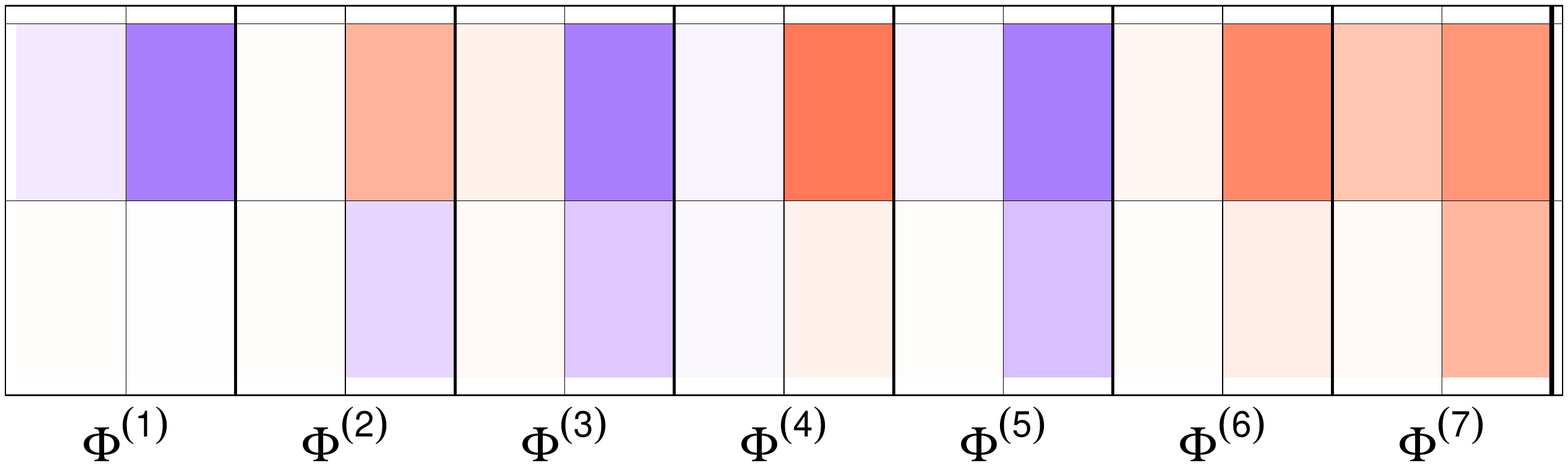}
         \subcaption{Riverside}
     \end{subfigure}
     \begin{subfigure}[b]{0.19\textwidth}
         \centering
         \includegraphics[width=\textwidth]{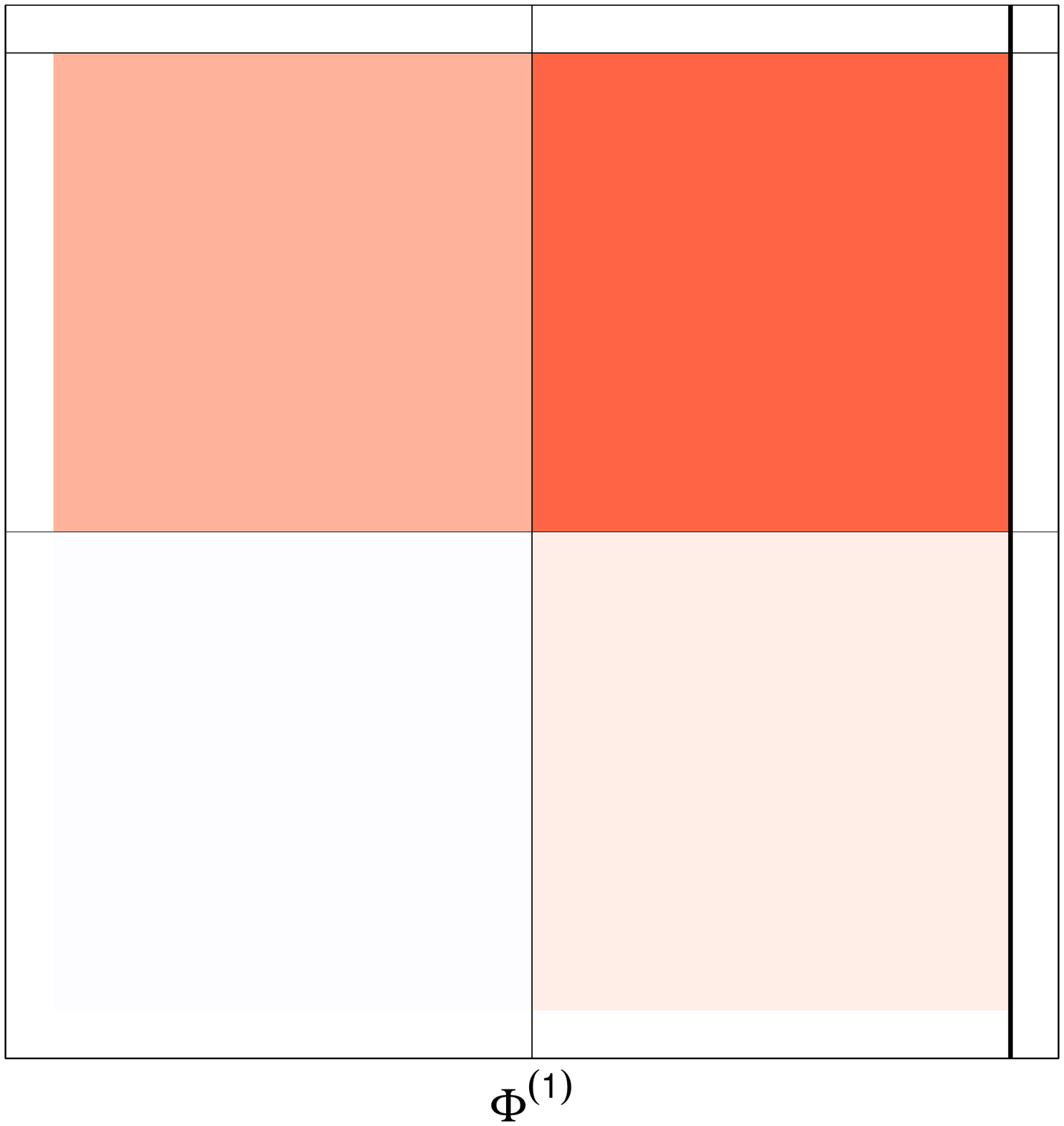}
         \subcaption{Santa Barbara}
     \end{subfigure}
        \caption{VAR($p$) parameters matrices of residuals
        in selected regions by Model 2.3 (piecewise constant SIR model + spatial effect). The red color stands for the positive value, the blue color stands for the negative value, and the white color stands for zero.}
        \label{fig:residual_coef}
\end{figure*}

\begin{table*}[!ht]
\caption{\label{table_cp_residual}
The detected change points of the residuals from Model 2.3 by block fused lasso in the training dataset (excluding the latest 2 weeks data). The BIC criterion is used to select the lag $p$.  }
\tiny
\centering
\begin{tabular}{lcccccccccc} 
  \hline
    & NY & OR & FL & CA & TX  \\
    \hline
Change point (date) & -& March 11 2020, May 01 2020 & Aug 27 2020 & July 30 2020& May 16 2020, July 20 2020 \\
 \hline
 \hline
  & NYC& {King} & {Miami}  \\
  \hline
 Change point (date) & March 19 2020 & - & \multicolumn{2}{c}{March 29 2020, May 08 2020, May 26 2020, July 20 2020}\\
  \hline
 \hline
  &  {Charleston} & {Greenville}  \\
  \hline
Change point (date) & -&\multicolumn{3}{c}{ March 24 2020, May 09 2020, Aug 03 2020, Aug 27 2020}\\
  \hline
 \hline
& Riverside & Santa Barbara & {Richland}  & & {Horry}  \\
  \hline
Change point (date)  & - & - & \multicolumn{2}{c}{April 05 2020, April 10 2020}, June 10 2020 & March 26 2020, May 05 2020 \\
\hline
\end{tabular}
\end{table*}

Before fitting the VAR($p$) model, we first use the block fused lasso method to check whether there are any change points in the structure of the VAR($p$) parameters. If there are some change points in the VAR component, we use the last segment refitted estimated VAR parameters for forecasting. The detected change points are shown in Table \ref{table_cp_residual}. In state-level, we detect a change point in Oregon, Florida, California and Texas; in county-level, we detect some change points in New York City, Miami-Dade County,  Greenville County, Richland County and Horry County. The estimated lag $p$ and VAR($p$) parameters matrices for each state and counties are displayed in Figure~\ref{fig:residual_coef}.


%

\ifCLASSOPTIONcaptionsoff
  \newpage
\fi

\end{document}